\newcommand{\imp}{\mathrm{imp}}

\documentclass[12pt]{article}
\usepackage[colorlinks=true, linkcolor=blue, citecolor=blue, urlcolor=black, pdfpagemode={UseNone}]{hyperref}
\usepackage{graphicx}

\usepackage{natbib}
\usepackage{multibib}
\newcites{supp}{Supplementary References}

\usepackage{amsmath}
\usepackage{amssymb,bm,amsthm}

\usepackage{multirow}
\usepackage{dcolumn,booktabs}
\usepackage{algorithm,algorithmicx,algpseudocode,algcompatible}
\usepackage{mathtools}
\usepackage{authblk}
\usepackage{url}
\usepackage[utf8]{inputenc}
\usepackage{textgreek}
\usepackage{enumerate}
\usepackage{enumitem}
\usepackage[font=small]{caption}

\usepackage{tikz}
\usetikzlibrary{arrows.meta,positioning}
\usepackage{titletoc}

\newtheorem{theorem}{Theorem}
\newtheorem{proposition}{Proposition}
\newtheorem{lemma}{Lemma}

\newtheorem{corollary}{Corollary}
\newtheorem{definition}{Definition}

\algrenewcommand{\algorithmiccomment}[1]{\hfill$\blacktriangleright$ #1}

\graphicspath{{Figures}}

\DeclareMathOperator{\Var}{Var}
\DeclareMathOperator{\Cov}{Cov}

\DeclareMathOperator{\E}{E}

\DeclareMathOperator*{\argmin}{argmin}  
 
\DeclareMathOperator*{\argzero}{argzero} 
 
\DeclareMathOperator{\Tr}{Tr}
\DeclareMathOperator{\diag}{diag}

\DeclareMathOperator{\IFu}{IF}
\DeclareMathOperator{\vect}{vec}
\DeclareMathOperator{\vecth}{vech}

\DeclareMathOperator{\interior}{int}
\newcommand{\ind}{I}

\newcommand{\ki}{\ell}

\newcommand{\MCD}{\text{\tiny\upshape MCD}}
\newcommand{\PD}{\text{\upshape PD}}
\newcommand{\uPD}{\mbox{\tiny{\upshape PD}}} 
\newcommand{\RD}{\text{\upshape RD}}
\newcommand{\uRD}{\mbox{\tiny{\upshape RD}}} 
\newcommand{\rk}{k}

\newcommand{\rt}{t}
\newcommand{\ximp}{x^{\mbox{\scriptsize{\upshape imp}}}}
\newcommand{\bximp}{\bx^{\mbox{\scriptsize{\upshape imp}}}}

\newcommand{\sub}{k}
\newcommand{\zort}{\boldsymbol{z}^{\perp}}

\newcommand{\cell}{\mbox{\scriptsize{\upshape cell}}} 
\newcommand{\case}{\mbox{\scriptsize{\upshape case}}} 

\newcommand{\CD}{{\boldsymbol{\cdot}}}
\newcommand{\eps}{\varepsilon}
\newcommand{\beps}{\boldsymbol\eps}

\newcommand{\hpi}{\widehat{\pi}}
\newcommand{\hsigma}{\widehat{\sigma}}
\newcommand{\tsigma}{\widetilde{{\sigma}}}

\newcommand{\tl}{\widetilde{\ell}}

\newcommand{\tn}{\widetilde{n}}

\newcommand{\hx}{\widehat{x}}

\newcommand{\hy}{\widehat{y}}
\newcommand{\hz}{\widehat{z}}
\newcommand{\htheta}{\widehat{\theta}}

\newcommand{\bzero}{\boldsymbol 0}
\newcommand{\bone}{\boldsymbol 1}

\newcommand{\ba}{\boldsymbol a}

\newcommand{\bb}{\boldsymbol b}

\newcommand{\bhb}{\boldsymbol{\widehat{b}}}
\newcommand{\bmed}{\boldsymbol{m}}

\newcommand{\bc}{\boldsymbol c}
\newcommand{\hc}{\hat c}
\newcommand{\bE}{\boldsymbol E}

\newcommand{\be}{\boldsymbol e}

\newcommand{\bG}{\boldsymbol G}

\newcommand{\bL}{\boldsymbol L}

\newcommand{\bp}{\boldsymbol p}
\newcommand{\br}{\boldsymbol r}
\newcommand{\hr}{\widehat{r}}

\newcommand{\hq}{\hat q}
\newcommand{\bs}{\boldsymbol s}
\newcommand{\bt}{\boldsymbol t}
\newcommand{\hti}{\widehat t}
\newcommand{\bu}{\boldsymbol u}
\newcommand{\bhu}{\boldsymbol{\widehat{u}}}

\newcommand{\bv}{\boldsymbol v}

\newcommand{\bw}{\boldsymbol w}

\newcommand{\bx}{\boldsymbol x}

\newcommand{\bhx}{\boldsymbol{\widehat{x}}}
\newcommand{\by}{\boldsymbol y}

\newcommand{\bhy}{\boldsymbol{\widehat{y}}}
\newcommand{\bz}{\boldsymbol z}
\newcommand{\bhz}{\boldsymbol{\widehat{z}}}
\newcommand{\btz}{\boldsymbol{\widetilde{z}}}
\newcommand{\tz}{\widetilde{z}}
\newcommand{\bA}{\boldsymbol A}
\newcommand{\bB}{\boldsymbol B}
\newcommand{\bhB}{\boldsymbol{\widehat{B}}}

\newcommand{\bC}{\boldsymbol C}
\newcommand{\hC}{\widehat C}
\newcommand{\bD}{\boldsymbol D}
\newcommand{\bH}{\boldsymbol{H}}
\newcommand{\bI}{\boldsymbol I}
\newcommand{\bJ}{\boldsymbol J}
\newcommand{\bK}{\boldsymbol K}

\newcommand{\bM}{\boldsymbol M}
\newcommand{\bN}{\boldsymbol{ N}}

\newcommand{\bO}{\boldsymbol O}
\newcommand{\bP}{\boldsymbol P}
\newcommand{\hP}{\widehat P}
\newcommand{\bhP}{\boldsymbol{\widehat P}}
\newcommand{\bPhi}{\boldsymbol{\Phi}}
\newcommand{\bQ}{\boldsymbol Q}
\newcommand{\bR}{\boldsymbol R}
\newcommand{\btr}{\boldsymbol{\widetilde{r}}}

\newcommand{\bS}{\boldsymbol S}
\newcommand{\bhS}{\boldsymbol{\widehat{S}}}
\newcommand{\bT}{\boldsymbol T}

\newcommand{\bU}{\boldsymbol U}

\newcommand{\bhU}{\boldsymbol{\widehat{U}}}

\newcommand{\bV}{\boldsymbol V}
\newcommand{\btV}{\boldsymbol{\widetilde{V}}}
\newcommand{\bh}{\boldsymbol h}
\newcommand{\bhV}{\boldsymbol{\widehat{V}}}
\newcommand{\bW}{\boldsymbol W}
\newcommand{\btW}{\boldsymbol{\widetilde{W}}}
\newcommand{\bX}{\boldsymbol X}

\newcommand{\bY}{\boldsymbol Y}
\newcommand{\bZ}{\boldsymbol Z}
\newcommand{\bhZ}{\boldsymbol{\widehat{Z}}}
\newcommand{\btZ}{\boldsymbol{\widetilde{Z}}}
\newcommand{\tZ}{\widetilde{Z}}

\newcommand{\bg}{\boldsymbol g}

\newcommand{\bdelta}{\boldsymbol \delta}

\newcommand{\boeta}{\boldsymbol{\eta}}
\newcommand{\boheta}{\boldsymbol{\widehat{\eta}}}

\newcommand{\btheta}{\bm \theta}
\newcommand{\bTheta}{\boldsymbol \Theta}

\newcommand{\pibar}{\overline \pi}

\newcommand{\bhtheta}{\boldsymbol{\widehat{\theta}}}

\newcommand{\bmu}{\boldsymbol \mu}

\newcommand{\bhmu}{\boldsymbol{\widehat{\mu}}}
\newcommand{\btmu}{\boldsymbol{\widetilde{\mu}}}
\newcommand{\bsigma}{\boldsymbol \sigma}
\newcommand{\bhsigma}{ \boldsymbol{\widehat{\sigma}}}
\newcommand{\bSigma}{\boldsymbol \Sigma}
\newcommand{\bhSigma}{\boldsymbol{\widehat{\Sigma}}}
\newcommand{\btSigma}{\boldsymbol{\widetilde{\Sigma}}}

\newcommand{\bLambda}{\boldsymbol \Lambda}

\newcommand{\bhD}{\boldsymbol{\widehat{D}}}

\newcolumntype{M}[1]{>{\centering\arraybackslash}m{#1}}

\begin{document}

\def\spacingset#1{\renewcommand{\baselinestretch}
{#1}\small\normalsize} \spacingset{1}

\newcommand{\blind}{0}

\if0\blind
{
\title{\bf Cellwise and Casewise Robust\\ Multivariate 
           Regression with Inference}
\author[1]{Fabio Centofanti}
\author[1]{Mia Hubert}
\author[1]{Peter J. Rousseeuw}

\affil[1]{Section of Statistics and Data Science, Department 
          of Mathematics, KU Leuven, Belgium}

\setcounter{Maxaffil}{0}
\renewcommand\Affilfont{\itshape\small}
\date{October 17, 2026}
\maketitle
} \fi

\if1\blind
{
  \phantom{abc}\\
  \vspace{20mm}
  
  \begin{center}
    {\LARGE \bf  Cellwise and Casewise Robust\\
           \vskip5mm  
           Multivariate Regression with Inference}\\
    \vskip12mm
    \large{October 17, 2026}
\end{center}
  \medskip
} \fi

\bigskip
\begin{abstract}
Multivariate linear regression is a fundamental 
statistical task, but classical estimators such 
as ordinary least squares are highly sensitive to 
outliers. These may occur as casewise outliers that
affect entire observations, or as outlying cells, 
that are individual contaminated entries 
in the predictor and/or response matrix.  
Moreover, modern datasets frequently contain missing 
values and are high-dimensional.
To address these challenges we propose the cellwise 
multivariate regression (cellMR) estimator, a robust 
regression method that simultaneously accommodates 
casewise and cellwise outliers, missing data, and 
high dimensionality. The approach builds on a  
cellwise robust covariance estimator and uses ridge 
regularization for numerical stability. We further 
introduce cellBoot, a novel bootstrap-based 
inference procedure tailored to the cellMR framework. 
Relying on indirect inference, cellBoot provides 
asymptotically valid confidence intervals that are 
robust to casewise and cellwise contamination. 
We derive influence functions and breakdown 
bounds of the regression 
estimator and prove the asymptotic validity of the 
cellBoot confidence intervals. Simulations and 
a real genomics application illustrate the strong 
finite-sample performance of the proposed methods.
\end{abstract}

\mbox{Indirect} Inference.

\newpage
\spacingset{1.8}
\section{Introduction} \label{sec:intro}

Multivariate linear regression is a fundamental 
tool in statistics that models the relationship 
between $p$-dimensional predictors $\bx_i$ and
$q$-dimensional responses $\by_i$ from a sample 
$(\bx_i,\by_i)$ for $i=1,\dots,n$. It plays a central 
role in a wide range of scientific disciplines, as 
it provides an interpretable framework for 
understanding complex multivariate relationships and 
for making predictions. The model is typically fit 
by the ordinary least squares (OLS) method. 

However, nowadays people routinely collect large 
and complex datasets, that are often 
contaminated by outliers, also called anomalies.
Those are parts of the data that deviate markedly 
from the majority. 
They may arise from a variety of sources, such as 
measurement errors, data entry mistakes, sensor 
malfunctions, or rare and unexpected events. 
Because OLS is highly sensitive to such 
contamination, its performance can deteriorate 
substantially.
It can be attracted by outliers so strongly that 
its residuals hide the outliers. This 
is called the masking effect. Additionally, some
regular values might even appear to be outlying,
which is known as swamping.  
An illustrative example showing how masking and
swamping can arise in classical regression, and how these
effects are mitigated by the robust regression method
presented in this paper, is provided in Supplementary
Material~\ref{app:masking_swamping}.
The paradigm of robust statistics \citep{huber1981, 
hampel1986, RL1987, maronna2019robust} provides
a strategy to mitigate these effects. Robust 
methods produce estimates that are only mildly 
affected by the presence of outliers. 
The outliers can then be detected by their 
large deviations from the estimated model.

Research on outliers has traditionally focused on
outlying cases, also known as casewise outliers.
These are observations $(\bx_i,\by_i)$ that were 
not generated by the same underlying mechanism as the 
majority of the data. Many robust regression methods 
have been developed to address casewise outliers 
in order to detect and downweight them.
All casewise robust 
methods require that at least 50\% of the pairs 
$(\bx_i,\by_i)$ are clean.
This requirement ensures that the clean 
component is dominant, allowing the target 
regression model to be distinguished from a 
competing model generated by the contamination.

In recent years, increasing attention has been 
devoted to cellwise outliers \citep{alqallaf2009}. 
In the context of the linear regression model, 
these correspond to anomalous cells (entries) in
the combined matrix $[\bX\,;\,\bY]$ of predictors 
and responses. Cellwise outliers might only 
contaminate a few coordinates of a case 
$(\bx_i,\by_i)$. They are particularly common in 
high-dimensional datasets, that is, with a large 
number of predictor variables and/or response 
variables. Even a relatively small proportion of 
outlying cells can contaminate many cases. 
When random cells in the predictors and responses 
are contaminated with probability $\eps$, the 
expected fraction of contaminated cases is 
$1 - (1-\eps)^{(p+q)}$. This grows fast with $p$ 
and $q$: even if only 1\% of the cells are 
contaminated with $p=100$ and $q=2$, we can expect 
64\% of the cases to be contaminated. In such 
situations casewise robust methods become ineffective. 

\citet{Challenges} reviewed
the challenges of dealing with cellwise outliers. 
Several proposals for cellwise robust regression 
with $q = 1$ have been made. One of the earliest 
was the shooting S-estimator of 
\citet{ollerer2015robust}, that iteratively updates 
each coefficient by a regression with $p=1$.
Afterward \citet{bottmer2022sparse} constructed 
a sparse version of the shooting S-estimator.
\citet{leung2016robust} proposed an alternative 
approach that is based on robust estimation of the 
joint location and scatter of the pairs 
$(\bx_i, \by_i)$ using the 2SGS estimator of 
\citet{agostinelli2015robust}. 
The cellwise robust M-regression estimator 
of \citet{filzmoser2020cellwise} relies on an 
iteratively reweighted least squares procedure. 
More recently, \citet{CRlasso} introduced a  
regularization approach that minimizes both a 
regression loss and a cell deviation measure.

Contemporary applications increasingly involve
high-dimensional settings, where the number
of predictors $p$ is comparable to or exceeds the
sample size $n$. In this paper, we use the term 
high-dimensional specifically for settings with
$p\geqslant n$ but not $p \gg n$.
The number of responses $q$ may also
exceed $n$, further increasing the number of 
regression coefficients. In such regimes, 
classical estimation procedures may become unstable 
or ill-posed, and some regularization is required. 
Several casewise robust approaches for regularized 
regression have been proposed, including  
MM-estimation with a ridge penalty 
\citep{maronna2011robust} and robust elastic net 
methods \citep{cohenfreue2019robust}; 
see \citet{filzmoser2021robust} for a 
comprehensive review. A further challenge arises 
from missing data, that is common in practice 
and may occur together with 
contamination. Most robust procedures were 
developed for fully observed data. So far there
was no multivariate regression method capable of 
jointly addressing cellwise contamination, high 
dimensionality, and missing data.

Beyond robust estimation and prediction in
these challenging settings, reliable inference
is also essential. An advantage of classical
multivariate linear regression is the availability
of well-established inferential procedures.
Inference enables uncertainty 
quantification for regression coefficients, the 
construction of confidence regions, and 
hypothesis testing for linear contrasts or 
groups of parameters. 
Inference based on robust estimators is 
less straightforward.
One approach is
based on the asymptotic distribution of robust
estimators. Those are typically derived under 
elliptical assumptions, and 
rely on normal approximations with estimated 
asymptotic covariance matrices. However, such 
assumptions rarely hold. The bootstrap offers an 
attractive alternative without relying heavily 
on distributional assumptions, see 
\citet{efron1994introduction}. 
It repeatedly draws samples with replacement from 
the observed cases, and recalculates the estimator 
on each resample. The empirical distribution of
these estimates then approximates the estimator's
distribution. This approach can be extended to 
robust regression settings. 
For casewise contamination, robust bootstrap 
procedures have been developed along two main 
directions. The first is the fast and robust 
bootstrap \citep{salibian2002bootstrapping,
van2005multivariate}. It provides an 
asymptotically consistent and 
computationally efficient framework for robust 
inference, reviewed in \citet{salibian2008fast}. 
A second line of work, initiated by 
\citet{amado2004robust}, modifies the resampling 
scheme by assigning lower sampling probabilities 
to potentially harmful cases.

In this paper, we develop a framework for
multivariate linear regression with cellwise
and casewise contamination and missing values.
Its novel contributions are:

\begin{enumerate}[nosep]
\item \textit{Multivariate linear regression estimation.} We introduce cellMR that 
simultaneously accommodates cellwise and 
casewise contamination and missing values, and 
allows for high dimensions by ridge 
regularization. It is the first method to do so.
It builds on the cellwise robust covariance 
estimator of \citet{centofanti2025cellwise} and
the robust low-rank approximation proposed by
\citet{cellPCA}.

\item \textit{Prediction and tuning 
parameter selection.} We develop an out-of-sample prediction strategy that can handle contaminated 
and missing entries in
new predictor vectors. We also introduce a robust
cross-validation criterion for tuning parameter
selection.

\item \textit{Robustness theory of 
the regression estimates.}
We derive the casewise and cellwise influence
functions of cellMR and establish lower bounds 
for its finite-sample casewise and cellwise 
breakdown values. This is the first such 
joint analysis of 
casewise and cellwise robustness for a  
multivariate regression estimator.

\item \textit{Outlier diagnostics.}
We develop a regression outlier map as well
as predictor and residual cellmaps, that 
distinguish unusual predictors from departures 
from the regression relationship, and identify
the individual cells involved.

\item \textit{Bootstrap inference.}
We complement cellMR with a novel 
inference procedure, named the 
\textit{cellwise bootstrap}
(cellBoot). This method combines nonparametric
bootstrap resampling with indirect inference (II),
a simulation-based bias correction approach that
constructs a consistent estimator from an
inconsistent auxiliary estimator
\citep{gourieroux1993indirect,guerrier2019}.
For this purpose we introduce FastCellCov,
a new auxiliary estimator that reuses quantities
from a single initial cellMR fit, avoiding 
repeated execution of the full cellMR during
resampling and II. To the best of our 
knowledge, cellBoot is the first bootstrap 
inference procedure for multivariate regression 
designed to handle cellwise and casewise 
contamination combined with missing values.

\item \textit{Theoretical validity of 
cellBoot.} We prove asymptotic exactness of the
cellBoot confidence intervals. This is the first
such validity result for a multivariate regression
bootstrap procedure designed to handle both
cellwise and casewise contamination.

\item \textit{Robustness of confidence 
intervals.} We derive casewise and cellwise 
influence functions for the center and length of
the cellBoot confidence intervals, 
characterizing their robustness.
This is the first joint analysis of these interval characteristics for multivariate regression
under both contamination paradigms.
\end{enumerate}
So whereas cellMR builds on earlier work,
it goes substantially further  
by developing a comprehensive framework for 
multivariate regression, prediction, outlier 
diagnostics, and inference, with new 
theoretical results on inference validity
and robustness.

Section~\ref{sec:framework-roadmap} 
provides a roadmap of the methods
developed in this paper.
Section~\ref{sec:method} introduces the cellMR 
method. Section~\ref{sec:outdetect} 
constructs graphical displays  to facilitate
outlier detection. 
Section~\ref{sec:inf} describes the 
cellBoot inference procedure. The finite-sample 
performance of cellMR and cellBoot is evaluated 
through simulation in Section~\ref{sec:sim}, and 
Section~\ref{sec:realdata} illustrates them on a 
real dataset from genomics.
Section~\ref{sec:conc} concludes.

\section{Roadmap of the methodological
               framework}
\label{sec:framework-roadmap}
This section provides a roadmap of the
two procedures developed in this paper, cellMR 
and cellBoot, which are summarized in Figure~\ref{fig:method-workflow}.
The left panel shows the steps of
cellMR. The combined predictor and response 
datapoints $\btz_i = (\bx_i^\top,\by_i^\top)^\top$ 
are first robustly standardized to account for 
different measurement units. CellPCA
then provides a low-rank fit to the $\btz_i$ that 
is robust to cellwise and casewise outliers and 
imputes missing values. Because fitted cases may
be outlying inside this subspace, the minimum 
covariance determinant 
(MCD) of \cite{rousseeuw1984least} estimates 
location and scatter of the fitted points.
We then estimate the scatter in the orthogonal 
complement, while limiting the effect of anomalous 
cells and cases. The two components of scatter
are then combined
to obtain the cellCov estimates of the location 
and scatter of the $\btz_i$\,. Finally, we carry 
out ridge regression based on these estimates to 
obtain stable regression coefficients.
These steps will be detailed in 
Section~\ref{sec:estimator}. The robust selection 
of the cellPCA rank and the ridge parameter will 
be described in Section~\ref{sec:tuning}.

\begin{figure}[H]
\vspace{2mm}
\centering

\begin{minipage}[t]{0.49\textwidth}
\centering
\textbf{cellMR for regression estimation}\par\medskip
\begin{tikzpicture}[
  node distance=3mm,
  box/.style={
    draw=black!60,
    rounded corners=2pt,
    fill=black!3,
    align=center,
    text width=0.86\linewidth,
    inner sep=4pt,
    font=\footnotesize
  },
  result/.style={
    box,
    fill=blue!6,
    draw=blue!50!black
  },
  arrow/.style={
    -{Stealth[length=2mm]},
    semithick
  }
]
\node[box] (scale)
  {\textbf{Robust standardization} of the\\[-2mm]
  combined predictor and response matrix};

\node[box, below=of scale] (pca)
  {Casewise and cellwise robust dimension\\[-2mm]
  reduction via \textbf{cellPCA}};

\node[box, below=of pca] (parts1)
  {Robust location and scatter estimation\\[-2mm]
  of the fitted points by \textbf{MCD}};

\node[box, below=of parts1] (parts2)
  {Estimation of \textbf{scatter in the orthogonal \\[-2mm]
   complement} by weighted sample covariance};

\node[box, below=of parts2] (cov)
  {Combine the estimates in the
   subspace  and \\[-2mm]its orthogonal complement to obtain \textbf{cellCov} };

\node[box, below=of cov] (out)
  {Estimate the slope, intercept, and error\\[-2mm]
  covariance by \textbf{ridge regression}
  formulas};

\draw[arrow] (scale) -- (pca);

\draw[arrow] (pca) -- (parts1);
\draw[arrow] (parts1) -- (parts2);
\draw[arrow] (parts2) -- (cov);

\draw[arrow] (cov) -- (out);
\end{tikzpicture}
\end{minipage}\hfill
\begin{minipage}[t]{0.49\textwidth}
\centering
\textbf{cellBoot for confidence intervals\vspace{0.6mm}}\par\medskip
\begin{tikzpicture}[
  node distance=3mm,
  box/.style={
    draw=black!60,
    rounded corners=2pt,
    fill=black!3,
    align=center,
    text width=0.86\linewidth,
    inner sep=4pt,
    font=\footnotesize
  },
  result/.style={
    box,
    fill=blue!6,
    draw=blue!50!black
  },
  arrow/.style={
    -{Stealth[length=2mm]},
    semithick
  }
]
\node[box] (fit)
  {Initial \textbf{cellMR} fit to estimate the robust \\[-2mm]
   structure and select the tuning parameters};

\node[box, below=of fit] (resample)
  {\textbf{Nonparametric bootstrap}\\[-2mm] resampling
   of observation pairs};

\node[box, below=of resample] (fast)
  {\textbf{FastCellCov} estimation on \\[-2mm] each bootstrap
   sample};

\node[box, below=of fast] (ii)
  {Apply \textbf{II} to correct
   the auxiliary\\[-2mm] location and scatter estimates};

\node[box, below=of ii] (coef)
  {Estimate the corrected slope and intercept\\[-2mm]
   via the \textbf{ridge regression} formulas};

\node[box, below=of coef] (ci)
  {Compute \textbf{bootstrap quantiles} to \\[-2mm]obtain
   the cellBoot confidence intervals};

\draw[arrow] (fit) -- (resample);
\draw[arrow] (resample) -- (fast);
\draw[arrow] (fast) -- (ii);
\draw[arrow] (ii) -- (coef);
\draw[arrow] (coef) -- (ci);
\end{tikzpicture}
\end{minipage}
\caption{Roadmap of cellMR and cellBoot.
}
\label{fig:method-workflow}
\end{figure}

\vspace{-2mm}
The cellBoot method in the right panel uses nonparametric bootstrap resampling to quantify 
uncertainty in the regression estimates. 
Repeatedly fitting the full cellCov procedure 
would be computationally costly, so FastCellCov 
reuses quantities from the initial cellMR fit 
when evaluating bootstrap and simulated samples
(Section~\ref{sec:FastCellCov}). Because this
auxiliary estimator may target a different 
population location and scatter, II corrects 
the discrepancy (Section~\ref{sec:II}). 
Applying the regression formulas to the 
corrected estimates and taking bootstrap 
quantiles then yields confidence intervals, as 
described in Section~\ref{sec:cellBoot_algo}.
Table~\ref{tab:method-names} summarizes the 
names of the main components of 
the methodology used throughout the paper.

\begin{table}[H]
\centering
\caption{Main components of the methodological framework.}
\label{tab:method-names}
\vspace{-2mm}
\small
\begin{tabular}{@{}p{0.14\textwidth}p{0.78\textwidth}@{}}
\toprule
cellPCA &
Robust low-rank fit of the combined data. \\
cellCov &
Robust joint location and scatter estimator built on cellPCA. \\
cellMR &
Ridge regression using the cellCov estimates. \\
FastCellCov &
Auxiliary location and scatter estimator for resampled
and simulated data. \\
cellBoot &
Bootstrap inference using FastCellCov and 
  indirect inference (II). \\
\bottomrule
\end{tabular}
\end{table}

\section{Robust multivariate regression by cellMR}
\label{sec:method}

Let us consider a random sample 
$\{(\bx_i,\by_i)\}_{i=1}^n$\,, where $\bx_i$ 
denotes the predictor vector and $\by_i$ the 
response vector, that satisfies the 
multivariate linear regression model given by
{\setlength{\abovedisplayskip}{5pt}
\setlength{\belowdisplayskip}{5pt}
\begin{equation}
\label{eq:mainmodel}
   \by_i
  = \bb+\bB^{T}\bx_i
  + \beps_i \quad \mbox{for} \quad
   i = 1,\ldots,n,
\end{equation}}
\noindent where the $q$-dimensional vector $\bb$ is the 
intercept, $\bB$ is the $p \times q$ slope matrix,  
and $\beps_1,\dots\beps_n$ are $q$-dimensional 
random errors satisfying
$\E(\beps_i\mid\bx_i)=\bzero$, with unconditional
covariance $\bSigma_{\beps}=\Var(\beps_i)$. 
The conditional covariance $\Var(\beps_i\mid\bx_i)$
is allowed to depend on $\bx_i$.
Denote by $\btz_i$ the $d$-dimensional vector 
$\left(\bx_i^T,\by_i^T\right)^T$,
with $d=p+q$. The sample mean and 
empirical covariance matrix of the $\btz_i$ 
can be partitioned as 
\begin{equation*}
\bhmu=
\begin{pmatrix}
  \bhmu_x \\[2pt]
  \bhmu_y
\end{pmatrix}
\quad\text{and}\quad
\bhSigma =
\begin{pmatrix}
  \bhSigma_{x} & \bhSigma_{xy} \\[2pt]
  \bhSigma_{yx} & \bhSigma_{y}
\end{pmatrix}.
\end{equation*}
Under model~\eqref{eq:mainmodel}, the classical
ridge regression is given by
{
\setlength{\abovedisplayskip}{5pt}
\setlength{\belowdisplayskip}{5pt}
\begin{equation} \label{eq:ridge}
  \bhB=(\bhSigma_{x}+\lambda\bI_p)^{-1}
  \bhSigma_{xy} \quad \mbox{and} \quad 
  \bhb=\bhmu_y-\bhB^T\bhmu_x\,,
\end{equation}
}
with estimated error covariance matrix
{
\setlength{\abovedisplayskip}{4pt}
\setlength{\belowdisplayskip}{4pt}
\begin{equation}
\label{eq:errcov}
  \bhSigma_{\beps}=\bhSigma_{y}-\bhSigma_{yx}
  (\bhSigma_{x}+\lambda\bI_p)^{-1}\bhSigma_{xy}\,.
\end{equation}
}
The advantage of the regularization by
$\lambda > 0$ is that it avoids the
inversion of a potentially ill-conditioned
matrix $\bhSigma_x$ that can occur due to 
multicollinearity or a high dimension $p$,
possibly even with $p \geqslant n$. The 
tuning parameter $\lambda$ is typically 
chosen by cross-validation.

\subsection{The cellMR estimator}
\label{sec:estimator}

Formulas~\eqref{eq:ridge} and~\eqref{eq:errcov} 
allow to compute other regression estimators
by plugging in suitable estimators of the location 
$\bmu$ and scatter $\bSigma$ of the $\btz_i$. For 
this we will use the cellCov estimator, based on 
\citet{centofanti2025cellwise}, that combines 
estimates of a robust low-rank fit with the 
variation in its orthogonal complement. It performs
well with both casewise and cellwise outliers, even
in high dimensions, and will be described in
the steps below.

\vspace{3mm}
\noindent{\bf Robust standardization.} To 
compute the cellCov estimator, the 
combined $n \times d$ data matrix 
$\btZ=\lbrace \tz_{ij}\rbrace=\left[
\btz_1,\dots,\btz_n\right]^T$  is first
standardized to $\bZ=\lbrace z_{ij}\rbrace=
[\bz_1,\dots,\bz_n]^T=\btZ\bhD^{-1}$, where 
$\bhD=\diag(\hsigma^{\tZ}_1,\dots,
\hsigma^{\tZ}_{d})$. Here 
$\hsigma^{\tZ}_j:=
\sigma_M(\{\tz_{ij}-m_j\}_{i=1}^n)$, where  
$m_j$ is the median of the $j$-th variable and
$\sigma_M$ is the robust scale M-estimator 
defined in Section~\ref{app:adddet} of the 
Supplementary Material.

\vspace{3mm}
\noindent{\bf Dimension reduction via cellPCA.} 
Next, cellCov applies the cellPCA method 
\citep{cellPCA}, a robust principal component 
analysis (PCA) method designed to handle both 
cellwise and casewise outliers as well as 
missing values. We model the matrix 
$\bZ$ as
{
\setlength{\abovedisplayskip}{4pt}
\setlength{\belowdisplayskip}{4pt}
\begin{equation}\label{eq:model}
    \bZ=\bone_n \bmu_z^T + \bU\bV^T + \bE,
\end{equation}
}
where $\bone_n$ is a column vector with all 
$n$ components equal to $1$, the scores
matrix $\bU  =\lbrace u_{i\ell} \rbrace = 
\left[\bu_1,\dots,\bu_n\right]^T$ is 
$n \times \rk$,  the loadings matrix 
$\bV=\lbrace v_{j\ell}\rbrace =
\left[\bv_1,\dots,\bv_{d}\right]^T$ is $
d\times k$, and the matrix 
$\bE=\left[\zort_1,\dots,\zort_n\right]^T$ 
is  the noise term. Here $\rk$ is the
rank of the low-rank component $\bU\bV^T$, which
for now we assume to be known. Its selection 
will be described in Section~\ref{sec:tuning}.
The cellPCA method computes 
estimates $\bhV$, $\bhU$, and  $\bhmu_z$ by 
minimizing the loss function
\begin{equation} \label{eq:objP}
  L_{\rho_1,\rho_2}(\bZ,\bV,\bU,\bmu_z) := 
  \frac{\hsigma_2^2}{o}\sum_{i=1}^{n}o_i \rho_2\! 
  \left(\frac{1}{\hsigma_2}\sqrt{\frac{1}{o_i}
  \sum_{j=1}^{d} o_{ij}\, \hsigma_{1,j}^2\,
  \rho_1\!\left(\frac{r_{i j}}
  {\hsigma_{1,j}}\right)}\, \right),
\end{equation}
where the $r_{i j} := z_{i j}-\mu_{z,j}-
\sum_{\ell=1}^k u_{i \ell}v_{j \ell}$\,, the 
indicator $o_{ij}$ is 1 if $x_{ij}$ is 
observed and 0 if it is missing, 
$o_i=\sum_{j=1}^{d} o_{ij}$\,,
and $o=\sum_{i=1}^{n} o_i$\,. 
The scale $\hsigma_{1,j}:=
\sigma_M(\{r_{ij}\}_{i=1}^n)$ 
standardizes the {\it cellwise PCA residual} 
$r_{ij}$\,, and 
$\hsigma_2:=\sigma_M(\{\rt_i\}_{i=1}^n)$ 
divides the {\it casewise total deviation} 
\begin{equation}\label{eq:rt_i}
  \rt_i := \sqrt{\frac{1}{o_i} \sum_{j=1}^{d}
  o_{ij}\, \hsigma_{1,j}^2\, \rho_1\!\left(
  \frac{r_{ij}}{\hsigma_{1,j}} \right)}.
\end{equation} 

Pseudocode of the algorithm used to minimize
\eqref{eq:objP} is provided in Supplementary
Material~\ref{app:pseudocode}.
In the fully observed setting, for 
$\rho_1(z) = \rho_2(z) = z^2$ the 
objective~\eqref{eq:objP} becomes 
the objective of classical PCA.
But instead cellPCA uses bounded
functions $\rho_1$ and $\rho_2$\,.

\begin{definition}\label{def:validrho}
A function 
$\rho:\mathbb{R}\rightarrow\mathbb{R}$ 
is called a valid $\rho$-function if it is 
continuous and differentiable, even, bounded,
nondecreasing in $|z|$, has $\rho(0)=0$, 
and is such that the mapping 
$z \mapsto \rho(\sqrt{z})$ is concave for 
$z \geqslant 0$.
\end{definition}

The cellPCA method uses the valid hyperbolic 
tangent $\rho$-function \citep{hampel1986}, that 
is described in the Supplementary 
Material~\ref{app:adddet}.
This makes
$\bhmu_z$, $\bhU$, and $\bhV$ robust against 
both cellwise and casewise outliers.
Indeed, 
a cellwise outlier $z_{ij}$
yields a cellwise PCA residual $r_{ij}$ with a 
large absolute value, but the boundedness of
$\rho_1$ reduces its effect on the estimates.
Similarly, a casewise outlier results in a large 
casewise total deviation $\rt_i$ but its effect
is reduced by $\rho_2$\,.
Note that in the computation of $\rt_i$ the effect 
of cellwise outliers is tempered by the presence 
of $\rho_1$\,. This avoids that a single very
outlying cell could give its case a large $\rt_i$\,.

This intuition is supported theoretically by the
influence function analysis of \citet{cellPCA}, which
demonstrates that cellwise or casewise
outliers have bounded influence on the estimated
principal subspace. Extensive simulations
in that work further demonstrate cellPCA's
robustness and favorable performance relative to
competing approaches. 
The use of valid $\rho$-functions ensures the
redescending property, meaning that sufficiently
extreme outliers have
no influence on the estimates.
Moreover, it has the ability to accommodate missing values. 
These properties motivate our choice of cellPCA.
Note that we do not require the clean data to lie
exactly in a low-dimensional subspace, as the 
subsequent covariance estimation also accounts
for variation outside the fitted subspace.
Alternative robust PCA methods
could in principle be considered, provided that
they offer suitable protection against both
contamination types and accommodate missing values.
A prominent example is the class of sparsity-based
methods, which explicitly model contamination
through additional parameters penalized to
encourage sparsity. Although these methods benefit
from recovery guarantees under suitable assumptions,
such guarantees do not, by themselves, establish
their suitability for the setting considered here.
In particular, the principal component pursuit
method of \citet{candes2011robust} uses an entrywise
absolute reconstruction loss, which lacks the
redescending property, and has no separate mechanism
for downweighting entire outlying cases.
Therefore its theoretical guarantees for sparse 
corruption do not automatically extend to combined 
cellwise and casewise contamination, as also 
illustrated by the simulation results in 
\citet{cellPCA}. A detailed comparison of the
cellPCA formulation with sparsity-based approaches 
is provided in Supplementary Material~\ref{app:robust_pca_sparsity}.

\vspace{3mm}
\noindent{\bf Location and scatter estimation 
of the fitted points via MCD.} Define the matrix 
\mbox{$\bhZ  =\lbrace \hz_{ij}\rbrace=
\left[\bhz_1,\dots,\bhz_n\right]^T$}, where
$\bhz_i =\bhmu_z + \bhV\bhu_i$ are the fitted points
in the $k$-dimensional principal subspace. 
A point $\bz_i$ may be close to the subspace
whereas its fit $\bhz_i$ is outlying within the
subspace. This can distort the location and scatter
of the fitted points despite the robust cellPCA 
fit. We therefore estimate
their location and scatter by
{
\setlength{\abovedisplayskip}{4pt}
\setlength{\belowdisplayskip}{4pt}
\begin{equation}\label{eq:SigmaMCD}
  \btmu_{\bz^\sub}=
  \bhmu_z+\bhV\bhmu_{\text{\tiny MCD}}(\bhu_i) 
  \quad \mbox{and} \quad \btSigma_{\bz^\sub}=\bhV
  \bhSigma_{\text{\tiny MCD}}(\bhu_i)\bhV^T,
\end{equation} 
}
where 
$\bhmu_{\text{\tiny MCD}}(\bhu_i)$ and 
$\bhSigma_{\MCD}(\bhu_i)$ are the MCD estimates of 
$\bhu_1,\dots,\bhu_n$ \citep{rousseeuw1984least}.

\vspace{3mm}
\noindent{\bf Scatter estimation in the 
orthogonal complement.} Then we define 
$\btW=\bW^{\cell} \odot \bO$, where the 
Hadamard product $\odot$ multiplies matrices
entry by entry, and where the $n \times d$ matrix 
$\bW^{\cell}=\lbrace w_{ij}^{\cell}\rbrace$ 
contains the {\it cellwise PCA weights}
\begin{equation}\label{eq:cellweight}
   w_{ij}^{\cell}=w^{\cell}\left( 
   \frac{\hr_{ij}}{\hsigma_{1,j}}\right)=
   \psi_1\!\left(
   \frac{\hr_{ij}}{\hsigma_{1,j}}\right)
   \Big/ \frac{\hr_{ij}}{\hsigma_{1,j}} , 
   \quad i=1,\dots,n,\quad j=1\dots,d,
\end{equation}
where $\hr_{ij} = z_{ij} - \hz_{ij}$ and 
$\psi_1=\rho_1'$ with the convention 
$w_{ij}^{\cell}(0) = 1$. 
The $n \times d$ matrix $\bO$ contains the 
indicators $o_{ij}$ of observed cells.
We then estimate the scatter in the
orthogonal complement of the principal
subspace as
\begin{equation}\label{eq:Sigmay}
   \btSigma_{\zort}=
   \lbrace \tsigma_{(\zort)j\ell}
   \rbrace = \frac{1}{b}\sum_{i=1}^n
   w_i^{\case}\btW_i(\bz_i - \bhz_i)
   (\bz_i - \bhz_i)^T\btW_i\;,
\end{equation}
\sloppy where $\btW_i$  is a diagonal matrix 
whose diagonal is the $i$-th row of $\btW$, 
and $b$ is given by
$b=\sum_{i=1}^n\sum_{j=1}^{d}
\sum_{\ell=1}^{d} o_{ij}o_{i\ell}w_i^{\case}
w_{ij}^{\cell} w_{i\ell}^{\cell}/d^2$.
The {\it casewise PCA weights} $w_i^{\case}$ 
are defined as
\begin{equation}\label{eq:caseweight}
   w_{i}^{\case}=w^{\case}\left(
   \frac{\hti_i}{\hsigma_2}\right)=
   \psi_2\!\left(
   \frac{\hti_i}{\hsigma_2}\right)
   \Big/ \frac{\hti_i}{\hsigma_2} , 
   \quad i=1,\dots,n\,,
\end{equation}
with $\psi_2=\rho_2'$\,, and $\hti_i$ is 
obtained from \eqref{eq:rt_i} with 
$\hr_{ij}$ in place of $r_{ij}$. 
We do not need to regularize
$\btSigma_{\zort}$ because we will regularize 
the regression 
by~\eqref{eq:ridge}--\eqref{eq:errcov}
anyway.

\vspace{3mm}
\noindent{\bf Combining the components into the 
cellCov estimates.} The two scatter estimates 
$\btSigma_{\bz^\sub}$ and $ \btSigma_{\zort}$ describe 
complementary sources of variation. The first captures
variation among fitted points inside the principal
subspace, while $ \btSigma_{\zort}$  captures
variation outside that subspace. Adding them gives
a robust scatter estimate for the standardized data.
These are combined into the cellCov estimates 
$\bhmu$ and $\bhSigma$ of the overall $\bmu$ and
$\bSigma$ in $d$ dimensions, given by 
{
\setlength{\abovedisplayskip}{4pt}
\setlength{\belowdisplayskip}{4pt}
\begin{equation} \label{eq:cellCov}
 \bhmu:=\bhD\btmu_{\bz^\sub} \quad 
 \mbox{and} \quad \bhSigma :=
  \bhD\big(\btSigma_{\bz^\sub}+
  \btSigma_{\zort}\big)\bhD\,,
\end{equation}
}
where the diagonal matrix $\bhD$ undoes the
original standardization.
The standardization step in the beginning ensures 
that cellCov is scale equivariant, meaning that
for any diagonal matrix 
$\bA = \operatorname{diag}(a_1, \dots, a_d)$ 
with $a_j > 0$, the cellCov location and scatter 
estimators of $\btZ\bA$ are $ \bA \bhmu $ and
$ \bA \bhSigma \bA$. Therefore, $\bhmu$ and 
$\bhSigma$ react in the usual way to changes of 
variable units. 

\vspace{3mm}
\noindent{\bf Ridge regression estimates.}
We then define the cellMR 
estimates of $\bB$, $\bb$, and $\bSigma_\eps$ as 
in~\eqref{eq:ridge}--\eqref{eq:errcov}.
We will discuss the tuning of the ridge
parameter $\lambda$ later. The estimator 
remains well defined when $p\geqslant n$.
Indeed, the cellCov scatter estimate 
in~\eqref{eq:cellCov} is positive semidefinite, 
and so is its predictor block $\bhSigma_x$. For
any $\lambda>0$, $\lambda_{\min}(\bhSigma_x+\lambda\bI_p)\geqslant\lambda$.
Consequently, the inverse in~\eqref{eq:ridge} exists
regardless of the relation between $p$ and $n$.

\vspace{2mm}
\noindent{\bf Remark 1.} Having multiple 
responses ($q>1$) may appear to facilitate outlier
detection compared with the univariate setting 
($q=1$). However, this benefit depends on the
dependence structure of the predictors and 
responses. Cellwise downweighting is 
determined by the standardized residuals $\hr_{ij}/\hsigma_{1,j}$ 
in~\eqref{eq:cellweight}. Response variables 
that improve the reconstruction of another 
variable can therefore make anomalous cells 
more distinguishable from regular cells.
Dependence among regression errors can provide
information beyond that already available from
the predictors. By contrast, response variables
that are conditionally independent of each 
other given the predictors provide no additional information for this purpose.

For casewise outliers, the weights
in~\eqref{eq:caseweight} depend on the 
standardized total deviation $\hti_i/\hsigma_2$.
Additional responses can increase the separation
between clean and contaminated cases when they
exhibit departures from the fitted subspace.
However, adding response variables with 
substantial unexplained variation can dilute 
this separation by adding noise. 

The following proposition summarizes 
the computational complexity of cellMR.
\begin{proposition}
\label{prop:complexity-cellMR}
Assume that $\rk\leqslant \rk_{\max}$.
Then, for a given $\lambda>0$, the time
complexity of estimating the slope and 
intercept by cellMR is
$O\!\left(nd^2+nd\log(n)+p^3+p^2q\right)$.
If $p\leqslant n$, this simplifies to
$O\!\left(nd^2+nd\log(n)\right)$.
\end{proposition}
The derivation is provided in Supplementary Material~\ref{sec:supp-complexity-theor}.

\vspace{2mm}
\noindent{\bf Remark 2.} 
Although cellMR accommodates settings with
$p\geqslant n$, it does involve some relatively
computation-intensive steps. As such it is not
suited for extremely large dimensions. 
Nevertheless, the runtime experiments in the 
Supplementary Material~\ref{app:runtime} show 
that cellMR remains computationally feasible for
fairly large problems. For example, with $n=100$ 
and $p=q=400$, cellMR ran in 70 seconds on the 
computing infrastructure described there.

\subsection{Out-of-sample prediction}
\label{sec:oos}

When a new datapoint
$\bx_*=\left(x_{*1},\dots,x_{*p}\right)^T$ arrives,
we wish to predict the unknown response $\by^*$.
If $\bx_*$ were clean and fully observed, we could 
simply put
$\bhy_* = \bhb+\bhB^{T}\bx_*$\,.
However, the task becomes nontrivial since $\bx_*$
can also contain NA's and/or cellwise outliers.
Even when the regression coefficients are
estimated robustly, contamination in $\bx_*$
can propagate through $\bhB^T\bx_*$ and produce
unreliable predictions. To obtain meaningful
predictions, we therefore need to limit the
influence of anomalous predictor cells and
impute missing values. Fortunately, cellPCA 
can produce an imputed version of
$\bx_*$ in which suspicious cells are cleaned and
missing cells are filled in, whereas the other
cells are kept as they were. For this 
purpose, we fit a separate cellPCA model
to the training predictors $\{\bx_1,\dots,\bx_n\}$.
This predictor-only fit exploits the dependence 
among the predictors without requiring the 
response, which is unavailable for a new 
datapoint. The imputed point is given by
{
\setlength{\abovedisplayskip}{4pt}
\setlength{\belowdisplayskip}{4pt}
\begin{equation} \label{eq:impxi}
   \ximp_{*j} := \hx_{*j} +
   w_{*j}^{\cell} o_{*j}(x_{*j} - \hx_{*j})
\end{equation}
}
where $\bhx_* = \btV_X \bu^X_* + \btmu_X$.
Here $\btV_X$ and $\btmu_X$ are robust estimates
computed by the predictor-only 
cellPCA fit, with
$\btV_X\in\mathbb{R}^{p\times k_X}$.
The dimension $k_X$ is selected separately from
the dimension $k$ used for $\bZ$, and need not
coincide with it. We select $k_X$ using the 
robust Parallel Analysis procedure proposed in
\cite{centofanti2025cellwise}. The 
scores $\bu^X_*$ are also provided by cellPCA,
with the training center and loadings 
held fixed, as detailed in \cite{cellPCA}.
We set $w_{*j}^{\cell} = w^{\cell}\!\left(
\frac{x_{*j}-\hx_{*j}}{\hsigma_{1,j}^{X}}
\right)$, where $w^{\cell}$ is given 
by~\eqref{eq:cellweight}
and $\hsigma_{1,j}^{X}$ is the corresponding
residual scale. Suspicious cells are thus
flagged by large standardized residuals
$(x_{*j}-\hx_{*j})/\hsigma_{1,j}^{X}$ and 
receive small weights, moving their imputed 
values toward $\hx_{*j}$ in~\eqref{eq:impxi}.
The cellMR prediction $\bhy_*$ of the response
of $\bx_*$ is then given by
{
\setlength{\abovedisplayskip}{4pt}
\setlength{\belowdisplayskip}{4pt}
\begin{equation*}
    \bhy_* := \bhb+\bhB^{T}\bximp_*.
\end{equation*}
}
This prediction procedure assumes that the
clean predictors of new observations follow the 
same distribution as those of the training data.

\subsection{Selecting tuning parameters}
\label{sec:tuning}
The cellMR method contains two knobs: the 
dimension $\rk$ of the principal subspace in 
the PCA model~\eqref{eq:model}, and the 
parameter $\lambda$ in the ridge 
regularization~\eqref{eq:ridge}. 
If the condition number of $\bhSigma_x$
given by $\lambda_{\max}(\bhSigma_x)/
\lambda_{\min}(\bhSigma_x)$ is small,
say under $10^5$, we can just 
put $\lambda=0$.

The most common method to select tuning 
parameters in regression is cross-validation 
(CV), see e.g. \citet{hastie2009elements}. 
$K$-fold CV randomly splits the dataset into 
$K$ folds of size $n_h$ for $h=1,\dots,K$. 
For each $h$, the model is trained on the 
union of the other $K-1$ folds, and the
predictions are computed on fold $h$. The 
resulting squared regression residuals are 
then averaged to estimate the out-of-sample
mean squared error (MSE). The computation is 
repeated on a grid of parameter values,
and the parameter with the smallest MSE 
is selected.

We need to modify this approach for selecting 
the pair $(\rk,\lambda)$ because each 
fold can contain outliers and NA's. We will 
estimate the prediction error by
\begin{equation}\label{eq:cv}
   \text{CV}(\rk,\lambda) = \frac{1}{q} \sum_{j=1}^{q}
   \frac{1}{K}\sum_{h=1}^K
   \mbox{WMSE}_{(h)j}\;,
\end{equation}
where $\mbox{WMSE}_{(h)j}:=\sum_{i=1}^{n_h}
 w_{(h)ij}(y_{(h)ij}-\hy_{(h)ij})^2
 /\big(\sum_{i=1}^{n_h} w_{(h)ij}\big)$ 
is a weighted mean of squared residuals in 
fold $h$, that downweights cellwise and casewise 
outliers. The weights are specified in 
Supplementary Material~\ref{app:cv}.
We then select the pair ($\rk, \lambda)$  
that minimizes~\eqref{eq:cv}.

\subsection{Robustness properties}
\label{sec:robustness_cellMR}
We now study the robustness properties of the 
cellMR estimators $\bhb$ and $\bhB$. 
For this purpose we use the influence 
function (IF), a standard tool for studying local
robustness \citep{hampel1986}. The IF quantifies 
the standardized change in an estimator caused
by placing contamination at a given position in 
the sample space. Examining the IF over all 
possible contamination positions reveals where
outliers have the greatest effect, and when the
IF is bounded we know that no single contamination 
can produce an arbitrarily large change.

To define the IF, we first model how contamination affects the data.
Let $Z=(X^T,Y^T)^T$ denote the $d$-variate random
variable obtained by concatenating the $p$-variate 
predictor $X$ and the $q$-variate response $Y$, and let $H_0$ be the distribution 
of $Z$ without contamination. 
We now add both casewise and cellwise contamination.
Let $C\in\mathbb R^{d}$ be a random vector with 
unspecified outlier-generating distribution $H_C$. 
The \emph{mixed contamination model} says that
we observe
{
\setlength{\abovedisplayskip}{4pt}
\setlength{\belowdisplayskip}{4pt}
\begin{equation}
\label{eq:cont_reg_mixed}
  Z_{\eps} =
  A\odot Z + (\bone_d-A)\odot C,
\end{equation}
}
where $A=A^{\case}\odot A^{\cell}$. The casewise 
contamination factor $A^{\case}$ has Bernoulli 
marginals with $\Pr(A^{\case}_j=1)=1-\eps^{\case}$ 
for $j=1,\ldots,d$, 
and its entries are fully dependent in the sense that
$\Pr(A^{\case}_1=\cdots=A^{\case}_d)=1$.
The cellwise contamination factor $A^{\cell}$
has independent entries $A^{\cell}_j$ for 
$j=1,\ldots,d$, that are Bernoulli random variables
with success probabilities
$\Pr(A^{\cell}_j=1)=1-\eps^{\cell}_j$.
Therefore the mixed model~\eqref{eq:cont_reg_mixed} 
captures the simultaneous presence of casewise and 
cellwise outliers in the regression setting.

We now specialize the mixed contamination 
model to its casewise and cellwise versions and 
define the corresponding influence functions.
For computing the IF we let the outlier distribution
$H_C$ be a point mass $\Delta_{\bc}$\,, which is a 
distribution that assigns all its mass to a point 
$\bc=(c_1,\ldots,c_d)^T$. Note that we do
not assume that actual outliers are concentrated 
at a single point. Rather, the IF is the standard 
device to isolate the effect of a small mass of
contamination at a particular point $\bc$. By 
allowing $\bc$ to vary over $\mathbb R^d$, the IF describes the local sensitivity of the estimator 
to all possible contamination positions.
The classical casewise IF of \citet{hampel1986}  
replaces entire cases by outliers. It puts 
$A^{\cell} = \bone_{d}$ so $A = A^{\case}$, and 
assumes that $A^{\case}$ is independent of $Z$. 
In that situation, the distribution of 
$Z_{\eps}$ simplifies to $(1-\eps^{\case})H_0 + 
\eps^{\case}\Delta_{\bc}$\,. We denote the 
distribution of $A^{\case}$ as $G_\eps^D$ and the 
distribution of $Z_\eps$ as $H(G_\eps^D,\bc)$. 
Writing an estimator as a functional $T$ defined on
distributions on $\mathbb R^{d}$, the casewise 
influence function at the contamination point $\bc$ 
is defined as
{
\begin{equation}
\label{eq:caseIF}
  \IFu_{\case}(\bc,T,H_0) :=
  \left.\frac{\partial}{\partial\eps}
  T\bigl(H(G_\eps^{D},\bc)\bigr)\right|_{\eps=0}
  = \; \lim_{\eps\downarrow 0}
  \frac{T\bigl(H(G_\eps^{D},\bc)\bigr)-T(H_0)}{\eps}.
\end{equation}
}
\citet{giordano2026automatic} propose a 
complementary casewise notion, the \textit{approximate
maximum influence perturbation} that assesses the 
largest change in a prespecified statistic, produced 
by {\it deleting} a small fraction of the cases. This
is related to the casewise empirical influence 
function \citep{hampel1986}, which {\it adds} a case 
to an existing dataset.

To capture the influence of 
outlying \textit{cells} we adopt 
the cellwise IF introduced by \citet{alqallaf2009}. 
Now $A=A^{\cell}$, where the entries $A^{\cell}_j$ 
are mutually independent and independent of $Z$.
The distribution of $A^{\cell}$ is denoted as
$G_\eps^{I}$\,. The resulting distribution of 
$Z_\eps$ with $H_C = \Delta_{\bc}$ is denoted by 
$H(G_\eps^{I},\bc)$. The cellwise influence function 
of $T$ at $\bc$ is then defined as
{
\setlength{\abovedisplayskip}{4pt}
\setlength{\belowdisplayskip}{4pt}
\begin{equation}
\label{eq:cellIF}
  \IFu_{\cell}(\bc,T,H_0) :=
  \left.\frac{\partial}{\partial\eps}
  T\bigl(H(G_\eps^{I},\bc)\bigr)\right|_{\eps=0}
  = \; \lim_{\eps\downarrow 0}
  \frac{T\bigl(H(G_\eps^{I},\bc)\bigr)-T(H_0)}{\eps}.
\end{equation}
}

We denote the functionals corresponding to the 
cellMR estimators $\bhb$ and $\bhB$ by 
$\bb(H)$ and $\bB(H)$.
We will derive the IFs of $\bb(H)$ and of the 
$pq \times 1$ column vector $\vect(\bB(H))$, 
where $\vect(\CD)$ converts a matrix to a vector 
by stacking its columns on top of each other.
The derivation of these IFs relies on the IFs of 
the cellPCA and cellCov estimators, and on the
various components that are used for their 
construction, such as $\bmu_z(H)$, $\bV(H)$,  
$\bmu^{\bu}_{\text{\tiny MCD}}(H)$, 
$\bSigma^{\bu}_{\MCD}(H)$, and $\bSigma_{\zort}(H)$, 
which are the functionals corresponding to 
$\bhmu_z$ and $\bhV$ from~\eqref{eq:objP}, 
$\bhmu_{\text{\tiny MCD}}(\bhu_i)$ and 
$\bhSigma_{\text{\tiny MCD}}(\bhu_i)$ 
in~\eqref{eq:SigmaMCD}, and $\btSigma_{\zort}$ 
in~\eqref{eq:Sigmay}.

\begin{proposition}\label{IFBb}
The casewise and cellwise influence 
functions of $\vect(\bB)$ and $\bb$ are
{
\setlength{\abovedisplayskip}{4pt}
\setlength{\belowdisplayskip}{4pt}
\begin{align}
\label{eq:IFB_tau}
  \IFu_{\tau}\!\left(\bc,\vect(\bB),H_0\right)
  &= \bR_{B}\Big[\bR_{\Sigma,1}\IFu_{\tau}\!\left(
  \bc,\vect(\bV),H_0\right)\nonumber\\
  &\hspace{1cm}
  +\bR_{\Sigma,2}\IFu_{\tau}\!\left(
  \bc,\vect(\bSigma^{\bu}_{\MCD}),H_0\right)
  +\IFu_{\tau}\!\left(\bc,\vect(
  \btSigma_{\zort}),H_0\right) \Big],\\
\label{eq:IFb_tau}
  \IFu_{\tau}\!\left(\bc,\bb,H_0\right)
  &= \bR_{b,1}\Big[\bR_{\mu}\IFu_{\tau}\!\left(
  \bc,\vect(\bV),H_0\right) \nonumber\\
  &\hspace{1cm}
  +\bV_0\IFu_{\tau}\!\left(
  \bc,\bmu^{\bu}_{\text{\MCD}},H_0\right)
  +\IFu_{\tau}\!\left(\bc,\bmu_z,H_0\right)
  \Big]\nonumber\\
  &\hspace{0.5cm}
  -\bR_{b,2}
  \IFu_{\tau}\!\left(\bc,\vect(\bB),H_0\right),
\end{align}
}
with  $\tau\in\{\case,\cell\}$, where $\bV_0=\bV(H_0)$. The matrices $\bR_{\Sigma,1}$, 
$\bR_{\Sigma,2}$, $\bR_{B}$, $\bR_{b,1}$, $\bR_{b,2}$, 
and $\bR_{\mu}$ and the proofs are provided in
Supplementary Material~\ref{app:IF_cellMR}.
\end{proposition}
Let us look at a special case to get a feel for 
these results. Consider i.i.d. data 
$z_i=(x_i,y_i)^T$ for $i=1,\ldots,n$ with 
$z_i \sim N\!\left(\bzero,\;
\begin{bsmallmatrix} 1 & 0.9\\ 
 0.9 & 1 \end{bsmallmatrix}
\right)$. This obeys the simple linear model without
intercept $y_i = 0.9\,x_i + \eps_i$ with errors 
$\eps_i \sim N(0,0.19)$. Figure~\ref{fig_IFs_cellMR} 
shows the casewise and cellwise IFs of the cellMR
slope. They are bounded, indicating that the estimator 
is robust to both casewise and cellwise contamination. 
The casewise IF (left) has some flat regions, where
moving an outlier further away makes no difference
because its weight is zero there. The cellwise IF 
(right) is more smooth, with the influence largest 
near the center and gradually decreasing as the 
contamination becomes more distant. Together, these 
plots illustrate that the estimator effectively 
controls the influence of both casewise and cellwise 
perturbations. 
In Supplementary Material~\ref{app:lambda_IF} 
we also investigate how the ridge parameter 
$\lambda$ affects the casewise and cellwise IFs in Proposition~\ref{prop:lambda_IF}, and provide 
additional simulation results.
\begin{figure}[ht]
\centering
\vspace{4mm}
\includegraphics[width=0.45\textwidth]
    {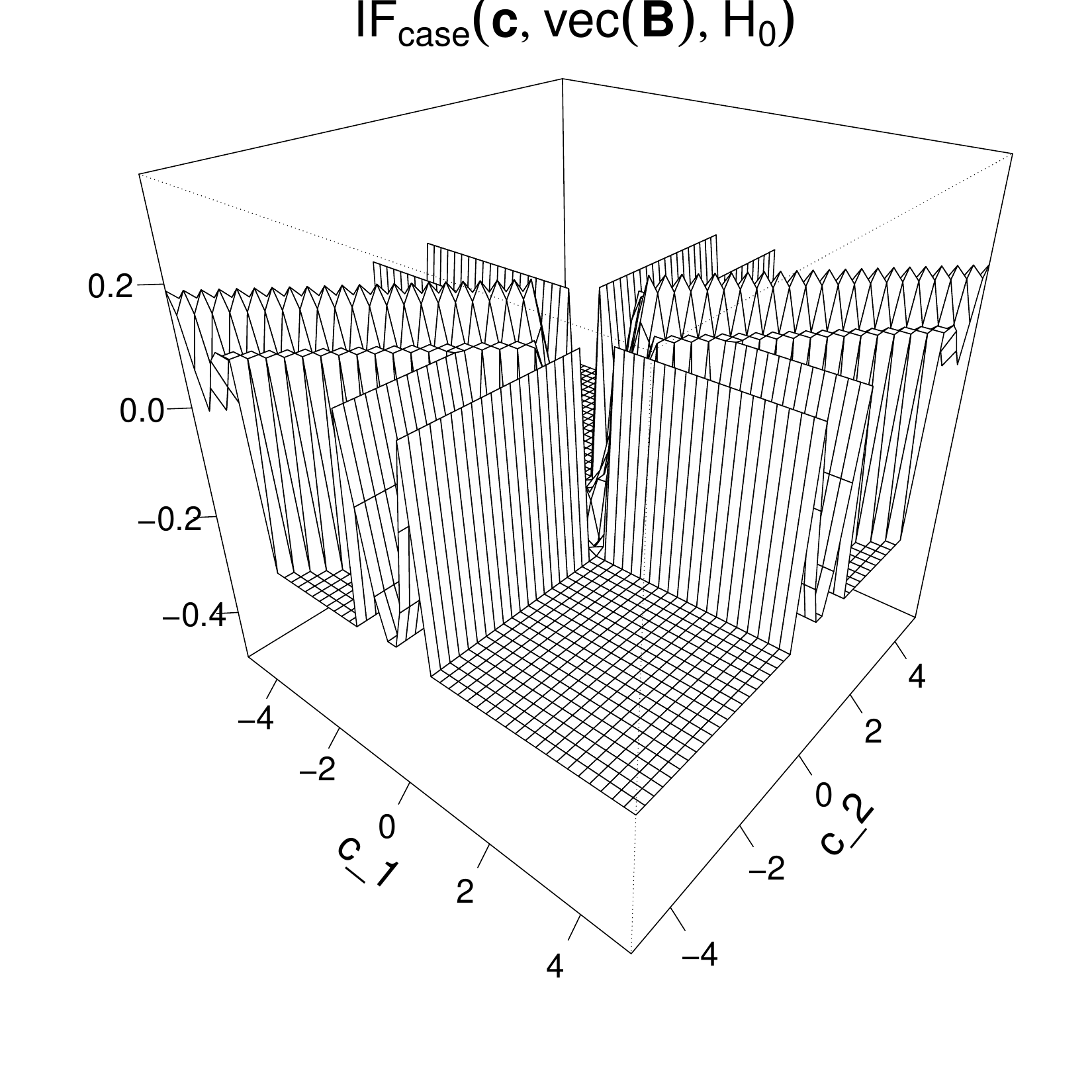}
\includegraphics[width=0.45\textwidth]
    {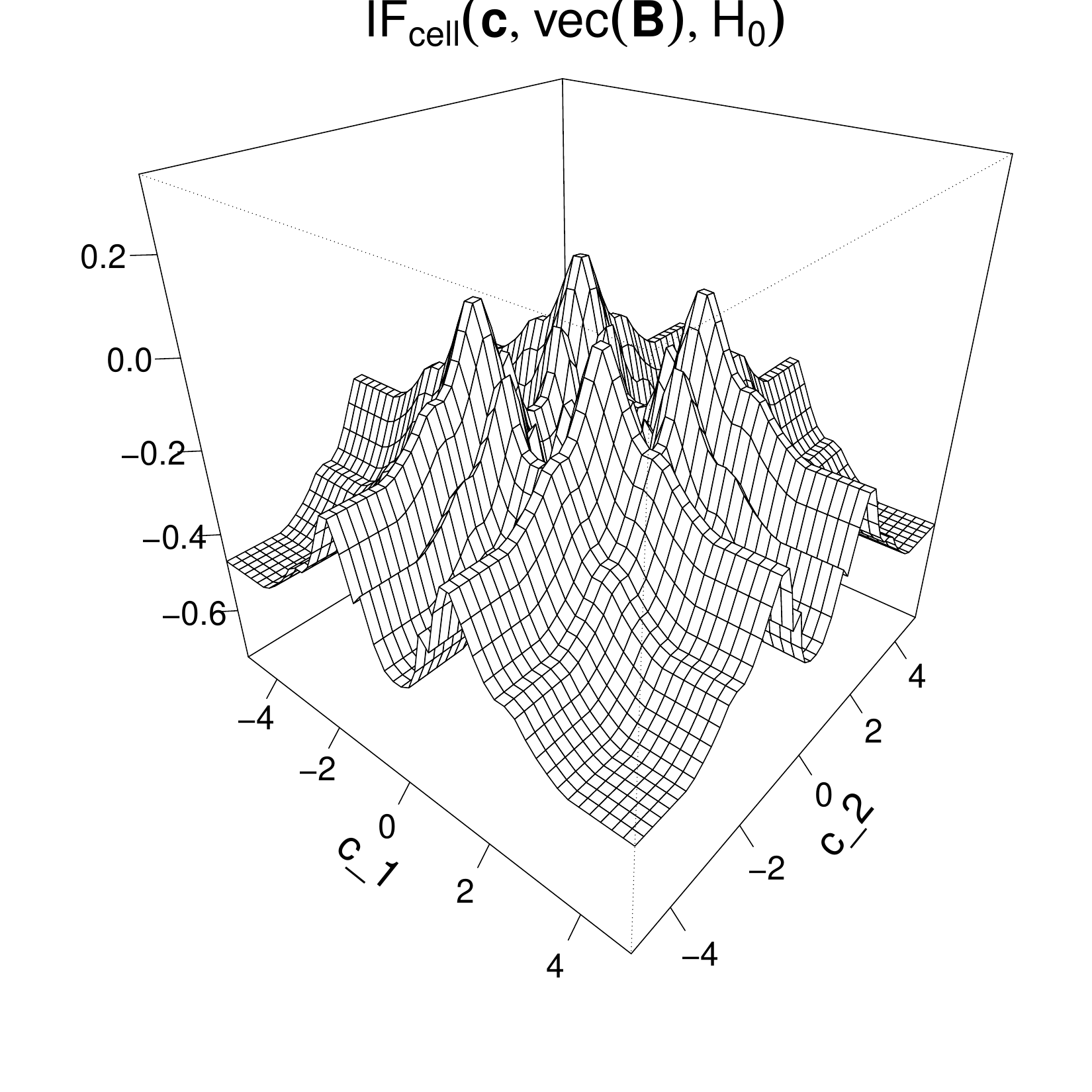}  
\vspace{-7mm}
\caption{The casewise (left) and cellwise (right)
IF of $\bB$ in a simple linear model.}
\label{fig_IFs_cellMR}
\end{figure}

Since the IF is a local robustness measure,
we complement it with a breakdown analysis, which
studies how much contamination is required to make 
the cellMR estimator unbounded. We assume that
the data are already standardized, so $\btZ=\bZ$.
For $m = 0,\ldots,n$ we define the cellwise and 
casewise replacement neighborhoods as
{
\setlength{\abovedisplayskip}{4pt}
\setlength{\belowdisplayskip}{4pt}
\begin{equation*}
\mathcal C_m(\bZ) := \left\{
\bZ^{(m)}=\{z_{ij}^{(m)}\}\in
\mathbb R^{n\times d}:
\max_{1\leqslant j\leqslant d}
\#\{i:z_{ij}^{(m)}\neq z_{ij}\}
\leqslant m \right\},
\end{equation*}
\begin{equation*}
  \mathcal R_m(\bZ) := \left\{
  \bZ^{(m)}= [\bz^{(m)}_1,\dots,\bz^{(m)}_n]^T
  \in\mathbb R^{n\times d}: \#
  \left\{i:\bz_i^{(m)}\neq\bz_i\right\}
 \leqslant m \right\}.
\end{equation*}
}
That is, a matrix in $\mathcal C_m(\bZ)$ is 
obtained by replacing at most $m$ cells in each 
column of $\bZ$ by arbitrary values. On the 
other hand, a matrix in $\mathcal R_m(\bZ)$ is 
obtained by replacing at most $m$ complete rows 
of $\bZ$ by arbitrary vectors in $\mathbb R^d$.
Following \cite{Challenges}, the 
finite-sample cellwise and casewise breakdown 
values of the cellMR slope are defined as 
\begin{equation*}
\eps_{n,\mathrm{cell}}^{*}(\bhB,\bZ)
  := \min_{m\in\{0,\ldots,n\}}\left\{
  \frac{m}{n}:
  \sup_{\bZ^{(m)}\in\mathcal C_m(\bZ)}
  \left\|\bhB(\bZ^{(m)}) - \bhB(\bZ)
  \right\|_F = +\infty \right\},
\end{equation*}
\begin{equation*}
  \eps_{n,\mathrm{case}}^{*}(\bhB,\bZ)
  := \min_{m\in\{0,\ldots,n\}} \left\{
  \frac{m}{n}:
  \sup_{\bZ^{(m)}\in\mathcal R_m(\bZ)}
  \left\|\bhB(\bZ^{(m)})-\bhB(\bZ)
  \right\|_F = +\infty \right\},
\end{equation*}
where $\bhB(\bZ^{(m)})$ is the estimate $\bhB$ 
on the dataset $\bZ^{(m)}$. These breakdown 
values are the smallest fractions of replaced 
cells per column, or replaced rows, that can drive
the cellMR slope arbitrarily far from its value 
on the original data. The following proposition
provides a lower bound for these quantities.
\begin{proposition}
\label{prop:slope-total}
Under assumptions~F1--F5 in Supplementary Material~\ref{app:cellwise-breakdown},
\begin{equation*}
  \eps_{n,\mathrm{case}}^{*}
  \left(\bhB,\bZ\right)\geqslant
  \eps_{n,\mathrm{cell}}^{*}
  \left(\bhB,\bZ\right)\geqslant
  \min\Bigg\{\left\lceil\frac{
  \Delta_{\bP}(\bZ)}{2W}\right\rceil,
  \left\lceil\frac{
  \Delta_{\bmu}^{\infty}(\bZ)}{2W}\right\rceil,
  \frac{1}{n}\left\lceil\frac{r_0(n-h+1)}{d}
  \right\rceil\Bigg\},
\end{equation*}
where $\lceil x\rceil$ denotes the smallest integer greater than or equal to $x$, and 
$\Delta_{\bP}(\bZ)$, $\Delta_{\bmu}^{\infty}(\bZ)$, 
$W$, $r_0$, $h$, and the proof are provided in Supplementary Material~\ref{app:cellwise-breakdown}.
\end{proposition}

A simpler bound is obtained when $\bmu_z$ 
and $\bP$ are fixed. Then $r_0=1$, and for 
$h=\left\lceil\frac{n+2}{2}\right\rceil$ we have
$\eps_{n,\mathrm{case}}^{*}(\bhB,\bZ)
\geqslant \eps_{n,\mathrm{cell}}^{*}(\bhB,\bZ)
\geqslant \frac{1}{n} \left\lfloor \frac{n+2}{4}
\right\rfloor$ so the lower bound is $0.25$ for
large $n$.
Figure~\ref{fig:results_varyingeps_reg_msee} in
Section~\ref{app:addsim} of the Supplementary Material plots
the empirical MSE of the slope in function of 
the contamination fraction, for
$p = q = 10, 25, 50$. In all settings the
MSE of cellMR is quite low up to 30\% of 
contamination and starts to go up at 40\%,
whereas the competing methods do less well. 

\section{Outlier detection}
\label{sec:outdetect}
We construct numerical and graphical diagnostics to 
gain further insight into outlying cells and cases 
in the responses and the predictors. For each pair 
$(\bx_i,\by_i)$ we define the regression residual
$\br_i:=\by_i-\bhy_i$\,, where 
$\bhy_i=\bhb+\bhB^{T}\bximp_i$ with $\bhb$ and $\bhB$ 
denoting the cellMR estimates, and where $\bximp_i$ is 
the imputed version of $\bx_i$ as in
Section~\ref{sec:oos}.  

The outlier map of the regression shows the distances 
$\RD_i=\sqrt{\br_i^{T}\bhSigma_{\beps}^{-1}\br_i}$ 
of the regression residuals versus the distances 
$\PD_i=\sqrt{(\bx_i-\btmu_X)^{T}\btSigma_{X}^{-1}
(\bx_i-\btmu_X)}$ of the predictors, where $\btmu_X$ 
and $\btSigma_X$ are robust estimates of the 
location and scatter of $\bx_i$ computed as 
in \citet{centofanti2025cellwise}. 
Figure~\ref{fig:outmap} illustrates the 
outlier map on a well-known genomic dataset from
\citet{shankavaram2007transcript}, also analyzed by
\citet{alfons2013sparseLTS}. It contains protein 
and gene expression measurements of $n=59$ human 
cancer cell lines. The goal is to predict the 
expression levels of the $q=3$ proteins  
MLH1, PRKCI, and RELA from $p=50$ genes.
\begin{figure}[ht]
\vspace{3mm}
\centering
\includegraphics[width=0.6\linewidth]
  {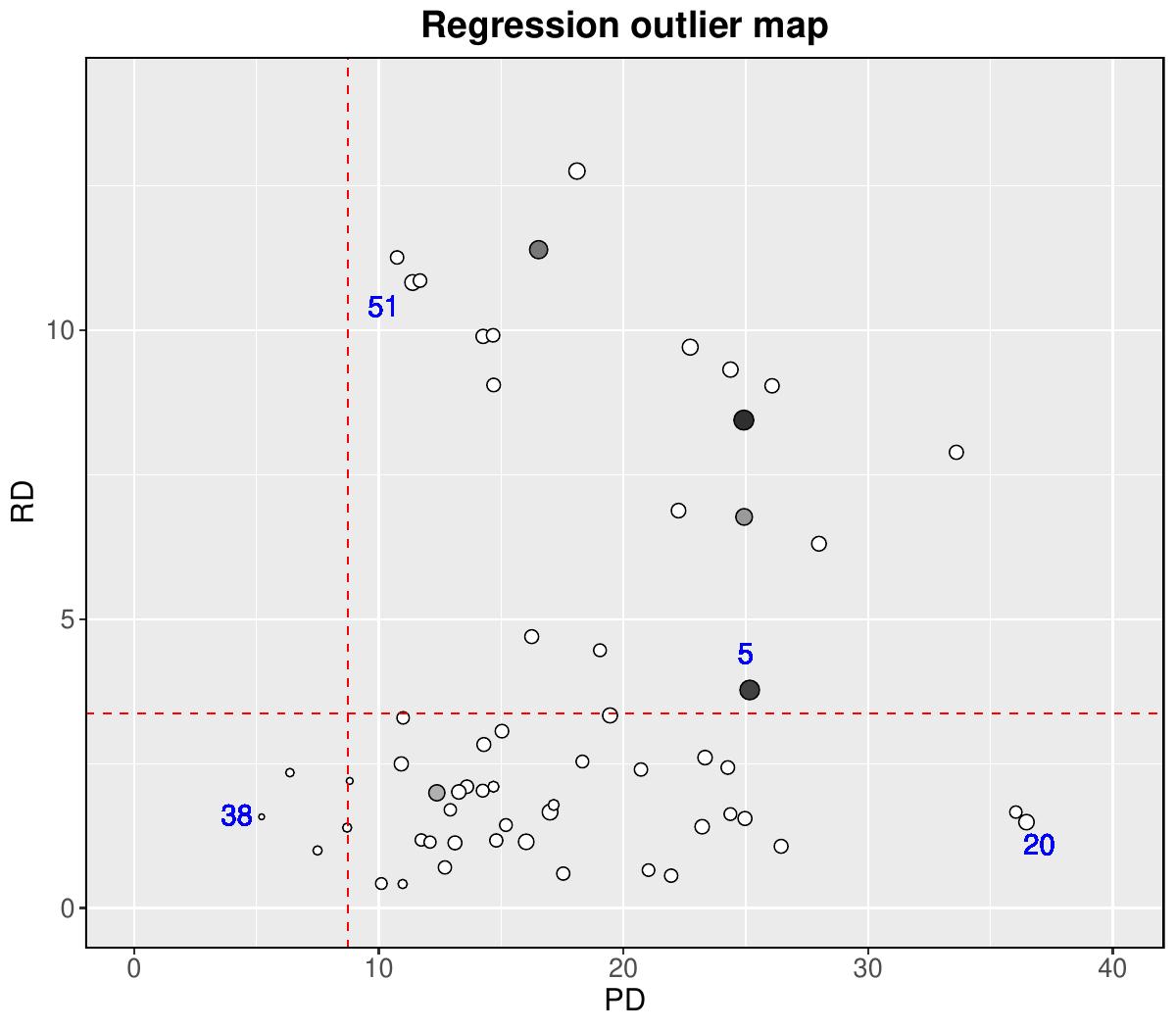}
\vspace{-3mm}
\caption{A regression outlier map of cellMR.}
\label{fig:outmap}
\end{figure}

The vertical dashed line is at the cutoff
$c_{\uPD}=\sqrt{\chi^2_{p,0.99}}$\,, while the 
horizontal one shows the cutoff 
$c_{\uRD}=\sqrt{\chi^2_{q,0.99}}$\,. Cases with 
small $\RD_i \leqslant c_{\uRD}$ and small 
$\PD_i \leqslant c_{\uPD}$ are considered regular. 
Cases with large $\PD_i$ and small 
$\RD_i$ are referred to as good leverage points. 
Cases with large $\RD_i$ are considered vertical 
outliers when $\PD_i$ is small, and bad leverage 
points when $\PD_i$ is large. 
Figure~\ref{fig:outmap} contains two of the three 
types of atypical observations. 

The size of each point is made proportional to 
$1-\frac{1}{d}\sum_{j=1}^{d} o_{ij}w_{ij}^{\cell}$.
A large point therefore indicates a case with many 
outlying cells in the predictor and/or the response.
The casewise outlyingness is visualized by coloring 
the points according to their standardized total 
deviation $\rt_i/\hsigma_2$ from~\eqref{eq:rt_i}.
The points are colored black when 
$\rt_i/\hsigma_2 > 1.5\,c_{\rt,0.99}$\,, white 
when $\rt_i/\hsigma_2 < c_{\rt,0.99}$\,, and use 
an interpolated grayscale in between. Here the 
cutoff $c_{\rt,0.99}$ is the $99$th percentile of 
the distribution of $\rt_i/\hsigma_2$ simulated 
for uncontaminated data. This display adapts 
the enhanced outlier map in \cite{cellPCA} from 
the PCA setting to regression.

We can also visualize outlying cells by a cellmap
\citep{DDC2018}. We construct the predictor 
cellmap by computing the vector 
$\br_i^X=\bD_{r^X}^{-1}(\bx_i-\bhx_i)$
for each $\bx_i$\, where $\bhx_i$ is the fitted
point $\bhx_i = \btV_X \bu^X_i + \btmu_X$ as in
Section~\ref{sec:oos}, and $\bD_{r^X}=\diag
(\hsigma^{r^X}_1,\dots,\hsigma^{r^X}_{p})$ with
$\hsigma^{r^X}_j=\sigma_M(
\{x_{ij}-\hx_{ij}\}_{i=1}^n)$.
These vectors are combined in an $n \times p$
matrix, and visualized by coloring. 
The left panel of Figure~\ref{fig:cellmap} shows
this map for the 4 cases labeled in
Figure~\ref{fig:outmap}.

\begin{figure}[ht]
\centering
\includegraphics[width=1\linewidth]
  {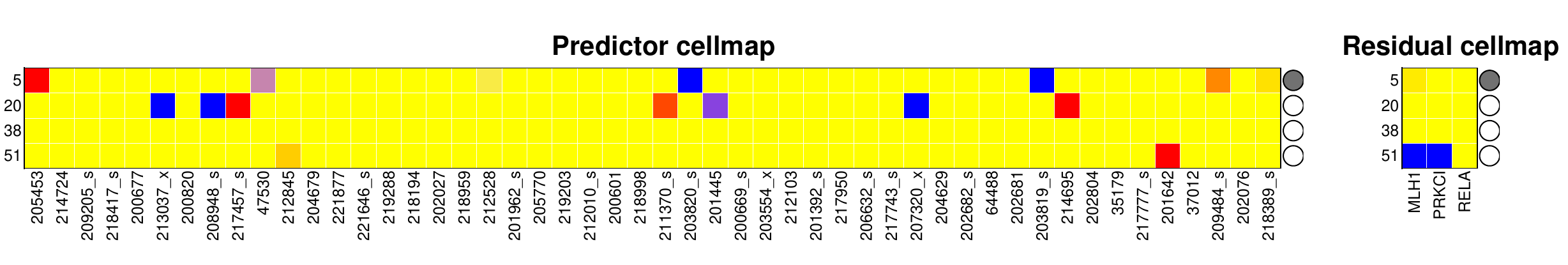} 

\vspace{-3mm}
\caption{cellMR predictor and residual cellmaps 
  of the 4 labeled cases.}
\label{fig:cellmap}
\end{figure}

In this map, cells with $|r_{ij}^X| \leqslant
c_{\cell}:=\sqrt{\chi^2_{1,0.99}}$ are considered 
regular and colored yellow. Missing cells would be
white. The remaining cells are flagged as cellwise 
outliers. Cells with $r_{ij}^X > c_{\cell}$ range 
from light orange to red, while cells with
$r_{ij}^X < -c_{\cell}$ are colored from light 
purple to dark blue. We add information on casewise 
outlyingness by drawing a circle to the right of 
each row, using the same color scheme as in the 
regression outlier map.

The right panel of Figure~\ref{fig:cellmap} is
the residual cellmap of the same cases. It 
shows the standardized residual vectors 
$\btr_i=\bD_{\beps}^{-1}(\by_i-\bhy_i)$ of 
the regression, where $\bD_{\beps}=\diag(
\hsigma^{\beps}_1,\dots,\hsigma^{\beps}_{q})$
 with $\hsigma^{\beps}_1,\dots,\hsigma^{\beps}_q$ 
the diagonal entries of $\bhSigma_{\beps}$. The
cells are colored according to the same scheme.

The predictor cellmap applies the cellPCA 
display of \citet{cellPCA} to the predictors,
showing deviations from their robust 
reconstruction.
The residual cellmap shows which response cells deviate
from the fitted regression relationship.

\section{Robust inference via cellBoot}
\label{sec:inf}

We want a confidence interval for a parameter of 
interest $\theta_0$ of the form
$\theta_0=(\bb^T , \vect(\bB)^T)\ba$ for some 
fixed vector $\ba \in \mathbb{R}^{q+pq}$. In
particular, each $b_j$ and $B_{j\ell}$ are of
this form. Let $\htheta_n$ denote the estimator 
of $\theta_0$ computed from the sample 
$\{(\bx_i,\by_i)\}_{i=1}^n$\,. From these data
the bootstrap draws a sample of size $n$ with 
replacement, denoted as 
$\{(\bx_i^*,\by_i^*)\}_{i=1}^n$ and used
to compute a bootstrap analog $\htheta_n^{*}$ 
of $\htheta_n$. Drawing many bootstrap samples
yields an empirical distribution of $\htheta_n^{*}$
from which one constructs a confidence interval 
$C_n\subseteq\mathbb{R}$ for $\theta_0$ with a 
certain level $1-\alpha$. An important question
is whether the coverage probability of $C_n$
is correct for $n \rightarrow \infty$.

\begin{definition} \label{def:asy_exact}
A confidence interval $C_n\subseteq\mathbb{R}$ for 
$\theta_0$ is asymptotically exact at 
level $1-\alpha$ if
{
\setlength{\abovedisplayskip}{4pt}
\setlength{\belowdisplayskip}{4pt}
\begin{equation}\label{eq:reg-coverage}
  \Pr\{C_n \;\mbox{\upshape{contains}}\; \theta_0\} \;
  \rightarrow\; 1-\alpha
  \quad \text{ as }\quad n \rightarrow \infty.
\end{equation}
}
\end{definition}
In the standard bootstrap setting, where
the interval endpoints approach $\htheta_n$ in
probability and $\htheta_n$ converges in probability
to a fixed value,
a necessary condition for \eqref{eq:reg-coverage}
is that $\htheta_n$ is a consistent estimator of
$\theta_0$, that is, 
$\htheta_n \rightarrow_p \theta_0$
where $\rightarrow_p$ denotes convergence in
probability. Indeed, if $\htheta_n$ converged
to a value different from $\theta_0$, the
coverage probability would tend to zero as
$C_n$ shrinks around $\htheta_n$ for
$n\rightarrow\infty$.

\subsection{The Indirect Inference estimator}
\label{sec:II}
We would like to use the robust cellMR estimator
$(\bhb,\bhB)$ of Section~\ref{sec:method} to
estimate $\htheta_n$. However, this estimator is 
not guaranteed to be consistent due to its 
cellwise construction.
A solution comes from \emph{indirect inference} 
(II), a simulation-based method that provides a 
consistent estimator starting from an inconsistent 
one \citep{gourieroux1993indirect, guerrier2019}. 
This approach follows a two-step procedure: in the 
first step, an auxiliary estimator is computed 
that may not be consistent; in the second step, 
simulation-based bias correction is applied, 
yielding a consistent estimator. Consistency of 
the cellMR estimator relies on the consistency of 
the estimators of the location $\bmu$ and scatter 
$\bSigma$ of the joint distribution $F$ of 
$(\bx_i,\by_i)$. We will employ II to compute 
consistent estimators of $\bmu$ and $\bSigma$.

Let $\{F_{\btheta} : \btheta \in \bTheta\}$ 
denote a parametric family of distributions 
indexed by the parameter
$\btheta = (\bmu^T,\vecth_s(\bSigma)^T)^T$.
Here $\vecth_s(\bSigma)$ denotes the scaled 
half--vectorization, that first multiplies
the off-diagonal entries of $\bSigma$ by 
$\sqrt{2}$, and then stacks the lower 
triangular entries (including the diagonal) 
on top of each other. With this convention, the 
Euclidean norm of $\vecth_s(\bSigma)$ coincides
with the Frobenius norm of $\bSigma$.
The parameter space $\bTheta$ contains the 
admissible parameters: $\bTheta = 
\Bigl\{\,(\bmu^T,\vecth_s(\bSigma)^T)^T 
\in \mathbb R^{d + d(d+1)/2} : \bSigma \in 
\mathbb S^d,\ \|\bmu\|\leqslant M, 
c \leqslant \lambda_{\min}(\bSigma)\leqslant 
\lambda_{\max}(\bSigma)\leqslant C \Bigr\}$ 
for $M<\infty$ and $0< c \leqslant C<\infty$.
Here $\mathbb S^d$ denotes the space of 
symmetric $d\times d$ matrices. 
We assume that  $F=F_{\btheta_0}$ for some
$\btheta_0 \in \bTheta$.

Let $\hpi_n(\boheta_n)$ denote the auxiliary estimator
of $\btheta_0$ computed from the observed sample 
$\{(\bx_i,\by_i)\}_{i=n}^n$ which depends on a 
vector of tuning parameters 
$\boeta \in \bH \subseteq \mathbb{R}^r$, where 
$\boheta_n$ is an estimate of $\boeta$ computed
from the observed sample. We further denote by 
$\hpi(\btheta,\boheta_n,n)$ the auxiliary estimator 
evaluated on a generic sample of size $n$ generated 
from $F_{\btheta}$, with $\btheta\in\bTheta$, and 
computed with the tuning parameters in $\boheta_n$.
The indirect estimator is then defined as
{
\setlength{\abovedisplayskip}{4pt}
\setlength{\belowdisplayskip}{4pt}
\begin{equation}
\label{eq:indirect_obj}
  \bhtheta_n = \argzero_{\btheta \in \bTheta}
  \bigl\{\, \hpi_n(\boheta_n) - 
  \pibar(\btheta,\boheta_n,n) \bigr\},
\end{equation}
}
where
{
\setlength{\abovedisplayskip}{4pt}
\setlength{\belowdisplayskip}{4pt}
\begin{equation}
  \pibar(\btheta,\boheta_n,n) = \frac{1}{H} 
  \sum_{h=1}^{H} \hpi_{h}(\btheta,\boheta_n,n),
\end{equation}
}
and $\hpi_{h}(\btheta,\boheta_n,n)$ denotes the
value of $\hpi(\btheta,\boheta_n,n)$ computed on 
the $h$-th simulated dataset of size $n$ drawn from 
$F_{\btheta}$. The observations drawn from 
$F_{\btheta}$ are generated using the same random
number generator seed for all $\btheta$ to ensure 
the objective function is deterministic.

\begin{proposition}
\label{prop:2}
Under assumptions~A1-A3 in Supplementary 
Material~\ref{app:asyCBoot}, any sequence 
$\bhtheta_n$ satisfying
$\bigl\|\hpi_n(\boheta_n)-\pibar(\bhtheta_n,\boheta_n,n)
\bigr\| \rightarrow_p 0$ is consistent for $\btheta_0$, 
that is, $\bhtheta_n \rightarrow_p \btheta_0$.
\end{proposition}

The proof is presented in Supplementary Material~\ref{app:asyCBoot}.

To obtain a solution of \eqref{eq:indirect_obj}, 
we iteratively update the estimate 
$\bhtheta_n^{(\ell)}$ as
{
\setlength{\abovedisplayskip}{4pt}
\setlength{\belowdisplayskip}{4pt}
\begin{equation}
\label{eq:thetaupdate}
  \bhtheta_n^{(\ell)}
  \;=\;
  \Pi_{\bTheta}\!\left(
  \hpi_n(\boheta_n)+
  \Big[
    \bhtheta_n^{(\ell-1)}
    -
    \pibar(\bhtheta_n^{(\ell-1)},\boheta_n, n)
  \Big]\right),
\end{equation}
}
starting from $\bhtheta_n^{(0)}=\hpi_n(\boheta_n)$.
Here the function $\Pi_{\bTheta}$ brings its 
argument into $\bTheta$ if it wasn't already in it.
Starting from a vector $(\ba^T,\vecth_s(\bL)^T)^T$ 
in $\mathbb R^{d + d(d+1)/2}$, 
$\Pi_{\bTheta}(\ba^T,\vecth_s(\bL)^T)^T := \bigl( 
\tilde\ba^T,\ \vecth_s(\tilde\bL))^T \bigr)^T$,
where $\tilde\ba$ is equal to $\ba$ if 
$\|\ba\|\leqslant M$, and 
$\displaystyle M(\ba/\|\ba\|)$ otherwise. This 
clips $\ba$ to the ball of radius $M$. Also, $\bL$ 
is decomposed as $\bL = \bQ \bLambda \bQ^T$ with $\bLambda=\diag(\lambda_1,\dots,\lambda_d)$, and
turned into $\tilde\bL=\bQ\bLambda_c\bQ^T$ where
$\bLambda_c=\diag(\tilde\lambda_1,\dots,
\tilde\lambda_d)$, with
$\tilde\lambda_j=\min\{\max\{\lambda_j,c\},C\}$.
This clips the eigenvalues of $\bL$ to the 
interval $[c,C]$, leaving eigenvectors unchanged.
Pseudocode of the algorithm for
\eqref{eq:indirect_obj} is provided
in Supplementary
Material~\ref{app:pseudocode}.

Under appropriate conditions, the limit of the 
sequence $\bhtheta_n^{(\ell)}$ for 
$\ell \rightarrow \infty$ indeed exists and is 
the unique solution $\bhtheta_n$ 
of~\eqref{eq:indirect_obj}, as shown in the 
following proposition.

\begin{proposition}
\label{pro:proj-fp-contraction}
Under assumptions~B1-B2 in Supplementary 
Material~\ref{app:asyCBoot}, we have that $\bhtheta_n$ is 
unique and the sequence $\bhtheta_n^{(\ell)}$ converges
in norm to $\bhtheta_n$ with linear rate for every 
$\bhtheta_n^{(0)}\in\bTheta$, that is, 
$\|\bhtheta_n^{(\ell)}-\bhtheta_n\| \leqslant L^\ell 
\|\bhtheta_n^{(0)}-\bhtheta_n\|$ for some $L<1$.
\end{proposition}

\subsection{The auxiliary FastCellCov estimator}
\label{sec:FastCellCov}

The robust cellMR regression is derived 
directly from the cellCov estimates $\bhmu$ 
and $\bhSigma$ through formula~\eqref{eq:ridge}.
In our setting, where the data may be 
contaminated by cellwise and casewise outliers, 
we would like to use cellCov as our auxiliary 
estimator. However, in the II algorithm the 
auxiliary estimator has to be recomputed many 
times, and using cellCov would be too expensive 
computationally. This motivates the construction 
of the FastCellCov estimator, which preserves 
the main robustness ideas of cellCov but avoids 
repeating the most expensive parts of the 
algorithm.

FastCellCov starts by applying cellCov to the 
original sample $\btz_1, \dots, \btz_n$. Now 
consider another sample
$\btz_1^*, \dots, \btz_n^*$\,, that may be 
a bootstrap sample or a simulated sample of 
size $n$ generated from $F_{\btheta}$. We then 
standardize it by computing
$\bz_i^*=\bhD^{-1}\btz^*_i$ where $\bhD$ is 
the diagonal matrix of scale estimators of the 
original sample, as in 
Section~\ref{sec:estimator}.
Then, FastCellCov computes a robust center 
$\btmu_F$ and a robust covariance matrix 
$\btSigma_F$ of the $\bz_i^*$. These quantities 
are computed as weighted versions of the sample 
mean and covariance matrix. The weights are 
designed to downweight both cellwise and 
casewise outliers, while also accounting for 
the possible presence of missing values, and 
are constructed by reusing the robust structure 
learned from the original sample. This 
construction follows the same principles as 
the Detecting Deviating Cells (DDC) method of 
\citet{DDC2018} where each cell is compared 
with a robust prediction obtained from the 
variables that are sufficiently correlated 
with it, and cells that are incompatible with 
this prediction receive a smaller weight. Thus, 
FastCellCov should be viewed as a one-step 
approximation to \mbox{cellCov}. The entire 
reasoning can be found in 
Supplementary Material~\ref{app:FastCellCov}.
The final FastCellCov estimates of $\bmu$ 
and $\bSigma$ are then given by 
$\bhmu_{F} := \bhD \btmu_{F}$ and 
$\bhSigma_{F} := \bhD \btSigma_{F}\bhD$.

\subsection{The cellBoot algorithm}
\label{sec:cellBoot_algo}
The cellBoot algorithm consists of the following
steps:
\begin{enumerate}[label={Step \arabic*.}, 
  align=left, leftmargin=*, nosep]
\item Apply the \texttt{cellMR} estimator to the 
  original dataset $\{(\bx_i,\by_i)\}_{i=1}^n$ 
  to compute estimates of $\bmu$, $\bSigma$, and 
  the tuning parameters $\rk$ and $\lambda$.
\item Generate bootstrap samples  
  $\{(\bx^*_{bi},\by^*_{bi})\}_{i=1}^n$ from the
  data, for $b=1,\dots,B$.
\item For each 
  $\{(\bx^*_{bi},\by^*_{bi})\}_{i=1}^n$\,, 
  compute the FastcellCov estimates $\bhmu_{F,b}$ 
  and $\bhSigma_{F,b}$ as in 
  \mbox{Section~\ref{sec:FastCellCov}}. Then 
  compute the corrected versions 
  $\bhmu_{II,b}^{*}$ and 
  $\bhSigma_{II,b}^{*}$ by applying II as in
  Section~\ref{sec:II}. Next, compute the 
  estimates $(\bhb^*_{b}, \bhB^*_b)$ of the 
  coefficient matrix $\bB$ and intercept $\bb$ 
  according to~\eqref{eq:ridge}.
\item Construct the confidence interval for the 
  parameter of interest at level $1-\alpha$ by 
  computing the $\alpha/2$ and $1-\alpha/2$ 
  quantiles of the bootstrap estimates 
  $\{\htheta^*_b\}_{b=1}^B$ with $\htheta^*_b=
  ((\bhb_{b}^*)^T, \vect(\bhB_b^*)^T) \ba$.
\end{enumerate}
The number of bootstrap replications $B$ determines 
the accuracy: higher $B$ values yield more stable 
confidence intervals, at the cost of increased 
computation. In our implementation we set
$B = 1000$, which offers a suitable trade-off 
between precision and computation time.

The following theorem establishes the 
consistency of the bootstrap distribution of 
$\sqrt{n}(\bhtheta_n^*-\bhtheta_n)$ for 
approximating the sampling distribution of 
$\sqrt{n}(\bhtheta_n-\btheta_0)$. Here
$\bhtheta_n=\big((\bhmu_{II})^T,
\vecth_s(\bhSigma_{II})^T\big)^T$ denotes the 
II estimator using FastcellCov as auxiliary 
estimator, \mbox{applied} to the observed sample,
and $\bhtheta_n^*=\big((\bhmu_{II}^*)^T,
\vecth_s(\bhSigma_{II}^*)^T\big)^T$
is its bootstrap counterpart.

\begin{theorem} \label{the:bootstrap-theta}
Under assumptions~D1--D7 in Supplementary 
Material~\ref{app:asyCBoot}, the bootstrap 
distribution of $\sqrt{n}(\bhtheta_n^*-\bhtheta_n)$ 
is consistent for the distribution of
$\sqrt{n}(\bhtheta_n-\btheta_0)$ in 
Kolmogorov-Smirnov distance, that is,
{
\setlength{\abovedisplayskip}{4pt}
\setlength{\belowdisplayskip}{4pt}
\begin{equation*}
  \sup_{\bx\in\mathbb{R}^{d+d(d+1)/2}} \Big|
  \Pr\bigl(\sqrt{n}(\bhtheta_n-\btheta_0)
  \leqslant \bx \bigr) - {\Pr}^* 
  \bigl(\sqrt{n}(\bhtheta_n^*-\bhtheta_n)
  \leqslant \bx \,\bigr)\Big|\rightarrow_p 0,
\end{equation*}
}
where ${\Pr}^*$ is computed under the bootstrap 
distribution, conditional on the observed data.
\end{theorem}

Define the scalar parameters 
$\htheta_n=\bigl(\bhb^T,\vect(\bhB)^T\bigr)\ba$
and $\htheta_n^*=
\bigl((\bhb^*)^T,\vect(\bhB^*)^T\bigr)\ba$,
where $(\bhb,\bhB)$ and $(\bhb^*,\bhB^*)$ 
are obtained from $\bhtheta_n$ and
$\bhtheta_n^*$ through \eqref{eq:ridge}.
Then Theorem~\ref{the:bootstrap-theta} applies
directly to the distributions of 
$\sqrt{n}(\htheta_n-\theta_0)$ and
$\sqrt{n}(\htheta_n^*-\htheta_n)$:

\begin{corollary}
\label{prop:bootstrap-theta-linear-reg}
Under assumptions~D1-D8 in Supplementary 
Material~\ref{app:asyCBoot}, it holds that
{
\setlength{\abovedisplayskip}{4pt}
\setlength{\belowdisplayskip}{4pt}
\begin{equation*}
  \sup_{x\in\mathbb R} \Big| \Pr\bigl(
  \sqrt{n}(\htheta_n-\theta_0)\leqslant x \bigr)
  - {\Pr}^*\bigl(\sqrt{n}(\htheta_n^*-\htheta_n)
  \leqslant x \bigr)\Big|\rightarrow_p 0.
\end{equation*}
}
\end{corollary}

The next corollary of 
Theorem~\ref{the:bootstrap-theta} guarantees
that the proposed bootstrap confidence interval 
for $\theta_0$ is asymptotically exact at level
$(1-\alpha)$.

\begin{corollary}
\label{cor:efron-percentile-exact}
Under assumptions~D1--D8 in Supplementary 
Material~\ref{app:asyCBoot}, it holds for all 
\mbox{$\alpha\in(0,1)$} that 
$\hC_n:= \bigl[\hc_n^*(\alpha/2),\ 
\hc_n^*(1-\alpha/2)\bigr]$ is asymptotically 
exact at level $1-\alpha$, where 
$\hc_n^*(\alpha/2)$ and 
$\hc_n^*(1-\alpha/2)$ are the $\alpha/2$ 
and $1-\alpha/2$ empirical quantiles
of the distribution of $\htheta_n^{\,*}$.
\end{corollary}

The next proposition gives the computational 
complexity of FastCellCov and cellBoot.
\begin{proposition}
\label{prop:complexity-cellBoot}
Assume $\rk\leqslant \rk_{\max}$.
Then the time complexity of FastCellCov 
is $O(nd^2)$.
Suppose further that $B$ and $H$  are fixed.
For Gaussian simulated samples, 
the total time complexity of cellBoot is $O\!\left(nd^2+nd\log(n)+d^3\right)$.
If $d\leqslant n$, this simplifies to $O\!\left(nd^2+nd\log(n)\right)$.
\end{proposition} 

The proof is in Supplementary
Material~\ref{sec:supp-complexity-theor},
together with runtime experiments illustrating
the large speedup from using FastCellCov
instead of cellCov as the auxiliary estimator.

\subsection{Robustness properties of cellBoot}
\label{sec:robustness_inf}

We now study the influence of data contamination
on the cellBoot confidence interval $\hC_n$ for
$\theta_0$. We will use the contamination model 
of Section~\ref{sec:robustness_cellMR}, 
with $H_0:=F_{\btheta_0}$.
The center of $\hC_n$ is $m_n :=
(\hc^*(\alpha/2)+\hc^*(1-\alpha/2))/2$, 
and its length is $\ell_{n,\alpha}=
\hc^*(1-\alpha/2)-\hc^*(\alpha/2)$.
Since the bootstrap estimator admits the 
decomposition 
$\htheta_n^{\,*}=\htheta_n+n^{-1/2}
\sqrt n(\htheta_n^{\,*}-\htheta_n)$, a
bootstrap $\gamma$-quantile can be written as
$\hc_n^*(\gamma) = \htheta_n+\frac{1}{\sqrt n}\,
\hq_n^*(\gamma)$, where $\hq_n^*(\gamma)$ 
denotes the $\gamma$-quantile of the conditional 
distribution of
$\sqrt n(\htheta_n^{\,*}-\htheta_n)$ given the 
observed sample. The center and the length of 
the interval can thus be written as
$m_n = \htheta_n+\frac{1}{2\sqrt n}\bigl\{
\hq^*(\alpha/2)+\hq^*(1-\alpha/2)\bigr\}$
and $\ell_{n,\alpha} = \frac{1}{\sqrt n}
\bigl\{\hq_n^*(1-\alpha/2)-\hq_n^*(\alpha/2)
\bigr\}$.

Under $H_0$, Theorem~\ref{the:bootstrap-theta} 
implies that the center $m_n$ and the scaled 
length $\sqrt n\,\ell_{n,\alpha}$ converge in 
probability to $\theta_0$ and 
$q(1-\alpha/2)-q(\alpha/2)$, where
$q(\alpha/2)$ and $q(1-\alpha/2)$ are the 
$\alpha/2$ and $(1-\alpha/2)$ quantiles of the 
asymptotic distribution of
$\sqrt n(\htheta_n-\theta_0)$.

At a generic distribution $H$ we define the center 
functional as $m(H)=T(H)$ where $T(H)$ is the 
functional version of the estimator $\htheta_n$
under $H$. We define the scaled length functional as
$\tl_{\alpha}(H)=Q_H(1-\alpha/2)-Q_H(\alpha/2)$,
where $Q_H(\gamma)$ denotes the $\gamma$-quantile 
of the asymptotic distribution of 
$\sqrt n\{T(\widehat H_n)-T(H)\}$, where 
$\widehat H_n$ is the empirical distribution 
obtained from an i.i.d.\ sample of size $n$ drawn
from $H$. 

\begin{proposition}
\label{lem:IF-length-final-weakA2}
Under assumptions~E1--E2 in Supplementary 
Material~\ref{sec:derivationIF_cellboot} it holds
that
{
\setlength{\abovedisplayskip}{4pt}
\setlength{\belowdisplayskip}{5pt}
\begin{equation*}
\IFu_{\case}(\bc,m,H_0)=\bmed_c^TZ_F(\bc,\boeta_0),
\end{equation*}
\begin{equation*}
\IFu_{\cell}(\bc,m,H_0)= d\, \bmed_c^T\left(
\sum_{j=1}^d \E_{H(j,\bc)}[Z_F(X,\boeta_0)]
\right),
\end{equation*}
\begin{align*}
  \IFu_{\case}(\bc,\tl_{\alpha},H_0)=
  2z_{1-\alpha/2}\Big[m_{s,1}&
  +\bmed_{s,2}^T\,\IFu_{\case}(\bc,\boeta,H_0)\\
& +\bmed_{s,3}^T Z_F(\bc,\boeta_0)
  +\bmed_{s,4}^T\!\bigl(
  Z_F(\bc,\boeta_0)\otimes Z_F(\bc,\boeta_0)
  \bigr)\Big],
\end{align*}
\begin{align*}
  \IFu_{\cell}(\bc,\tl_{\alpha},H_0)=
  2z_{1-\alpha/2}\Big[dm_{s,1}
  +&\bmed_{s,2}^T\,\IFu_{\cell}(\bc,\boeta,H_0)
  +d\, \bmed_{s,3}^T\left( \sum_{j=1}^{d}\E_{H(j,\bc)}
  \!\big[Z_F(X,\boeta_0)\big]\right)\\
  &+d\, \bmed_{s,4}^T\left(\sum_{j=1}^{d}
  \E_{H(j,\bc)}\!\bigl[Z_F(X,\boeta_0)\otimes 
  Z_F(X,\boeta_0)\bigr]\right)\Big].
\end{align*}
}
Here $\otimes$ denotes the Kronecker product,
$\boeta_0:=\boeta(H_0)$ with $\boeta(\cdot)$ the 
functional corresponding to the tuning parameter 
estimator $\boheta_n$ in the auxiliary estimator, 
$H(j,\bc)$ is the distribution of $X \sim H_0$ 
but with its $j$-th component fixed at the 
constant $c_j$\,, and $z_\gamma$ denotes the 
$\gamma$-quantile of the standard normal
distribution. The quantities $\bmed_c$, $m_{s,1}$, 
$\bmed_{s,2}$, $\bmed_{s,3}$, $\bmed_{s,4}$ and 
the function $Z_F$ are defined in Supplementary 
Material~\ref{sec:derivationIF_cellboot} where 
these results are proved, and 
$\IFu_{\case}(\bc,\boeta,H_0)$ 
and $\IFu_{\cell}(\bc,\boeta,H_0)$ 
are the casewise and cellwise IFs of $\boeta$.
\end{proposition}

\begin{figure}[!ht]
\centering
\includegraphics[width=0.45\textwidth]
    {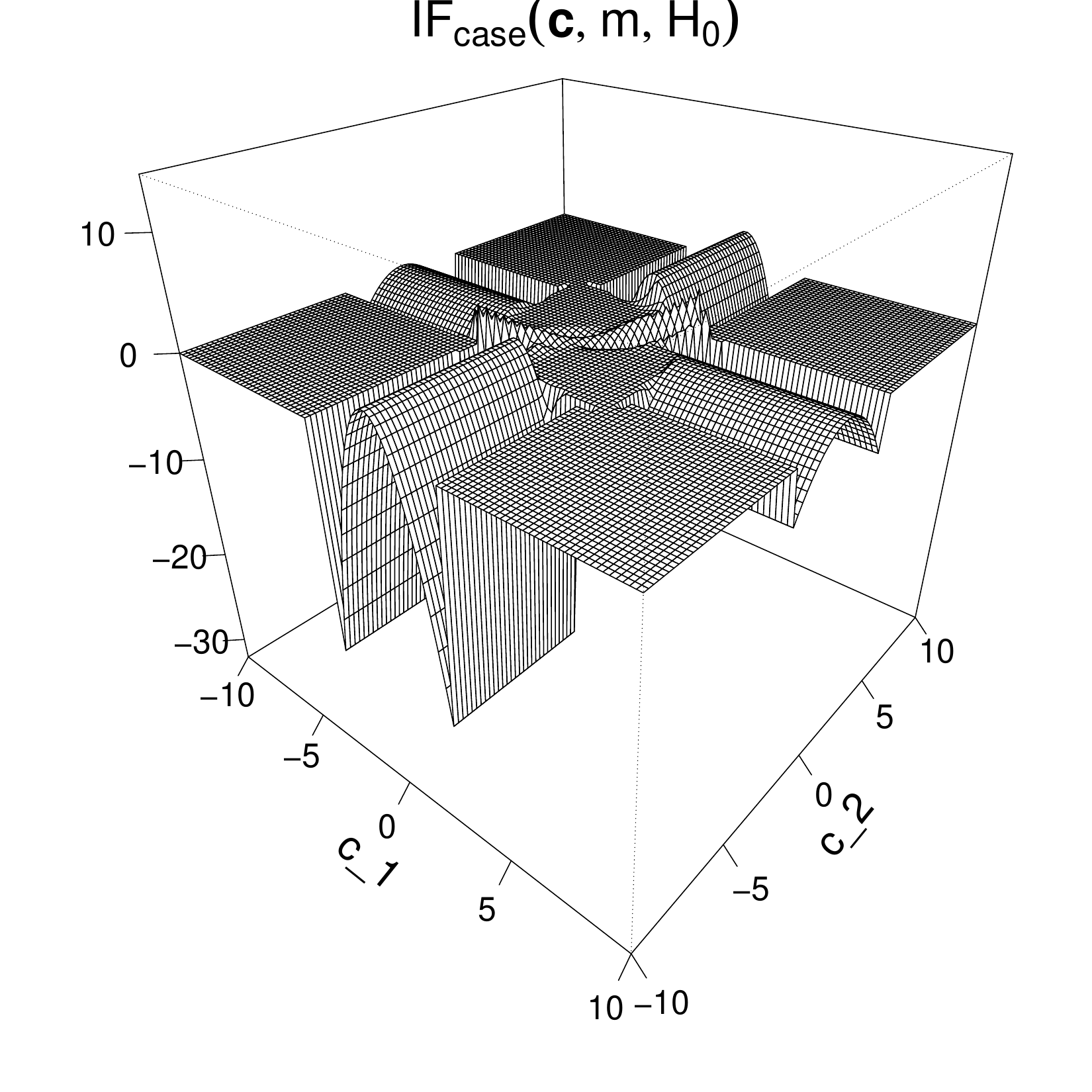}
\includegraphics[width=0.45\textwidth]
    {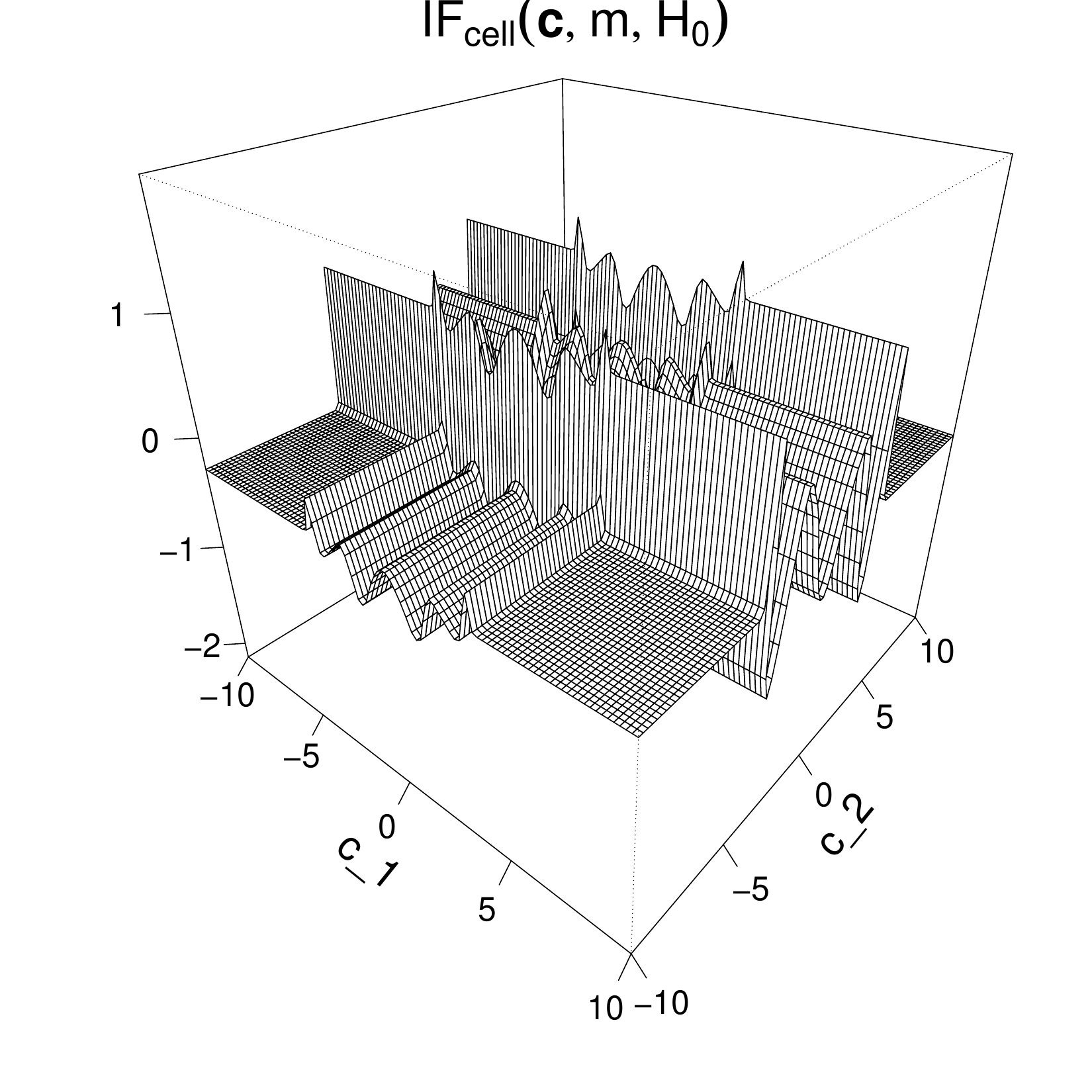}\\
\includegraphics[width=0.45\textwidth]
    {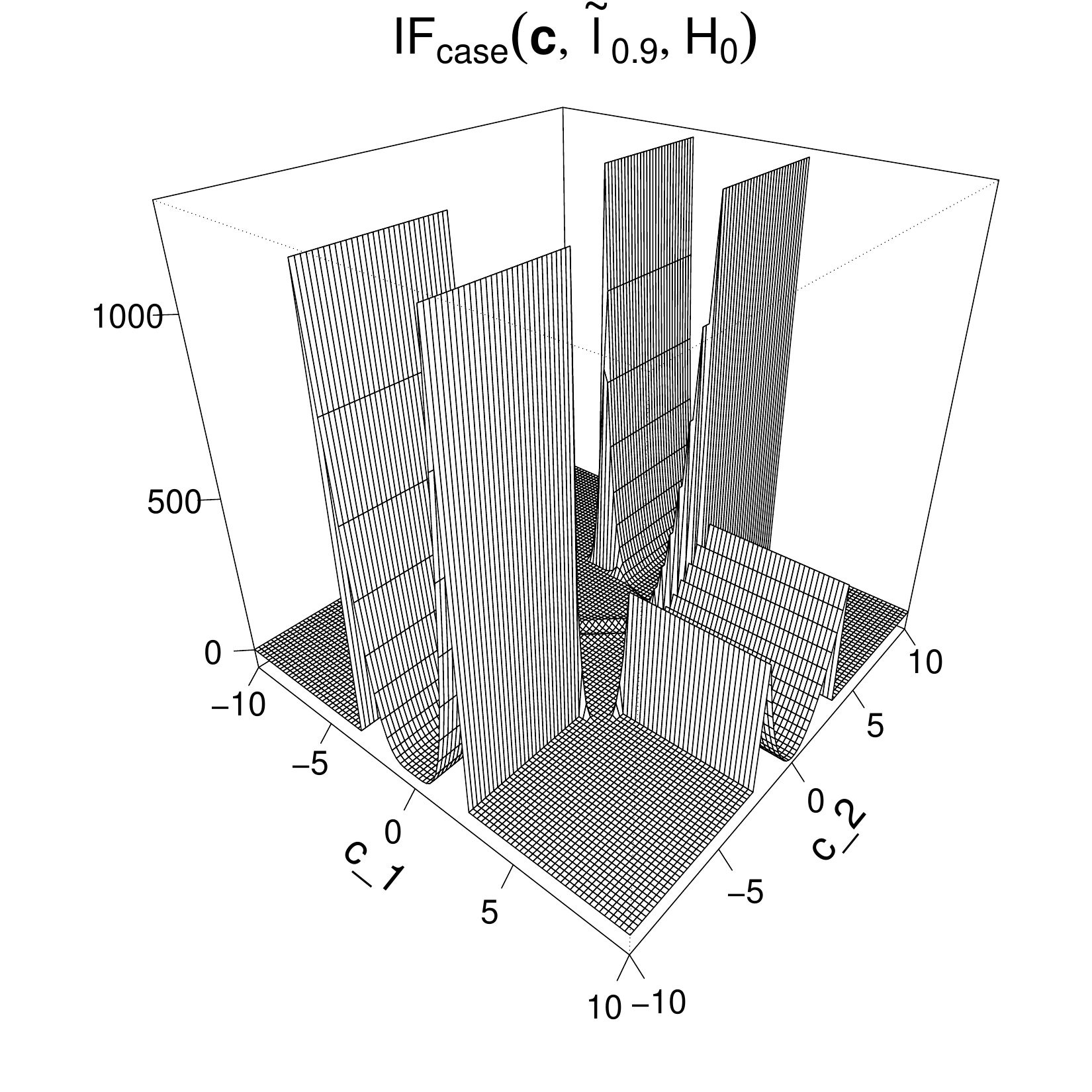}
\includegraphics[width=0.45\textwidth]
    {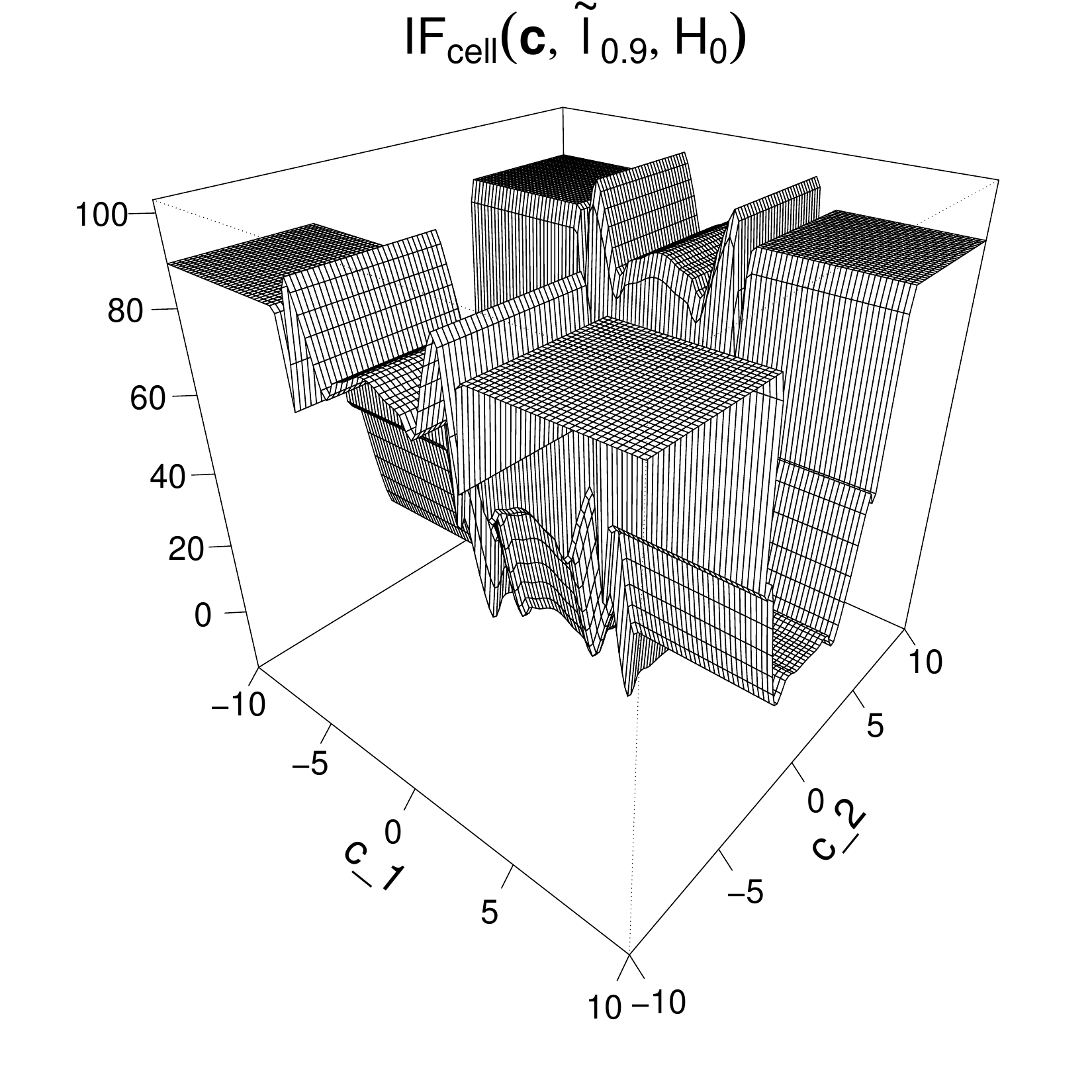}\\
\vspace{-4mm}
\caption{The casewise (left) and cellwise (right)
influence functions of $m$ and $\tl_{\alpha}$ 
for the proposed bootstrap percentile confidence 
interval of the slope, with $1-\alpha=0.9$.}
\label{fig_IFs}
\end{figure}

To gain intuition for these results, we 
consider the bivariate model introduced in 
Section~\ref{sec:robustness_cellMR}.
Figure~\ref{fig_IFs} shows the casewise and 
cellwise IF of the center $m$ and the scaled
length $\tl_{\alpha}$ of the $1-\alpha=0.9$
confidence interval of the slope. Their 
shape is involved due to the complex 
construction of the cellBoot intervals. 
The main feature is that all four IFs are 
bounded, indicating that $m$ and 
$\tl_{\alpha}$ are robust against casewise 
and cellwise contamination.

\section{Simulation study}
\label{sec:sim}
We assess the performance of cellMR and cellBoot
through a Monte Carlo experiment, where the
clean data are generated according to the linear
model~\eqref{eq:mainmodel}. The predictors follow
a multivariate normal distribution with
$\bmu=\bzero$ and covariance $\bSigma_x$ with
entries $\sigma_{x,j\ell}=(-0.4)^{|j-\ell|}$.
The coefficient matrix $\bB$ has entries drawn
from {$N(0.5,0.04)$}. To control the
signal-to-noise ratio (SNR), the error
covariance is $\bSigma_\eps=\bI_q(\Tr(
\bB^T\bSigma_x\bB)/q)/\text{SNR}$.
We fix {$\text{SNR}=30$}.

Three contamination scenarios are considered.
In the cellwise outlier scenario, we replace
a fraction $\eps$ of random entries
in each of $\bX$ and $\bY$ by independent
draws from $N(\gamma\sqrt{\sigma_{jj}},
0.25\sigma_{jj})$,
where $\sigma_{j\ell}$ are the entries of the
$(p+q)\times(p+q)$ covariance matrix $\bSigma$.
The contamination position $\gamma$
varies from $1$ to $7$. In the casewise
outlier scenario, a fraction $\eps$ of the
$(\bx_i^{T},\by_i^{T})^{T}$
are generated from
$N(0.2\gamma d\sqrt{d}\,
\be/\sqrt{\be^T\bSigma^{-1}\be},\bSigma)$,
where $\be$ is
selected among the eigenvectors corresponding
to the five smallest eigenvalues of $\bSigma$.
These directions generate departures from the joint
dependence structure that can remain relatively
inconspicuous in individual variables, providing a
challenging contamination setting. Among these
candidates, we favor directions with balanced
contributions from predictors and responses.
Writing each candidate direction as
$\be_j=(\be_{x,j}^T,\be_{y,j}^T)^T$,
we select the one maximizing
$\min(\|\be_{x,j}\|_2,\|\be_{y,j}\|_2)/
 \max(\|\be_{x,j}\|_2,\|\be_{y,j}\|_2)$.
In the mixed contamination scenario, the data 
contains a fraction $\eps/2$ of cellwise 
outliers and a fraction $\eps/2$ of casewise 
outliers. The clean data is labeled as
$\gamma=0$.

\subsection{Predictive performance of cellMR}
\label{sec:simreg}
We measure predictive performance on a clean 
test set $(\bx_{test,i},\by_{test,i})$ of size 
$\tn=1000$ by the mean squared error
$\mbox{MSE}=\frac{1}{\tn}
 \sum_{i=1}^{\tn} \|\by_{test,i} -
 \bhb-\bhB^T\bx_{test,i}\|^2$.
In each simulation setting, we generate 200 
datasets and report the average MSE. 
Each dataset is composed of $n=200$
observations with $p=q=10,25,50$, 
with $\eps = 0.2$\,.
The competing approaches include the classical 
RIDGE regression, the CRM estimator of 
\citet{filzmoser2020cellwise} implemented in the 
\texttt{R} package \texttt{crmReg}, the REGCELL 
method of \citet{CRlasso} implemented in the 
\texttt{R} package \texttt{regcell}, the 
multivariate S-estimator (SEST) of 
\citet{van2005multivariate} implemented in the 
\texttt{R} package \texttt{FRB}, the robust 
penalized ridge version of the adaptive elastic 
net proposed by \citet{cohenfreue2019robust}, 
denoted as PENSE and implemented in the 
\texttt{R} package \texttt{pense}, and the 
Shooting S regression method of 
\citet{bottmer2022sparse} referred to as SHOOT. 
Methods that were not explicitly designed for 
multivariate regression are carried out by 
fitting $q$ separate regression models, one for 
each response. The cellMR tuning 
parameters $\rk$ and $\lambda$ are selected as 
described in Section~\ref{sec:tuning} using 
$10$-fold cross-validation.

\begin{figure}[!ht]
\centering
\begin{tabular}{M{0.0005\textwidth}M{0.29\textwidth}M{0.29\textwidth}M{0.32\textwidth}}
   &\large \textbf{Cellwise}  & \large \textbf{Casewise} &\large{\textbf{Casewise \& Cellwise}} \\
[-4mm]

 \rotatebox{90}{\textbf{\footnotesize{$p=q=10$}}}&\includegraphics[width=.31\textwidth]{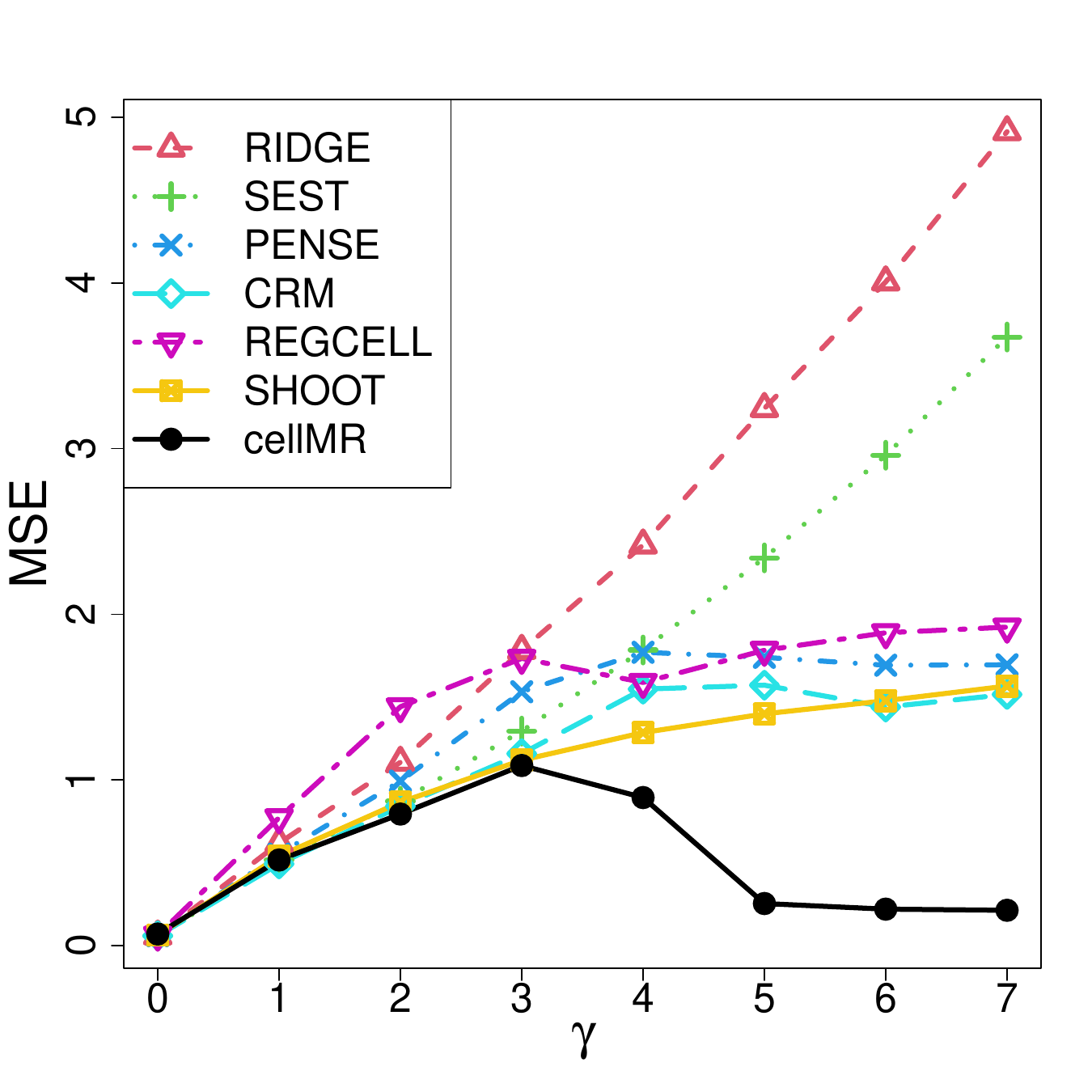} &\includegraphics[width=.31\textwidth]{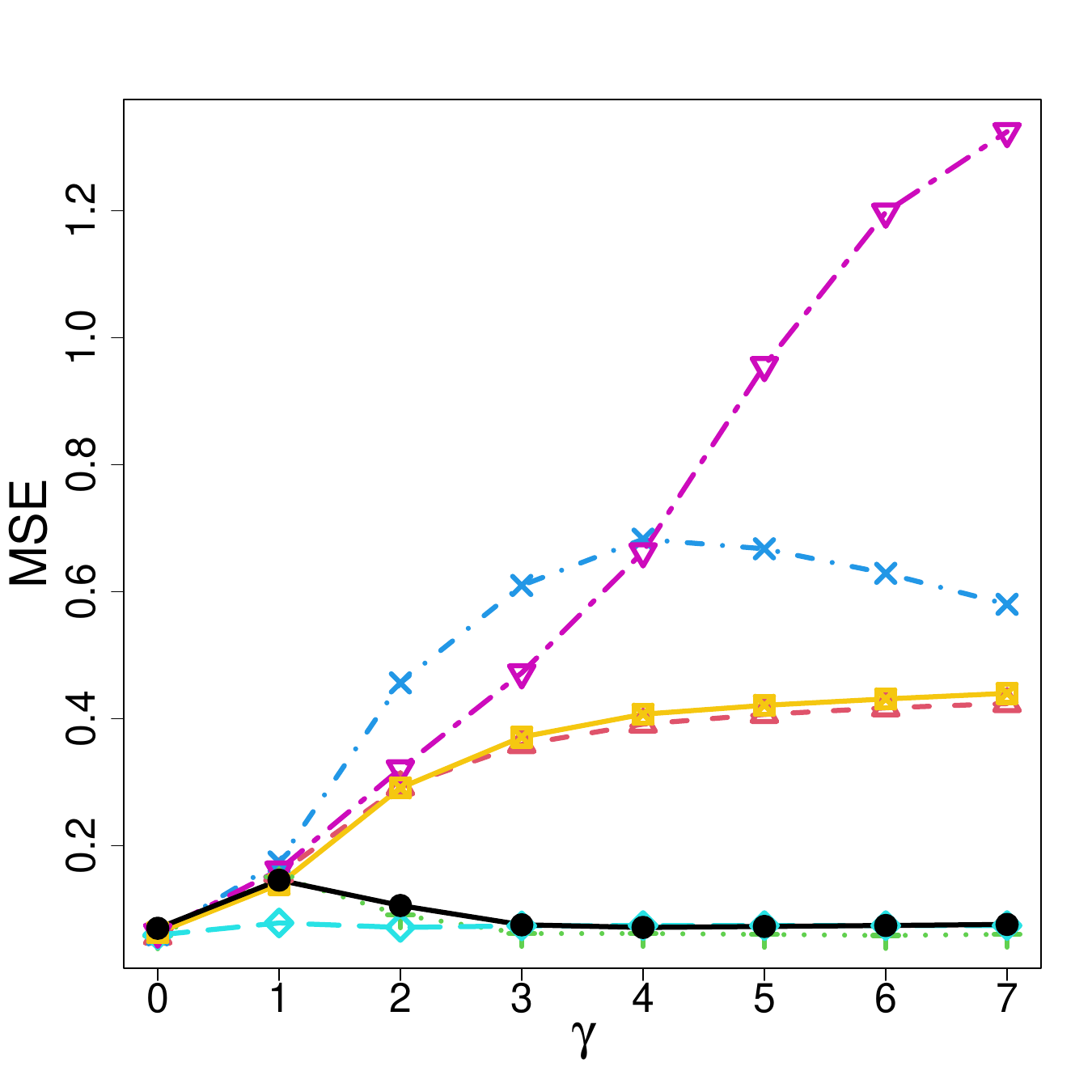} &\includegraphics[width=.31\textwidth]{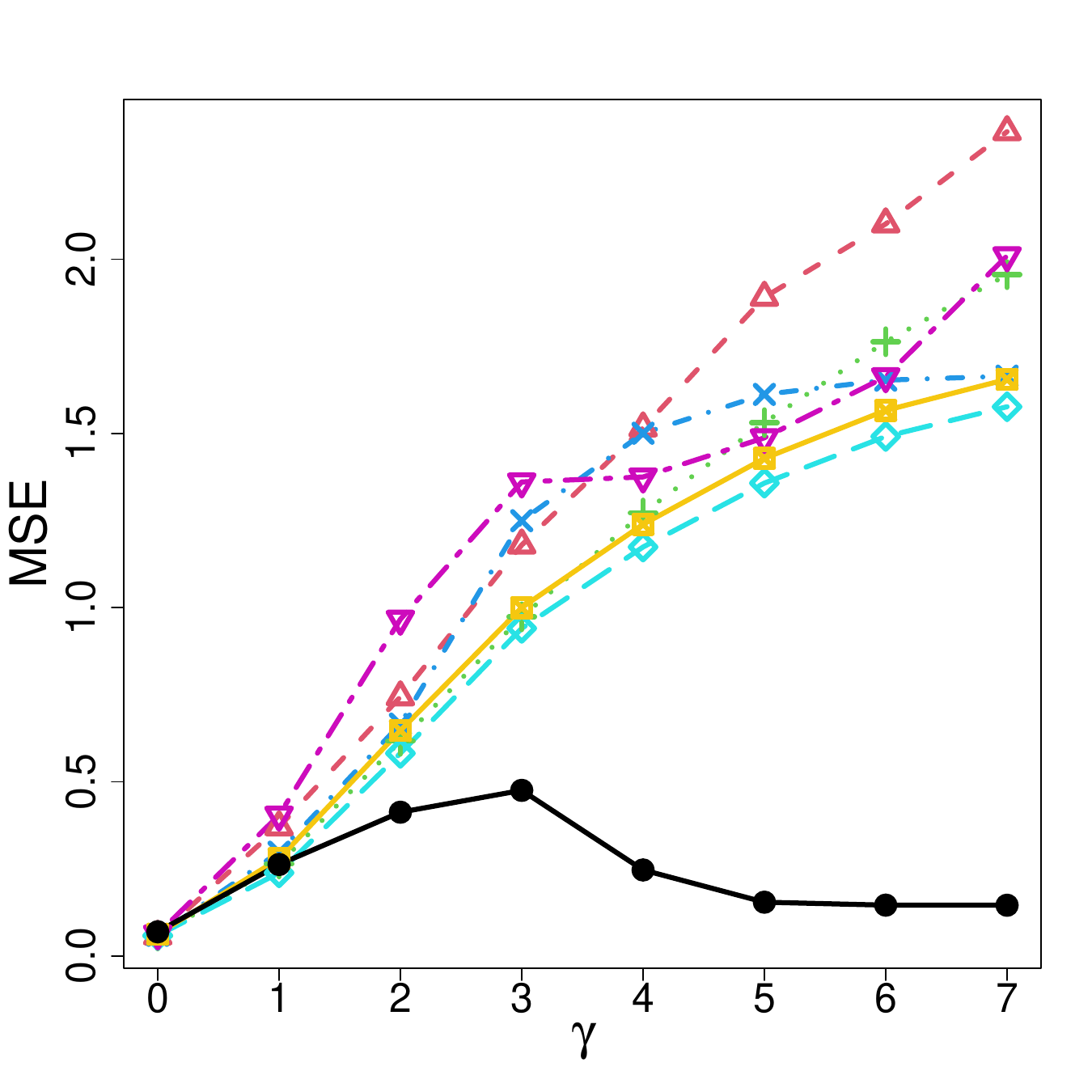} \\

 \rotatebox{90}{\textbf{\footnotesize{$p=q=25$}}}&\includegraphics[width=.31\textwidth]{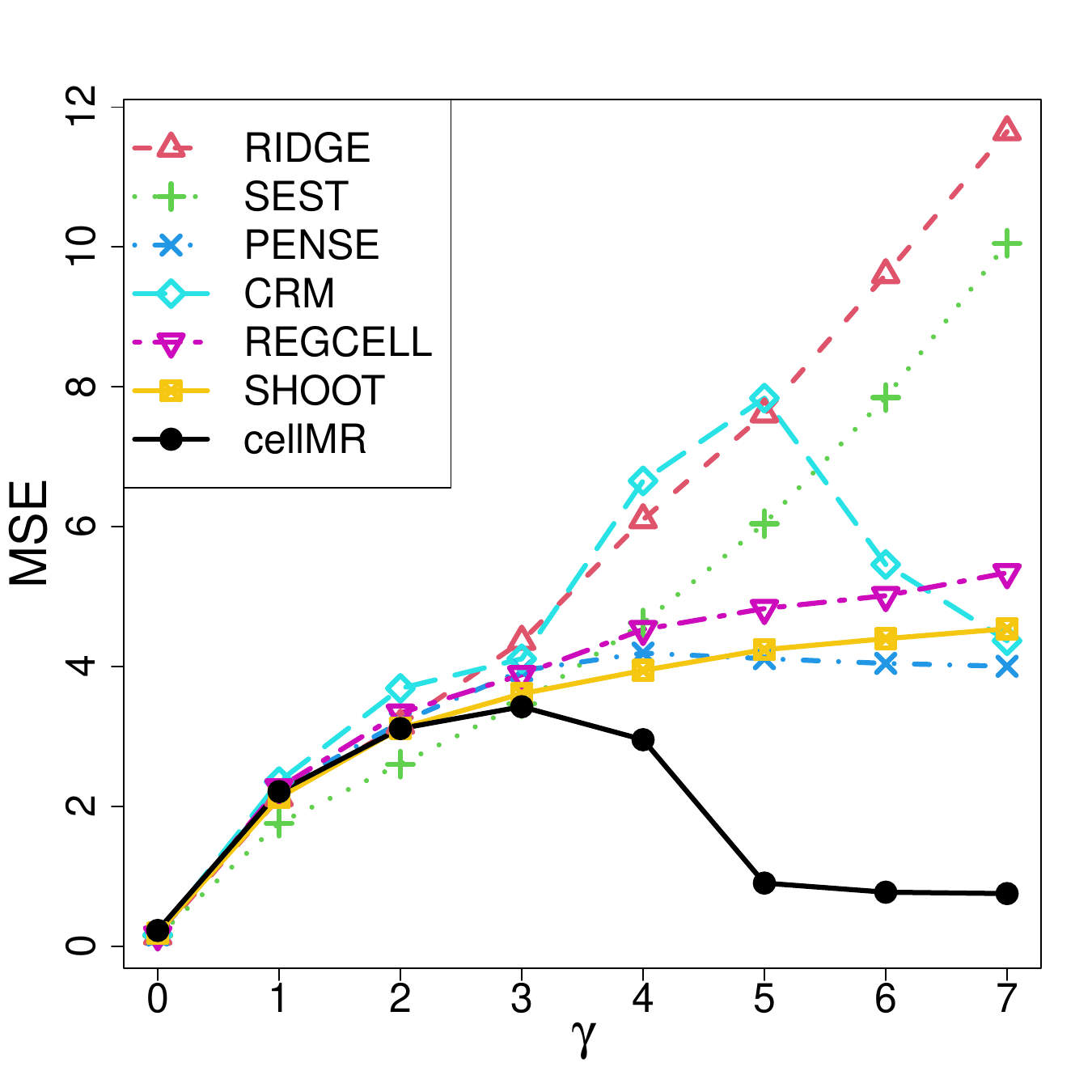} &\includegraphics[width=.31\textwidth]{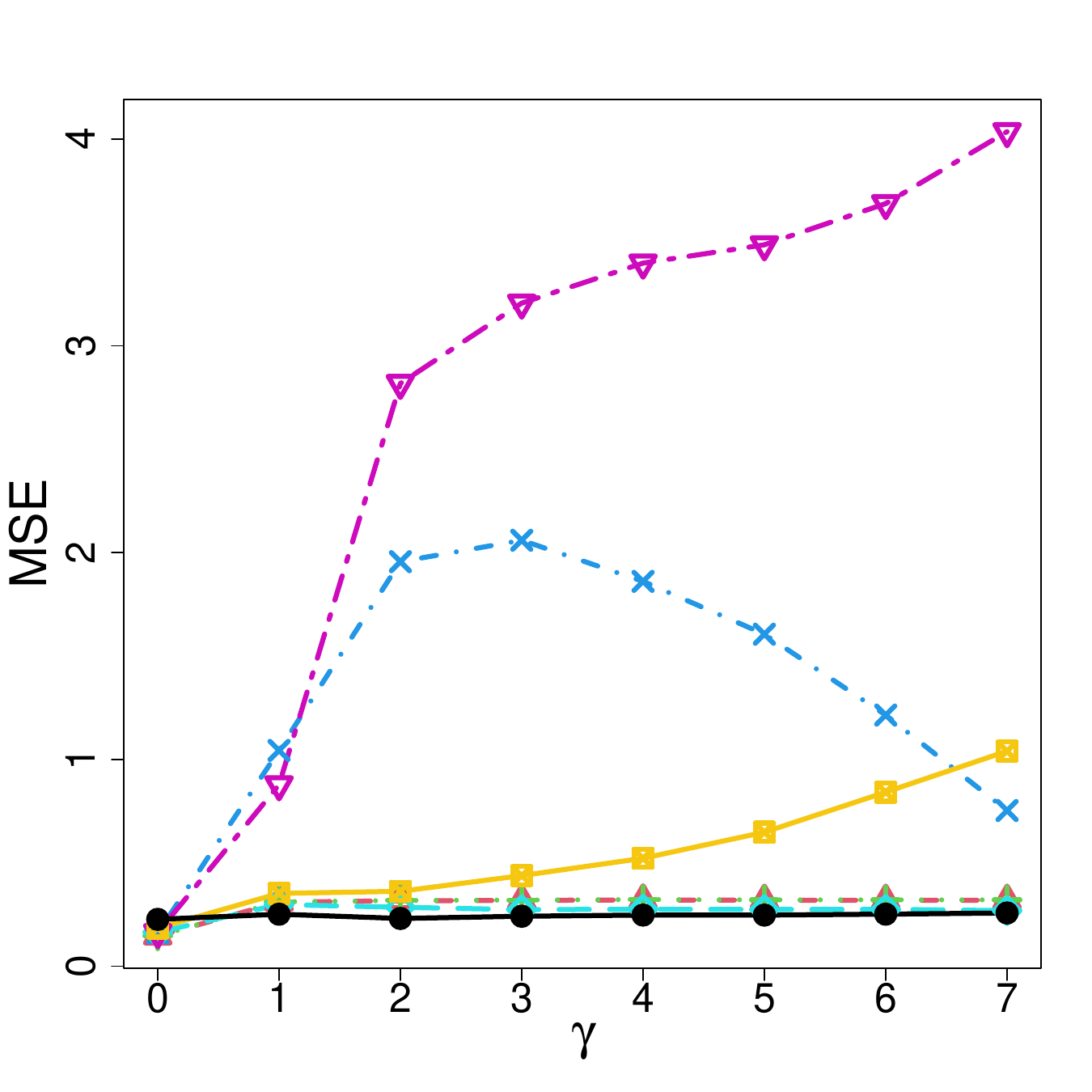} &\includegraphics[width=.31\textwidth]{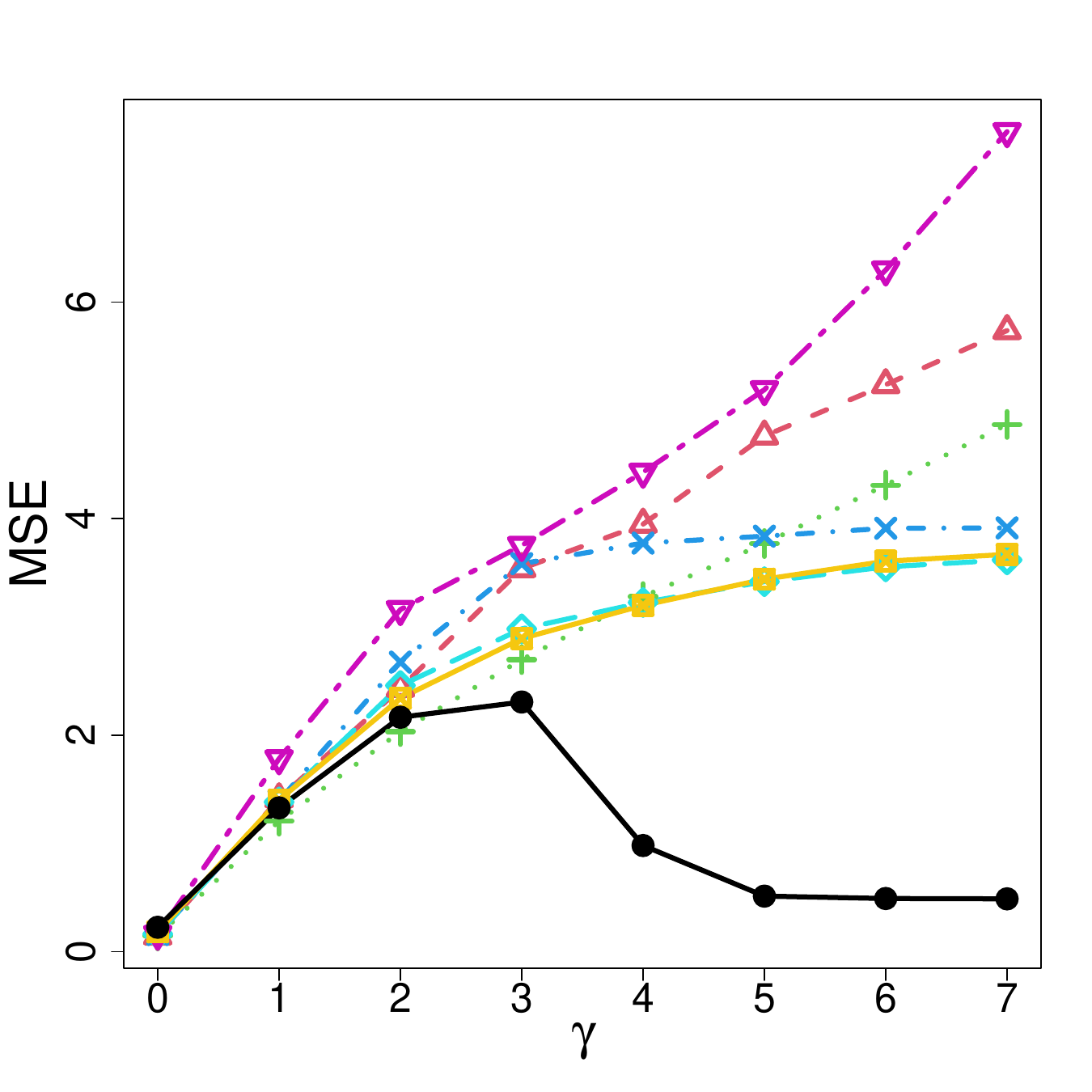}\\
 \rotatebox{90}{\textbf{\footnotesize{$p=q=50$}}}&\includegraphics[width=.31\textwidth]{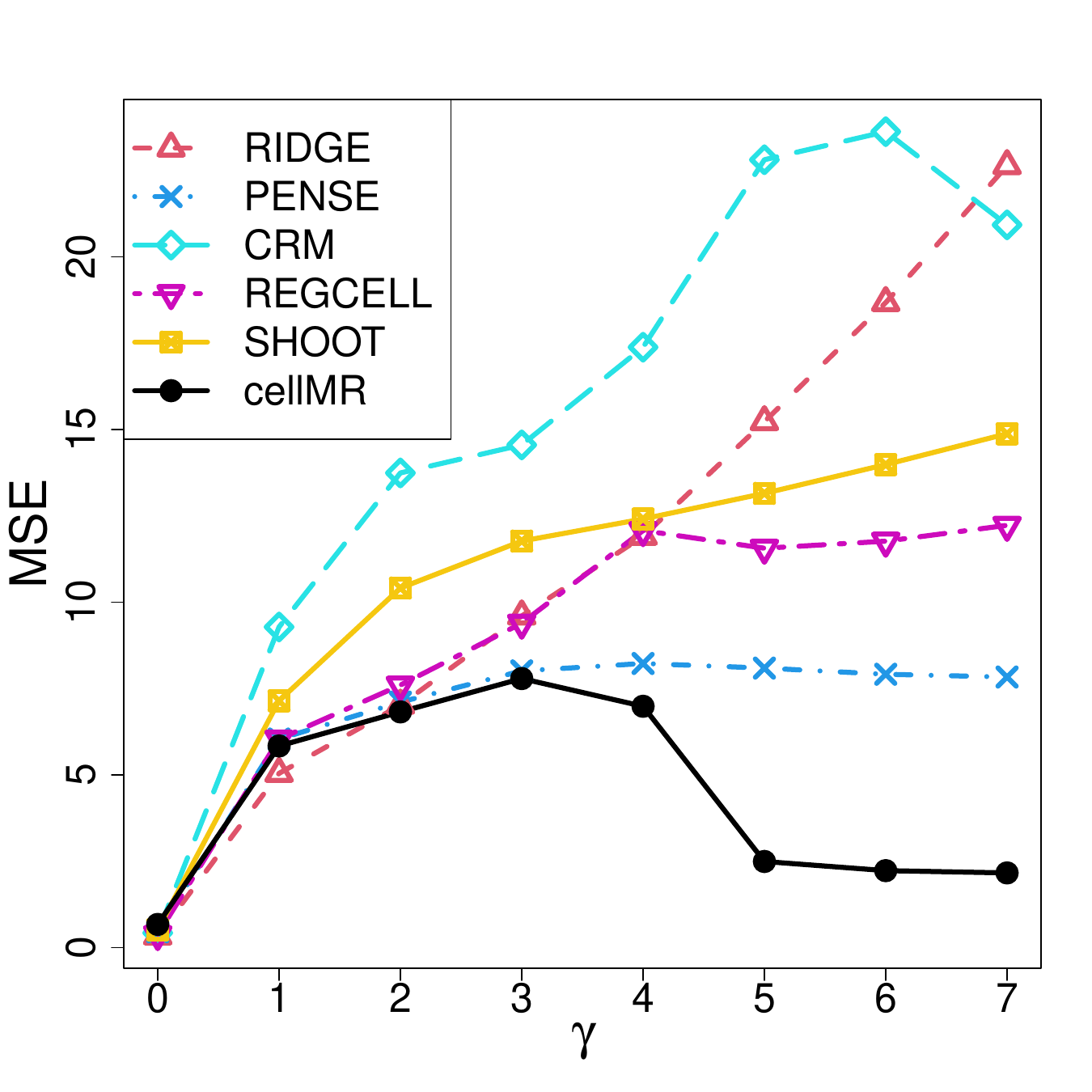} &\includegraphics[width=.31\textwidth]{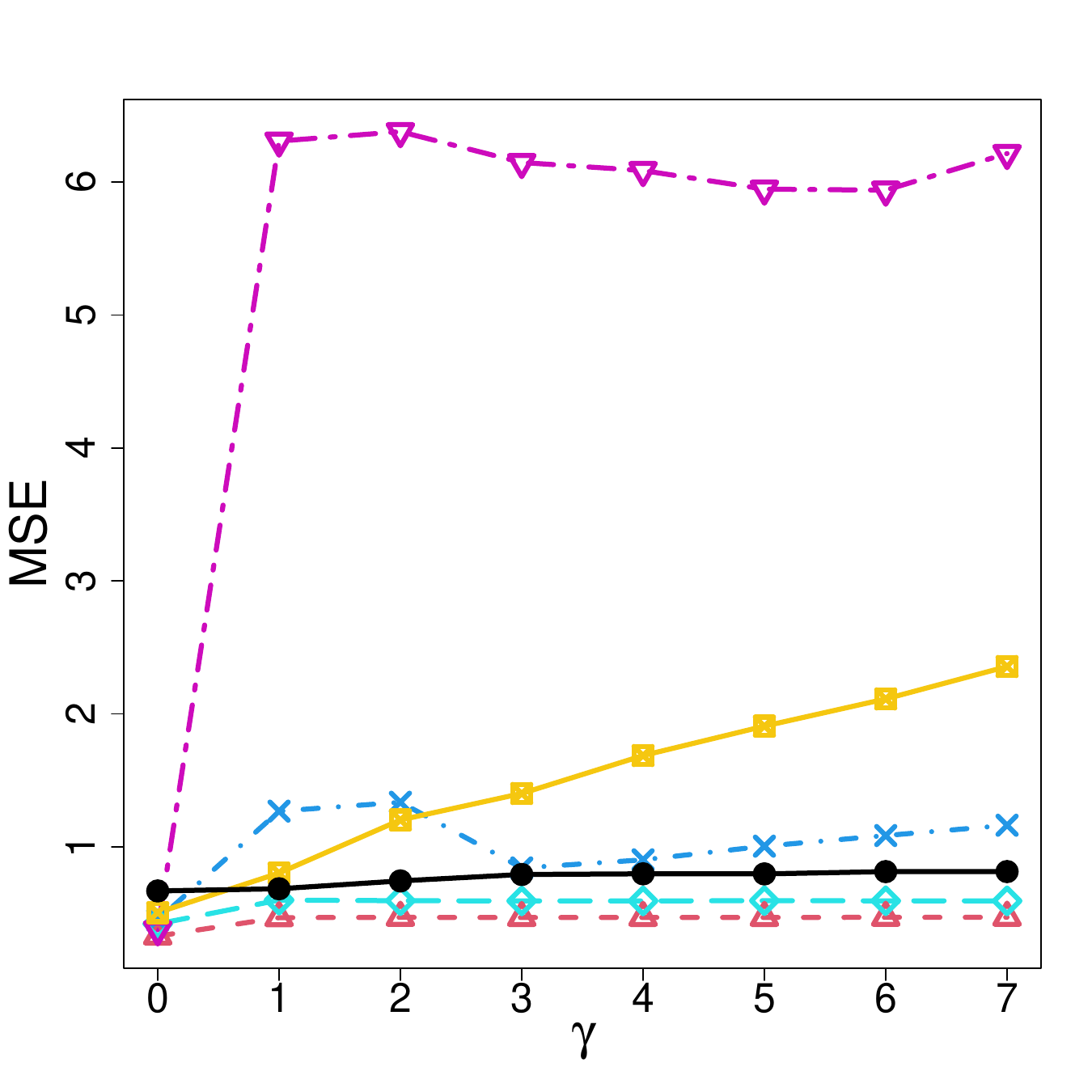} &\includegraphics[width=.31\textwidth]{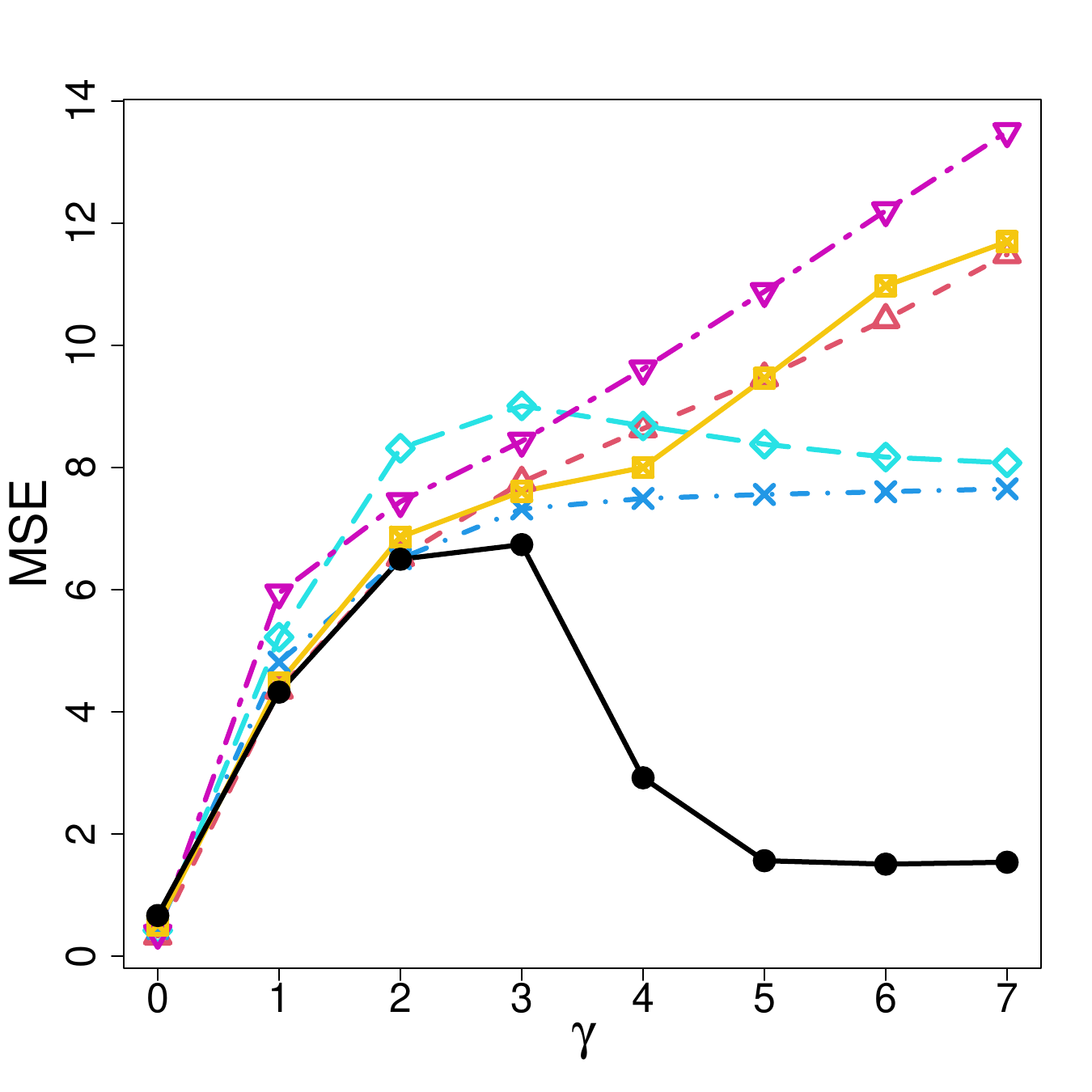}\\
\end{tabular}

\vspace{-5mm}
\caption{Average MSE attained by RIDGE, SEST, 
PENSE, CRM, REGCELL, SHOOT, and cellMR in the 
presence of cellwise outliers, casewise 
outliers, or both, without missing data.}
\label{fig:results_NA1_perout2_reg}
\end{figure}

Figure~\ref{fig:results_NA1_perout2_reg} shows 
the results, without missing data.
In the top row we see the results for the 
low-dimensional setting with $p = q = 10$. 
The nonrobust RIDGE estimator gets a high MSE 
in all three contamination settings, where 
also REGCELL and SHOOT perform poorly. SEST, 
PENSE, and CRM perform well under casewise
contamination, but not when cellwise outliers 
are present. In contrast, cellMR performs well 
in all three scenarios. The middle and bottom 
rows of Figure~\ref{fig:results_NA1_perout2_reg} 
displays qualitatively similar results for 
$p = q = 25$ and $p = q = 50$. In this last 
setting, SEST did not yield results. Overall, 
we conclude that 
cellMR is the only method that achieves 
satisfactory performance across all settings.
Very similar results are obtained when setting
10\% of randomly selected cells in both the 
predictors and responses to NA, as seen in 
Figure~\ref{fig:results_NA2_perout2_reg} in
Supplementary Material~\ref{app:addsim}.
Supplementary Material~\ref{app:addsim} 
also reports additional results across varying 
contamination fractions $\eps$ and an ablation
study assessing the contributions of the ridge 
regularization and the robust cross-validation.

\subsection{Inference performance of cellBoot}
\label{sec:siminf}

To evaluate the inference performance of cellBoot 
we compute the empirical coverage probability as 
the proportion of Monte Carlo replications where 
the true parameter value lies inside its 
confidence interval. In each setting we generate 
200 datasets of $n=400$ observations with 
$p=q=10,25,50$ and $\eps = 0.1$\,, and 
compute the confidence 
intervals with nominal level of 90\% of all 
entries of $\bB$. Ideally, the empirical
coverage should be close to 90\%.

We compare cellBoot with the classical bootstrap 
percentile method applied to the OLS estimator, 
and with the Fast and Robust Bootstrap (FRB) 
method of \citet{van2005multivariate}. 
We also
include the percentile intervals obtained by 
bootstrapping the cellMR estimate itself,
labeled Boot-cellMR, to assess the benefit of 
the II bias correction in cellBoot.

\begin{figure}[!ht]
\centering
\begin{tabular}{M{0.0005\textwidth}M{0.29\textwidth}M{0.29\textwidth}M{0.32\textwidth}}
   &\large \textbf{Cellwise}  & \large \textbf{Casewise} &\large{\textbf{Casewise \& Cellwise}} \\
[-4mm]

 \rotatebox{90}{\textbf{\footnotesize{$p=q=10$}}}&\includegraphics[width=.31\textwidth]{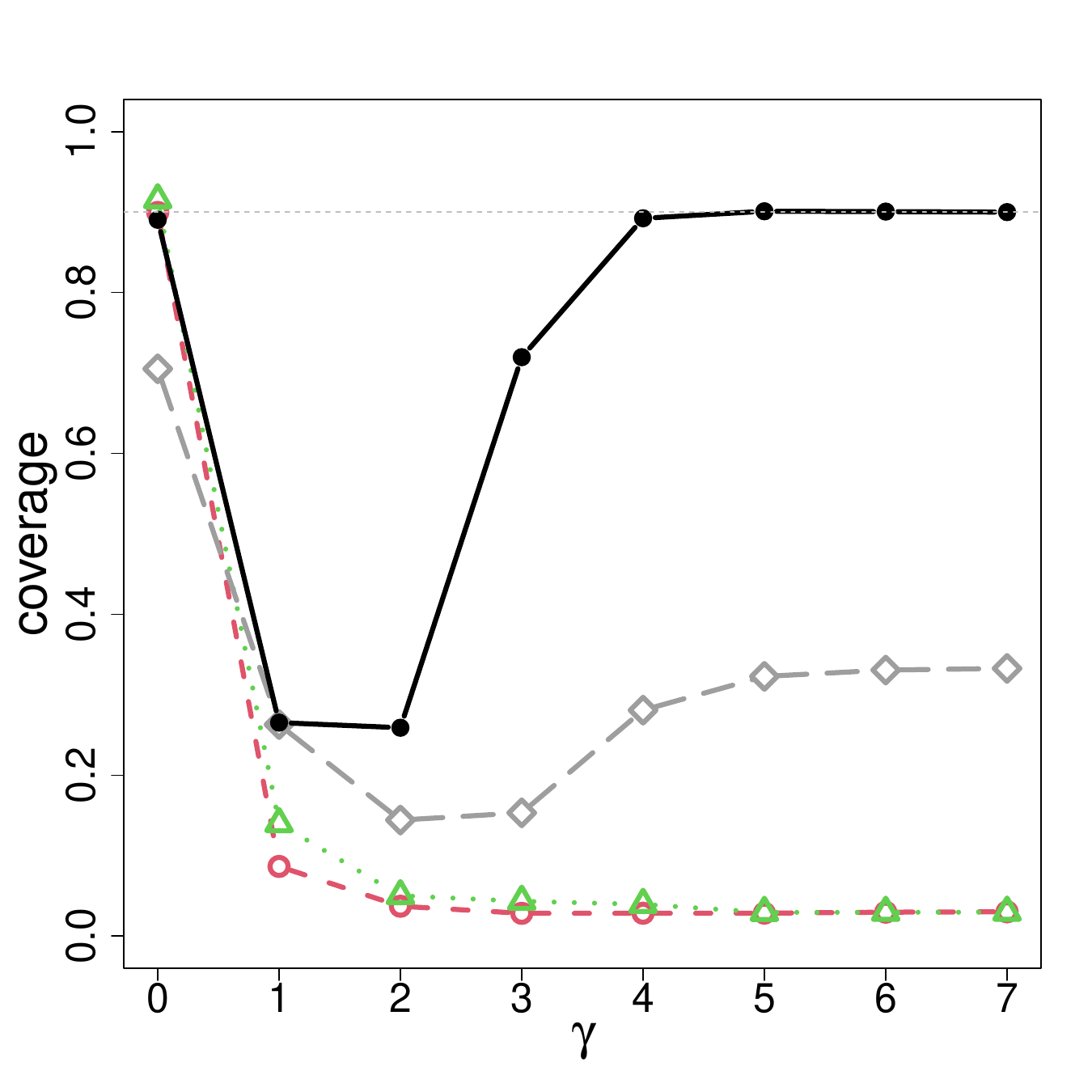} &\includegraphics[width=.31\textwidth]{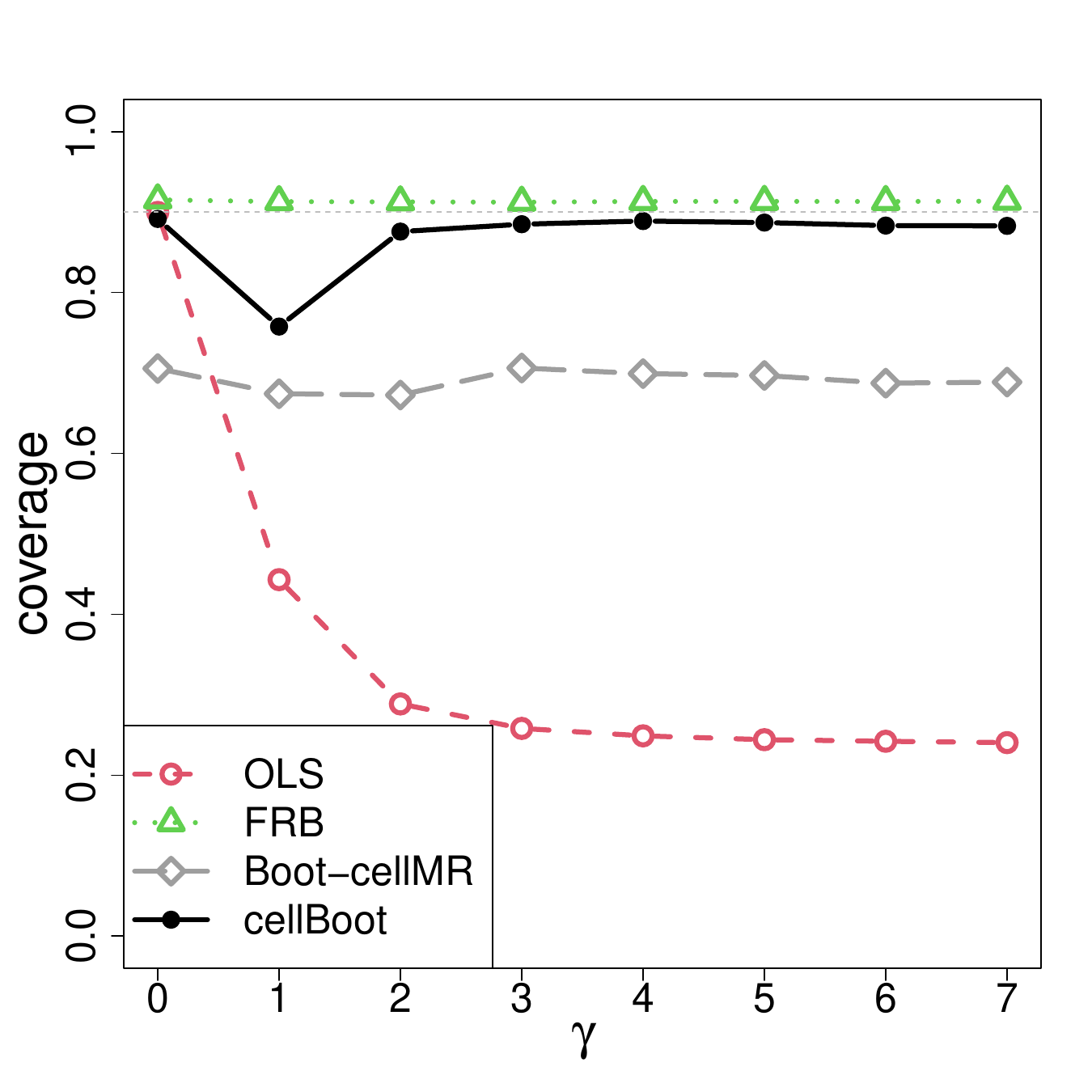} &\includegraphics[width=.31\textwidth]{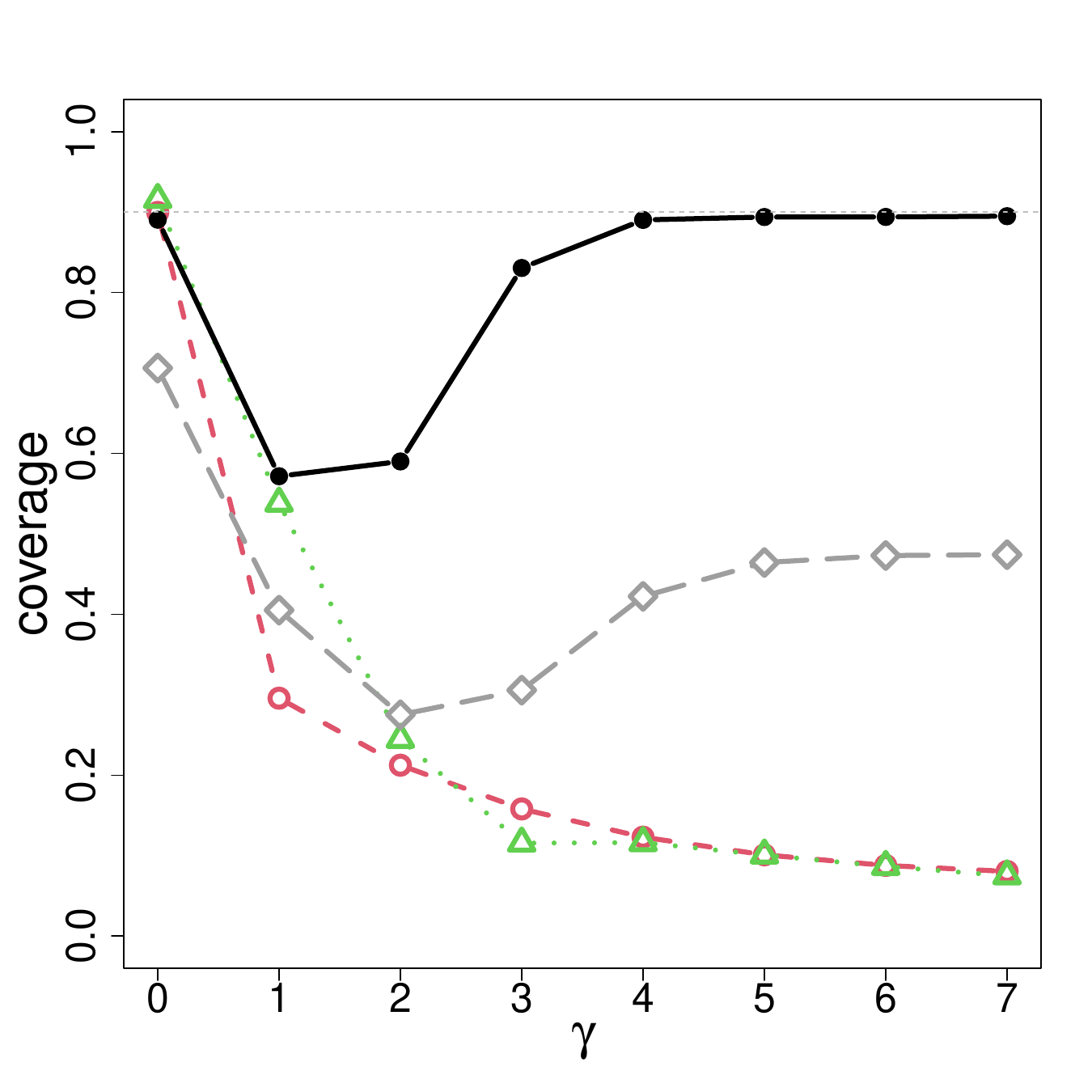} \\

 \rotatebox{90}{\textbf{\footnotesize{$p=q=25$}}}&\includegraphics[width=.31\textwidth]{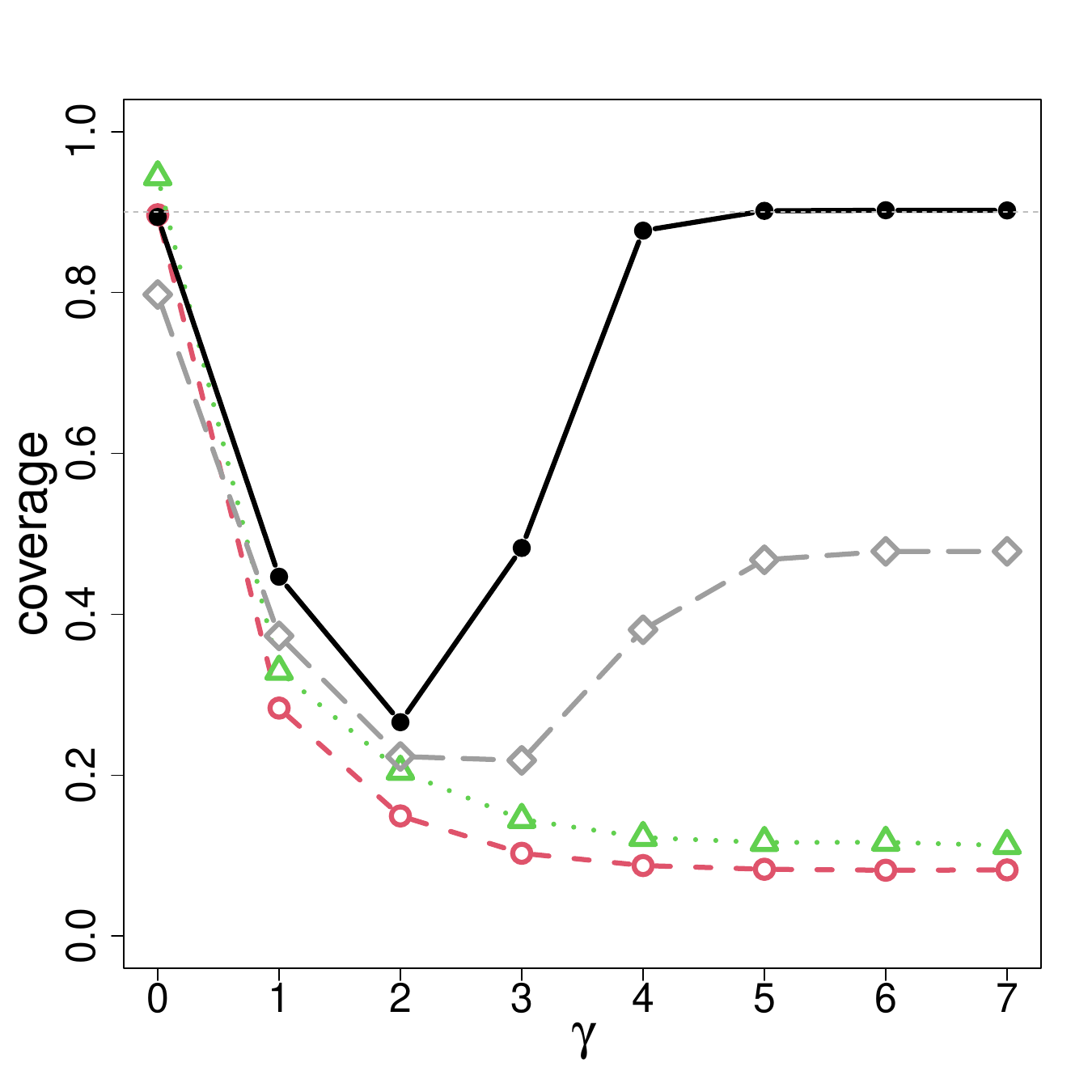} &\includegraphics[width=.31\textwidth]{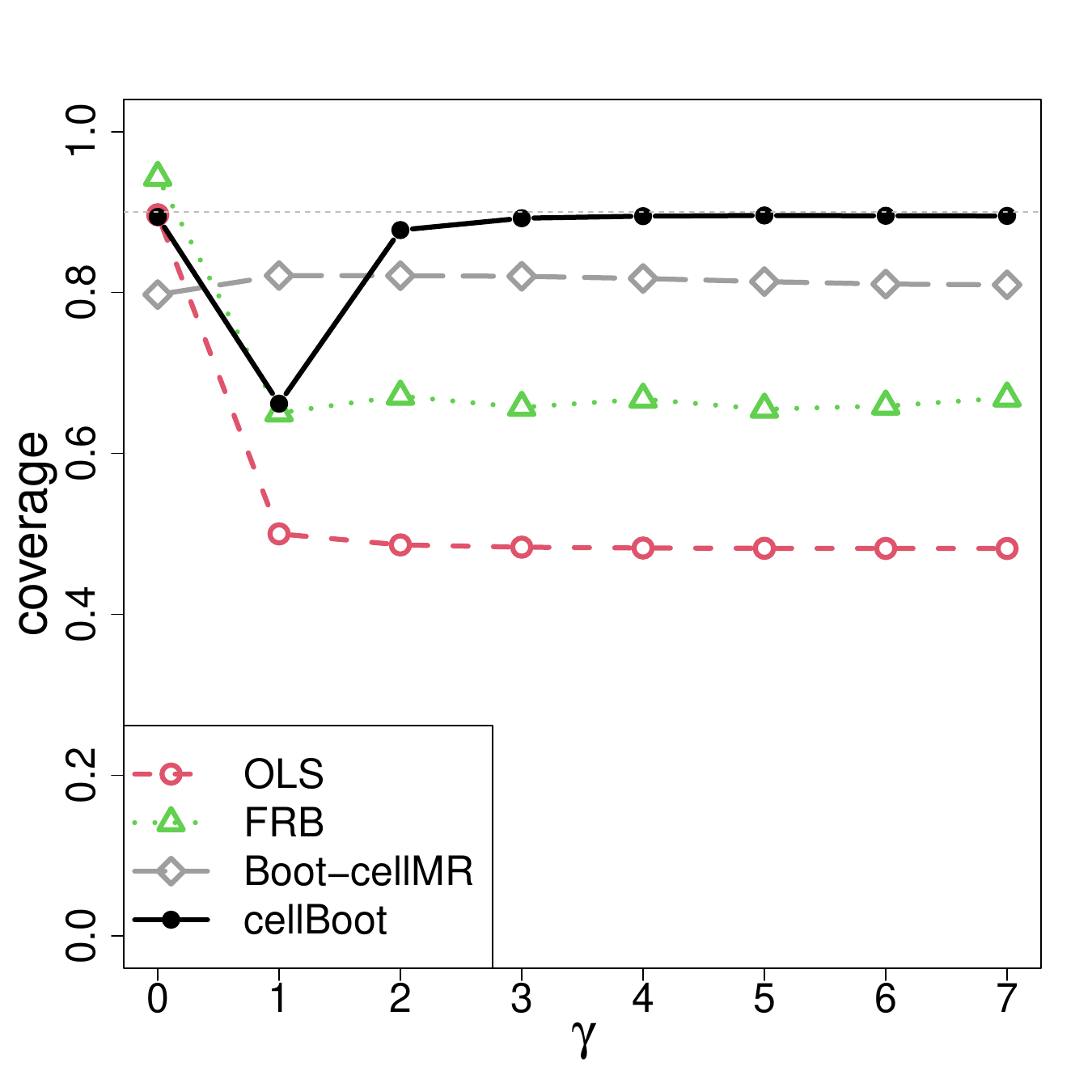} &\includegraphics[width=.31\textwidth]{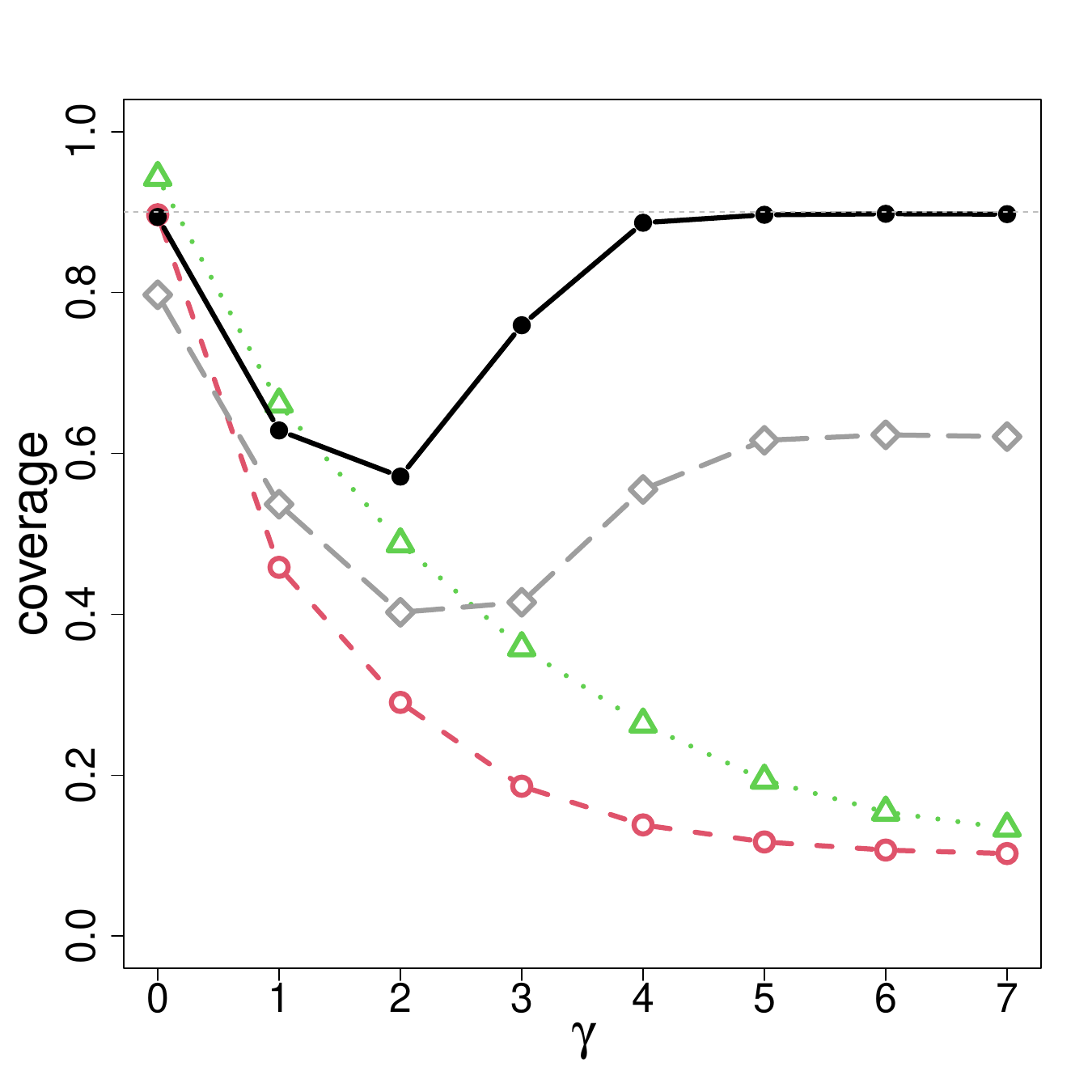}\\
 \rotatebox{90}{\textbf{\footnotesize{$p=q=50$}}}&\includegraphics[width=.31\textwidth]{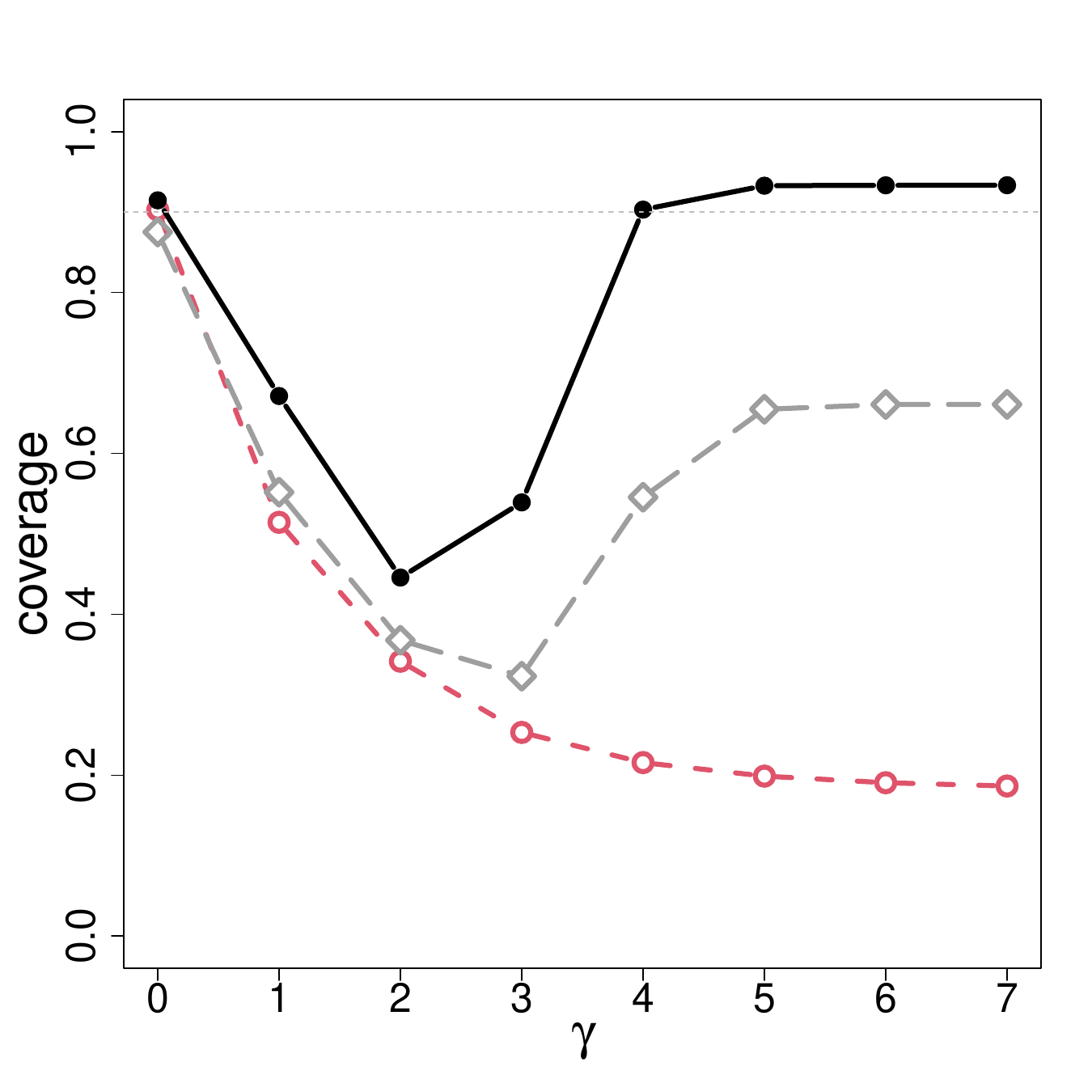} &\includegraphics[width=.31\textwidth]{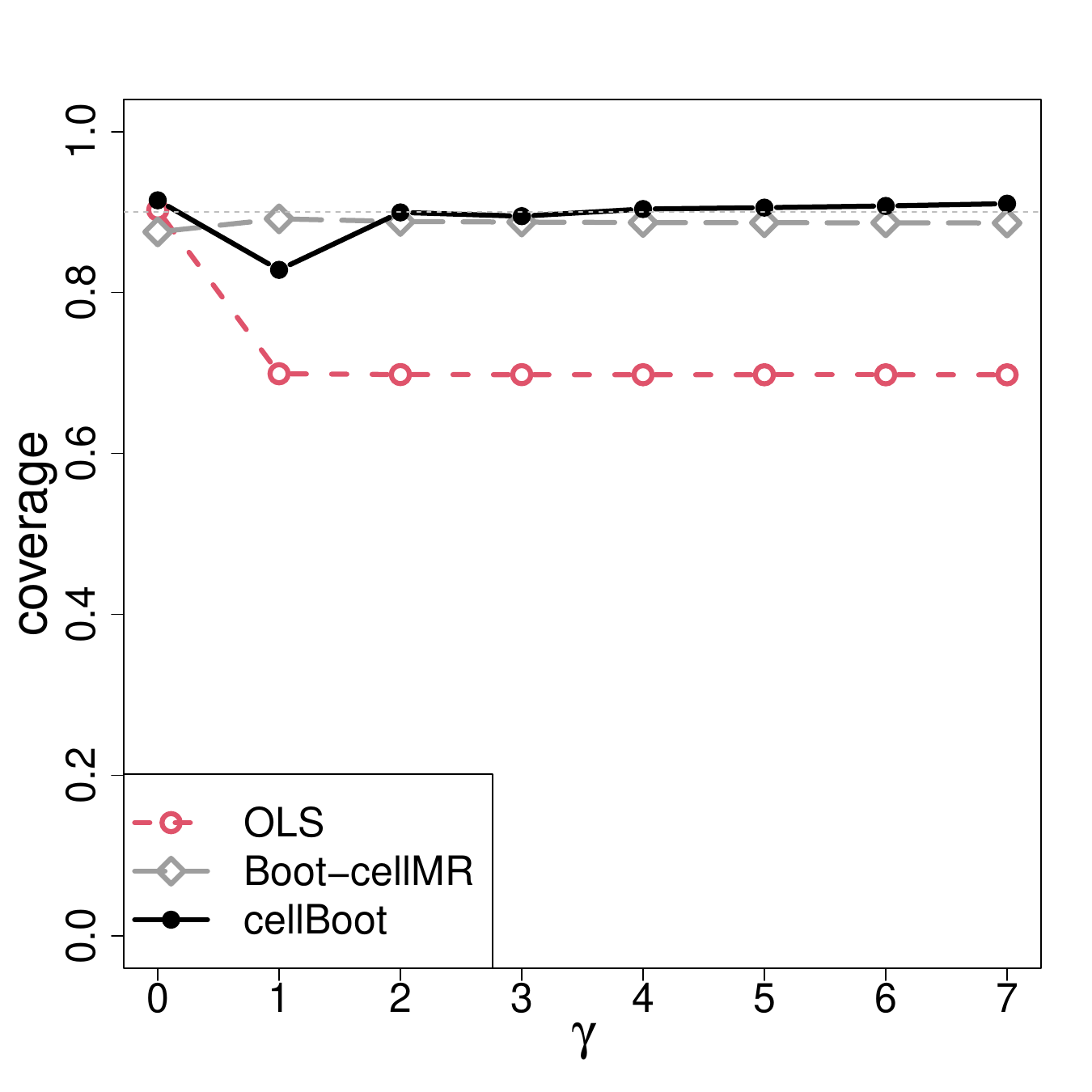} &\includegraphics[width=.31\textwidth]{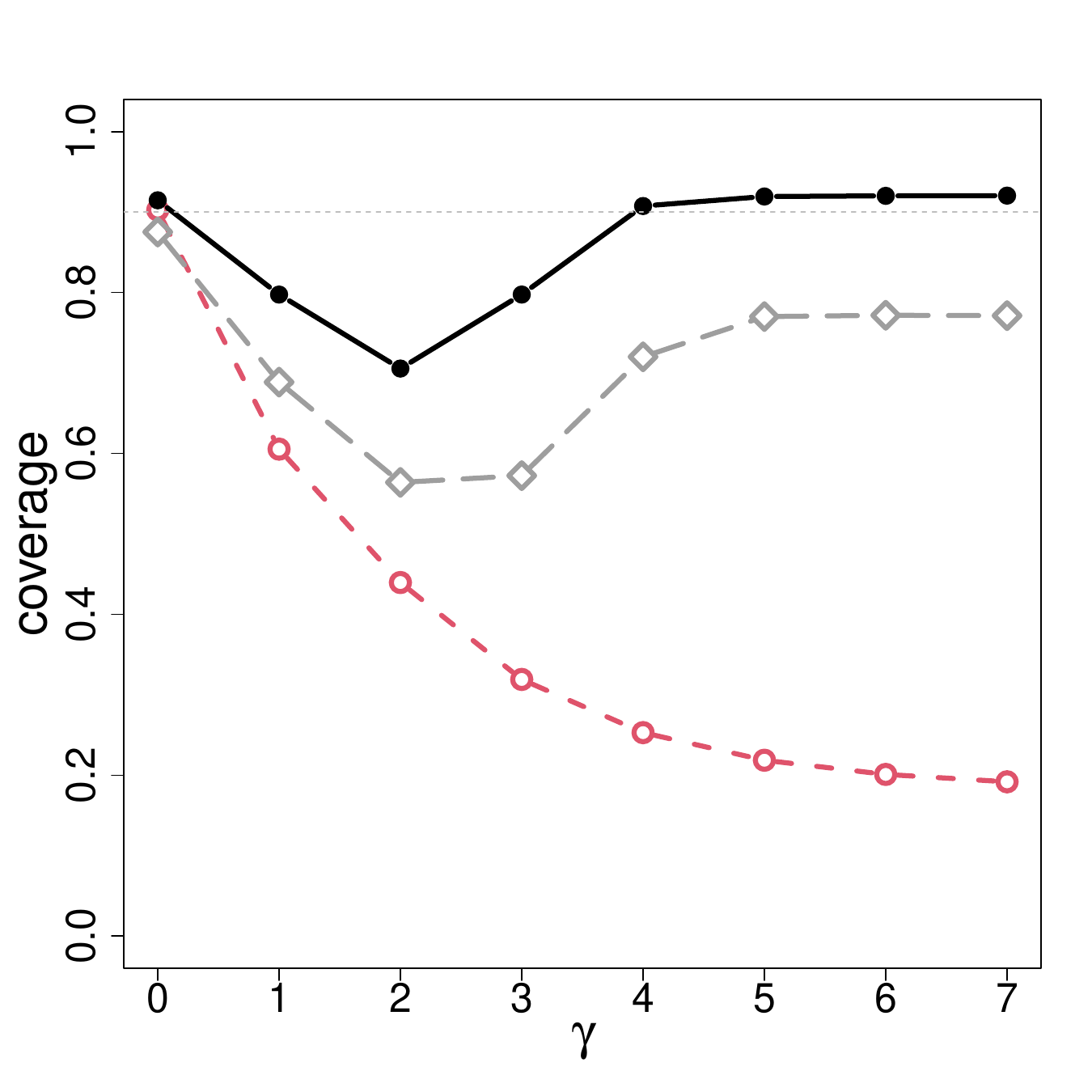}\\
\end{tabular}
\vspace{-5mm}
\caption{Average coverage attained by OLS, FRB,
Boot-cellMR, and cellBoot for the $0.9$-level
confidence intervals of the regression coefficients,
without missing data.}
\label{fig:results_NA1_coverage}
\end{figure}

Figure~\ref{fig:results_NA1_coverage} 
shows the resulting coverages. In the 
uncontaminated setting ($\gamma=0$), cellBoot 
attains the nominal $0.9$ coverage, 
corroborating the theoretical results 
in Section~\ref{sec:inf}. As expected, OLS and 
FRB attain the desired level as well. But under 
contamination, OLS exhibits severe undercoverage 
in all contaminated  settings. FRB remains 
reliable in the low-dimensional setting 
($p=q=10$) under casewise contamination, but
not in the other settings, and crashed at 
$p=q=50$. In contrast, cellBoot maintains its 
coverage close to nominal for large $\gamma$, 
in all settings. This is because far away cells 
or cases receive a small weight.
The Boot-cellMR results show that directly
bootstrapping cellMR can produce substantial
undercoverage even without contamination.
This is particularly apparent for $p=q=10$,
where coverage at $\gamma=0$ is approximately
70\% instead of 90\%. Under cellwise and mixed 
contamination, substantial undercoverage persists 
even for large $\gamma$, where cellBoot attains
coverage close to nominal. This provides strong 
empirical support for the II bias 
correction in cellBoot.

We also repeated this simulation when 
setting 10\% of the cells to NA in the 
predictors and the responses. 
The results are very similar, and shown in 
Supplementary Material~\ref{app:addsim}.
There also additional results across 
different contamination fractions $\eps$ are
reported.

\section{Real data example}
\label{sec:realdata}
To illustrate the proposed regression and 
inference procedure we analyze the genomic 
dataset introduced in Section~\ref{sec:outdetect}. 
We model the expression levels of the three 
proteins MLH1, PRKCI, and RELA using the selected 
genetic predictors. In the interest of space, the
biological motivation and the comparison with
earlier analysis of these data are relegated to
Section~\ref{app_realdata} of the Supplementary
Material.

We compare the performance of cellMR with the 
competing methods RIDGE, REGCELL, and PENSE, 
presented in Section~\ref{sec:simreg}. In this 
setting where $n \approx p$, SEST and CRM 
did not work and SHOOT gave a bad fit. To evaluate 
prediction performance while accounting for potential 
outliers, we sort the squared regression residuals 
$r_{ij}^2 = (y_{ij} - \hy_{ij})^2$ to
$r_{(1)j}^2 \leqslant \cdots \leqslant r_{(n)j}^2$
and define the robust trimmed Root Mean Squared 
Error (trimRMSE) as
\begin{equation*}
  \mathrm{trimRMSE}_{\alpha} = \left(
  \frac{1}{q H_{\alpha}} \sum_{j=1}^q
  \sum_{i=1}^{H_{\alpha}} r_{(i)j}^2
  \right)^{1/2},
\end{equation*}
where $H_{\alpha} = \lceil \alpha n \rceil$ and 
$0 < \alpha < 1$ denotes the trimming level. 
To avoid an optimistic bias in the performance
assessment, the residuals 
$r_{ij} = y_{ij} - \hy_{ij}$ are computed by 
10-fold CV. 

\begin{figure}[!ht]
\centering
\includegraphics[width=0.5\linewidth]
  {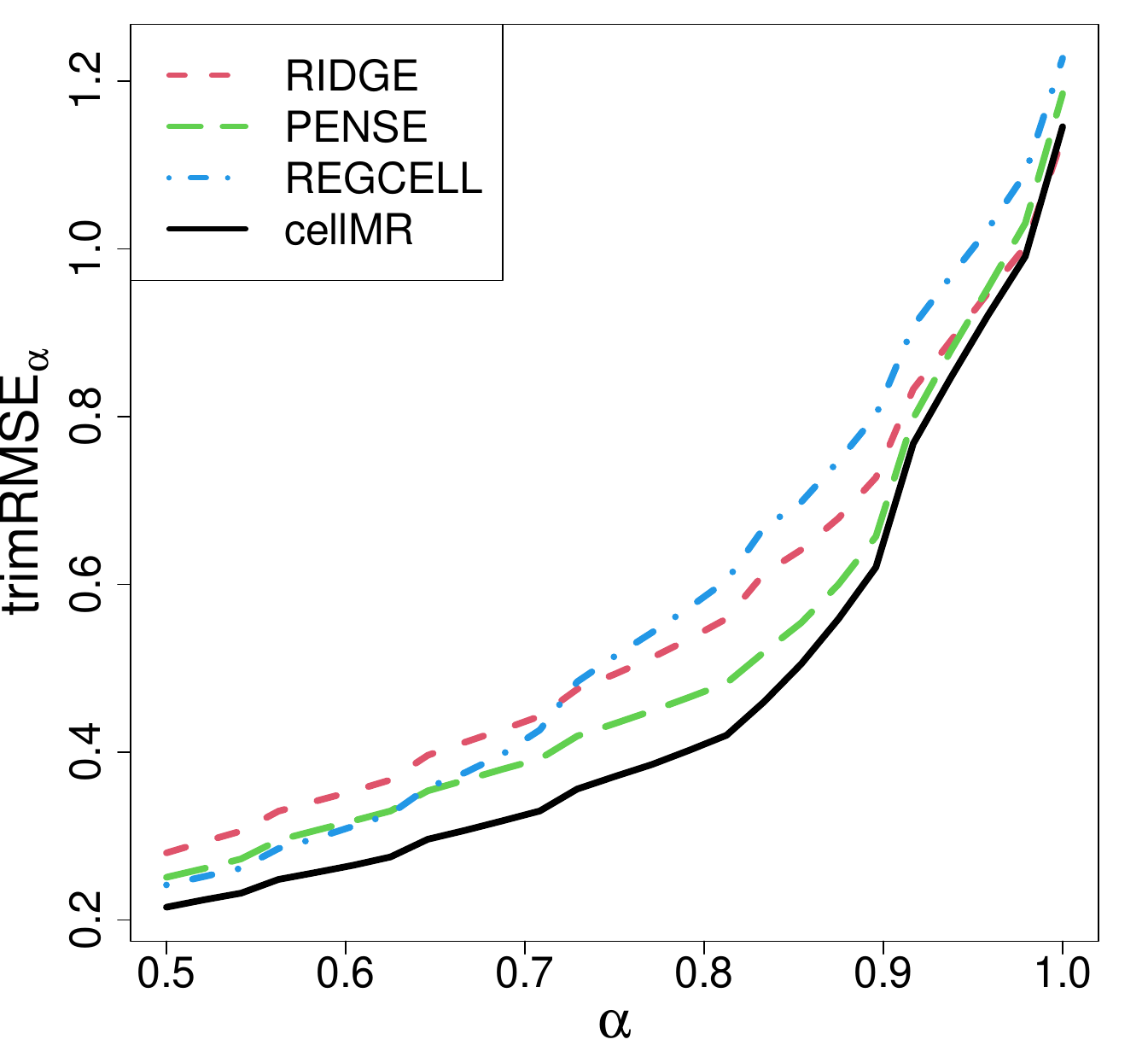}
\vspace{-3mm}
\caption{$\mathrm{trimRMSE}_{\alpha}$ of 
  RIDGE, PENSE, REGCELL, and cellMR in function 
  of $\alpha\in [0.5,1]$.}
\label{fig:trimRMSE}
\end{figure}

Figure~\ref{fig:trimRMSE} shows the 
$\mathrm{trimRMSE}_{\alpha}$ attained by RIDGE, 
PENSE, REGCELL, and cellMR as a function of 
$\alpha\in [0.5,1]$. SHOOT is not shown as it 
had very poor predictive performance here.
CellMR achieves the best predictive performance 
for $\alpha \leqslant 0.85$ where trimming 
excluded the largest residuals. 
For $\alpha$ over 85\% the curve
shoots up, indicating about 15\% of poorly
fitted responses that might contain outliers.

To better understand these results, 
we look at the regression outlier map in 
Figure~\ref{fig:outmap} in 
Section~\ref{sec:outdetect}. Case 51 has
a large residual. There are several 
leverage points, with predictor distances 
exceeding the vertical cutoff. 
Case 20 is a good leverage point, and 5
is a bad one. The cases with several 
flagged outlying cells are shown 
as points with larger sizes. The casewise 
outliers are plotted as dark grey and black 
points.

Figure~\ref{fig:cellmap} displays the 
predictor and residual cellmaps of cases 5, 
20, 38, and 51. Case 5 is flagged as a 
casewise outlier, as indicated by its dark 
circles. Case 20 is a good leverage point, 
with typical residuals but several outlying 
cells in the predictors. Case 38 seems clean. 
We also see that case 51 has outlying negative 
residuals in the variables MLH1 and PRKCI.

To further investigate  the relationship between 
genes and proteins, Figure~\ref{fig:forest_plot} 
displays a forest plot. For each protein, it shows 
the cellMR regression coefficients and their 95\%
cellBoot confidence intervals for a selected set
of genes. The plot facilitates comparisons between
gene coefficients and their uncertainties. 
We see that the effects of all selected genes on 
PRKCI are deemed nonzero, indicating a stable 
association between these genes and the protein 
PRKCI. In contrast, for predicting MLH1  
fewer genes have confidence intervals excluding 
zero, suggesting a weaker association (based on 
this dataset with low $n$). The predictive 
strength for RELA is similar. Moreover, the 
relatively narrow bootstrap intervals of several 
coefficients indicate that cellBoot provides stable 
inference despite the presence of contamination.

\begin{figure}[ht]
\centering
\vspace{2mm}
\includegraphics[width=1\linewidth]
  {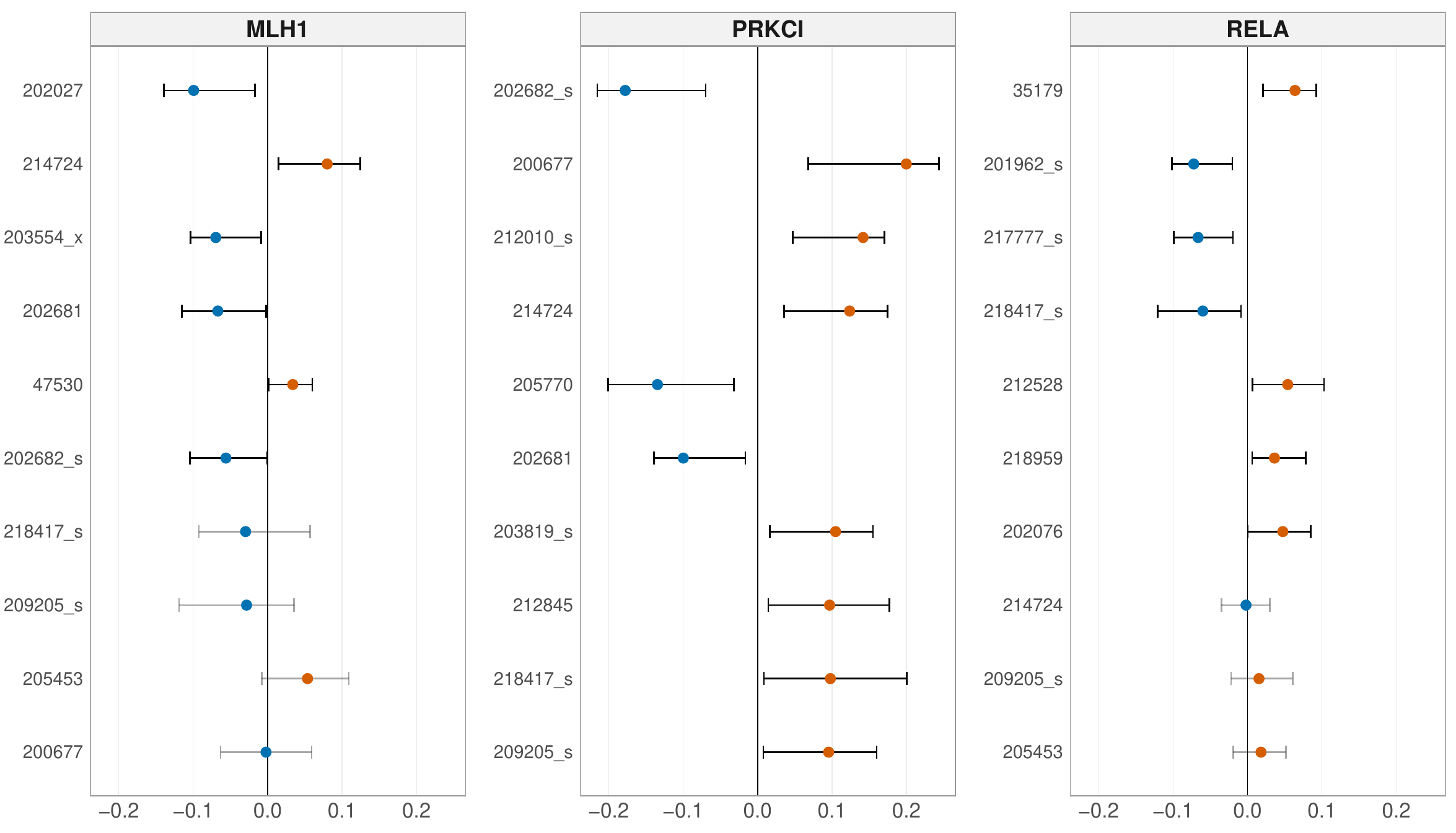} 
\caption{cellMR forest plot with level 0.95 bootstrap 
  confidence intervals for the gene-protein data.}
\label{fig:forest_plot}
\end{figure}

\section{Conclusions}
\label{sec:conc}

We have introduced the cellwise multivariate regression 
(cellMR) estimator, a novel robust regression method 
capable of simultaneously handling cellwise outliers, 
casewise outliers, missing data, and high dimensions. 
The method builds upon a recent robust covariance 
estimator and integrates it within a multivariate linear 
regression structure with ridge regularization. To the 
best of our knowledge, this is the first multivariate
regression approach that addresses all these 
challenges, and it does so with a unified methodology.
We also constructed visualizations of both cellwise 
and casewise outliers, facilitating 
anomaly detection and interpretation.

We complemented cellMR with cellBoot, a new 
bootstrap-based inference procedure. It leverages 
indirect inference to construct a consistent estimator 
of the sampling distribution. This procedure 
provides valid confidence intervals that remain stable 
under the possibly simultaneous presence of cellwise
and casewise contamination and missing values. We think
that cellBoot is the first inference framework 
specifically designed for this setting.

We established several theoretical properties of the 
proposed methodology. We derived the influence 
functions of the cellMR regression estimators 
as well as finite-sample breakdown bounds, 
and we were able to prove the asymptotic consistency 
of cellBoot. We also derived the influence functions of 
the center and the length of the resulting confidence 
intervals.

The excellent finite-sample performance of cellMR and 
cellBoot was verified through extensive
simulation, confirming their robustness and stability 
across a wide range of contamination scenarios and 
dimensions. A real data application from genomics 
illustrated the practical utility of the proposed 
approach. 

The proposed methodology is developed under a
multivariate linear regression model, and the 
theoretical results rely on this assumption. 
Nevertheless, when $\E(\bY\mid\bX)$ is not exactly
linear, a linear fit can still be used as a best
linear approximation to the data 
\citep{white1980using}.
A more detailed discussion of this issue, and a
simulation with an additional quadratic term, can
be found in Supplementary
Material~\ref{sec:misspecification}. 

Future research directions include extensions of the cellMR framework to generalized linear models and adaptations to structured high-dimensional settings such as graphical or functional regression models. Moreover, the cellBoot principle could serve as a general template for robust inference in the presence of cellwise outliers.

\vspace{6mm}
\noindent{\bf Software availability.} A zip file with
the R code, an example script, and the dataset is at
\url{https://wis.kuleuven.be/statdatascience/code/cellmr_r_code.zip}

\setlength{\bibsep}{5pt plus 0.2ex}
{\small
\spacingset{1}

}

\clearpage
\pagenumbering{arabic}
\appendix
\begin{center}
\phantom{abc}\\ 

\Large{Supplementary Material to: 
  Cellwise and Casewise Robust Multivariate Regression 
  with Inference}\\
\end{center}
\vspace{5mm}

\setcounter{equation}{0}
\renewcommand{\theequation}
  {A.\arabic{equation}}

\spacingset{1}
\startcontents[supplement]

{\small
\printcontents[supplement]{}{1}{\setcounter{tocdepth}{2}}
\par}

\clearpage
\spacingset{1.45}

\section{Illustration of masking and swamping}
\label{app:masking_swamping}

We provide an example to illustrate the 
masking and swamping effects discussed in
Section~\ref{sec:intro}, and how they are mitigated
by cellMR, the regression method presented in this article. The data are generated according to the 
linear regression model and the casewise 
contamination scheme described in 
Section~\ref{sec:sim}. We consider $n=200$ 
cases with $p=q=10$, and the first 20 cases
are generated as outliers with contamination 
magnitude $\gamma=6$.

To assess whether an observation is 
identified as outlying, for each fitted regression
method we compute the residual distance
$\RD_i = \sqrt{\br_i^T \bSigma_{\beps}^{-1}
\br_i}$ for $ i=1,\ldots,n$, where 
$\br_i = \by_i-\bhb-\bhB^T\bx_i$ and 
$\bSigma_{\beps}$ is the true error covariance 
matrix.  An observation is flagged as outlying 
when $\RD_i > c_{\uRD}$ with
$c_{\uRD}=\sqrt{\chi^2_{q,0.99}}$\,.

\begin{figure}[!ht]
\centering
\includegraphics[width=0.85\linewidth]
  {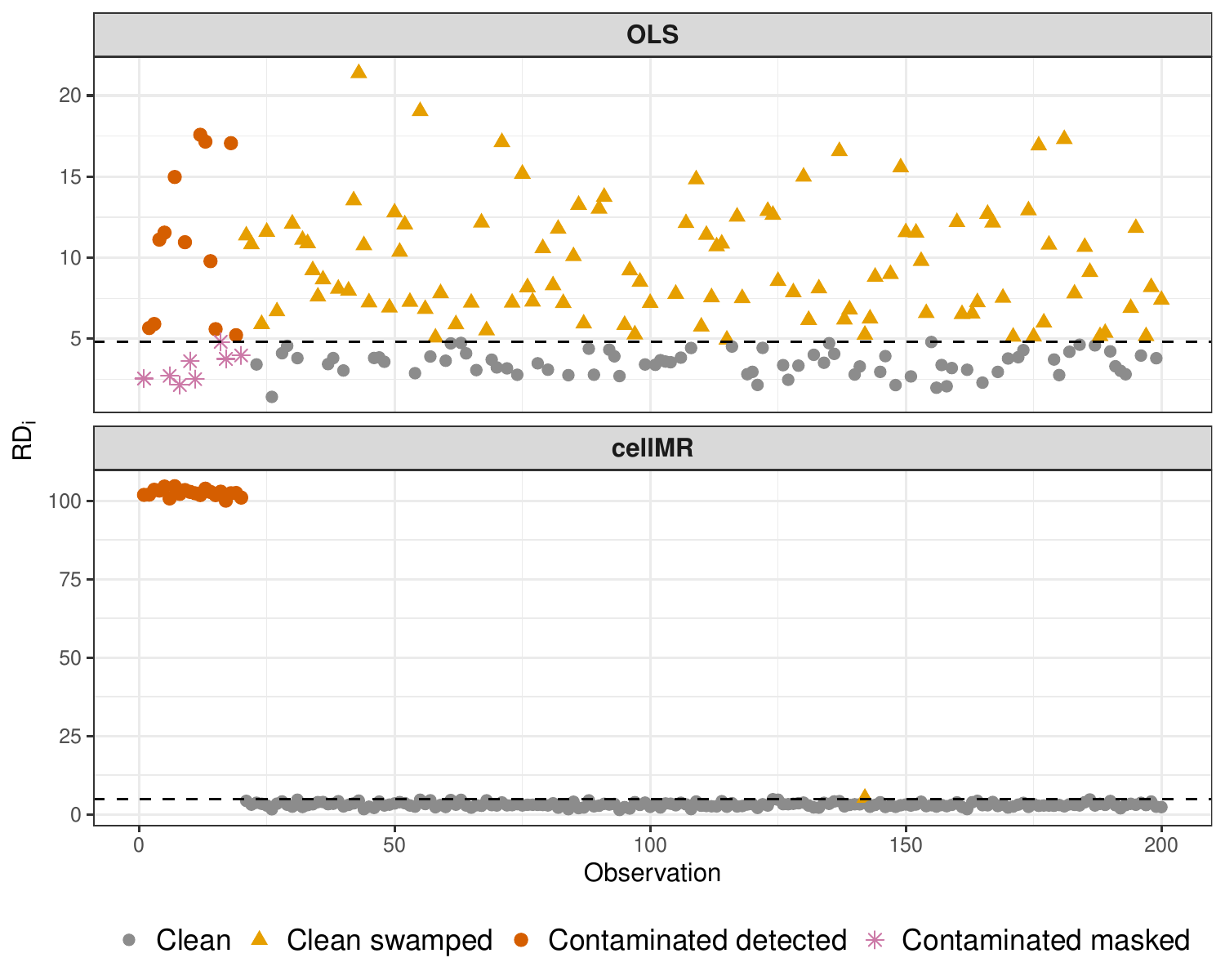}
\caption{Residual distances $\RD_i$ for OLS 
(top) and cellMR (bottom) under casewise
contamination. The horizontal dashed line denotes
the cutoff $c_{\uRD}$.}
\label{fig:masking_swamping}
\end{figure}

Figure~\ref{fig:masking_swamping} displays 
the $\RD_i$ obtained with OLS and cellMR.
For OLS, several contaminated cases have
$\RD_i \leqslant c_{\uRD}$ (purple stars) and 
are therefore not flagged as outlying. This is 
the \emph{masking}
effect where the contaminated cases distort the
classical regression fit in their direction,
thereby reducing their own residual distances.
At the same time, many uncontaminated cases have
$\RD_i>c_{\uRD}$ (orange triangles) so they are 
incorrectly flagged as outliers. That is the 
\emph{swamping} effect, where distortion of the 
fitted regression relationship makes regular
observations appear anomalous.

In contrast, neither masking nor swamping 
occurs for cellMR. As shown in the lower panel of
Figure~\ref{fig:masking_swamping}, all outliers
have residual distances above $c_{\uRD}$ and are 
therefore correctly identified as outlying, 
whereas almost all clean cases remain below the
cutoff. This illustrates how robust estimation
of the regression relationship prevents outlying 
cases from pulling the fit toward them, and 
thereby facilitates their subsequent detection.

\section{More on M-estimation}
\label{app:adddet}

This appendix provides details on the robust loss and scale
functions used in the proposed method.  We first introduce the hyperbolic tangent loss and
explain how it determines the weights assigned to cells and cases.
We then discuss the theoretical and practical reasons for this
choice. Finally, we describe the redescending M-scale used in
the procedure.

\paragraph*{Robust loss and weights}

The loss function $\rho$ determines how deviations contribute
to the fitting criterion, while its derivative $\psi$ determines
the corresponding weights. The hyperbolic tangent choice treats
small deviations quadratically, progressively reduces the weights
of larger deviations, and assigns zero weight to deviations beyond
a cutoff.
The hyperbolic 
tangent (\textit{tanh}) function $\rho_{b,c}$ 
introduced by \citetsupp{tanh1981}
is defined piecewise by
\begin{equation}\label{eq:rhotanh}
\rho_{b,c}(z) = 
\begin{cases}
 z^2/2 &\mbox{ if } 0 \leqslant |z| \leqslant b,\\
  d - (q_1/q_2) \ln(\cosh(q_2(c - |z|)))    
    &\mbox{ if } b \leqslant |z| \leqslant c,\\
  d &\mbox{ if } c \leqslant |z|,\\
\end{cases}
\end{equation}
where $d = (b^2/2) + (q_1/q_2)\ln(\cosh(q_2(c - b)))$.
Its first derivative $\psi_{b,c} = \rho'_{b,c}$ has 
been used as the \textit{wrapping function}
\citepsupp{FROC2021} and equals
\begin{equation}\label{eq:psitanh}
\psi_{b,c}(z) = 
\begin{cases}
  z &\mbox{ if } 0 \leqslant |z| \leqslant b,\\
  q_1\tanh(q_2(c - |z|))\,\mbox{sign}(z)    
  &\mbox{ if } b \leqslant |z| \leqslant c,\\
  0 &\mbox{ if } c \leqslant |z|\,.\\
\end{cases}
\end{equation}
The function $\psi_{b,c}$ is continuous, which 
implies certain constraints on $q_1$ and $q_2$. 
CellMR uses the default wrapping 
function shown in Figure~\ref{fig:rho_psi},
which has $b=1.5$ and $c=4$ with 
$q_1=1.54$ and $q_2=0.86$.
\begin{figure}[!ht]
\centering
\includegraphics[width=0.8\textwidth]
  {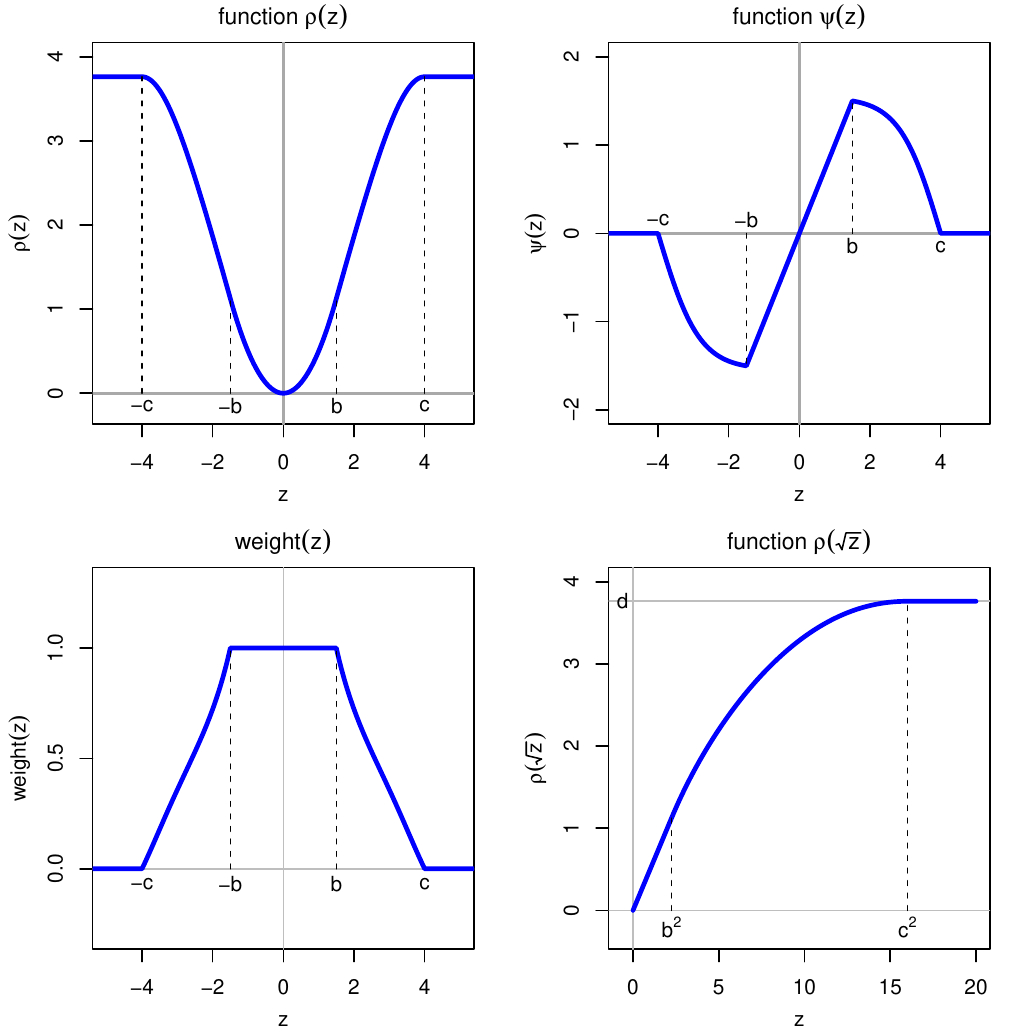}
\vspace{-4mm}
\caption{The function $\rho_{b,c}$ with 
  $b=1.5$ and $c=4$ (top left), its derivative
  $\psi_{b,c}$ (top right), its
  weight function
  used in~\eqref{eq:cellweight}
  and~\eqref{eq:caseweight} (bottom left), and 
  the function $z \mapsto \rho(\sqrt{z})$ (bottom right).}
\label{fig:rho_psi}
\end{figure}
The bottom left panel of 
Figure~\ref{fig:rho_psi}
shows the weight function $w(z)$ used
in~\eqref{eq:cellweight} 
and~\eqref{eq:caseweight}.
\begin{proposition} \label{the_validtanh}
The hyperbolic tangent function $\rho_{b,c}$ in
\eqref{eq:rhotanh} is a valid $\rho$-function.
\end{proposition}
The proof is given in \citet{cellPCA}. The 
bottom right panel of 
Figure~\ref{fig:rho_psi} illustrates the 
concavity of the function 
$z \mapsto \rho_{b,c}(\sqrt{z})$.

\paragraph*{Motivation for the 
  hyperbolic tangent loss}
The choice of the hyperbolic tangent 
$\rho$-function is further motivated by both 
theoretical and practical considerations.
First note that, since $\rho_{b,c}$ is 
constant outside of $[-c,c]$, its derivative 
$\psi_{b,c}(z)$ and the corresponding weight 
function are zero for large positive and 
negative values of $z$. Such functions are 
said to be redescending in the sense that 
very extreme values receive zero weight
in the estimation. This favorable property 
is not shared by the well-known Huber 
$\rho$-function \citep{huber1981}, 
that is not suitable in our framework.
Moreover, \citet{hampel1986} show that the 
function $\rho_{b,c}$ arises as the unique 
solution to the V-robustness problem. This 
problem seeks to maximize the estimator's
efficiency subject to an upper bound on its 
change-of-variance function, that measures 
how much the asymptotic variance changes 
under point contamination. 
The optimization is performed over all 
bounded redescending $\rho$-functions. 
Moreover, the tanh $\rho$-function 
possesses several other desirable robustness 
properties, including qualitative robustness, 
low gross-error sensitivity, and a maximal 
breakdown value. For a comparison between
redescending $\rho$-functions see Table~3 
in Section~2.6 of \citet{hampel1986}.
From a practical point of view, another 
advantage of $\psi_{b,c}$ is that it is 
linear in the central region $[-b,b]$. This 
makes the weight exactly 1 in that 
region, so inlying cells will not be 
downweighted at all. That is an advantage 
over other valid $\rho$-functions 
that could have been used, such as
Tukey's biweight $\rho$-function.

\paragraph*{Robust scale estimation}

 A scale
estimate that is inflated by extreme observations can make
anomalous deviations appear less pronounced. We therefore use
a redescending M-scale, whose estimating equation limits the
contribution of extreme observations.
An M-scale $\sigma_M(z_1,\dots,z_n)$ of a 
univariate sample $(z_1,\dots,z_n)$ is 
defined as the solution $\hat\sigma>0$ of 
the estimating equation
\begin{equation}\label{eq:Mscale_est}
\frac{1}{n}\sum_{i=1}^n
\chi\!\left(\frac{z_i}{\sigma}\right)=0.
\end{equation}
We employ a redescending M-scale whose 
associated $\chi$-function is given by
\[
\chi_{b,c}(x)=
\begin{cases}
x^2 - 1 + a, & |x|\leqslant b,\\[6pt]
\sqrt{A(k-1)}\;
\tanh\!\left(\dfrac{1}{2}\sqrt{\dfrac{(k-1)B^2}{A}}
\bigl(\log(c)-\log|x|\bigr)\right), & b \leqslant |x| \leqslant c,\\[10pt]
0, & |x|\geqslant c,
\end{cases}
\]
where the constant $a$ is chosen to ensure continuity at $|x|=b$, namely
$a
=
\sqrt{A(k-1)}\,
\tanh\!\left(
\frac{1}{2}\sqrt{\frac{(k-1)B^2}{A}}
\log\!\left(\frac{c}{b}\right)
\right)
-(b^2-1)$.
The parameters $A$, $B$, and $k$ are determined 
according to the optimal V-robust redescending 
construction described in Section~2.6 of 
\citet{hampel1986}. In the implementation of 
$\sigma_M$, we use $\chi_{b,c}$ with $b=1.5$ and 
$c=4$. The resulting redescending function is 
displayed in Figure~\ref{fig:rho_psi2}. The 
function $\chi$ is bounded and redescends to 
zero for $|x|>c$, so extreme observations 
receive zero weight. In the simulations, the 
reduced bias proved to be very helpful for
the resulting inference.

\begin{figure}[!ht]
\centering
\includegraphics[width=0.47\textwidth]
  {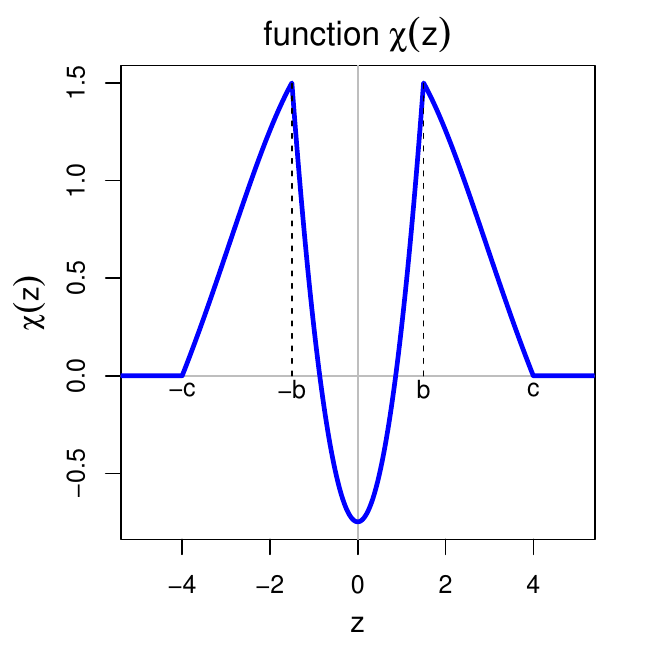}
\caption{The function $\chi_{b,c}$ with 
  $b=1.5$ and $c=4$.}
\label{fig:rho_psi2}
\end{figure}

\newpage

\section{Pseudocode for cellPCA and indirect inference}
\label{app:pseudocode}

Algorithm~\ref{alg:cellpca-irls} computes the
cellPCA fit described in Section~\ref{sec:estimator}.
Here $\bW_i$ is a $d \times d$ diagonal matrix, 
whose diagonal entries are equal to the 
$i$th row of the $n\times d$ weight matrix 
\[ \bW=\lbrace w_{i j}\rbrace= \bW^{\case}
\odot \bW^{\cell} \odot \bO\;.\]
Th $n \times d$ matrix $\bW^{\case}$ has constant
rows: each entry of row $i$ equals $w_i^{\case}$.
Below $\bW^j$ is an $n \times n$ 
diagonal matrix, whose diagonal entries are 
the $j$th column of the matrix $\bW$.

\begin{algorithm}[H]
\caption{\textbf{cellPCA algorithm, adapted from \citet{cellPCA}}}
\label{alg:cellpca-irls}
\begin{algorithmic}[1]
\STATE Initialize $\bV_{(0)}$, $\bU_{(0)}$, and
       $\bmu_{z,(0)}$ using the MacroPCA initial fit \citep{hubert2019macropca}.
\STATE Compute the robust residual scales
       $\hsigma_{1,j}$ and $\hsigma_2$ according to \eqref{eq:Mscale_est}.
\STATE Compute the initial weight matrix $\bW_{(0)}$ using \eqref{eq:cellweight} and \eqref{eq:caseweight}.
\STATE Set $t=0$.
\REPEAT
  \STATE Update each row of $\bV$ by 
  \[
  (\bv_{(t+1)})_j =
  \bigl(\bU_{(t)}^T\bW_{(t)}^j\bU_{(t)}\bigr)^\dagger
  \bU_{(t)}^T\bW_{(t)}^j
  \bigl(\bz^j-(\bmu_{z,(t)})_j\bone_n\bigr),
  \quad j=1,\ldots,d,
  \]
  where $^{\dagger}$ denotes the generalized 
inverse of a matrix, and $\bZ=
[\bz^1,\dots,\bz^d]$.
  \STATE Update each row of $\bU$ by 
  \[
  (\bu_{(t+1)})_i =
  \bigl(\bV_{(t+1)}^T\btW_{i,(t)}
        \bV_{(t+1)}\bigr)^\dagger
  \bV_{(t+1)}^T\btW_{i,(t)}
  \bigl(\bz_i-\bmu_{z,(t)}\bigr),
  \quad i=1,\ldots,n,
  \]
  where $\btW_{i,(t)}$ is a diagonal matrix whose diagonal
  is the $i$th row of $\bW_{(t)}^{\cell} \odot \bO$.
  \STATE Update $\bmu_z$  as \[
        \bmu_{z,(t+1)} = \Big(\sum_{i=1}^n \bW_{i,(t)}\Big)^{-1}
        \sum_{i=1}^n \bW_{i,(t)}\big(\bz_i - \bV_{(t+1)}
        (\bu_{(t+1)})_i\big).
        \]
  \STATE Update $\bW$  using  \eqref{eq:cellweight} and \eqref{eq:caseweight}.
  \STATE Set $t\leftarrow t+1$.
\UNTIL{the relative change in
       $\bU_{(t)}\bV_{(t)}^T$ is below a
       specified tolerance.}
\end{algorithmic}
\end{algorithm}

Algorithm~\ref{alg:ii-cellboot} summarizes the implementation of the II algorithm used in cellBoot. 

\begin{algorithm}[H]
\caption{\textbf{Indirect inference algorithm 
used in cellBoot}}
\label{alg:ii-cellboot}
\begin{algorithmic}[1]
\STATE Compute the FastCellCov auxiliary estimate
       $\hpi_n(\boheta_n)$ on the sample being corrected.
\STATE Initialize
       $\btheta^{(0)}=\hpi_n(\boheta_n)$ and fix
       $H$ simulation seeds.
\STATE Set $\ell=0$.
\REPEAT
  \FOR{$h=1,\ldots,H$}
    \STATE Generate a sample of size $n$ from
           $F_{\btheta^{(\ell)}}$ using simulation
           seed $h$.
    \STATE Compute its FastCellCov estimate
           $\hpi_h(\btheta^{(\ell)},\boheta_n,n)$,
           reusing the
           $\boheta_n$ of the initial fit.
  \ENDFOR
  \STATE Compute   \[
    \pibar(\btheta^{(\ell)},\boheta_n,n)
    =\frac{1}{H}\sum_{h=1}^{H}
      \hpi_h(\btheta^{(\ell)},\boheta_n,n).
  \]
  \STATE Update the parameter by
  \[
    \btheta^{(\ell+1)}
    =\Pi_{\bTheta}\!\left[
      \btheta^{(\ell)}
      +\hpi_n(\boheta_n)
      -\pibar(\btheta^{(\ell)},\boheta_n,n)
    \right].
  \]
  \STATE Set $\ell\leftarrow\ell+1$.
\UNTIL{$\bigl\|\hpi_n(\boheta_n)-
    \pibar(\btheta^{(\ell)},\boheta_n,n)\bigr\|^2$ is below a specified  tolerance.}
\end{algorithmic}
\end{algorithm}

\section{Connections of cellPCA with sparsity-based approaches for robust PCA}
\label{app:robust_pca_sparsity}

The cellPCA method belongs to a broad class 
of robust statistical procedures that modify the estimation criterion to limit the influence of 
anomalous parts of the data. The M-estimator
approach replaces the squared loss with a loss 
function that reduces the contribution of large 
residuals. This principle underlies many robust 
methods for multivariate statistical analysis, 
such as estimation of location and covariance, 
regression, generalized linear models, and time
series analysis \citep{maronna2019robust}.

In principal component analysis, this 
principle has been applied by 
\citetsupp{maronna2005LTSPCA} to robust subspace 
estimation under casewise contamination and by \citetsupp{maronna2008robust} to robust low-rank 
approximation under cellwise contamination.
The cellPCA method \citep{cellPCA} builds on 
these developments through a nested objective 
that combines a loss $\rho_1$ applied to 
individual cell residuals and a loss $\rho_2$
applied to aggregated casewise deviations. 
These losses limit the influence of cellwise 
and casewise outliers, allowing cellPCA to 
address both types of contamination
simultaneously.

Another approach, widely studied in the
computer science and signal processing 
communities, represents outliers explicitly
by introducing additional parameters into the 
loss function and penalizing them to encourage 
sparsity. This approach can be applied to 
multiple linear regression, where a univariate 
response $y_i$ is predicted from a
$p$-dimensional covariate vector $\bx_i$\,. 
The classical ordinary least-squares
minimization is changed to
\begin{equation}
\label{eq:sparse_shift_regression}
    \min_{\substack{b_0\in\mathbb{R},\,\bb
    \in\mathbb{R}^p\\
    \bdelta\in\mathbb{R}^n}}
    \frac{1}{2}\sum_{i=1}^{n}
    \left(y_i-b_0-\bx_i^T\bb-\delta_i\right)^2
    + \sum_{i=1}^{n} P_\lambda(\delta_i),
\end{equation}
where $P_\lambda$ is a sparsity-inducing penalty
with tuning parameter $\lambda\geqslant0$, so 
that nonzero estimates of $\delta_i$ flag
$y_i$ values requiring a shift.
\citetsupp{she2011outlier} showed a remarkable 
connection between the solution of \eqref{eq:sparse_shift_regression} and the 
classical M-estimator for regression that 
minimizes 
\begin{equation}
\label{eq:sparse_shift_regression_2}
    \min_{b_0\in \mathbb R, \bb\in \mathbb R^p}\sum_{i=1}^{n}
    \rho_{\lambda}(y_i-b_0-\bx_i^T\bb),
\end{equation}
where $\rho_\lambda$ is a robust loss function. 
Under suitable conditions linking $P_\lambda$ 
and $\rho_\lambda$, Proposition~4.1 of \citetsupp{she2011outlier} says that the stationary 
points of \eqref{eq:sparse_shift_regression} 
obtained by their algorithm satisfy the 
first-order conditions of the M-estimation 
objective
in \eqref{eq:sparse_shift_regression_2}.

In PCA, the sparse-shift approach 
represents the observed matrix $\bZ$ as the 
sum  of a low-rank component and a sparse
component \citep{candes2011robust}, that is
\[
    \bZ=\bL_0+\bS_0,
\]
where $\bL_0$ is a low rank $n\times d$ matrix 
and $\bS_0$ is an $n\times d$ sparse matrix. 
The principal component pursuit (PCP) estimator 
proposed by \citet{candes2011robust} solves
\begin{equation}
\label{eq:pcp_supplement}
    \min_{\bL,\bS}
    \|\bL\|_*+\lambda\|\bS\|_1
    \qquad\text{subject to}\qquad
    \bZ=\bL+\bS,
\end{equation}
where $\|\bL\|_*$ is the nuclear norm of $\bL$, 
i.e., the sum of its singular values, $\|\bS\|_1=\sum_{i,j}|s_{ij}|$, and $\lambda>0$.

Classical
PCA considers the model in \eqref{eq:model} and
solves
\begin{equation}
\label{eq:classical_pca_reconstruction}
    \min_{\bmu_z,\bU,\bV}
  \|\bZ-\bone_n \bmu_z^T - \bU\bV^T\|_F^2.
\end{equation}
This connects to the PCP estimator by 
considering $\bS=\bZ-\bL$ into
\eqref{eq:pcp_supplement}, which gives the
equivalent problem
\begin{equation}
\label{eq:pcp_reconstruction}
    \min_{\bL}
    \|\bZ-\bL\|_1
    +\lambda^{-1}\|\bL\|_*.
\end{equation}

The formulations \eqref{eq:classical_pca_reconstruction} and \eqref{eq:pcp_reconstruction} differ in three 
respects.
First, the classical PCA model explicitly 
decomposes the fitted matrix into a location 
component $\bone_n\bmu_z^T$ and a component 
$\bU\bV^T$ describing variation around that
location. This separation prevents the mean 
from dominating the leading principal 
directions, allowing the estimated subspace 
to capture the low-rank structure of the 
variability around the center. Without 
centering, the leading directions may 
primarily reflect the location of the data 
rather than their spread. In PCP, one could 
first subtract a center of $\bZ$ before 
applying \eqref{eq:pcp_reconstruction}.

Second, classical PCA imposes a 
prescribed rank bound through the dimensions 
of the factors. With
$\bU\in\mathbb{R}^{n\times k}$ and
$\bV\in\mathbb{R}^{d\times k}$, the centered 
reconstruction $\bU\bV^T$ has rank at most $k$, 
and the complete reconstruction 
$\bone_n\bmu_z^T+\bU\bV^T$ has rank at most
$k+1$. PCP instead encourages low rank through 
the nuclear-norm penalty, without imposing a 
prescribed rank bound. The rank of its fitted
matrix is thus determined by the data and the 
tuning parameter.

Third, classical PCA minimizes squared 
reconstruction errors, whereas PCP uses the 
entrywise $\ell_1$ loss. The latter grows 
linearly rather than quadratically with the 
magnitude of each residual, reducing the 
contribution of large individual reconstruction 
errors and providing resistance to cellwise 
outliers.

The reconstruction-based formulation of 
PCP in \eqref{eq:pcp_reconstruction} 
facilitates comparison with cellPCA. When there 
are no missing values, choosing $\rho_2(z)=z^2$ 
and $\rho_1(z)=|z|$ (which we do not recommend)
would reduce the objective in \eqref{eq:objP} to
\begin{equation}
\label{eq:objP_l1}
   \left\|
   \left(\bZ-\bone_n\bmu_z^T-\bU\bV^T\right)
   \bD \right\|_1,
\end{equation}
where
$\bD=\operatorname{diag}(\hsigma_{1,1},\ldots,\hsigma_{1,d})$.
So in this special case the cellPCA objective 
is somewhat similar to the reconstruction loss 
in PCP. But it still explicitly decomposes 
the fitted matrix into a center and a 
component describing variation around it. 
Moreover, a prescribed rank bound is imposed 
through the dimensions of $\bU$ and $\bV$, 
rather than encouraged through a 
nuclear-norm penalty. Comparing \eqref{eq:pcp_reconstruction} with
\eqref{eq:objP_l1} also helps clarify 
potential limitations of PCP under cellwise
and casewise contamination. Its entrywise 
absolute loss is unbounded and not 
redescending, so arbitrarily large residuals 
retain a nonzero contribution to the 
estimating equations. Moreover, its 
reconstruction loss lacks a separate 
mechanism for downweighting entire outlying 
cases, because it has no function $\rho_2$\,. 
Consequently, robustness to sparse cellwise 
contamination does not automatically extend 
to casewise contamination. These limitations 
were illustrated in the simulation study in 
Section~6 of \citet{cellPCA}.

The seminal contributions of 
\citet{candes2011robust} and 
\citetsupp{chandrasekaran2011rank} helped 
establish an extensive literature on low-rank 
recovery under contamination. Representative 
developments include
\citesupp{zhou2010stable,xu2010robust,
netrapalli2014arXiv,yi2016fast,chen2021bridging} 
and \cite{she2016robust} with more recent work
addressing heavy-tailed noise and robust matrix 
completion in 
\citesupp{wang2025robust} and \citesupp{vaknin2026rgnmr}.
We refer to \citesupp{vaswani2018static} for a 
comprehensive review of robust PCA, its 
dynamic extensions, and matrix completion.

\section{More on cellMR}
\label{app:cellMR}

\subsection{More on robust cross-validation}
\label{app:cv}
We first address the selection of the pair 
$(\rk,\lambda)$ by cross-validation, where 
the folds can contain outliers and NA's.
Section~\ref{sec:tuning} estimates the 
prediction error by formula~\eqref{eq:cv}:
\begin{equation*}
   \text{CV}(\rk,\lambda) = \frac{1}{q} \sum_{j=1}^{q}
   \frac{1}{K}\sum_{h=1}^K
   \mbox{WMSE}_{(h)j}\;,
\end{equation*}
where $\mbox{WMSE}_{(h)j}:=\sum_{i=1}^{n_h}
 w_{(h)ij}(y_{(h)ij}-\hy_{(h)ij})^2
 /\big(\sum_{i=1}^{n_h} w_{(h)ij}\big)$
is a weighted mean of the squared residuals in 
fold $h$. The weights in this formula are given by
 $$w_{(h)ij} = o_{(h)ij}^y 
  w^{\case,y}_{(h)i}
 (w^{\cell,y}_{(h)ij})^2\;.$$
Here $o_{(h)ij}^y$ indicates whether $y_{(h)ij}$
is observed, as missing responses are of course 
excluded. We also downweight outlying cells, by 
$$w^{\cell,y}_{(h)ij}=w^{\cell}
\left( \frac{y_{(h)ij}-\hy_{(h)ij}}{\hsigma_{1,j}^y}
\right)$$
where $w^{\cell}$ is the univariate function in the
bottom left panel of Figure~\ref{fig:rho_psi}, and
$\hsigma_{(h)1,j}^y: = 
\sigma_M(\{y_{(h)ij}-\hy_{(h)ij}\}_{i=1}^n)$.

The weight $w^{\case}_{(h)i}$ is given by 
$w^{\case,y}_{(h)i} = w^{\case}\left(d_{(h)i}/
\hsigma_{(h)2}^y\right)$ where $w^{\case}$
is a function of the same type,
\begin{equation*}
  d_{(h)i} := \frac{1}{\sum_{j=1}^{q}
  o_{(h)ij}^yw^{\cell,y}_{(h)ij}} \sum_{j=1}^{q}
  o_{(h)ij}^yw^{\cell,y}_{(h)ij}
  (y_{(h)ij}-\hy_{(h)ij})^2\,,
\end{equation*}
and $\hsigma_{(h)2}^y :=
\sigma_M(\{d_{(h)i}\}_{i=1}^n)$. 
We then select the pair ($\rk, \lambda)$  
that minimizes~\eqref{eq:cv}.\\

\subsection{Derivation of the influence functions}
\label{app:IF_cellMR}
We now turn to the proofs of the influence functions
of cellMR. We consider the contamination 
model~\eqref{eq:cont_reg_mixed} with $H_C=\Delta_{\bc}$, 
where $\Delta_{\bc}$ denotes the distribution that 
assigns unit mass to a vector $\bc=(c_1,\ldots,c_d)^T$, 
thus
\begin{equation}
\label{eq:cont_reg_mixed2}
Z_{\eps}
=
A\odot Z + (\bone_d-A)\odot \bc,
\end{equation}
where $Z\sim H_0$ and $A=A^{\case}\odot A^{\cell}$. 
We denote the distribution of $Z_\eps$ as $H_\eps$\,, 
and the distribution of $A$ as $G_\eps$.
The casewise contamination component $A^{\case}$ has Bernoulli marginals with
$\Pr(A^{\case}_j=1)=1-\eps^{\case}$ for $j=1,\ldots,d$, and its components are fully dependent, in the sense that
$
\Pr(A^{\case}_1=\cdots=A^{\case}_d)=1$.
The cellwise contamination component $A^{\cell}$
has components $A^{\cell}_j$, $j=1,\ldots,d$, that are Bernoulli random variables with success probabilities
$\Pr(A^{\cell}_j=1)=1-\eps^{\cell}_j$.

Under the \textit{fully dependent contamination
model} (FDCM) we have 
$A = A^{\case}$ with independent $Z$ and 
$A^{\case}$.  We denote $G_\eps$ as 
$G_\eps^D$ and the distribution $H_\eps$ of $Z_{\eps}$ as 
$H(G_\eps^D,\bc)$. 
In the \textit{fully independent
contamination model} (FICM), $G_\eps$ is denoted by $G_\eps^I$ and the corresponding contaminated distribution $H_\eps$ by $H(G_\eps^{I},\bc)$.
In this setting $A=A^{\cell}$, where the components $A^{\cell}_j$ are mutually independent and independent of $Z$.

Under both the dependent and independent 
contamination models,
the distribution of $A$ satisfies
$\Pr\left(A_j=1\right)=1-\eps$, $j=1, \dots, d$, 
and (ii) for any sequence 
$\left(j_1, j_2, \ldots, j_d\right)$ of zeroes 
and ones with $d-\ki$ ones and $\ki$ zeroes, 
$\Pr(A_1=j_1, \ldots, A_d=j_d)$ has the same 
value, denoted as $\delta_\ki(\eps)$. Obviously, $\eps=\eps^{\case}$ 
under the FDCM and $\eps=\eps^{\cell}$ under the FICM. 
Under FDCM we have that $\Pr(A_1=\cdots=A_d)=1$, and then 
$\delta_0(\eps)=(1-\eps),\; \delta_1(\eps)=\cdots=\delta_{d-1}(\eps)=0$, and 
$\delta_d(\eps)=\eps$.
In that situation, the 
distribution of $Z_{\eps}$ simplifies to 
$(1-\eps)H_0 + \eps\Delta_{\bc}$ where 
$\Delta_{\bc}$ is the distribution which puts 
all of its mass in the point $\bc$. The FICM instead 
assumes that $A_1, \dots, A_d$ are independent, 
hence
$$\delta_\ell(\eps)=\binom{d}{\ell}
   (1-\eps)^{d-\ell} \eps^\ell, 
  \quad \ell=0,1, \ldots, d\;. $$
  
We first present the proposition deriving the IFs of $\bmu_z(H)$, $\bV(H)$, $\bmu^{\bu}_{\text{\tiny MCD}}(H)$, $\bSigma^{\bu}_{\MCD}(H)$, and $\bSigma_{\zort}(H)$, which are the functionals corresponding to $\bhmu_z$ and $\bhV$ from~\eqref{eq:objP}, $\bhmu_{\text{\tiny MCD}}(\bhu_i)$ and $\bhSigma_{\text{\tiny MCD}}(\bhu_i)$ in \eqref{eq:SigmaMCD}, and $\btSigma_{\zort}$ in \eqref{eq:Sigmay}.
The proofs of the following propositions can be 
found in \citet{centofanti2025cellwise}.

When there are no missing values we can write 
the functional version $(\bV(H),\bmu_z(H))$ of 
the minimizer of \eqref{eq:objP} as
\begin{align} \label{eq:ux}
\hspace{-6mm}
  \left(\bV(H),\bmu_z(H)\right)&=\argmin_{\bV,\bmu_z}
  \E_H\left[\rho_2 \left(\frac{1}{\sigma_2(H)}
  \sqrt{\frac{1}{d}\sum_{j=1}^{d} \sigma_{1,j}^2(H)
  \rho_1\left(\frac{z_{j}-\mu_{z,j}-\bu^T\bv_j}
  {\sigma_{1,j}(H)}\right)}\right)\right]\nonumber\\
  \text{such that}\quad \bu&= \argmin_{\bu}
  \rho_2 
  \left(\frac{1}{\sigma_2(H)}\sqrt{\frac{1}{d}
  \sum_{j=1}^{d} \sigma_{1,j}^2(H)\rho_1 \left(
  \frac{z_{j}-\mu_{z,j}-\bu^T\bv_j}{\sigma_{1,j}(H)}
  \right)} \right)
\end{align}
where $\bu=\left(u_{1},\dots,u_{\rk}\right)^T$ 
and $\bz=\left(z_1,\dots,z_d\right)^T\sim H$.
Here $\sigma_{1,j}(H)$ and $\sigma_2(H)$ are 
estimating functionals corresponding to 
the initial scale estimates of 
$r_{j} :=z_{j}-\mu_{z,j}- \bu^T\bv_j$ and 
$\rt :=\sqrt{\sum_{j=1}^{d} \sigma_{1,j}^2(H)
\rho_1(r_{j}/\sigma_{1,j}(H))/d}$\,. 
This is subject to the first-order conditions 
in \citet{cellPCA} given by
\begin{align*}
  \E_{H}\left[ \bW\bV\bu\bu^T\right]&=  
  \E_{H}\left[ \bW(\bz-\bmu_z) \bu^T\right],\\
  \E_{H}\left[\bW\bV\bu\right] &=
   \E_{H}\left[\bW 
   \left(\bz-\bmu_z\right)\right],\\
  \left(\bV^T\bW\bV\right)\bu &=
     \bV^T\bW(\bz-\bmu_z)\,.
\end{align*}
Here $\bW=\diag(\bw)$ for 
    $\bw=\left(w_1,\dots,w_d\right)^T$. 
The components of $\bw$ are $w_j=w^{\cell}_jw^{\case}$ 
with  cellwise weights 
$w_{j}^{\cell}=w^{\cell}\left( \frac{r_{j}}
{\sigma_{1,j}}\right)$ 
and casewise weights
$w^{\case}=w^{\case}\left(\frac{\rt}{\sigma_2}\right)$\,.  
We also denote 
$\btW=\diag(w^{\cell}_1,\dots,w^{\cell}_d)$.

\begin{proposition}\label{prop1}
The casewise and cellwise influence functions
of $\vect(\bV)$ and $\bmu_z$ are
\begin{equation*}
\IFu_{\case}\left(\bc,\vect(\bV),H_0\right)= -\bD_{1}\Big[\bS\IFu_{\case}\left(\bc,\bsigma,H_0\right) + 
   \bg\left(\Delta_{\bc},\vect\left(\bV_0\right),\bmu_{z,0},\bsigma_0\right)\Big],
\end{equation*}
\begin{equation*}
\IFu_{\case}\left(\bc,\bmu_z,H_0\right)= -\bD_{2}\Big[\bS\IFu_{\case}\left(\bc,\bsigma,H_0\right)+ 
   \bg\left(\Delta_{\bc},\vect\left(\bV_0\right),\bmu_{z,0},\bsigma_0\right)\Big],
\end{equation*}
and
\begin{equation}\label{eq:IFPFICM1}
\IFu_{\cell}\left(\bc,\vect(\bV),H_0\right)= -\bD_{1}\Big[\bS\IFu_{\cell}\left(\bc,\bsigma,H_0\right)+ 
   d\sum_{j=1}^d
  \bg\left(H\left(j,\bc\right),\vect\left(\bV_0\right),\bmu_{z,0},\bsigma_0\right)\Big],
\end{equation}
\begin{equation}\label{eq:IFPFICM2}
\IFu_{\cell}\left(\bc,\bmu_z,H_0\right)= -\bD_{2}\Big[\bS\IFu_{\cell}\left(\bc,\bsigma,H_0\right)+ 
   d\sum_{j=1}^d
  \bg\left(H\left(j,\bc\right),\vect\left(\bV_0\right),\bmu_{z,0},\bsigma_0\right)\Big],
\end{equation}
with $\bsigma(H)=\left(\sigma_{1,1}(H), \dots, 
\sigma_{1,d}(H), \sigma_{2}(H)\right)^T$, $\bmu_{z,0}:=\bmu_z(H_0)$,
$\bV_0:=\bV(H_0)$, $\bsigma_0:=\bsigma(H_0)$,
\begin{align}\label{eq:gmuL1}
  \bg_1(H,\bmu_z,\bV,\bsigma)&=
  \vect\left(\E_{H}\big[\bW\right(\bV\bu-\bz+\bmu_z\left)\bu^T\big]\right),\\\label{eq:gmuL2}
  \bg_2(H,\bmu_z,\bV,\bsigma)&=
 \E_{H}\big[\bW\left(\bV\bu-\bz+\bmu_z\right)\big],
\end{align}
and $\bg(H,\bmu_z,\bV,\bsigma)=\left(\bg_1(H,\bmu_z,\bV,
\bsigma)^T,\bg_2(H,\bmu_z,\bV,\bsigma)^T\right)^T$. 
The matrices $\bD_{1}$\,, $\bD_{2}$\,, and 
$\bS$ are described in 
\citet{centofanti2025cellwise}, and 
$\IFu_{\case}(\bc,\bsigma,H_0)$ and $\IFu_{\cell}(\bc,\bsigma,H_0)$ 
are the casewise and cellwise influence functions of 
$\bsigma$. 
\end{proposition}

Note that the solutions $\bV(H)$ and $\bmu_z(H)$ of \eqref{eq:ux} are not unique. Therefore, the influence functions of $\vect(\bV)$ and $\bmu_z$ should be interpreted with respect to the functionals $\bV(H)$ and $\bmu_z(H)$ defined as the output of the algorithm used to minimize \eqref{eq:objP}, translated from the finite-sample setting to the population distribution setting and initialized at $\bV_0$ and $\bmu_{z,0}$.
Moreover, note that \eqref{eq:gmuL1} and \eqref{eq:gmuL2} 
express two of the first order conditions, but 
the other first-order condition 
$\left(\bV^T\bW\bV\right)\bu =
\bV^T\bW(\bz-\bmu_z)$ must hold as well, and acts 
as a constraint. Moreover, $\bg_1$ and $\bg_2$ 
depend on $\bsigma$ through $\bW$ and $\bu$\,. 
Also note that $H(j,\bc)$ in \eqref{eq:IFPFICM1} 
and \eqref{eq:IFPFICM2} is the distribution of 
$Z \sim H_0$ but with its \mbox{$j$-th} 
component fixed at the constant $c_j$\,. It is 
thus a degenerate distribution concentrated on 
the hyperplane $z_j = c_j$\,.

To derive the IF of the cellMR estimates, other important 
pieces are the IFs of the functionals  
$\bmu^{\bu}_{\MCD}(H)$ and $\bSigma^{\bu}_{\MCD}(H)$  corresponding to the
MCD estimator of location and covariance $\bmu_{\MCD}$  and $\bSigma_{\MCD}$ 
applied to $\bhu_1,\dots,\bhu_n$ under both FDCM and FICM.  
That is, $\bSigma^{\bu}_{\MCD}(H)$ corresponds 
to the MCD functional $\bSigma_{\MCD}(\cdot)$ 
of scatter with parameter $0.5<\alpha<1$, applied 
to the distribution $H^{\bu}(H, \bT(H))$ of 
$\bu$ when $\bz$ is distributed as $H$, for 
$\bT(H):=\left(\vect(\bV(H)), \bmu_z(H),
\bsigma(H)\right)^T$. Similarly,
$\bmu^{\bu}_{\MCD}(H)$ is the MCD functional 
$\bmu_{\MCD}(\cdot)$ of location with parameter 
$\alpha$ applied to $H^{\bu}(H, \bT(H))$.

\begin{proposition}\label{prop2}
The casewise and cellwise influence functions 
of $\vect(\bSigma^{\bu}_{\MCD})$  and  $\bmu^{\bu}_{\MCD}$ are
\begin{align*}
\IFu_{\case}(\bc,\vect(\bSigma^{\bu}_{\MCD}),H_0)= 
  & -\bD^{\bu}_{1}\Big[\bS^{\bu}  
  \IFu_{\case}\left(\bc,\bsigma,H_0\right) 
  + \bB^{\bu}_{1} \IFu_{\case}
    \left(\bc,\vect(\bV),H_0\right)\\
 &\;\; + \bB^{\bu}_{2} \IFu_{\case}\left(
   \bc,\bmu_z,H_0\right)+\bg^{\bu}(\Delta_{\bc}, 
   \bT_0,\bSigma^{\bu}_{\MCD,0},
   \bmu^{\bu}_{\MCD,0},q_{\alpha,0})\Big]
\end{align*}
\begin{align*}
\IFu_{\case}(\bc,\bmu^{\bu}_{\MCD},H_0)
  =& -\bD^{\bu}_{2}\Big[\bS^{\bu} 
     \IFu_{\case}\left(\bc,\bsigma,H_0\right)
     + \bB^{\bu}_{1} \IFu_{\case}
       \left(\bc,\vect(\bV),H_0\right)\\
  &+ \bB^{\bu}_{2} \IFu_{\case}\left(\bc,\bmu_z,H_0\right)
     +\bg^{\bu}(\Delta^{\bu}_{\bc}, 
     \bT_0,\bSigma^{\bu}_{\MCD,0},
     \bmu^{\bu}_{\MCD,0},q_{\alpha,0})\Big],
  \end{align*}
and
\begin{align*}
  \IFu_{\cell}(\bc,\vect(\bSigma^{\bu}_{\MCD}),H_0)= 
  & -\bD^{\bu}_{1}\Big[\bS^{\bu} 
    \IFu_{\cell}\left(\bc,\bsigma,H_0\right) +
    \bB^{\bu}_{1} \IFu_{\cell}
    \left(\bc,\vect(\bV),H_0\right)\\
  &\;\; +
    \bB^{\bu}_{2} \IFu_{\cell}\left(\bc,\bmu_z,H_0\right)\\
  &\;\; +
    d\sum_{j=1}^d\bg^{\bu}(H(j,\bc), 
    \bT_0,\bSigma^{\bu}_{\MCD,0},
    \bmu^{\bu}_{\MCD,0},q_{\alpha,0}) \Big]\,,
\end{align*}
\begin{align*}
\IFu_{\cell}(\bc,\bmu^{\bu}_{\MCD},H_0)
  =& -\bD^{\bu}_{2}\Big[\bS^{\bu} \IFu_{\cell}
      \left(\bc,\bsigma,H_0\right)
     +\bB^{\bu}_{1} \IFu_{\cell}
     \left(\bc,\vect(\bV),H_0\right)\\
  &+ \bB^{\bu}_{2} \IFu_{\cell}\left(\bc,\bmu_z,H_0\right)
     + d\sum_{j=1}^d\bg^{\bu}(H(j,\bc), 
     \bT_0,\bSigma^{\bu}_{\MCD,0},
     \bmu^{\bu}_{\MCD,0},q_{\alpha,0})\Big].
\end{align*}
with $\bsigma(H)=\left(\sigma_{1,1}(H), \dots, 
\sigma_{1,d}(H), \sigma_{2}(H)\right)^T$, 
$\bmu^{\bu}_{\MCD,0}:=\bmu^{\bu}_{\MCD}(H_0)$, 
$\bSigma^{\bu}_{\MCD,0}:=\bSigma^{\bu}_{\MCD}(H_0)$,
$\bT_0:=\left(\vect(\bV_0), \bmu_{z,0},\bsigma_0\right)^T$ and
\begin{equation}
    \label{eq:newgf2}
  \bg^{\bu}(H,\bT,\bSigma^{\bu}_{\MCD},\bmu^{\bu}_{\MCD},q_{\alpha})=\left[\begin{array}{cc}
       \bg^{\bu}_1(H,\bT,\bSigma^{\bu}_{\MCD},\bmu^{\bu}_{\MCD},q_{\alpha}) \\
       \bg^{\bu}_2(H,\bT,\bSigma^{\bu}_{\MCD},\bmu^{\bu}_{\MCD},q_{\alpha}) \\
       g^{\bu}_3(H,\bT,\bSigma^{\bu}_{\MCD},\bmu^{\bu}_{\MCD},q_{\alpha})
  \end{array}\right].
\end{equation}
The functions $\bg^{\bu}_1$, $\bg^{\bu}_2$, and 
$g^{\bu}_3$ are defined as
\begin{align*}\label{eq:newgf1}
    \bg^{\bu}_1(H,\bT,\bSigma^{\bu}_{\MCD},\bmu^{\bu}_{\MCD},q_{\alpha})
  &=\vect\left(\E_{H^{\bu}(H, \bT)}\right.\Big[\ind\left(
    {\bu \in A(\bSigma^{\bu}_{\MCD},\bmu^{\bu}_{\MCD},q_{\alpha})}\right)\\
  &\;\;\;\;\; \left.\left(c_{\alpha}\left(\bu-\bmu^{\bu}_{\MCD}\right)\left(\bu-\bmu^{\bu}_{\MCD}
  \right)^T-\bSigma^{\bu}_{\MCD}\right)\Big]\right)\,,
\end{align*}
\begin{equation*}\label{eq:newgf2.b}
    \bg^{\bu}_2(H,\bT,\bSigma^{\bu}_{\MCD},\bmu^{\bu}_{\MCD},q_{\alpha})\nonumber=\E_{H^{\bu}(H, \bT)}\Big[\ind\left(
    {\bu \in A(\bSigma^{\bu}_{\MCD},\bmu^{\bu}_{\MCD},q_{\alpha})}\right)\left(\bu-\bmu^{\bu}_{\MCD}\right)\Big],
\end{equation*}
and
\begin{equation*}\label{eq:newgf2.c}
    g^{\bu}_3(H,\bT,\bSigma^{\bu}_{\MCD},\bmu^{\bu}_{\MCD},q_{\alpha})=\E_{H^{\bu}(H, \bT)}\Big[\ind\left(
    {\bu \in A(\bSigma^{\bu}_{\MCD},\bmu^{\bu}_{\MCD},q_{\alpha})}\right)-(1-\alpha)\Big]
\end{equation*}
where 
\begin{align*}
   A(\bSigma^{\bu}_{\MCD},\bmu^{\bu}_{\MCD},q_{\alpha})=\lbrace \bu \in \mathbb{R}^k:\left(\bu-\bmu^{\bu}_{\MCD}\right)^T\left(\bSigma^{\bu}_{\MCD}\right)^{-1}\left(\bu-\bmu^{\bu}_{\MCD}\right)\leqslant q_{\alpha}\rbrace.
\end{align*}
The functional $q_{\alpha}(H)$  satisfies
\begin{equation*}
    \int_{\mathbb{R}^k}\ind\left(
    {\bu \in A( \bSigma^{\bu}_{\MCD}(H),\bmu^{\bu}_{\MCD}(H),q_{\alpha}(H))}\right)dH^{\bu}(H, \bT)(\bu)=1-\alpha,
\end{equation*}
in which $q_{\alpha,0}:=q_{\alpha}(H_0)$ and
$c_{\alpha}$ is chosen in such a way that 
consistency is obtained at a prespecified model. 
The matrices $\bB^{\bu}_{1}$, 
$\bB^{\bu}_{2}$,  $\bD^{\bu}_{1}$ and $\bD^{\bu}_{2}$ are 
defined in \citet{centofanti2025cellwise}.
\end{proposition}

\begin{proposition}\label{prop3}
The casewise and cellwise influence functions 
of $\vect(\btSigma_{\zort})$ are
\begin{align*}
  \IFu_{\case}\left(\bc,\vect(\btSigma_{\zort}),H_0\right)
   =&\;\bD^{\zort}\Big[\bB^{\zort}_{1}\IFu_{\case}
      \left(\bc,\bmu_z,H_0\right)+\bB^{\zort}_{2}\IFu_{\case}
      \left(\bc,\vect(\bV),H_0\right)\\
   &+\bS^{\zort}\IFu_{\case}\left(\bc,\bsigma,H_0\right)
      +\bg^{\zort}(\Delta_{\bc},
       \bT_0,\bSigma_{\zort,0})\Big]
\end{align*}
and
\begin{align*}
  \IFu_{\cell}\left(\bc,\vect(\btSigma_{\zort}),H_0\right)
  =&\;\bD^{\zort}\Big[\bB^{\zort}_{1}\IFu_{\cell}
     \left(\bc,\bmu_z,H_0\right)+\bB^{\zort}_{2}
     \IFu_{\cell}\left(\bc,\vect(\bV),H_0\right)\\
   &+\bS^{\zort}\IFu_{\cell}\left(\bc,\bsigma,H_0\right)
    +d\sum_{j=1}^d \bg^{\zort}(H\left(j,\bc\right),
    \bT_0,\bSigma_{\zort,0})\Big],
\end{align*}
with  $ \bSigma_{\zort,0}= \btSigma_{\zort}(H_0)$,
$\bT_0=\left(\vect(\bV_0), \bmu_{z,0},\bsigma_0\right)^T$ and
\begin{multline*}
  \bg^{\zort}(H,\bT,\btSigma_{\zort}):=
  \vect\Big(\E_{H}\Big[b\btSigma_{\zort}
  - w^{\case}\btW(\bz - \bmu_z-\bV\bu)(\bz - \bmu_z-\bV\bu)^T
  \btW\Big]\Big),
\end{multline*}
where $\bT=\left(\vect(\bV), \bmu_z,\bsigma\right)^T$, 
$b=\sum_{j=1}^d\sum_{\ell=1}^d w^{\case}w_{j}^{\cell}w_{\ell}^{\cell}/d^2$\,, 
 $\bu$ depends on $\bz$ and $\bT$ through \eqref{eq:ux}.
The matrices $\bD^{\zort}$, $\bB_1^{\zort}$, 
$\bB_2^{\zort}$, and $\bS^{\zort}$  are 
defined in \citet{centofanti2025cellwise}. 
\end{proposition}

\vspace{3mm}

\begin{proof}[\textbf{Proof of Proposition~\ref{IFBb}}]
Let us introduce the functionals $\bmu$ and $\bSigma$ corresponding to $\bhmu$ and $\bhSigma$ defined in \eqref{eq:cellCov}, which can be written by using \eqref{eq:SigmaMCD} as
\begin{equation}\label{eq:functionalmuSigma}
\bmu(H)=\bmu_z(H)+\bV(H)\bmu^{u}_{\text{\tiny MCD}}(H),\quad \bSigma(H)=\bV(H)
\bSigma^u_{\text{\tiny MCD}}(H)(\bV(H))^T+\btSigma_{\zort}(H).
\end{equation}
Define the  selection matrices $\bS_x\in\mathbb R^{p\times d}$ and
$\bS_y\in\mathbb R^{q\times d}$ such that
$\bmu_x(H)=\bS_x\bmu(H)$, $\bmu_y(H)=\bS_y\bmu(H)$ and
$\bSigma_x(H)=\bS_x\bSigma(H) \bS_x^T$, $\bSigma_{xy}(H)=\bS_x\bSigma(H) \bS_y^T$, where $ \bmu_x$, $ \bmu_y$,   $\bSigma_{x}$ and $\bSigma_{xy}$ represent the usual partitions of $\bmu$ and $\bSigma$. 
From the definition of  $\bhb$ and $\bhB$ in 
\eqref{eq:ridge} we have
\begin{equation}
\label{eq:blue2_IF}
\bB(H)=\bigl(\bS_x\bSigma(H) \bS_x^T+\lambda\bI_p\bigr)^{-1}(\bS_x\bSigma(H) \bS_y^T),
\qquad
\bb(H)=\bS_y\bmu(H)-(\bB(H)
)^T(\bS_x\bmu(H)).
\end{equation}

Let $\bA(H):=\bS_x\bSigma(H)\bS_x^T+\lambda\bI_p$ and
$\bC(H):=\bS_x\bSigma(H)\bS_y^T$, so that $\bB(H)=\bA(H)^{-1}\bC(H)$.
Using the identity
\begin{equation*}
\left.\frac{\partial}{\partial\eps}\bA(H_\eps)^{-1}\right|_{\eps=0}
=
-\bA_0^{-1}
\left.\frac{\partial}{\partial\eps}\bA(H_\eps)\right|_{\eps=0}
\bA_0^{-1},
\end{equation*}
we obtain
\begin{equation*}
\left.\frac{\partial}{\partial\eps}\bB(H_\eps)\right|_{\eps=0}
=
-\bA_0^{-1}
\left.\frac{\partial}{\partial\eps}\bA(H_\eps)\right|_{\eps=0}
\bA_0^{-1}\bC_0
+\bA_0^{-1}
\left.\frac{\partial}{\partial\eps}\bC(H_\eps)\right|_{\eps=0},
\end{equation*}
where $\bA_0:=\bA(H_0)$ and $\bC_0:=\bC(H_0)$.
Equivalently, using $\bB_0:=\bB(H_0)=\bA_0^{-1}\bC_0$,
\begin{equation*}
\left.\frac{\partial}{\partial\eps}\bB(H_\eps)\right|_{\eps=0}
=
-\bA_0^{-1}
\left.\frac{\partial}{\partial\eps}\bA(H_\eps)\right|_{\eps=0}
\bB_0
+\bA_0^{-1}
\left.\frac{\partial}{\partial\eps}\bC(H_\eps)\right|_{\eps=0}.
\end{equation*}
Since
\begin{equation*}
\left.\frac{\partial}{\partial\eps}\bA(H_\eps)\right|_{\eps=0}
=
\bS_x
\left.\frac{\partial}{\partial\eps}\bSigma(H_\eps)\right|_{\eps=0}
\bS_x^T,
\qquad
\left.\frac{\partial}{\partial\eps}\bC(H_\eps)\right|_{\eps=0}
=
\bS_x
\left.\frac{\partial}{\partial\eps}\bSigma(H_\eps)\right|_{\eps=0}
\bS_y^T,
\end{equation*}
it follows that
\begin{equation*}
\left.\frac{\partial}{\partial\eps}\bB(H_\eps)\right|_{\eps=0}
=
\bA_0^{-1}\bS_x
\left.\frac{\partial}{\partial\eps}\bSigma(H_\eps)\right|_{\eps=0}
\bigl(\bS_y^T-\bS_x^T\bB_0\bigr).
\end{equation*}
Finally, vectorizing and using $\vect(\bM\bX\bN)=(\bN^T\otimes\bM)\vect(\bX)$, we obtain
\begin{align*}
\left.\frac{\partial}{\partial\eps}\vect\bigl(\bB(H_\eps)\bigr)\right|_{\eps=0}
&=
\Bigl(
\bigl(\bS_y-\bB_0^T\bS_x\bigr)\otimes
((\bS_x\bSigma_0\bS_x^T+\lambda\bI_p)^{-1}\bS_x
)\Bigr)
\left.\frac{\partial}{\partial\eps}\vect\bigl(\bSigma(H_\eps)\bigr)\right|_{\eps=0}\\
&=\bR_{B}\left.\frac{\partial}{\partial\eps}\vect\bigl(\bSigma(H_\eps)\bigr)\right|_{\eps=0}
\end{align*}
with $\bSigma_0:=\bSigma(H_0)$, and $\bR_{B}:=
\bigl(\bS_y-\bB_0^T\bS_x\bigr)\otimes
((\bS_x\bSigma_0\bS_x^T+\lambda\bI_p)^{-1}\bS_x
)$.

Now consider $\bb(H)$, we have
\begin{align*}
\left.\frac{\partial}{\partial\eps}\bb(H_\eps)\right|_{\eps=0}
&=
\bS_y
\left.\frac{\partial}{\partial\eps}\bmu(H_\eps)\right|_{\eps=0}
-
\left.\frac{\partial}{\partial\eps}
\bigl(\bB(H_\eps)^T\bS_x\bmu(H_\eps)\bigr)
\right|_{\eps=0}.
\end{align*}
For the second term, applying the product rule yields
\begin{align*}
\left.\frac{\partial}{\partial\eps}
\bB(H_\eps)^T\bS_x\bmu(H_\eps)
\right|_{\eps=0}
&=
\bigl(\bmu_0^T\bS_x^T\otimes \bI_q\bigr)
\left.\frac{\partial}{\partial\eps}
\vect\bigl(\bB(H_\eps)^T\bigr)
\right|_{\eps=0}
+
\bB_0^T\bS_x
\left.\frac{\partial}{\partial\eps}
\bmu(H_\eps)
\right|_{\eps=0},
\end{align*}
with $\bmu_0:=\bmu(H_0)$.
Define the permutation matrix satisfying $\vect(\bB(H_\eps)^T)=\bP_{B}\vect(\bB(H_\eps))$. Then
\begin{align*}
\left.\frac{\partial}{\partial\eps}
\bb(H_\eps)
\right|_{\eps=0}
&=
\Bigl(
\bS_y
-
\bB_0^T\bS_x
\Bigr)
\left.\frac{\partial}{\partial\eps}
\bmu(H_\eps)
\right|_{\eps=0}
-
\bigl(\bmu_0^T\bS_x^T\otimes \bI_q\bigr)
\bP_B
\left.\frac{\partial}{\partial\eps}
\vect\bigl(\bB(H_\eps)\bigr)
\right|_{\eps=0}\\
&=\bR_{b,1}
\left.\frac{\partial}{\partial\eps}
\bmu(H_\eps)
\right|_{\eps=0} 
-
\bR_{b,2}
\left.\frac{\partial}{\partial\eps}
\vect\bigl(\bB(H_\eps)\bigr)
\right|_{\eps=0},
\end{align*}
where $\bR_{b,1}:=
\bS_y
-
\bB_0^T\bS_x$ and $\bR_{b,2}:=\bigl(\bmu_0^T\bS_x^T\otimes \bI_q\bigr)
\bP_B$.

By using \eqref{eq:functionalmuSigma}, we obtain 
\begin{align*}
    \left.\frac{\partial}{\partial \eps} \bSigma(H_{\eps})\right|_{\eps=0}
    &=\left.\frac{\partial}{\partial \eps}\bV(H_{\eps})\bSigma^{u}_{\MCD}(H_{\eps})\left(\bV(H_{\eps})\right)^T\right|_{\eps=0}+\left.\frac{\partial}{\partial \eps}\btSigma_{\zort}(H_{\eps})\right|_{\eps=0}\\
    &=\left.\frac{\partial}{\partial \eps}\bV(H_{\eps})\right|_{\eps=0}\bSigma^{u}_{\MCD,0}\bV_0^T+\bV_0\left.\frac{\partial}{\partial \eps}\bSigma^{u}_{\MCD}(H_{\eps})\right|_{\eps=0}\bV_0^T\\
    &\;\;\;\;\; +\bV_0\bSigma^{u}_{\MCD,0}\left.\frac{\partial}{\partial \eps}\left(\bV(H_{\eps})\right)^T\right|_{\eps=0}+\left.\frac{\partial}{\partial \eps}\btSigma_{\zort}(H_{\eps})\right|_{\eps=0}\;,
\end{align*}
where $\bV_{0}:=\bV(H_0)$, and 
$\bSigma^{u}_{\MCD,0}:=\bSigma^{u}_{\MCD}(H_0)$.
Applying $\vect$ to both sides yields 
\begin{align*}
    &\left.\frac{\partial}{\partial \eps} \vect\left(\bSigma(H_{\eps})\right)\right|_{\eps=0}\\
    &=\left(\bV_0\bSigma^{u}_{\MCD,0}\otimes\bI_d\right)\left.\frac{\partial}{\partial \eps}\vect\left(\bV(H_{\eps})\right)\right|_{\eps=0}
+ \left(\bV_0\otimes\bV_0\right)\left.\frac{\partial}{\partial \eps}\vect\left(\bSigma^{u}_{\MCD}(H_{\eps})\right)\right|_{\eps=0}\\
    &\;\;\;\;\; +\left(\bI_d\otimes\bV_0\bSigma^{u}_{\MCD,0}\right)\bP_V\left.\frac{\partial}{\partial \eps}\vect\left(\bV(H_{\eps})\right)\right|_{\eps=0}
      +\left.\frac{\partial}{\partial \eps}\vect\left(\btSigma_{\zort}(H_{\eps})\right)\right|_{\eps=0}\\
        &=\bR_{\Sigma,1}\left.\frac{\partial}{\partial \eps}\vect\left(\bV(H_{\eps})\right)\right|_{\eps=0}
      +\bR_{\Sigma,2}\left.\frac{\partial}{\partial \eps}\vect\left(\bSigma^{u}_{\MCD}(H_{\eps})\right)\right|_{\eps=0}
      +\left.\frac{\partial}{\partial \eps}\vect\left(\btSigma_{\zort}(H_{\eps})\right)\right|_{\eps=0},
\end{align*}
where the matrix $\bP_{V}$ is a  permutation 
matrix such that  
$\vect\left(\left.\frac{\partial}{\partial \eps}
\left(\bV(H_{\eps})\right)^T\right|_{\eps=0}\right)=\bP_{V}\vect\left(\left.\frac{\partial}{\partial \eps}
\left(\bV(H_{\eps})\right)\right|_{\eps=0}\right)$.
The matrices $\bR_{\Sigma,1}$ and $\bR_{\Sigma,2}$ are
\begin{align*}
    \bR_{\Sigma,1}&=\left(\bV_0\bSigma^{u}_{\MCD,0}\otimes\bI_d\right)+\left(\bI_d\otimes\bV_0\bSigma^{u}_{\MCD,0}\right)\bP_V,\\
    \bR_{\Sigma,2}&=\bV_0\otimes\bV_0.
\end{align*}

Similarly, by \eqref{eq:functionalmuSigma}, we have
\begin{align*}
    \left.\frac{\partial}{\partial \eps}\bmu(H_{\eps})\right|_{\eps=0}
    &=
    \left.\frac{\partial}{\partial \eps}\bmu_z(H_{\eps})\right|_{\eps=0}
    +\left.\frac{\partial}{\partial \eps}\Bigl(\bV(H_{\eps})\bmu^{u}_{\text{\tiny MCD}}(H_{\eps})\Bigr)\right|_{\eps=0}\\
    &=
    \left.\frac{\partial}{\partial \eps}\bmu_z(H_{\eps})\right|_{\eps=0}
    +\left.\frac{\partial}{\partial \eps}\bV(H_{\eps})\right|_{\eps=0}\bmu^{u}_{\text{\tiny MCD},0}
    +\bV_0\left.\frac{\partial}{\partial \eps}\bmu^{u}_{\text{\tiny MCD}}(H_{\eps})\right|_{\eps=0}\\
    &=
    \left.\frac{\partial}{\partial \eps}\bmu_z(H_{\eps})\right|_{\eps=0}
    +\left(\left(\bmu^{u}_{\text{\tiny MCD},0}\right)^T\otimes \bI_d\right)
    \left.\frac{\partial}{\partial \eps}\vect\bigl(\bV(H_{\eps})\bigr)\right|_{\eps=0}
    +\bV_0\left.\frac{\partial}{\partial \eps}\bmu^{u}_{\text{\tiny MCD}}(H_{\eps})\right|_{\eps=0}\\
    &=
    \bR_{\mu}\left.\frac{\partial}{\partial \eps}\vect\bigl(\bV(H_{\eps})\bigr)\right|_{\eps=0}
    +\bV_0\left.\frac{\partial}{\partial \eps}\bmu^{u}_{\text{\tiny MCD}}(H_{\eps})\right|_{\eps=0}
    +\left.\frac{\partial}{\partial \eps}\bmu_z(H_{\eps})\right|_{\eps=0},
\end{align*}
where $\bmu^{u}_{\text{\tiny MCD},0}=\bmu^{u}_{\text{\tiny MCD}}(H_0)$
and $\bR_{\mu}=\left(\left(\bmu^{u}_{\text{\tiny MCD},0}\right)^T\otimes \bI_d\right)$.

So under the FDCM, we have
\begin{align*}
\left.\frac{\partial}{\partial\eps}
\vect\bigl(\bB(H(G_\eps^D,\bc))\bigr)
\right|_{\eps=0}
&=
\bR_B\,
\left.\frac{\partial}{\partial\eps}
\vect\bigl(\bSigma(H(G_\eps^D,\bc))\bigr)
\right|_{\eps=0} \nonumber\\
&=
\bR_B\Bigg[
\bR_{\Sigma,1}
\left.\frac{\partial}{\partial \eps}
\vect\bigl(\bV(H(G_\eps^D,\bc))\bigr)
\right|_{\eps=0} \nonumber\\
&\hspace{2.2em}
+\bR_{\Sigma,2}
\left.\frac{\partial}{\partial \eps}
\vect\bigl(\bSigma^{u}_{\MCD}(H(G_\eps^D,\bc))\bigr)
\right|_{\eps=0} \nonumber\\
&\hspace{2.2em}
+\left.\frac{\partial}{\partial \eps}
\vect\bigl(\btSigma_{\zort}(H(G_\eps^D,\bc))\bigr)
\right|_{\eps=0}
\Bigg] \label{eq:IF_B_intermediate}\\
&=
\bR_B\Big[
\bR_{\Sigma,1}\,
\IFu_{\case}\!\left(\bc,\vect(\bV),H_0\right)\\&\hspace{2.2em}
+\bR_{\Sigma,2}\,
\IFu_{\case}\!\left(\bc,\vect(\bSigma^{u}_{\MCD}),H_0\right)
+\IFu_{\case}\!\left(\bc,\vect(\btSigma_{\zort}),H_0\right)
\Big].
\end{align*}
Similarly
\begin{align*}
\left.\frac{\partial}{\partial\eps}
\bb\bigl(H(G_\eps^D,\bc)\bigr)
\right|_{\eps=0}
&=
\bR_{b,1}
\left.\frac{\partial}{\partial\eps}
\bmu\bigl(H(G_\eps^D,\bc)\bigr)
\right|_{\eps=0} 
-
\bR_{b,2}
\left.\frac{\partial}{\partial\eps}
\vect\bigl(\bB(H(G_\eps^D,\bc))\bigr)
\right|_{\eps=0}\\
&=
\bR_{b,1}\Bigg[
\bR_{\mu}
\left.\frac{\partial}{\partial \eps}
\vect\bigl(\bV(H(G_\eps^D,\bc))\bigr)
\right|_{\eps=0}
+\bV_0
\left.\frac{\partial}{\partial \eps}
\bmu^{u}_{\text{\tiny MCD}}\bigl(H(G_\eps^D,\bc)\bigr)
\right|_{\eps=0}\\
&\hspace{2.2em}
+\left.\frac{\partial}{\partial \eps}
\bmu_z\bigl(H(G_\eps^D,\bc)\bigr)
\right|_{\eps=0}
\Bigg]
-
\bR_{b,2}
\left.\frac{\partial}{\partial\eps}
\vect\bigl(\bB(H(G_\eps^D,\bc))\bigr)
\right|_{\eps=0}\\
&=
\bR_{b,1}\Big[
\bR_{\mu}\,
\IFu_{\case}\!\left(\bc,\vect(\bV),H_0\right)
+\bV_0\,
\IFu_{\case}\!\left(\bc,\bmu^{u}_{\text{\tiny MCD}},H_0\right)
\\
&\hspace{2.2em}+\IFu_{\case}\!\left(\bc,\bmu_z,H_0\right)
\Big]
-\bR_{b,2}\,
\IFu_{\case}\!\left(\bc,\vect(\bB),H_0\right).
\end{align*}
Using similar arguments, we obtain for the FICM
\begin{align*}
\left.\frac{\partial}{\partial\eps}
\vect\bigl(\bB(H(G_\eps^I,\bc))\bigr)
\right|_{\eps=0}&=
\bR_B\Big[
\bR_{\Sigma,1}\,
\IFu_{\cell}\!\left(\bc,\vect(\bV),H_0\right)\\&\hspace{2.2em}
+\bR_{\Sigma,2}\,
\IFu_{\cell}\!\left(\bc,\vect(\bSigma^{u}_{\MCD}),H_0\right)
+\IFu_{\cell}\!\left(\bc,\vect(\btSigma_{\zort}),H_0\right)
\Big],
\end{align*}
and
\begin{align*}
\left.\frac{\partial}{\partial\eps}
\bb\bigl(H(G_\eps^I,\bc)\bigr)
\right|_{\eps=0}
&=
\bR_{b,1}\Big[
\bR_{\mu}\,
\IFu_{\cell}\!\left(\bc,\vect(\bV),H_0\right)
+\bV_0\,
\IFu_{\cell}\!\left(\bc,\bmu^{u}_{\text{\tiny MCD}},H_0\right)
\\
&\hspace{2.2em}+\IFu_{\cell}\!\left(\bc,\bmu_z,H_0\right)
\Big]
-\bR_{b,2}\,
\IFu_{\cell}\!\left(\bc,\vect(\bB),H_0\right).
\end{align*}
This completes the proof.
\end{proof}

\subsection{Effect of ridge regularization on the influence function}
\label{app:lambda_IF}

We now investigate more explicitly how the ridge parameter
$\lambda$ affects the influence functions of the cellMR regression
coefficients. The boundedness established in
Proposition~\ref{IFBb} originates from the robust location and
covariance estimator underlying cellMR. For a fixed $\lambda$,
ridge regularization preserves this boundedness, but affects how
perturbations of the location and covariance estimates are
propagated to the regression coefficients.

We use the same notation as in the proof of
Proposition~\ref{IFBb}, making the dependence on $\lambda$
explicit where needed. In particular,
\[
\bB_0(\lambda)
=
\bigl(\bSigma_{x,0}+\lambda\bI_p\bigr)^{-1}
\bSigma_{xy,0}\;,
\]
and
\[
\bR_B(\lambda)
=
\bigl(\bS_y-\bB_0(\lambda)^T\bS_x\bigr)
\otimes
\left[
\bigl(\bSigma_{x,0}+\lambda\bI_p\bigr)^{-1}\bS_x
\right],
\]
whereas
\[
\bR_{b,1}(\lambda)
=
\bS_y-\bB_0(\lambda)^T\bS_x\;.
\]
The matrix $\bR_{b,2}$ defined in the proof of
Proposition~\ref{IFBb} does not depend on $\lambda$.
The influence functions of $\bmu$ and $\bSigma$ 
also do not depend on $\lambda$, since $\lambda$
enters only in the ridge regression step.

\begin{proposition}
\label{prop:lambda_IF}
For any fixed contamination point $\bc$ and
$\tau\in\{\case,\cell\}$, we have
\[
\frac{\partial}{\partial\lambda}
\IFu_{\tau}\!\left(
\bc,\vect(\bB),H_0
\right)
=
\frac{\partial}{\partial\lambda}\bR_B(\lambda)
\IFu_{\tau}\!\left(
\bc,\vect(\bSigma),H_0
\right),
\]
where
\begin{align}
\frac{\partial}{\partial\lambda}\bR_B(\lambda)
={}&
\left[
\bB_0(\lambda)^T
\bigl(\bSigma_{x,0}+\lambda\bI_p\bigr)^{-1}
\bS_x
\right]
\otimes
\left[
\bigl(\bSigma_{x,0}+\lambda\bI_p\bigr)^{-1}
\bS_x
\right]
\nonumber\\
&-
\left[
\bS_y-\bB_0(\lambda)^T\bS_x
\right]
\otimes
\left[
\bigl(\bSigma_{x,0}+\lambda\bI_p\bigr)^{-2}
\bS_x
\right].
\label{eq:RBprime}
\end{align}
Moreover,
\begin{align}
\frac{\partial}{\partial\lambda}
\IFu_{\tau}\!\left(
\bc,\bb,H_0
\right)
={}&
\bB_0(\lambda)^T
\bigl(\bSigma_{x,0}+\lambda\bI_p\bigr)^{-1}
\bS_x
\IFu_{\tau}\!\left(
\bc,\bmu,H_0
\right)
\nonumber\\
&-
\bR_{b,2}\frac{\partial}{\partial\lambda}\bR_B(\lambda)
\IFu_{\tau}\!\left(
\bc,\vect(\bSigma),H_0
\right).
\label{eq:dIFb_dlambda}
\end{align}
\end{proposition}

\begin{proof}
From the proof of Proposition~\ref{IFBb}, for
$\tau\in\{\case,\cell\}$,
\[
\IFu_{\tau}\!\left(
\bc,\vect(\bB),H_0
\right)
=
\bR_B(\lambda)
\IFu_{\tau}\!\left(
\bc,\vect(\bSigma),H_0
\right).
\]
Since the covariance estimator underlying cellMR
is computed before using the ridge parameter, its
influence function does not depend on $\lambda$.
Therefore,
\[
\frac{\partial}{\partial\lambda}
\IFu_{\tau}\!\left(
\bc,\vect(\bB),H_0
\right)
=
\frac{\partial\bR_B(\lambda)}{\partial\lambda}
\IFu_{\tau}\!\left(
\bc,\vect(\bSigma),H_0
\right).
\]

Using
\[
\frac{\partial}{\partial\lambda}
\bigl(\bSigma_{x,0}+\lambda\bI_p\bigr)^{-1}
=
-
\bigl(\bSigma_{x,0}+\lambda\bI_p\bigr)^{-1}
\bigl(\bSigma_{x,0}+\lambda\bI_p\bigr)^{-1},
\]
we obtain
\[
\frac{\partial\bB_0(\lambda)}
{\partial\lambda}
=
-
\bigl(\bSigma_{x,0}+\lambda\bI_p\bigr)^{-1}
\bB_0(\lambda).
\]
and
\[
\frac{\partial\bR_{b,1}(\lambda)}
{\partial\lambda}
=
\bB_0(\lambda)^T
\bigl(\bSigma_{x,0}+\lambda\bI_p\bigr)^{-1}
\bS_x.
\]

Furthermore,
\[
\frac{\partial}{\partial\lambda}
\left[
\bigl(\bSigma_{x,0}+\lambda\bI_p\bigr)^{-1}
\bS_x
\right]
=
-
\bigl(\bSigma_{x,0}+\lambda\bI_p\bigr)^{-2}
\bS_x.
\]
Since
\[
\bR_B(\lambda)
=
\bR_{b,1}(\lambda)
\otimes
\left[
\bigl(\bSigma_{x,0}+\lambda\bI_p\bigr)^{-1}
\bS_x
\right],
\]
the product rule for the Kronecker product gives
\eqref{eq:RBprime}.

Finally, Proposition~\ref{IFBb} gives
\[
\IFu_{\tau}\!\left(
\bc,\bb,H_0
\right)
=
\bR_{b,1}(\lambda)
\IFu_{\tau}\!\left(
\bc,\bmu,H_0
\right)
-
\bR_{b,2}
\IFu_{\tau}\!\left(
\bc,\vect(\bB),H_0
\right).
\]
The influence function of $\bmu$ and the matrix $\bR_{b,2}$
do not depend on $\lambda$. Therefore, differentiating the
previous expression and using the result for
$\IFu_{\tau}(\bc,\vect(\bB),H_0)$ gives
\eqref{eq:dIFb_dlambda}.
\end{proof}
Proposition~\ref{prop:lambda_IF} shows that the effect of
$\lambda$ on the influence function is not necessarily monotone.
Increasing $\lambda$ reduces the magnitude of
$(\bSigma_{x,0}+\lambda\bI_p)^{-1}$, but it also changes
$\bB_0(\lambda)$ and therefore the way in which the different
components of the covariance influence function are combined.
These components may reinforce or partially offset each other,
and the extent of this cancellation may change with $\lambda$.

This is particularly transparent when $p=q=1$. In this case,
\[
B_0(\lambda)
=
\frac{\sigma_{xy,0}}
{\sigma_{x,0}^2+\lambda},
\]
and, for $\tau\in\{\case,\cell\}$,
\begin{align}
\IFu_{\tau}(\bc,B,H_0)
=
\frac{
\IFu_{\tau}(\bc,\sigma_{xy},H_0)
}{
\sigma_{x,0}^2+\lambda
}
-
\frac{
\sigma_{xy,0}
\IFu_{\tau}(\bc,\sigma_x^2,H_0)
}{
(\sigma_{x,0}^2+\lambda)^2
}.
\label{eq:IF_scalar_lambda}
\end{align}
Therefore,
\begin{align}
\frac{\partial}{\partial\lambda}
\IFu_{\tau}(\bc,B,H_0)
=
-
\frac{
\IFu_{\tau}(\bc,\sigma_{xy},H_0)
}{
(\sigma_{x,0}^2+\lambda)^2
}+
\frac{
2\sigma_{xy,0}
\IFu_{\tau}(\bc,\sigma_x^2,H_0)
}{
(\sigma_{x,0}^2+\lambda)^3
}.
\label{eq:dIF_scalar_lambda}
\end{align}
The sign of~\eqref{eq:dIF_scalar_lambda} is not fixed. Hence,
even though ridge regularization attenuates the individual terms,
their different rates of attenuation may reduce or increase their
mutual cancellation. For sufficiently large $\lambda$, however,
both terms in~\eqref{eq:IF_scalar_lambda} converge to zero and thus
\[
\IFu_{\tau}(\bc,B,H_0)\rightarrow0,
\qquad
\lambda\rightarrow\infty,
\]
for every fixed contamination point $\bc$.

We illustrate this behavior in the bivariate normal setting used
in Section~\ref{sec:robustness_cellMR}. Figure~\ref{fig:IF_lambda}
shows the casewise and cellwise influence functions of $\bB$
as functions of $\lambda$ at the fixed contamination point
$\bc=(1,1)^T$. The figure confirms the results in
Proposition~\ref{prop:lambda_IF}, as the magnitude of the IF does not vary monotonically for intermediate values of $\lambda$,
whereas it approaches zero as the regularization becomes
sufficiently large.

\begin{figure}[H]
\centering
\includegraphics[width=0.4\textwidth]
  {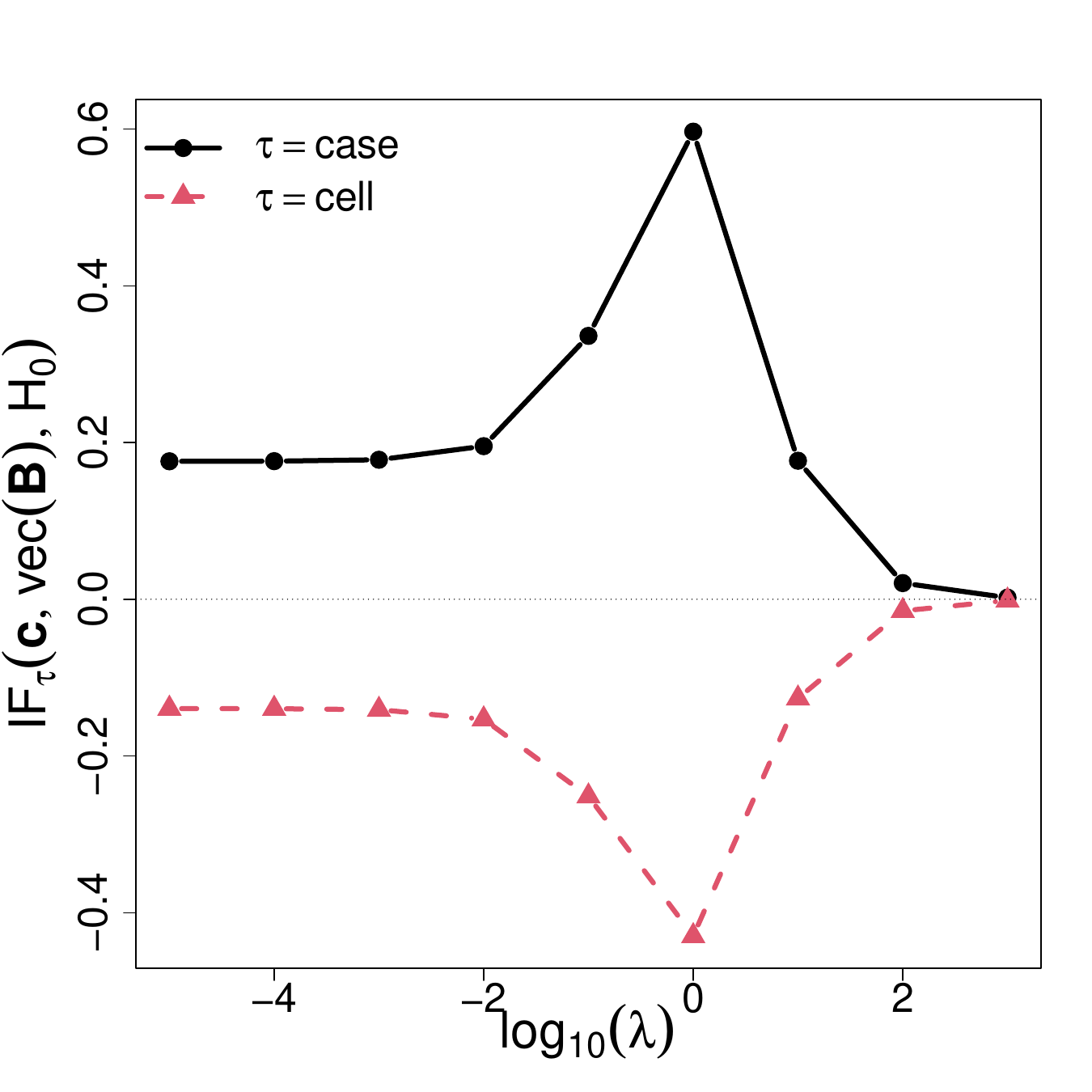}
\caption{The casewise and cellwise IF of $\bB$ in a simple linear model
as functions of $\lambda$ at the contamination point
$\bc=(1,1)^T$.}
\label{fig:IF_lambda}
\end{figure}

To complement this pointwise analysis,
Figure~\ref{fig:IF_lambda_surfaces} shows the 
casewise (top row) and cellwise (bottom row)
IF of $\bB$ in a simple linear model for
$\lambda=\{0.001,1,100\}$. The surfaces show 
that the ridge regularization modifies mostly 
the magnitude of the influence functions, in 
accordance with Figure~\ref{fig:IF_lambda}.

\vspace{6mm}
\begin{figure}[H]
\centering

\includegraphics[width=0.32\textwidth]
  {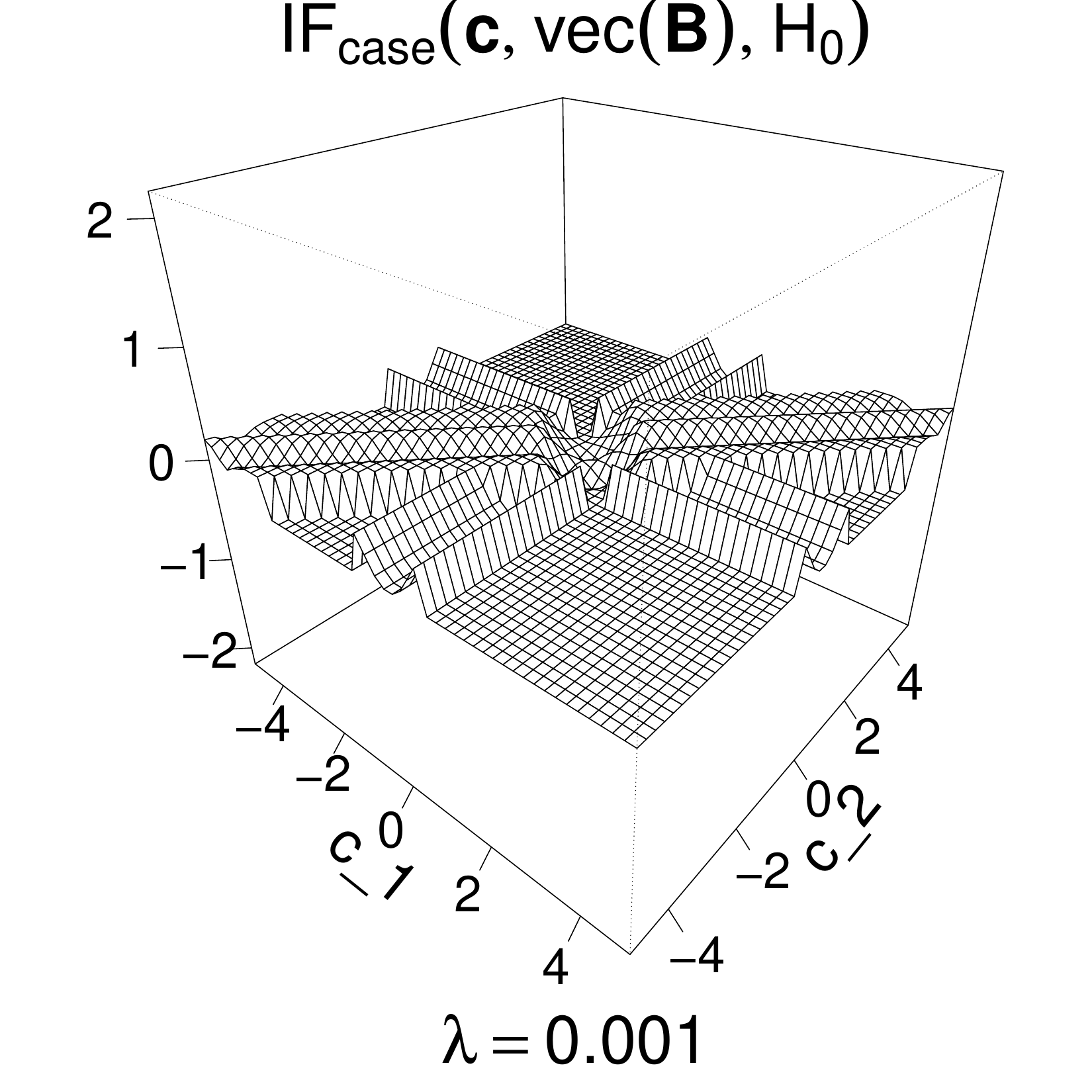}
\includegraphics[width=0.32\textwidth]
  {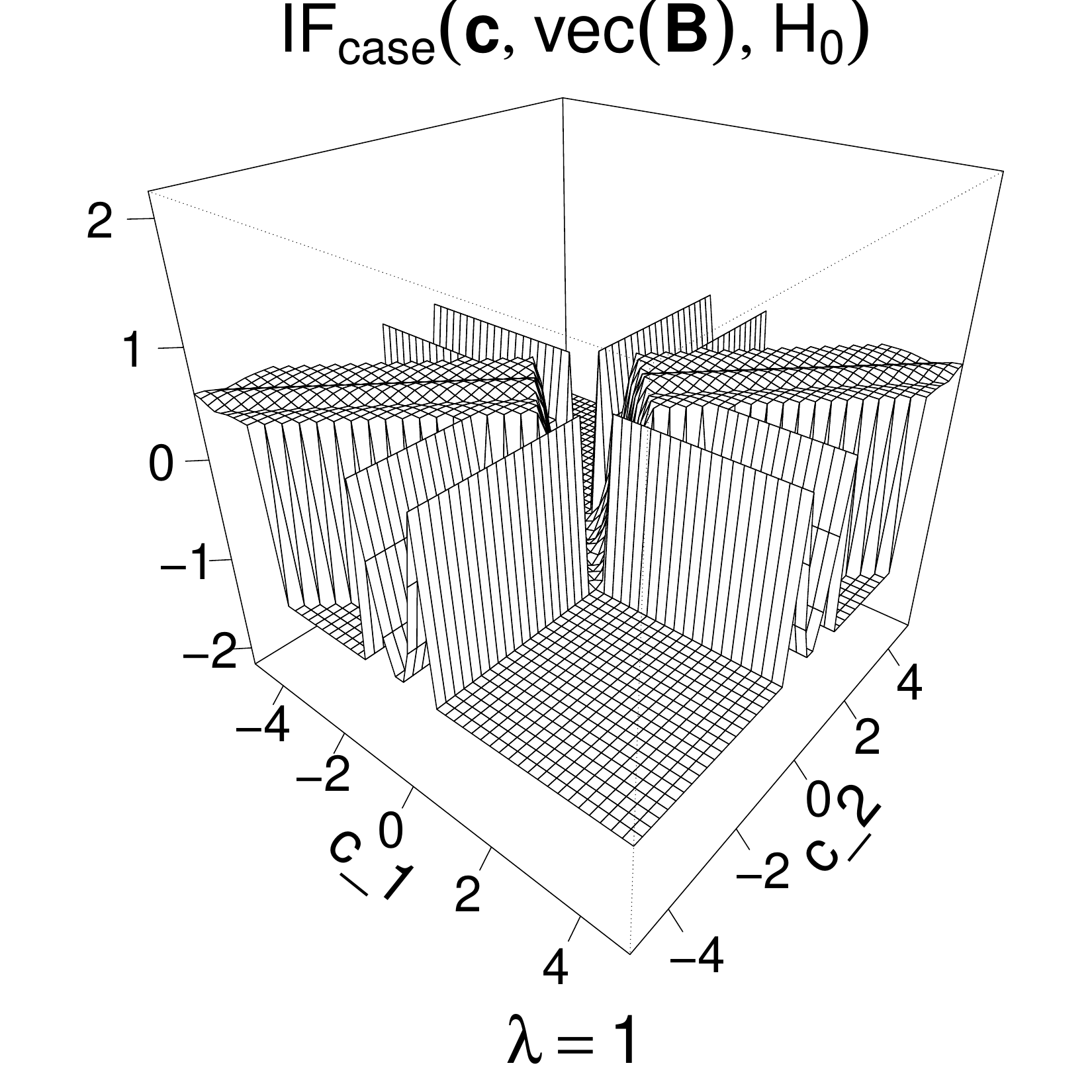}
\includegraphics[width=0.32\textwidth]
  {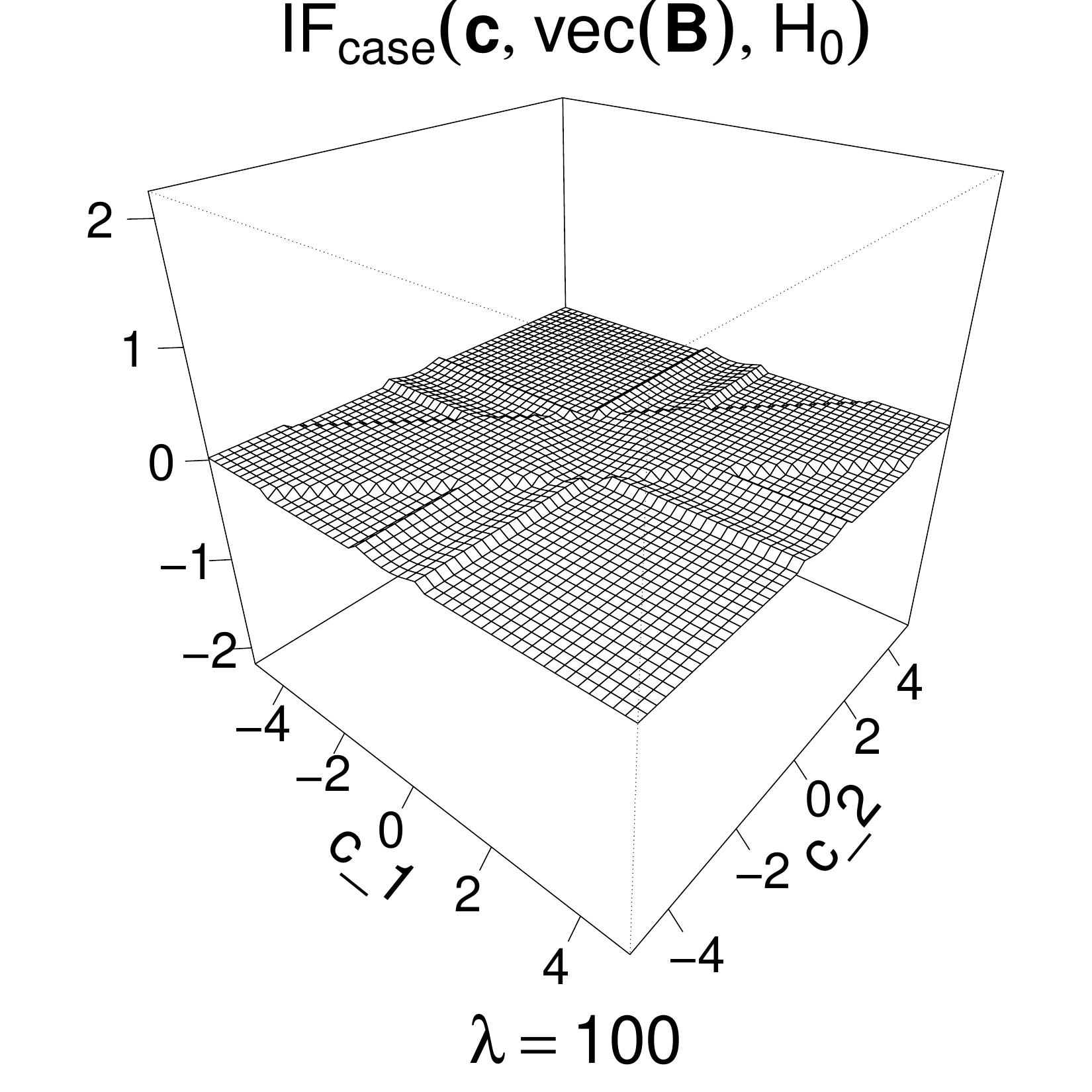}

\vspace{0.3cm}

\includegraphics[width=0.32\textwidth]
  {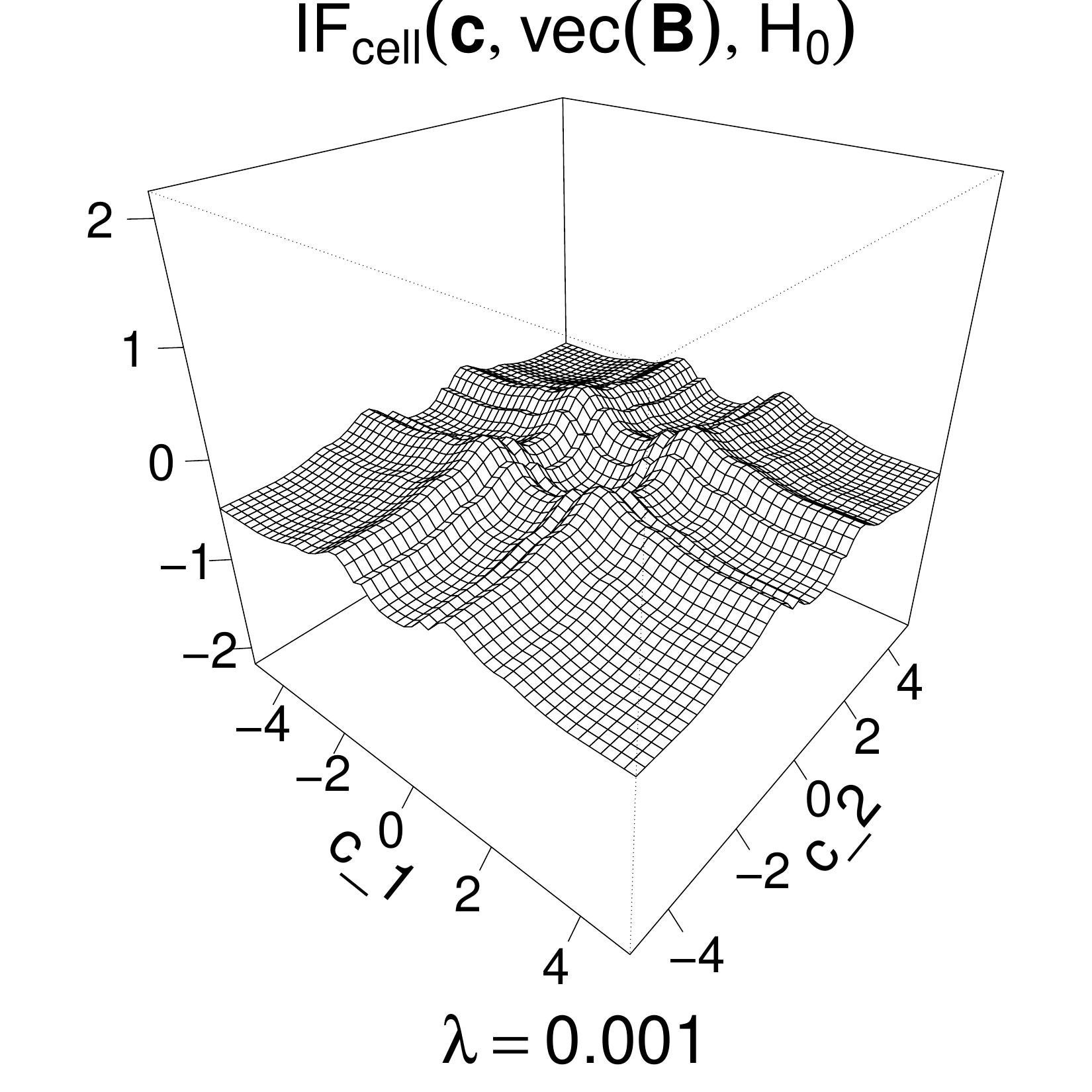}
\includegraphics[width=0.32\textwidth]
  {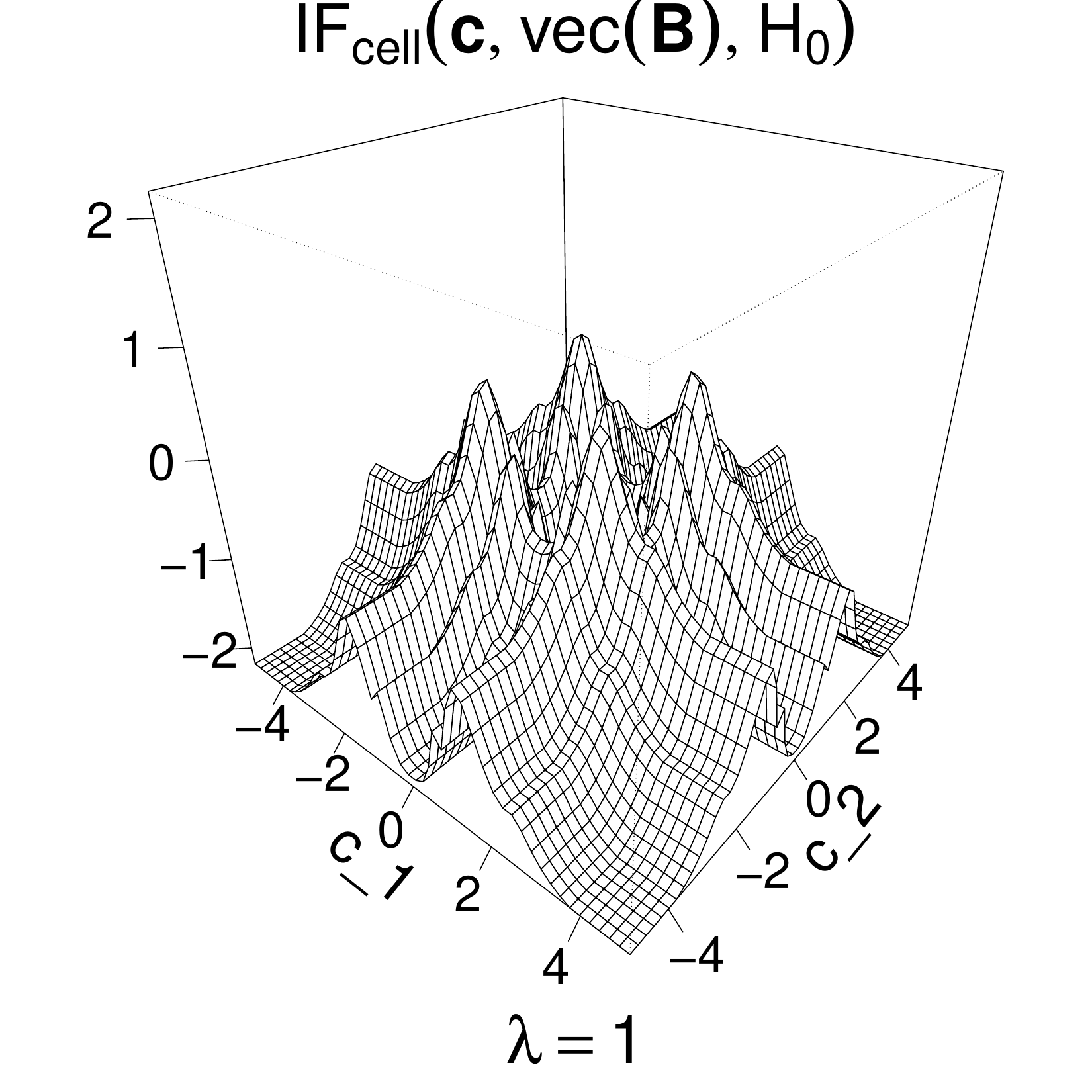}
\includegraphics[width=0.32\textwidth]
  {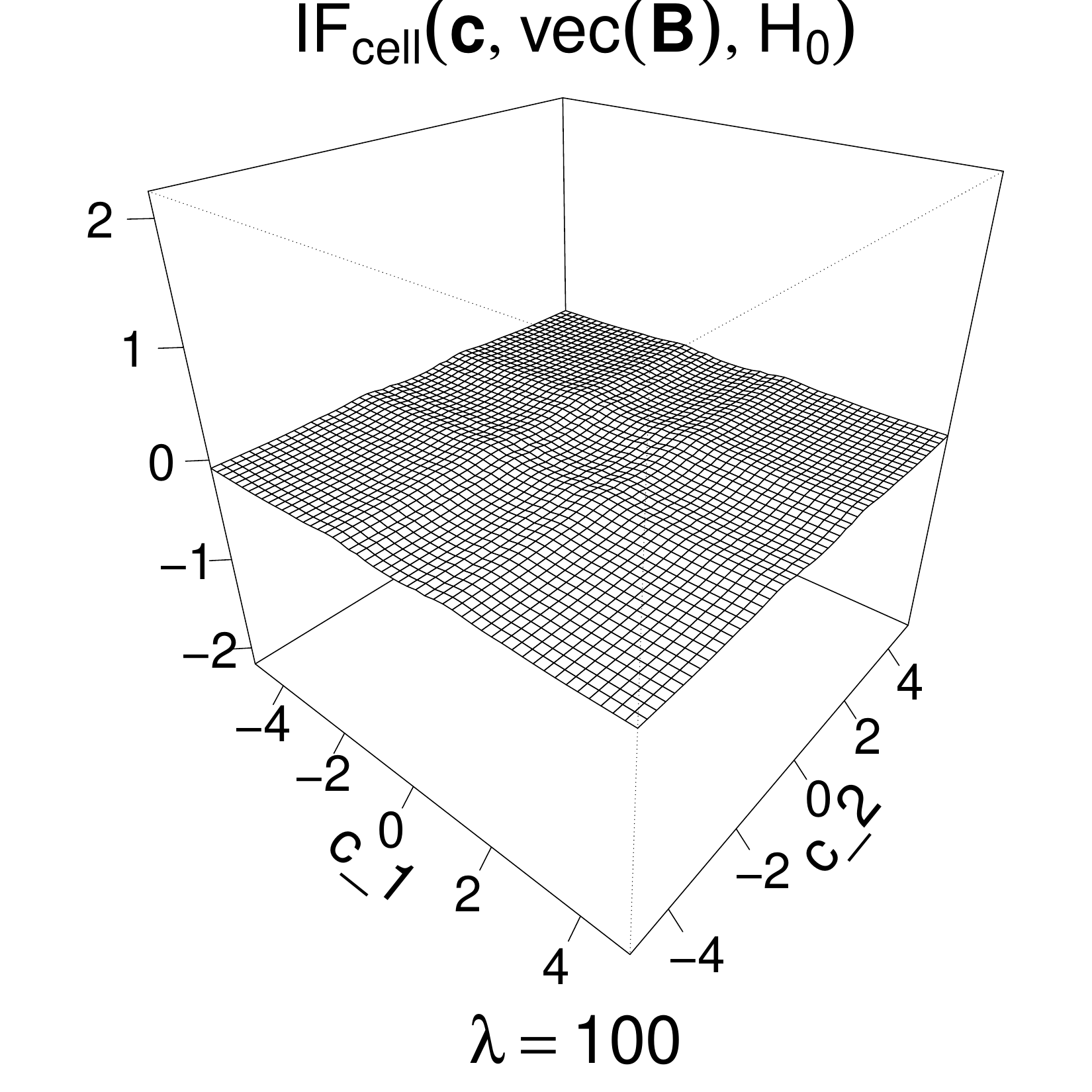}

\caption{The casewise (top row) and cellwise
(bottom row) IF of $\bB$ in a simple linear 
model for $\lambda=\{0.001,1,100\}$.}
\label{fig:IF_lambda_surfaces}
\end{figure}

\newpage

\subsection{Finite-sample casewise and 
cellwise breakdown}
\label{app:cellwise-breakdown}

For $m\in\{0,\ldots,n\}$, the cellwise 
replacement neighborhood was defined as
\begin{equation*}
\mathcal C_m(\bZ)
:=
\left\{
\bZ^{(m)}=\{z_{ij}^{(m)}\}\in\mathbb R^{n\times d}:
\max_{1\leqslant j\leqslant d}
\#
\left\{
i:z_{ij}^{(m)}\neq z_{ij}
\right\}
\leqslant m
\right\}.
\end{equation*}
That is, a matrix in $\mathcal C_m(\bZ)$ is obtained
by replacing at most $m$ cells in each column of 
$\bZ$ by arbitrary values. The cellPCA estimates 
$\bhV$, and  $\bhmu_z$ satisfy
$\bhV^{T}\bhV = 
\bm I_k$ and $\bhV^{T}\bhmu_z=\bzero$. We put
$\bhP:= \bhV\bhV^{T}$. We define $\bhB(\bZ^{(m)})$,
$\bhmu_z(\bZ^{(m)})$, and $\bhV(\bZ^{(m)})$ as the
estimates $\bhB$, $\bhmu_z$, and $\bhV$ obtained 
on the dataset $\bZ^{(m)}$ with $\bhP(\bZ^{(m)}):=
\bhV(\bZ^{(m)})(\bhV(\bZ^{(m)}))^{T}$.

The finite-sample cellwise breakdown values
of the cellMR slope and the cellPCA location and subspace orientation are defined as 
\begin{equation*}
\eps_{n,\mathrm{cell}}^{*}(\bhB,\bZ)
:=
\min_{m\in\{0,\ldots,n\}}
\left\{
\frac{m}{n}:
\sup_{\bZ^{(m)}\in\mathcal C_m(\bZ)}
\left\|
\bhB(\bZ^{(m)})
-
\bhB(\bZ)
\right\|_F
=
+\infty
\right\},
\end{equation*}
\begin{equation*}
\label{eq:def-mmu}
\eps_{n,\mathrm{cell}}^{*}(\bhmu_z,\bZ)
:=
\min_{m\in\{0,\ldots,n\}}
\left\{
\frac{m}{n}:
\sup_{\bZ^{(m)}\in\mathcal C_m(\bZ)}
\left\|
\bhmu_z(\bZ^{(m)})
-
\bhmu_z(\bZ)
\right\|
=
+\infty
\right\},
\end{equation*}
and
\begin{equation*}
\label{eq:def-mP}
\eps_{n,\mathrm{cell}}^{*}(\bhP,\bZ)
:=
\min_{m\in\{0,\ldots,n\}}
\left\{
\frac{m}{n}:
\sup_{\bZ^{(m)}\in\mathcal C_m(\bZ)}
\theta_{\max}
\left(
\bhP(\bZ^{(m)}),
\bhP(\bZ)
\right)
=
\frac{\pi}{2}
\right\}.
\end{equation*}
For two rank-$k$ projection matrices
$\bP_1=\bV_1\bV_1^T$ and $\bP_2=\bV_2\bV_2^T$, 
with orthonormal bases $\bV_1$ and $\bV_2$, their largest principal angle is defined as
\[
\theta_{\max}(\bP_1,\bP_2)
:=
\arccos
\left[
\sigma_{\min}
\left(
\bV_1^T\bV_2
\right)
\right],
\]
where $\sigma_{\min}(\cdot)$ denotes the smallest singular value.

We also define the casewise replacement
neighborhood for $m\in\{0,\ldots,n\}$ as 
\begin{equation*}
\mathcal R_m(\bZ)
:=
\left\{
\bZ^{(m)}=
[\bz^{(m)}_1,\dots,\bz^{(m)}_n]^T\in\mathbb R^{n\times d}:
\#
\left\{
i:\bz_i^{(m)}\neq\bz_i
\right\}
\leqslant m
\right\}.
\end{equation*}
Thus, a matrix in $\mathcal R_m(\bZ)$ is obtained by replacing at
most $m$ complete rows of $\bZ$ by arbitrary vectors in
$\mathbb R^d$.
The finite-sample casewise breakdown points of the cellMR slope, and the cellPCA location and orientation are
\begin{equation*}
\eps_{n,\mathrm{case}}^{*}(\bhB,\bZ)
:=
\min_{m\in\{0,\ldots,n\}}
\left\{
\frac{m}{n}:
\sup_{\bZ^{(m)}\in\mathcal R_m(\bZ)}
\left\|
\bhB(\bZ^{(m)})-\bhB(\bZ)
\right\|_F
=
+\infty
\right\},
\end{equation*}
\begin{equation*}
\label{eq:def-case-mmu}
\eps_{n,\mathrm{case}}^{*}(\bhmu_z,\bZ)
:=
\min_{m\in\{0,\ldots,n\}}
\left\{
\frac{m}{n}:
\sup_{\bZ^{(m)}\in\mathcal R_m(\bZ)}
\left\|
\bhmu_z(\bZ^{(m)})-\bhmu_z(\bZ)
\right\|
=
+\infty
\right\},
\end{equation*}
and
\begin{equation*}
\label{eq:def-case-mP}
\eps_{n,\mathrm{case}}^{*}(\bhP,\bZ)
:=
\min_{m\in\{0,\ldots,n\}}
\left\{
\frac{m}{n}:
\sup_{\bZ^{(m)}\in\mathcal R_m(\bZ)}
\theta_{\max}
\left(
\bhP(\bZ^{(m)}),\bhP(\bZ)
\right)
=
\frac{\pi}{2}
\right\}.
\end{equation*}

The cellwise and casewise breakdown points are related, as established in the following lemma.
\begin{lemma}
\label{lem:case-cell-breakdown-relation}
For
$\bT\in\{\bhB,\bhmu_z,\bhP\}$, $\eps_{n,\mathrm{case}}^{*}(\bT,\bZ)
\geqslant
\eps_{n,\mathrm{cell}}^{*}(\bT,\bZ)$.
\end{lemma}
\begin{proof}
Every matrix in $\mathcal R_m(\bZ)$ differs from 
$\bZ$ in at most $m$ rows. Therefore, in each column,
at most the cells belonging to those $m$ rows have 
been replaced. Hence, each column contains at most 
$m$ replaced cells, which proves
$\mathcal R_m(\bZ) \subseteq \mathcal C_m(\bZ)$.
For any fixed $m$, the supremum appearing in the
definition of casewise breakdown is taken over a 
subset of $ C_m(\bZ)$. Consequently, if $m$ 
replacements are insufficient to cause cellwise breakdown, they are also insufficient to cause
casewise breakdown. The minimum number of casewise replacements required for breakdown is therefore 
at least the minimum number of cellwise 
replacements, proving the lemma.
\end{proof}

For  $\bm z_i'\in\mathbb R^d$, 
$\bmu_z\in\mathbb R^d$, and $\bV\in\mathbb R^{d\times\rk}$, define
\begin{equation*}
L_i
\left(
\bm u;
\bm z_i',
\bmu_z,
\bV
\right)
:=
\sum_{j=1}^{d}
\sigma_{1,j}^{2}
\rho_1
\left(
\frac{
z_{ij}'-\mu_{z,j}-\bu^T\bv_j
}{
\sigma_{1,j}
}
\right),
\end{equation*}
 and
\begin{equation*}
\bhu_i
\left(
\bm z_i';
\bmu_z,\bV
\right)
=
\operatorname*{argmin}_{\bm u\in\mathbb R^k}
L_i
\left(
\bm u;
\bm z_i',
\bmu_z,\bV
\right).
\end{equation*}
Moreover, we use the abbreviations
$\widehat{\bm u}_i(\bmu_z,\bV)
=
\widehat{\bm u}_i
\left(
\bm z_i;
\bmu_z,\bV
\right)$ and $\widehat{\bm u}_i^{(m)}
=
\widehat{\bm u}_i
\left(
\bm z_i^{(m)};
\bhmu_z(\bZ^{(m)}),
\bhV(\bZ^{(m)})
\right)$.
Define, for $\bmu_z\in\mathbb R^d$, and $\bV\in\mathbb R^{d\times\rk}$,
\begin{equation*}
R_i(\bmu_z,\bV)
:=
L_i
\left(
\widehat{\bm u}_i(\bmu_z,\bV);
\bm z_i,\bmu_z,\bV
\right).
\end{equation*}
For $\ell=0,\ldots,d$, define
$S_{\ell}^{-}
:=
\sum_{j=1}^{\ell}\sigma_{1,(j)}^2$ and $S_{\ell}^{+}
:=
\sum_{j=d-\ell+1}^{d}\sigma_{1,(j)}^2$,
with $S_0^-=S_0^+=0$, where $\sigma_{1,(1)}^2
\leqslant
\cdots
\leqslant
\sigma_{1,(d)}^2$.
Moreover,
\begin{equation*}
\label{eq:def-R0}
R_0
:=
\sup_{\substack{
m<n\min\{\eps_{n,\mathrm{cell}}^{*}(\bhmu_z,\bZ),\eps_{n,\mathrm{cell}}^{*}(\bhP,\bZ)\}\\
\bZ^{(m)}\in\mathcal C_m(\bZ)
}}
\max_{1\leqslant i\leqslant n}
R_i
\left(
\bhmu_z(\bZ^{(m)}),
\bhV(\bZ^{(m)})
\right).
\end{equation*}
The quantity $R_0$ is the largest minimum objective value attained by
an original row over all contaminations preceding breakdown of the
location and subspace orientation.
Define
\begin{equation*}
r_{0,s}
:=
\min
\left\{
r\in\{0,\ldots,s\}:
\overline\rho_1S_{s-r}^{-}
\leqslant
R_0+\overline\rho_1S_r^{+}
\right\},
\end{equation*}
 where $\overline\rho_1
:=
\sup_{t\in\mathbb R}\rho_1(t)
<\infty$.
Thus, $r_0$ is the first number of rowwise replacements for which the
lower bound on the objective of an escaping score is no longer
strictly larger than the upper bound attained by a finite comparison
score.
For any nonempty subset $A\subseteq\{1,\ldots,n\}$, 
\[
\overline{\bm u}_{A}^{(m)}
:=
\frac{1}{|A|}
\sum_{i\in A}
\widehat{\bm u}_{i}^{(m)},
\]
and
\begin{equation*}
\bm S_A^{(m)}
:=
\frac{1}{|A|}
\sum_{i\in A}
\left(
\widehat{\bm u}_{i}^{(m)}
-
\overline{\bm u}_{A}^{(m)}
\right)
\left(
\widehat{\bm u}_{i}^{(m)}
-
\overline{\bm u}_{A}^{(m)}
\right)^T.
\end{equation*}
When all cells are observed, the normalizing constant of the
orthogonal residual scatter is
\begin{equation*}
\label{eq:orthogonal-normalizer}
b\left(\bZ^{(m)}\right)
:=
\frac{1}{d^2}
\sum_{i=1}^n
w_i^{\case,(m)}
\left[
\sum_{j=1}^d
w_{ij}^{\cell,(m)}
\right]^2.
\end{equation*}
where $w_i^{\case,(m)}$ and $w_{ij}^{\cell,(m)}$ are, respectively, the casewise and cellwise weights computed from
the contaminated dataset $\bZ^{(m)}$ at  $\left(
\bhmu_z(\bZ^{(m)}),
\bhV(\bZ^{(m)})
\right)$.

\vspace{3mm}
We make five assumptions here.
\begin{enumerate}
\item[(F1)]
The scales
$\sigma_{1,1},\ldots,\sigma_{1,d}, \sigma_{2}$ are fixed and strictly positive, and the data are already standardized and have no missing cells.
\item[(F2)]
There exist an integer $s_0\in\{1,\dots,d\}$ and a $\kappa_0>0$ such that
\[
\inf_{\substack{
\bZ^{(m)}\in\mathcal C_m(\bZ)\\
\|\bm c\|=1
}}
\left|
\bhV(\bZ^{(m)})\bm c
\right|_{[s_0]}
\geqslant\kappa_0,
\]
 for every integer $m$ satisfying
$m<n\eps_{n,\mathrm{cell}}^{*}(\bhmu_z,\bZ)$ and $
m<n\eps_{n,\mathrm{cell}}^{*}(\bhP,\bZ)$. Moreover, 
$
R_0
<
\overline\rho_1 S_{s_0}^{-}$, and put $r_0:=r_{0,s_0}$.
Here, $|\bm a|_{[s]}$ denotes the $s$-th largest value among
$|a_1|,\ldots,|a_d|$.
\item[(F3)]
There exists a constant $\gamma_0>0$ such that
\[
\inf_{\substack{
\bZ^{(m)}\in\mathcal C_m(\bZ)\\
I\subseteq\{1,\ldots,n\},\ |I|=k+1\\
r_i(\bZ^{(m)})<r_0\ \text{for every }i\in I
}}
\lambda_{\min}
\left(
\bS_I^{(m)}
\right)
\geqslant\gamma_0,
\]
for every integer $m$ satisfying
$m<n\eps_{n,\mathrm{cell}}^{*}(\bhmu_z,\bZ)$ and $
m<n\eps_{n,\mathrm{cell}}^{*}(\bhP,\bZ)$,
where $r_i(\bZ^{(m)})$ is the number of replaced cells in row $i$.
\item[(F4)]
The MCD score covariance 
$\bhSigma_{\text{\tiny MCD}}^{(m)}(\bhu_i)$ 
computed on $\bZ^{(m)}$ has parameter $h$ with
$k<h\leqslant n$ and $2h\geqslant n+k+1$.
\item[(F5)]
There exists a constant $b_0>0$ such that
\[
\inf_{\bZ^{(m)}\in\mathcal C_m(\bZ)}
b\left(\bZ^{(m)}\right)
\geqslant b_0,
\]
for every integer $m$ satisfying
$m<n\eps_{n,\mathrm{cell}}^{*}(\bhmu_z,\bZ)$ and $
m<n\eps_{n,\mathrm{cell}}^{*}(\bhP,\bZ)$.
\end{enumerate}

Assumption~(F1) places the analysis in the fixed-scale setting
and excludes missing cells. It therefore isolates the breakdown of
the cellPCA geometry and the resulting cellMR slope from that of the
preliminary scale estimators.
Assumption~(F2) requires every score direction to affect at
least $s_0$ fitted coordinates uniformly before breakdown of the
cellPCA location or orientation. The inequality
$R_0<\overline\rho_1S_{s_0}^{-}$ ensures that a bounded comparison
score remains preferable to an escaping score, while $r_0$ is the
minimum number of replacements within a row that can remove this
margin. This condition is plausible for nondegenerate loading spaces
that are not concentrated on very few variables.
Assumption~(F3) requires any $k+1$ scores from rows containing
fewer than $r_0$ replacements to remain in general position.
Assumption~(F4) then ensures that every MCD subset contains at least
$k+1$ such scores. These conditions prevent the MCD score scatter
from becoming degenerate and are compatible with the usual
high-breakdown choice
$h\geqslant\lceil(n+k+1)/2\rceil$.
Assumption~(F5) keeps the normalizing constant of the
orthogonal residual scatter bounded away from zero.

\begin{proposition}
\label{prop:slope-breakdown-mu-P}
Under the assumptions F1-F5,
\begin{equation}
\eps_{n,\mathrm{case}}^{*}
\left(
\bhB,\bZ
\right)\geqslant\eps_{n,\mathrm{cell}}^{*}
\left(
\bhB,\bZ
\right)
\geqslant
\min\Bigg\{
\eps_{n,\mathrm{cell}}^{*}
\left(
\bhP,\bZ
\right),
\eps_{n,\mathrm{cell}}^{*}
\left(
\bhmu_z,\bZ
\right),
\frac{1}{n}
\left\lceil
\frac{r_0(n-h+1)}{d}
\right\rceil
\Bigg\}.
\label{eq:slope-breakdown-mu-P-fraction}
\end{equation}
\end{proposition}

\begin{proof}
Consider a contamination level $m<n\min\{\eps_{n,\mathrm{cell}}^{*}
\left(
\bhP,\bZ
\right),
\eps_{n,\mathrm{cell}}^{*}
\left(
\bhmu_z,\bZ
\right)\}$. Then neither
$\bhmu_z$ nor $\bhP$ has broken down.

We first show that fewer than $r_0$ replacements in row $i$ cannot
make its estimated score unbounded. Suppose, contrary to this claim,
that there exists a sequence
$\left\{
\bm Z_{\ell}^{(m)}
\right\}_{\ell\geqslant1}
\subset\mathcal C_m(\bZ)$ 
such that, for every $\ell$, fewer than $r_0$ cells are replaced in
row $i$, while the corresponding estimated scores satisfy
$\left\|
\bhu_{i,\ell}^{(m)}
\right\|
\rightarrow\infty$, for $\ell\rightarrow\infty$,
where
$\bhu_{i,\ell}^{(m)}
=\bhu_i
\left(
\bz_{i,\ell}^{(m)},
\bmu_{\ell}^{(m)},
\bV_{\ell}^{(m)}
\right)$, with $\bm z_{i,\ell}^{(m)}$ are the rows of $\bm Z_{\ell}^{(m)}$, and
$\bmu_{\ell}^{(m)}
=
\bhmu_z(\bZ_{\ell}^{(m)})$,
$\bV_{\ell}^{(m)}
=
\bhV(\bZ_{\ell}^{(m)})$.
For each $\ell$, let
$\bm u_{i,\ell}^{0}
=
\bhu_i
\left(
\bz_i;
\bmu_{\ell}^{(m)},
\bV_{\ell}^{(m)}
\right)$
be the comparison score obtained by fitting the original row
$\bz_i$ using $\bmu_{\ell}^{(m)}$ and $\bV_{\ell}^{(m)}$, and define
\[
J_{i,\ell}
=
\left\{
j:z_{ij,\ell}^{(m)}\neq z_{ij}
\right\},
\qquad
r_{i,\ell}=|J_{i,\ell}|<r_0.
\]
Note that by definition $r_0\leqslant s_0$.
For every $\ell$, the definition of $R_i$ and $R_0$ gives
\begin{equation}
\label{eq:clean-comparison-score}
L_i
\left(
\bm u_{i,\ell}^{0};
\bz_i,
\bmu_{\ell}^{(m)},
\bV_{\ell}^{(m)}
\right)
=
R_i
\left(
\bmu_{\ell}^{(m)},
\bV_{\ell}^{(m)}
\right)\leqslant R_0.
\end{equation}

Although the sets $J_{i,\ell}$ may vary with $\ell$, they take values
in the finite collection of subsets of $\{1,\ldots,d\}$. Therefore,
there exist a set $J_i\subseteq\{1,\ldots,d\}$ and a strictly
increasing sequence of indices
$\ell_1<\ell_2<\cdots$
such that
$J_{i,\ell_t}=J_i$ for every $t$.
Consequently,
$r_{i,\ell_t}=r_i:=|J_i|<r_0$ for every $t$.
Moreover, because the score norms diverge along the original
sequence, they also diverge along this subsequence, i.e., 
$\left\|
\bhu_{i,\ell_t}^{(m)}
\right\|
\rightarrow\infty$ as $t\rightarrow\infty$.
For notational simplicity, we now relabel
this subsequence by $\ell$, so that
$J_{i,\ell}=J_i$ and
$r_{i,\ell}=r_i<r_0$.

Evaluating the
contaminated row objective at $\bm u_{i,\ell}^0$ gives
\begin{align*}
L_i
\left(
\bm u_{i,\ell}^0;
\bz_{i,\ell}^{(m)},
\bmu_{\ell}^{(m)},
\bV_{\ell}^{(m)}
\right)
&=
\sum_{j\notin J_i}
\sigma_{1,j}^{2}
\rho_1
\left(
\frac{
z_{ij}-\mu_{j,\ell}^{(m)}
-(\bu_{i,\ell}^0)^T\bv_{j,\ell}^{(m)}
}{\sigma_{1,j}}
\right)
\\
&\quad+
\sum_{j\in J_i}
\sigma_{1,j}^{2}
\rho_1
\left(
\frac{
z_{ij,\ell}^{(m)}-\mu_{j,\ell}^{(m)}
-(\bu_{i,\ell}^0)^T\bv_{j,\ell}^{(m)}
}{\sigma_{1,j}}
\right).
\end{align*}
The first sum is bounded above by the
complete objective for the original row evaluated at
$\bm u_{i,\ell}^{0}$, which by
\eqref{eq:clean-comparison-score},  is at most $R_0$.
 Moreover, 
\[
\sum_{j\in J_i}
\sigma_{1,j}^{2}
\rho_1
\left(
\frac{
z_{ij,\ell}^{(m)}-\mu_{j,\ell}^{(m)}
-(\bu_{i,\ell}^0)^T\bv_{j,\ell}^{(m)}
}{\sigma_{1,j}}
\right)
\leqslant
\overline\rho_1
\sum_{j\in J_i}\sigma_{1,j}^{2}
\leqslant
\overline\rho_1S_{r_i}^{+}.
\]
Consequently, since $\bhu_{i,\ell}^{(m)}$ is a 
minimizer,
\begin{equation}
L_i
\left(
\bhu_{i,\ell}^{(m)};
\bz_{i,\ell}^{(m)},
\bmu_{\ell}^{(m)},
\bV_{\ell}^{(m)}
\right)
\leqslant
L_i
\left(
\bm u_{i,\ell}^0;
\bz_{i,\ell}^{(m)},
\bmu_{\ell}^{(m)},
\bV_{\ell}^{(m)}
\right)
\leqslant
R_0+\overline\rho_1S_{r_i}^{+}.
\label{eq:bounded-score-candidate}
\end{equation}

\vspace{-3mm}
We now obtain a lower bound for the objective along the supposed
diverging sequence. 
For large $\ell$, define the normalized directions 
$\bm c_{i,\ell}
:=
\bhu_{i,\ell}^{(m)}
/
\left\|
\bhu_{i,\ell}^{(m)}
\right\|
$.
 Assumption~(F2) therefore gives
\[
\left|
\bV_{\ell}^{(m)}\bm c_{i,\ell}
\right|_{[s_0]}
\geqslant
\kappa_0.
\]
By the definition of the $s_0$th largest absolute coordinate, there
exists a set
$A_{\ell}\subseteq\{1,\ldots,d\}$, with
$|A_{\ell}|\geqslant s_0$, such that
$\left|
(\bv_{j,\ell}^{(m)})^T\bm c_{i,\ell}
\right|
\geqslant
\kappa_0$ for every $j\in A_{\ell}$.
Consequently,
\begin{align*}
\min_{j\in A_{\ell}}
\left|
(\bv_{j,\ell}^{(m)})^T
\bhu_{i,\ell}^{(m)}
\right|
=
\left\|
\bhu_{i,\ell}^{(m)}
\right\|
\min_{j\in A_{\ell}}
\left|
(\bv_{j,\ell}^{(m)})^T
\bm c_{i,\ell}
\right|
\geqslant
\kappa_0
\left\|
\bhu_{i,\ell}^{(m)}
\right\|
\rightarrow\infty.
\end{align*}
Thus, for every $\ell$, at least $s_0$ fitted coordinates diverge
with the score norm.

At most $r_i$ indices in $A_{\ell}$ can belong to $J_i$. Therefore,
the set
$A_{\ell}^{0}=A_{\ell}\setminus J_i$
contains at least $s_0-r_i>0$ indices corresponding to unchanged
cells.
Since $\bm z_i$ is fixed and
$\{\bmu_{\ell}^{(m)}\}_{\ell\geqslant1}$ is bounded, we have
\begin{equation*}
\left|
z_{ij}
-\mu_{j,\ell}^{(m)}
-(\bv_{j,\ell}^{(m)})^T
\bhu_{i,\ell}^{(m)}
\right|
\geqslant
\kappa_0
\left\|
\bhu_{i,\ell}^{(m)}
\right\|
-|z_{ij}|
-|\mu_{j,\ell}^{(m)}|
\rightarrow\infty
\end{equation*}
uniformly over $j\in A_{\ell}^{0}$. It follows that
\begin{equation}
\liminf_{\ell\rightarrow\infty}
L_i
\left(
\bhu_{i,\ell}^{(m)};
\bz_{i,\ell}^{(m)},
\bmu_{\ell}^{(m)},
\bV_{\ell}^{(m)}
\right)
\geqslant
\overline\rho_1
\liminf_{\ell\rightarrow\infty}
\sum_{j\in A_{\ell}^{0}}\sigma_{1,j}^{2}
\geqslant
\overline\rho_1S_{s_0-r_i}^{-}.
\label{eq:diverging-score-lower-bound}
\end{equation}

By definition, $r_0$ is the smallest integer
$r\in\{0,\ldots,s_0\}$ for which
$\overline\rho_1S_{s_0-r}^{-}
\leqslant
R_0+\overline\rho_1S_r^{+}$.
Therefore, no integer $r<r_0$ satisfies this inequality. Since
$r_i<r_0$, we necessarily have
$\overline\rho_1S_{s_0-r_i}^{-}
>
R_0+\overline\rho_1S_{r_i}^{+}$.
On the other hand,
\eqref{eq:bounded-score-candidate} and
\eqref{eq:diverging-score-lower-bound} imply
\[
\overline\rho_1S_{s_0-r_i}^{-}
\leqslant
\liminf_{\ell\rightarrow\infty}
L_i
\left(
\bhu_{i,\ell}^{(m)};
\bz_{i,\ell}^{(m)},
\bmu_{\ell}^{(m)},
\bV_{\ell}^{(m)}
\right)
\leqslant
R_0+\overline\rho_1S_{r_i}^{+},
\]
which is a contradiction.
This contradiction rules out the diverging
subsequence extracted above.  Hence, the contradiction also rules out the original
sequence. We have therefore shown that
making the score of row $i$ unbounded
requires at least $r_0$ replaced cells in that row.

For each row $i$, the preceding result gives
\[
M_i
:=
\sup_{\substack{
\bZ^{(m)}\in\mathcal C_m(\bZ)\\
r_i(\bZ^{(m)})<r_0
}}
\left\|
\bhu_{i}^{(m)}
\right\|
<\infty.
\]
Thus, we may define
$M_0
=
\max_{1\leqslant i\leqslant n}M_i
<\infty$.
Suppose that fewer than $n-h+1$ rows contain at least $r_0$
replaced cells. Then at most $n-h$ rows contain at least $r_0$
replacements. Consequently, at least $h$ rows contain fewer than
$r_0$ replacements and have score norms bounded by $M_0$.
For each contaminated dataset $\bZ^{(m)}$, select any $h$ such rows
and denote their index set by
$H_0^{(m)}
=
H_0\left(\bZ^{(m)}\right)$.
Although $H_0^{(m)}$ may vary with the contaminated dataset, every
score indexed by $H_0^{(m)}$ is bounded by the same constant $M_0$.
For this subset,
$\|
\overline{\bu}_{H_0^{(m)}}^{(m)}
\|
\leqslant M_0$,
and hence
\[
\|
\bhu_{i}^{(m)}
-
\overline{\bu}_{H_0^{(m)}}^{(m)}
\|
\leqslant 2M_0,
\]
for every $i\in H_0^{(m)}$.
It follows that
\[
\|
\bS_{H_0^{(m)}}^{(m)}
\|_2
\leqslant
\frac{1}{h}
\sum_{i\in H_0^{(m)}}
\|
\bhu_{i}^{(m)}
-
\overline{\bu}_{H_0^{(m)}}^{(m)}
\|^2
\leqslant 4M_0^2,
\]
where $\|\cdot\|_2$ denotes the spectral norm of a matrix.
For
every $\bm x\in\mathbb R^k$,
\begin{align*}
\bm x^T
\bS_{H_0^{(m)}}^{(m)}
\bm x
&=
\frac{1}{h}
\sum_{i\in H_0^{(m)}}
\left[
\bm x^T
\left(
\bhu_i^{(m)}
-
\overline{\bu}_{H_0^{(m)}}^{(m)}
\right)
\right]^2
\geqslant0.
\end{align*}
Hence, $\bS_{H_0^{(m)}}^{(m)}$ is positive semidefinite. Its
eigenvalues are therefore nonnegative and bounded above by its
spectral norm, so
\[
\det
(
\bS_{H_0^{(m)}}^{(m)}
)
\leqslant
\|
\bS_{H_0^{(m)}}^{(m)}
\|_2^k
\leqslant
\left(4M_0^2\right)^k
=:C_0.
\]
Let
$H_{\mathrm{MCD}}^{(m)}
=
H_{\mathrm{MCD}}\left(\bZ^{(m)}\right)$
denote the subset selected by the MCD estimator. Since the MCD
minimizes the determinant over all subsets of size $h$,
\begin{equation}
    \label{eq:boundC0mcd}
\det
(\bS_{H_{\mathrm{MCD}}^{(m)}}^{(m)}
)
\leqslant
\det
(
\bS_{H_0^{(m)}}^{(m)}
)
\leqslant C_0.
\end{equation}
Since at most $n-h$ rows in the complete dataset contain at least
$r_0$ replacements, the MCD subset
$H_{\mathrm{MCD}}^{(m)}$, which contains $h$ rows, contains at least
$h-(n-h)=2h-n$
rows having fewer than $r_0$ replacements. By (F4),
$2h\geqslant n+k+1$, thus, we have
$2h-n\geqslant k+1$.
Therefore, we can select a subset
$I\subseteq H_{\mathrm{MCD}}^{(m)}$,
$|I|=k+1$,
such that
$r_i(\bZ^{(m)})<r_0$ for every $i\in I$.
By (F3), $\lambda_{\min}
\left(
\bS_I^{(m)}
\right)
\geqslant\gamma_0$.
The decomposition of the scatter around the corresponding subset
means gives
\begin{align*}
h\bS_{H_{\mathrm{MCD}}^{(m)}}^{(m)}
&=
\sum_{i\in I}
\left(
\bhu_{i}^{(m)}
-
\overline{\bu}_{I}^{(m)}
\right)
\left(
\bhu_{i}^{(m)}
-
\overline{\bu}_{I}^{(m)}
\right)^T
\\
&\quad+
(k+1)
\left(
\overline{\bu}_{I}^{(m)}
-
\overline{\bu}_{H_{\mathrm{MCD}}^{(m)}}^{(m)}
\right)
\left(
\overline{\bu}_{I}^{(m)}
-
\overline{\bu}_{H_{\mathrm{MCD}}^{(m)}}^{(m)}
\right)^T
\\
&\quad+
\sum_{i\in H_{\mathrm{MCD}}^{(m)}\setminus I}
\left(
\bhu_{i}^{(m)}
-
\overline{\bu}_{H_{\mathrm{MCD}}^{(m)}}^{(m)}
\right)
\left(
\bhu_{i}^{(m)}
-
\overline{\bu}_{H_{\mathrm{MCD}}^{(m)}}^{(m)}
\right)^T.
\end{align*}
The last two terms in the preceding decomposition are positive
semidefinite. It follows that the matrix
$\bS_{H_{\mathrm{MCD}}^{(m)}}^{(m)}
-
\frac{k+1}{h}\bS_I^{(m)}$
is positive semidefinite. Therefore, for every unit vector
$\bm a\in\mathbb R^k$,
\[
\bm a^T
\bS_{H_{\mathrm{MCD}}^{(m)}}^{(m)}
\bm a
\geqslant
\frac{k+1}{h}
\bm a^T\bS_I^{(m)}\bm a.
\]
Since $\bS_I^{(m)}$ is symmetric, then
$
\lambda_{\min}
\left(
\bS_I^{(m)}
\right)
=
\min_{\|\bm a\|=1}
\bm a^T\bS_I^{(m)}\bm a$.
Hence, for every unit vector $\bm a\in\mathbb R^k$,
\[
\bm a^T\bS_I^{(m)}\bm a
\geqslant
\lambda_{\min}
\left(
\bS_I^{(m)}
\right)
\geqslant\gamma_0.
\]
 Consequently,
\[
\lambda_{\min}
\left(
\bS_{H_{\mathrm{MCD}}^{(m)}}^{(m)}
\right)
\geqslant
\frac{k+1}{h}\gamma_0
=:\gamma_h>0.
\]
Let
$0<
\lambda_1^{(m)}
\leqslant
\cdots
\leqslant
\lambda_k^{(m)}$
denote the eigenvalues of
$\bS_{H_{\mathrm{MCD}}^{(m)}}^{(m)}$, then
\begin{equation*}
\det
\left(
\bS_{H_{\mathrm{MCD}}^{(m)}}^{(m)}
\right)
=
\lambda_k^{(m)}
\prod_{j=1}^{k-1}\lambda_j^{(m)}
\geqslant
\lambda_k^{(m)}\gamma_h^{\,k-1}.
\end{equation*}
Combining this inequality with \eqref{eq:boundC0mcd}
gives
\[
\lambda_{\max}
\left(
\bS_{H_{\mathrm{MCD}}^{(m)}}^{(m)}
\right)
\leqslant
\frac{C_0}{\gamma_h^{\,k-1}}.
\]
We have therefore proved that, whenever fewer than $n-h+1$ rows
contain at least $r_0$ replaced cells, the raw MCD score scatter
cannot explode. The MCD score covariance $\bhSigma_{\text{\tiny MCD}}^{(m)}(\bhu_i)$ computed on $\bZ^{(m)}$ is obtained by multiplying the raw
scatter $\bS_{H_{\mathrm{MCD}}^{(m)}}^{(m)}$ by its consistency and finite-sample correction factors.
Since these correction factors depend only on the fixed quantities
$n$, $h$, and $k$ and are finite, they do not affect boundedness.
Consequently, $\bhSigma_{\text{\tiny MCD}}^{(m)}(\bhu_i)$ cannot explode either.
Equivalently, explosion of the MCD
score scatter requires at least $n-h+1$ rows to contain at least
$r_0$ replaced cells. 
Every matrix in $\mathcal C_m(\bZ)$ contains at most $md$ replaced
cells in total. If at least $n-h+1$ rows  contain at least
$r_0$ replaced cells, then the total number of replaced cells would be at
least $r_0(n-h+1)$.
Therefore, this necessary condition for explosion cannot be
satisfied when
\[
m
<
\left\lceil
\frac{r_0(n-h+1)}{d}
\right\rceil.
\]

We next consider
$\btSigma_{\zort}^{(m)}$, defined as in \eqref{eq:Sigmay} for the
contaminated dataset $\bZ^{(m)}$. Since all cells are observed, the
diagonal entries of $\btW_i^{(m)}$ are
$w_{ij}^{\cell,(m)}$. Therefore,
\begin{align*}
\left\|
\btW_i^{(m)}(\bz_i^{(m)}-\bhz_i^{(m)})
\right\|^2
=
\sum_{j=1}^d
\sigma_{1,j}^2
\psi_1^2
\left(
\frac{
z_{ij}^{(m)}-\widehat z_{ij}^{(m)}
}{
\sigma_{1,j}
}
\right)
\leqslant
\|\psi_1\|_\infty^2
\sum_{j=1}^d\sigma_{1,j}^2,
\end{align*}
where
$\bhz_i^{(m)}=\bmu^{(m)}
+\bV^{(m)}\bhu_i^{(m)}$ and
$\btW_i^{(m)}$ is the diagonal matrix whose diagonal entries are the
cellwise weights $w_{ij}^{\cell,(m)}$. Then
\begin{align*}
&
\left\|
\sum_{i=1}^n
w_i^{\case,(m)}
\btW_i^{(m)}(\bz_i^{(m)}-\bhz_i^{(m)})
(\bz_i^{(m)}-\bhz_i^{(m)})^T\btW_i^{(m)}
\right\|_2
\\
&\qquad\leqslant
\sum_{i=1}^n
w_i^{\case,(m)}
\left\|
\btW_i^{(m)}(\bz_i^{(m)}-\bhz_i^{(m)})
\right\|^2
\\
&\qquad\leqslant
n
\|\psi_1\|_\infty^2
\sum_{j=1}^d\sigma_{1,j}^2.
\end{align*}
From (F5), it follows that
\[
\|
\btSigma_{\zort}^{(m)}
\|_2
\leqslant
\frac{
nC_{\mathrm{case}}
\|\psi_1\|_\infty^2
\sum_{j=1}^d\sigma_{1,j}^2
}{
b_0
}
<\infty.
\]
Therefore, $\btSigma_{\zort}^{(m)}$ remains 
bounded.

The  cellCov estimate is
\[
\bhSigma^{(m)}
=
\bhV^{(m)}
\bhSigma_{\text{\tiny MCD}}^{(m)}(\bhu_i)
(\bhV^{(m)})^{T}
+
 \btSigma_{\zort}^{(m)}.
\]
Since
$(\bhV^{(m)})^{T}\bhV^{(m)}=\bm I_k$,
\[
\|
\bhV^{(m)}
\bhSigma_{\text{\tiny MCD}}^{(m)}(\bhu_i)
(\bhV^{(m)})^{T}
\|_2
=
\|
\bhSigma_{\text{\tiny MCD}}^{(m)}(\bhu_i)
\|_2.
\]
Consequently,
\begin{equation*}
\label{eq:joint-scatter-bound}
\|
\bhSigma^{(m)}
\|_2
\leqslant
\|
\bhSigma_{\text{\tiny MCD}}^{(m)}(\bhu_i)
\|_2
+
\|
 \btSigma_{\zort}^{(m)}
\|_2.
\end{equation*}

Finally, since
$\lambda\geqslant\lambda_0>0$ and
$\bhSigma_x^{(m)}$ is positive semidefinite,
\[
\|
\left(
\bhSigma_x^{(m)}+\lambda\bm I_p
\right)^{-1}
\|_2
\leqslant
\frac{1}{\lambda_0}.
\]
It follows that
\begin{align*}
\|
\bhB^{(m)}
\|_F
&\leqslant
\|
\left(
\bhSigma_{x}^{(m)}
+
\lambda\bm I_p
\right)^{-1}
\|_2
\|
\bhSigma_{xy}^{(m)}
\|_F\\
&\leqslant
\frac{1}{\lambda_0}
\|
\bhSigma_{xy}^{(m)}
\|_F\\
&\leqslant
\frac{1}{\lambda_0}
\|
\bhSigma^{(m)}
\|_F\\
&\leqslant
\frac{1}{\lambda_0}
\sqrt{d}\|
\bhSigma^{(m)}
\|_2.
\end{align*}
 Since
$\widehat{\bB}(\bZ)$ is fixed,
\[
\sup_{\bZ^{(m)}\in\mathcal C_m(\bZ)}
\left\|
\widehat{\bB}(\bZ^{(m)})
-
\widehat{\bB}(\bZ)
\right\|_F
\leqslant
\sup_{\bZ^{(m)}\in\mathcal C_m(\bZ)}
\left\|
\widehat{\bB}(\bZ^{(m)})
\right\|_F
+
\left\|
\widehat{\bB}(\bZ)
\right\|_F
<\infty.
\]
Therefore, the cellMR slope cannot break down  for every integer $m$ satisfying
\[
m<n\eps_{n,\mathrm{cell}}^{*}(\bhP,\bZ),
\qquad
m<n\eps_{n,\mathrm{cell}}^{*}(\bhmu_z,\bZ),
\]
and
\[
m<
\left\lceil
\frac{r_0(n-h+1)}{d}
\right\rceil,
\]
 which, by using Lemma~\ref{lem:case-cell-breakdown-relation}, proves
\eqref{eq:slope-breakdown-mu-P-fraction}.
\end{proof}

Define
\[
\mathcal H_k
:=
\left\{
(\bmu_z,\bP):
\bP^T=\bP,\ 
\bP^2=\bP,\ 
\operatorname{rank}(\bP)=k,\ 
\bP\bmu_z=\bzero
\right\}.
\]
For $(\bmu_z,\bP)\in\mathcal H_k$, let $\bV_{\bP}$  so that
$\bV_{\bP}^T\bV_{\bP}=\bm I_k$, and $\bP=\bV_{\bP}\bV_{\bP}^T$.
The profiled cellPCA objective is
\begin{align*}
Q_{\bZ}(\bmu_z,\bP)
&:=
\inf_{\bU\in\mathbb R^{n\times k}}
L^{\sigma}_{\rho_1,\rho_2}
\left(
\bZ,\bV_{\bP},\bU,\bmu_z
\right).
\label{eq:profiled-cellpca-objective}
\end{align*}
where $L^{\sigma}_{\rho_1,\rho_2}$ is as in \eqref{eq:objP} with  scales fixed to $\sigma_{1,1},\ldots,\sigma_{1,d}$, and $\sigma_{2}$.
Let
\[
(\bhmu_z,\bhP)
=
\operatorname*{argmin}_{(\bmu_z,\bP)\in\mathcal H_k}
Q_{\bZ}(\bmu_z,\bP),
\]
and define
\begin{equation*}
\label{eq:def-omega-W}
\omega_j
:=
\rho_2
\left(
\frac{\sigma_{1,j}}{\sigma_2}
\sqrt{
\frac{\overline\rho_1}{d}
}
\right),
\qquad
W
:=
\sum_{j=1}^{d}\omega_j.
\end{equation*}
Define the orientation separation by
\begin{equation*}
\label{eq:delta-P}
\Delta_{\bP}(\bZ)
:=
\lim_{\eta\downarrow0}
\inf_{\substack{
(\bmu_z,\bP)\in\mathcal H_k\\
\theta_{\max}(\bP,\bhP)
\geqslant
\frac{\pi}{2}-\eta
}}
\left\{
Q_{\bZ}(\bmu_z,\bP)
-
Q_{\bZ}(\bhmu_z,\bhP)
\right\}.
\end{equation*}
and define the location separation by
\begin{equation*}
\label{eq:delta-mu}
\Delta_{\bmu}^{\infty}(\bZ)
:=
\liminf_{R\rightarrow\infty}
\inf_{\substack{
(\bmu_z,\bP)\in\mathcal H_k\\
\|\bmu_z-\bhmu_z\|\geqslant R
}}
\left\{
Q_{\bZ}(\bmu_z,\bP)
-
Q_{\bZ}(\bhmu_z,\bhP)
\right\}.
\end{equation*}
Thus, $\Delta_{\bP}(\bZ)$ measures the objective separation from
subspaces approaching maximum angular distance from $\bhP$, whereas
$\Delta_{\bmu}^{\infty}(\bZ)$ measures the objective separation
from locations escaping to infinity.

\begin{proposition}
\label{pro:breakdown-mu-P}
Under conditions F1--F5, the 
finite-sample cellwise breakdown values
of the cellPCA location and subspace satisfy
\begin{equation}
\label{eq:bound-mP}
\eps_{n,\mathrm{case}}^{*}(\bhP,\bZ)
\geqslant\eps_{n,\mathrm{cell}}^{*}(\bhP,\bZ)
\geqslant
\left\lceil
\frac{
\Delta_{\bP}(\bZ)
}{
2W
}
\right\rceil,
\end{equation}
and
\begin{equation}
\label{eq:bound-mmu}
\eps_{n,\mathrm{case}}^{*}(\bhmu_z,\bZ)
\geqslant\eps_{n,\mathrm{cell}}^{*}(\bhmu_z,\bZ)
\geqslant
\left\lceil
\frac{
\Delta_{\bmu}^{\infty}(\bZ)
}{
2W
}
\right\rceil.
\end{equation}
\end{proposition}

\begin{proof}
Consider two datasets $\bZ$ and $\bZ'$ that differ only in cell
$(i,j)$. For $(\bmu_z,\bP)\in\mathcal H_k$, define 
\begin{equation*}
A_i
\left(
\bm u;\bZ,\bmu_z,\bV_{\bP}
\right)
:=
\frac{1}{d\sigma_2^2}
\sum_{j=1}^{d}
\sigma_{1,j}^2
\rho_1
\left(
\frac{
z_{ij}-\mu_{z,j}-\bm u^T\bv_{\bP,j}
}{
\sigma_{1,j}
}
\right).
\end{equation*}
Then,
\begin{align*}
&
\left|
A_i
\left(
\bm u;\bZ',\bmu_z,\bV_{\bP}
\right)
-
A_i
\left(
\bm u;\bZ,\bmu_z,\bV_{\bP}
\right)
\right|
\\
&=
\frac{\sigma_{1,j}^2}{d\sigma_2^2}
\left|
\rho_1
\left(
\frac{
z_{ij}'-\mu_{z,j}-\bm u^T\bv_{\bP,j}
}{
\sigma_{1,j}
}
\right)
-
\rho_1
\left(
\frac{
z_{ij}-\mu_{z,j}-\bm u^T\bv_{\bP,j}
}{
\sigma_{1,j}
}
\right)
\right|\leqslant
c_j.
\end{align*}
where
$
c_j
:=
\sigma_{1,j}^2\overline\rho_1/(d\sigma_2^2)
$.
 Define
$
g_2(t):=\rho_2(\sqrt t)
$
for $t\geqslant0$. By Definition~\ref{def:validrho}, $g_2$ is
nondecreasing and concave, with $g_2(0)=0$. Concavity implies that
 if $a\geqslant b\geqslant0$ and
$\delta=a-b$, then
\begin{equation*}
g_2(a)-g_2(b)
=
g_2(b+\delta)-g_2(b)
\leqslant
g_2(\delta)-g_2(0)
=
g_2(\delta).
\end{equation*}
The same argument with $a$ and $b$ interchanged applies when
$b\geqslant a$. Consequently, for every $a,b\geqslant0$,
\begin{equation*}
|g_2(a)-g_2(b)|
\leqslant
g_2(|a-b|).
\end{equation*}
Then
\begin{align}
&
\left|
g_2
\left[
A_i
\left(
\bm u;\bZ',\bmu_z,\bV_{\bP}
\right)
\right]
-
g_2
\left[
A_i
\left(
\bm u;\bZ,\bmu_z,\bV_{\bP}
\right)
\right]
\right|
\nonumber\\
&\leqslant
g_2(c_j)
=
\rho_2(\sqrt{c_j})
=
\rho_2
\left(
\frac{\sigma_{1,j}}{\sigma_2}
\sqrt{
\frac{\overline\rho_1}{d}
}
\right)
=
\omega_j.
\label{eq:one-cell-outer-bound}
\end{align}
Importantly, \eqref{eq:one-cell-outer-bound} holds uniformly over
$\bm u\in\mathbb R^k$.
Finally, define the profiled contribution of row $i$ by
\begin{equation*}
q_i
\left(
\bZ;\bmu_z,\bP
\right)
:=
\inf_{\bm u\in\mathbb R^k}
g_2
\left[
A_i
\left(
\bm u;\bZ,\bmu_z,\bV_{\bP}
\right)
\right],
\end{equation*}
then,
\begin{equation}
\label{eq:Q-row-decomposition}
Q_{\bZ}(\bmu_z,\bP)
=
\frac{1}{n}
\sum_{i=1}^{n}
q_i(\bZ;\bmu_z,\bP).
\end{equation}
From \eqref{eq:one-cell-outer-bound}, for every $\bm u$,
\begin{equation*}
g_2
\left[
A_i
\left(
\bm u;\bZ',\bmu_z,\bV_{\bP}
\right)
\right]
\leqslant
g_2
\left[
A_i
\left(
\bm u;\bZ,\bmu_z,\bV_{\bP}
\right)
\right]
+
\omega_j.
\end{equation*}
Taking the infimum over $\bm u$ on both sides yields
\begin{equation*}
q_i
\left(
\bZ';\bmu_z,\bP
\right)
\leqslant
q_i
\left(
\bZ;\bmu_z,\bP
\right)
+
\omega_j.
\end{equation*}
Conversely, the absolute bound in
\eqref{eq:one-cell-outer-bound} also implies that, for every
$\bm u\in\mathbb R^k$,
\begin{equation*}
g_2
\left[
A_i
\left(
\bm u;\bZ,\bmu_z,\bV_{\bP}
\right)
\right]
\leqslant
g_2
\left[
A_i
\left(
\bm u;\bZ',\bmu_z,\bV_{\bP}
\right)
\right]
+
\omega_j.
\end{equation*}
Taking the infimum over $\bm u$ gives
\begin{equation*}
q_i
\left(
\bZ;\bmu_z,\bP
\right)
\leqslant
q_i
\left(
\bZ';\bmu_z,\bP
\right)
+
\omega_j.
\end{equation*}
Therefore,
\begin{equation}
\label{eq:one-cell-profiled-bound}
\left|
q_i
\left(
\bZ';\bmu_z,\bP
\right)
-
q_i
\left(
\bZ;\bmu_z,\bP
\right)
\right|
\leqslant
\omega_j.
\end{equation}
Thus, replacing one cell in column $j$ changes the profiled
contribution of its row by at most $\omega_j$.

For each column $j$, let
$N_j^{(m)}
:=
\#
\left\{
i:z_{ij}^{(m)}\neq z_{ij}
\right\}$
denote the number of replaced cells in that column. By the definition
of $\mathcal C_m(\bZ)$, we have
$N_j^{(m)}\leqslant m$, for $j=1,\ldots,d$.
Enumerate all replaced cells and construct a sequence of intermediate
datasets
$\bZ^{[0]}=\bZ,\,
\bZ^{[1]},\ldots,\bZ^{[N]}=\bZ^{(m)}$,
where each pair $\bZ^{[\ell-1]}$ and $\bZ^{[\ell]}$ differs in
exactly one cell and $N:=\sum_{j=1}^{d}N_j^{(m)}$. If the cell replaced at step $\ell$ belongs to
column $j_\ell$, the preceding one-cell result gives
\[
\left|
Q_{\bZ^{[\ell]}}(\bmu_z,\bP)
-
Q_{\bZ^{[\ell-1]}}(\bmu_z,\bP)
\right|
\leqslant
\frac{\omega_{j_\ell}}{n}.
\]
Applying the triangle inequality along this sequence gives
\begin{align*}
\left|
Q_{\bZ^{(m)}}(\bmu_z,\bP)
-
Q_{\bZ}(\bmu_z,\bP)
\right|
&\leqslant
\sum_{\ell=1}^{N}
\left|
Q_{\bZ^{[\ell]}}(\bmu_z,\bP)
-
Q_{\bZ^{[\ell-1]}}(\bmu_z,\bP)
\right|
\\
&\leqslant
\frac{1}{n}
\sum_{\ell=1}^{N}
\omega_{j_\ell}
=
\frac{1}{n}
\sum_{j=1}^{d}
N_j^{(m)}\omega_j
\\
&\leqslant
\frac{m}{n}
\sum_{j=1}^{d}\omega_j
=
\frac{mW}{n}.
\end{align*}
Consequently,
\begin{equation}
\label{eq:uniform-objective-bound}
\left|
Q_{\bZ^{(m)}}(\bmu_z,\bP)
-
Q_{\bZ}(\bmu_z,\bP)
\right|
\leqslant
\frac{mW}{n}.
\end{equation}
Let
\[
(\bhmu_z^{(m)},\bhP^{(m)})
=
\operatorname*{argmin}_{(\bmu_z,\bP)\in\mathcal H_k}
Q_{\bZ^{(m)}}(\bmu_z,\bP).
\]
By its minimizing property and
\eqref{eq:uniform-objective-bound},
\begin{align}
Q_{\bZ}
\left(
\bhmu_z^{(m)},\bhP^{(m)}
\right)
-
Q_{\bZ}
\left(
\bhmu_z,\bhP
\right)
&\leqslant
\left|
Q_{\bZ}
\left(
\bhmu_z^{(m)},\bhP^{(m)}
\right)
-
Q_{\bZ^{(m)}}
\left(
\bhmu_z^{(m)},\bhP^{(m)}
\right)
\right|
\nonumber\\
&\quad+
\left|
Q_{\bZ^{(m)}}
\left(
\bhmu_z,\bhP
\right)
-
Q_{\bZ}
\left(
\bhmu_z,\bhP
\right)
\right|
\nonumber\\
&\leqslant
\frac{2mW}{n}.
\label{eq:contaminated-minimizer-gap}
\end{align}

Suppose now that
\[
\frac{2mW}{n}
<
\Delta_{\bP}(\bZ).
\]
For $\eta>0$, define
\[
D_{\bP}(\eta)
:=
\inf_{\substack{
(\bmu_z,\bP)\in\mathcal H_k\\
\theta_{\max}(\bP,\bhP)
\geqslant
\frac{\pi}{2}-\eta
}}
\left\{
Q_{\bZ}(\bmu_z,\bP)
-
Q_{\bZ}(\bhmu_z,\bhP)
\right\}.
\]
Since
$
D_{\bP}(\eta)\rightarrow\Delta_{\bP}(\bZ)
$
as $\eta\downarrow0$, there exists $\eta_0>0$ such that
\[
D_{\bP}(\eta_0)
>
\frac{2mW}{n}.
\]
Consequently, every $(\bmu_z,\bP)\in\mathcal H_k$ satisfying
\[
\theta_{\max}(\bP,\bhP)
\geqslant
\frac{\pi}{2}-\eta_0
\]
has
\[
Q_{\bZ}(\bmu_z,\bP)
-
Q_{\bZ}(\bhmu_z,\bhP)
>
\frac{2mW}{n}.
\]
The contaminated minimizer cannot satisfy this angular condition,
because that would contradict
\eqref{eq:contaminated-minimizer-gap}. Hence, for every
$\bZ^{(m)}\in\mathcal C_m(\bZ)$,
\[
\theta_{\max}
\left(
\bhP^{(m)},\bhP
\right)
<
\frac{\pi}{2}-\eta_0.
\]
It follows that
\[
\sup_{\bZ^{(m)}\in\mathcal C_m(\bZ)}
\theta_{\max}
\left(
\bhP^{(m)},\bhP
\right)
\leqslant
\frac{\pi}{2}-\eta_0
<
\frac{\pi}{2}.
\]
Therefore, the cellPCA subspace orientation cannot break down at
contamination level $m$.
Suppose that
\[
\frac{2mW}{n}
<
\Delta_{\bmu}^{\infty}(\bZ).
\]
For $R>0$, define
\[
D_{\bmu}(R)
:=
\inf_{\substack{
(\bmu_z,\bP)\in\mathcal H_k\\
\|\bmu_z-\bhmu_z\|\geqslant R
}}
\left\{
Q_{\bZ}(\bmu_z,\bP)
-
Q_{\bZ}(\bhmu_z,\bhP)
\right\}.
\]
Since
\[
\liminf_{R\rightarrow\infty}
D_{\bmu}(R)
=
\Delta_{\bmu}^{\infty}(\bZ)
>
\frac{2mW}{n},
\]
there exists a finite $R_{\mu,0}$ such that
\[
D_{\bmu}(R_{\mu,0})
>
\frac{2mW}{n}.
\]
The contaminated minimizer cannot satisfy
\[
\|\bhmu_z^{(m)}-\bhmu_z\|
\geqslant R_{\mu,0},
\]
because this would contradict
\eqref{eq:contaminated-minimizer-gap}. Hence,
\[
\sup_{\bZ^{(m)}\in\mathcal C_m(\bZ)}
\|\bhmu_z^{(m)}-\bhmu_z\|
\leqslant R_{\mu,0}
<\infty.
\]
Therefore, the cellPCA location cannot break down at contamination
level $m$.

Therefore, the orientation and location cannot break down,
respectively, for any integers satisfying
\[
m<
\left\lceil
\frac{
n\Delta_{\bP}(\bZ)
}{
2W
}
\right\rceil
\]
and
\[
m<
\left\lceil
\frac{
n\Delta_{\bmu}^{\infty}(\bZ)
}{
2W
}
\right\rceil.
\]
By using Lemma~\ref{lem:case-cell-breakdown-relation}, this proves \eqref{eq:bound-mP} and \eqref{eq:bound-mmu}.
\end{proof}
\newcounter{oldtheo}
    \setcounter{oldtheo}{\value{proposition}}

\setcounter{proposition}{2}

\begin{proposition}[Restated]
Assume that conditions F1--F5 hold.
Then, 
\begin{equation}
\eps_{n,\mathrm{case}}^{*}
\left(
\bhB,\bZ
\right)\geqslant\eps_{n,\mathrm{cell}}^{*}
\left(
\bhB,\bZ
\right)
\geqslant
\min\Bigg\{
\left\lceil
\frac{
\Delta_{\bP}(\bZ)
}{
2W
}
\right\rceil,
\left\lceil
\frac{
\Delta_{\bmu}^{\infty}(\bZ)
}{
2W
}\right\rceil,
\frac{1}{n}\left\lceil
\frac{r_0(n-h+1)}{d}
\right\rceil
\Bigg\},
\label{eq:slope-breakdown-final}
\end{equation}
where $\lceil x\rceil$ denotes the smallest integer greater than or
equal to $x$.
\end{proposition}

\setcounter{proposition}{\value{oldtheo}}
\begin{proof}
The result follows by substituting the lower bounds for the cellPCA
location and subspace breakdown points established in
Proposition~\ref{pro:breakdown-mu-P} into the bound of
Proposition~\ref{prop:slope-breakdown-mu-P}.
\end{proof}

\section{More on cellBoot}
\label{app:cellBoot}
\subsection{More on the construction of FastCellCov}
\label{app:FastCellCov}

Section~\ref{sec:FastCellCov} introduced the
FastCellCov auxiliary estimator. It 
starts by applying cellCov to the 
actual sample $\btz_1, \dots, \btz_n$. 
Now consider another sample
$\btz_1^*, \dots, \btz_n^*$\,, 
that may be a bootstrap sample or a simulated 
sample of size $n$ generated from $F_{\btheta}$.
We then standardize it by computing
$\bz_i^*=\bhD^{-1}\btz^*_i$ where 
$\bhD$ is the diagonal matrix of scale 
estimators of the original sample, as in 
Section~\ref{sec:estimator}.

To compute fitted values $\bhz_i^{\,*}$ from 
the $\bz_i^*$\,, we first need to 
construct an outlier-free version 
$\bz_{i,\mathrm{imp}}^*$ of $\bz_i^*$.
We first standardize each $\bz_{i}^*$ to 
\begin{equation}
  \bz_{s,i}^* = \bhS^{-1}\bigl(\bz_{i}^* - \bhmu_z),
  \label{eq:standardization}
\end{equation}
where $\bhS := \diag(s_1,\ldots,s_d)$ with
$s_j := \sqrt{(\bhD^{-1}\bhSigma\bhD^{-1})_{jj}}$\,.
Then we compute the filtering weights vector 
$\bw^{f}_i=(w^{f}_{i1},\dots,w^{f}_{id})^T$, where 
$w^{f}_{ij}=w^{\cell}(z_{s,ij}^*)$.
Next we compute the predicted vectors 
$\bhz_{s,i}^{\,*} = (\hz_{s,i1}^{\,*}, \dots, 
\hz_{s,id}^{\,*})^T$ as
\begin{equation}
  \hz_{s,ij}^{\,*} \;=\; \frac{1}
  {\sum_{h \in H_j} | \mathrm{corr}_{jh} |\,
  w^{f}_{ih}}\;\, 
  \sum_{h \in H_j} |\mathrm{corr}_{jh}|
  \, w^{f}_{ih}\, b_{jh}\, z_{s,ih}^*,
  \label{eq:predicted}
\end{equation}
which is a weighted mean of the terms 
$b_{jh} z_{s,ih}^*$. We obtain the $b_{jh}$ 
from the Detecting Deviating Cells (DDC) method, 
used in the initialization of cellCov. The
$b_{jh}$ is the DDC-estimated slope from a robust 
no-intercept regression of variable $j$ on 
variable $h$. We also use the DDC-estimated 
absolute correlation $\mathrm{corr}_{jh}$ between 
variables $j$ and $h$, combined with 
$\omega^{f}_{ih}$. The index set 
$H_j$ contains all variables $h$ for which 
$|\mathrm{corr}_{jh}| \geqslant 0.5$. 

Note that the prediction \eqref{eq:predicted} 
typically shrinks the scale of the entries. 
To correct for this, we compute the slope
$a^{shr}_j$ of a robust no-intercept regression 
of the observed $z_{s,ij}^*$ on the predicted 
values $\hz_{s,ij}^{\,*}$\,, and then rescale 
the predictions by setting all $\hz_{s,ij}^{\,*}
\leftarrow a^{shr}_j \hz_{s,ij}^{\,*}$\,. 

We then compute the standardized cell residuals
\begin{equation}\label{eq:residuals}
  r_{ij}^* = \frac{z_{s,ij}^* - 
  \hz_{s,ij}^{\,*}}{s_{r,j}}\;,
\end{equation}
where $s_{r,j}=\sigma_M(\{z_{ij} - 
\hz_{ij}\}_{i=1}^n)$ was computed 
on the observed sample.
Next we compute the residual weights vector 
$\bw^r_i=(w^{r}_{i1},\dots,w^r_{id})^T$, where
$w^r_{ij}=w^{\cell}(r_{ij}^*)$. Finally we set 
\begin{equation}
 \bz_{i,\imp}^*=  (\bw^f_i\odot \bw^r_i 
 \odot \bmed_i^*)\odot\bz_{i}^*+
 (\bone_d-\bw^f_i\odot \bw^r_i \odot \bmed_i^*)
 \odot\btz_{s,i}^*\;,
\end{equation}
with $\bmed_i^*$ the  missing indicator vector
whose entries are 0 for missing and 1 otherwise, 
and $\btz_{s,i}^*=\bhS\bhz_{s,i}^{\,*}+\bhmu_z$.
Then we compute $\bhz_i^{\,*} := \bhmu_z +
\bhV\bhV^T(\bz_{i,\imp}^*-\bhmu_z)$.

Next we compute the matrix
$\btW^*=\bW^{\cell}_* \odot \bO_*$, where the 
$n \times d$ matrix $\bW^{\cell}_*$ contains the 
weights $w_{*,ij}^{\cell}=w^{\cell}\left( 
\frac{\hr_{ij}^{\,*}}{\hsigma_{1,j}^*}\right)$ 
with $\hr_{ij}^{\,*}=z_{ij}^*-\hz_{ij}^{\,*}$ and 
$\hsigma_{1,j}= \sigma_M(\{\hr_{ij}\}_{i=1}^n)$ 
was already computed on the observed sample.
The $n \times d$ matrix $\bO_*$ contains the 
indicators of the observed cells in 
$\bz_{1}^*,\dots,\bz_{n}^*$. Next we compute
the new center
\begin{equation} \label{eq:btmuF}
  \btmu_{F} := \bC^{\bmu}\odot\sum_{i=1}^n
  w_{*,i}^{\case}w_{*,i}^{\mathrm{sub} }\btW_i^*\bz_i^*\,,
\end{equation}
where $\btW_i^*$ is a diagonal matrix whose diagonal 
is the $i$-th row of $\btW^*$ and 
$\bC^{\bmu}=\diag(1/c^{\bmu}_{1},\dots,1/c^{\bmu}_{d})$ 
with $c^{\bmu}_{j}=\sum_{i=1}^n w_{*,i}^{\case}
w_{*,i}^{\mathrm{sub} } w_{*,ij}^{\cell}o_{ij}^*$.
We also compute the new covariance
\begin{equation} \label{eq:btSigmaF}
   \btSigma_{F} :=
   \bC^{\bSigma}\odot\sum_{i=1}^n
   w_{*,i}^{\case}w_{*,i}^{\mathrm{sub} }\btW_i^*
   (\bz_i^* - \bhmu_z)
   (\bz_i^* - \bhmu_z)^T\btW_i^*,
\end{equation}
where $\bC^{\bSigma}=\{c^{\bSigma}_{j\ell}\}$ with 
$c^{\bSigma}_{j\ell}=1/\sum_{i=1}^n o_{ij}^
*o_{i\ell}^*w_{*,i}^{\case}w_{*,i}^{\mathrm{sub}}
w_{*,ij}^{\cell} w_{*,i\ell}^{\cell}$\,.
The weights $w_{*,i}^{\case}$ 
in~\eqref{eq:btmuF} and~\eqref{eq:btSigmaF}
are given by 
$w_{*,i}^{\case} = w^{\case}\left(\frac{
\hr_{*,i}^{\,\imp}}{\hsigma_2}\right)$
where $\hr_{*,i}^{\;\imp}$ is given by
\begin{equation}
 \hr_{*,i}^{\;\imp} := \frac{1}
 {\sum_{j=1}^{d}o_{*,ij}w_{*,ij}^{\cell}} \;
 \sum_{j=1}^{d}o_{*,ij}w_{*,ij}^{\cell}
   (\hr_{ij}^{\,*})^2
\end{equation}
and $\hsigma_2 = \sigma_M(\hr_i^{\;\imp})$ 
was computed on the observed sample.

As $w_{*,i}^{\mathrm{case}}$ and 
$w_{*,ij}^{\mathrm{cell}}$ measure the outlyingness 
of $\bz_i^*$ in the principal subspace formed
by $\bhmu_z$ and $\bhV$, we included 
weights $w_{*,i}^{\mathrm{sub}}$
in~\eqref{eq:btmuF} and~\eqref{eq:btSigmaF} to 
limit the influence of potentially outlying fitted 
values $\bhz_i^{\;*}$. Let $d_{i,\mathrm{mah}}^2 :=
(\bhz_i^{\,*}-\bhmu_{\bhz})^T \bhSigma_{\bhz}^{-1}
(\bhz_i^{\,*}-\bhmu_{\bhz})$ denote the robust squared 
Mahalanobis distance of $\bhz_i^{\,*}$ from the robust
location and scatter estimates $\bhmu_{\bhz}$ and 
$\bhSigma_{\bhz}$\,, computed by applying the MCD 
estimator to the fitted values computed on the
observed sample, as in~\eqref{eq:SigmaMCD}. 
The subspace weights are defined as
$w_{*,i}^{\mathrm{sub}} =
\frac{\psi^s\!\left(d^2_{i,\mathrm{mah}}\right)}
     {d^2_{i,\mathrm{mah}}},$
where $\psi^s$ is the derivative of the hyperbolic 
tangent $\rho$-function defined in Supplementary 
Material~\ref{app:adddet}.

The final FastCellCov estimates $\bhmu_{F}$ and 
$\bhSigma_{F}$ of $\bmu$ and $\bSigma$ are given by
\begin{equation}
  \bhmu_{F} := \bhD \btmu_{F}\qquad \mbox{and} 
  \qquad  \bhSigma_{F} := \bhD \btSigma_{F}\bhD.
\end{equation}

\newpage
\subsection{Proofs of asymptotic exactness of cellBoot}
\label{app:asyCBoot}
Throughout the paper, for sequences of random variables $(X_n)$ and positive
deterministic sequences $(a_n)$, we use the following asymptotic notation:
\begin{itemize}
\item $X_n = o_p(a_n)$ if $X_n/a_n \rightarrow_p 0$;
\item $X_n = O_p(a_n)$ if for every $\eps>0$ there exists an $M<\infty$ such that
$\Pr(|X_n|/a_n > M) < \eps$ for all sufficiently large $n$.
\item $X_n = o(a_n)$ if $X_n/a_n \rightarrow 0$.
\end{itemize}
We write $X_n \rightarrow_p X$ for convergence in probability.
For any event $A$, we denote by $I_A$ its indicator function, that is,
\[
I_A
:=
\begin{cases}
1, & \text{if } A \text{ occurs},\\
0, & \text{otherwise}.
\end{cases}
\]
In particular, for any integrable random vector $X$ and any event $A$,
\[
\E[X]
=
\E[X I_A]
+
\E[X I_{A^c}]
=
\Pr(A)\,\E[X\mid A]
+
\Pr(A^c)\,\E[X\mid A^c],
\]
where $A^c$ denotes the complement of the event $A$.
The bootstrap probability measure ${\Pr}^*$ denotes probability computed
under the bootstrap distribution, conditional on the observed data
$\bx_1,\ldots,\bx_n$. Correspondingly, $\E^*(\cdot)$ and $\Cov^*(\cdot)$ denote
expectation and covariance with respect to ${\Pr}^*$.

The admissible parameter space is defined as
\begin{equation*}
\bTheta =
\Bigl\{\,(\bmu^T,\vecth_s(\bSigma)^T)^T 
\in \mathbb R^{d + d(d+1)/2} :
\bSigma \in \mathbb S^d,\ 
\|\bmu\|\leqslant M,\ 
c \leqslant \lambda_{\min}(\bSigma)\leqslant 
\lambda_{\max}(\bSigma)\leqslant C \Bigr\},
\end{equation*}
where $M<\infty$ and $0< c \leqslant C<\infty$.
Here, $\mathbb S^d$ denotes the space of symmetric $d\times d$ matrices.
The operator $\vecth_s(\cdot)$ denotes the \emph{scaled half-vectorization}
defined as follows: for $\bSigma\in\mathbb S^d$,
$\vecth_s(\bSigma)$ is obtained by stacking the lower triangular elements
(including the diagonal), and multiplying each off-diagonal element by $\sqrt{2}$.
With this convention, the Euclidean norm of $\vecth_s(\bSigma)$ coincides
with the Frobenius norm of $\bSigma$, consequently, for any $(\bmu,\bSigma)\in\mathbb R^d\times\mathbb S^d$,
\[
\|(\bmu^T,\vecth_s(\bSigma)^T)^T\|^2
=
\|\bmu\|^2+\|\vect(\bSigma)\|^2,
\]
so that the standard Euclidean geometry on
$\mathbb R^{d+d(d+1)/2}$ corresponds exactly to the natural
Frobenius geometry on the covariance matrix component.
Hence, $\bTheta$ is a subset of $\mathbb R^{d + d(d+1)/2}$ whose last $d(d+1)/2$ components, 
when reshaped and rescaled into a  symmetric $d\times d$ matrix, form a  matrix 
with eigenvalues bounded between $c$ and $C$.
The interior of $A$ is defined as
\[
\interior(A)
:=
\Bigl\{x \in A :
\exists\, \delta > 0 \text{ such that}\ 
x+h\in A
\ \text{for all }h\in\mathbb R^{m}
\ \text{with }\|h\|<\delta
\Bigr\}.
\]

Let us consider $\bt=(\ba^T,\vecth(\bL)^T)^T\in
\mathbb R^{d + d(d+1)/2}$, 
where $\bL=\bQ\bLambda\bQ^T$, with $\bQ$ the $d\times d$ orthogonal matrix whose columns are the eigenvectors of $\bL$, and $\bLambda=\mathrm{diag}(\lambda_1,\ldots,\lambda_d)$ the diagonal matrix containing the corresponding eigenvalues. Here, $\vecth(\cdot)$ denotes the half-vectorization operator, mapping a
symmetric matrix to the vector obtained by stacking its lower triangular
elements (including the diagonal). 
The operator $\Pi_{\bTheta}$ projects $\bt$ onto $\bTheta$ and is defined as
\[
\Pi_{\bTheta}(\bt)
=
(\tilde\bmu^T,\vecth_s(\btSigma)^T)^T,
\]
where
\begin{equation*}
\tilde\bmu =
\begin{cases}
\ba, & \text{if } \|\ba\|\leqslant M,\\[3pt]
\dfrac{M}{\|\ba\|}\,\ba, & \text{if } \|\ba\|>M,
\end{cases}
\qquad
\btSigma = \bQ\bLambda_c\bQ^T,
\end{equation*}
and $\bLambda_c=\mathrm{diag}(\tilde\lambda_1,\dots,\tilde\lambda_d)$ with
\[
\tilde\lambda_j=\min\{\max\{\lambda_j,c\},C\},
\qquad j=1,\ldots,d.
\]

A bootstrap sample is obtained by resampling from the empirical
distribution $\hat F$, and the corresponding bootstrap version of
$\hpi_n(\boheta_n)$ is denoted by $\hpi^*_n(\boheta_n)$. Then, the bootstrap estimator $\bhtheta_n^*$ is 
\begin{equation*}
\bhtheta_n^*
=
\argzero_{\btheta\in\bTheta}
  \bigl\{
    \hpi^*_n(\boheta_n)-\pibar(\btheta,\boheta_n,n)
  \bigr\}.
\end{equation*}

Let $\bx_{1,\btheta},\ldots,\bx_{n,\btheta}$ be
a random sample generated from the distribution
$F_{\btheta}$, with parameter $\btheta=(\bmu^T,\vecth_s(\bSigma)^T)^T$.
We rewrite the FastCellCov estimator of the location as
\begin{equation*}
\bhmu_{F,n}(\btheta,\boeta)
=
\Big( \bC^{\bmu}(\btheta,\boeta) \Big)^{-1}
\odot
\sum_{i=1}^n
   w^{\mathrm{case}}(\bx_{i,\btheta},\boeta)\,
   w^{\mathrm{sub}}(\bx_{i,\btheta},\boeta)\,
   \btW(\bx_{i,\btheta},\boeta)\bx_{i,\btheta},
\end{equation*}
where $\btW(\bx_{i,\btheta},\boeta)=\diag(\tilde w^{\mathrm{cell}}(x_{i1,\btheta},\boeta),\dots,\tilde w^{\mathrm{cell}}(x_{id,\btheta},\boeta))$, and $\bC^{\bmu}(\btheta,\boeta)$ is the diagonal matrix with entries, for $j=1,\dots,d$,
\[
\big(\bC^{\bmu}(\btheta,\boeta)\big)_{jj}
:=
\max\!\left\{
\sum_{i=1}^n
w^{\mathrm{case}}(\bx_{i,\btheta},\boeta)\,
w^{\mathrm{sub}}(\bx_{i,\btheta},\boeta)\,
w^{\mathrm{cell}}(x_{ij,\btheta},\boeta)
,
\, n\tilde\delta
\right\}.
\]
Note that the diagonal elements of the summed denominator matrix are clipped
at $n\tilde\delta$, with $\tilde\delta>0$, making explicit the regularization step that is implicitly applied in the FastCellCov estimator. This prevents division by arbitrarily small values and ensures numerical stability.
 We make explicit the dependence on
$\boeta\in\bH\subseteq\mathbb R^{r}$, namely the tuning parameters that determine
the weights
$w^{\mathrm{case}}(\bx_{i,\btheta},\boeta)$,
$w^{\mathrm{sub}}(\bx_{i,\btheta},\boeta)$,
and the diagonal matrices
$\btW(\bx_{i,\btheta},\boeta)$ associated with each observation
$\bx_{i,\btheta}$.
Specifically,
\begin{align*}
\boeta=
\left(
\diag(\bhD)^T,
\bhmu_z^T,\right.&
\vecth(\bhV\bhV^T)^T,
\diag(\bhS)^T,
\vecth_{-d}(\bB^{slope})^T,\\&
\vecth_{-d}(\bC^{corr})^T,
(\ba^{shr})^T,
\bs_{r}^T,
(\bhsigma^*)^T,
\bhmu_{\bhz}^T,
\vecth_{-d}(\bhSigma_{\bhz})^T
\left.\right)^T,
\end{align*}
where $\bB^{slope}=\{b_{jh}\}$, $\bC^{corr}=\{corr_{jh}\}$,
$\ba^{shr}=(a^{shr}_1,\dots,a^{shr}_d)^T$,
$\bs_{r}=(s_{r,1},\dots,s_{r,d})^T$, and
$\bhsigma^*=(\hsigma^*_{1,1},\dots,\hsigma^*_{1,d},\hsigma^*_2)^T$.
Here, $\vecth_{-d}(\cdot)$ denotes the
half-vectorization without the diagonal, obtained by stacking only the
strictly lower triangular elements of a symmetric matrix.
Analogously, the FastCellCov estimator of the scatter matrix is
\begin{align*}
\bhSigma_{F,n}(\btheta,\boeta)
=
\Big( &\bC^{\bSigma}(\btheta,\boeta) \Big)^{-1}
\odot\\&
\sum_{i=1}^n
   w^{\mathrm{case}}(\bx_{i,\btheta},\boeta)\,
   w^{\mathrm{sub}}(\bx_{i,\btheta},\boeta)\,
   \btW(\bx_{i,\btheta},\boeta)\,
   (\bx_{i,\btheta}-\bhmu_z)\,(\bx_{i,\btheta}-\bhmu_z)^T\,
   \btW(\bx_{i,\btheta},\boeta),
\end{align*}
where $\bC^{\bSigma}(\btheta,\boeta)
=\{c^{\bSigma}_{j\ell}(\btheta,\boeta)\}$ is defined elementwise by
\[
c^{\bSigma}_{j\ell}(\btheta,\boeta)
:=
\max\!\left\{
\sum_{i=1}^n
w^{\mathrm{case}}(\bx_{i,\btheta},\boeta)\,
w^{\mathrm{sub}}(\bx_{i,\btheta},\boeta)\,
\tilde w^{\mathrm{cell}}(x_{ij,\btheta},\boeta)\,
\tilde w^{\mathrm{cell}}(x_{i\ell,\btheta},\boeta)
,
\, n\tilde\delta
\right\}.
\]

Let
\begin{align*}
a_\mu(\bx_{i,\btheta},\boeta)
&:= w^{\mathrm{case}}(\bx_{i,\btheta},\boeta)\,
   w^{\mathrm{sub}}(\bx_{i,\btheta},\boeta)\,
   \btW(\bx_{i,\btheta},\boeta)\,
   \bx_{i,\btheta},\\
b_\mu(\bx_{i,\btheta},\boeta)
&:= \diag\!\big(\bC^{\bmu}(\bx_{i,\btheta},\boeta)\big),\\[0.2cm]
a_\Sigma(\bx_{i,\btheta},\boeta)
&:= \vecth_s\Big(
   w^{\mathrm{case}}(\bx_{i,\btheta},\boeta)\,
   w^{\mathrm{sub}}(\bx_{i,\btheta},\boeta)\,
   \btW(\bx_{i,\btheta},\boeta)\,
   (\bx_{i,\btheta}-\bmu)\,(\bx_{i,\btheta}-\bmu)^T\,
   \btW(\bx_{i,\btheta},\boeta)
   \Big),\\
b_\Sigma(\bx_{i,\btheta},\boeta)
&:= \vecth_s\big(\bC^{\bSigma}(\bx_{i,\btheta},\boeta)\big),
\end{align*}
where \[
\bC^{\bmu}(\bx_{i,\btheta},\boeta)
=
w^{\mathrm{case}}(\bx_{i,\btheta},\boeta)\,
w^{\mathrm{sub}}(\bx_{i,\btheta},\boeta)\,
\btW(\bx_{i,\btheta},\boeta).
\] and $\bC^{\bSigma}(\bx_{i,\btheta},\boeta)
=\{c^{\bSigma}_{j\ell}(\bx_{i,\btheta},\boeta)\}$ is defined elementwise by
\[
c^{\bSigma}_{j\ell}(\bx_{i,\btheta},\boeta)
=
w^{\mathrm{case}}(\bx_{i,\btheta},\boeta)\,
w^{\mathrm{sub}}(\bx_{i,\btheta},\boeta)\,
\tilde w^{\mathrm{cell}}(x_{ij,\btheta},\boeta)\,
\tilde w^{\mathrm{cell}}(x_{i\ell,\btheta},\boeta).
\]
We now collect the location and scatter components into the vector-valued
functions
\begin{equation*}
a_F(\bx_{i,\btheta},\boeta)
:=
\bigl(
   a_\mu(\bx_{i,\btheta},\boeta)^T,\;
   a_\Sigma(\bx_{i,\btheta},\boeta)^T
\bigr)^T,
\qquad
b_F(\bx_{i,\btheta},\boeta)
:=
\bigl(
   b_\mu(\bx_{i,\btheta},\boeta)^T,\;
   b_\Sigma(\bx_{i,\btheta},\boeta)^T
\bigr)^T.
\end{equation*}
Write
\begin{equation*}
\bar a_{F,n}(\btheta,\boeta)
=
\frac{1}{n}\sum_{i=1}^n a_F(\bx_{i,\btheta},\boeta),
\qquad
\bar b_{F,n}(\btheta,\boeta)
=
\frac{1}{n}\sum_{i=1}^n b_F(\bx_{i,\btheta},\boeta).
\end{equation*}
With this notation, the
FastCellCov estimator $\hP_F(\btheta,\boeta,n)$ of $(\bmu^T,\vecth_s(\bSigma)^T)^T$ is defined componentwise,
for $j=1,\dots,d+d(d+1)/2$, as
\[
\big(\hP_F(\btheta,\boeta,n)\big)_j
:=
\begin{cases}
\dfrac{(\bar a_{F,n})_j(\btheta,\boeta)}
      {(\bar b_{F,n})_j(\btheta,\boeta)},
&  \big|(\bar b_{F,n})_j(\btheta,\boeta)\big|
\geqslant \tilde \delta, \\[6pt]
\dfrac{(\bar a_{F,n})_j(\btheta,\boeta)}
      {\tilde \delta},
&  \big|(\bar b_{F,n})_j(\btheta,\boeta)\big|
< \tilde \delta.
\end{cases}
\]

Finally, the corresponding population target is
\begin{equation*}
P_F(\btheta,\boeta)
=
\mu_{a,F}(\btheta,\boeta)\oslash \mu_{b,F}(\btheta,\boeta),
\end{equation*}
where  $\oslash$ denotes the elementwise (Hadamard) division of two vectors of the same
dimension, 
\begin{equation*}
\mu_{a,F}(\btheta,\boeta)=\E[a_F(\bx_{\btheta},\boeta)],
\qquad
\mu_{b,F}(\btheta,\boeta)=\E[b_F(\bx_{\btheta},\boeta)],
\end{equation*}
and $\bx_{\btheta}$ denotes a generic random vector from $F_{\btheta}$.

Figure~\ref{fig:proof-roadmap} summarizes the organization
of the proofs and the main dependencies among the results.
\begin{figure}[H]
\centering
\begin{tikzpicture}[
  box/.style={
    draw=black!60,
    rounded corners=2pt,
    fill=black!3,
    text width=5.4cm,
    align=center,
    inner sep=6pt,
    font=\small
  },
  result/.style={box,fill=blue!6,draw=blue!50!black},
  arrow/.style={-{Stealth[length=2mm]},semithick},
  node distance=8mm and 12mm
]

\node[box] (geometry) {
  Lemmas~\ref{prop:closed},
  \ref{prop:projection-btheta},
  \ref{lem:proj-firm}, and
  \ref{prop:proj-implies-zero-btheta}\\[-1mm]
  Parameter space and \\[-2mm]projection properties
};

\node[box,right=of geometry] (algorithm) {
  Proposition~\ref{prop:proj-fp-contraction_supp}\\[-1mm]
  Convergence of the II algorithm
};

\draw[arrow] (geometry) -- (algorithm);

\node[box,below=of geometry] (consistency) {
  Proposition~\ref{prop:2_supp}\\[-1mm]
  Consistency of the II estimator
};

\node[box,right=of consistency] (uniform) {
  Lemma~\ref{lem:ratio-uniform-eta}\\[-1mm]
  Uniform bounds for auxiliary\\[-2mm] estimation
  and simulation error
};

\draw[arrow] (geometry) -- (consistency);

\node[box,below=of consistency] (expansions) {
  Propositions~\ref{pro:II_WI},
  \ref{prop:bootstrap-consistency-theta},
  and~\ref{pro:II_WI_boot}\\[-1mm]
  II linear expansions and\\[-2mm]
  bootstrap consistency
};

\node[box,right=of expansions] (auxiliary) {
  Propositions~\ref{pro:adphat}
  and~\ref{pro:converbootmean}\\[-1mm]
  Asymptotic normality and \\[-2mm]bootstrap
  validity of \\[-2mm]the auxiliary estimator
};

\draw[arrow] (consistency) -- (expansions);
\draw[arrow] (uniform) -- (auxiliary);
\draw[arrow] (uniform.south west)
  -- (expansions.north east);

\node[result,below=12mm of expansions,
      xshift=3.5cm] (general) {
  Proposition~\ref{pro:bootstrap-theta}\\[-1mm]
  Bootstrap validity \\[-2mm]for the II estimator
};

\draw[arrow] (expansions.south)
  -- (general.north west);
\draw[arrow] (auxiliary.south)
  -- (general.north east);

\node[result,below=of general] (theorem) {
  Theorem~\ref{the:bootstrap-theta}\\[-1mm]
  Bootstrap validity for cellBoot
};

\draw[arrow] (general) -- (theorem);

\node[result,below=of theorem] (regression) {
  Corollary~\ref{prop:bootstrap-theta-linear-reg}\\[-1mm]
  Bootstrap validity for \\[-2mm]regression parameters
};

\node[result,below=of regression] (coverage) {
  Corollary~\ref{cor:efron-percentile-exact}\\[-1mm]
  Asymptotically exact \\[-2mm]confidence intervals
};

\draw[arrow] (theorem) -- (regression);
\draw[arrow] (regression) -- (coverage);

\end{tikzpicture}
\caption{Roadmap of the main proof dependencies.
Related intermediate results are grouped.}
\label{fig:proof-roadmap}
\end{figure}

\begin{lemma}\label{prop:closed}
The set $\bTheta$ is closed, convex, and compact.
\end{lemma}

\begin{proof}
To prove convexity, let 
$(\bmu_1^T,\vecth_s(\bSigma_1)^T)^T$ and 
$(\bmu_2^T,\vecth_s(\bSigma_2)^T)^T$ 
be two elements of $\bTheta$, and let $t\in[0,1]$.  
Define
\begin{equation*}
\bmu_t = t\bmu_1+(1-t)\bmu_2, 
\qquad 
\vecth_s(\bSigma_t) = t\vecth_s(\bSigma_1)+(1-t)\vecth_s(\bSigma_2).
\end{equation*}
By convexity of the Euclidean norm,
\begin{equation*}
\|\bmu_t\| 
\leqslant 
t\|\bmu_1\| + (1-t)\|\bmu_2\| 
\leqslant 
M,
\end{equation*}
so $\bmu_t$ satisfies $\|\bmu_t\|\leqslant M$.

We have $\bSigma_t\in\mathbb S^d$.
We now verify that the eigenvalues of $\bSigma_t$ remain within $[c,C]$.
It is well known that the maps 
$\bSigma\mapsto\lambda_{\min}(\bSigma)$ 
and $\bSigma\mapsto\lambda_{\max}(\bSigma)$ 
are, respectively, concave and convex on the space of symmetric matrices.
Hence, for any $t\in[0,1]$,
\begin{equation*}
\lambda_{\min}(\bSigma_t)
\geqslant 
t\,\lambda_{\min}(\bSigma_1)
+
(1-t)\,\lambda_{\min}(\bSigma_2),
\qquad
\lambda_{\max}(\bSigma_t)
\leqslant 
t\,\lambda_{\max}(\bSigma_1)
+
(1-t)\,\lambda_{\max}(\bSigma_2).
\end{equation*}
It follows that
$\lambda_{\min}(\bSigma_t)\geqslant c$,
 and
$\lambda_{\max}(\bSigma_t)\leqslant C$.
Therefore,
$(\bmu_t^T,\vecth_s(\bSigma_t)^T)^T\in\bTheta$,
proving that $\bTheta$ is convex.

To prove closedness, let 
$\{(\bmu_n^T,\vecth_s(\bSigma_n)^T)^T\}$ 
be a convergent sequence in $\bTheta$, with
\[
(\bmu_n^T,\vecth_s(\bSigma_n)^T)^T
\rightarrow
(\bmu^T,\vecth_s(\bSigma)^T)^T
\quad \text{in } \mathbb R^{d + d(d+1)/2}.
\]
We show that $(\bmu^T,\vecth_s(\bSigma)^T)^T \in \bTheta$.
Since the Euclidean norm is continuous and $\|\bmu_n\|\leqslant M$ for all $n$, 
passing to the limit gives $\|\bmu\|\leqslant M$.
Moreover, $\vecth_s(\bSigma_n)\rightarrow\vecth_s(\bSigma)$ implies that
$\bSigma_n\rightarrow\bSigma$ in the matrix entries.
Since each $\bSigma_n$ is symmetric, the limit $\bSigma$ is also symmetric.
The eigenvalues of symmetric matrices depend continuously on the matrix entries.
Hence,
\begin{equation*}
\lambda_{\min}(\bSigma_n)\rightarrow\lambda_{\min}(\bSigma),
\qquad
\lambda_{\max}(\bSigma_n)\rightarrow\lambda_{\max}(\bSigma).
\end{equation*}
Because $c\leqslant\lambda_{\min}(\bSigma_n)\leqslant\lambda_{\max}(\bSigma_n)\leqslant C$
for all $n$, continuity implies
\begin{equation*}
c\leqslant\lambda_{\min}(\bSigma)\leqslant\lambda_{\max}(\bSigma)\leqslant C.
\end{equation*}
Thus $(\bmu^T,\vecth_s(\bSigma)^T)^T\in\bTheta$, showing that $\bTheta$ is closed.

Moreover, note that for any 
$\btheta=(\bmu^T,\vecth_s(\bSigma)^T)\in\bTheta$ we have 
$\|\bmu\|\leqslant M$ and, since 
$c\leqslant\lambda_{\min}(\bSigma)\leqslant\lambda_{\max}(\bSigma)\leqslant C$,
\begin{equation*}
\|\vect(\bSigma)\|
= \Big(\sum_{j=1}^{d}\lambda_j(\bSigma)^2\Big)^{1/2}
\leqslant \sqrt{d}\,\lambda_{\max}(\bSigma)
\leqslant \sqrt{d}\,C.
\end{equation*}
Hence,
\begin{equation*}
\|\btheta\|^2
= \|\bmu\|^2+\|\vect(\bSigma)\|^2
\leqslant M^2 + d\,C^2,
\end{equation*}
which proves that $\bTheta$ is bounded. Therefore, $\bTheta$ is closed and bounded, and by the Heine-Borel theorem, it is also compact.
\end{proof}

\begin{lemma}
\label{prop:projection-btheta}
For any $\bt=(\ba^T,\vecth_s(\bL)^T)^T\in 
\mathbb R^{d + d(d+1)/2}$, 
define
\[
\Pi_{\bTheta}(\bt)
=
(\tilde\bmu^T,\vecth_s(\btSigma)^T)^T.
\]
Then $\Pi_{\bTheta}$ is the metric projection of $\bt$ onto $\bTheta$
with respect to the norm $\|\cdot\|$.
\end{lemma}

\begin{proof}
We seek the metric projection of 
$\bt=(\ba^T,\vecth_s(\bL)^T)^T$ onto $\bTheta$ with respect to $\|\cdot\|$, i.e.,
\begin{equation*}
\min_{\bp=(\bmu^T,\vecth_s(\bSigma)^T)^T\in\bTheta}
\left\|
\bp
-
(\ba^T,\vecth_s(\bL)^T)^T
\right\|^2.
\end{equation*}
By definition of $\|\cdot\|$, this equals
\begin{equation*}
\min_{(\bmu^T,\vecth_s(\bSigma)^T)^T\in\bTheta}
\bigl(
\|\bmu-\ba\|^2
+
\|\vect(\bSigma)-\vect(\bL)\|^2
\bigr).
\end{equation*}
Since $\bTheta$ is the Cartesian product
\[
\{\bmu:\|\bmu\|\leqslant M\}
\times
\{\vecth_s(\bSigma):\bSigma\in\mathbb S^d, c \leqslant \lambda_{\min}(\bSigma)
\leqslant \lambda_{\max}(\bSigma)\leqslant C\},
\]
and the objective is separable, the minimization splits into two independent parts.
The problem
\begin{equation*}
\min_{\bmu:\|\bmu\|\leqslant M}\ \|\bmu-\ba\|^2
\end{equation*}
is the Euclidean projection of $\ba$ onto the ball of radius $M$ in $\mathbb R^d$,
which yields the stated $\tilde\bmu$.
The problem
\begin{equation*}
\min_{\vecth_s(\bSigma):c \leqslant \lambda_{\min}(\bSigma)\leqslant 
\lambda_{\max}(\bSigma)\leqslant C}
\ \|\vect(\bSigma)-\vect(\bL)\|^2
\end{equation*}
is the projection of $\bL$ onto the set of symmetric matrices whose eigenvalues lie in $[c,C]$.
Let $\bL=\bQ\,\mathrm{diag}(\lambda_1,\ldots,\lambda_d)\bQ^T$
be the spectral decomposition of $\bL$, where $\bQ$ is orthogonal.
For any $\bSigma=\bU\,\mathrm{diag}(y_1,\ldots,y_d)\bU^T$
satisfying $c\leqslant y_i\leqslant C$,
the Hoffman-Wielandt inequality (and the unitary invariance of the Frobenius norm) gives
\begin{equation*}
\|\vect(\bSigma)-\vect(\bL)\|^2
\geqslant
\sum_{i=1}^d (y_i-\lambda_{\pi(i)})^2,
\end{equation*}
for some permutation $\pi$, with equality if and only if $\bU=\bQ$.
Hence the optimal choice is $\bU=\bQ$, and the minimization reduces to
$d$ independent scalar problems
\begin{equation*}
\min_{c\leqslant y_j\leqslant C}(y_j-\lambda_j)^2,
\qquad j=1,\ldots,d,
\end{equation*}
whose unique solutions are
$y_j=\tilde\lambda_j=\min\{\max\{\lambda_j,c\},C\}$.
Therefore,
$\btSigma=\bQ\bLambda_c\bQ^T$.

Since both subproblems are strictly convex and are defined on closed convex sets (Lemma~\ref{prop:closed}), the solution 
$(\tilde\bmu,\btSigma)$ is unique.
Thus
$\Pi_{\bTheta}(\bt)
=
(\tilde\bmu^T,\vecth_s(\btSigma)^T)^T$
is the metric projection of $\bt$ onto $\bTheta$ with respect to the norm $\|\cdot\|$.
\end{proof}

\begin{lemma}
\label{lem:proj-firm}
Then, for all
$\bu,\bv\in\mathbb R^{d+d(d+1)/2}$,
\begin{equation}\label{eq:firmH}
  \|\Pi_{\bTheta}(\bu)-\Pi_{\bTheta}(\bv)\|^2
  \;\leqslant\;
  \big\langle \Pi_{\bTheta}(\bu)-\Pi_{\bTheta}(\bv),\, \bu-\bv\big\rangle.
\end{equation}
In particular, $\Pi_{\bTheta}$ is $1$-Lipschitz with respect to $\|\cdot\|$,
that is,
\begin{equation}\label{eq:nonexpH}
  \|\Pi_{\bTheta}(\bu)-\Pi_{\bTheta}(\bv)\|
  \;\leqslant\; \|\bu-\bv\|.
\end{equation}
\end{lemma}

\begin{proof}
The characterization of the metric projection in a Hilbert space states that
\[
\big\langle \bu-\Pi_{\bTheta}(\bu),\, \bz-\Pi_{\bTheta}(\bu)\big\rangle \leqslant 0
\quad \forall \bz\in\bTheta,
\]
and similarly for $\bv$.  
Taking $\bz=\Pi_{\bTheta}(\bv)$ in the first inequality and
$\bz=\Pi_{\bTheta}(\bu)$ in the second yields
\[
\big\langle \bu-\Pi_{\bTheta}(\bu),\, \Pi_{\bTheta}(\bv)-\Pi_{\bTheta}(\bu)\big\rangle \leqslant 0,
\]
\[
\big\langle \bv-\Pi_{\bTheta}(\bv),\, \Pi_{\bTheta}(\bu)-\Pi_{\bTheta}(\bv)\big\rangle \leqslant 0.
\]
Adding these two inequalities yields
\begin{align*}
0
&\;\geqslant\; \langle \bu-\Pi_{\bTheta}(\bu),\, \Pi_{\bTheta}(\bv)-\Pi_{\bTheta}(\bu)\rangle + \langle \bv-\Pi_{\bTheta}(\bv),\, \Pi_{\bTheta}(\bu)-\Pi_{\bTheta}(\bv)\rangle \\
&\;=\; \langle \bu-\Pi_{\bTheta}(\bu),\, \Pi_{\bTheta}(\bv)-\Pi_{\bTheta}(\bu)\rangle - \langle \bv-\Pi_{\bTheta}(\bv),\, \Pi_{\bTheta}(\bv)-\Pi_{\bTheta}(\bu)\rangle \\
&\;=\; \langle (\bu-\bv) - (\Pi_{\bTheta}(\bu)-\Pi_{\bTheta}(\bv)),\, \Pi_{\bTheta}(\bv)-\Pi_{\bTheta}(\bu)\rangle \\
&\;=\; \langle \bu-\bv,\, \Pi_{\bTheta}(\bv)-\Pi_{\bTheta}(\bu)\rangle - \|\Pi_{\bTheta}(\bu)-\Pi_{\bTheta}(\bv)\|^2 .
\end{align*}
Rearranging yields~\eqref{eq:firmH}.  

Finally, applying Cauchy-Schwarz with respect to $\langle\cdot,\cdot\rangle$,
\[
\|\Pi_{\bTheta}(\bu)-\Pi_{\bTheta}(\bv)\|^2
\leqslant
\|\Pi_{\bTheta}(\bu)-\Pi_{\bTheta}(\bv)\|\,\|\bu-\bv\|,
\]
which implies~\eqref{eq:nonexpH}.
\end{proof}

\begin{lemma}
\label{prop:proj-implies-zero-btheta}
Fix $\btheta^*\in\operatorname{int}(\bTheta)$ and 
$\bg\in\mathbb R^{d+d(d+1)/2}$.
If
$\Pi_{\bTheta}(\btheta^*+\bg)=\btheta^*$,
then $\bg=\bzero$.
\end{lemma}

\begin{proof}
Since $\Pi_{\bTheta}$ is the metric projection with respect to
$\|\cdot\|$, the projection optimality condition states that for any
$\bu\in\mathbb R^{d+d(d+1)/2}$ and $\bx\in\bTheta$, such that
$\bx=\Pi_{\bTheta}(\bu)$ implies $\big\langle \bu-\bx,\; \bz-\bx\big\rangle \leqslant 0$, for each  $\bz\in\bTheta$.
Apply this with $\bu=\btheta^*+\bg$ and $\bx=\btheta^*$ to obtain
$\big\langle \bg,\; \bz-\btheta^*\big\rangle \leqslant 0$
$ \forall\,\bz\in\bTheta$.
Since $\btheta^*\in\operatorname{int}(\bTheta)$, there exists a $\delta>0$
such that for every $\bh$ with $\|\bh\|<\delta$ we have
$\btheta^*+\bh\in\bTheta$.
Fix any $\bzero \ne \bv \in\mathbb R^{d+d(d+1)/2}$.
For any $t\in(0,\delta/\|\bv\|)$,
 $\bz=\btheta^*+t \bv\in\bTheta$,
and hence
$\big\langle \bg,\; t \bv\big\rangle \leqslant 0$.
Since $t>0$, we obtain
$\big\langle \bg,\bv\big\rangle \leqslant 0$.
Repeating the argument with $-\bv$ in place of $\bv$ yields
$\big\langle \bg,\bv\big\rangle \geqslant 0$.
Therefore,
$\big\langle \bg,\bv\big\rangle=0$ for any $\bzero \ne \bv\in\mathbb R^{d+d(d+1)/2}$.
By positive definiteness of the inner product,
this implies $\bg=\bzero$.
\end{proof}

\newcounter{oldprop}
\setcounter{oldprop}{\value{proposition}}

\setcounter{proposition}{\getrefnumber{prop:2}-1}
\begin{proposition}[Restated]
\label{prop:2_supp}
Assume that the following conditions hold:
\begin{enumerate}[label=(A\arabic*)]
\item there exists a $\boeta_0 \in \bH$ such that $\boheta_n \rightarrow_p \boeta_0$.

\item $\hpi(\btheta,\boheta_n,n)$ converges uniformly in probability to a
limit $\pi(\btheta,\boeta_0)$ over $\bTheta$, that is,
\begin{equation*}
  \sup_{\btheta \in \bTheta}
  \bigl\| \hpi(\btheta,\boheta_n,n) - \pi(\btheta,\boeta_0) \bigr\|
  \;\rightarrow_p\; 0,
\end{equation*}
\item $\pi(\cdot,\boeta_0):\bTheta\rightarrow\mathbb R^{d+d(d+1)/2}$
is continuous on $\bTheta$ and for any $\eps>0$,
\[
\inf_{\btheta\in\bTheta:\ \|\btheta-\btheta_0\|\geqslant\eps}
\|\pi(\btheta,\boeta_0)-\pi(\btheta_0,\boeta_0)\|>0.
\]
\end{enumerate}

Then any sequence $\bhtheta_n$ satisfying
\[
\bigl\|\hpi_n(\boheta_n)-\pibar(\bhtheta_n,\boheta_n,n)\bigr\|  \rightarrow_p0,
\]
is consistent for $\btheta_0$.
\end{proposition}
\setcounter{proposition}{\value{oldprop}}
\begin{proof}

The proof is obtained by verifying the conditions of Proposition~2.1 of \citetsupp{newey1994large}. Let
\[
\widehat Q_n(\btheta)
=
\bigl\|\hpi_n(\boheta_n)-\pibar(\btheta,\boheta_n,n)\bigr\|^2
\qquad\text{and}\qquad
Q(\btheta)
=
\bigl\|\pi(\btheta_0,\boeta_0)-\pi(\btheta,\boeta_0)\bigr\|^2.
\]
By assumption,
\[
\widehat Q_n(\bhtheta_n)
=
\bigl\|\hpi_n(\boheta_n)-\pibar(\bhtheta_n,\boheta_n,n)\bigr\|^2
\rightarrow_p 0.
\]
Since $\widehat Q_n(\btheta)\geqslant 0$ for all $\btheta\in\bTheta$, we have
\[
0 \leqslant \inf_{\btheta\in\bTheta}\widehat Q_n(\btheta)
\leqslant \widehat Q_n(\bhtheta_n),
\]
and therefore
\[
\widehat Q_n(\bhtheta_n)
\leqslant
\inf_{\btheta\in\bTheta}\widehat Q_n(\btheta)
+
o_p(1).
\]
Hence $\bhtheta_n$ is an approximate minimizer 
of $\widehat Q_n$
in the sense of \citetsupp{newey1994large}.
Theorem~2.1 of \citetsupp{newey1994large} states 
that if $Q(\btheta)$ satisfies:
\begin{enumerate}[label=(\roman*)]
  \item $Q$ is continuous on $\bTheta$,
  \item $Q$ is uniquely minimized at $\btheta_0$,
  \item $\bTheta$ is compact,
  \item $\widehat Q_n$ converges uniformly in probability to $Q$, that is,
  \[
  \sup_{\btheta\in\bTheta}
  \bigl|
  \widehat Q_n(\btheta)-Q(\btheta)
  \bigr|
  \;\rightarrow_p\;0,
  \]
\end{enumerate}
then $\bhtheta_n\rightarrow_p\btheta_0$.

By assumption,
\[
\sup_{\btheta\in\bTheta}
\bigl\|\hpi(\btheta,\boheta_n,n)-\pi(\btheta,\boeta_0)\bigr\|
\;\rightarrow_p\;0,
\]
and hence also
\[
\sup_{\btheta\in\bTheta}
\bigl\|\pibar(\btheta,\boheta_n,n)-\pi(\btheta,\boeta_0)\bigr\|
\;\rightarrow_p\;0.
\]
Since $\pi(\cdot,\boeta_0)$ is continuous on $\bTheta$, the objective function
\[
Q(\btheta)
=
\bigl\|\pi(\btheta_0,\boeta_0)-\pi(\btheta,\boeta_0)\bigr\|^2,
\]
is continuous on $\bTheta$, which establishes condition~(i).

Since $Q(\btheta)=0$ if and only if $\pi(\btheta,\boeta_0)=\pi(\btheta_0,\boeta_0)$ and, by the
separation condition, this can occur only when $\btheta=\btheta_0$, it follows
that $Q$ is uniquely minimized at $\btheta_0$, which establishes condition~(ii).

By Lemma~\ref{prop:closed}, $\bTheta$ is compact, hence (iii) holds. 

Consider the absolute value of the difference between
$\widehat Q_n(\btheta)$ and $Q(\btheta)$, that is,
\begin{align*}
\bigl|\widehat Q_n(\btheta) - Q(\btheta)\bigr|
&=
\Bigl|
\bigl\|\hpi_n(\boheta_n) - \pibar(\btheta,\boheta_n,n)\bigr\|^2
-
\bigl\|\pi(\btheta_0,\boeta_0)-\pi(\btheta,\boeta_0)\bigr\|^2
\Bigr| \\[2mm]
&=
\Bigl|
\bigl\|
\big[\hpi_n(\boheta_n)-\pi(\btheta_0,\boeta_0)\big]
+
\big[\pi(\btheta_0,\boeta_0)-\pi(\btheta,\boeta_0)\big]  \\
&\hspace{3.5cm}
+
\big[\pi(\btheta,\boeta_0)-\pibar(\btheta,\boheta_n,n)\big]
\bigr\|^2
-
\bigl\|\pi(\btheta_0,\boeta_0)-\pi(\btheta,\boeta_0)\bigr\|^2
\Bigr| \\[2mm]
&=
\Bigl|
\|\hpi_n(\boheta_n)-\pi(\btheta_0,\boeta_0)\|^2
+
\|\pi(\btheta,\boeta_0)-\pibar(\btheta,\boheta_n,n)\|^2  \\
&\quad
+2\big(\hpi_n(\boheta_n)-\pi(\btheta_0,\boeta_0)\big)^{\!T}
   \big(\pi(\btheta_0,\boeta_0)-\pi(\btheta,\boeta_0)\big) \\[1mm]
&\quad
+2\big(\pi(\btheta,\boeta_0)-\pibar(\btheta,\boheta_n,n)\big)^{\!T}
   \big(\pi(\btheta_0,\boeta_0)-\pi(\btheta,\boeta_0)\big) \\[1mm]
&\quad
+2\big(\hpi_n(\boheta_n)-\pi(\btheta_0,\boeta_0)\big)^{\!T}
   \big(\pi(\btheta,\boeta_0)-\pibar(\btheta,\boheta_n,n)\big)
\Bigr| \\[2mm]
&\leqslant
\|\hpi_n(\boheta_n)-\pi(\btheta_0,\boeta_0)\|^2
+
\|\pi(\btheta,\boeta_0)-\pibar(\btheta,\boheta_n,n)\|^2 \\[1mm]
&\quad
+2\|\hpi_n(\boheta_n)-\pi(\btheta_0,\boeta_0)\|
  \,\|\pi(\btheta_0,\boeta_0)-\pi(\btheta,\boeta_0)\| \\[1mm]
&\quad
+2\|\pi(\btheta,\boeta_0)-\pibar(\btheta,\boheta_n,n)\|
  \,\|\pi(\btheta_0,\boeta_0)-\pi(\btheta,\boeta_0)\| \\[1mm]
&\quad
+2\|\hpi_n(\boheta_n)-\pi(\btheta_0,\boeta_0)\|
  \,\|\pi(\btheta,\boeta_0)-\pibar(\btheta,\boheta_n,n)\| .
\end{align*}
Therefore,
\begin{align*}
\sup_{\btheta\in\bTheta}
\bigl|\widehat Q_n(\btheta)-Q(\btheta)\bigr|
&\leqslant
\|\hpi_n(\boheta_n)-\pi(\btheta_0,\boeta_0)\|^2
+
\sup_{\btheta\in\bTheta}
\|\pi(\btheta,\boeta_0)-\pibar(\btheta,\boheta_n,n)\|^2 \\[1mm]
&\quad
+2\|\hpi_n(\boheta_n)-\pi(\btheta_0,\boeta_0)\|
 \sup_{\btheta\in\bTheta}
 \|\pi(\btheta_0,\boeta_0)-\pi(\btheta,\boeta_0)\| \\[1mm]
&\quad
+2\sup_{\btheta\in\bTheta}
 \|\pi(\btheta,\boeta_0)-\pibar(\btheta,\boheta_n,n)\|
 \sup_{\btheta\in\bTheta}
 \|\pi(\btheta_0,\boeta_0)-\pi(\btheta,\boeta_0)\| \\[1mm]
&\quad
+2\|\hpi_n(\boheta_n)-\pi(\btheta_0,\boeta_0)\|
 \sup_{\btheta\in\bTheta}
 \|\pi(\btheta,\boeta_0)-\pibar(\btheta,\boheta_n,n)\| .
\end{align*}
Since $\pi(\cdot,\boeta_0)$ is continuous on $\bTheta$ and $\bTheta$ is compact,
it follows that
$\sup_{\btheta\in\bTheta}
\|\pi(\btheta_0,\boeta_0)-\pi(\btheta,\boeta_0)\|$
is bounded.
Therefore, since
\[
\sup_{\btheta\in\bTheta}
\|\bar\pi(\btheta,\boheta_n,n)-\pi(\btheta,\boeta_0)\|
\;\rightarrow_p\; 0,
\]
it follows in particular that, at $\btheta=\btheta_0$,
\[
\|\bar\pi(\btheta_0,\boheta_n,n)-\pi(\btheta_0,\boeta_0)\|
\;\rightarrow_p\; 0.
\]
Because the observed sample $(\bx_i,\by_i)$ is a random sample from
$F_{\btheta_0}$, the observed auxiliary estimator $\hpi_n(\boheta_n)$
has the same distribution as $\bar\pi(\btheta_0,\boheta_n,n)$.
Consequently,
\[
\|\hpi_n(\boheta_n)-\pi(\btheta_0,\boeta_0)\|
\;\rightarrow_p\; 0,
\]
and, thus,
\[
\sup_{\btheta\in\bTheta}
\bigl|\widehat Q_n(\btheta)-Q(\btheta)\bigr|
\rightarrow_p 0,
\]
which establishes condition~(iv).
\end{proof}

\vspace{2mm}
We define $T_n(\btheta) = \hpi_n(\boheta_n) +
\Big[\btheta - \pibar(\btheta,\boheta_n,n)\Big]$, 
and make the following assumptions:
\begin{enumerate}
\item[(B1)] 
The mapping $T_n$
is uniformly $L$-Lipschitz with $L<1$, that is,
\[
\|T_n(\btheta_1)-T_n(\btheta_2)\|
\leqslant
L\|\btheta_1-\btheta_2\|
\qquad
\forall\,\btheta_1,\btheta_2\in\bTheta.
\]
\item[(B2)]
$\bhtheta_n\in\operatorname{int}(\bTheta)$.
\end{enumerate} 

Assumption~(B1) requires $T_n$ to be 
globally contractive on $\bTheta$. Since the 
projection onto $\bTheta$ is nonexpansive, 
this also makes the projected II update a
contraction and guarantees convergence from 
any starting value. Indeed, defining
$\ba_n(\btheta)=\pibar(\btheta,\boheta_n,n)
-\btheta$, we have $T_n(\btheta)=
\hpi_n(\boheta_n)-\ba_n(\btheta)$,
so (B1) is equivalent to requiring $\ba_n$ 
to be Lipschitz with constant strictly 
smaller than one.

Assumption~(B2) requires the fixed 
point of the projected iteration to lie in 
the interior of $\bTheta$. This condition
ensures that the projection is locally 
inactive at the solution and, consequently, 
that the fixed point satisfies the original
II matching equation rather than being a 
boundary fixed point created by the 
projection.

\setcounter{oldprop}{\value{proposition}}
\setcounter{proposition}{\getrefnumber{pro:proj-fp-contraction}-1}

\vspace{2mm}
\begin{proposition}[Restated]
\label{prop:proj-fp-contraction_supp}
Under conditions B1--B2, $\bhtheta_n$ is 
unique and the sequence $\bhtheta_n^{(\ell)}$ 
converges in norm to $\bhtheta_n$ with linear
rate for every $\bhtheta_n^{(0)}\in\bTheta$, 
that is
 \[
\|\bhtheta_n^{(\ell)}-\bhtheta_n\|
\leqslant
L^\ell
\|\bhtheta_n^{(0)}-\bhtheta_n\|.
\]
\end{proposition}

\setcounter{proposition}{\value{oldprop}}

\begin{proof}
Define $F(\btheta):=\Pi_{\bTheta}(T_n(\btheta))$ for $\btheta\in\bTheta$.
By Lemma~\ref{lem:proj-firm}, the metric projection is nonexpansive, hence
for any $\btheta_1,\btheta_2\in\bTheta$,
\[
\|F(\btheta_1)-F(\btheta_2)\|
=
\|\Pi_{\bTheta}(T_n(\btheta_1))-\Pi_{\bTheta}(T_n(\btheta_2))\|
\leqslant
\|T_n(\btheta_1)-T_n(\btheta_2)\|
\leqslant
L\|\btheta_1-\btheta_2\|.
\]
Thus $F$ is a contraction with constant $L<1$.
Since $\bTheta$ is a closed subset of a Hilbert space
(Lemma~\ref{prop:closed}), it is complete. Therefore, by the
Banach fixed-point theorem, there exists a unique
$\bhtheta_n\in\bTheta$ such that
$F(\bhtheta_n)=\bhtheta_n$,
that is,
$\bhtheta_n=\Pi_{\bTheta}(T_n(\bhtheta_n))$.
Moreover, for every $\bhtheta_n^{(0)}\in\bTheta$, the iterates
$\bhtheta_n^{(\ell)}=F(\bhtheta_n^{(\ell-1)})$,
 $\ell=1,2,\dots$,
converge to $\bhtheta_n$ with linear rate
\[
\|\bhtheta_n^{(\ell)}-\bhtheta_n\|
\leqslant
L^\ell
\|\bhtheta_n^{(0)}-\bhtheta_n\|.
\]
Finally, as $\bhtheta_n\in\operatorname{int}(\bTheta)$,
then from $\bhtheta_n=\Pi_{\bTheta}(T_n(\bhtheta_n))$
Lemma~\ref{prop:proj-implies-zero-btheta}
implies $T_n(\bhtheta_n)=\bhtheta_n$.
\end{proof}

\setcounter{oldprop}{\value{proposition}}

\begin{proposition}
\label{pro:II_WI}
Let the tuning parameter estimator $\boheta_n\in\bH$, computed on the observed sample,  satisfies $\boheta_n-\boeta_0=O_p(n^{-1/2})$, with $\boeta_0\in\interior(\bH)
$.
Assume
\begin{equation}\label{eq:asslemma5_eta}
\sup_{\btheta \in \bTheta}
\|\hpi(\btheta,\boheta_n,n)-\pi(\btheta,\boeta_0)\|
=O_p(n^{-1/2}),
\end{equation}
and 
\begin{equation}\label{eq:simerr_eta}
\sup_{\btheta \in \bTheta}
\|\bar\pi(\btheta,\boheta_n,n)-\pi(\btheta,\boheta_n)\|
=o_p(n^{-1/2}).
\end{equation}
Assume that 
$\pi:\bTheta\times\bH\rightarrow\mathbb R^{d+d(d+1)/2}$ is continuous on
$\bTheta\times\bH$ and that, for any $\eps>0$,
\[
\inf_{\btheta\in\bTheta:\ \|\btheta-\btheta_0\|\geqslant\eps}
\|\pi(\btheta,\boeta_0)-\pi(\btheta_0,\boeta_0)\|>0,
\]
with $\btheta_0\in\interior(\bTheta)$.
Moreover, assume that $\pi$ is differentiable in $\btheta$ and $\boeta$ at
$(\btheta_0,\boeta_0)$ with full-rank Jacobian
$
\bA_0
=
\left.
\frac{\partial\, \pi(\btheta,\boeta)}{\partial \btheta}
\right|_{(\btheta,\boeta)=(\btheta_0,\boeta_0)}$ and $\bK_0:=\bA_0^{-1}$.

Then,
\[
\bhtheta_n-\btheta_0
=
\bK_0\big(\hpi_n(\boheta_n)-\pi(\btheta_0,\boheta_n)\big)
+\btr_n,
\qquad
\btr_n=o_p(n^{-1/2}).
\]
\end{proposition}
\begin{proof}

Using Proposition~\ref{prop:2}, we have $\bhtheta_n\rightarrow_p\btheta_0$.

Let,
\begin{equation*}
\zeta_n(\bhtheta_n,\boheta_n)
=
\bar\pi(\bhtheta_n,\boheta_n,n)-\pi(\bhtheta_n,\boheta_n).
\end{equation*}
By assumption,
\begin{equation*}
\sup_{\btheta\in\bTheta}\|\zeta_n(\btheta,\boheta_n)\|
=
\sup_{\btheta\in\bTheta}
\|\bar\pi(\btheta,\boheta_n,n)-\pi(\btheta,\boheta_n)\|
=o_p(n^{-1/2}),
\end{equation*}
and in particular $\zeta_n(\bhtheta_n,\boheta_n)=o_p(n^{-1/2})$.

By a first-order Taylor expansion of $\pi(\btheta,\boeta)$ at $(\btheta_0,\boeta_0)$,
we have
\begin{equation*}
\pi(\btheta,\boeta)
=
\pi(\btheta_0,\boeta_0)
+
\bA_0(\btheta-\btheta_0)
+
\bB_0(\boeta-\boeta_0)
+
\br(\btheta,\boeta),
\end{equation*}
where
\[
\bB_0
:=
\left.\frac{\partial\,\pi(\btheta,\boeta)}{\partial\boeta}\right|_{(\btheta,\boeta)=(\btheta_0,\boeta_0)},
\]
and
\[
\br(\btheta,\boeta)
:=
\pi(\btheta,\boeta)-\pi(\btheta_0,\boeta_0)
-\bA_0(\btheta-\btheta_0)-\bB_0(\boeta-\boeta_0).
\]
Since $\pi$ is differentiable at $(\btheta_0,\boeta_0)$, the remainder $\br(\btheta,\boeta)$ satisfies
\begin{equation}\label{eq:reminderdet_joint_sum1}
\frac{\|\br(\btheta,\boeta)\|}{\|\btheta-\btheta_0\|+\|\boeta-\boeta_0\|}\rightarrow 0
\quad\text{as }(\btheta,\boeta)\rightarrow(\btheta_0,\boeta_0).
\end{equation}
Define, for $(\btheta,\boeta)\neq(\btheta_0,\boeta_0)$,
\[
f(\btheta,\boeta)
:=
\frac{\|\br(\btheta,\boeta)\|}{\|\btheta-\btheta_0\|+\|\boeta-\boeta_0\|},
\qquad
f(\btheta_0,\boeta_0)=0.
\]
Due to \eqref{eq:reminderdet_joint_sum1}, $f$ is continuous at $(\btheta_0,\boeta_0)$ and
$f(\btheta_0,\boeta_0)=0$.
By considering the sequence $(\bhtheta_n,\boheta_n)\rightarrow_p(\btheta_0,\boeta_0)$,
the continuous mapping theorem \citep{van2000asymptotic} yields
\begin{equation*}
f(\bhtheta_n,\boheta_n)
=
\frac{\|\br(\bhtheta_n,\boheta_n)\|}{\|\bhtheta_n-\btheta_0\|+\|\boheta_n-\boeta_0\|}
\rightarrow_p0,
\end{equation*}
that is,
\begin{equation}\label{eq:reminaderop_joint_sum}
\br_n:=\br(\bhtheta_n,\boheta_n)
=
o_p\bigl(\|\bhtheta_n-\btheta_0\|+\|\boheta_n-\boeta_0\|\bigr).
\end{equation}

Then,
\begin{align*}
\mathbf 0
&=\hpi_n(\boheta_n)-\bar\pi(\bhtheta_n,\boheta_n,n)\\
&=\hpi_n(\boheta_n)-\pi(\bhtheta_n,\boheta_n)-\zeta_n(\bhtheta_n,\boheta_n)\\
&=\hpi_n(\boheta_n)-\pi(\btheta_0,\boeta_0)
-\bA_0(\bhtheta_n-\btheta_0)
-\bB_0(\boheta_n-\boeta_0)
-\br_n
-\zeta_n(\bhtheta_n,\boheta_n).
\end{align*}
Thus,
\[
\bhtheta_n-\btheta_0
=
\bK_0\big(\hpi_n(\boheta_n)-\pi(\btheta_0,\boeta_0)
-\bB_0(\boheta_n-\boeta_0))
-\bK_0\br_n
-\bK_0\zeta_n(\bhtheta_n,\boheta_n),
\]
Taking norms and using the triangle inequality,
\begin{align*}
\|\bhtheta_n-\btheta_0\|
&\leqslant
C\|\hpi_n(\boheta_n)-\pi(\btheta_0,\boeta_0)\|
+CB\|\boheta_n-\boeta_0\|
+C\|\br_n\|
+C\|\zeta_n(\bhtheta_n,\boheta_n)\|,
\end{align*}
with $C:=\sigma_{\max}(\bK_0)<\infty$ and $B:=\sigma_{\max}(\bB_0)<\infty$ by assumption and, where $\sigma_{\max}(\cdot)$ extracts the maximum singular value.

Moreover, since $\br_n = o_p\bigl(\|\bhtheta_n-\btheta_0\|+\|\boheta_n-\boeta_0\|\bigr)$, by the
definition of $o_p(\cdot)$, this means, for every $\eps>0$,
\begin{equation*}
\Pr\Big(
 \|\br_n\|
\leqslant
\eps\bigl(\|\bhtheta_n-\btheta_0\|+\|\boheta_n-\boeta_0\|\bigr)
\Big) \rightarrow 1.
\end{equation*}
On the event
\(
E_{n,\eps}
:=
\bigl\{\|\br_n\|\leqslant \eps(\|\bhtheta_n-\btheta_0\|+\|\boheta_n-\boeta_0\|)\bigr\}
\), we obtain
\begin{multline*}
\|\bhtheta_n-\btheta_0\|
\leqslant
C\|\hpi_n(\boheta_n)-\pi(\btheta_0,\boeta_0)\|
+CB\,\|\boheta_n-\boeta_0\|
+C\eps\|\bhtheta_n-\btheta_0\| 
+C\eps\|\boheta_n-\boeta_0\| \\
+C\|\zeta_n(\bhtheta_n,\boheta_n)\|.
\end{multline*}
Rearranging terms yields, on $E_{n,\eps}$,
\begin{equation}\label{eq:theta_rearranged}
(1-C\eps)\|\bhtheta_n-\btheta_0\|
\leqslant
C\|\hpi_n(\boheta_n)-\pi(\btheta_0,\boeta_0)\|
+C(B+\eps)\|\boheta_n-\boeta_0\|
+C\|\zeta_n(\bhtheta_n,\boheta_n)\|.
\end{equation}
Choose $\eps>0$ such that $C\eps<1$. Dividing both sides of
\eqref{eq:theta_rearranged} by $(1-C\eps)$, we obtain on $E_{n,\eps}$
\begin{equation}\label{eq:theta_bound_final}
\|\bhtheta_n-\btheta_0\|
\leqslant
C_1\|\hpi_n(\boheta_n)-\pi(\btheta_0,\boeta_0)\|
+C_2\|\boheta_n-\boeta_0\|
+C_1\|\zeta_n(\bhtheta_n,\boheta_n)\|,
\end{equation}
for finite constants $C_1,C_2>0$.
Therefore, for any $t>0$,
\[
\{\|\bhtheta_n-\btheta_0\|>t\}
\subseteq
E_{n,\eps}^c
\;\cup\;
\Bigl\{
C_1\|\hpi_n(\boheta_n)-\pi(\btheta_0,\boeta_0)\|
+
C_2\|\boheta_n-\boeta_0\|
+
C_1\|\zeta_n(\bhtheta_n,\boheta_n)\|
>t
\Bigr\}.
\]
This implies
\[
\Pr(\|\bhtheta_n-\btheta_0\|>t)
\leqslant
\Pr(E_{n,\eps}^c)
+
\Pr\Bigl(
C_1\|\hpi_n(\boheta_n)-\pi(\btheta_0,\boeta_0)\|
+
C_2\|\boheta_n-\boeta_0\|
+
C_1\|\zeta_n(\bhtheta_n,\boheta_n)\|
>t
\Bigr).
\]
Since the bound holds for all $t>0$, it also holds for any deterministic
sequence $t=t_n>0$. In particular, since by assumption
\[
\sup_{\btheta \in \bTheta}
\|\hpi(\btheta,\boheta_n,n)-\pi(\btheta,\boeta_0)\|
=O_p(n^{-1/2}),
\]
it follows by evaluating at $\btheta_0$ and using that
the observed sample is a random sample from $F_{\btheta_0}$ that
\[
\|\hpi_n(\boheta_n)-\pi(\btheta_0,\boeta_0)\|=O_p(n^{-1/2}).
\]
Therefore, taking $t=Mn^{-1/2}$ and using
\[
\|\hpi_n(\boheta_n)-\pi(\btheta_0,\boeta_0)\|=O_p(n^{-1/2}),\qquad
\|\boheta_n-\boeta_0\|=O_p(n^{-1/2}),\qquad
\|\zeta_n(\bhtheta_n,\boheta_n)\|=o_p(n^{-1/2}),
\]
together with $\Pr(E_{n,\eps}^c)\rightarrow 0$, we conclude that
\begin{equation}\label{eq:ratehtheta}
\|\bhtheta_n-\btheta_0\|=O_p(n^{-1/2}).
\end{equation}
Moreover, since $\|\bhtheta_n-\btheta_0\|+\|\boheta_n-\boeta_0\|=O_p(n^{-1/2})$ and
$\|\br_n\|=o_p(\|\bhtheta_n-\btheta_0\|+\|\boheta_n-\boeta_0\|)$, we conclude that
$\br_n=o_p(n^{-1/2})$.

Define the combined remainder term
\[
\btr_n
:=
-\bK_0\br_n
-\bK_0\zeta_n(\bhtheta_n,\boheta_n).
\]
 By the triangle inequality,
\begin{equation*}
\sqrt{n}\,\|\btr_n\|
\leqslant
C\sqrt{n}\,\|\br_n\|
+
C\sqrt{n}\,\|\zeta_n(\bhtheta_n,\boheta_n)\|.
\end{equation*}
From above we have $\|\br_n\| = o_p(n^{-1/2})$ and
$\|\zeta_n(\bhtheta_n,\boheta_n)\| = o_p(n^{-1/2})$, hence
$\sqrt{n}\,\|\br_n\|\rightarrow_p0$ and
$\sqrt{n}\,\|\zeta_n(\bhtheta_n)\|\rightarrow_p0$. Therefore,
$\sqrt{n}\,\|\btr_n\|\rightarrow_p0$,
which is equivalent to
$\btr_n = o_p(n^{-1/2})$.
Note that
\[
\bhtheta_n-\btheta_0
=
\bK_0\big(\hpi_n(\boheta_n)-\pi(\btheta_0,\boeta_0)-\bB_0(\boheta_n-\boeta_0))
+\btr_n.
\]
Add and subtract $\pi(\btheta_0,\boheta_n)$:
\[
\hpi_n(\boheta_n)-\pi(\btheta_0,\boeta_0)-\bB_0(\boheta_n-\boeta_0)
=
\big(\hpi_n(\boheta_n)-\pi(\btheta_0,\boheta_n)\big)
+\br_{\eta,n},
\]
where
\[
\br_{\eta,n}
:=
\pi(\btheta_0,\boheta_n)-\pi(\btheta_0,\boeta_0)-\bB_0(\boheta_n-\boeta_0).
\]
Since $\pi(\btheta_0,\boeta)$ is differentiable at $\boeta_0$ with derivative
$\bB_0$, by definition of differentiability we have
\[
\frac{
\big\|
\pi(\btheta_0,\boeta)-\pi(\btheta_0,\boeta_0)-\bB_0(\boeta-\boeta_0)
\big\|
}{
\|\boeta-\boeta_0\|
}
\;\rightarrow\;0
\quad\text{as }\boeta\rightarrow\boeta_0.
\]
Define
\[
h(\boeta)
:=
\frac{
\big\|
\pi(\btheta_0,\boeta)-\pi(\btheta_0,\boeta_0)-\bB_0(\boeta-\boeta_0)
\big\|
}{
\|\boeta-\boeta_0\|
},
\qquad
h(\boeta_0):=0.
\]
Then $h$ is continuous at $\boeta_0$ and $h(\boeta_0)=0$. Since
$\boheta_n\rightarrow_p\boeta_0$, the   continuous mapping theorem \citep{van2000asymptotic} yields
$h(\boheta_n)\rightarrow_p 0$. Therefore,
\[
\|\br_{\eta,n}\|
=
\|\boheta_n-\boeta_0\|\,h(\boheta_n)
=
o_p(\|\boheta_n-\boeta_0\|).
\]
As $\boheta_n-\boeta_0=O_p(n^{-1/2})$, it follows that
$\br_{\eta,n}=o_p(n^{-1/2})$.
Thus,
\[
\bhtheta_n-\btheta_0
=
\bK_0\big(\hpi_n(\boheta_n)-\pi(\btheta_0,\boheta_n)\big)
+\btr_n^{\,*},
\qquad
\btr_n^{\,*}=o_p(n^{-1/2}),
\]
where
$\btr_n^{\,*}
:=
\btr_n+\bK_0\br_{\eta,n}$.

\end{proof}

\begin{lemma}\label{lem:ratio-uniform-eta}
Let $\bTheta$ and $\bH$ be compact, and let the tuning parameters $\boeta\in\bH$.
The observed sample $\bx_1,\ldots,\bx_n$ consists of i.i.d.\
observations drawn from $F_{\btheta_0}$, where
$\btheta_0\in\interior(\bTheta)$.
Let $\{F_{\btheta} : \btheta \in \bTheta\}$, where 
$\btheta=(\bmu^T,\vecth_s(\bSigma)^T)^T$,  be a $d$-dimensional
parametric family such that there exist i.i.d.\ random vectors 
$\bu_1,\ldots,\bu_n$ defined on $\mathcal U$ with known distribution $P_U$, not depending on $\btheta$,
with  $\bx_{i,\btheta}
=
G(\btheta,\bu_i)$, $i=1,\ldots,n$,
where $\E\|\bu_1\|^2<\infty$. Assume moreover, that there exists a measurable function
$m:\mathcal U\rightarrow[0,\infty)$ with $\E[m(\bu_1)^2]<\infty$
such that for all $\btheta_1,\btheta_2\in\bTheta$ and all
$\bu\in\mathcal U$,
\[
\|G(\btheta_1,\bu)-G(\btheta_2,\bu)\|
\leqslant
m(\bu)\,\|\btheta_1-\btheta_2\|.
\]
For each $(\btheta,\boeta)\in\bTheta\times\bH$, define
\begin{equation*}
 \bar a_n(\btheta,\boeta)
=
\frac{1}{n}\sum_{i=1}^n a(\bx_{i,\btheta},\boeta),
\qquad
\bar b_n(\btheta,\boeta)
=
\frac{1}{n}\sum_{i=1}^n b(\bx_{i,\btheta},\boeta),
\end{equation*}
with population targets 
$\mu_{a}(\btheta,\boeta)=\E[a(\bx_{i,\btheta},\boeta)]$ and 
$\mu_{b}(\btheta,\boeta)=\E[b(\bx_{i,\btheta},\boeta)]$, 
where the  function $a: \mathbb R^{d}\times \mathbb R^{r}\rightarrow\mathbb R^{d+d(d+1)/2}$ and $b:\mathbb R^{d}\times \mathbb R^{r}\rightarrow\mathbb R^{d+d(d+1)/2}$. Assume that there exists $c>0$ such that for all $(\btheta,\boeta)\in\bTheta\times\bH$, $\min_{ j}|(\mu_b)_j(\btheta,\boeta)| > c$.
Define
\[
E_n
:=
\Big\{
\inf_{(\btheta,\boeta)\in\bTheta\times\bH}
\ \min_{ j}
\big|(\bar b_n)_j(\btheta,\boeta)\big|
\geqslant c
\Big\},
\]
and assume that there exists $\alpha>1/2$ such that $\Pr(E_n^c)=O(n^{-\alpha})$
 as $n\rightarrow\infty$.
 The estimator $\hP(\btheta,\boeta,n)$ is defined componentwise, for $j=1,\dots,d+d(d+1)/2$, as
\[
\big(\hP(\btheta,\boeta,n)\big)_j
:=
\begin{cases}
\dfrac{(\bar a_n)_j(\btheta,\boeta)}
      {(\bar b_n)_j(\btheta,\boeta)},
&  \big|(\bar b_n)_j(\btheta,\boeta)\big|
\geqslant \tilde \delta, \\[6pt]
\dfrac{(\bar a_n)_j(\btheta,\boeta)}
      {\tilde \delta},
&  \big|(\bar b_n)_j(\btheta,\boeta)\big|
< \tilde \delta,
\end{cases}
\]
with $0<\tilde \delta \leqslant c$. Define
\begin{equation*}
P(\btheta,\boeta):=\mu_{a}(\btheta,\boeta)\oslash \mu_{b}(\btheta,\boeta).
\end{equation*}
Assume, $a$ is Lipschitz continuous, which means that for some $L_a<\infty$, and any $\bx_1,\bx_2\in\mathbb R^d$, and 
$\boeta_1,\boeta_2\in\bH$,
\begin{equation*}
\|a(\bx_1,\boeta_1)-a(\bx_2,\boeta_2)\|
\leqslant
L_a\bigl(\|\bx_1-\bx_2\|+\|\boeta_1-\boeta_2\|\bigr).
\end{equation*}
An analogous Lipschitz condition holds for $b(\bx,\boeta)$.
Moreover, assume that there exists an $M<\infty$ such that for all
$\bx\in\mathbb R^d$ and all $\boeta\in\bH$,
$\|a(\bx,\boeta)\|\leqslant M$,
$\|b(\bx,\boeta)\|\leqslant M$.

The tuning parameter estimator $\boheta_n$, computed on the observed sample,  satisfies $\boheta_n-\boeta_0=O_p(n^{-1/2})$. 
Further define
\[
\bar P(\btheta,\boeta,n):=\frac{1}{H}\sum_{h=1}^H \hP^{(h)}(\btheta,\boeta,n),
\]
where $\hP^{(h)}(\btheta,\boeta,n)$ is $\hP(\btheta,\boeta,n)$ computed on $H$ random samples $\{\bx^{(h)}_{1,\btheta},\ldots,\bx^{(h)}_{n,\btheta}\}_{h=1}^H$,  with law
$F_{\btheta}$. 
Assume that the number of Monte Carlo replications $H$ satisfies
$\sqrt{\frac{\log H}{H}}=o(n^{-1/2})$  as $n\rightarrow\infty$.

Then
\begin{equation*}
\sup_{\btheta\in\bTheta}\big\|
\hP(\btheta,\boheta_n,n)-P(\btheta,\boeta_0)
\big\| = O_p(n^{-1/2}),
\end{equation*}
and
\[
\sup_{\btheta\in\bTheta}
\left\|
\bar P(\btheta,\boheta_n,n)-P(\btheta,\boheta_n)
\right\|
=o_p(n^{-1/2}).
\]
\end{lemma}
\begin{proof}
Let $\btheta_i=(\bmu_i^T,\vecth_s(\bSigma_i)^T)^T$, $i=1,2$. Then, for any $\bu\in\mathcal U$,
\begin{equation*}
\|G(\btheta_1,\bu)-G(\btheta_2,\bu)\|
\leqslant m(\bu)\,\|\btheta_1-\btheta_2\|,
\end{equation*}
where 
$\E[m(U)^2]<\infty$ for $U\sim P_U$.
Define for any $(\btheta,\boeta)\in\bTheta\times\bH$, $f_{\btheta,\boeta}(\bu):=a\big(G(\btheta,\bu),\boeta\big)$,
with components $f_{\btheta,\boeta,j}(\bu)$, $j=1,\dots,d+d(d+1)/2$.
By the assumptions on $G$ and $a$, for every $\bu\in\mathcal U$ and $j$,
\begin{align*}
\bigl|f_{\btheta_1,\boeta_1,j}(\bu)-f_{\btheta_2,\boeta_2,j}(\bu)\bigr|
&\leqslant
\|f_{\btheta_1,\boeta_1}(\bu)-f_{\btheta_2,\boeta_2}(\bu)\|\\
&\leqslant L_a\bigl(\|G(\btheta_1,\bu)-G(\btheta_2,\bu)\|+\|\boeta_1-\boeta_2\|\bigr)\\
&\leqslant L_a\bigl(m(\bu)\|\btheta_1-\btheta_2\|+\|\boeta_1-\boeta_2\|\bigr)\\
&\leqslant m_G(\bu)\,\bigl(\|\btheta_1-\btheta_2\|^2+\|\boeta_1-\boeta_2\|^2\bigr)^{1/2},
\end{align*}
where $m_G(\bu):=L_a\sqrt{2}\,\max\{m(\bu),1\}$.
Moreover, $|f_{\btheta,\boeta,j}(\bu)|\leqslant M<\infty$ for all $(\btheta,\boeta,\bu)$.
Thus $\{f_{\btheta,\boeta,j}:(\btheta,\boeta)\in\bTheta\times\bH\}$ is a
Lipschitz-parametric class indexed by a compact, and thus bounded, subset of
$\mathbb R^{d+d(d+1)/2+r}$, with $\E[m_G(U)^2]<\infty$.
The boundedness assumption ensures that the class admits a finite (hence
square-integrable) envelope.
By Example~19.7 in \citet{van2000asymptotic} the bracketing entropy integral is finite,
and therefore this class is $P_{ U}$-Donsker by Theorem~19.5. 
Consequently, for each fixed
$j$,
\begin{equation}
\sup_{(\btheta,\boeta)\in\bTheta\times\bH}
| (\bar a_n(\btheta,\boeta))_j-(\mu_a(\btheta,\boeta))_j|
=O_p(n^{-1/2}).
\end{equation}
and, thus,
\begin{equation}\label{eq:barAmuA}
\sup_{(\btheta,\boeta)\in\bTheta\times\bH}
\|\bar a_n(\btheta,\boeta)-\mu_a(\btheta,\boeta)\|
=O_p(n^{-1/2}).
\end{equation}
An identical argument gives
\begin{equation}\label{eq:barBmuB}
\sup_{(\btheta,\boeta)\in\bTheta\times\bH}
\|\bar b_n(\btheta,\boeta)-\mu_b(\btheta,\boeta)\|
=O_p(n^{-1/2}).
\end{equation}

On the event $E_n$, for each $j$ and each $(\btheta,\boeta)\in\bTheta\times\bH$,
both pairs $
\big((\bar a_n)_j(\btheta,\boeta),(\bar b_n)_j(\btheta,\boeta)\big)$,
 and 
$\big((\mu_a)_j(\btheta,\boeta),(\mu_b)_j(\btheta,\boeta)\big)$
belong to the domain $\{(a,b):|b|\geqslant c\}$ as by assumption $\inf_{(\btheta,\boeta)\in\bTheta\times\bH}\min_j |(\mu_b)_j(\btheta,\boeta)|> c$. 
The function $g$ is continuously differentiable there with
gradient $\nabla g(a,b)=(1/b,- a /b^2)$, which is uniformly bounded since
$|b|\geqslant c$ and $|a|\leqslant M$. Thus $g$ is Lipschitz, thus, there exists an $L_g<\infty$
such that
\begin{equation}\label{eq:gLineq}
|g(a_1,b_1)-g(a_2,b_2)|
\leqslant L_g\big(|a_1-a_2|+|b_1-b_2|\big).
\end{equation}
Applying this to the points to
$(a_1,b_1)=((\bar a_n)_{j}(\btheta,\boeta),(\bar b_n)_{j}(\btheta,\boeta))$ and
$(a_2,b_2)=((\mu_{a})_{j}(\btheta,\boeta),(\mu_{b})_{j}(\btheta,\boeta))$ yields
\begin{equation*}
\big|\hP_{j}(\btheta,\boeta,n)-P_{j}(\btheta,\boeta)\big|
\leqslant
L_g\Big(\big|(\bar a_n)_{j}(\btheta,\boeta)-(\mu_{a})_{j}(\btheta,\boeta)\big|
+
\big|(\bar b_n)_{j}(\btheta,\boeta)-(\mu_{b})_{j}(\btheta,\boeta)\big|
\Big),
\end{equation*}
as $\hP_{j}(\btheta,\boeta,n)=g(a_1,b_1)$ and $P_{j}(\btheta,\boeta)=g(a_2,b_2)$ on $E_n$.
Taking the supremum over $(\btheta,\boeta)\in\bTheta\times\bH$, we obtain on the event $E_n$ that
$Y_n
\leqslant
L_g
Z_n$,
where
\[
Y_n :=
\sup_{(\btheta,\boeta)\in\bTheta\times\bH}
\|\hP(\btheta,\boeta,n)-P(\btheta,\boeta)\|,\quad Z_n :=
\sup_{(\btheta,\boeta)\in\bTheta\times\bH}
\big(
\|\bar a_n-\mu_a\|+\|\bar b_n-\mu_b\|
\big).
\]
Thus, we have
\[
\Pr\big(\{Y_n > M n^{-1/2}\}\cap E_n\big)\leqslant
\Pr(L_g Z_n > M n^{-1/2}).
\]
By \eqref{eq:barAmuA} and \eqref{eq:barBmuB}, $Z_n = O_p(n^{-1/2})$. Hence, for every $\eps>0$ there exists $M<\infty$ and $N_1$
such that for all $n\geqslant N_1$,
$\Pr(L_g Z_n > M n^{-1/2}) \leqslant \eps/2$.
Since by assumption $\Pr(E_n^c)\rightarrow 0$, by definition of convergence there exists
$N_2$ such that for all $n\geqslant N_2$,
$\Pr(E_n^c) \leqslant \eps/2$.
Moreover,
\begin{align*}
\Pr\big( Y_n > M n^{-1/2} \big)
&=
\Pr\big(\{Y_n > M n^{-1/2}\}\cap E_n\big)
+
\Pr\big(\{Y_n > M n^{-1/2}\}\cap E_n^c\big) \\
&\leqslant
\Pr\big(\{Y_n > M n^{-1/2}\}\cap E_n\big)
+
\Pr(E_n^c).
\end{align*}
Therefore, for all $n\geqslant N:=\max(N_1,N_2)$,
$\Pr\big( Y_n > M n^{-1/2} \big)
\leqslant
\eps/2 + \eps/2
=
\eps$.
Because $\eps>0$ is arbitrary, we conclude that 
\begin{equation}\label{eq:sumP}
\sup_{(\btheta,\boeta)\in\bTheta\times\bH}
\|\hP(\btheta,\boeta,n)-P(\btheta,\boeta)\|
=O_p(n^{-1/2}).
\end{equation}

Thus,
\begin{equation}\label{eq:inePhatetahat}
\sup_{\btheta\in\bTheta}
\|\hP(\btheta,\boheta_n,n)-P(\btheta,\boeta_0)\|
\leqslant
\sup_{\btheta\in\bTheta}
\|\hP(\btheta,\boheta_n,n)-\hP(\btheta,\boeta_0,n)\|
+
\sup_{\btheta\in\bTheta}
\|\hP(\btheta,\boeta_0,n)-P(\btheta,\boeta_0)\|.
\end{equation}
By \eqref{eq:sumP}, the second term on the right hand side is $O_p(n^{-1/2})$ .

By assumption, for any $\btheta\in\bTheta$, any $\boeta_1,\boeta_2\in\bH$, and any
$\bu\in\mathcal U$,
\begin{equation*}
\|a(G(\btheta,\bu),\boeta_1)-a(G(\btheta,\bu),\boeta_2)\|
\leqslant
L_a\,\|\boeta_1-\boeta_2\|,
\end{equation*}
and analogously,
\begin{equation*}
\|b(G(\btheta,\bu),\boeta_1)-b(G(\btheta,\bu),\boeta_2)\|
\leqslant
L_b\,\|\boeta_1-\boeta_2\|.
\end{equation*}
Averaging over $i$ preserves the Lipschitz constant, so
\begin{equation*}
\|\bar a_n(\btheta,\boeta_1)-\bar a_n(\btheta,\boeta_2)\|
\leqslant L_a\|\boeta_1-\boeta_2\|,
\qquad
\|\bar b_n(\btheta,\boeta_1)-\bar b_n(\btheta,\boeta_2)\|
\leqslant L_b\|\boeta_1-\boeta_2\|.
\end{equation*}
This means  that
\begin{equation*}
\big|(\bar a_n)_{j}(\btheta,\boeta_1)-(\bar a_n)_{j}(\btheta,\boeta_2)\big|
\leqslant L_a\|\boeta_1-\boeta_2\|,
\qquad
\big|(\bar b_n)_{j}(\btheta,\boeta_1)-(\bar b_n)_{j}(\btheta,\boeta_2)\big|
\leqslant L_b\|\boeta_1-\boeta_2\|.
\end{equation*}
On $E_n$, apply \eqref{eq:gLineq} to 
$(a_1,b_1)=((\bar a_n)_{j}(\btheta,\boeta_1),(\bar b_n)_{j}(\btheta,\boeta_1))$ and
$(a_2,b_2)=((\bar a_n)_{j}(\btheta,\boeta_2),(\bar b_n)_{j}(\btheta,\boeta_2))$,
we obtain
\begin{equation*}
\big|\hP_{j}(\btheta,\boeta_1,n)-\hP_{j}(\btheta,\boeta_2,n)\big|
\leqslant
L_g\Big(
\big|(\bar a_n)_{j}(\btheta,\boeta_1)-(\bar a_n)_{j}(\btheta,\boeta_2)\big|
+
\big|(\bar b_n)_{j}(\btheta,\boeta_1)-(\bar b_n)_{j}(\btheta,\boeta_2)\big|
\Big).
\end{equation*}
Thus, for each $j$,
\begin{equation*}
\big|\hP_{j}(\btheta,\boeta_1,n)-\hP_{j}(\btheta,\boeta_2,n)\big|
\leqslant
L_g(L_a+L_b)\,\|\boeta_1-\boeta_2\|.
\end{equation*}
Since there are finitely many entries,
there exists a finite constant $C_P$ such that for any $\btheta\in\bTheta,\ \boeta_1,\boeta_2\in\bH$,
\begin{equation*}
\|\hP(\btheta,\boeta_1,n)-\hP(\btheta,\boeta_2,n)\|
\leqslant C_P\|\boeta_1-\boeta_2\|.
\end{equation*}
Taking the supremum over $\btheta\in\bTheta$ preserves the bound for any $\boeta_1,\boeta_2\in\bH$
\begin{equation*}
\sup_{\btheta\in\bTheta}
\|\hP(\btheta,\boeta_1,n)-\hP(\btheta,\boeta_2,n)\|
\leqslant C_P\|\boeta_1-\boeta_2\|.
\end{equation*}
Applying this with $\boeta_1=\boheta_n$ and $\boeta_2=\boeta_0$ gives on $E_n$
\begin{equation*}
W_n \leqslant C_P\|\boheta_n-\boeta_0\|.
\end{equation*}
where
\[
W_n
:=
\sup_{\btheta\in\bTheta}
\|\hP(\btheta,\boheta_n,n)-\hP(\btheta,\boeta_0,n)\|.
\]
Therefore, for any $M>0$,
\[
\Pr\big(\{W_n > M n^{-1/2}\}\cap E_n\big)
\leqslant
\Pr\big(C_P\|\boheta_n-\boeta_0\| > M n^{-1/2}\big).
\]
Since $\|\boheta_n-\boeta_0\|=O_p(n^{-1/2})$,
for every $\eps>0$ there exist $M<\infty$ and $N_1$ such that
for all $n\geqslant N_1$,
$\Pr\big(C_P\|\boheta_n-\boeta_0\| > M n^{-1/2}\big)
\leqslant \eps/2$.
Moreover, since $\Pr(E_n^c)\rightarrow 0$, there exists $N_2$ such that
for all $n\geqslant N_2$,
$\Pr(E_n^c)\leqslant \eps/2$.
Finally,
\begin{align*}
\Pr\big(W_n > M n^{-1/2}\big)
&=
\Pr\big(\{W_n > M n^{-1/2}\}\cap E_n\big)
+
\Pr\big(\{W_n > M n^{-1/2}\}\cap E_n^c\big) \\
&\leqslant
\Pr\big(\{W_n > M n^{-1/2}\}\cap E_n\big)
+
\Pr(E_n^c).
\end{align*}
Hence, for all $n\geqslant N:=\max(N_1,N_2)$,
$\Pr\big(W_n > M n^{-1/2}\big)
\leqslant
\eps/2+\eps/2
=
\eps$.
Because $\eps>0$ is arbitrary, we conclude that
\begin{equation}\label{eq:comb2}
\sup_{\btheta\in\bTheta}
\|\hP(\btheta,\boheta_n,n)-\hP(\btheta,\boeta_0,n)\|
=O_p(n^{-1/2}).
\end{equation}
By using \eqref{eq:comb2}
in \eqref{eq:inePhatetahat}, it yields
\begin{equation*}
\sup_{\btheta\in\bTheta}
\|\hP(\btheta,\boheta_n,n)-P(\btheta,\boeta_0)\|
=O_p(n^{-1/2}).
\end{equation*}

Next, for each replication $h$, define
\[
\bar a^{(h)}_n(\btheta,\boeta):=\frac{1}{n}\sum_{i=1}^n a(\bx^{(h)}_{i,\btheta},\boeta),
\qquad
\bar b^{(h)}_n(\btheta,\boeta):=\frac{1}{n}\sum_{i=1}^n b(\bx^{(h)}_{i,\btheta},\boeta).
\]
Moreover, let $P(\btheta,\boeta,n):=\E\!\left[\hP(\btheta,\boeta,n)\right]$
denotes the finite-sample target.
We have for all $(\btheta,\boeta)\in\bTheta\times\bH$ and all $h$,
\begin{equation}\label{eq:Pj-bounded}
|\hP^{(h)}_{j}(\btheta,\boeta,n)| \leqslant M/\tilde\delta,
\qquad
\forall\,(\btheta,\boeta)\in\bTheta\times\bH,\ \forall\,h=1,\dots,H.
\end{equation}
Moreover, 
we have $|\hP^{(h)}_j(\btheta,\boeta,n)|\leqslant|(\bar a_n^{(h)})_j(\btheta,\boeta)|/\tilde\delta$.
Since $(\bar a_n^{(h)})_j(\btheta,\boeta)$ is Lipschitz in $(\btheta,\boeta)$
with constant $L_a$, it follows that
\[
|\hP^{(h)}_{j}(\btheta_1,\boeta_1,n)
-
\hP^{(h)}_{j}(\btheta_2,\boeta_2,n)|
\leqslant
\frac{L_a}{\tilde\delta}
\bigl(\|\btheta_1-\btheta_2\|+\|\boeta_1-\boeta_2\|\bigr).
\]
Therefore, letting $L:=L_a/\tilde\delta$, for all $h$,
\begin{equation}\label{eq:Pj-Lip}
|\hP^{(h)}_{j}(\btheta_1,\boeta_1,n)
-
\hP^{(h)}_{j}(\btheta_2,\boeta_2,n)|
\leqslant
L\bigl(\|\btheta_1-\btheta_2\|+\|\boeta_1-\boeta_2\|\bigr),
\qquad
(\btheta_k,\boeta_k)\in\bTheta\times\bH.
\end{equation}
Since $P_{j}(\btheta,\boeta,n)=\E[\hP_{j}(\btheta,\boeta,n)]$,
the same Lipschitz bound holds for $P_{j}(\btheta,\boeta,n)$ by Jensen's inequality.
Let $\mathcal N_\delta$ be a finite subset of $\bTheta\times\bH$
such that for every $(\btheta,\boeta)\in\bTheta\times\bH$
there exists a $(\btheta^\delta,\boeta^\delta)\in\mathcal N_\delta$ with $\|\btheta-\btheta^\delta\|
+
\|\boeta-\boeta^\delta\|
\leqslant \delta$.
Such a set exists because $\bTheta\times\bH$ is compact,  hence totally bounded,
and $\#\mathcal N_\delta \leqslant C_0\,\delta^{-(r+d+d(d+1)/2)}$ for some $C_0<\infty$.
For any $(\btheta,\boeta)$, choose
$(\btheta^\delta,\boeta^\delta)\in\mathcal N_\delta$
with distance at most $\delta$.
By adding and subtracting
$\bar P_j(\btheta^\delta,\boeta^\delta,n)$
and
$P_{j}(\btheta^\delta,\boeta^\delta,n)$
and using the triangle inequality,
\begin{align*}
|\bar P_j(\btheta,\boeta,n)-P_{j}(\btheta,\boeta,n)|
&\leqslant
|\bar P_j(\btheta,\boeta,n)
-
\bar P_j(\btheta^\delta,\boeta^\delta,n)|
\\
&\quad
+
|\bar P_j(\btheta^\delta,\boeta^\delta,n)
-
P_{j}(\btheta^\delta,\boeta^\delta,n)|
\\
&\quad
+
|P_{j}(\btheta^\delta,\boeta^\delta,n)
-
P_{j}(\btheta,\boeta,n)|.
\end{align*}
By \eqref{eq:Pj-Lip}, 
$\hP^{(h)}_{j}(\btheta,\boeta,n)$
is $L$-Lipschitz, hence so is its average $\bar P_j(\btheta,\boeta,n)$.
Therefore
\[
|\bar P_j(\btheta,\boeta,n)
-
\bar P_j(\btheta^\delta,\boeta^\delta,n)|
\leqslant L\delta.
\]
Since $P_{j}(\btheta,\boeta,n)$ is also $L$-Lipschitz,
\[
|P_{j}(\btheta^\delta,\boeta^\delta,n)
-
P_{j}(\btheta,\boeta,n)|
\leqslant L\delta.
\]
Combining the bounds gives
\[
|\bar P_j(\btheta,\boeta,n)-P_{j}(\btheta,\boeta,n)|
\leqslant
|\bar P_j(\btheta^\delta,\boeta^\delta,n)
-
P_{j}(\btheta^\delta,\boeta^\delta,n)|
+
2L\delta,
\]
and therefore
\begin{equation}\label{eq:sup-net-P}
\sup_{(\btheta,\boeta)\in\bTheta\times\bH}
|\bar P_j(\btheta,\boeta,n)-P_{j}(\btheta,\boeta,n)|
\leqslant
\max_{(\btheta,\boeta)\in\mathcal N_\delta}
|\bar P_j(\btheta,\boeta,n)-P_{j}(\btheta,\boeta,n)|
+
2L\delta.
\end{equation}
For each fixed $(\btheta,\boeta)\in\mathcal N_\delta$,
the variables
$\{\hP^{(h)}_{j}(\btheta,\boeta,n)\}_{h=1}^H$
are i.i.d.\ and bounded by $[-B,B]$, with $B:=M/\tilde\delta$, by construction. 
Hence Hoeffding's inequality yields, for any $t>0$
\[
\Pr\!\left(
|\bar P_j(\btheta,\boeta,n)
-
P_{j}(\btheta,\boeta,n)|>t
\right)
\leqslant
2\exp\!\left(-\frac{H t^2}{2B^2}\right).
\]
Choose $\delta=H^{-1/2}$ and define
\[
t_H
:=
B\sqrt{\frac{(r+d+\frac{d(d+1)}{2}+2)\log H+2\log(2C_0)}{H}}.
\]
Then
\[
2C_0\,\delta^{-(r+d+\frac{d(d+1)}{2})}
\exp\!\left(-\frac{H t_H^2}{2B^2}\right)
=
H^{-1},
\]
so that
\[
\Pr\!\left(
\max_{(\btheta,\boeta)\in\mathcal N_{H^{-1/2}}}
|\bar P_j(\btheta,\boeta,n)
-
P_{j}(\btheta,\boeta,n)|>t_H
\right)
\leqslant
H^{-1}.
\]
Using \eqref{eq:sup-net-P} with $\delta=H^{-1/2}$, we have the implication
\begin{align*}
&
\Big\{
\sup_{(\btheta,\boeta)\in\bTheta\times\bH}
|\bar P_j(\btheta,\boeta,n)-P_j(\btheta,\boeta,n)|
>
t_H + 2L\,H^{-1/2}
\Big\}
\\
&\qquad\subseteq
\Big\{
\max_{(\btheta,\boeta)\in\mathcal N_{H^{-1/2}}}
|\bar P_j(\btheta,\boeta,n)-P_j(\btheta,\boeta,n)|
>
t_H
\Big\}.
\end{align*}
Therefore,
\[
\Pr\!\left(
\sup_{(\btheta,\boeta)\in\bTheta\times\bH}
|\bar P_j(\btheta,\boeta,n)-P_j(\btheta,\boeta,n)|
>
t_H + 2L\,H^{-1/2}
\right)
\leqslant
H^{-1}.
\]
Since $\sqrt{\frac{\log H}{H}} = o(n^{-1/2})$, it follows $H\rightarrow\infty$, we have 
\[
\sup_{(\btheta,\boeta)\in\bTheta\times\bH}
|\bar P_j(\btheta,\boeta,n)-P_{j}(\btheta,\boeta,n)|
=
O_p\!\bigl(t_H + 2L\,H^{-1/2}\bigr).
\]
Now,
\[
t_H
=
B\sqrt{\frac{(r+d+\frac{d(d+1)}{2}+2)\log H + 2\log(2C_0)}{H}}
=
O\!\left(\sqrt{\frac{\log H}{H}}\right),
\]
and clearly $H^{-1/2} = O\!\left(\sqrt{\frac{\log H}{H}}\right)$ for large $H$.
Therefore,
\[
t_H + 2L\,H^{-1/2}
=
O\!\left(\sqrt{\frac{\log H}{H}}\right),
\]
and we conclude
\begin{equation}\label{eq:MC-rate}
\sup_{(\btheta,\boeta)\in\bTheta\times\bH}
|\bar P_j(\btheta,\boeta,n)-P_{j}(\btheta,\boeta,n)|
=
O_p\!\left(\sqrt{\frac{\log H}{H}}\right).
\end{equation}

Recall that
\[
P_j(\btheta,\boeta,n)
=
\E\!\left[\hP_j(\btheta,\boeta,n)\right],
\qquad
P_j(\btheta,\boeta)
=
g\!\left((\mu_a)_j(\btheta,\boeta),(\mu_b)_j(\btheta,\boeta)\right),
\]
where $g(a,b)=a/b$.
On the event $E_n$, we have $|(\bar b_n)_j(\btheta,\boeta)|\geqslant c$ and also
$|(\mu_b)_j(\btheta,\boeta)|> c$ by assumption, hence both points
$\big((\bar a_n)_j,(\bar b_n)_j\big)$ and $\big((\mu_a)_j,(\mu_b)_j\big)$
belong to the domain $\{(a,b):|b|\geqslant c\}$.
Since $g(a,b)=a/b$ is twice continuously differentiable on the domain
$\{|b|\geqslant c\}$ and its Hessian is uniformly bounded there, a second-order
Taylor expansion of $g$ around
$\big((\mu_a)_j(\btheta,\boeta),(\mu_b)_j(\btheta,\boeta)\big)$ gives,
for some point
$\xi$ on the line segment joining
$\big((\bar a_n)_j(\btheta,\boeta),(\bar b_n)_j(\btheta,\boeta)\big)$, and
$\big((\mu_a)_j(\btheta,\boeta),(\mu_b)_j(\btheta,\boeta)\big)$,
the expansion
\begin{align*}
g\!\left((\bar a_n)_j(\btheta,\boeta),(\bar b_n)_j(\btheta,\boeta)\right)
&=
g\!\left((\mu_a)_j(\btheta,\boeta),(\mu_b)_j(\btheta,\boeta)\right) +
\left(\nabla g_{0}(\btheta,\boeta)\right)^T
\begin{pmatrix}
\Delta_{a,n}(\btheta,\boeta)\\
\Delta_{b,n}(\btheta,\boeta)
\end{pmatrix}\\
&\quad+
\frac12
\begin{pmatrix}
\Delta_{a,n}(\btheta,\boeta)\\
\Delta_{b,n}(\btheta,\boeta)
\end{pmatrix}^T
H(\xi)
\begin{pmatrix}
\Delta_{a,n}(\btheta,\boeta)\\
\Delta_{b,n}(\btheta,\boeta)
\end{pmatrix}.
\end{align*}
where 
\[
\Delta_{a,n}(\btheta,\boeta):=(\bar a_n)_j(\btheta,\boeta)-(\mu_a)_j(\btheta,\boeta),
\qquad
\Delta_{b,n}(\btheta,\boeta):=(\bar b_n)_j(\btheta,\boeta)-(\mu_b)_j(\btheta,\boeta),
\]
and 
\[
\nabla g_{0}(\btheta,\boeta):=
\nabla g\!\left((\mu_a)_j(\btheta,\boeta),(\mu_b)_j(\btheta,\boeta)\right),
\qquad
H(\xi):=\nabla^2 g(\xi).
\]
Let $A_i:=a_j(\bx_{i,\btheta},\boeta)$. Then $\E[A_i]=(\mu_a)_j(\btheta,\boeta)$ and
\begin{equation}\label{eq:varaince bounddeltaa}
\E\!\left[\big(\Delta_{a,n}(\btheta,\boeta)\big)^2\right]
=\Var\!\left(\frac1n\sum_{i=1}^n A_i\right)
=\frac{1}{n}\Var(A_1)
\leqslant \frac{1}{n}\E[A_1^2]
\leqslant \frac{M^2}{n},
\end{equation}
where the last inequality follows from $|A_1|\leqslant M$ , uniformly in $(\btheta,\boeta)$.
An identical argument applies to $\Delta_{b,n}$.
Multiplying the Taylor expansion by $I_{E_n}$ and taking expectations yields
\begin{align*}
&\E\!\Big[
\Big(
g\!\left((\bar a_n)_j(\btheta,\boeta),(\bar b_n)_j(\btheta,\boeta)\right)
-
g\!\left((\mu_a)_j(\btheta,\boeta),(\mu_b)_j(\btheta,\boeta)\right)
\Big)I_{E_n}
\Big]
\\
&\qquad=
\left(\nabla g_{0}(\btheta,\boeta)\right)^T
\E\!\Big[
\begin{pmatrix}
\Delta_{a,n}(\btheta,\boeta)\\
\Delta_{b,n}(\btheta,\boeta)
\end{pmatrix}
I_{E_n}
\Big]
+
\frac12\,\E\!\Big[
\begin{pmatrix}
\Delta_{a,n}(\btheta,\boeta)\\
\Delta_{b,n}(\btheta,\boeta)
\end{pmatrix}^T
H(\xi)
\begin{pmatrix}
\Delta_{a,n}(\btheta,\boeta)\\
\Delta_{b,n}(\btheta,\boeta)
\end{pmatrix}
I_{E_n}
\Big].
\end{align*}
Note that $\E[\Delta_{a,n}(\btheta,\boeta)]=0$ and $\E[\Delta_{b,n}(\btheta,\boeta)]=0$.
Hence
\[
\E[\Delta_{a,n}(\btheta,\boeta)I_{E_n}]
=
-\E[\Delta_{a,n}(\btheta,\boeta)I_{E_n^c}],
\qquad
\E[\Delta_{b,n}(\btheta,\boeta)I_{E_n}]
=
-\E[\Delta_{b,n}(\btheta,\boeta)I_{E_n^c}].
\]
By Cauchy-Schwarz and the variance bound \eqref{eq:varaince bounddeltaa},
\begin{align*}
\sup_{(\btheta,\boeta)\in\bTheta\times\bH}
\big|\E[\Delta_{a,n}(\btheta,\boeta)I_{E_n}]\big|
&=
\sup_{(\btheta,\boeta)\in\bTheta\times\bH}
\big|\E[\Delta_{a,n}(\btheta,\boeta)I_{E_n^c}]\big|
\\
&\leqslant
\sup_{(\btheta,\boeta)\in\bTheta\times\bH}
\sqrt{\E\!\big[\Delta_{a,n}(\btheta,\boeta)^2\big]}\,
\sqrt{\Pr(E_n^c)}
\\
&\leqslant
\frac{M}{\sqrt n}\sqrt{\Pr(E_n^c)}.
\end{align*}
and similarly for $\Delta_{b,n}$.
Since $\|\nabla g_0(\btheta,\boeta)\|$ is uniformly bounded
(because $|(\mu_b)_j(\btheta,\boeta)|> c$ and $|(\mu_a)_j(\btheta,\boeta)|\leqslant M$),
it follows that the linear term is
\[
\sup_{(\btheta,\boeta)\in\bTheta\times\bH}
\left|
\left(\nabla g_{0}(\btheta,\boeta)\right)^T
\E\!\Big[
\begin{pmatrix}
\Delta_{a,n}(\btheta,\boeta)\\
\Delta_{b,n}(\btheta,\boeta)
\end{pmatrix}
I_{E_n}
\Big]
\right|
=
O\!\left(\frac{1}{\sqrt n}\sqrt{\Pr(E_n^c)}\right).
\]
On $E_n$, the point $\xi$ lies in $\{|b|\geqslant c\}$, hence $\|H(\xi)\|\leqslant C$ for a constant $C<\infty$.
Therefore,
\begin{align*}
&\sup_{(\btheta,\boeta)\in\bTheta\times\bH}
\left|
\E\!\Big[
\begin{pmatrix}
\Delta_{a,n}(\btheta,\boeta)\\
\Delta_{b,n}(\btheta,\boeta)
\end{pmatrix}^T
H(\xi)
\begin{pmatrix}
\Delta_{a,n}(\btheta,\boeta)\\
\Delta_{b,n}(\btheta,\boeta)
\end{pmatrix}
I_{E_n}
\Big]
\right|
\\
&\qquad\leqslant
C\,
\sup_{(\btheta,\boeta)\in\bTheta\times\bH}
\E\!\Big[
\big(\Delta_{a,n}(\btheta,\boeta)^2+\Delta_{b,n}(\btheta,\boeta)^2\big) I_{E_n}
\Big]
\\
&\qquad\leqslant
C\,
\sup_{(\btheta,\boeta)}\E\!\big[\Delta_{a,n}(\btheta,\boeta)^2\big]
+
C\,
\sup_{(\btheta,\boeta)}\E\!\big[\Delta_{b,n}(\btheta,\boeta)^2\big]
=
O(n^{-1}).
\end{align*}
The last equality uses \eqref{eq:varaince bounddeltaa} and the analogous for $\Delta_{b,n}$.
Recall that
$|\hP_j(\btheta,\boeta,n)|
\leqslant |(\bar a_n)_j(\btheta,\boeta)|/\tilde\delta$
and that
$|P_j(\btheta,\boeta)|
\leqslant \frac{M}{c}$.
Hence, using $|(\bar a_n)_j(\btheta,\boeta)|\leqslant M$,
\[
\sup_{(\btheta,\boeta)\in\bTheta\times\bH}
\Big|
\E\big[(\hP_j(\btheta,\boeta,n)-P_j(\btheta,\boeta))I_{E_n^c}\big]
\Big|
\leqslant
M\left(\frac{1}{\tilde\delta}+\frac{1}{c}\right)
\Pr(E_n^c).
\]
Using that for any integrable random vector $X$, we have
$\E[X]
=
\E[XI_{E_n}]
+
\E[XI_{E_n^c}]$, then
\begin{align*}
P_j(\btheta,\boeta,n)-P_j(\btheta,\boeta)
&=
\E\!\big[(\hP_j(\btheta,\boeta,n)-P_j(\btheta,\boeta))I_{E_n}\big]
\\
&\quad+
\E\!\big[(\hP_j(\btheta,\boeta,n)-P_j(\btheta,\boeta))I_{E_n^c}\big].
\end{align*}
Taking absolute values and the supremum over $(\btheta,\boeta)$ yields
\begin{align*}
\sup_{(\btheta,\boeta)\in\bTheta\times\bH}
\big|P_j(\btheta,\boeta,n)-P_j(\btheta,\boeta)\big|
&\leqslant
\sup_{(\btheta,\boeta)}
\Big|
\E\big[(\hP_j(\btheta,\boeta,n)-P_j(\btheta,\boeta))I_{E_n}\big]
\Big|
\\
&\quad+
\sup_{(\btheta,\boeta)}
\Big|
\E\big[(\hP_j(\btheta,\boeta,n)-P_j(\btheta,\boeta))I_{E_n^c}\big]
\Big|.
\end{align*}
Combining the bounds above yields
\[
\sup_{(\btheta,\boeta)\in\bTheta\times\bH}
\big|P_j(\btheta,\boeta,n)-P_j(\btheta,\boeta)\big|
=
O(n^{-1})
+
O\!\left(\frac{1}{\sqrt n}\sqrt{\Pr(E_n^c)}\right)
+
O\!\big(\Pr(E_n^c)\big).
\]
By assumption, $\Pr(E_n^c)=O(n^{-\alpha})$, so
\[
\sup_{(\btheta,\boeta)\in\bTheta\times\bH}
\big|P_j(\btheta,\boeta,n)-P_j(\btheta,\boeta)\big|
=
O(n^{-1})
+
O\!\big(n^{-(1+\alpha)/2}\big)
+
O(n^{-\alpha}).
\]
Since $\alpha>1/2$ then the right-hand side is $o(n^{-1/2})$.
Then,
\begin{equation}\label{eq:bias-rate}
\sup_{(\btheta,\boeta)\in\bTheta\times\bH}
\left\|
P(\btheta,\boeta,n)-P(\btheta,\boeta)
\right\|
=
o(n^{-1/2}).
\end{equation}

By the triangle inequality,
\begin{align*}
\sup_{\btheta\in\bTheta}
\left\|
\bar P(\btheta,\boheta_n,n)-P(\btheta,\boheta_n)
\right\|
&\leqslant
\sup_{\btheta\in\bTheta}
\left\|
\bar P(\btheta,\boheta_n,n)-P(\btheta,\boheta_n,n)
\right\|
+
\sup_{\btheta\in\bTheta}
\left\|
P(\btheta,\boheta_n,n)-P(\btheta,\boheta_n)
\right\| \\
&\leqslant
\sup_{(\btheta,\boeta)\in\bTheta\times\bH}
\left\|
\bar P(\btheta,\boeta,n)-P(\btheta,\boeta,n)
\right\| \\
& \quad +
\sup_{(\btheta,\boeta)\in\bTheta\times\bH}
\left\|
P(\btheta,\boeta,n)-P(\btheta,\boeta)
\right\|.
\end{align*}
Combining \eqref{eq:MC-rate} and \eqref{eq:bias-rate} yields
\[
\sup_{\btheta\in\bTheta}
\left\|
\bar P(\btheta,\boheta_n,n)-P(\btheta,\boheta_n)
\right\|
=
O_p\left(\sqrt{\frac{\log H}{H}}\right)+o(n^{-1/2}).
\]
By assumption $\sqrt{\frac{\log H}{H}} = o(n^{-1/2})$, we conclude that
\[
\sup_{\btheta\in\bTheta}
\left\|
\bar P(\btheta,\boheta_n,n)-P(\btheta,\boheta_n)
\right\|
=o_p(n^{-1/2}).
\]

\end{proof}

\begin{proposition}
\label{prop:bootstrap-consistency-theta}
Under the assumptions of Proposition~\ref{prop:2}, further assume that for each $\eps>0$ there exists an $L>0$ such that
\begin{equation}
\label{eq:bootstrap-uniform-pi-hat}
\Pr\Big( {\Pr}^*\Big( \sqrt{n}\| \hpi^*_n(\boheta_n) - \hpi_n(\boheta_n) \| > L \Big) >\eps \Big) \rightarrow 0.
\end{equation}
Then, for every $\eps>0$,
\begin{equation*}
{\Pr}^*\big(\|\bhtheta_n^* - \btheta_0\|>\eps)
\;\rightarrow_p\; 0.
\end{equation*}
\end{proposition}

\begin{proof}

Set $\Delta_n^* := \|\hpi_n^*(\boheta_n)-\hpi_n(\boheta_n)\|$. From \eqref{eq:bootstrap-uniform-pi-hat}, for each $\eps>0$ there exists a constant
$L_\eps>0$ such that
\begin{equation}
\label{eq:assumption-sqrt}
\Pr\Big(
  {\Pr}^*\big( \sqrt{n}\,\Delta_n^* > L_\eps \big) > \eps
\Big)
\rightarrow 0.
\end{equation}
Fix $\eps>0$ and an arbitrary $L>0$.  Since $L>0$, we can choose
$N_1$ such that $
\sqrt{n}\,L \geqslant L_\eps$ for all $n\geqslant N_1$.
For such $n$ we have the inclusion of events
$
\{\Delta_n^*>L\}
\subseteq 
\{\sqrt{n}\,\Delta_n^*>L_\eps\}$,
which implies, for every fixed sample, that
${\Pr}^*(\Delta_n^*>L)
\leqslant
{\Pr}^*(\sqrt{n}\,\Delta_n^*>L_\eps)$.
Therefore, for any $n\geqslant N_1$
\begin{equation*}
\Pr\Big(
  {\Pr}^*(\Delta_n^*>L) > \eps
\Big)
\leqslant
\Pr\Big(
  {\Pr}^*(\sqrt{n}\,\Delta_n^*>L_\eps) > \eps
\Big).
\end{equation*}
Since $L>0$, $\eps>0$ are arbitrary, this proves that the for every $L>0$ and
$\eps>0$,
\begin{equation}
\label{eq:bootstrap-uniform-pi-hat2}
\Pr\Big(
  {\Pr}^*\Big(
    \| \hpi^*_n(\boheta_n) - \hpi_n(\boheta_n) \|
    > L
  \Big)
  > \eps
\Big)
\;\rightarrow\; 0.
\end{equation}
Using the triangle inequality
\begin{equation}\label{eq:trineq11}
\|\hpi^*_n(\boheta_n)-\pi(\btheta_0,\boeta_0)\|
\leqslant
\|\hpi^*_n(\boheta_n)-\hpi_n(\boheta_n)\|
+
\|\hpi_n(\boheta_n)-\pi(\btheta_0,\boeta_0)\|.
\end{equation}
Hence, for $L>0$, 
\begin{equation*}
\Big\{
\|\hpi^*_n(\boheta_n)-\pi(\btheta_0,\boeta_0)\|>L
\Big\}
\subseteq
\Big\{
\|\hpi_n(\boheta_n)-\pi(\btheta_0,\boeta_0)\|>L/2
\Big\}
\cup
\Big\{
\|\hpi^*_n(\boheta_n)-\hpi_n(\boheta_n)\|>L/2
\Big\},
\end{equation*}
as \eqref{eq:trineq11}
implies that whenever the left-hand side exceeds \(L\), at least one of the
two terms on the right must exceed \(L/2\).
Taking bootstrap probability ${\Pr}^*$ of both sides and then outer
probability $\Pr$, we obtain for any fixed $L>0$ and $\eps>0$,
\begin{align*}
&\Pr\Big(
  {\Pr}^*\Big(
    \|\hpi^*_n(\boheta_n)-\pi(\btheta_0,\boeta_0)\|
    > L
  \Big)
  > \eps
\Big)\\
&\qquad\leqslant
\Pr\Big(
  \|\hpi_n(\boheta_n)-\pi(\btheta_0,\boeta_0)\|
  > L/2
\Big)
+
\Pr\Big(
  {\Pr}^*\Big(
    \|\hpi^*_n(\boheta_n)-\hpi_n(\boheta_n)\|
    > L/2
  \Big)
  > \eps/2
\Big).
\end{align*}
The first term converges to zero by assumption, while the
second term converges to zero by \eqref{eq:bootstrap-uniform-pi-hat2} .
Therefore, for every $L>0$ and $\eps>0$,
\begin{equation}
\label{eq:bootstrap-uniform-to-pi}
\Pr\Big(
  {\Pr}^*\Big(
   \|\hpi^*_n(\boheta_n)-\pi(\btheta_0,\boeta_0)\|
    > L
  \Big)
  > \eps
\Big)
\;\rightarrow\; 0.
\end{equation}

 Note that, $\bhtheta_n^*$ satisfies
\begin{equation*}
\bhtheta_n^* = \argmin_{\btheta\in\bTheta} \widehat Q_n^*(\btheta),
\end{equation*}
where $\widehat Q_n^*(\btheta)
=
\big\|
  \hpi^*_n(\boheta_n)
  - \pibar(\btheta,\boheta_n,n) 
\big\|^2$.
Let $Q(\btheta)=\|\pi(\btheta_0,\boeta_0)-\pi(\btheta,\boeta_0)\|^2$.
 Repeating the same algebra as in the proof of
Proposition~\ref{prop:2}, but replacing $\hpi_n(\boheta_n)$  with $\hpi^*_n(\boheta_n)$ we obtain, 
\begin{equation*}
\sup_{\btheta\in\bTheta}
|\widehat Q_n^*(\btheta)-Q(\btheta)|
\;\leqslant\;
C_1
\|\hpi^*_n(\boheta_n)-\pi(\btheta_0,\boeta_0)\|^2
+
C_2
\sup_{\btheta\in\bTheta}
\|\overline\pi(\btheta,\boheta_n,n)-\pi(\btheta,\boeta_0)\|^2,
\end{equation*}
for suitable finite constants $C_1,C_2>0$.
Now fix $\delta>0$.
Choose $L_1>0$ and $L_2>0$ such that $
C_1 L_1^2 + C_2 L_2^2 < \delta$.
Consider the event
\begin{equation*}
E_n = 
\Big\{
  \sup_{\btheta\in\bTheta}
  |\widehat Q_n^*(\btheta)-Q(\btheta)|
  > \delta
\Big\}.
\end{equation*}
By the inequality above, if both
\(
\|\hpi_n^*(\boheta_n)-\pi(\btheta_0,\boeta_0)\|\leqslant L_1
\)
and
\(
\sup_{\btheta}\|\overline\pi(\btheta,\boheta_n,n)-\pi(\btheta,\boeta_0)\|\leqslant L_2
\),
then
\(
\sup_{\btheta}|\widehat Q_n^*(\btheta)-Q(\btheta)|\leqslant C_1L_1^2+C_2L_2^2<\delta
\).
Taking complements yields
\begin{equation*}
E_n \subseteq
\Big\{
    \|\hpi^*_n(\boheta_n)-\pi(\btheta_0,\boeta_0)\|
  > L_1
\Big\}
\;\cup\;
\Big\{
  \sup_{\btheta\in\bTheta}
  \| \pibar(\btheta,\boheta_n,n) -\pi(\btheta,\boeta_0)\|
  > L_2
\Big\}.
\end{equation*}
Taking ${\Pr}^*$ of both sides
\begin{equation*}
{\Pr}^*(E_n)
\leqslant
{\Pr}^*\Big(
  \|\hpi^*_n(\boheta_n)-\pi(\btheta_0,\boeta_0)\|
  > L_1
\Big)
+
\bone
\Big\{
  \sup_{\btheta\in\bTheta}
  \| \pibar(\btheta,\boheta_n,n) -\pi(\btheta,\boeta_0)\|
  > L_2
\Big\}.
\end{equation*}
Finally, taking $\Pr$ of both sides and using
\eqref{eq:bootstrap-uniform-to-pi} for the first term and the fact that the assumptions implies $\sup_{\btheta\in\bTheta}
\| \pibar(\btheta,\boheta_n,n) -\pi(\btheta,\boeta_0)\|
\;\rightarrow_p\; 0$ (see the proof of Proposition~\ref{prop:2}) for the second,
we obtain for any $\eps>0$
\begin{equation*}
\Pr\big({\Pr}^*(E_n)>\eps\big)
\rightarrow 0.
\end{equation*}
This proves that for any $\delta>0$
\begin{equation}\label{eq:supQboot}
{\Pr}^*\Big(
\sup_{\btheta\in\bTheta}
|\widehat Q_n^*(\btheta)-Q(\btheta)|
> \delta
\Big)
\;\rightarrow_p\; 0.
\end{equation}

From Proposition~\ref{prop:2},  we know that $\bTheta$ is compact, $Q$ is continuous, and $Q(\btheta)$ is uniquely minimized at $\btheta_0$.
Let
$A_\delta := \{\btheta\in\bTheta:\|\btheta-\btheta_0\|\geqslant\delta\}$ , then by assumption, we have $\inf_{\btheta \in A_\delta} Q(\btheta) > 0$.
Fix 
$0<\eta\leqslant\Delta_\delta/3$ and $\delta>0$, with $\Delta_\delta := \inf_{\btheta \in A_\delta} Q(\btheta)>0$. 
Then
\begin{equation}
\inf_{\btheta:\|\btheta-\btheta_0\|\geqslant\delta} Q(\btheta)
\geqslant  3\eta.
\label{eq:identification}
\end{equation}
Let $U:=\{\btheta:\|\btheta-\btheta_0\|<\delta\}$ and $U^c:=\{\btheta:\|\btheta-\btheta_0\|\geqslant\delta\}$.
Consider the event
\begin{equation*}
B_n = \Big\{
\sup_{\btheta\in\bTheta}
|\widehat Q_n^*(\btheta)-Q(\btheta)| \leqslant \eta
\Big\}.
\end{equation*}
Then on $B_n$,  we have that   
  \begin{equation*}
  \widehat Q_n^*(\btheta_0)
  \leqslant \eta.
  \end{equation*}
and by \eqref{eq:identification}, for $\btheta\in U^c$:  
  \begin{equation*}
  \widehat Q_n^*(\btheta)
  \geqslant
  Q(\btheta)-\eta\geqslant 3\eta-\eta= 2\eta.
  \end{equation*}
Thus, on $B_n$, and for any $\btheta\in U^c$
\begin{equation*}
\widehat Q_n^*(\btheta_0)\leqslant \eta\leqslant \widehat Q_n^*(\btheta)- \eta< \widehat Q_n^*(\btheta).
\end{equation*}
On $B_n$, $\widehat Q_n^*(\btheta_0)$ is strictly smaller than
$\widehat Q_n^*(\btheta)$ for all $\btheta\in U^c$, hence the minimizer
$\bhtheta_n^*$ must belong to $U$. 
Hence,
\begin{equation*}
\{\bhtheta_n^*\in U^c\}
\subseteq 
\Big\{
\sup_{\btheta\in\bTheta}
|\widehat Q_n^*(\btheta)-Q(\btheta)| > \eta
\Big\},
\end{equation*}
and taking conditional probability ${\Pr}^*$ yields
\begin{equation*}
{\Pr}^*(\bhtheta_n^*\in U^c)
\;\leqslant\;
{\Pr}^*\Big(
\sup_{\btheta\in\bTheta}
|\widehat Q_n^*(\btheta)-Q(\btheta)| > \eta
\Big).
\end{equation*}
Consequently, for any $t>0$,
\begin{equation*}
\big\{{\Pr}^*(\bhtheta_n^*\in U^c) > t\big\}
\subseteq
\Big\{
{\Pr}^*\Big(
\sup_{\btheta\in\bTheta}
|\widehat Q_n^*(\btheta)-Q(\btheta)| > \eta
\Big) > t
\Big\},
\end{equation*}
and hence by \eqref{eq:supQboot}
\begin{equation*}
\Pr\Big(
{\Pr}^*(\bhtheta_n^*\in U^c) > t
\Big)
\leqslant
\Pr\Big(
{\Pr}^*\Big(
\sup_{\btheta\in\bTheta}
|\widehat Q_n^*(\btheta)-Q(\btheta)| > \eta
\Big) > t
\Big)\rightarrow 0.
\end{equation*}
 This implies
\begin{equation*}
{\Pr}^*(\|\bhtheta_n^*-\btheta_0\|>\delta)\rightarrow_p 0.
\end{equation*}

\end{proof}

\begin{proposition}
\label{pro:II_WI_boot}
Assume the conditions of Proposition~\ref{pro:II_WI} hold and that for each $\eps>0$ there exists an $L>0$ such that 
\begin{equation}
\label{eq:asslemma5_bootpointwise}
\Pr\Big( {\Pr}^*\Big(  \sqrt{n}\big\| \hpi^*_n(\boheta_n) - \hpi_n(\boheta_n) \big\| > L \Big) >\eps \Big)\rightarrow 0.
\end{equation}
Then
\begin{align*}
\bhtheta_n^*-\btheta_0
=
\bK_0\big(\hpi^*_n(\boheta_n)-\pi(\btheta_0,\boheta_n)\big)
+\btr_n^*,
\end{align*}
where ${\Pr}^*(\sqrt{n}\,\|\btr_n^*\|>\eps)\rightarrow_p 0$.
\end{proposition}

\begin{proof}
 From Proposition~\ref{prop:bootstrap-consistency-theta}, we have  for any $\eps>0$, ${\Pr}^*(\|\bhtheta_n^*-\btheta_0\|>\eps)\rightarrow_p 0$.
Let,
\begin{equation*}
\zeta_n(\bhtheta_n^*,\boheta_n)
=
\bar\pi(\bhtheta_n^*,\boheta_n,n)-\pi(\bhtheta_n^*,\boheta_n).
\end{equation*}
By assumption,
\begin{equation*}
\sup_{\btheta\in\bTheta}\|\zeta_n(\btheta,\boheta_n)\|
=
\sup_{\btheta\in\bTheta}
\|\bar\pi(\btheta,\boheta_n,n)-\pi(\btheta,\boheta_n)\|
=o_p(n^{-1/2}),
\end{equation*}
and in particular $\zeta_n(\bhtheta_n^*,\boheta_n)=o_p(n^{-1/2})$.
Then
\[
X_n:=\sqrt n\sup_{\btheta\in\bTheta}\|\zeta_n(\btheta,\boheta_n)\|=o_p(1).
\]
Fix $M>0$ and $\eps>0$, for large $n$, we have
\begin{equation}\label{eq:Prsupbarpi_hateta}
\Pr(X_n>M)\leqslant \eps.
\end{equation}
Define the event $E_n:=\{X_n\leqslant M\}$. On $E_n$, since $\bhtheta_n^*\in\bTheta$,
\[
\sqrt n\,\|\zeta_n(\bhtheta_n^*,\boheta_n)\|
\leqslant \sqrt n\sup_{\btheta\in\bTheta}\|\zeta_n(\btheta,\boheta_n)\|
= X_n \leqslant M,
\]
hence
\[
{\Pr}^*\!\big(\sqrt n\,\|\zeta_n(\bhtheta_n^*,\boheta_n)\|>M\big)=0.
\]
Therefore, for any $\delta>0$,
\[
\Big\{{\Pr}^*\!\big(\sqrt n\,\|\zeta_n(\bhtheta_n^*,\boheta_n)\|>M\big)>\delta\Big\}
\subseteq \{X_n>M\}.
\]
Taking $\Pr$ on both sides and using \eqref{eq:Prsupbarpi_hateta} yields, for all $n\geqslant N$,
\[
\Pr\Big({\Pr}^*\!\big(\sqrt n\,\|\zeta_n(\bhtheta_n^*,\boheta_n)\|>M\big)>\delta\Big)
\leqslant \Pr(X_n>M)\leqslant \eps.
\]
Thus, for each $\delta,M,\eps>0$  we have for large $n$
\begin{equation}\label{eq:bound-zeta-mixed}
\Pr\Big(
{\Pr}^*\big(
\sqrt{n}\,\|\zeta_n(\bhtheta_n^*)\|>M
\big)>\delta
\Big)
\leqslant \eps.
\end{equation}

By a first-order Taylor expansion of $\pi(\btheta,\boeta)$ at $(\btheta_0,\boeta_0)$,
we have
\begin{equation*}
\pi(\btheta,\boeta)
=
\pi(\btheta_0,\boeta_0)
+
\bA_0(\btheta-\btheta_0)
+
\bB_0(\boeta-\boeta_0)
+
\br(\btheta,\boeta),
\end{equation*}
where
\[
\bB_0
:=
\left.\frac{\partial\,\pi(\btheta,\boeta)}{\partial\boeta}\right|_{(\btheta,\boeta)=(\btheta_0,\boeta_0)},
\]
and
\[
\br(\btheta,\boeta)
:=
\pi(\btheta,\boeta)-\pi(\btheta_0,\boeta_0)
-\bA_0(\btheta-\btheta_0)-\bB_0(\boeta-\boeta_0).
\]
Since $\pi$ is differentiable at $(\btheta_0,\boeta_0)$, the remainder $\br(\btheta,\boeta)$ satisfies
\begin{equation}\label{eq:reminderdet_joint_sum}
\frac{\|\br(\btheta,\boeta)\|}{\|\btheta-\btheta_0\|+\|\boeta-\boeta_0\|}\rightarrow 0
\quad\text{as }(\btheta,\boeta)\rightarrow(\btheta_0,\boeta_0).
\end{equation}
Define, for $(\btheta,\boeta)\neq(\btheta_0,\boeta_0)$,
\[
f(\btheta,\boeta)
:=
\frac{\|\br(\btheta,\boeta)\|}{\|\btheta-\btheta_0\|+\|\boeta-\boeta_0\|},
\qquad
f(\btheta_0,\boeta_0)=0.
\]
Due to \eqref{eq:reminderdet_joint_sum}, $f$ is continuous at $(\btheta_0,\boeta_0)$ and
$f(\btheta_0,\boeta_0)=0$.
Hence, for any $\eps>0$ there exists an $\eta>0$ such that
\begin{equation}\label{eq:joint-continuity-fixedeta}
\|\btheta-\btheta_0\|+\|\boeta-\boeta_0\|\leqslant \eta
\quad\text{implies}\quad
f(\btheta,\boeta)\leqslant \eps.
\end{equation}
From \eqref{eq:joint-continuity-fixedeta},
\[
\big\{ f(\bhtheta_n^*,\boheta_n)>\eps \big\}
\subseteq
\big\{ \|\bhtheta_n^*-\btheta_0\|
      +\|\boheta_n-\boeta_0\|>\eta \big\},
\]
and therefore
\begin{equation*}
{\Pr}^*\!\big( f(\bhtheta_n^*,\boheta_n)>\eps \big)
\;\leqslant\;
{\Pr}^*\!\big(
\|\bhtheta_n^*-\btheta_0\|
+\|\boheta_n-\boeta_0\|>\eta
\big).
\end{equation*}
Let $\delta>0$. Taking probability $\Pr$ on both sides yields
\begin{align*}
\Pr\Big(
{\Pr}^*\!\big( f(\bhtheta_n^*,\boheta_n)>\eps \big)>\delta
\Big)
&\leqslant
\Pr\Big(
{\Pr}^*\!\big(
\|\bhtheta_n^*-\btheta_0\|
+\|\boheta_n-\boeta_0\|>\eta
\big)>\delta
\Big).
\end{align*}
Fix $\eta>0$ and $\delta>0$. By the union bound,
\[
\big\{\|\bhtheta_n^*-\btheta_0\|+\|\boheta_n-\boeta_0\|>\eta\big\}
\subseteq
\big\{\|\bhtheta_n^*-\btheta_0\|>\tfrac{\eta}{2}\big\}
\;\cup\;
\big\{\|\boheta_n-\boeta_0\|>\tfrac{\eta}{2}\big\}.
\]
Hence,
\begin{align}
{\Pr}^*\big(\|\bhtheta_n^*-\btheta_0\|+\|\boheta_n-\boeta_0\|>\eta\big)
&\leqslant
{\Pr}^*\big(\|\bhtheta_n^*-\btheta_0\|>\tfrac{\eta}{2}\big)
+
{\Pr}^*\big(\|\boheta_n-\boeta_0\|>\tfrac{\eta}{2}\big).
\label{eq:split1}
\end{align}
The second term is
\[
{\Pr}^*\big(\|\boheta_n-\boeta_0\|>\tfrac{\eta}{2}\big)
=
\bone\Big\{\|\boheta_n-\boeta_0\|>\tfrac{\eta}{2}\Big\}.
\]
Therefore, from \eqref{eq:split1},
\[
{\Pr}^*\big(\|\bhtheta_n^*-\btheta_0\|+\|\boheta_n-\boeta_0\|>\eta\big)
\leqslant
{\Pr}^*\big(\|\bhtheta_n^*-\btheta_0\|>\tfrac{\eta}{2}\big)
+
\bone\Big\{\|\boheta_n-\boeta_0\|>\tfrac{\eta}{2}\Big\}.
\]
Consequently,
\begin{align}
&\Pr\Big(
  {\Pr}^*\big(\|\bhtheta_n^*-\btheta_0\|+\|\boheta_n-\boeta_0\|>\eta\big)>\delta
\Big)
\nonumber\\
&\qquad\leqslant
\Pr\Big(
  {\Pr}^*\big(\|\bhtheta_n^*-\btheta_0\|>\tfrac{\eta}{2}\big)>\tfrac{\delta}{2}
\Big)
+
\Pr\Big(
  \bone\{\|\boheta_n-\boeta_0\|>\tfrac{\eta}{2}\}>\tfrac{\delta}{2}
\Big).
\label{eq:outer-split}
\end{align}
The first term in \eqref{eq:outer-split} converges to $0$ because
${\Pr}^*(\|\bhtheta_n^*-\btheta_0\|>\eta/2)\rightarrow_p 0$.
The second term is bounded by
\[
\Pr\Big(\|\boheta_n-\boeta_0\|>\tfrac{\eta}{2}\Big),
\]
which converges to $0$ since $\boheta_n-\boeta_0=o_p(1)$.
Hence,
\[
\Pr\Big(
  {\Pr}^*\big(\|\bhtheta_n^*-\btheta_0\|+\|\boheta_n-\boeta_0\|>\eta\big)>\delta
\Big)\rightarrow 0,
\]
and
consequently, for all $\eps,\delta>0$,
\begin{equation}\label{eq:bootstrap-remainder-final}
\Pr\Big(
{\Pr}^*\!\Big(
\frac{\|\br_n^*\|}
{\|\bhtheta_n^*-\btheta_0\|+\|\boheta_n-\boeta_0\|}
>\eps
\Big)>\delta
\Big)\;\rightarrow\;0.
\end{equation}
where $\br_n^*:=\br(\bhtheta_n^*,\boheta_n)$.

Since $\bhtheta_n^*$ satisfies the bootstrap normal equations with the nuisance
estimate $\boheta_n$,
\[
\hpi^*_n(\boheta_n)-\bar\pi(\bhtheta_n^*,\boheta_n,n)=0,
\]
we have
\begin{align*}
\mathbf 0
&=\hpi^*_n(\boheta_n)-\bar\pi(\bhtheta_n^*,\boheta_n,n)\\
&=\hpi^*_n(\boheta_n)-\pi(\bhtheta_n^*,\boheta_n)-\zeta_n(\bhtheta_n^*,\boheta_n)\\
&=\hpi^*_n(\boheta_n)-\pi(\btheta_0,\boeta_0)
-\bA_0(\bhtheta_n^*-\btheta_0)
-\bB_0(\boheta_n-\boeta_0)
-\br_n^*
-\zeta_n(\bhtheta_n^*,\boheta_n),
\end{align*}
where the last equality follows from the first-order Taylor expansion of
$\pi(\btheta,\boeta)$ at $(\btheta_0,\boeta_0)$, evaluated at
$(\bhtheta_n^*,\boheta_n)$.
Thus,
\[
\bhtheta_n^*-\btheta_0
=
\bK_0\big(\hpi^*_n(\boheta_n)-\pi(\btheta_0,\boeta_0)-\bB_0(\boheta_n-\boeta_0)\big)
-\bK_0\br_n^*
-\bK_0\zeta_n(\bhtheta_n^*,\boheta_n).
\]
With $C:=\sigma_{\max}(\bK_0)<\infty$ and $B:=\sigma_{\max}(\bB_0)<\infty$, we have
\begin{align*}
\|\bhtheta_n^*-\btheta_0\|
&\leqslant
C\|\hpi^*_n(\boheta_n)-\pi(\btheta_0,\boeta_0)\|
+CB\|\boheta_n-\boeta_0\|
+C\|\br_n^*\|
+C\|\zeta_n(\bhtheta_n^*,\boheta_n)\|.
\end{align*}
From
\eqref{eq:bootstrap-remainder-final}, for each $\eps,\eta>0$,
\begin{equation}\label{eq:Anstar}
\Pr\Big(
{\Pr}^*\Big(
\|\br_n^*\|>\eps\big(\|\bhtheta_n^*-\btheta_0\|+\|\boheta_n-\boeta_0\|\big)
\Big)>\eta
\Big)\rightarrow 0.
\end{equation}
Define the event
\[
E_{n,\eps}^*
:=
\Big\{\|\br_n^*\|\leqslant \eps\big(\|\bhtheta_n^*-\btheta_0\|+\|\boheta_n-\boeta_0\|\big)\Big\}.
\]
On $E_{n,\eps}^*$ we obtain
\begin{align*}
\|\bhtheta_n^*-\btheta_0\|
&\leqslant
C\|\hpi^*_n(\boheta_n)-\pi(\btheta_0,\boeta_0)\|
+CB\,\|\boheta_n-\boeta_0\|
+C\|\br_n^*\|
+C\|\zeta_n(\bhtheta_n^*,\boheta_n)\|\\
&\leqslant
C\|\hpi^*_n(\boheta_n)-\pi(\btheta_0,\boeta_0)\|
+CB\,\|\boheta_n-\boeta_0\|
+C\eps\|\bhtheta_n^*-\btheta_0\|
+C\eps\|\boheta_n-\boeta_0\|
+C\|\zeta_n(\bhtheta_n^*,\boheta_n)\|.
\end{align*}
 Rearranging yields, on $E_{n,\eps}^*$,
\begin{equation}\label{eq:theta-star-rearranged}
(1-C\eps)\|\bhtheta_n^*-\btheta_0\|
\leqslant
C\|\hpi^*_n(\boheta_n)-\pi(\btheta_0,\boeta_0)\|
+C(B+\eps)\|\boheta_n-\boeta_0\|
+C\|\zeta_n(\bhtheta_n^*,\boheta_n)\|.
\end{equation}
Choose $\eps>0$ such that $C\eps<1$. Dividing both sides of
\eqref{eq:theta-star-rearranged} by $(1-C\eps)$, we obtain on $E_{n,\eps}^*$
\begin{equation*}
\|\bhtheta_n^*-\btheta_0\|
\leqslant
C_1\|\hpi^*_n(\boheta_n)-\pi(\btheta_0,\boeta_0)\|
+C_2\|\boheta_n-\boeta_0\|
+C_1\|\zeta_n(\bhtheta_n^*,\boheta_n)\|,
\end{equation*}
for finite constants $C_1,C_2>0$.
Therefore, for any $t>0$,
\[
\{\|\bhtheta_n^*-\btheta_0\|>t\}
\subseteq
(E_{n,\eps}^*)^c
\;\cup\;
\Bigl\{
C_1\|\hpi^*_n(\boheta_n)-\pi(\btheta_0,\boeta_0)\|
+
C_2\|\boheta_n-\boeta_0\|
+
C_1\|\zeta_n(\bhtheta_n^*,\boheta_n)\|
>t
\Bigr\}.
\]
Hence,
\begin{align*}
{\Pr}^*(\|\bhtheta_n^*-\btheta_0\|>t)
&\leqslant
{\Pr}^*((E_{n,\eps}^*)^c)
+
{\Pr}^*\!\Bigl(
C_1\|\hpi^*_n(\boheta_n)-\pi(\btheta_0,\boeta_0)\|
+
C_2\|\boheta_n-\boeta_0\|
+
C_1\|\zeta_n(\bhtheta_n^*,\boheta_n)\|
>t
\Bigr).
\end{align*}
Taking $\Pr$ on both sides gives
\begin{align*}
\Pr\Big({\Pr}^*(\|\bhtheta_n^*-\btheta_0\|>t)>\eta\Big)
&\leqslant
\Pr\Big({\Pr}^*((E_{n,\eps}^*)^c)>\tfrac{\eta}{2}\Big)\\
&\quad+
\Pr\Bigg(
{\Pr}^*\!\Bigl(
C_1\|\hpi^*_n(\boheta_n)-\pi(\btheta_0,\boeta_0)\|
+\\&\qquad \qquad
C_2\|\boheta_n-\boeta_0\|
+
C_1\|\zeta_n(\bhtheta_n^*,\boheta_n)\|
>t
\Bigr)>\tfrac{\eta}{2}
\Bigg).
\end{align*}
Now set $t=Mn^{-1/2}$ and note that, by the triangle inequality, if
\[
C_1\sqrt n\|\hpi^*_n(\boheta_n)-\pi(\btheta_0,\boeta_0)\|
+
C_2\sqrt n\|\boheta_n-\boeta_0\|
+
C_1\sqrt n\|\zeta_n(\bhtheta_n^*,\boheta_n)\|
>M,
\]
then at least one of the following events must occur:
\[
C_1\sqrt n\|\hpi^*_n(\boheta_n)-\pi(\btheta_0,\boeta_0)\|>\tfrac{M}{3},
\qquad
C_2\sqrt n\|\boheta_n-\boeta_0\|>\tfrac{M}{3},
\qquad
C_1\sqrt n\|\zeta_n(\bhtheta_n^*,\boheta_n)\|>\tfrac{M}{3}.
\]
Hence, for every $M>0$,
\begin{align*}
&\Pr\Bigg(
{\Pr}^*\!\Bigl(
C_1\sqrt n\|\hpi^*_n(\boheta_n)-\pi(\btheta_0,\boeta_0)\|
+
C_2\sqrt n\|\boheta_n-\boeta_0\|
+
C_1\sqrt n\|\zeta_n(\bhtheta_n^*,\boheta_n)\|
>M
\Bigr)>\tfrac{\eta}{2}
\Bigg)
\nonumber\\
&\qquad\leqslant
\Pr\Bigg(
{\Pr}^*\Big(
C_1\sqrt n\|\hpi^*_n(\boheta_n)-\pi(\btheta_0,\boeta_0)\|>\tfrac{M}{3}
\Big)>\tfrac{\eta}{6}
\Bigg)
+
\Pr\Bigg(
{\Pr}^*\Big(
C_2\sqrt n\|\boheta_n-\boeta_0\|>\tfrac{M}{3}
\Big)>\tfrac{\eta}{6}
\Bigg)
\nonumber\\
&\qquad\quad+
\Pr\Bigg(
{\Pr}^*\Big(
C_1\sqrt n\|\zeta_n(\bhtheta_n^*,\boheta_n)\|>\tfrac{M}{3}
\Big)>\tfrac{\eta}{6}
\Bigg).
\end{align*}
Then by using \eqref{eq:Anstar}, we have  
\begin{align}
&\Pr\Big({\Pr}^*(\sqrt n\|\bhtheta_n^*-\btheta_0\|>M)>\eta\Big)
\nonumber\\
&\qquad\leqslant \Pr\Big({\Pr}^*((E_{n,\eps}^*)^c)>\tfrac{\eta}{2}\Big)+
\Pr\Bigg(
{\Pr}^*\Big(
C_1\sqrt n\|\hpi^*_n(\boheta_n)-\pi(\btheta_0,\boeta_0)\|>\tfrac{M}{3}
\Big)>\tfrac{\eta}{6}
\Bigg)\nonumber\\&\qquad
+
\Pr\Big(
C_2\sqrt n\|\boheta_n-\boeta_0\|>\tfrac{M}{3}
\Big)+
\Pr\Bigg(
{\Pr}^*\Big(
C_1\sqrt n\|\zeta_n(\bhtheta_n^*,\boheta_n)\|>\tfrac{M}{3}
\Big)>\tfrac{\eta}{6}
\Bigg).
\label{eq:outer-split-three}
\end{align}

Next, we obtain a bound for
$\sqrt{n}\,\|\hpi^*_n(\boheta_n)-\pi(\btheta_0,\boeta_0)\|$.
Fix $\eps>0$.
Assumption \eqref{eq:asslemma5_bootpointwise} implies that there exists an $L_1>0$ such that for large $n$
\[
\Pr\Big(
  {\Pr}^*\big(
    \sqrt{n}\,\|\hpi^*_n(\boheta_n)-\hpi_n(\boheta_n)\|>L_1
  \big)
  >\eps
\Big)\leqslant \eps.
\]
Moreover, since $\hpi_n(\boheta_n)-\pi(\btheta_0,\boeta_0)=O_p(n^{-1/2})$, there exists an
$L_2>0$ such that for large $n$,
\[
\Pr\Big(
  \sqrt{n}\,\|\hpi_n(\boheta_n)-\pi(\btheta_0,\boeta_0)\|>L_2
\Big)\leqslant\eps.
\]
By the triangle inequality,
\[
\sqrt{n}\,\|\hpi^*_n(\boheta_n)-\pi(\btheta_0,\boeta_0)\|
\leqslant
\sqrt{n}\,\|\hpi^*_n(\boheta_n)-\hpi_n(\boheta_n)\|
+
\sqrt{n}\,\|\hpi_n(\boheta_n)-\pi(\btheta_0,\boeta_0)\|.
\]
Hence,
\begin{align*}
&\Big\{
{\Pr}^*\big(
\sqrt{n}\,\|\hpi^*_n(\boheta_n)-\pi(\btheta_0,\boeta_0)\|>L_1+L_2
\big)
>\eps
\Big\}\\
&\qquad\subseteq
\Big\{
{\Pr}^*\big(
\sqrt{n}\,\|\hpi^*_n(\boheta_n)-\hpi_n(\boheta_n)\|>L_1
\big)
>\eps
\Big\}
\;\cup\;
\Big\{
\sqrt{n}\,\|\hpi_n(\boheta_n)-\pi(\btheta_0,\boeta_0)\|>L_2
\Big\}.
\end{align*}
Taking $\Pr$ of both sides and using the two bounds established earlier, for large $n$,
\[
\Pr\Big(
{\Pr}^*\big(
\sqrt{n}\,\|\hpi^*_n(\boheta_n)-\pi(\btheta_0,\boeta_0)\|>L_1+L_2
\big)
>\eps
\Big)
\leqslant
2\eps.
\]
This implies, for each $\eps>0$, there is an $L>0$ such that for large $n$,
\begin{equation}\label{ineq-pistar-final-joint}
\Pr\Big(
{\Pr}^{*}\big(
\sqrt{n}\,\|\hpi^*_n(\boheta_n)-\pi(\btheta_0,\boeta_0)\|>L
\big)
>\eps
\Big)
\leqslant \eps.
\end{equation}

In \eqref{eq:outer-split-three}, fix $\eta>0$ and choose $M$ large enough such that
$\tfrac{M}{3C_1}\geqslant L_\pi$ and $\tfrac{M}{3C_1}\geqslant L_\zeta$, where $L_\pi$ is the $L$  in
\eqref{ineq-pistar-final-joint} with $\eps=\eta/6$ and $L_\zeta$ is the $M$ in
\eqref{eq:bound-zeta-mixed} with $\delta=\eta/6$ and $\eps=\eta/6$.
Moreover, since $\|\boheta_n-\boeta_0\|=O_p(n^{-1/2})$, we can choose $M$ large
enough so that, for all large $n$,
\[
\Pr\Big(
\sqrt n\|\boheta_n-\boeta_0\|>\tfrac{M}{3C_2}
\Big)\leqslant \tfrac{\eta}{6}.
\]
From \eqref{eq:outer-split-three}, with this choice of $M$, for large $n$,
\[
\Pr\Big({\Pr}^*(\sqrt n\|\bhtheta_n^*-\btheta_0\|>M)>\eta\Big)
\leqslant \eps+\tfrac{\eta}{6}+\tfrac{\eta}{6}+\tfrac{\eta}{6}.
\]
For any fixed $\eps>0$, since
\[
\Pr\Big(
  {\Pr}^*\big((E_{n,\eps}^*)^c\big)>\tfrac{\eta}{2}
\Big)\;\rightarrow\;0,
\]
implies that for large $n$,
\[
\Pr\Big(
  {\Pr}^*\big((E_{n,\eps}^*)^c\big)>\tfrac{\eta}{2}
\Big)
\leqslant \eps.
\]
As $\eps$ is arbitrary fix $\eps=\tfrac{\eta}{2}$.
Thus, for every $\eta>0$, there exists an
$M>0$ such that, for large $n$,
\begin{equation*}
\Pr\Big(
  {\Pr}^*\big(
    \sqrt{n}\,\|\bhtheta_n^*-\btheta_0\|>M
  \big)
  >\eta
\Big)
\leqslant\eta.
\end{equation*}
Moreover, fix arbitrary $\eps>0$ and $\delta>0$ and set
$\eta := \min(\eps,\delta)$. Let $M$ be the corresponding
constant. Then, for large $n$,
\begin{equation*}
\Pr\Big(
  {\Pr}^*\big(
    \sqrt{n}\,\|\bhtheta_n^*-\btheta_0\|>M
  \big)
  >\eps
\Big)
\leqslant
\Pr\Big(
  {\Pr}^*\big(
    \sqrt{n}\,\|\bhtheta_n^*-\btheta_0\|>M
  \big)
  >\eta
\Big)
\leqslant\eta
\leqslant\delta.
\end{equation*}
Thus, for every $\eps>0$ and $\delta>0$ there exists an  $M>0$ such that,
for large $n$,
\begin{equation}\label{eq:boundboothat}
\Pr\Big(
  {\Pr}^*\big(
    \sqrt{n}\,\|\bhtheta_n^*-\btheta_0\|>M
  \big)
  >\eps
\Big)
\leqslant \delta.
\end{equation}

Recall that for all $\eps,\delta>0$,
\begin{equation}\label{eq:bootstrap-remainder-final-again}
\Pr\Big(
{\Pr}^*\!\Big(
\frac{\|\br_n^*\|}
{\|\bhtheta_n^*-\btheta_0\|+\|\boheta_n-\boeta_0\|}
>\eps
\Big)>\delta
\Big)\;\rightarrow\;0,
\end{equation}
 and note that
\[
\sqrt n\,\|\br_n^*\|
=
\sqrt n\big(\|\bhtheta_n^*-\btheta_0\|+\|\boheta_n-\boeta_0\|\big)
\frac{\|\br_n^*\|}{\|\bhtheta_n^*-\btheta_0\|+\|\boheta_n-\boeta_0\|}.
\]
Let $\eps,\eta>0$. Then
\begin{equation}\label{eq:event-inclusion-joint}
\big\{\sqrt n\,\|\br_n^*\|>\eps\big\}
\subseteq
\Big\{
\frac{\|\br_n^*\|}{\|\bhtheta_n^*-\btheta_0\|+\|\boheta_n-\boeta_0\|}>\eta
\Big\}
\;\cup\;
\Big\{
\sqrt n\big(\|\bhtheta_n^*-\btheta_0\|+\|\boheta_n-\boeta_0\|\big)>\frac{\eps}{\eta}
\Big\}.
\end{equation}
Indeed, if both
$\frac{\|\br_n^*\|}{\|\bhtheta_n^*-\btheta_0\|+\|\boheta_n-\boeta_0\|}\leqslant \eta$
and
$\sqrt n(\|\bhtheta_n^*-\btheta_0\|+\|\boheta_n-\boeta_0\|)\leqslant \eps/\eta$
held, then
\[
\sqrt n\,\|\br_n^*\|
=
\sqrt n(\|\bhtheta_n^*-\btheta_0\|+\|\boheta_n-\boeta_0\|)
\frac{\|\br_n^*\|}{\|\bhtheta_n^*-\btheta_0\|+\|\boheta_n-\boeta_0\|}
\leqslant
\frac{\eps}{\eta}\eta
=
\eps,
\]
contradicting $\sqrt n\,\|\br_n^*\|>\eps$.
We further decompose the second event as
\[
\Big\{\sqrt n\big(\|\bhtheta_n^*-\btheta_0\|+\|\boheta_n-\boeta_0\|\big)>\frac{\eps}{\eta}\Big\}
\subseteq
\Big\{\sqrt n\|\bhtheta_n^*-\btheta_0\|>\frac{\eps}{2\eta}\Big\}
\ \cup\
\Big\{\sqrt n\|\boheta_n-\boeta_0\|>\frac{\eps}{2\eta}\Big\}.
\]
Hence,
\begin{align}
\big\{\sqrt n\,\|\br_n^*\|>\eps\big\}
\subseteq\;
&\Big\{
\frac{\|\br_n^*\|}{\|\bhtheta_n^*-\btheta_0\|+\|\boheta_n-\boeta_0\|}>\eta
\Big\}
\label{eq:event-inclusion-joint-split}\nonumber\\
&\cup\;
\Big\{
\sqrt n\|\bhtheta_n^*-\btheta_0\|>\frac{\eps}{2\eta}
\Big\}
\;\cup\;
\Big\{
\sqrt n\|\boheta_n-\boeta_0\|>\frac{\eps}{2\eta}
\Big\}.
\end{align}
Taking ${\Pr}^*$ of both sides yields
\begin{align}
{\Pr}^*(\sqrt n\,\|\br_n^*\|>\eps)
\leqslant\;
&{\Pr}^*\Big(
\frac{\|\br_n^*\|}{\|\bhtheta_n^*-\btheta_0\|+\|\boheta_n-\boeta_0\|}>\eta
\Big)
\label{eq:star-decomp-joint-split}\nonumber\\
&+
{\Pr}^*\Big(
\sqrt n\|\bhtheta_n^*-\btheta_0\|>\frac{\eps}{2\eta}
\Big)
+
{\Pr}^*\Big(
\sqrt n\|\boheta_n-\boeta_0\|>\frac{\eps}{2\eta}
\Big).
\end{align}
Note that,
\[
{\Pr}^*\Big(
\sqrt n\|\boheta_n-\boeta_0\|>\frac{\eps}{2\eta}
\Big)
=
\bone\Big\{
\sqrt n\|\boheta_n-\boeta_0\|>\frac{\eps}{2\eta}
\Big\}.
\]
Therefore, for each $t>0$, using that if $a+b+c>t$ then at least one term exceeds $t/3$,
\begin{align}
\Big\{{\Pr}^*(\sqrt n\,\|\br_n^*\|>\eps)>t\Big\}
\subseteq\;
&\Big\{{\Pr}^*\Big(
\frac{\|\br_n^*\|}{\|\bhtheta_n^*-\btheta_0\|+\|\boheta_n-\boeta_0\|}>\eta
\Big)>\frac{t}{3}\Big\}
\label{eq:outer-split-3terms}\nonumber\\
&\cup
\Big\{{\Pr}^*\Big(
\sqrt n\|\bhtheta_n^*-\btheta_0\|>\frac{\eps}{2\eta}
\Big)>\frac{t}{3}\Big\}\nonumber\\
&
\cup
\Big\{\bone\Big\{
\sqrt n\|\boheta_n-\boeta_0\|>\frac{\eps}{2\eta}
\Big\}>\frac{t}{3}\Big\}.
\end{align}
Taking $\Pr$ on both sides gives
\begin{align}
\Pr\Big({\Pr}^*(\sqrt n\,\|\br_n^*\|>\eps)>t\Big)
\leqslant\;
&\Pr\Bigg({\Pr}^*\Big(
\frac{\|\br_n^*\|}{\|\bhtheta_n^*-\btheta_0\|+\|\boheta_n-\boeta_0\|}>\eta
\Big)>\frac{t}{3}\Bigg)
\label{eq:outer-bound-3terms}\nonumber\\
&+
\Pr\Bigg({\Pr}^*\Big(
\sqrt n\|\bhtheta_n^*-\btheta_0\|>\frac{\eps}{2\eta}
\Big)>\frac{t}{3}\Bigg)\nonumber\\
&+
\Pr\Big(
\sqrt n\|\boheta_n-\boeta_0\|>\frac{\eps}{2\eta}
\Big).
\end{align}
Fix $\eps,t,\gamma>0$. 
First, by \eqref{eq:boundboothat} there exists 
an $M_1>0$ such that, for all large $n$,
\begin{equation*}
\Pr\Bigg({\Pr}^*\Big(
\sqrt n\|\bhtheta_n^*-\btheta_0\|>M_1
\Big)>\frac{t}{3}\Bigg)\leqslant \frac{\gamma}{3}.
\end{equation*}
Second, since $\|\boheta_n-\boeta_0\|=O_p(n^{-1/2})$, 
there exists an $M_2>0$ such that, for all large $n$,
\begin{equation*}
\Pr\Big(
\sqrt n\|\boheta_n-\boeta_0\|>M_2
\Big)\leqslant \frac{\gamma}{3}.
\end{equation*}
Let $M:=\max\{M_1,M_2\}$ and define $\eta:=\frac{\eps}{2M}$, so that
$\frac{\eps}{2\eta}=M$.
Then, for all large $n$,
\begin{align}
\Pr\Bigg({\Pr}^*\Big(
\sqrt n\|\bhtheta_n^*-\btheta_0\|>\frac{\eps}{2\eta}
\Big)>\frac{t}{3}\Bigg)
&=
\Pr\Bigg({\Pr}^*\Big(
\sqrt n\|\bhtheta_n^*-\btheta_0\|>M
\Big)>\frac{t}{3}\Bigg)
\leqslant \frac{\gamma}{3},
\label{eq:second-term-small}
\\
\Pr\Big(
\sqrt n\|\boheta_n-\boeta_0\|>\frac{\eps}{2\eta}
\Big)
&=
\Pr\Big(
\sqrt n\|\boheta_n-\boeta_0\|>M
\Big)
\leqslant \frac{\gamma}{3}.
\label{eq:third-term-small}
\end{align}
Finally, by \eqref{eq:bootstrap-remainder-final-again}, for the fixed $\eta>0$ chosen above,
\begin{equation*}
\Pr\Bigg({\Pr}^*\Big(
\frac{\|\br_n^*\|}{\|\bhtheta_n^*-\btheta_0\|+\|\boheta_n-\boeta_0\|}>\eta
\Big)>\frac{t}{3}\Bigg)\rightarrow 0,
\end{equation*}
and therefore for large $n$,
\[
\Pr\Bigg({\Pr}^*\Big(
\frac{\|\br_n^*\|}{\|\bhtheta_n^*-\btheta_0\|+\|\boheta_n-\boeta_0\|}>\eta
\Big)>\frac{t}{3}\Bigg)\leqslant \frac{\gamma}{3}.
\]
Combining \eqref{eq:outer-bound-3terms}, \eqref{eq:second-term-small},
\eqref{eq:third-term-small}, and the last display, we conclude that for large $n$
\[
\Pr\Big({\Pr}^*(\sqrt n\,\|\br_n^*\|>\eps)>t\Big)\leqslant \gamma.
\]
Since $\gamma>0$ is arbitrary, this shows that for every $\eps,t>0$,
\begin{equation}\label{eq:bootreminderzero2}
\Pr\Big({\Pr}^*(\sqrt n\,\|\br_n^*\|>\eps)>t\Big)\;\rightarrow\;0.
\end{equation}

Finally, rewrite the decomposition of $\bhtheta_n^*-\btheta_0$ as
\[
\bhtheta_n^*-\btheta_0
=
\bK_0\big(\hpi^*_n(\boheta_n)-\pi(\btheta_0,\boeta_0)-\bB_0(\boheta_n-\boeta_0)\big)
-\bK_0\br_n^*
-\bK_0\zeta_n(\bhtheta_n^*,\boheta_n).
\]
Add and subtract $\pi(\btheta_0,\boheta_n)$, then we have 
\[
\hpi^*_n(\boheta_n)-\pi(\btheta_0,\boeta_0)-\bB_0(\boheta_n-\boeta_0)
=
\hpi^*_n(\boheta_n)-\pi(\btheta_0,\boheta_n)
+\br_{\eta,n},
\]
where $\br_{\eta,n}=o_p(n^{-1/2})$ as shown at the end of the proof of Proposition~\ref{pro:II_WI}.
Thus
\[
\bhtheta_n^*-\btheta_0
=
\bK_0\big(\hpi^*_n(\boheta_n)-\pi(\btheta_0,\boheta_n)\big)
+\bK_0\br_{\eta,n}
-\bK_0\br_n^*
-\bK_0\zeta_n(\bhtheta_n^*,\boheta_n).
\]
Define the combined bootstrap remainder (including the nuisance-centering term)
\begin{equation*}
\btr_n^*
:=
\bK_0\br_{\eta,n}
-\bK_0\br_n^*
-\bK_0\zeta_n(\bhtheta_n^*,\boheta_n).
\end{equation*}
Then the decomposition becomes
\begin{equation*}
\bhtheta_n^*-\btheta_0
=
\bK_0\big(\hpi^*_n(\boheta_n)-\pi(\btheta_0,\boheta_n)\big)
+\btr_n^*.
\end{equation*}
Set $C:=\|\bK_0\|<\infty$. By the triangle inequality,
\begin{equation*}
\sqrt n\,\|\btr_n^*\|
\leqslant
C\sqrt n\,\|\br_{\eta,n}\|
+
C\sqrt n\,\|\br_n^*\|
+
C\sqrt n\,\|\zeta_n(\bhtheta_n^*,\boheta_n)\|.
\end{equation*}
Fix $\eps>0$. Then
\begin{equation*}
\big\{\sqrt n\,\|\btr_n^*\|>\eps\big\}
\subseteq
\Big\{
\sqrt n\,\|\br_{\eta,n}\|>\tfrac{\eps}{3C}
\Big\}
\;\cup\;
\Big\{
\sqrt n\,\|\br_n^*\|>\tfrac{\eps}{3C}
\Big\}
\;\cup\;
\Big\{
\sqrt n\,\|\zeta_n(\bhtheta_n^*,\boheta_n)\|>\tfrac{\eps}{3C}
\Big\}.
\end{equation*}
Taking ${\Pr}^*$ on both sides yields
\begin{align*}
{\Pr}^*\big(\sqrt n\,\|\btr_n^*\|>\eps\big)
\leqslant\;
&
{\Pr}^*\Big(
\sqrt n\,\|\br_{\eta,n}\|>\tfrac{\eps}{3C}
\Big)
+
{\Pr}^*\Big(
\sqrt n\,\|\br_n^*\|>\tfrac{\eps}{3C}
\Big)
+
{\Pr}^*\Big(
\sqrt n\,\|\zeta_n(\bhtheta_n^*,\boheta_n)\|>\tfrac{\eps}{3C}
\Big).
\end{align*}
Since $\br_{\eta,n}$ is computed on the original sample, it is fixed in the bootstrap
world, hence
\[
{\Pr}^*\Big(
\sqrt n\,\|\br_{\eta,n}\|>\tfrac{\eps}{3C}
\Big)
=
\bone\Big\{
\sqrt n\,\|\br_{\eta,n}\|>\tfrac{\eps}{3C}
\Big\}.
\]
Now fix any $t>0$. Using that if $a+b+c>t$ with $a,b,c\geqslant 0$ then at least one term
exceeds $t/3$, we obtain
\begin{align*}
\Big\{
{\Pr}^*\big(\sqrt n\,\|\btr_n^*\|>\eps\big)>t
\Big\}
\subseteq\;
&
\Big\{
\bone\{\sqrt n\,\|\br_{\eta,n}\|>\tfrac{\eps}{3C}\}>\tfrac{t}{3}
\Big\}
\\&\cup
\Big\{
{\Pr}^*\Big(
\sqrt n\,\|\br_n^*\|>\tfrac{\eps}{3C}
\Big)>\tfrac{t}{3}
\Big\}
\cup
\Big\{
{\Pr}^*\Big(
\sqrt n\,\|\zeta_n(\bhtheta_n^*,\boheta_n)\|>\tfrac{\eps}{3C}
\Big)>\tfrac{t}{3}
\Big\}.
\end{align*}
Taking $\Pr$ gives
\begin{align}
\Pr\Big(
{\Pr}^*(\sqrt n\,\|\btr_n^*\|>\eps)>t
\Big)
\leqslant\;
&
\Pr\Big(
\sqrt n\,\|\br_{\eta,n}\|>\tfrac{\eps}{3C}
\Big)
\nonumber\\
&+
\Pr\Big(
{\Pr}^*\Big(
\sqrt n\,\|\br_n^*\|>\tfrac{\eps}{3C}
\Big)>\tfrac{t}{3}
\Big)\nonumber
\\
&+
\Pr\Big(
{\Pr}^*\Big(
\sqrt n\,\|\zeta_n(\bhtheta_n^*,\boheta_n)\|>\tfrac{\eps}{3C}
\Big)>\tfrac{t}{3}
\Big).
\label{eq:tilde-rn-outer-bound-joint}
\end{align}

By assumption, $\br_{\eta,n}=o_p(n^{-1/2})$, hence
$\Pr(\sqrt n\,\|\br_{\eta,n}\|>\eps/(3C))\rightarrow 0$.
Moreover, by \eqref{eq:bootreminderzero2} and \eqref{eq:bound-zeta-mixed},
the last two terms on the right-hand side of \eqref{eq:tilde-rn-outer-bound-joint}
converge to $0$. Therefore, for any $\eps,t>0$,
\[
\Pr\Big(
{\Pr}^*(\sqrt n\,\|\btr_n^*\|>\eps)>t
\Big)\;\rightarrow\;0,
\]
which is equivalent to ${\Pr}^*(\sqrt n\,\|\btr_n^*\|>\eps)\rightarrow_p 0$.

\end{proof}

\begin{proposition}\label{pro:adphat}
Under the same assumptions as in Lemma~\ref{lem:ratio-uniform-eta}, let $\bx_1,\ldots,\bx_n$ denote the observed sample, assumed to be randomly drawn
from $F_{\btheta_0}$.
For each tuning parameters $\boeta\in\bH$, define
\begin{equation*}
\bar a_n(\boeta)
=
\frac{1}{n}\sum_{i=1}^n a(\bx_i,\boeta),
\qquad
\bar b_n(\boeta)
=
\frac{1}{n}\sum_{i=1}^n b(\bx_i,\boeta).
\end{equation*}
Denote the corresponding population targets by
$
\mu_a(\boeta):=\E[a(\bx,\boeta)]$, and $\mu_b(\boeta):=\E[b(\bx,\boeta)]$.
The estimator $\hP_n(\boeta)$ is defined componentwise, for 
$j=1,\dots,d+d(d+1)/2$, as
\[
\big(\hP_n(\boeta)\big)_j
:=
\begin{cases}
\dfrac{(\bar a_n)_j(\boeta)}
      {(\bar b_n)_j(\boeta)},
&  \big|(\bar b_n)_j(\boeta)\big|
\geqslant \tilde \delta, \\[6pt]
\dfrac{(\bar a_n)_j(\boeta)}
      {\tilde \delta},
&  \big|(\bar b_n)_j(\boeta)\big|
< \tilde \delta,
\end{cases}
\]
Moreover
$P(\boeta):=\mu_a(\boeta)\oslash \mu_b(\boeta)$.
Let $Z(\bx,\boeta):=(a(\bx,\boeta)^T,b(\bx,\boeta)^T)^T$,  $\bSigma_{ab}:=\Cov(Z(\bx,\boeta_0))$, and 
\begin{equation*}
\bD_0
:=
\Big[
\diag\big(\mu_b(\boeta_0)\big)^{-1},
\;
-\diag\!\left(
\mu_a(\boeta_0)\oslash \mu_b(\boeta_0)^2
\right)
\Big].
\end{equation*}

Then
\begin{equation*}
\sqrt n\Big(\hP_n(\boheta_n)-P(\boheta_n)\Big)\ \rightarrow_d N\!\Big(\bzero,\ \bD_0\,\bSigma_{ab}\,\bD_0^T\Big).
\end{equation*}

\end{proposition}

\begin{proof}
Define 
$\bar Z_n(\boeta):=\frac1n\sum_{i=1}^nZ(\bx_i,\boeta)$, and 
$\mu_Z(\boeta):=\left(\mu_a(\boeta)^T, \mu_b(\boeta)^T\right)^T$,
and let
$\bar \Delta_n:=\bar Z_n(\boheta_n)-\bar Z_n(\boeta_0)$, and
$\bar\Delta:=\mu_Z(\boheta_n)-\mu_Z(\boeta_0)
$.

With these definitions, the decomposition
\begin{equation}
\sqrt n\big(\bar Z_n(\boheta_n)-\mu_Z(\boheta_n)\big)
=
\sqrt n\big(\bar Z_n(\boeta_0)-\mu_Z(\boeta_0)\big)
+\sqrt n(\bar\Delta_n-\bar\Delta),
\label{eq:decomp}
\end{equation}
holds.
By boundedness of $a(\bx,\boeta)$ and $b(\bx,\boeta)$, all coordinates of
$a(\bx,\boeta_0)$ and $b(\bx,\boeta_0)$ have finite second moments.   Moreover, since the observations
$\bx_1,\ldots,\bx_n$ are i.i.d., the vectors
$Z_1(\boeta_0),\ldots,Z_n(\boeta_0)$ are i.i.d.
Hence, by the multivariate central limit theorem \citep{van2000asymptotic},
\begin{equation*}
\sqrt n\big(\bar Z_n(\boeta_0)-\mu_Z(\boeta_0)\big)\ \rightarrow_{d}\ N(\bzero,\bSigma_{ab}).
\end{equation*}

We now show that $\sqrt n(\bar\Delta_n-\bar\Delta)=o_p(1)$.
Let $Z_j(\bx,\boeta)$ indicates the $j$-th entry of $Z(\bx,\boeta)$.
Define the real-valued class
\[
\mathcal F_j:=\bigl\{\bx\mapsto Z_j(\bx,\boeta):\ \boeta\in\bH\bigr\},
\]
for $j\in\{1,\ldots,2(d+d(d+1)/2)\}$.
By uniform boundedness, there exists an $M<\infty$ such that
$|Z_j(\bx,\boeta)|\leqslant M$ for all $\bx$ and all $\boeta\in\bH$, hence
$\mathcal F_j$ admits a constant envelope.
Moreover, by Lipschitz continuity of $a$ and $b$ in $\boeta$ uniformly in $\bx$,
there exist $L_j<\infty$ such that for all $\boeta_1,\boeta_2\in\bH$ and all $\bx$,
\[
|Z_j(\bx,\boeta_1)-Z_j(\bx,\boeta_2)|\leqslant L_j\|\boeta_1-\boeta_2\|.
\]
Thus $\mathcal F_j$ is a bounded
Lipschitz class indexed by a compact, and thus bounded, subset of
$\mathbb R^{r}$. By Example~19.7 of \citet{van2000asymptotic}, the bracketing entropy integral is
finite, and therefore $\mathcal F_j$ is $F_{\btheta_0}$-Donsker by Theorem~19.5 of
\citet{van2000asymptotic}.
Define the random and limit functions
\[
f_{n,j}(\bx):=Z_j(\bx,\boheta_n),
\qquad
f_{0,j}(\bx):=Z_j(\bx,\boeta_0).
\]
Since $\boheta_n\in\bH$, then $f_{n,j}\in\mathcal F_j$. Moreover
\[
\E\!\left[f_{0,j}(X)^2\right]
=
\E\!\left[ Z_j(X,\boeta_0)^2 \right]
< \infty,
\]
by the uniform boundedness assumption on $a(\bx,\boeta)$ and $b(\bx,\boeta)$. Furthermore, by the Lipschitz bound,
\[
|f_{n,j}(\bx)-f_{0,j}(\bx)|
=|Z_j(\bx,\boheta_n)-Z_j(\bx,\boeta_0)|
\leqslant L_j\|\boheta_n-\boeta_0\|,
\]
hence
\[
\int(f_{n,j}(X)-f_{0,j}(X))^2dF_{\btheta_0}(X)
\leqslant L_j^2\|\boheta_n-\boeta_0\|^2\rightarrow_p 0,
\]
since $\boheta_n-\boeta_0=O_p(n^{-1/2})$.

By Lemma~19.24 of \citet{van2000asymptotic}, for 
each fixed $j$,
\[
\sqrt n(\bar\Delta_n-\bar\Delta)_j=\sqrt n(\frac{1}{ n}\sum_{i=1}^n
\Bigl[
\bigl(Z_j(\bx_i,\boheta_n)-Z_j(\bx_i,\boeta_0)\bigr)
-(
(\mu_Z(\boheta_n))_j-(\mu_Z(\boeta_0))_j)
\Bigr])
\ \rightarrow_p\ 0 .
\]
Thus $\sqrt n(\bar\Delta_n-\bar\Delta)=o_p(1)$.

Applying Slutsky’s theorem \citep{van2000asymptotic}
to the decomposition in \eqref{eq:decomp} yields
\begin{equation}
\sqrt n\big(\bar Z_n(\boheta_n)-\mu_Z(\boheta_n)\big)\ \rightarrow_{d}\ N(\bzero,\bSigma_{ab}).
\label{eq:CLT-vector}
\end{equation}

Define $g:\mathbb R^{2(d+d(d+1)/2)}\rightarrow \mathbb R^{d+d(d+1)/2}$ by $g(\bz):=\bu\oslash \bv,$
with $ \bz=(\bu^T,\bv^T)^T$, $\bu,\bv\in\mathbb R^{d+d(d+1)/2}$, $\min_j |v_j| > \tilde{c}$ for some $\tilde{c}>0$.
Noting that $\mu_{b}(\boeta)=\mu_{b}(\btheta_0,\boeta)$, by assumption, for every $\boeta\in\bH$,
$\min_j|\mu_{b,j}(\boeta)|> c$ so $g$ is differentiable at $\mu_Z(\boheta_n)$.
On the event $E_n$ we have 
$\min_j |(\bar b_n)_j(\boheta_n)|\geqslant c$ since $(\bar b_n)_j(\boheta_n)$  has the same distribution as $(\bar b_n)_j(\btheta_0,\boheta_n)$,  it follows that $g$ is well-defined at $\bar Z_n(\boheta_n)$ on $E_n$.
Therefore, on the event $E_n$, a Taylor expansion of $g$ at $\mu_Z(\boheta_n)$ gives
\begin{equation}\label{eq:expansiong}
g(\bar Z_n(\boheta_n))
= g(\mu_Z(\boheta_n))
+ Dg(\mu_Z(\boheta_n))(\bar Z_n(\boheta_n)-\mu_Z(\boheta_n))
+ \br_{n,g},
\end{equation}
where
\begin{equation*}
Dg(\mu_Z(\boheta_n))
=\big[\,\diag(\mu_b(\boheta_n))^{-1},\ -\,\diag(\mu_a(\boheta_n)\oslash \mu_b(\boheta_n)^2)\,\big].
\end{equation*}
The remainder term satisfies for all $\eps>0$,
\begin{equation*}
\Pr\!\left(
E_n\ \cap\
\left\{
\frac{\|\br_{n,g}\|}{\|\bar Z_n(\boheta_n)-\mu_Z(\boheta_n)\|}
>\eps
\right\}
\right)
\ \rightarrow\ 0,
\end{equation*}
equivalently
\begin{equation*}
I_{E_n}\,\|\br_{n,g}\|
=
o_p\!\left(
\|\bar Z_n(\boheta_n)-\mu_Z(\boheta_n)\|
\right).
\end{equation*}
This holds because on $E_n$ the map $g$ is differentiable at $\mu_Z(\boheta_n)$ and
$\bar Z_n(\boheta_n)-\mu_Z(\boheta_n)\rightarrow_p0$, by using similar argument as in  Proposition~\ref{pro:II_WI}.
Since $\bar Z_n(\boheta_n)-\mu_Z(\boheta_n)=O_p(n^{-1/2})$, then
$I_{E_n}\br_{n,g} =o_p(n^{-1/2})$.

By assumption, there exists an $L<\infty$ such that for all
$\boeta_1,\boeta_2$ and all $\bx$,
\[
\|a(\bx,\boeta_1)-a(\bx,\boeta_2)\|
\leqslant L\|\boeta_1-\boeta_2\|.
\]
Thus, by Jensen's inequality
\[
\|\mu_a(\boeta_1)-\mu_a(\boeta_2)\|
=
\big\|\E[a(X,\boeta_1)-a(X,\boeta_2)]\big\|
\leqslant
\E\|a(X,\boeta_1)-a(X,\boeta_2)\|
\leqslant
L\|\boeta_1-\boeta_2\|.
\]
An analogous result holds for $\mu_b(\boeta)$.
Hence each coordinate of $\mu_a(\boeta)$ and $\mu_b(\boeta)$ is Lipschitz continuous in $\boeta$, and therefore continuous at $\boeta_0$.
Since $\boheta_n \rightarrow_p \boeta_0$, the 
continuous mapping theorem
\citep{van2000asymptotic} implies $\mu_Z(\boheta_n)\rightarrow_p \mu_Z(\boeta_0)$.
Moreover, as $Dg(\cdot)$ is continuous at $\mu_Z(\boeta_0)$,
we obtain $
Dg(\mu_Z(\boheta_n))\rightarrow_pDg(\mu_Z(\boeta_0))$.

Recall that $\hP_n(\boheta_n)=g(\bar Z_n(\boheta_n))$ on $E_n$, hence
\[
\sqrt n\big(\hP_n(\boheta_n)-P(\boheta_n)\big)
=
I_{E_n}\,\sqrt n\big(g(\bar Z_n(\boheta_n))-g(\mu_Z(\boheta_n))\big)
+
\sqrt n\,I_{E_n^c}\big(\hP_n(\boheta_n)-g(\mu_Z(\boheta_n))\big).
\]
Since $\|\hP_n(\boheta_n)\|\leqslant M/\tilde \delta$ as $\|\bar a_n(\boheta_n)\|\leqslant M$ and
$\|g(\mu_Z(\boheta_n))\|\leqslant M/c$, we have
\[
\Big\|
\hP_n(\boheta_n)-g(\mu_Z(\boheta_n))
\Big\|
\leqslant
\|\hP_n(\boheta_n)\|+\|g(\mu_Z(\boheta_n))\|
\leqslant
M\left(\frac{1}{\tilde \delta}+\frac{1}{c}\right).
\]
Hence
\[
\sqrt n\,I_{E_n^c}
\big(P_n(\boheta_n)-g(\mu_Z(\boheta_n))\big)
=
O_p\!\big(\sqrt n\,I_{E_n^c}\big).
\]
Since $\Pr(E_n^c)=O(n^{-\alpha})$ with $\alpha>1/2$, we have  that
\[
\sqrt n\,I_{E_n^c}
\big(P_n(\boheta_n)-g(\mu_Z(\boheta_n))\big)
=o_p(1).
\]
Moreover,  because $\Pr(E_n)\rightarrow 1$,  $I_{E_n}\rightarrow_p 1$, and therefore
\[
I_{E_n}\,\sqrt n\big(g(\bar Z_n(\boheta_n))-g(\mu_Z(\boheta_n))\big)
\ \rightarrow_d
N\!\Big(\bzero,\ \bD_0\,\bSigma_{ab}\,\bD_0^T\Big),
\]
by \eqref{eq:CLT-vector}, \eqref{eq:expansiong}, and Slutsky’s theorem
\citep{van2000asymptotic}. Therefore,
\begin{equation}\label{eq:conphatetahat}
\sqrt n\big(\hP_n(\boheta_n)-P(\boheta_n)\big)
\ \rightarrow_d\ 
N\!\Big(\bzero,\ \bD_0\,\bSigma_{ab}\,\bD_0^T\Big).
\end{equation}
\end{proof}

\begin{proposition}
\label{pro:converbootmean}
Under the same assumptions and notation of Proposition~\ref{pro:adphat}, let $\bx_1^*,\ldots,\bx_n^*$ be a  bootstrap sample from the empirical distribution function $F_n$ from the observed sample $\bx_1,\ldots,\bx_n$, and set
\begin{equation*}
\bar a_n^*(\boeta)=\frac{1}{n}\sum_{i=1}^n a(\bx_i^*,\boeta),\qquad
\bar b_n^*(\boeta)=\frac{1}{n}\sum_{i=1}^n b(\bx_i^*,\boeta).
\end{equation*}
The estimator  $\hP_n^*(\boeta)$ is defined componentwise, for $j=1,\dots,d+d(d+1)/2$, as
\[
\big(\hP_n^*(\boeta)\big)_j
:=
\begin{cases}
\dfrac{(\bar a_n^*)_j(\boeta)}
      {(\bar b_n^*)_j(\boeta)},
&  \big|(\bar b_n^*)_j(\boeta)\big|
\geqslant \tilde \delta, \\[6pt]
\dfrac{(\bar a_n^*)_j(\boeta)}
      {\tilde \delta},
&  \big|(\bar b_n^*)_j(\boeta)\big|
< \tilde \delta.
\end{cases}
\]
Then 
\begin{equation*}
\sqrt n\big(\hP_n^*(\boheta_n)-\hP_n(\boheta_n)\big)
\ \rightarrow_d\ N\!\big(\bzero,\,\bD_0\,\bSigma_{ab}\,\bD_0^T\big),
\end{equation*}
in probability.
\end{proposition}

\begin{proof}
Define 
$\bar Z_n(\boeta):=\frac1n\sum_{i=1}^nZ(\bx_i,\boeta)$, and 
$\mu_Z(\boeta):=\left(\mu_a(\boeta)^T, \mu_b(\boeta)^T\right)^T$,
and let
$\bar \Delta_n:=\bar Z_n(\boheta_n)-\bar Z_n(\boeta_0)$, and
$\bar\Delta:=\mu_Z(\boheta_n)-\mu_Z(\boeta_0)
$, and $\Delta_i:=Z(\bx_i,\boheta_n)-Z(\bx_i,\boeta_0)$,

Define the bootstrap analogues $\bar Z_n^*(\boeta)$, $\bar\Delta_n^*$, and
$\Delta_i^*$ of $\bar Z_n(\boeta)$, $\bar\Delta_n$, and $\Delta_i$,
respectively, obtained by replacing each observation $\bx_i$ with its
bootstrap counterpart $\bx_i^*$.
Then
\begin{equation}
\sqrt n\big(\bar Z_n^*(\boheta_n)-\bar Z_n(\boheta_n)\big)
=\sqrt n\big(\bar Z_n^*(\boeta_0)-\bar Z_n(\boeta_0)\big)
+\sqrt n(\bar\Delta_n^*-\bar\Delta_n).
\label{eq:boot-decomp}
\end{equation}
By boundedness of $a(\bx,\boeta)$ and 
$b(\bx,\boeta)$, all coordinates of
$a(\bx,\boeta_0)$ and $b(\bx,\boeta_0)$ have 
finite second moments.  
Hence, by \citetsupp{bickel1981some} 
Theorem~2.2(a), we have 
\begin{equation}\label{eq:bootCLT}
\mathcal \!\sqrt n(\bar Z_n^*(\boeta_0)-\bar Z_n(\boeta_0))
\rightarrow_d N(\bzero,\bSigma_{ab}),
\end{equation}
holds almost surely.
Since 
$\Delta_i^*$ are i.i.d. from the law assigning
mass $1/n$ to each $\Delta_i$, then
\begin{equation*}
\E^*(\Delta_i^*)=\bar\Delta_n
\qquad\text{and}\qquad
\Cov^*(\Delta_i^*)=\frac{1}{n}\sum_{i=1}^n(\Delta_i-\bar\Delta_n)(\Delta_i-\bar\Delta_n)^T.
\end{equation*}
Hence,
\begin{equation*}
\Cov^*(\bar\Delta_n^*)
= \Cov^*\!\left(\frac{1}{n}\sum_{i=1}^n \Delta_i^*\right)
= \frac{1}{n^2}\sum_{i=1}^n \Cov^*(\Delta_i^*)
= \frac{1}{n}\Cov^*(\Delta_i^*).
\end{equation*}
Then, it follows that
\begin{equation*}
\Tr[ \Cov^*\!\big(\sqrt n(\bar\Delta_n^*-\bar\Delta_n)\big)]
= n\,\Tr[\Cov^*(\bar\Delta_n^*)]
= \Tr[\Cov^*(\Delta_i^*)]
= \frac{1}{n}\sum_{i=1}^n\|\Delta_i-\bar\Delta_n\|^2.
\end{equation*}
Moreover, using $\bar\Delta_n=\E^*\bar\Delta_n^*$, we obtain
\begin{align*}
\E^*\Big\|\sqrt n(\bar\Delta_n^*-\bar\Delta_n)\Big\|^2
&=
\Tr\!\Big(\Cov^*\big(\sqrt n(\bar\Delta_n^*-\bar\Delta_n)\big)\Big)
\\
&=
\frac{1}{n}\sum_{i=1}^n\|\Delta_i-\bar\Delta_n\|^2.
\end{align*}
Furthermore, by the inequality
\(
\|\bu-\bv\|^2\leqslant 2\|\bu\|^2+2\|\bv\|^2
\)
and Jensen's inequality,
\begin{align*}
\frac{1}{n}\sum_{i=1}^n\|\Delta_i-\bar\Delta_n\|^2
&\leqslant
\frac{2}{n}\sum_{i=1}^n\|\Delta_i\|^2
+2\|\bar\Delta_n\|^2
\\
&\leqslant
\frac{2}{n}\sum_{i=1}^n\|\Delta_i\|^2
+\frac{2}{n}\sum_{i=1}^n\|\Delta_i\|^2
=
\frac{4}{n}\sum_{i=1}^n\|\Delta_i\|^2.
\end{align*}
By the Lipschitz assumptions on $a$ and $b$, there exist constants $L_a,L_b<\infty$ such that, for every $i$,
\[
\|\Delta_i\|
\leqslant
\|a(\bx_i,\boheta_n)-a(\bx_i,\boeta_0)\|
+
\|b(\bx_i,\boheta_n)-b(\bx_i,\boeta_0)\|
\leqslant
(L_a+L_b)\|\boheta_n-\boeta_0\| .
\]
Hence,
\[
\frac{1}{n}\sum_{i=1}^n\|\Delta_i\|^2
\leqslant
(L_a+L_b)^2\|\boheta_n-\boeta_0\|^2,
\]
and therefore
\begin{equation}\label{eq:expbootdiff}
\E^*\Big\|\sqrt n(\bar\Delta_n^*-\bar\Delta_n)\Big\|^2
\leqslant
4(L_a+L_b)^2\|\boheta_n-\boeta_0\|^2
\rightarrow_p 0.
\end{equation}

Let 
\begin{equation*}
T_n^*:=\sqrt n\big(\bar Z_n^*(\boeta_0)-\bar Z_n(\boeta_0)\big),\qquad
R_n^*:=\sqrt n\big(\bar\Delta_n^*-\bar\Delta_n\big),
\end{equation*}
so by \eqref{eq:boot-decomp},
\begin{equation*}
\sqrt n\big(\bar Z_n^*(\boheta_n)-\bar Z_n(\boheta_n)\big)=T_n^*+R_n^*.
\end{equation*}
Fix $\varphi\in\bPhi$, where $\bPhi:=\lbrace\varphi:\mathbb{R}^{2(d+d(d+1)/2)}\rightarrow\mathbb{R}:|\varphi(\bx)-\varphi(\by)|\leqslant L \|\bx-\by\|, |\varphi(\bx)|\leqslant B, \quad \bx,\by\in\mathbb{R}^{2(d+d(d+1)/2)}, L< \infty,B< \infty\rbrace$ is a set of all
 bounded $L$-Lipschitz functions, then
\begin{equation}\label{eq:split}
\Big|\E^*[\varphi(T_n^*+R_n^*)]-\E[\varphi(Z)]\Big| \leqslant
\Big|\E^*[\varphi(T_n^*+R_n^*)]-\E^*[\varphi(T_n^*)]\Big|
+
\Big|\E^*[\varphi(T_n^*)]-\E[\varphi(Z)]\Big|,
\end{equation}
where $Z\sim N(\bzero,\bSigma_{ab})$.
 By linearity,
\begin{equation*}
\Big|\E^*[\varphi(T_n^*+R_n^*)]-\E^*[\varphi(T_n^*)]\Big|
=\Big|\E^*\!\big[\varphi(T_n^*+R_n^*)-\varphi(T_n^*)\big]\Big|.
\end{equation*}
 By Jensen's inequality,
\begin{equation*}
\Big|\E^*\!\big[\varphi(T_n^*+R_n^*)-\varphi(T_n^*)\big]\Big|
\leqslant \E^*\!\big|\varphi(T_n^*+R_n^*)-\varphi(T_n^*)\big|.
\end{equation*}
As, $\varphi$ is Lipschitz,
\begin{equation*}
\big|\varphi(T_n^*+R_n^*)-\varphi(T_n^*)\big|
\leqslant L\|R_n^*\|.
\end{equation*}
Taking conditional expectations, we have
\begin{equation*}
\E^*\!\big|\varphi(T_n^*+R_n^*)-\varphi(T_n^*)\big|
\leqslant L\E^*\|R_n^*\|.
\end{equation*}
Thus, 
\begin{align}\label{eq:lipschitz}
\Big|\E^*[\varphi(T_n^*+R_n^*)]-\E^*[\varphi(T_n^*)]\Big|
&=\Big|\E^*\big[\varphi(T_n^*+R_n^*)-\varphi(T_n^*)\big]\Big|
\leqslant \E^*\big|\varphi(T_n^*+R_n^*)-\varphi(T_n^*)\big|\nonumber
\\&\leqslant L\E^*\|R_n^*\|\leqslant L\sqrt{\E^*\|R_n^*\|^2},
\end{align}
where the last inequality follows from the Cauchy-Schwarz inequality.
By \eqref{eq:expbootdiff}, $\E^*\|R_n^*\|^2\rightarrow_p 0$, thus,
for all $\eps>0$,
\begin{equation}\label{eq:rem-prob}
\Pr\!\big(\E^*\|R_n^*\|^2>\eps\big)\ \rightarrow\ 0.
\end{equation}

Moreover, by  \eqref{eq:bootCLT}, $
T_n^*\ \rightarrow_{d}\ N(\bzero,\bSigma_{ab}),
$
almost surely, and, thus, by the Portmanteau 
theorem \citep{van2000asymptotic}, as $\varphi$ 
is bounded and $L$-Lipschitz, we have
\begin{equation}\label{eq:clt-bl}
\E^*[\varphi(T_n^*)]\ \rightarrow_p\ \E[\varphi(Z)].
\end{equation}

Combining \eqref{eq:split} with the Boole’s inequality and using \eqref{eq:lipschitz}, \eqref{eq:rem-prob}, and \eqref{eq:clt-bl}, for any $\eps>0$,
\begin{align*}
\Pr(\Big|&\E^*[\varphi(T_n^*+R_n^*)]-\E[\varphi(Z)]\Big|>\eps)
\\&\leqslant
\Pr\!\Big(\Big|\E^*[\varphi(T_n^*+R_n^*)]-\E^*[\varphi(T_n^*)]\Big|>\tfrac{\eps}{2}\Big)+\Pr\!\Big(\Big|\E^*[\varphi(T_n^*)]-\E[\varphi(Z)]\Big|>\tfrac{\eps}{2}\Big)
\\&\leqslant
\Pr\!\Big(\E^*\|R_n^*\|^2>\left(\frac{\eps}{2L}\right)^2\Big)
+\Pr\!\Big(\Big|\E^*[\varphi(T_n^*)]-\E[\varphi(Z)]\Big|>\tfrac{\eps}{2}\Big)
\ \rightarrow\ 0.
\end{align*}
We implicitly assume $L>0$ as in the case $L=0$ the bound holds trivially.
Therefore, for every bounded $L$-Lipschitz function $\varphi\in \bPhi$, and any $\eps>0$
\begin{equation*}
\Pr\!\Big(\big|\E^*[\varphi(T_n^*+R_n^*)]-\E[\varphi(Z)]\big|>\eps\Big)\ \rightarrow\ 0,
\end{equation*}
which implies 
\begin{equation}\label{eq:adbarZ}
    \sqrt n\big(\bar Z_n^*(\boheta_n)-\bar Z_n(\boheta_n)\big)
\ \rightarrow_d\ N(\bzero,\bSigma_{ab}),
\end{equation}
in probability.

Define
\[
E_n^*
:=
\Big\{
\inf_{\boeta\in\bH}\ \min_{j}\big|(\bar b_n^*)_j(\boeta)\big|\geqslant c
\Big\},
\]
Define
$G_n:=E_n\cap  E_n^*$.
On the event $G_n$, all denominators of $\bar Z_n(\boheta_n)$ and $\bar Z_n^*(\boheta_n)$
are bounded away from $0$ by $c$.
Define $g:\mathbb R^{2(d+d(d+1)/2)}\rightarrow \mathbb R^{d+d(d+1)/2}$ by
$g(\bz)=\bu\oslash \bv$, where $\bz=(\bu^T,\bv^T)^T$ and 
$\bu,\bv\in\mathbb R^{d+d(d+1)/2}$, $\min_j |v_j| > \tilde{c}$ for some $\tilde{c}>0$.
On $G_n$, the denominator blocks of both $\bar Z_n(\boheta_n)$ and
$\bar Z_n^*(\boheta_n)$ satisfy
$\min_j|(\bar b_n)_j(\boheta_n)|\geqslant c$, and $\min_j|(\bar b_n^*)_j(\boheta_n)|\geqslant c$,
so $g$ is well defined and differentiable at both
$\bar Z_n(\boheta_n)$ and $\bar Z_n^*(\boheta_n)$.
Hence, on $G_n$, a first-order Taylor expansion of $g$ at $\bar Z_n(\boheta_n)$ gives
\begin{equation}\label{eq:tylorexpansionZ*}
g(\bar Z_n^*(\boheta_n))- g(\bar Z_n(\boheta_n))
=
Dg(\bar Z_n(\boheta_n))
\big(\bar Z_n^*(\boheta_n)-\bar Z_n(\boheta_n)\big)
+ \br_{n,g}^*,
\end{equation}
where 
\begin{equation*}
\br_{n,g}^*
:=
g(\bar Z_n^*(\boheta_n))
-
g(\bar Z_n(\boheta_n))
-
Dg(\bar Z_n(\boheta_n))
\big(\bar Z_n^*(\boheta_n)-\bar Z_n(\boheta_n)\big).
\end{equation*}
For similar arguments as in the proof of Proposition~\ref{pro:II_WI_boot}
to obtain \eqref{eq:bootstrap-remainder-final}, for each $\eps,\delta>0$,
\begin{equation}\label{eq:ratereminder2}
\Pr\Big(
  {\Pr}^*\big( \Big\{ \frac{\|\br_{n,g}^*\|}
       {\|\bar Z_n^*(\boheta_n)-\bar Z_n(\boheta_n)\|} >\eps \Big\}\cap G_n\big)>\delta
\Big)\rightarrow 0.
\end{equation}

Moreover, as shown at the end of the proof of Proposition~\ref{pro:adphat},  $\mu_Z(\boheta_n)\rightarrow_p \mu_Z(\boeta_0)$ and $\bar Z_n(\boheta_n)-\mu_Z(\boheta_n)=O_p(n^{-1/2})$, thus, as $Dg(\cdot)$ is continuous at $\mu_Z(\boeta_0)$, the continuous mapping theorem
\citep{van2000asymptotic} implies 
\begin{equation}\label{eq:convDn}
Dg(\bar Z_n(\boheta_n)) \rightarrow_p Dg(\mu_Z(\boeta_0))=\bD_0.
\end{equation}

Write
$X_n^*:=\sqrt n(\bar Z_n^*(\boheta_n)-\bar Z_n(\boheta_n))$ and
$Y_n^*:=\sqrt n\,\br_{n,g}^*$, with
$\bD_n:=Dg(\bar Z_n(\boheta_n))$.
Fix a bounded $L$-Lipschitz function $\varphi\in\bPhi$.
By using the triangle inequality and splitting according to $G_n$,
\begin{align*}
&\Big|\E^*\!\big[\varphi(\sqrt n(\hP_n^*(\boheta_n)-\hP_n(\boheta_n)))\big]
-\E\big[\varphi(Z_P)\big]\Big|
\nonumber\\
&\leqslant
\Big|\E^*\!\big[\varphi(\sqrt n(\hP_n^*(\boheta_n)-\hP_n(\boheta_n)))\,I_{G_n}\big]
-\E\big[\varphi(Z_P)\big]\Big|
+
\Big|\E^*\!\big[\varphi(\sqrt n(\hP_n^*(\boheta_n)-\hP_n(\boheta_n)))\,I_{G_n^c}\big]\Big|
\nonumber\\
&\leqslant
\Big|\E^*\!\big[\varphi(\sqrt n(g(\bar Z_n^*(\boheta_n))-g(\bar Z_n(\boheta_n))))\,I_{G_n}\big]
-\E\big[\varphi(Z_P)\big]\Big|
+
C_{\varphi}\,{\Pr}^*(G_n^c),
\end{align*}
where $C_{\varphi}=\sup_{\bx\in\mathbb{R}^{d+d(d+1)/2}}\|\varphi(\bx)\|$, and $Z_P\sim N\!\big(\bzero,\,\bD_0\bSigma_{ab}\bD_0^T\big)$.
By using that the Taylor expansion holds on $G_n$ and  adding and
subtracting $\E^*[\varphi(\bD_n X_n^*)I_{G_n}]$ yields
\begin{align}
&\Big|\E^*\!\big[\varphi(\sqrt n(g(\bar Z_n^*(\boheta_n))-g(\bar Z_n(\boheta_n))))I_{G_n}\big]
-\E\big[\varphi(Z_P)\big]\Big|
\nonumber\\
&\leqslant
\Big|\E^*\!\big[\varphi(\bD_n X_n^*+Y_n^*)\,I_{G_n}\big]
-
\E^*\!\big[\varphi(\bD_n X_n^*)\,I_{G_n}\big]\Big|
\nonumber\\
&\quad+
\Big|\E^*\!\big[\varphi(\bD_n X_n^*)\,I_{G_n}\big]-\E\big[\varphi(Z_P)\big]\Big|.
\label{eq:split2}
\end{align}
In the last term in \eqref{eq:split2}, add and subtract
$\E^*[\varphi(\bD_0 X_n^*)]$
\begin{align}\label{eq:middleterm}
\Big|\E^*\!\big[\varphi(\bD_n X_n^*)\,I_{G_n}\big]
-\E\big[\varphi(Z_P)\big]\Big|
&\leqslant
\Big|\E^*\!\big[(\varphi(\bD_n X_n^*)-\varphi(\bD_0 X_n^*))\,I_{G_n}\big]\Big|
\nonumber\\
&\quad+
\Big|\E^*\!\big[\varphi(\bD_0 X_n^*)\big]
-\E\big[\varphi(Z_P)\big]\Big|
\nonumber\\
&\quad+
C_{\varphi}{\Pr}^*(G_n^c)
\end{align}
Combining \eqref{eq:split2} and \eqref{eq:middleterm} gives
\begin{align}
\Big|\E^*\!\big[\varphi(\sqrt n(\hP_n^*(\boheta_n)-\hP_n(\boheta_n)))\big]
-\E\big[\varphi(Z_P)\big]\Big|
&\leqslant
\Big|\E^*\!\big[\varphi(\bD_n X_n^*+Y_n^*)\,I_{G_n}\big]
-
\E^*\!\big[\varphi(\bD_n X_n^*)\,I_{G_n}\big]\Big|
\nonumber\\
&\quad+
\Big|\E^*\!\big[(\varphi(\bD_n X_n^*)-\varphi(\bD_0 X_n^*))\,I_{G_n}\big]\Big|
\nonumber\\
&\quad+
\Big|\E^*\!\big[\varphi(\bD_0 X_n^*)\big]-\E\big[\varphi(Z_P)\big]\Big|
+
2\,C_{\varphi}\,{\Pr}^*(G_n^c).
\label{eq:split-g-detailed}
\end{align}

From the definition of $L$-Lipschitz function and Jensen's inequality,
\begin{align}
\Big|\E^*\big[\varphi(\bD_n X_n^*+Y_n^*)\,I_{G_n}\big]
-\E^*\big[\varphi(\bD_n X_n^*)\,I_{G_n}\big]\Big|
&\leqslant
\E^*\!\Big[\big|\varphi(\bD_n X_n^*+Y_n^*)-\varphi(\bD_n X_n^*)\big|\,I_{G_n}\Big]
\nonumber\\
&\leqslant
L\,\E^*\!\big[\|Y_n^*\|\,I_{G_n}\big].
\label{eq:E*var}
\end{align}

From \eqref{eq:adbarZ},
for every bounded Lipschitz function $\varphi$,
\begin{equation}\label{eq:converexpectboot}
\E^*\!\left[\varphi(X_n^*)\right]
\;\rightarrow_p\;
\E\!\left[\varphi(Z)\right],
\end{equation}
where $Z\sim N(\bzero,\bSigma_{ab})$.
Fix $\eps>0$.  
Since the Gaussian law is tight, there exists an $M<\infty$ such that
\begin{equation*}
\Pr(\|Z\|>M/2)<\eps/4.
\end{equation*}
Define the bounded Lipschitz function $\phi_M:\mathbb R^{d+d(d+1)/2}\rightarrow[0,1]$ by
\begin{equation*}
\phi_M(\bx)=
\begin{cases}
0, & \|\bx\|\leqslant M/2,\\[3pt]
\displaystyle \frac{\|\bx\|-M/2}{M/2}, & M/2<\|\bx\|<M,\\[6pt]
1, & \|\bx\|\geqslant M,
\end{cases}
\end{equation*}
so that $
I_{\{\|\bx\|>M\}} \;\leqslant\; \phi_M(\bx)
\;\leqslant\;
I_{\{\|x\|>M/2\}}$ .
From the definition of $\phi_M$ we have
$
I_{\{\|\bx\|>M\}} \;\leqslant\; \phi_M(\bx)$, 
and therefore, taking conditional expectations,
\begin{equation}\label{eq:prstarTstarM}
{\Pr}^*\bigl(\|X_n^*\|>M\bigr)
= \E^*\!\left[I_{\{\|X_n^*\|>M\}}\right]
\;\leqslant\;
\E^*\!\left[\phi_M(X_n^*)\right].
\end{equation}
As $
 \phi_M(\bx)
\;\leqslant\; 
I_{\{\|\bx\|>M/2\}}$,
taking expectations with respect to $Z$ gives
\begin{equation*}
\E[\phi_M(Z)]
\;\leqslant\;
\E\!\left[I_{\{\|Z\|>M/2\}}\right]
=
\Pr(\|Z\|>M/2).
\end{equation*}
By \eqref{eq:converexpectboot} applied to $\phi_M$,
\begin{equation*}
\E^*[\phi_M(X_n^*)] \rightarrow_p \E[\phi_M(Z)]
\leqslant \Pr(\|Z\|>M/2) < \eps/4.
\end{equation*}
Since $\E^*[\phi_M(X_n^*)]\to_p \E[\phi_M(Z)]<\eps/4<\eps$,
it follows that
\[
\Pr\!\left(\E^*[\phi_M(X_n^*)]>\eps\right)\rightarrow 0.
\]
Thus, by using \eqref{eq:prstarTstarM}, we have
\begin{equation*}
\Pr\!\left(
  {\Pr}^*(\|X_n^*\|>M)>\eps
\right)
\;\leqslant\;
\Pr\!\left(
  \E^*[\phi_M(X_n^*)]>\eps
\right)
\;\rightarrow\;
0.
\end{equation*}
This shows that \eqref{eq:adbarZ} implies that for every $\eps>0$
there exists an $M<\infty$ such that
\begin{equation}\label{eq:boundsqrtnZ}
\Pr\!\left(
  {\Pr}^*(\|X_n^*\|>M)>\eps
\right)\;\rightarrow\;0.
\end{equation}

Define $\Delta_n^*:=\bar Z_n^*(\boheta_n)-\bar Z_n(\boheta_n)$.
Then, from \eqref{eq:ratereminder2}, for every $\eps,\delta>0$,
\begin{equation}\label{eq:ratereminder-repeat}
\Pr\!\Big(
  {\Pr}^*\!\Big(
  \Big\{ \tfrac{\|\br_{n,g}^*\|}{\|\Delta_n^*\|} > \eps \Big\}\cap G_n
  \Big) > \delta
\Big)
\;\rightarrow\;0 .
\end{equation}
Moreover, \eqref{eq:boundsqrtnZ} implies that for every $\eps,\delta>0$  
there exists an $L<\infty$ such that for large $n$
\begin{equation}\label{eq:Deltan-epsdelta}
\Pr\!\Big(
  {\Pr}^*\!\big( \sqrt n\,\|\Delta_n^*\|>L \big)>\eps
\Big)\;<\;\delta.
\end{equation}

Fix $\delta'>0$.  
For any $\eps'>0$, consider the events
\begin{equation*}
A_n := \Big\{\frac{\|\br_{n,g}^*\|}{\|\Delta_n^*\|}\leqslant \eps'\Big\}\cap G_n,
\qquad
A_n^c := \Big\{\frac{\|\br_{n,g}^*\|}{\|\Delta_n^*\|}>\eps'\Big\}\cap G_n.
\end{equation*}
On $A_n$ we have
\begin{equation*}
\|Y_n^*\|
=\sqrt n\,\|\br_{n,g}^*\|
\leqslant \eps'\sqrt n\,\|\Delta_n^*\|.
\end{equation*}
Therefore, whenever both $A_n$ and
$\{\sqrt n\,\|\Delta_n^*\|\leqslant\delta'/\eps'\}$ occur, we obtain
\begin{equation*}
\|Y_n^*\|\leqslant \delta'.
\end{equation*}
Equivalently,
\begin{equation*}
\{\|Y_n^*\|>\delta'\}
\;\subseteq\;
A_n^c
\;\cup\;
\Big\{\sqrt n\,\|\Delta_n^*\|>\delta'/\eps'\Big\}.
\end{equation*}
Applying ${\Pr}^*$ and using the union bound gives
\begin{equation}\label{eq:Y-split}
{\Pr}^*\!\Big(
\Big\{\|Y_n^*\|>\delta'\Big\}\cap G_n
\Big)
\;\leqslant\;
{\Pr}^*\!\Big(
\Big\{\tfrac{\|\br_{n,g}^*\|}{\|\Delta_n^*\|}>\eps'\Big\}\cap G_n
\Big)
\;+\;
{\Pr}^*\!\Big(
\sqrt n\,\|\Delta_n^*\|>\delta'/\eps'
\Big).
\end{equation}
From \eqref{eq:Y-split}, for $\eta>0$,
\begin{align*}
\Pr\!\Big(
  {\Pr}^*\!\Big(
    \Big\{\|Y_n^*\|>\delta'\Big\}\cap G_n
  \Big)>\eta
\Big)
&\leqslant
\Pr\!\Big(
  {\Pr}^*\!\Big(
    \Big\{\tfrac{\|\br_{n,g}^*\|}{\|\Delta_n^*\|}>\eps'\Big\}\cap G_n
  \Big)>\eta/2
\Big)
\\
&\quad+
\Pr\!\Big(
  {\Pr}^*\!\Big(
    \sqrt n\,\|\Delta_n^*\|>\delta'/\eps'
  \Big)>\eta/2
\Big).
\end{align*}
The first probability tends to zero by \eqref{eq:ratereminder-repeat}.
For the second probability, choose $\eps'$ such that
$\delta'/\eps' > L$, where $L$ is from \eqref{eq:Deltan-epsdelta}.
Then by \eqref{eq:Deltan-epsdelta} the second probability also tends to zero.
Thus for any $\delta',\eta>0$,
\begin{equation*}
\Pr\!\Big(
  {\Pr}^*\!\Big(\big\{\|Y_n^*\|>\delta'\big\}\cap G_n\Big)>\eta
\Big)\;\rightarrow\;0,
\end{equation*}
that is,
\begin{equation}\label{eq:convY^*}
{\Pr}^*\!\Big(\big\{\|Y_n^*\|>\delta'\big\}\cap G_n\Big)\;\rightarrow_p\;0.
\end{equation}

Fix $\delta'>0$. 
Write
\[
\E^*\big[\|Y_n^*\| I_{G_n}\big]
=
\E^*\big[\|Y_n^*\| I_{\{\|Y_n^*\|\leqslant \delta'\}\cap G_n}\big]
+
\E^*\big[\|Y_n^*\| I_{\{\|Y_n^*\|>\delta'\}\cap G_n}\big].
\]
On $\{\|Y_n^*\|\leqslant \delta'\}\cap G_n$ we have $\|Y_n^*\|\leqslant \delta'$, hence
\[
\E^*\big[\|Y_n^*\| I_{\{\|Y_n^*\|\leqslant \delta'\}\cap G_n}\big]
\leqslant
\delta'.
\]
Moreover, by Cauchy-Schwarz,
\[
\E^*\big[\|Y_n^*\| I_{\{\|Y_n^*\|>\delta'\}\cap G_n}\big]
\leqslant
\sqrt{\E^*\big[\|Y_n^*\|^2 I_{G_n}\big]}\;
\sqrt{{\Pr}^*\big(\{\|Y_n^*\|>\delta'\}\cap G_n\big)}.
\]
Therefore,
\begin{equation}\label{eq:EY-Gn}
\E^*\big[\|Y_n^*\| I_{G_n}\big]
\leqslant
\delta'
+
\sqrt{\E^*\big[\|Y_n^*\|^2 I_{G_n}\big]}\;
\sqrt{{\Pr}^*\big(\{\|Y_n^*\|>\delta'\}\cap G_n\big)}.
\end{equation}
By \eqref{eq:convY^*}, for every fixed $\delta'>0$,
\[
{\Pr}^*\big(\{\|Y_n^*\|>\delta'\}\cap G_n\big)\ \rightarrow_p\ 0,
\]
and hence
\[
\sqrt{{\Pr}^*\big(\{\|Y_n^*\|>\delta'\}\cap G_n\big)}\ \rightarrow_p\ 0.
\]
Moreover, on $G_n$ the map $g(\bu,\bv)=\bu\oslash \bv$ has uniformly bounded second
derivatives, so there
exists a constant $K<\infty$ such that on $G_n$,
\[
\|\br_{n,g}^*\|\leqslant K\|\Delta_n^*\|^2,
\qquad
\Delta_n^*:=\bar Z_n^*(\boheta_n)-\bar Z_n(\boheta_n).
\]
Therefore, on $G_n$,
\[
\|Y_n^*\|
=
\sqrt n\,\|\br_{n,g}^*\|
\leqslant
K\sqrt n\,\|\Delta_n^*\|^2,
\]
and thus
\[
\|Y_n^*\|^2 I_{G_n}
\leqslant
K^2\,n\,\|\Delta_n^*\|^4.
\]
Taking conditional expectations yields
\begin{equation}\label{eq:EY2-Gn}
\E^*\big[\|Y_n^*\|^2 I_{G_n}\big]
\leqslant
K^2\,n\,\E^*\|\Delta_n^*\|^4.
\end{equation}
By assumption $\|\bar Z_n(\boheta_n)\|\leqslant M$ and $\|\bar Z_n^*(\boheta_n)\|\leqslant M$, and hence
\[
\|\Delta_n^*\|
=
\|\bar Z_n^*(\boheta_n)-\bar Z_n(\boheta_n)\|
\leqslant
2M,
\qquad
\text{so}
\qquad
\|\Delta_n^*\|^4\leqslant (2M)^2\|\Delta_n^*\|^2.
\]
Therefore,
\begin{equation}\label{eq:Delta4-Delta2}
\E^*\|\Delta_n^*\|^4
\leqslant
(2M)^2\,\E^*\|\Delta_n^*\|^2.
\end{equation}
Since $\Delta_i^*:=Z(\bx_i^*,\boheta_n)-\bar Z_n(\boheta_n)$ are i.i.d.\ under ${\Pr}^*$
with $\E^*[\Delta_i^*]=0$ and $\|\Delta_i^*\|\leqslant 2M$, we have
\[
\E^*\|\Delta_n^*\|^2
=
\E^*\Big\|\frac1n\sum_{i=1}^n \Delta_i^*\Big\|^2
=
\frac1{n^2}\E^*\Big\|\sum_{i=1}^n \Delta_i^*\Big\|^2.
\]
Expanding the square and using independence and $\E^*[\Delta_i^*]=0$ yields
\[
\E^*\Big\|\sum_{i=1}^n \Delta_i^*\Big\|^2
=
\sum_{i=1}^n \E^*\|\Delta_i^*\|^2
+
2\sum_{i<j}\E^*\big[(\Delta_i^*)^T\Delta_j^*\big]
=
n\,\E^*\|\Delta_1^*\|^2.
\]
Therefore,
\[
\E^*\|\Delta_n^*\|^2
=
\frac1n\,\E^*\|\Delta_1^*\|^2
\leqslant
\frac1n(2M)^2.
\]
Combining this bound with \eqref{eq:Delta4-Delta2} yields
\[
\E^*\|\Delta_n^*\|^4
\leqslant
(2M)^2\,\E^*\|\Delta_n^*\|^2
\leqslant
(2M)^2\cdot \frac{(2M)^2}{n}
=
\frac{16M^4}{n},
\]
and hence
\[
n\,\E^*\|\Delta_n^*\|^4 \leqslant 16M^4.
\]
Substituting into \eqref{eq:EY2-Gn} gives
\[
\E^*\big[\|Y_n^*\|^2 I_{G_n}\big]
\leqslant
K^2\,n\,\E^*\|\Delta_n^*\|^4
\leqslant
16K^2M^4.
\]
In particular, $\E^*[\|Y_n^*\|^2 I_{G_n}]=O_p(1)$, and thus
\[
\sqrt{\E^*\big[\|Y_n^*\|^2 I_{G_n}\big]}\;
\sqrt{{\Pr}^*\big(\{\|Y_n^*\|>\delta'\}\cap G_n\big)}
\ \rightarrow_p\ 0.
\]
Substituting this and ${\Pr}^*(\{\|Y_n^*\|>\delta'\}\cap G_n)\to_p 0$ into
\eqref{eq:EY-Gn}, and letting $\delta'\downarrow 0$, we obtain
\[
\E^*\big[\|Y_n^*\| I_{G_n}\big]\ \rightarrow_p\ 0.
\]

Thus, from \eqref{eq:E*var}, we obtain
\begin{equation}\label{eq:cont1}
\Big|\E^*\big[\varphi(\bD_n X_n^*+Y_n^*)\,I_{G_n}\big]
-\E^*\big[\varphi(\bD_n X_n^*)\,I_{G_n}\big]\Big|
\ \rightarrow_p\ 0.
\end{equation}

It remains to control the second term in \eqref{eq:split-g-detailed}, namely
\[
\Big|\E^*\!\big[(\varphi(\bD_n X_n^*)-\varphi(\bD_0 X_n^*))\,I_{G_n}\big]\Big|.
\]
On $G_n$, for any $\bx$, by the $L$-Lipschitz property of $\varphi$,
\[
|\varphi(\bD_n \bx)-\varphi(\bD_0 \bx)|
\leqslant
L\,\|(\bD_n-\bD_0)\bx\|
\leqslant
L\,\|\bD_n-\bD_0\|\,\|\bx\|
\leqslant
L\,\|\vect(\bD_n)-\vect(\bD_0)\|\,\|\bx\|.
\]
Therefore,
\begin{align}
\Big|\E^*\!\big[(\varphi(\bD_n X_n^*)-\varphi(\bD_0 X_n^*))\,I_{G_n}\big]\Big|
&\leqslant
\E^*\!\Big[|\varphi(\bD_n X_n^*)-\varphi(\bD_0 X_n^*)|\,I_{G_n}\Big]
\nonumber\\
&\leqslant
L\,\|\vect(\bD_n)-\vect(\bD_0)\|\;\E^*\!\big[\|X_n^*\|\,I_{G_n}\big].
\label{eq:Dndec}
\end{align}

Recall the decomposition
\[
X_n^* = T_n^* + R_n^*,
\qquad
T_n^* = \sqrt n\big(\bar Z_n^*(\boeta_0)-\bar Z_n(\boeta_0)\big),
\qquad
R_n^* = \sqrt n(\bar\Delta_n^*-\bar\Delta_n).
\]
By the triangle inequality,
\begin{equation}\label{eq:EX-decomp}
\E^*\|X_n^*\|
\leqslant
\E^*\|T_n^*\|
+
\E^*\|R_n^*\|.
\end{equation}
Since $a(\bx,\boeta)$ and $b(\bx,\boeta)$ are uniformly bounded, there exists an $M>0$ such that
\[
\|Z(\bx,\boeta_0)\|
\leqslant
\|a(\bx,\boeta_0)\|+\|b(\bx,\boeta_0)\|
\leqslant 2M .
\]
Conditionally on the data, $Z(\bx_1^*,\boeta_0),\ldots,Z(\bx_n^*,\boeta_0)$ are i.i.d. with
${\Pr}^*\big(Z(\bx_1^*,\boeta_0)=Z(\bx_i,\boeta_0)\big)=1/n$.
Hence $\E^*[T_n^*]=0$ and by similar reasoning as in the proof of \eqref{eq:expbootdiff}, we have
\begin{align*}
\E^*\|T_n^*\|^2
&=\Tr\!\big(\Cov^*(T_n^*)\big)
=\Tr\!\big(\Cov^*(Z(\bx_1^*,\boeta_0))\big) \\
&=\E^*\big\|Z(\bx_1^*,\boeta_0)-\E^*(Z(\bx_1^*,\boeta_0))\big\|^2 \\
&=\frac1n\sum_{i=1}^n \big\|Z(\bx_i,\boeta_0)-\bar Z_n(\boeta_0)\big\|^2 \\
&\leqslant \frac1n\sum_{i=1}^n \|Z(\bx_i,\boeta_0)\|^2\leqslant (2M)^2.
\end{align*}
Hence, by the Cauchy-Schwarz inequality,
\[
\E^*\|T_n^*\|
\leqslant
\sqrt{\E^*\|T_n^*\|^2}
\leqslant
2M .
\]
Moreover, by \eqref{eq:expbootdiff},
$\E^*\|R_n^*\|^2 \to_p 0$, and therefore
\[
\E^*\|R_n^*\|
\leqslant
\sqrt{\E^*\|R_n^*\|^2}
\to_p 0 .
\]
Combining these bounds in \eqref{eq:EX-decomp} yields
\[
\E^*\|X_n^*\|
\leqslant
2M + o_p(1),
\]
so that $\E^*\|X_n^*\| = O_p(1)$.
Moreover, since $I_{G_n}\leqslant 1$, we have
$\E^*\!\big[\|X_n^*\|\,I_{G_n}\big]\leqslant \E^*\|X_n^*\|$.
By \eqref{eq:convDn},
$\|\vect(\bD_n)-\vect(\bD_0)\|\to_p 0$.
Therefore, by \eqref{eq:Dndec},
\begin{equation}\label{eq:cont2}
\Big|\E^*\!\big[(\varphi(\bD_n X_n^*)-\varphi(\bD_0 X_n^*))\,I_{G_n}\big]\Big|
\ \rightarrow_p\ 0.
\end{equation}
By \eqref{eq:adbarZ} and using the same reasoning as in the proof of Proposition~\ref{pro:bootstrap-theta} to obtain \eqref{eq:D0convergence}, $\bD_0 X_n^*\rightarrow_dZ_P$ in  probability, and as the map $\varphi$ is bounded Lipschitz, we have
\begin{equation}
\label{eq:cont3}
\Big|\E^*[\varphi(\bD_0 X_n^*)]-\E[\varphi(Z_P)]\Big|\ \rightarrow_p\ 0.
\end{equation}

Recall that $G_n = E_n \cap  E_n^*$. Then
$G_n^c \subseteq E_n^c \cup ( E_n^*)^c$,
hence
\[
{\Pr}^*(G_n^c)
\leqslant
{\Pr}^*(E_n^c)+{\Pr}^*(( E_n^*)^c).
\]
Since $E_n$ depends only on the original sample, $
{\Pr}^*(E_n^c)=I_{E_n^c}$.
By assumption, ${\Pr}(E_n^c)=O(n^{-\alpha})$ with $\alpha>1/2$, hence
$I_{E_n^c}\to_p 0$, and therefore
${\Pr}^*(E_n^c)\to_p 0$.

It remains to show that ${\Pr}^*(( E_n^*)^c)\to_p 0$.
Fix $\gamma>0$ and define the buffered event
\[
E_n^{(c+\gamma)}
:=
\Big\{
\inf_{\boeta\in\bH}\min_{j}\big|(\bar b_n)_j(\boeta)\big|\geqslant c+\gamma
\Big\}.
\]
On the event $E_n^{(c+\gamma)}$, if $( E_n^*)^c$ occurs, then there exist
$\tilde \boeta\in\bH$ and $j$ such that $|(\bar b_n^*)_j(\tilde\boeta)|<c$ and
$|(\bar b_n)_j(\tilde\boeta)|\geqslant c+\gamma$. Hence,
\[
\big|(\bar b_n^*)_j(\tilde\boeta)-(\bar b_n)_j(\tilde\boeta)\big|
\geqslant
\big|(\bar b_n)_j(\tilde\boeta)\big|-\big|(\bar b_n^*)_j(\tilde\boeta)\big|
>
(c+\gamma)-c
=
\gamma,
\]
and therefore
\[
( E_n^*)^c\cap E_n^{(c+\gamma)}
\subseteq
\Big\{
\sup_{\boeta\in\bH}\max_{j}
\big|(\bar b_n^*)_j(\boeta)-(\bar b_n)_j(\boeta)\big|
>\gamma
\Big\}.
\]
Consequently,
\begin{equation}\label{eq:tildeEnstar_split}
{\Pr}^*\!\big(( E_n^*)^c\cap E_n^{(c+\gamma)}\big)
\leqslant
{\Pr}^*\!\Big(
\sup_{\boeta\in\bH}\max_{j}
\big|
(\bar b_n^*)_j(\boeta)
-
(\bar b_n)_j(\boeta)
\big|
>\gamma
\Big).
\end{equation}
Let $\rho:=\gamma/(4L_b)$. Since $\bH$ is compact,  hence totally bounded, there exists a finite $\rho$-net $\mathcal N_\rho\subset\bH$,
that is, for every $\boeta\in\bH$ there exists $\boeta^\rho\in\mathcal N_\rho$ with
$\|\boeta-\boeta^\rho\|\leqslant\rho$, and $\#\mathcal N_\rho<\infty$.
Since $b(\bx,\boeta)$ is $L_b$-Lipschitz in $\boeta$, both
$\bar b_n(\boeta)$ and $\bar b_n^*(\boeta)$ inherit the same Lipschitz constant $L_b$, that is, for all $\boeta_1,\boeta_2\in\bH$,
\[
\max_j
\big|
(\bar b_n)_j(\boeta_1)
-
(\bar b_n)_j(\boeta_2)
\big|
\leqslant
L_b \|\boeta_1-\boeta_2\|,
\]
and
\[
\max_j
\big|
(\bar b_n^*)_j(\boeta_1)
-
(\bar b_n^*)_j(\boeta_2)
\big|
\leqslant
L_b \|\boeta_1-\boeta_2\|.
\]
Then, by the triangle inequality,
\begin{align*}
\max_j
\big|
(\bar b_n^*)_j(\boeta)
-
(\bar b_n)_j(\boeta)
\big|
&\leqslant
\max_j
\big|
(\bar b_n^*)_j(\boeta)
-
(\bar b_n^*)_j(\boeta^\rho)
\big|
\\
&\quad+
\max_j
\big|
(\bar b_n^*)_j(\boeta^\rho)
-
(\bar b_n)_j(\boeta^\rho)
\big|
\\
&\quad+
\max_j
\big|
(\bar b_n)_j(\boeta^\rho)
-
(\bar b_n)_j(\boeta)
\big|.
\end{align*}
By Lipschitz continuity,
\[
\max_j
\big|
(\bar b_n^*)_j(\boeta)
-
(\bar b_n^*)_j(\boeta^\rho)
\big|
\leqslant L_b\rho,
\qquad
\max_j
\big|
(\bar b_n)_j(\boeta^\rho)
-
(\bar b_n)_j(\boeta)
\big|
\leqslant L_b\rho,
\]
hence
\[
\max_j
\big|
(\bar b_n^*)_j(\boeta)
-
(\bar b_n)_j(\boeta)
\big|
\leqslant
\max_j
\big|
(\bar b_n^*)_j(\boeta^\rho)
-
(\bar b_n)_j(\boeta^\rho)
\big|
+
2L_b\rho.
\]
As $\rho=\gamma/(4L_b)$, we have $2L_b\rho=\gamma/2$, and therefore
\[
\max_j
\big|
(\bar b_n^*)_j(\boeta)
-
(\bar b_n)_j(\boeta)
\big|
\leqslant
\max_j
\big|
(\bar b_n^*)_j(\boeta^\rho)
-
(\bar b_n)_j(\boeta^\rho)
\big|
+
\gamma/2.
\]
Hence,
\[
\Big\{
\sup_{\boeta\in\bH}\max_j
\big|
(\bar b_n^*)_j(\boeta)
-
(\bar b_n)_j(\boeta)
\big|
>\gamma
\Big\}
\subseteq
\Big\{
\sup_{\boeta\in\mathcal N_\rho}\max_j
\big|
(\bar b_n^*)_j(\boeta)
-
(\bar b_n)_j(\boeta)
\big|
>\gamma/2
\Big\}.
\]
Fix $\boeta\in\mathcal N_\rho$ and $j$. Conditionally on the data,
$b_j(\bx_1^*,\boeta),\dots,b_j(\bx_n^*,\boeta)$ are i.i.d.\ bounded in $[-M,M]$
with mean $(\bar b_n)_j(\boeta)$, and
$(\bar b_n^*)_j(\boeta)=n^{-1}\sum_{i=1}^n b_j(\bx_i^*,\boeta)$. Therefore, by conditional Hoeffding's inequality,
\[
{\Pr}^*\!\Big(
\big|(\bar b_n^*)_j(\boeta)-(\bar b_n)_j(\boeta)\big|>\gamma/2
\Big)
\leqslant
2\exp\!\left(-\frac{n\gamma^2}{8M^2}\right).
\]
A union bound over $j$ and $\mathcal N_\rho$ yields
\[
{\Pr}^*\!\Big(
\sup_{\boeta\in\mathcal N_\rho}\max_j
\big|
(\bar b_n^*)_j(\boeta)
-
(\bar b_n)_j(\boeta)
\big|
>\gamma/2
\Big)
\leqslant
2\big(d+d(d+1)/2\big)\,\#\mathcal N_\rho\,
\exp\!\left(-\frac{n\gamma^2}{8M^2}\right)
\rightarrow 0.
\]
Combining with \eqref{eq:tildeEnstar_split} gives
\[
{\Pr}^*\!\big(( E_n^*)^c\cap E_n^{(c+\gamma)}\big)\rightarrow 0,
\]
almost surely.

Define $c_0
:=
\inf_{\boeta\in\bH}\min_{j}|(\mu_b)_j(\boeta)|$.
By assumption, $c_0>c$. Fix $\gamma>0$ such that $c+\gamma<c_0$, and set
$\eps
:=
c_0-(c+\gamma)
>0$.
Recall that
$(\bar b_n)_j(\boeta)
=
(\bar b_n)_j(\btheta_0,\boeta)$, then by \eqref{eq:barBmuB}, we have
\[
\sup_{\boeta\in\bH}\max_{j}
\big|
(\bar b_n)_j(\boeta)-(\mu_b)_j(\boeta)
\big|
\to_p\;0.
\]
Hence,
\[
\Pr\!\left(
\sup_{\boeta\in\bH}\max_{j}
\big|
(\bar b_n)_j(\boeta)-(\mu_b)_j(\boeta)
\big|
< \eps
\right)
\rightarrow 1.
\]
On this event, for every $\boeta\in\bH$ and every $j$,
\[
|(\bar b_n)_j(\boeta)|
\geqslant
|(\mu_b)_j(\boeta)|-\eps
\geqslant
c_0-\eps
=
c+\gamma,
\]
and therefore
\[
\inf_{\boeta\in\bH}\min_{j}|(\bar b_n)_j(\boeta)|
\geqslant
c+\gamma,
\]
that is, $E_n^{(c+\gamma)}$ occurs. Consequently,
$\Pr(E_n^{(c+\gamma)})\rightarrow 1$.
Finally,
\[
{\Pr}^*(( E_n^*)^c)
\leqslant
{\Pr}^*\!\big(( E_n^*)^c\cap E_n^{(c+\gamma)}\big)
+
I_{(E_n^{(c+\gamma)})^c}.
\]
Since $\Pr(E_n^{(c+\gamma)})\rightarrow 1$, we have $I_{(E_n^{(c+\gamma)})^c}\to_p 0$,
and therefore ${\Pr}^*(( E_n^*)^c)\to_p 0$.
Consequently,
\[
{\Pr}^*(G_n^c)
\leqslant
{\Pr}^*(E_n^c)+{\Pr}^*(( E_n^*)^c)
\ \rightarrow_p\ 0.
\]

From \eqref{eq:split-g-detailed}, combining \eqref{eq:cont1}, \eqref{eq:cont2},
\eqref{eq:cont3}, and using Boole's inequality, for every $\eps>0$,
\begin{align*}
\Pr\Big(
\Big|\E^*\!\big[\varphi(\sqrt n(\hP_n^*(\boheta_n)-&\hP_n(\boheta_n)))\big]
-\E\big[\varphi(Z_P)\big]\Big|>\eps
\Big)
\\&\leqslant
\Pr\Big(
\Big|\E^*\!\big[\varphi(\bD_n X_n^*+Y_n^*)\,I_{G_n}\big]
-\E^*\!\big[\varphi(\bD_n X_n^*)\,I_{G_n}\big]\Big|>\eps/4
\Big)
\\
&\quad+
\Pr\Big(
\Big|\E^*\!\big[(\varphi(\bD_n X_n^*)-\varphi(\bD_0 X_n^*))\,I_{G_n}\big]\Big|>\eps/4
\Big)
\\
&\quad+
\Pr\Big(
\Big|\E^*\!\big[\varphi(\bD_0 X_n^*)\big]-\E\big[\varphi(Z_P)\big]\Big|>\eps/4
\Big)
\\
&\quad+
\Pr\Big(2C_\varphi\,{\Pr}^*(G_n^c)>\eps/4\Big)
\ \rightarrow\ 0.
\end{align*}
Therefore, for every bounded $L$-Lipschitz $\varphi\in\bPhi$,
\begin{equation}\label{eq:convcondboot}
\Big|\E^*\!\big[\varphi(\sqrt n(\hP_n^*(\boheta_n)-\hP_n(\boheta_n)))\big]
-\E\big[\varphi(Z_P)\big]\Big|
\ \rightarrow_p\ 0,
\end{equation}
which is equivalent to
\begin{equation*}
\sqrt n\big(\hP_n^*(\boheta_n)-\hP_n(\boheta_n)\big)
\ \rightarrow_d\ N\!\big(\bzero,\,\bD_0\,\bSigma_{ab}\,\bD_0^T\big),
\end{equation*}
in probability.
\end{proof}

\begin{proposition}
\label{pro:bootstrap-theta}
Let us assume
\begin{enumerate}[label=(C\arabic*)]
\item 
The map $\pi$ is  differentiable in $\btheta$ and $\boeta$ at
$(\btheta_0,\boeta_0)$ with full-rank Jacobian
$\bA_0
=
\left.
\frac{\partial\, \pi(\btheta,\boeta)}{\partial \btheta}
\right|_{(\btheta,\boeta)=(\btheta_0,\boeta_0)}$, and let
  $\bK_0:=\bA_0^{-1}$.
\item  
The estimator $\bhtheta_n$ satisfies
  \begin{equation*}\bhtheta_n-\btheta_0
=\bK_0\big(\hpi_n(\boheta_n)-\pi(\btheta_0,\boheta_n)\big)
   +\tilde \br_n,
  \end{equation*}
  where $\tilde \br_n=o_p(n^{-1/2})$. 
\item 
For the bootstrap estimator $\bhtheta_n^*$ we have
  \begin{align*}
\bhtheta_n^*-\btheta_0
&=\bK_0\big(\hpi^*_n(\boheta_n)-\pi(\btheta_0,\boheta_n)\big)
   +\btr_n^*.
\end{align*}
where ${\Pr}^*(\sqrt{n}\,\|\btr_n^*\|>\eps)\rightarrow_p 0$.
\item
  Consider $Z_\pi\sim N(\mathbf{0},\bSigma)$ such that
  \begin{equation*}
  \sqrt{n}\bigl(
    \hpi_n(\boheta_n)-\pi(\btheta_0,\boheta_n)
  \bigr)
  \rightarrow_d Z_\pi,
  \end{equation*}
  and 
  \begin{equation*}
  \sqrt{n}\bigl(
    \hpi^*_n(\boheta_n)-\hpi_n(\boheta_n)
  \bigr)
  \rightarrow_d Z_\pi
  \quad\text{in probability.}
  \end{equation*}
\end{enumerate}

Then the bootstrap distribution of
$\sqrt{n}(\bhtheta_n^*-\bhtheta_n)$ is consistent for 
the distribution of 
\mbox{$\sqrt{n}(\bhtheta_n-\btheta_0)$} in the 
Kolmogorov-Smirnov distance, i.e.,
\begin{equation*}
\sup_{\bx\in\mathbb{R}^{d+d(d+1)/2}}
\Big|
  \Pr\bigl(
    \sqrt{n}(\bhtheta_n-\btheta_0)\leqslant \bx
  \bigr)
  -
  {\Pr}^*\bigl(
    \sqrt{n}(\bhtheta_n^*-\bhtheta_n)\leqslant \bx \,
  \bigr)
\Big|
\rightarrow_p0.
\end{equation*}
\end{proposition}

\begin{proof}

Assumption (C2) gives the linear expansion
\begin{equation*}
\bhtheta_n-\btheta_0
=
\bK_0\big(\hpi_n(\boheta_n)-\pi(\btheta_0,\boheta_n)\big)
+\btr_n,
\qquad
\btr_n=o_p(n^{-1/2}).
\end{equation*}
Thus,
\begin{equation*}
\sqrt{n}\,(\bhtheta_n-\btheta_0)
=
\bK_0
\sqrt{n}\big(\hpi_n(\boheta_n)-\pi(\btheta_0,\boheta_n)\big)
+
\sqrt{n}\,\btr_n,
\end{equation*}
and, $\sqrt{n}\,\btr_n\rightarrow_p 0$. By (C4), and the   continuous mapping theorem \citep{van2000asymptotic},
\begin{equation*}
\sqrt{n}\bK_0\big(\hpi_n(\boheta_n)-\pi(\btheta_0,\boheta_n)\big)
\rightarrow_d \tilde Z_\pi:=\bK_0 Z_\pi
\sim N\bigl(\mathbf 0,\bK_0\bSigma\bK_0^T\bigr),
\end{equation*}
and by Slutsky’s theorem \citep{van2000asymptotic},
\begin{equation}
\label{eq:theta-limit}
\sqrt{n}\,(\bhtheta_n-\btheta_0)\rightarrow_d \tilde Z_\pi.
\end{equation}

From (C3),
\begin{equation*}
\bhtheta_n^*-\btheta_0
=
\bK_0\big(\hpi^*_n(\boheta_n)-\pi(\btheta_0,\boheta_n)\big)
+\btr_n^*,
\end{equation*}
and subtracting the (C2) expansion for $\bhtheta_n-\btheta_0$ gives
\begin{equation*}
\bhtheta_n^*-\bhtheta_n
=
\bK_0\big(\hpi^*_n(\boheta_n)-\hpi_n(\boheta_n)\big)
+(\btr_n^*-\btr_n).
\end{equation*}
Thus
\begin{equation}
\label{eq:theta-star-expansion}
\sqrt{n}\,(\bhtheta_n^*-\bhtheta_n)
=
\sqrt{n}\bK_0\big(\hpi^*_n(\boheta_n)-\hpi_n(\boheta_n)\big)
+
\sqrt{n}\,(\btr_n^*-\btr_n).
\end{equation}
 By the triangle
inequality,
\begin{equation*}
\sqrt{n}\,\|\btr_n^*-\btr_n\|
\leqslant
\sqrt{n}\,\|\btr_n\|+\sqrt{n}\,\|\btr_n^*\|.
\end{equation*}
Therefore,
\begin{equation*}
\Bigl\{\sqrt{n}\,\|\btr_n^*-\btr_n\|>\eps\Bigr\}
\subseteq
\Bigl\{\sqrt{n}\,\|\btr_n\|>\eps/2\Bigr\}
\ \cup\
\Bigl\{\sqrt{n}\,\|\btr_n^*\|>\eps/2\Bigr\}.
\end{equation*}
Taking conditional probabilities ${\Pr}^*$ and using that $\btr_n$ does not
depend on the bootstrap resample, we obtain
\begin{equation*}
{\Pr}^*\Bigl(\sqrt{n}\,\|\btr_n^*-\btr_n\|>\eps\Bigr)
\leqslant
I_{\{\sqrt{n}\,\|\btr_n\|>\eps/2\}}
+
{\Pr}^*\Bigl(\sqrt{n}\,\|\btr_n^*\|>\eps/2\Bigr).
\end{equation*}
Since $\sqrt{n}\,\| \tilde \br_n\|\rightarrow_p0$, the indicator term converges in
probability to $0$. Moreover, by assumption,
\begin{equation*}
{\Pr}^*\Bigl(\sqrt{n}\,\|\btr_n^*\|>\eps/2\Bigr)\rightarrow_p0.
\end{equation*}
Thus, for every $\eps>0$
\begin{equation}\label{eq:diffreminder}
{\Pr}^*\bigl(\sqrt{n}\,\|\btr_n^*-\btr_n\|>\eps\bigr)
\;\rightarrow_p\;0.
\end{equation}

By (C4),
\begin{equation*}
\sqrt{n}\big(\hpi^*_n(\boheta_n)-\hpi_n(\boheta_n)\big)
\rightarrow_d Z_\pi
\quad\text{in probability}.
\end{equation*}
This means that for every bounded $L$-Lipschitz function
$\varphi:\mathbb{R}^{d+d(d+1)/2}\rightarrow\mathbb{R}$ we have
\begin{equation}\label{eq:BL-for-pi}
\big|
  \E^*[\varphi(\sqrt{n}\big(\hpi^*_n(\boheta_n)-\hpi_n(\boheta_n)\big))]
  - \E[\varphi(Z_\pi)]
\big|
\;\rightarrow_p\;0.
\end{equation}
Let $\psi:\mathbb{R}^{d+d(d+1)/2}\rightarrow\mathbb{R}$ be any bounded $L$-Lipschitz function, $L<\infty$,
and define $\varphi:\mathbb{R}^{d+d(d+1)/2}\rightarrow\mathbb{R}$ by $\varphi(\bx) := \psi(\bK_0 \bx)$, 
for $ \bx\in\mathbb{R}^{d+d(d+1)/2}$.
Since $\psi:\mathbb{R}^{d+d(d+1)/2}\rightarrow\mathbb{R}$ is bounded, say $|\psi(\bu)|\leqslant M$, $M<\infty$, for all
$\bu\in\mathbb{R}^{d+d(d+1)/2}$, we immediately have
\begin{equation*}
|\varphi(\bx)| = |\psi(\bK_0 \bx)| \leqslant M,
\end{equation*}
so $\varphi$ is bounded.
As $\psi$ is $L$-Lipschitz, for $\bx,\by\in\mathbb{R}^{d+d(d+1)/2}$,
\begin{equation*}
|\psi(\bx)-\psi(\by)| \leqslant L\|\bx-\by\|.
\end{equation*}
Thus, for any $\bx,\by\in\mathbb{R}^{d+d(d+1)/2}$,
\begin{equation*}
|\varphi(\bx)-\varphi(\by)|
= |\psi(\bK_0 \bx) - \psi(\bK_0 \by)|
\leqslant L\,\|\bK_0(\bx-\by)\|\leqslant L\,\|\vect(\bK_0)\|\,\|\bx-\by\|.
\end{equation*}
Thus $\varphi$ is Lipschitz with Lipschitz constant 
$L\,\|\vect(\bK_0)\| < \infty$.
Hence $\varphi$ is a bounded $L$-Lipschitz function.
So, by using \eqref{eq:BL-for-pi}, we have 
\begin{equation*}
\big|
  \E^*[\psi(\sqrt{n}\bK_0(\hpi^*_n(\boheta_n)-\hpi_n(\boheta_n)))]
  - \E[\psi(\tilde Z_\pi)]
\big|
=
\big|
  \E^*[\varphi(\sqrt{n}\big(\hpi^*_n(\boheta_n)-\hpi_n(\boheta_n)\big))]
  - \E[\varphi(Z_\pi)]
\big|
\;\rightarrow_p\;0.
\end{equation*}
Since this holds for every bounded Lipschitz $\psi$, we conclude that
\begin{equation}\label{eq:D0convergence}
\sqrt{n}\bK_0\big(\hpi^*_n(\boheta_n)-\hpi_n(\boheta_n)\big)
\rightarrow_d \tilde Z_\pi,
\end{equation}
in probability.

Define
\begin{equation*}
U_n^*
:=
\sqrt{n}\bK_0\big(\hpi^*_n(\boheta_n)-\hpi_n(\boheta_n)\big),
\qquad
V_n^* :=
\sqrt{n}\,(\btr_n^*-\btr_n).
\end{equation*}
We now prove that $U_n^*+V_n^* \rightarrow_d \tilde Z_\pi$ in probability.
Let $\psi:\mathbb{R}^{d+d(d+1)/2}\rightarrow\mathbb{R}$ be bounded and $L$-Lipschitz, with
$|\psi(\bx)|\leqslant M$, $M<\infty$.
Then
\begin{align}\label{eq:boundUV}
\big|
  \E^*[\psi(U_n^*+V_n^*)] - \E[\psi(\tilde Z_\pi)]
\big|
&\leqslant
\big|
  \E^*[\psi(U_n^*+V_n^*)] - \E^*[\psi(U_n^*)]
\big|
\nonumber\\
&\quad+
\big|
  \E^*[\psi(U_n^*)] - \E[\psi(\tilde Z_\pi)]
\big|.
\end{align}
The second term converges to $0$ in probability by the convergence of
$U_n^*$ to $\tilde Z_\pi$ in \eqref{eq:D0convergence}. For the first term, using the Jensen inequalities, we have
\begin{equation*}
\begin{aligned}
\big|\E^*\!\left[\psi(U_n^*+V_n^*)\right]
      - \E^*\!\left[\psi(U_n^*)\right]\big|
&
= \big|\E^*\!\left[\psi(U_n^*+V_n^*) - \psi(U_n^*)\right]\big|
\leqslant \E^*\!\left[\big|\psi(U_n^*+V_n^*) - \psi(U_n^*)\big|\right].
\end{aligned}
\end{equation*}
Now fix $\delta>0$ and split the expectation according to the law of total expectation applied to the partition 
$\{\|V_n^*\|\leqslant\delta\}$ and $\{\|V_n^*\|>\delta\}$,
\begin{align*}
\E^*\!\left[\big|\psi(U_n^*+V_n^*)-\psi(U_n^*)\big|\right]
&=
\E^*\!\left[
  \big|\psi(U_n^*+V_n^*)-\psi(U_n^*)\big|
  \,\big|\, \|V_n^*\|\leqslant\delta
\right]{\Pr}^*(\|V_n^*\|\leqslant\delta)
\\[4pt]
&\quad+
\E^*\!\left[
  \big|\psi(U_n^*+V_n^*)-\psi(U_n^*)\big|
  \,\big|\, \|V_n^*\|>\delta
\right]{\Pr}^*(\|V_n^*\|>\delta).
\end{align*}
When $\{\|V_n^*\|\leqslant\delta\}$, by the Lipschitz property of $\psi$,
\begin{equation*}
|\psi(U_n^*+V_n^*)-\psi(U_n^*)|
\leqslant L\,\|V_n^*\|
\leqslant L\delta.
\end{equation*}
Hence
\begin{equation*}
\E^*\!\left[
  \big|\psi(U_n^*+V_n^*)-\psi(U_n^*)\big|
  \,\big|\, \|V_n^*\|\leqslant\delta
\right]{\Pr}^*(\|V_n^*\|\leqslant\delta)
\leqslant L\delta\,{\Pr}^*(\|V_n^*\|\leqslant\delta)
\leqslant L\delta,
\end{equation*}
since ${\Pr}^*(\|V_n^*\|\leqslant\delta)\leqslant 1$.
When $\{\|V_n^*\|>\delta\}$, by using the boundedness of $\psi$,
\begin{equation*}
|\psi(U_n^*+V_n^*)-\psi(U_n^*)|
\leqslant |\psi(U_n^*+V_n^*)| + |\psi(U_n^*)|
\leqslant 2M.
\end{equation*}
Thus
\begin{equation*}
\E^*\!\left[
  \big|\psi(U_n^*+V_n^*)-\psi(U_n^*)\big|
  \,\big|\, \|V_n^*\|>\delta
\right]{\Pr}^*(\|V_n^*\|>\delta)
\leqslant
2M\,
{\Pr}^*\!\left(\|V_n^*\|>\delta\right).
\end{equation*}
Combining the bounds above, we have
\begin{equation*}
\big|\E^*\!\left[\psi(U_n^*+V_n^*)\right]
      - \E^*\!\left[\psi(U_n^*)\right]\big|
\leqslant
L\delta
+
2M\,{\Pr}^*\!\left(\|V_n^*\|>\delta\right).
\end{equation*}
By \eqref{eq:diffreminder}, for any 
$\eps>0$
\begin{equation*}
  \Pr\!\left( \, {\Pr}^*\!\bigl(\|V_n^*\|>\delta\bigr) > \eps \, \right) \rightarrow 0 .
\end{equation*} 
 The event
\begin{equation*}
\Big\{
  \big|
    \E^*[\psi(U_n^*+V_n^*)] - \E^*[\psi(U_n^*)]
  \big|
  > L\delta + 2M\,\eps
\Big\},
\end{equation*}
can only occur if 
\begin{equation*}
\Big\{
  L\delta
+
2M\,{\Pr}^*\!\left(\|V_n^*\|>\delta\right)
  > L\delta + 2M\,\eps
\Big\},
\end{equation*}
which simplifies to the event $
\{
{\Pr}^*\!\left(\|V_n^*\|>\delta\right)
  > \eps
\}
$.
Thus,
\begin{equation*}
\Pr\Big(
  \big|
    \E^*[\psi(U_n^*+V_n^*)] - \E^*[\psi(U_n^*)]
  \big|
  > L\delta + 2M\,\eps
\Big)
\;\leqslant\;
\Pr\Big(
  {\Pr}^*(\|V_n^*\|>\delta) > \eps
\Big)\rightarrow 0.
\end{equation*}
Since $\delta, \eps>0$ are arbitrary, we obtain
\begin{equation}\label{eq:convergence_expextation_kros}
\big|
  \E^*[\psi(U_n^*+V_n^*)] - \E^*[\psi(U_n^*)]
\big|
\;\rightarrow_p\; 0.
\end{equation}
This implies from \eqref{eq:boundUV} that
\begin{equation}\label{eq:distfinalboot}
\sqrt{n}(\bhtheta_n^*-\bhtheta_n)
= U_n^*+V_n^*
\rightarrow_d \tilde Z_\pi,
\end{equation}
in probability.

Let us define for $\bx\in\mathbb{R}^{d+d(d+1)/2}$
\begin{equation*}
F_n(\bx)
:=
\Pr\bigl(\sqrt{n}\,(\bhtheta_n-\btheta_0)\leqslant \bx\bigr),\qquad
G_n(\bx)
:=
{\Pr}^*\bigl(\sqrt{n}\,(\bhtheta_n^*-\bhtheta_n)\leqslant \bx\bigr).
\end{equation*}
Let us define
$\Phi$ as the cdf of $  \tilde Z_\pi$ that by construction is everywhere continuous on
$\mathbb R^{d+d(d+1)/2}$.  Hence the set $\mathcal S$ of continuity points of
$\Phi$ coincides with $\mathbb R^{d+d(d+1)/2}$.  
By using \eqref{eq:convergence_expextation_kros} 
to apply Lemma~10.11(i) of \citepsupp{kosorok2008introduction} where
$\mathcal Y_n$ denotes the observed sample, $X_n=\sqrt{n}\,(\bhtheta_n^*-\bhtheta_n)$, $X=\tilde Z_\pi$ and with the
closed set $A=\mathcal S=\mathbb R^{d+d(d+1)/2}$, we obtain
\begin{equation*}
\sup_{\bx\in\mathbb{R}^{d+d(d+1)/2}}
\bigl|G_n(\bx)-\Phi(\bx)\bigr|
\ \rightarrow_p\ 0.
\end{equation*}
In addition from Lemma~2.11 of \citet{van2000asymptotic}, \eqref{eq:theta-limit} implies
\begin{equation*}
\sup_{\bx\in\mathbb{R}^{d+d(d+1)/2}}
\bigl|F_n(\bx)-\Phi(\bx)\bigr|
\ \rightarrow\ 0,
\end{equation*}
Thus  
\begin{equation*}
\sup_{\bx\in\mathbb{R}^{d+d(d+1)/2}} |F_n(\bx)-G_n(\bx)| \leqslant \sup_{\bx\in\mathbb{R}^{d+d(d+1)/2}}
\bigl|F_n(\bx)-\Phi(\bx)\bigr|+\sup_{\bx\in\mathbb{R}^{d+d(d+1)/2}}
\bigl|G_n(\bx)-\Phi(\bx)\bigr|.
\end{equation*}
and  
\begin{equation*}
\sup_{\bx\in\mathbb{R}^{d+d(d+1)/2}} |F_n(\bx)-G_n(\bx)| \rightarrow_p 0.
\end{equation*}
\end{proof}

In the sequel we will use the following seven
conditions.
\begin{enumerate}
\item [(D1)] 
Let $\{F_{\btheta} : \btheta \in \bTheta\}$ 
be a $d$-dimensional
parametric family. That is, for each
$\btheta = (\bmu^T,\vecth_s(\bSigma)^T)^T \in \bTheta$,
a random sample
$\bx_{1,\btheta},\ldots,\bx_{n,\btheta}$ from $F_{\btheta}$
admits the representation
$\bx_{i,\btheta}
=
G(\btheta,\bu_i)$, $i=1,\ldots,n$,
where $\bu_1,\ldots,\bu_n$ are i.i.d.\ random vectors with
known common distribution $P_U$ that does not depend on $\btheta$,
and $\E[\|\bu_1\|^2]<\infty$.
Assume moreover, that there exists a measurable function
$m:\mathcal U\rightarrow[0,\infty)$ with $\E[m(\bu_1)^2]<\infty$
such that for all $\btheta_1,\btheta_2\in\bTheta$ and all
$\bu\in\mathcal U$,
\[
\|G(\btheta_1,\bu)-G(\btheta_2,\bu)\|
\leqslant
m(\bu)\,\|\btheta_1-\btheta_2\|.
\]
The observed sample $\bx_1,\ldots,\bx_n$ consists of i.i.d.\
observations drawn from $F_{\btheta_0}$, where
$\btheta_0\in\interior(\bTheta)$.
\item[(D2)]
Let $\bH\subseteq \mathbb R^{r}$ be compact, and let the tuning parameters $\boeta\in\bH$. Moreover, the tuning parameter estimator $\boheta_n$, computed from the observed sample,
satisfies
$\boheta_n-\boeta_0 = O_p(n^{-1/2})$.
\item[(D3)] 
There exists a constant $M<\infty$ such that for all $\bx\in\mathbb R^d$ and all
$\boeta\in\bH$,
$
\|a_F(\bx,\boeta)\|\leqslant M$, and
$\|b_F(\bx,\boeta)\|\leqslant M$.
\item[(D4)]
The function $a$ is Lipschitz continuous, that is, there exists an $L_a<\infty$ such that
for any $\bx_1,\bx_2\in\mathbb R^d$ and any $\boeta_1,\boeta_2\in\bH$,
$
\|a_F(\bx_1,\boeta_1)-a_F(\bx_2,\boeta_2)\|
\leqslant
L_a\bigl(\|\bx_1-\bx_2\|+\|\boeta_1-\boeta_2\|\bigr).
$
An analogous Lipschitz condition holds for $b_F(\bx,\boeta)$.
\item[(D5)]
There exists a constant $c>0$ such that for all $(\btheta,\boeta)\in\bTheta\times\bH$ , $\min_{ j}|(\mu_b)_j(\btheta,\boeta)| > c$.
Define
$E_n
:=
\Big\{
\inf_{(\btheta,\boeta)\in\bTheta\times\bH}
\ \min_{ j}
\big|(\bar b_n)_j(\btheta,\boeta)\big|
\geqslant c
\Big\}$, and assume that there exists 
\mbox{$\alpha>1/2$} such that 
$\Pr(E_n^c)=O(n^{-\alpha})$ as 
$n\rightarrow\infty$, and 
$\tilde \delta\leqslant c$.
\item[(D6)]  
The number $H$ of simulated datasets is selected such that  $\sqrt{\frac{\log H}{H}} = o(n^{-1/2})$  as $n\rightarrow\infty$.
\item[(D7)]
 For any $\eps>0$,
\[
\inf_{\btheta\in\bTheta:\ \|\btheta-\btheta_0\|\geqslant\eps}
\|P_F(\btheta,\boeta_0)-P_F(\btheta_0,\boeta_0)\|>0.
\]
Moreover, $P_F$ is  differentiable in $\btheta$ 
and $\boeta$ at $(\btheta_0,\boeta_0)$ with 
full-rank Jacobian\linebreak
$\bA_0
=
\left.
\frac{\partial\, P_F(\btheta,\boeta)}{\partial \btheta}
\right|_{(\btheta,\boeta)=(\btheta_0,\boeta_0)}$, with $\bK_0:=\bA_0^{-1}$.
\end{enumerate} 

\vspace{2mm}
Assumption~(D1) requires independent 
observations from a correctly specified parametric
family and a simulation mechanism whose output 
varies regularly with the parameter.
The Lipschitz condition ensures that small
changes in the parameter produce controlled 
changes in the simulated observations, and the 
interior-point condition excludes boundary cases.
Distributional assumptions of this kind are 
familiar in robust estimation of location and
scatter. For example,
\citetsupp{lopuhaa1989relation} studies 
S-estimators under elliptical distributions 
while \citetsupp{butler1993asymptotics}
establishes consistency of the MCD estimator 
under unimodal elliptical distributions. 
Assumption (D1) is 
readily satisfied for location-scatter models
$G(\btheta,\bu)=\bmu+\bSigma^{1/2}\bu$, where 
$\bu$ has a fixed, known distribution with a 
finite second moment, provided the eigenvalues 
of $\bSigma$ are bounded away from zero, as 
imposed by our parameter space. Examples include 
the Gaussian model and  multivariate $t$ 
distributions with fixed degrees of 
freedom $\nu>2$.

Assumption~(D2) requires the preliminary 
quantities determining the FastCellCov weights 
to lie in a compact set and converge at the 
root-$n$ rate. These quantities are constructed 
from DDC, robust M-scales, cellPCA, and MCD 
estimates computed from the original data. 
DDC combines basic robust univariate and 
bivariate estimation steps, while consistency
and root-$n$ asymptotic normality have been 
established for the cellPCA projection matrix 
under suitable assumptions \citep{cellPCA}.
Together with the established asymptotic theory
for M-scales and MCD estimators, these results
support the plausibility of (D2).

Assumptions~(D3)-(D5) impose boundedness,
regularity, and
nondegeneracy on the weighted contributions defining
FastCellCov. Boundedness can be achieved through
redescending weights that sufficiently control both
observations and their cross-products. The Lipschitz
condition requires these contributions to vary regularly
with the observations and preliminary estimates.
It is similar in spirit to the smoothness conditions
used in the fast and robust bootstrap
\citep{salibian2008fast}.
Finally, (D5) requires sufficient total weight for each
location and scatter component, uniformly over the
parameter space. Its probability bound ensures that
near-zero denominators become sufficiently rare and
that denominator clipping is asymptotically inactive.

Assumption~(D6) requires the number $H$ of simulated
datasets in indirect inference to grow sufficiently
fast relative to $n$, so that simulation error is
negligible at the root-$n$ scale. Equivalently, $H/(n\log H)\to\infty$, which is a
sufficient theoretical requirement that increases
the simulation size as $n$ grows.

Assumption~(D7) requires the population 
FastCellCov target to retain enough information
to identify $\btheta_0$. The separation condition
ensures that parameters away from $\btheta_0$ 
cannot reproduce the same auxiliary target, which 
is needed for consistency of the indirect
inference estimator. Differentiability and the
full-rank Jacobian permit a local linear inversion,
underlying its root-$n$ asymptotic distribution 
and bootstrap validity. These are standard
identification and regularity conditions in
indirect inference  
and estimating equation theory. For example,
\citetsupp{guerrier2018properties} assumes a 
continuous, injective binding function to 
establish consistency. The separation and 
nonsingular derivative conditions also appear in
Theorems~5.9 and~5.21 of
\citet{van2000asymptotic}.

\newcounter{oldtheo2}
   \setcounter{oldtheo2}{\value{theorem}}
\setcounter{theorem}{0}

\vspace{2mm}
\begin{theorem} [Restated] \label{app:theoCellBoot}
Under conditions D1--D7 the bootstrap 
distribution of
$\sqrt{n}(\bhtheta_n^*-\bhtheta_n)$ is 
consistent for the distribution of
$\sqrt{n}(\bhtheta_n-\btheta_0)$ in the 
Kolmogorov-Smirnov distance, i.e.,
\begin{equation*}
\sup_{\bx\in\mathbb{R}^{d+d(d+1)/2}}
\Big|
  \Pr\bigl(
    \sqrt{n}(\bhtheta_n-\btheta_0)\leqslant \bx
  \bigr)
  -
  {\Pr}^* \bigl(
    \sqrt{n}(\bhtheta_n^*-\bhtheta_n)\leqslant \bx \,
  \bigr)
\Big|
\rightarrow_p0,
\end{equation*}
where ${\Pr}^*$ denotes probability computed under the bootstrap distribution, conditional on the observed data.
\end{theorem}

\setcounter{theorem}{\value{oldtheo2}}

\begin{proof}
Note that $\bTheta$ is compact as it is shown in the proof of Lemma~\ref{prop:closed}.
Note that, we have the identity
$\hpi(\btheta,\boheta_n,n) = \hP_F(\btheta,\boheta_n,n)$, for $\btheta\in\bTheta$,
and $\hP_F$  evaluated using $\boheta_n$ on the observed sample concides with $\hpi_n(\boheta_n)$.
Moreover,  $\pi(\btheta,\boeta)=P_F(\btheta,\boeta)$, for $(\btheta,\boeta)\in\bTheta\times\bH$.

The proof consists in verifying the assumptions of
Proposition~\ref{pro:bootstrap-theta}.  
Assumption (C1) is satisfied.
To verify (C4), we apply Proposition~\ref{pro:adphat} with $a=a_F$ and
$b=b_F$, and use the identities $\hP_n(\boheta_n)=\hpi_n(\boheta_n)$ and
$P(\boheta_n)=\pi(\btheta_0,\boheta_n)$. This gives
\begin{equation}\label{eq:asynormalest}
  \sqrt{n}\bigl(
    \hpi_n(\boheta_n) - \pi(\btheta_0,\boheta_n)
  \bigr)
  \rightarrow_d Z_\pi,
\end{equation}
where $Z_\pi \sim N(\bzero,\bD_0\,\bSigma_{ab}\,\bD_0^T)$.

Moreover, Proposition~\ref{pro:converbootmean} with $a=a_F$ and $b=b_F$,
together with the identity $\hP_n^*(\boheta_n)=\hpi_n^*(\boheta_n)$, yields
\begin{equation}\label{eq:asynormalestboot}
  \sqrt{n}\bigl(
    \hpi_n^*(\boheta_n) - \hpi_n(\boheta_n)
  \bigr)
  \rightarrow_d Z_\pi,
\end{equation}
in probability.
Using Lemma~\ref{lem:ratio-uniform-eta} with $a=a_F$ and $b=b_F$, and noting that
$\hP(\btheta,\boheta_n,n)=\hpi(\btheta,\boheta_n,n) $, 
$P(\btheta,\boeta_0)=\pi(\btheta,\boeta_0)$, and $\bar P(\btheta,\boheta_n,n)=\bar\pi(\btheta,\boheta_n,n)$, we obtain
\begin{equation}\label{eq:pihatuniform}
\sup_{\btheta\in\bTheta}
\big\|
\hpi(\btheta,\boheta_n,n) -\pi(\btheta,\boeta_0)
\big\|
= O_p(n^{-1/2}),
\end{equation}
and
\begin{equation}\label{eq:pihatuniform2}
\sup_{\btheta\in\bTheta}
\big\|
\bar\pi(\btheta,\boheta_n,n) -\pi(\btheta,\boheta_n)
\big\|
= o_p(n^{-1/2}).
\end{equation}

Fix $(\btheta_1,\boeta_1),(\btheta_2,\boeta_2)\in\bTheta\times\bH$. Using the Lipschitz property in (D4) and Jensen's inequality,
\begin{align*}
\big\|\mu_{a,F}(\btheta_1,\boeta_1)-\mu_{a,F}(\btheta_2,\boeta_2)\big\|
&=
\Big\|\E\!\big[a_F(\bx_{\btheta_1},\boeta_1)-a_F(\bx_{\btheta_2},\boeta_2)\big]\Big\|\\
&\leqslant
\E\big\|a_F(\bx_{\btheta_1},\boeta_1)-a_F(\bx_{\btheta_2},\boeta_2)\big\|\\
&\leqslant
L_a\Big(\E\|\bx_{\btheta_1}-\bx_{\btheta_2}\|+\|\boeta_1-\boeta_2\|\Big).
\end{align*}
Under (D1), $\bx_{\btheta}=G(\btheta,\bu)$, and, for $\bu\in\mathcal U$,
\[
\|\bx_{\btheta_1}-\bx_{\btheta_2}\|
=\|G(\btheta_1,\bu)-G(\btheta_2,\bu)\|
\leqslant
m(\bu)\,\|\btheta_1-\btheta_2\|.
\]
Taking expectations and using Cauchy-Schwarz, for $U\sim P_U$, 
\[
\E\|\bx_{\btheta_1}-\bx_{\btheta_2}\|
\leqslant
\E[m(U)]\,\|\btheta_1-\btheta_2\|
\leqslant
\sqrt{\E[m(U)^2]}\,\|\btheta_1-\btheta_2\|.
\]
Since $\E[m(U)^2]<\infty$,  it follows that
$\E\|\bx_{\btheta_1}-\bx_{\btheta_2}\|\rightarrow 0$ whenever $\btheta_1\rightarrow\btheta_2$.
Hence $\mu_{a,F}(\btheta,\boeta)$ is continuous on $\bTheta\times\bH$. The same argument  yields continuity of $\mu_{b,F}(\btheta,\boeta)$ on $\bTheta\times\bH$.
Finally, by (D5) we have $\min_j |(\mu_{b,F})_j(\btheta,\boeta)|> c>0$ for all $(\btheta,\boeta)\in\bTheta\times\bH$.
Componentwise division by a function bounded away from zero is continuous; therefore
$P_F(\btheta,\boeta)=\mu_{a,F}(\btheta,\boeta)\oslash \mu_{b,F}(\btheta,\boeta)$ is continuous on $\bTheta\times\bH$.
Since $\pi(\btheta,\boeta)=P_F(\btheta,\boeta)$, this proves that $\pi(\btheta,\boeta)$ is continuous on $\bTheta\times\bH$.
Hence, the conditions of Proposition~\ref{pro:II_WI} are satisfied,
and we conclude that
\begin{equation*}
\bhtheta_n - \btheta_0
= \bK_0\big(\hpi_n(\boheta_n) - \pi(\btheta_0,\boheta_n)\big)
  + \btr_n,
\qquad
\btr_n = o_p(n^{-1/2}).
\end{equation*}
This establishes assumption (C2).

Note that \eqref{eq:asynormalestboot}, combined with the same argument used in the
proof of Proposition~\ref{pro:converbootmean} to obtain
\eqref{eq:boundsqrtnZ}, implies that for each $\eps>0$ there exists
an $L>0$ such that
\begin{equation}
\Pr\Big( {\Pr}^*\big( \sqrt{n}\,\|\hpi_n^*(\boheta_n) - \hpi_n(\boheta_n)\| > L \big)
        > \eps \Big)
\rightarrow 0.
\end{equation}
Together with \eqref{eq:pihatuniform}, \eqref{eq:pihatuniform2}, the
conditions of Proposition~\ref{pro:II_WI_boot} are
satisfied. 
Hence,
\begin{equation*}
\bhtheta_n^* - \btheta_0
= \bK_0\big(\hpi_n^*(\boheta_n) - \pi(\btheta_0,\boheta_n)\big)
  + \btr_n^*,
\end{equation*}
where ${\Pr}^*(\sqrt{n}\,\|\btr_n^*\|>\eps)\rightarrow_p 0$.
This guarantees assumption (C3).
\end{proof}

For cellMR regression we now add the condition\\
\noindent (D8) $\lambda_n=o(n^{-1/2})$.\\
This condition is quite natural because we will apply
Corollary~\ref{prop:bootstrap-theta-linear-reg}
to the II estimator, that by its definition
satisfies the constraint
$c \leqslant \lambda_{\min}(\bhSigma)\leqslant 
\lambda_{\max}(\bhSigma)\leqslant C$ from which
the same property follows for its submatrix
$\bhSigma_x$\,, that is,
$c \leqslant \lambda_{\min}(\bhSigma_x)\leqslant 
\lambda_{\max}(\bhSigma_x)\leqslant C$.
So if the imposed $c$ and $C$ are chosen such
that $C/c$ is small enough, we could even put 
$\lambda_n = 0$.

\begin{proof}[\textbf{Proof of Corollary~\ref{prop:bootstrap-theta-linear-reg}}]
For $(\bmu^T,\vecth_s(\bSigma)^T)^T\in\bTheta$, define the  selection matrices $\bS_x\in\mathbb R^{p\times d}$ and
$\bS_y\in\mathbb R^{q\times d}$ such that
$\bmu_x=\bS_x\bmu$, $\bmu_y=\bS_y\bmu$ and
$\bSigma_x=\bS_x\bSigma \bS_x^T$, $\bSigma_{xy}=\bS_x\bSigma \bS_y^T$.
Define for $\lambda\geqslant0$,
\begin{equation*}
\bB(\bmu,\bSigma,\lambda):=\bigl(\bS_x\bSigma \bS_x^T+\lambda\bI_p\bigr)^{-1}(\bS_x\bSigma \bS_y^T),
\qquad
\bb(\bmu,\bSigma,\lambda):=\bS_y\bmu-\bB(\bmu,\bSigma,\lambda)^T(\bS_x\bmu),
\end{equation*}
and set
\begin{equation*}
\psi(\bmu,\bSigma,\lambda)
:=
\bigl((\bb(\bmu,\bSigma,\lambda))^T,\ \vect(\bB(\bmu,\bSigma,\lambda))^T\bigr)^T.
\end{equation*}
 Let $\varphi(\bmu,\bSigma,\lambda):=\ba^T\psi(\bmu,\bSigma,\lambda)$, then by construction, $
\htheta_n=\varphi(\bhtheta_n,\lambda_n)$, $
\htheta_n^*=\varphi(\bhtheta_n^*,\lambda_n)$,
 and $\theta_0=\varphi(\btheta_0,0)$.

Since $(\bmu^T,\vecth_s(\bSigma)^T)^T\in\bTheta$, we have
$\lambda_{\min}(\bSigma)\geqslant c>0$. As $\bS_x$ selects the
$x$-coordinates, $\bSigma_x=\bS_x\bSigma\bS_x^T$ is a principal submatrix
of $\bSigma$. Hence, by the Cauchy interlacing theorem,
$\lambda_{\min}(\bSigma_x)\geqslant \lambda_{\min}(\bSigma)\geqslant c$.
Therefore, for every $\lambda\geqslant 0$,
$
\lambda_{\min}(\bSigma_x+\lambda\bI_p)
\geqslant c+\lambda
\geqslant c$,
so that $\bSigma_x+\lambda\bI_p$ is symmetric positive definite and
invertible. The map
$
(A,C,\lambda)\mapsto (A+\lambda\bI_p)^{-1}C
$
is continuously differentiable on the open set where $A+\lambda\bI_p$ is
invertible, because it is the composition of matrix inversion and matrix
multiplication. Moreover,
$
(\bmu,\bSigma)\mapsto
\bigl(\bS_x\bSigma\bS_x^T,\ \bS_x\bSigma\bS_y^T\bigr)
$
is linear in $\bSigma$ and hence continuously differentiable. It follows by
the chain rule that
$
(\bmu,\bSigma,\lambda)\mapsto \bB(\bmu,\bSigma,\lambda)
$
is continuously differentiable in a neighbourhood of $(\bmu_0,\bSigma_0,0)$.
An analogous argument applies to $\bb(\bmu,\bSigma,\lambda)$, and therefore
$\psi(\bmu,\bSigma,\lambda)$ is continuously differentiable at
$(\bmu_0,\bSigma_0,0)$. Hence
$\varphi(\btheta,\lambda)=\ba^T\psi(\btheta,\lambda)$ is differentiable at
$(\btheta_0,0)$.

Let
\[
\bJ_0
:=
\left.
\frac{\partial\,\varphi(\btheta,\lambda)}{\partial\btheta}
\right|_{(\btheta,\lambda)=(\btheta_0,0)},
\qquad
\dot\varphi_{\lambda,0}
:=
\left.
\frac{\partial\,\varphi(\btheta,\lambda)}{\partial\lambda}
\right|_{(\btheta,\lambda)=(\btheta_0,0)}.
\]
By a first-order Taylor expansion around $(\btheta_0,0)$,
\begin{align*}
\htheta_n-\theta_0
&=
\varphi(\bhtheta_n,\lambda_n)-\varphi(\btheta_0,0)  \\
&=
\bJ_0^T(\bhtheta_n-\btheta_0)
+
\dot\varphi_{\lambda,0}\lambda_n
+
\br_n,
\\[0.4em]
\htheta_n^*-\theta_0
&=
\varphi(\bhtheta_n^*,\lambda_n)-\varphi(\btheta_0,0)  \\
&=
\bJ_0^T(\bhtheta_n^*-\btheta_0)
+
\dot\varphi_{\lambda,0}\lambda_n
+
\br_n^*,
\end{align*}
where $\br_n
=
o_p\bigl(\|\bhtheta_n-\btheta_0\|+\lambda_n\bigr)$.
Since $\bhtheta_n\rightarrow_p\btheta_0$ and
$\lambda_n=o(n^{-1/2})$, in particular $\lambda_n\rightarrow 0$, this
remainder bound follows by the same arguments used to obtain
\eqref{eq:reminaderop_joint_sum} in the proof of
Proposition~\ref{pro:II_WI}. Moreover, since
$\|\bhtheta_n-\btheta_0\|=O_p(n^{-1/2})$ from \eqref{eq:ratehtheta} and
$\lambda_n=o(n^{-1/2})$, we have $\br_n=o_p(n^{-1/2})$.
By using the fact that, for any
$\eps>0$,
\[
{\Pr}^*(\|\bhtheta_n^*-\btheta_0\|>\eps)\rightarrow_p 0,
\]
and the same arguments used to obtain \eqref{eq:bootreminderzero2} in the
proof of Proposition~\ref{pro:II_WI_boot}, for every $\eps,\delta>0$,
\[
\Pr\Big(
  {\Pr}^*\bigl(
    \sqrt{n}\,\|\br_n^*\|>\eps
  \bigr)>\delta
\Big)
\rightarrow 0,
\]
as $\lambda_n=o(n^{-1/2})$.

Subtracting the two expansions yields
\begin{align*}
\htheta_n^*-\htheta_n
&=
\bJ_0^T(\bhtheta_n^*-\bhtheta_n)
+
\dot\varphi_{\lambda,0}\lambda_n
-
\dot\varphi_{\lambda,0}\lambda_n
+
(\br_n^*-\br_n) \\
&=
\bJ_0^T(\bhtheta_n^*-\bhtheta_n)+(\br_n^*-\br_n),
\end{align*}
where the ridge terms cancel because the same tuning parameter $\lambda_n$ is
used in the original and bootstrap estimators. Hence
\begin{equation*}
\sqrt{n}\bigl(\htheta_n^*-\htheta_n\bigr)
=
\bJ_0^T\sqrt{n}(\bhtheta_n^*-\bhtheta_n)
+
\sqrt{n}(\br_n^*-\br_n).
\end{equation*}
Using the same arguments leading to \eqref{eq:diffreminder} in the proof of
Proposition~\ref{pro:bootstrap-theta}, for every $\eps,\eta>0$,
\begin{equation}\label{eq:diffreminder2}
\Pr\Big(
{\Pr}^*\bigl(\sqrt{n}\,\|\br_n^*-\br_n\|>\eps\bigr)>\eta
\Big)\rightarrow 0.
\end{equation}

From the proof of Proposition~\ref{pro:bootstrap-theta},
\begin{equation*}
\sqrt{n}\,(\bhtheta_n-\btheta_0)\rightarrow_d \tilde Z_\pi,
\end{equation*}
where $\tilde Z_\pi
\sim N\bigl(\bzero,\bK_0\bD_0\,\bSigma_{ab}\,\bD_0^T\bK_0^T\bigr)$.
By the continuous mapping theorem \citep{van2000asymptotic},
\begin{equation*}
\bJ_0^T\sqrt{n}(\bhtheta_n-\btheta_0)
\rightarrow_d
\bJ_0^T\tilde Z_\pi.
\end{equation*}
Moreover, since $\lambda_n=o(n^{-1/2})$,
\[
\sqrt n\,\lambda_n\,\dot\varphi_{\lambda,0}=o(1).
\]
Therefore, by Slutsky's lemma \citep{van2000asymptotic},
\begin{align*}
\sqrt{n}(\htheta_n-\theta_0)
&=
\bJ_0^T\sqrt{n}(\bhtheta_n-\btheta_0)
+
\sqrt n\,\lambda_n\,\dot\varphi_{\lambda,0}
+
\sqrt n\,\br_n \\
&\rightarrow_d
\bJ_0^T\tilde Z_\pi .
\end{align*}
Moreover, again from Proposition~\ref{pro:bootstrap-theta},
\begin{equation*}
\sqrt{n}(\bhtheta_n^*-\bhtheta_n)
\rightarrow_d \tilde Z_\pi,
\end{equation*}
in probability.
Applying the same arguments used to obtain \eqref{eq:D0convergence}, we have
\begin{equation*}
\bJ_0^T\sqrt{n}(\bhtheta_n^*-\bhtheta_n)
\;\rightarrow_d\;
\bJ_0^T\tilde Z_\pi,
\end{equation*}
in probability. Together with \eqref{eq:diffreminder2}, this implies that
\begin{equation*}
\sqrt{n}(\htheta_n^*-\htheta_n)
\;\rightarrow_d\;
\bJ_0^T\tilde Z_\pi,
\end{equation*}
in probability, by using the same arguments to obtain
\eqref{eq:distfinalboot}.

Using the same argument based on Lemma~10.11(i) 
of \citepsupp{kosorok2008introduction} as in the 
final step of the proof of
Proposition~\ref{pro:bootstrap-theta}, we 
conclude that
\begin{equation*}
\sup_{x\in\mathbb R}
\Big|
  \Pr\bigl(
    \sqrt{n}(\htheta_n-\theta_0)\leqslant x
  \bigr)
  -
  {\Pr}^*\bigl(
    \sqrt{n}(\htheta_n^*-\htheta_n)\leqslant x
  \bigr)
\Big|
\rightarrow_p 0.
\end{equation*}
\end{proof}

\begin{proof}[\textbf{Proof of Corollary~\ref{cor:efron-percentile-exact}}]

Let $F_n^*$ denote the conditional distribution function of
$\sqrt n(\htheta_n^{\,*}-\htheta_n)$, i.e., $
F_n^*(x):={\Pr}^*\!\big(\sqrt n(\htheta_n^{\,*}-\htheta_n)\leqslant x\big)$.
From the last part of the proof of Corollary~\ref{prop:bootstrap-theta-linear-reg}, we know that 
$\sup_{x\in\mathbb R}|F_n^*(x)-\Phi_{\bJ}(x)|\rightarrow_p 0$, where $\Phi_{\bJ}$ is the cumulative distribution of normal distribution $N(0,\sigma_{\theta}^2)$, with $\sigma_{\theta}^2=\bJ_0^T\bK_0\bD_0\,\bSigma_{ab}\,\bD_0^T\bK_0^T\bJ_0$ . We also know that $\sqrt n(\htheta_n-\theta_0)\rightarrow_d Z_{\theta}$, where $Z_{\theta}\sim N(0,\sigma_{\theta}^2)$.
 Define $\hq_n^*(\gamma)$ denote the $\gamma$-quantile of the conditional distribution of $\sqrt n(\htheta_n^{\,*}-\htheta_n)$, $
\hq_n^*(\gamma)
:=
\inf\{x\in\mathbb R:\;F_n^*(x)\geqslant \gamma\}$,
for $\gamma\in(0,1)$, and  let $q_\gamma:=\Phi_{\bJ}^{-1}(\gamma)=
\inf\{x\in\mathbb R:\,\Phi_{\bJ}(x)\geqslant \gamma\}$.

Fix $\gamma\in(0,1)$ and let
$0<\eps<\min\{\gamma,1-\gamma\}$.
Suppose that
\[
\sup_{x\in\mathbb R}
\bigl|F_n^*(x)-\Phi_{\bJ}(x)\bigr|
< \eps,
\]
that is,
\begin{equation}\label{eq:uniform-bound}
\Phi_{\bJ}(x)-\eps
<
F_n^*(x)
<
\Phi_{\bJ}(x)+\eps
\qquad
\text{for all }x\in\mathbb R.
\end{equation}
Let $x<\Phi_{\bJ}^{-1}(\gamma-\eps)$.
By definition of the generalized inverse,
$\Phi_{\bJ}(x)<\gamma-\eps$.
Using \eqref{eq:uniform-bound},
\[
F_n^*(x)
<
\Phi_{\bJ}(x)+\eps
<
(\gamma-\eps)+\eps
=
\gamma.
\]
Hence no such $x$ can satisfy $F_n^*(x)\geqslant\gamma$, and therefore
\[
\hq_n^*(\gamma)
\geqslant
\Phi_{\bJ}^{-1}(\gamma-\eps).
\]
Let $x_1:=\Phi_{\bJ}^{-1}(\gamma+\eps)$.
Then
\[
\Phi_{\bJ}(x_1)\geqslant\gamma+\eps.
\]
Again using \eqref{eq:uniform-bound},
\[
F_n^*(x_1)
>
\Phi_{\bJ}(x_1)-\eps
\geqslant
(\gamma+\eps)-\eps
=
\gamma.
\]
Thus $x_1\in\{x:\,F_n^*(x)\geqslant\gamma\}$, and by definition of the infimum,
\[
\hq_n^*(\gamma)
\leqslant
x_1
=
\Phi_{\bJ}^{-1}(\gamma+\eps).
\]
Combining the two inequalities yields
\begin{equation}\label{eq:ineqq}
\Phi_{\bJ}^{-1}(\gamma-\eps)
\;\leqslant\;
\hq_n^*(\gamma)
\;\leqslant\;
\Phi_{\bJ}^{-1}(\gamma+\eps).
\end{equation}

Since $\Phi_{\bJ}^{-1}$ is continuous at $\gamma\in(0,1)$, for every
$\eta>0$ there exists $\eps_0>0$ such that, for every
$u\in(0,1)$ with $|u-\gamma|<\eps_0$,
$\bigl|\Phi_{\bJ}^{-1}(u)-q_\gamma\bigr|<\eta$.
Let
$0<\eps<\min\{\eps_0,\gamma,1-\gamma\}$, then $\gamma\pm\eps\in(0,1)$ and
$|\gamma\pm\eps-\gamma|=\eps<\eps_0$, and hence
\[
\bigl|\Phi_{\bJ}^{-1}(\gamma\pm\eps)-q_\gamma\bigr|
<\eta.
\]
This implies
\begin{equation}
    \label{eq:ineqpsi2}
q_\gamma-\eta<\Phi_{\bJ}^{-1}(\gamma-\eps)\leqslant\Phi_{\bJ}^{-1}(\gamma+\eps)
<q_\gamma+\eta.
\end{equation}
Consider the event
\[
\left\{
\sup_{x\in\mathbb R}
\bigl|F_n^*(x)-\Phi_{\bJ}(x)\bigr|
< \eps
\right\},
\]
by applying \eqref{eq:ineqq}, 
\[
\Phi_{\bJ}^{-1}(\gamma-\eps)
\;\leqslant\;
\hq_n^*(\gamma)
\;\leqslant\;
\Phi_{\bJ}^{-1}(\gamma+\eps),
\]
and from \eqref{eq:ineqpsi2}
\[
q_\gamma-\eta
<
\Phi_{\bJ}^{-1}(\gamma-\eps)
\leqslant
\hq_n^*(\gamma)
\leqslant
\Phi_{\bJ}^{-1}(\gamma+\eps)
<
q_\gamma+\eta,
\]
and hence
\begin{equation*}
\bigl|\hq_n^*(\gamma)-q_\gamma\bigr|<\eta.
\end{equation*}
Therefore,
\[
\Pr\!\left(
|\hq_n^*(\gamma)-q_\gamma|>\eta
\right)
\leqslant
\Pr\!\left(
\sup_{x\in\mathbb R}
\bigl|F_n^*(x)-\Phi_{\bJ}(x)\bigr|
\geqslant\eps
\right).
\]
Since
\[
\sup_{x\in\mathbb R}
\bigl|F_n^*(x)-\Phi_{\bJ}(x)\bigr|
\;\rightarrow_p\;0,
\]
 we conclude that
\[
\hq_n^*(\gamma)\;\rightarrow_p\; q_\gamma\qquad\text{for every }\gamma\in(0,1).
\]

Noting that for $\gamma\in(0,1)$, $\hc_n^*(\gamma)$ is defined as
\[
\hc_n^*(\gamma)
:=
\inf\Big\{x\in\mathbb R:\;{\Pr}^*(\htheta_n^{\,*}\leqslant x)\geqslant \gamma\Big\},
\]
and because $\htheta_n^{\,*}=\htheta_n+n^{-1/2}\sqrt n(\htheta_n^{\,*}-\htheta_n)$,
we have
\begin{equation*}
\hc_n^*(\gamma)=\htheta_n+\frac{1}{\sqrt n}\,\hq_n^*(\gamma).
\end{equation*}
Therefore,
\[
\{\theta_0\in C_n\}
=
\Big\{
\hq_n^*(\alpha/2)\leqslant \sqrt n(\theta_0-\htheta_n)\leqslant \hq_n^*(1-\alpha/2)
\Big\}
=
\Big\{
-\hq_n^*(1-\alpha/2)\leqslant \sqrt n(\htheta_n-\theta_0)\leqslant -\hq_n^*(\alpha/2)
\Big\}.
\]
Hence
\[
\Pr\{\theta_0\in C_n\}
=
\Pr\!\left(\sqrt n(\htheta_n-\theta_0)\leqslant -\hq_n^*(\alpha/2)\right)
-
\Pr\!\left(\sqrt n(\htheta_n-\theta_0)< -\hq_n^*(1-\alpha/2)\right).
\]
By Slutsky's theorem,
\[
\sqrt n(\htheta_n-\theta_0)+\hq_n^*(\alpha/2)\to_d Z_{\theta}+q_{\alpha/2}.
\]
Since $Z_{\theta}+q_{\alpha/2}$ is Gaussian, its distribution function is continuous
everywhere, in particular at $0$. Hence, by the definition of convergence in
distribution (convergence of cdfs at continuity points),
\begin{align*}
\Pr\!\left(\sqrt n(\htheta_n-\theta_0)\leqslant -\hq_n^*(\alpha/2)\right)
=&
\Pr\!\left(\sqrt n(\htheta_n-\theta_0)+\hq_n^*(\alpha/2)\leqslant 0\right)
\rightarrow\\
&\qquad \Pr(Z_{\theta}+q_{\alpha/2}\leqslant 0)
=
\Pr(Z_{\theta}\leqslant -q_{\alpha/2}).
\end{align*}
Similarly, using $\hq_n^*(1-\alpha/2)\rightarrow_p q_{1-\alpha/2}$,
\[
\sqrt n(\htheta_n-\theta_0)+\hq_n^*(1-\alpha/2)\to_d Z_{\theta}+q_{1-\alpha/2},
\]
and 
\begin{equation*}
\Pr\!\left(\sqrt n(\htheta_n-\theta_0)\leqslant -\hq_n^*(1-\alpha/2)\right)
\rightarrow
\Pr(Z_{\theta}\leqslant -q_{1-\alpha/2}).
\end{equation*}
Because $Z_{\theta}$ has a continuous distribution, we also have
$
\Pr(Z_{\theta}< -q_{1-\alpha/2})=\Pr(Z_{\theta}\leqslant -q_{1-\alpha/2})$,
and therefore by combining the two limits gives
\begin{align*}
\Pr\{\theta_0\in C_n\}
&\rightarrow
\Pr(Z_{\theta}\leqslant -q_{\alpha/2})-\Pr(Z_{\theta}\leqslant -q_{1-\alpha/2})
=
\Pr\big(-q_{1-\alpha/2}\leqslant Z_{\theta}\leqslant -q_{\alpha/2}\big).
\end{align*}
As $Z_{\theta}$ is symmetric about $0$, then $-q_{1-u}=q_{u}$ for $u\in(0,1)$,
so the last probability equals
\[
\Pr\Big\{q_{\alpha/2}\leqslant Z_{\theta}\leqslant q_{1-\alpha/2}\Big\}=1-\alpha,
\]
which proves \eqref{eq:reg-coverage}.
\end{proof}

\subsection{Derivation of the influence functions}
\label{sec:derivationIF_cellboot}

We consider the contamination model \eqref{eq:cont_reg_mixed2} and the same setting as introduced in Supplementary Material~\ref{app:cellMR}, with $H_0:=F_{\btheta_0}$.
Given  a $d$-dimensional vector $\bmu$ and a $d\times d$ matrix $\bSigma$,
define the
selection matrices $\bS_x\in\mathbb R^{p\times d}$ and
$\bS_y\in\mathbb R^{q\times d}$ that extract the first $p$ and the last $q$ components of $\bmu$ such that
$\bmu_x=\bS_x\bmu$, $\bmu_y=\bS_y\bmu$ and
$\bSigma_x=\bS_x\bSigma \bS_x^T$, $\bSigma_{xy}=\bS_x\bSigma \bS_y^T$, represent the usual partitions of $\bmu$ and $\bSigma$.
Recall that $\htheta_n=\bigl(\bhb^T,\ \vect(\bhB)^T\bigr)\ba$, where
\[
\bhB
=
\bigl(\bS_x\bhSigma_{II}\bS_x^T+\lambda\bI_p\bigr)^{-1}
(\bS_x\bhSigma_{II}\bS_y^T),
\qquad
\bhb
=
\bS_y\bhmu_{II}-\bhB^T(\bS_x\bhmu_{II}),
\]
and $\bhtheta_n=\bigl((\bhmu_{II})^T,\vecth_s(\bhSigma_{II})^T\bigr)^T$ denotes the II
estimator computed by applying to the observed sample the indirect inference procedure with
FastCellCov as auxiliary estimators.
The functional version of $\htheta_n$ is then defined, for a generic distribution $H$ as
\[
T(H)
:=\bigl(\bb(H)^T,\ \vect(\bB(H))^T\bigr)\ba,
\]
where
\[
\bB(H)
=
\bigl(\bS_x\bSigma_{II}(H)\bS_x^T+\lambda\bI_p\bigr)^{-1}
(\bS_x\bSigma_{II}(H)\bS_y^T),
\qquad
\bb(H)
=
\bS_y\bmu_{II}(H)-\bB(H)^T(\bS_x\bmu_{II}(H)),
\]
and $\bmu_{II}(H)$ and $\bSigma_{II}(H)$ are the functionals corresponding to the estimators $\bhmu_{II}$ and $\bhSigma_{II}$, respectively.
Thus, $T(H)=\varphi(\btheta(H))$, where
\[
\btheta(H)=\bigl((\bmu_{II}(H))^T,\ \vecth_s(\bSigma_{II}(H))^T\bigr)^T,
\]
and $\varphi$ denotes the mapping that transforms the pair 
$(\bmu,\bSigma)$ into the regression parameters 
$(\bb,\bB)$ through the above formulas. 

Let us consider $P_F(H,\boeta(H))$ the functional associated with the auxiliary estimator $\hP_F$ defined in Supplementary Material~\ref{app:asyCBoot}, where  $\boeta(H)$ represents the functional corresponding to the tuning parameter estimator $\hat \boeta$.
Specifically for a vector $\boeta\in \mathbb R^r$,
\[
P_F(H,\boeta)=\mu_{a,F}(H,\boeta)\oslash \mu_{b,F}(H,\boeta),
\]
with
\begin{equation*}
\mu_{a,F}(H,\boeta)=\E_H[a_F(X,\boeta)],
\qquad
\mu_{b,F}(H,\boeta)=\E_H[b_F(X,\boeta)],
\end{equation*}
where $a_F$ and $b_F$ are defined in  Supplementary Material~\ref{app:asyCBoot} . 
Further introduce
$Z_F(X,\boeta)
:=\left(a_F(X,\boeta)^T,b_F(X,\boeta)^T\right)^T$, 
and
$\mu_{Z_F}(H,\boeta):=\E_H[Z_F(X,\boeta)]=\left(
\mu_{a,F}(H,\boeta)^T,
\mu_{b,F}(H,\boeta)^T\right)^T.$
Then, $\btheta(H)$ satisfies
\begin{equation}\label{eq:popuII}
P_F(H,\boeta(H))-\pi\bigl(\btheta(H),\boeta(H)\bigr)=\bzero.
\end{equation}

Moreover, 
\[
\pi(\btheta,\boeta)
=
\mu_{a,F}(\btheta,\boeta)\oslash \mu_{b,F}(\btheta,\boeta),
\]
where $\mu_{Z_F}(\btheta,\boeta):=\E_{F_{\btheta}}[Z_F(X,\boeta)]=\left(
\mu_{a,F}(\btheta,\boeta)^T,
\mu_{b,F}(\btheta,\boeta)^T\right)^T$, with
\[
\mu_{a,F}(\btheta,\boeta):=\E_{F_{\btheta}}\!\big[a_F(X,\boeta)\big],
\qquad
\mu_{b,F}(\btheta,\boeta):=\E_{F_{\btheta}}\!\big[b_F(X,\boeta)\big].
\]

Under the appropriate regularity conditions, following the same arguments used in the proofs of Proposition~\ref{pro:bootstrap-theta},
Theorem~\ref{the:bootstrap-theta}, and Corollary~\ref{prop:bootstrap-theta-linear-reg},
under $H_0$ we have
\begin{equation*}
\sqrt{n}(\htheta_n-\theta_0)\rightarrow_d Z_\theta ,
\end{equation*}
where $Z_\theta\sim N(0,\sigma_{\theta,0}^2)$ and
the asymptotic variance is given by
\[
\sigma_{\theta,0}^2
=
\bJ_0^T\,\bK_0\,\bD_0\,\bSigma_{ab,0}\,\bD_0^T\,\bK_0^T\,\bJ_0.
\]
Here
\[
\bD_0
:=
\Big[
\diag\big(\mu_{b,F}(\btheta_0,\boeta_0)\big)^{-1},
\;
-\diag\!\left(
\mu_{a,F}(\btheta_0,\boeta_0)\oslash \mu_{b,F}(\btheta_0,\boeta_0)^{ 2}
\right)
\Big],
\]
\[
\bJ_0
:=
\left.
\frac{\partial\,\varphi(\btheta)}{\partial\btheta}
\right|_{\btheta=\btheta_0},
\qquad
\bK_0
:=\left(
\left.
\frac{\partial\,\pi(\btheta,\boeta_0)}{\partial\btheta}
\right|_{\btheta=\btheta_0}\right)^{-1},
\]
and
\[
\bSigma_{ab,0}
:=
\Cov_{H_0}\bigl(Z_F(X,\boeta_0)\bigr),
\]
where $\btheta_0:=\btheta(H_0)$ and $\boeta_0:=\boeta(H_0)$.

Consider the functional  
\[
\sigma_{\theta}(H)^2
:=
\bJ(H)^T\,\bK(H)\,\bD(H)\,\bSigma_{ab}(H)\,\bD(H)^T\,\bK(H)^T\,\bJ(H).
\]
with $\sigma_{\theta}(H_0)=\sigma_{\theta,0}$.

Here
\[
\bJ(H)
:=
\left.
\frac{\partial\,\varphi(\btheta)}{\partial\btheta}
\right|_{\btheta=\btheta(H)},
\qquad
\bK(H)
=
\left(\left.
\frac{\partial\,\pi(\btheta,\boeta(H)}{\partial\btheta}
\right|_{\btheta=\btheta(H)}\right)^{-1},
\]
\[
\bD(H)
:=
\Big[
\diag\bigl(\mu_{b,F}(H,\boeta(H))\bigr)^{-1},
\;
-\diag\!\left(
\mu_{a,F}(H,\boeta(H))\oslash \mu_{b,F} (H,\boeta(H))^2
\right)
\Big],
\]
and
\[
\bSigma_{ab}(H)
:=
\Cov_{H}\bigl(Z_F(X,\boeta(H))\bigr).
\]

\vspace{3mm}
We will be using the following two 
conditions:
\begin{enumerate}
\item[(E1)]
Under $H_0$,
\[
\sqrt{n}(\htheta_n-\theta_0)
\rightarrow_d
N(0,\sigma_{\theta,0}^2),
\]
with $\sigma_{\theta,0}^2>0$.
\item[(E2)] 
For $\gamma\in\{\alpha/2,\,1-\alpha/2\}$,
\begin{equation}\label{eq:weak-A2}
\frac{Q_{H_\eps}(\gamma)}{\sigma_{\theta}(H_\eps)}
=
\frac{Q_{H_0}(\gamma)}{\sigma_{\theta}(H_0)}
+
o(\eps),
\qquad \eps\downarrow 0,
\end{equation}
and $\sigma_{\theta}(H_\eps)=O(1)$ as $\eps\downarrow 0$.
\end{enumerate}

Assumption~(E1) requires a nondegenerate 
normal limiting distribution at the ideal model. 
Assumption~(E2) requires the two standardized 
quantiles defining the confidence interval to 
remain unchanged to first order under small
contamination, while the asymptotic scale remains 
bounded. Thus, the change in interval length is
determined by the change in scale.
This condition holds, for example, if the centered
limiting distributions remain normal along the
contamination path, with positive and bounded
standard deviations.

\setcounter{oldprop}{\value{proposition}}
\setcounter{proposition}{\getrefnumber{lem:IF-length-final-weakA2}-1}

\begin{proposition}[Restated]
\label{app:prop7}
Under conditions E1 and E2, 
\begin{equation*}
\IFu_{\case}(\bc,m,H_0)=\bmed_c^TZ_F(\bc,\boeta_0),
\end{equation*}
\begin{equation*}
\IFu_{\cell}(\bc,m,H_0)= d\bmed_c^T\left(\sum_{j=1}^d \E_{H(j,\bc)}[Z_F(X,\boeta_0)]\right),
\end{equation*}
and 

\begin{align*}
\IFu_{\case}(\bc,\tl_{\alpha},H_0)=2z_{1-\alpha/2}\Big[m_{s,1}&
+\bmed_{s,2}^T\,\IFu_{\case}(\bc,\boeta,H_0)\\
&+\bmed_{s,3}^T Z_F(\bc,\boeta_0)
+\bmed_{s,4}^T\!\bigl(
Z_F(\bc,\boeta_0)\otimes Z_F(\bc,\boeta_0)
\bigr)\Big],
\end{align*}
\begin{align*}
\IFu_{\cell}(\bc,\tl_{\alpha},H_0)=2z_{1-\alpha/2}\Big[dm_{s,1}
+&\bmed_{s,2}^T\,\IFu_{\cell}(\bc,\boeta,H_0)
+d\bmed_{s,3}^T\left( \sum_{j=1}^{d}\E_{H(j,\bc)}\!\big[Z_F(X,\boeta_0)\big]\right)
\\
&+d\bmed_{s,4}^T\left(\sum_{j=1}^{d}
\E_{H(j,\bc)}\!\bigl[
Z_F(X,\boeta_0)\otimes Z_F(X,\boeta_0)
\bigr]\right)\Big].
\end{align*}
where  $\otimes$ denotes the Kronecker product,
$\boeta_0:=\boeta(H_0)$ with $\boeta(\cdot)$ the functional corresponding to the tuning parameters estimator $\boheta_n$ in the auxiliary estimator,  $H(j,\bc)$ is the distribution of 
$X \sim H_0$ but with its $j$-th component fixed 
at the constant $c_j$\,, and $z_\gamma$ denotes the $\gamma$-quantile of the standard normal
distribution.
The quantities $\bmed_c$, $m_{s,1}$, $\bmed_{s,2}$, $\bmed_{s,3}$ and $\bmed_{s,4}$ are 
defined in  the proof, and 
$\IFu_{\case}(\bc,\boeta,H_0)$ and $\IFu_{\cell}(\bc,\boeta,H_0)$ 
are the casewise and cellwise influence functions of 
$\boeta$.

\end{proposition}

\begin{proof}[Proof of 
Proposition~\ref{lem:IF-length-final-weakA2}]
Assumption (E1) implies
$\sqrt n\{T(\hat H_0)-T(H_0)\}\rightarrow_d N(0,\sigma_{\theta,0}^2)$, where $\hat H_0$  is the empirical
distribution obtained from an i.i.d.\ sample of size $n$ drawn from $H_0$.
This implies
\[
\frac{Q_{H_0}(\gamma)}{\sigma_{\theta}(H_0)}=z_\gamma,
\qquad \gamma\in\{\alpha/2,1-\alpha/2\}.
\]
Combining this with \eqref{eq:weak-A2} yields for $\gamma\in\{\alpha/2,1-\alpha/2\}$,
\[
\frac{Q_{H_\eps}(\gamma)}{\sigma_{\theta}(H_\eps)}
=
z_\gamma+o(\eps), \qquad  \eps\downarrow 0.
\]
As $\sigma_{\theta}(H_\eps)=O(1)$ as $\eps\downarrow 0$, it follows that
\[
Q_{H_\eps}(\gamma)
=
\sigma_{\theta}(H_\eps)\,z_\gamma
+
o(\eps), \qquad  \eps\downarrow 0.
\]
Subtracting the expansions for $\gamma\in\{\alpha/2,1-\alpha/2\}$ yields
\[
\tl_{\alpha}(H_\eps)
=
\sigma_{\theta}(H_\eps)\{z_{1-\alpha/2}-z_{\alpha/2}\}
+
o(\eps), \qquad  \eps\downarrow 0.
\]
Since $z_{1-\alpha/2}=-z_{\alpha/2}$, the bracket simplifies to $2z_{1-\alpha/2}$, and hence
\[
\tl_{\alpha}(H_\eps)
=
2z_{1-\alpha/2}\,\sigma_{\theta}(H_\eps)
+
o(\eps), \qquad  \eps\downarrow 0.
\]
When the distribution $H_\eps$ is of the form $H(G_\eps^{D},\bc)$,
differentiating with respect to $\eps$ at $\eps=0$ yields
\begin{equation}\label{eq:IF_lproof}
    \IFu_{\case}(\bc,\tl_{\alpha},H_0)
=
2z_{1-\alpha/2}\,\IFu_{\case}(\bc,\sigma_{\theta},H_0).
\end{equation}
Similarly, differentiating the definition of $m$ gives
\begin{equation}\label{eq:IF_mproof}
\IFu_{\case}(\bc,m,H_0)
=
\IFu_{\case}(\bc,T,H_0).
\end{equation}
In the cellwise contamination setting, where $H_\eps$ is denoted by
$H(G_\eps^{I},\bc)$, an analogous result follows, that is 
\begin{equation}\label{eq:IF_lproof_cell}
    \IFu_{\cell}(\bc,\tl_{\alpha},H_0)
=
2z_{1-\alpha/2}\,\IFu_{\cell}(\bc,\sigma_{\theta},H_0).
\end{equation}
\begin{equation}\label{eq:IF_mproof_cell}
\IFu_{\cell}(\bc,m,H_0)
=
\IFu_{\cell}(\bc,T,H_0).
\end{equation}

Let consider the functional $T$ and  the contaminated distribution $H_{\eps}$.
Then
\[
\left.\frac{\partial}{\partial\eps}
T\bigl(H_{\eps}\bigr)\right|_{\eps=0}
=
\bJ_0^T
\left.\frac{\partial}{\partial\eps}
\btheta(H_{\eps})\right|_{\eps=0}.
\]
Differentiating \eqref{eq:popuII}, we have
\[
\left.\frac{\partial}{\partial\eps}\left(
P_F(H_{\eps},\boeta(H_{\eps}))
-\pi\bigl(\btheta(H_{\eps}),\boeta(H_{\eps})\bigr)
\right)\right|_{\eps=0}
=\bzero,
\]
and thus
\begin{align*}
\left.\frac{\partial}{\partial\eps}\btheta(H_{\eps})\right|_{\eps=0}
&=
\bK_0
\Bigg[
\left.\frac{\partial}{\partial\eps}P_F(H_{\eps},\boeta(H_{\eps}))\right|_{\eps=0}\\
&\hspace{4cm}-
\left.\frac{\partial \pi(\btheta_0,\boeta)}{\partial\boeta}
\right|_{\boeta=\boeta_0}
\left.\frac{\partial}{\partial\eps}\boeta(H_{\eps})\right|_{\eps=0}
\Bigg].
\end{align*}
Note that,
\[
\left.\frac{\partial \pi(\btheta,\boeta_0)}{\partial\btheta}\right|_{\btheta=\btheta_0}
=
\bD_0\,
\left.\frac{\partial\mu_{Z_F}(\btheta,\boeta_0)}{\partial\btheta}\right|_{\btheta=\btheta_0}.
\]

Note that $H_0=F_{\btheta_0}$, thus 
\[
\pi(\btheta_0,\boeta)=P_F(H_0,\boeta).
\]
Then
\[
\left.\frac{\partial \pi(\btheta_0,\boeta)}{\partial\boeta}\right|_{\boeta=\boeta_0}
=
\left.\frac{\partial P_F(H_0,\boeta)}{\partial\boeta}\right|_{\boeta=\boeta_0}.
\]
Moreover, by the chain rule,
\[
\left.\frac{\partial}{\partial\eps}P_F(H_{\eps},\boeta(H_{\eps}))\right|_{\eps=0}
=
\left.\frac{\partial}{\partial\eps}P_F(H_{\eps},\boeta_0)\right|_{\eps=0}
+
\left.\frac{\partial P_F(H_0,\boeta)}{\partial\boeta}\right|_{\boeta=\boeta_0}
\left.\frac{\partial}{\partial\eps}\boeta(H_{\eps})\right|_{\eps=0}.
\]
Substituting into the differentiated population condition, the $\boeta$-terms cancel, and hence
\begin{equation}\label{eq:dertheta_eps}
\left.\frac{\partial}{\partial\eps}\btheta(H_{\eps})\right|_{\eps=0}
=
\bK_0
\left.\frac{\partial}{\partial\eps}P_F(H_{\eps},\boeta_0)\right|_{\eps=0}.
\end{equation}
Therefore,
\begin{equation}\label{eq:IFTpart}
\left.\frac{\partial}{\partial\eps}T(H_{\eps})\right|_{\eps=0}
=
\bJ_0^T
\bK_0
\left.\frac{\partial}{\partial\eps}P_F(H_{\eps},\boeta_0)\right|_{\eps=0}.
\end{equation}

When $H_{\eps}$ is either the dependent or independent 
contamination models, we have, for any vector-valued function $g$,
\begin{equation}\label{eq:meandecomp}
\E_{H_{\eps}}[g(X)]
=
\sum_{j=0}^d\delta_j(\eps)\sum_{I\in\ell_j}\,\E_{H(I,\bc)}[g(X)],
\end{equation}
where
\[
\ell_j=\{I=\{i_1,\ldots,i_j\}: 1\leqslant i_1<\cdots<i_j\leqslant d,\ 1\leqslant j\leqslant d\},
\]
and where $H(I,\bc)$ is the distribution of a random vector $X=(X_{1},\ldots,X_{d})$ such that
\[
X_{i}=c_i \ \text{if } i\in I,
\qquad
X_{i}=Z_i \ \text{if } i\notin I,
\]
with $Z\sim H_0$.
In particular, $\ell_0=\{\emptyset\}$ and $H(\emptyset,\bc)=H_0$, while
$\ell_d=\{1,2,\ldots,d\}$ and $H(\{1,2,\ldots,d\},\bc)=\Delta_{\bc}$, the point-mass
distribution at $\bc$.

Under the FDCM we have 
$\delta_0(\eps)=(1-\eps),\; \delta_1(\eps)=\cdots=\delta_{d-1}(\eps)=0$, and 
$\delta_d(\eps)=\eps$.

Thus,
\[
\E_{H(G_\eps^D,\bc)}[g(X)]
=
(1-\eps)\E_{H_0}[g(X)]
+\eps\,g(\bc).
\]

Applying this to $g(X)=a_F(X,\boeta_0)$ and $g(X)=b_F(X,\boeta_0)$ yields
\[
\mu_{Z_F}(H(G_\eps^D,\bc),\boeta_0)
=
(1-\eps)\mu_{Z_F}(H_0,\boeta_0)
+\eps\,Z_F(\bc,\boeta_0).\]
Therefore,
\begin{equation}\label{eq:IFmu_FDCM}
\left.\frac{\partial}{\partial\eps}
\mu_{Z_F}(H(G_\eps^D,\bc),\boeta_0)\right|_{\eps=0}
=
Z_F(\bc,\boeta_0)-\mu_{Z_F}(H_0,\boeta_0).
\end{equation}

Using the componentwise quotient rule for
$P_F(H(G_\eps^D,\bc),\boeta_0)
=\mu_{a,F}(H(G_\eps^D,\bc),\boeta_0)\oslash
\mu_{b,F}(H(G_\eps^D,\bc),\boeta_0)$,
we obtain
\begin{align*}
\left.\frac{\partial}{\partial\eps}
P_F(H(G_\eps^D,\bc),\boeta_0)\right|_{\eps=0}
=&\;
\diag\!\big(\mu_{b,F}(H_0,\boeta_0)\big)^{-1}
\left.\frac{\partial}{\partial\eps}
\mu_{a,F}(H(G_\eps^D,\bc),\boeta_0)\right|_{\eps=0}\\
&-
\diag\!\Big(
\mu_{a,F}(H_0,\boeta_0)
\oslash
\mu_{b,F}(H_0,\boeta_0)^{2}
\Big)
\left.\frac{\partial}{\partial\eps}
\mu_{b,F}(H(G_\eps^D,\bc),\boeta_0)\right|_{\eps=0}\\
=&\;
\diag\!\big(\mu_{b,F}(H_0,\boeta_0)\big)^{-1}
\big(a_F(\bc,\boeta_0)-\mu_{a,F}(H_0,\boeta_0)\big)\\
&-
\diag\!\Big(
\mu_{a,F}(H_0,\boeta_0)
\oslash
\mu_{b,F}(H_0,\boeta_0)^{2}
\Big)
\big(b_F(\bc,\boeta_0)-\mu_{b,F}(H_0,\boeta_0)\big)\\
=&\diag\!\big(\mu_{b,F}(H_0,\boeta_0)\big)^{-1} a_F(\bc,\boeta_0)
\\
&-
\diag\!\Big(
\mu_{a,F}(H_0,\boeta_0)\oslash\mu_{b,F}(H_0,\boeta_0)^{2}
\Big)\, b_F(\bc,\boeta_0).
\end{align*}
Thus, we can write
\begin{equation}\label{eq:IFP_F}
\left.\frac{\partial}{\partial\eps}
P_F(H(G_\eps^D,\bc),\boeta_0)\right|_{\eps=0}
=
\bD_0Z_F(\bc,\boeta_0).
\end{equation}
Then from \eqref{eq:IFTpart}, we have
\begin{equation*}
\IFu_{\case}(\bc,T,H_0)
=\bmed_c^TZ_F(\bc,\boeta_0),
\end{equation*}
with $\bmed_c^T=
\bJ_0^T
\bK_0\bD_0$.

Under the FICM we have
\[
\delta_0(\eps)=(1-\eps)^d,\qquad
\delta_1(\eps)=d(1-\eps)^{d-1}\eps,
\]
so that $\delta_0(0)=1$, $\delta_1(0)=0$, $\delta_1'(0)=d$, and
$\delta_j(0)=\delta_j'(0)=0$ for all $j\geqslant 2$. 
By separating the first two terms in \eqref{eq:meandecomp}, we have
\begin{align*}
\E_{H(G_\eps^I,\bc)}[g(X)]
=
(1-\eps)^d\,\E_{H_0}[g(X)]
+
d(1-\eps)^{d-1}\eps&
\sum_{j=1}^d \E_{H(j,\bc)}[g(X)]\\&
+\sum_{k=2}^d \delta_k(\eps)\sum_{I\in\ell_k}\E_{H(I,\bc)}[g(X)],
\end{align*}
where $H(j,\bc):=H(\{j\},\bc)$.
This yields
\begin{align*}
\left.\frac{\partial}{\partial\eps}\E_{H(G_\eps^I,\bc)}[g(X)]\right|_{\eps=0}
=
-d\E_{H_0}[g(X)]
+d\sum_{j=1}^d \E_{H(j,\bc)}[g(X)].
\end{align*}
Applying this identity to $g(X)=a_F(X,\boeta_0)$ and $g(X)=b_F(X,\boeta_0)$ yields
\[
\left.\frac{\partial}{\partial\eps}
\mu_{Z_F}(H(G_\eps^I,\bc),\boeta_0)\right|_{\eps=0}
=
-d\,\mu_{Z_F}(H_0,\boeta_0)
+
d\sum_{j=1}^d \E_{H(j,\bc)}\!\big[Z_F(X,\boeta_0)\big].\]
Using the componentwise quotient rule for
$P=\mu_{a,F}\oslash\mu_{b,F}$ evaluated at $\eps=0$, we obtain
\begin{align*}
\left.\frac{\partial}{\partial\eps}
P_F(H(G_\eps^I,\bc),\boeta_0)\right|_{\eps=0}
=&\;
\diag\!\big(\mu_{b,F}(H_0,\boeta_0)\big)^{-1}
\left.\frac{\partial}{\partial\eps}
\mu_{a,F}(H(G_\eps^I,\bc),\boeta_0)\right|_{\eps=0}\\
&-
\diag\!\left(
\mu_{a,F}(H_0,\boeta_0)\oslash \mu_{b,F}(H_0,\boeta_0)^{ 2}
\right)
\left.\frac{\partial}{\partial\eps}
\mu_{b,F}(H(G_\eps^I,\bc),\boeta_0)\right|_{\eps=0}.
\end{align*}
And thus, by using $\bD_0$, this can be written compactly as
\[
\left.\frac{\partial}{\partial\eps}
P_F(H(G_\eps^I,\bc),\boeta_0)\right|_{\eps=0}
=
\bD_0
d\sum_{j=1}^d \E_{H(j,\bc)}[Z_F(X,\boeta_0)].
\]
Then from \eqref{eq:IFTpart}, we have
\begin{equation*}
\IFu_{\cell}(\bc,T,H_0)
=d\bmed_c^T\left(\sum_{j=1}^d \E_{H(j,\bc)}[Z_F(X,\boeta_0)]\right).
\end{equation*}

Let us consider the functional  $\sigma_{\theta}(\cdot)^2$ 
evaluated in $H_{\eps}$
\[
\sigma_{\theta}(H_{\eps})^2
=
\bJ_{\eps}^T\,\bK_{\eps}\,\bD_{\eps}\,\bSigma_{ab,\eps}\,\bD_{\eps}^T\,\bK_{\eps}^T\,\bJ_{\eps}.
\]
where
$\bJ_{\eps}
:=\bJ(H_{\eps})$, $
\bK_{\eps}
:=\bK(H_{\eps})$, $\bD_{\eps}
:=\bD(H_{\eps})$, and $\bSigma_{ab,\eps}:=\bSigma_{ab}(H_{\eps})
$.
Let
\[
\bM_\eps
:=
\bK_\eps \bD_\eps \bSigma_{ab,\eps}
\bD_\eps^T \bK_\eps^T,
\qquad
\sigma_{\theta}(H_\eps)^2=\bJ_{\eps}^T\bM_\eps\bJ_{\eps}.
\]
Then
\begin{align*}
\left.\frac{\partial}{\partial\eps}\sigma_{\theta}(H_\eps)^2\right|_{\eps=0}
&=
\bJ_0^T
\left.\frac{\partial\bM_\eps}{\partial\eps}\right|_{\eps=0}
\bJ_0
+
2\bJ_0^T \bM_0
\left.\frac{\partial\bJ_{\eps}}{\partial\eps}\right|_{\eps=0}\\&=(\bJ_0^T\otimes\bJ_0^T)\left.\frac{\partial\vect(\bM_\eps)}{\partial\eps}\right|_{\eps=0}+
2\bJ_0^T \bM_0
\left.\frac{\partial\bJ_{\eps}}{\partial\eps}\right|_{\eps=0}.
\end{align*}
The chain rule gives
\[
\left.\frac{\partial\bJ_{\eps}}{\partial\eps}\right|_{\eps=0}
=
\bJ_{2,0}
\left.
\frac{\partial\btheta(H_\eps)}{\partial\eps}
\right|_{\eps=0},
\]
where
\[\bJ_{2,0}:=
\left.
\frac{\partial^2\varphi(\btheta)}
{\partial\btheta\,\partial\btheta^T}
\right|_{\btheta=\btheta_0},\]
and, from \eqref{eq:dertheta_eps},
\[
\left.
\frac{\partial\btheta(H_\eps)}{\partial\eps}
\right|_{\eps=0}
=
\bK_0
\left.
\frac{\partial}{\partial\eps}
P_F(H_\eps,\boeta_0)
\right|_{\eps=0}.
\]

Differentiating $\vect(\bM_\eps)$
yields
\begin{align*}
\left.\frac{\partial\vect(\bM_\eps)}{\partial\eps}\right|_{\eps=0}
=&\;
(\bK_0\bD_0 \bSigma_{ab,0} \bD_0^T \otimes \bI_{d+d(d+1)/2})\,
\left.\frac{\partial\vect(\bK_\eps)}{\partial\eps}\right|_{\eps=0}
\\
&+
(\bK_0 \bD_0 \bSigma_{ab,0} \otimes \bK_0)\,
\left.\frac{\partial\vect(\bD_\eps)}{\partial\eps}\right|_{\eps=0}
\\
&+
(\bK_0 \bD_0 \otimes \bK_0 \bD_0)\,
\left.\frac{\partial\vect(\bSigma_{ab,\eps})}{\partial\eps}\right|_{\eps=0}
\\
&+
(\bK_0 \otimes \bK_0 \bD_0 \bSigma_{ab,0})\,
\bP_{D}\,
\left.\frac{\partial\vect(\bD_\eps)}{\partial\eps}\right|_{\eps=0}
\\
&+
(\bI_{d+d(d+1)/2} \otimes \bK_0 \bD_0 \bSigma_{ab,0} \bD_0^T)\,
\bP_{K}\,
\left.\frac{\partial\vect(\bK_\eps)}{\partial\eps}\right|_{\eps=0},
\end{align*}
where $\bP_{D}$ and $\bP_{K}$ are the permutation matrices satisfying $\vect(\bD_{\eps}^T)=\bP_{D}\vect(\bD_{\eps})$ and $\vect(\bK_{\eps}^T)=\bP_{K}\vect(\bK_{\eps})$.
Let
\[
\bK_\eps
=
\left(
\left.
\frac{\partial \pi(\btheta,\boeta)}{\partial\btheta}
\right|_{(\btheta,\boeta)=(\btheta(H_\eps),\boeta(H_\eps))}
\right)^{-1}
=: \bG_\eps^{-1}.
\]
Using the identity $\partial\bG^{-1} = -\bG^{-1}(\partial\bG)\bG^{-1}$, we obtain
\[
\left.\frac{\partial\vect(\bK_\eps)}{\partial\eps}\right|_{\eps=0}
=
-(\bK_0^T\otimes \bK_0)
\left.\frac{\partial\vect(\bG_\eps)}{\partial\eps}\right|_{\eps=0}.
\]
For each $j=1,\dots,d+d(d+1)/2$, define
\[
\bA_j
:=
\left.\frac{\partial}{\partial\btheta^T}\Big(\frac{\partial \pi(\btheta,\boeta)}{\partial \theta_j}\Big)\right|_{(\btheta,\boeta)=(\btheta_0,\boeta_0)},
\qquad
(\bA_j)_{ik}
=
\left.\frac{\partial^2\pi_i(\btheta,\boeta)}{\partial \theta_j\,\partial\theta_k}\right|_{(\btheta,\boeta)=(\btheta_0,\boeta_0)},
\]
\[
\bB_j
:=
\left.\frac{\partial}{\partial\boeta^T}\Big(\frac{\partial \pi(\btheta,\boeta)}{\partial \theta_j}\Big)\right|_{(\btheta,\boeta)=(\btheta_0,\boeta_0)}
\in\mathbb R^{(d+d(d+1)/2)\times r},
\qquad
(\bB_j)_{i\ell}
=
\left.\frac{\partial^2\pi_i(\btheta,\boeta)}{\partial \theta_j\,\partial\eta_\ell}\right|_{(\btheta,\boeta)=(\btheta_0,\boeta_0)}.
\]
Introduce the stacked matrices
\[
\bH_{\theta\theta}
:=
\begin{pmatrix}
\bA_1\\
\vdots\\
\bA_{d+d(d+1)/2}
\end{pmatrix},
\qquad
\bH_{\theta\eta}
:=
\begin{pmatrix}
\bB_1\\
\vdots\\
\bB_{d+d(d+1)/2}
\end{pmatrix}.
\]
 Then, stacking the derivatives by rows and using vectorization,
\[
\left.\frac{\partial\vect(\bG_\eps)}{\partial\eps}\right|_{\eps=0}
=
\bH_{\theta\theta}
\left.\frac{\partial\btheta(H_\eps)}{\partial\eps}\right|_{\eps=0}
+
\bH_{\theta\eta}
\left.\frac{\partial\boeta(H_\eps)}{\partial\eps}\right|_{\eps=0}.
\]
Thus,
\[
\left.\frac{\partial\vect(\bK_\eps)}{\partial\eps}\right|_{\eps=0}
=
-(\bK_0^T\otimes \bK_0)\bH_{\theta\theta}
\left.\frac{\partial\btheta(H_\eps)}{\partial\eps}\right|_{\eps=0}
-(\bK_0^T\otimes \bK_0)\bH_{\theta\eta}
\left.\frac{\partial\boeta(H_\eps)}{\partial\eps}\right|_{\eps=0}.
\]

Write
\begin{align*}
\bD_\eps
=&
\Big[
\diag(\mu_{b,F}(H_\eps,\boeta(H_\eps)))^{-1},
\;
-\diag\!\big(
\mu_{a,F}(H_\eps,\boeta(H_\eps))
\oslash \mu_{b,F}(H_\eps,\boeta(H_\eps))^2
\big)
\Big]\\&
=:
[\bD^{(1)}_\eps,\bD^{(2)}_\eps],
\end{align*}
and hence
\[
\vect(\bD_\eps)
=
\begin{pmatrix}
\vect(\bD^{(1)}_\eps)\\
\vect(\bD^{(2)}_\eps)
\end{pmatrix}.
\]
Using $d(A^{-1})=-A^{-1}(dA)A^{-1}$, we obtain
\begin{align*}
\left.\frac{\partial \bD^{(1)}_\eps}{\partial \eps}\right|_{\eps=0}
=&\;
-\diag(\mu_{b,F}(H_0,\boeta_0))^{-1}
\\
&\times
\diag\!\left(
\left.
\frac{\partial}{\partial \eps}
\mu_{b,F}(H_\eps,\boeta(H_\eps))
\right|_{\eps=0}
\right)
\diag(\mu_{b,F}(H_0,\boeta_0))^{-1}.
\end{align*}
Vectorizing yields
\begin{align*}
\left.\frac{\partial \vect(\bD^{(1)}_\eps)}{\partial \eps}\right|_{\eps=0}
&=
-\big(
\diag(\mu_{b,F}(H_0,\boeta_0))^{-1}
\otimes
\diag(\mu_{b,F}(H_0,\boeta_0))^{-1}
\big)
\\&\hspace{6cm}\times\bL_{d+d(d+1)/2}
\left.
\frac{\partial}{\partial \eps}
\mu_{b,F}(H_\eps,\boeta(H_\eps))
\right|_{\eps=0}\\
&=-\bB^D_1\left.
\frac{\partial}{\partial \eps}
\mu_{b,F}(H_\eps,\boeta(H_\eps))
\right|_{\eps=0}.
\end{align*}
with $\bB^D_1:=\big(
\diag(\mu_{b,F}(H_0,\boeta_0))^{-1}
\otimes
\diag(\mu_{b,F}(H_0,\boeta_0))^{-1}
\big)
\bL_{d+d(d+1)/2}$ and
where for $\bv\in\mathbb R^{q}$ and  $\diag(\bv)\in\mathbb R^{q\times q}$, $\bL_q$ is defined 
so that
$\vect(\diag(\bv))=\bL_q\,\bv$.
Similarly,
\begin{align*}
\left.\frac{\partial \bD^{(2)}_\eps}{\partial \eps}\right|_{\eps=0}
=&\;
-\diag\!\big(\mu_{b,F}(H_0,\boeta_0)^{-2}\big)\,
\diag\!\left(\left.
\frac{\partial}{\partial \eps}
\mu_{a,F}(H_\eps,\boeta(H_\eps))
\right|_{\eps=0}\right)
\\
&\qquad
+\,2\,\diag\!\big(\mu_{a,F}(H_0,\boeta_0)\oslash\mu_{b,F}(H_0,\boeta_0)^3\big)\,
\diag\!\left(\left.
\frac{\partial}{\partial \eps}
\mu_{b,F}(H_\eps,\boeta(H_\eps))
\right|_{\eps=0}
\right),
\end{align*}
and
\begin{align*}
\left.\frac{\partial \vect(\bD^{(2)}_\eps)}{\partial \eps}\right|_{\eps=0}
=&\;
-\bigl(\bI \otimes \diag(\mu_{b,F}(H_0,\boeta_0)^{-2})\bigr)\,
\left.\frac{\partial}{\partial \eps}\vect\left(
\diag(\mu_{a,F}(H_\eps,\boeta(H_\eps)))
\right)\right|_{\eps=0}
\\
&\qquad
+\,2\bigl(\bI \otimes \diag(\mu_{a,F}(H_0,\boeta_0)\oslash\mu_{b,F}(H_0,\boeta_0)^3)\bigr)\,\\&\hspace{4cm}\times
\left.\frac{\partial}{\partial \eps}\vect\left(
\diag(\mu_{b,F}(H_\eps,\boeta(H_\eps)))
\right)\right|_{\eps=0}
\\
=&\;
-\bB^D_2
\left.\frac{\partial}{\partial \eps}
\mu_{a,F}(H_\eps,\boeta(H_\eps))
\right|_{\eps=0}
+\bB^D_3
\left.\frac{\partial}{\partial \eps}
\mu_{b,F}(H_\eps,\boeta(H_\eps))
\right|_{\eps=0},
\end{align*}
where $\bB^D_2:=\bigl(\bI \otimes \diag(\mu_{b,F}(H_0,\boeta_0)^{-2})\bigr)\,
\bL_{d+d(d+1)/2}\,$ and $\bB^D_3:=\,2\bigl(\bI \otimes \diag(\mu_{a,F}(H_0,\boeta_0)\oslash\mu_{b,F}(H_0,\boeta_0)^3)\bigr)\,
\bL_{d+d(d+1)/2}\,$.
By the chain rule,
\begin{align*}
\left.
\frac{\partial}{\partial \eps}
\mu_{Z_F}\bigl(H_\eps,\boeta(H_\eps)\bigr)
\right|_{\eps=0}
&=
\left.
\frac{\partial}{\partial \eps}
\mu_{Z_F}(H_\eps,\boeta_0)
\right|_{\eps=0}
+\bD_{\mu,\eta}
\left.
\frac{\partial \boeta(H_\eps)}{\partial \eps}
\right|_{\eps=0}
\end{align*}
with  $ \bD_{\mu,\eta}:=
\left.
\frac{\partial \mu_{Z_F}(H_0,\boeta)}{\partial \boeta}
\right|_{\boeta=\boeta_0}$.
So \[\left.\frac{\partial \vect(\bD_\eps)}{\partial \eps}\right|_{\eps=0}
=\bB^D\left.
\frac{\partial}{\partial \eps}
\mu_{Z_F}(H_\eps,\boeta_0)
\right|_{\eps=0}
+\bB^D\bD_{\mu,\eta}
\left.
\frac{\partial \boeta(H_\eps)}{\partial \eps}
\right|_{\eps=0},\]
where \[
\bB^D=\begin{pmatrix}
0&-\bB_1^D\\
-\bB_2^D&\bB_3^D
\end{pmatrix}.
\]

Recall that
\begin{align*}
\bSigma_{ab,\eps}
&=
\Cov_{H_\eps}\!\bigl(
Z_F(X,\boeta(H_\eps))
\bigr) \\
&=
\E_{H_\eps}\!\Big[
Z_F(X,\boeta(H_\eps))
Z_F(X,\boeta(H_\eps))^T
\Big]
-
\mu_{Z_F}(H_\eps,\boeta(H_\eps))
\,\mu_{Z_F}(H_\eps,\boeta(H_\eps))^T.
\end{align*}
Under the FDCM or the FICM,  differentiating with respect
to $\eps$ at $\eps=0$ yields
\begin{align*}
\left.
\frac{\partial\,\vect(\bSigma_{ab,\eps})}{\partial\eps}
\right|_{\eps=0}
=&\;
\left.
\frac{\partial}{\partial\eps}
\vect\Big(
\E_{H_\eps}\!\bigl[ Z_F(X,\boeta_0) Z_F(X,\boeta_0)^T \bigr]
\Big)
\right|_{\eps=0}
\\
&+
\left.
\frac{\partial}{\partial\eps}\,
\vect\Big(
\E_{H_0}\!\big[ Z_F(X,\boeta(H_\eps))\, Z_F(X,\boeta(H_\eps))^T \big]
\Big)\right|_{\eps=0}
\\
&-
\left.
\frac{\partial}{\partial\eps}\,
\vect\Big(
\mu_{Z_F}(H_\eps,\boeta(H_\eps))\,
\mu_{Z_F}(H_0,\boeta_0)^T
\Big)
\right|_{\eps=0}
\\
&-
\left.
\frac{\partial}{\partial\eps}\,
\vect\Big(
\mu_{Z_F}(H_0,\boeta_0)\,
\mu_{Z_F}(H_\eps,\boeta(H_\eps))^T
\Big)
\right|_{\eps=0}.
\end{align*}
Using $\vect(\bu\bv^T)=\bv\otimes \bu$, the last two terms become
\begin{align*}
\left.
\frac{\partial\,\vect(\bSigma_{ab,\eps})}{\partial\eps}
\right|_{\eps=0}
=&\;
\left.
\frac{\partial}{\partial\eps}
\vect\Big(
\E_{H_\eps}\!\bigl[ Z_F(X,\boeta_0) Z_F(X,\boeta_0)^T \bigr]
\Big)
\right|_{\eps=0}
\\
&+
\left.
\frac{\partial}{\partial\eps}\,
\vect\Big(
\E_{H_0}\!\big[ Z_F(X,\boeta(H_\eps))\, Z_F(X,\boeta(H_\eps))^T \big]
\Big)\right|_{\eps=0}
\\
&-\left(
\mu_{Z_F}(H_0,\boeta_0)\otimes \bI+
\bI\otimes \mu_{Z_F}(H_0,\boeta_0)\right)
\left.
\frac{\partial\,\mu_{Z_F}(H_\eps,\boeta(H_\eps))}{\partial\eps}
\right|_{\eps=0}.
\end{align*}
Since the expectation is taken with respect to the fixed distribution $H_0$, the only
$\eps$-dependence is through $\boeta(H_\eps)$. Hence, by the product rule,
\begin{align*}
\left.\frac{\partial}{\partial\eps}
\Big(
Z_F(X,\boeta(H_\eps))Z_F(X,\boeta(H_\eps))^T
\Big)\right|_{\eps=0}
=&\;
\left.\frac{\partial Z_F(X,\boeta(H_\eps))}{\partial\eps}\right|_{\eps=0}
Z_F(X,\boeta_0)^T
\\
&\quad+
Z_F(X,\boeta_0)
\left.\frac{\partial Z_F(X,\boeta(H_\eps))^T}{\partial\eps}\right|_{\eps=0}.
\end{align*}
Moreover, by the chain rule,
\[
\left.\frac{\partial Z_F(X,\boeta(H_\eps))}{\partial\eps}\right|_{\eps=0}
=
\left.\frac{\partial Z_F(X,\boeta)}{\partial\boeta}\right|_{\boeta=\boeta_0}
\left.\frac{\partial \boeta(H_\eps)}{\partial\eps}\right|_{\eps=0}.
\]
Therefore,
\begin{align*}
\frac{\partial}{\partial\eps}\,
&\vect\Big(
\E_{H_0}\!\big[
Z_F(X,\boeta(H_\eps))\,
Z_F(X,\boeta(H_\eps))^T
\big]
\Big)\Big|_{\eps=0}
\\&=\;
\E_{H_0}\!\Bigg[
\vect\Big(
\left.\frac{\partial Z_F(X,\boeta(H_\eps))}{\partial\eps}\right|_{\eps=0}
Z_F(X,\boeta_0)^T
+
Z_F(X,\boeta_0)
\left.\frac{\partial Z_F(X,\boeta(H_\eps))^T}{\partial\eps}\right|_{\eps=0}
\Big)
\Bigg]
\\
&=\;
\E_{H_0}\!\Bigg[
\big(
Z_F(X,\boeta_0)\otimes \bI
+
\bI\otimes Z_F(X,\boeta_0)
\big)
\left.\frac{\partial Z_F(X,\boeta)}{\partial\boeta}\right|_{\boeta=\boeta_0}\Bigg]
\left.\frac{\partial \boeta(H_\eps)}{\partial\eps}\right|_{\eps=0}
\\&=\;
\bB^Z_1
\left.\frac{\partial \boeta(H_\eps)}{\partial\eps}\right|_{\eps=0},
\end{align*}
with $\bB^Z_1:=\E_{H_0}\!\Bigg[
\big(
Z_F(X,\boeta_0)\otimes \bI
+
\bI\otimes Z_F(X,\boeta_0)
\big)
\left.\frac{\partial Z_F(X,\boeta)}{\partial\boeta}\right|_{\boeta=\boeta_0}\Bigg]$.
Thus 

\begin{align*}
\left.
\frac{\partial\,\vect(\bSigma_{ab,\eps})}{\partial\eps}
\right|_{\eps=0}
=&\;
\left.
\frac{\partial}{\partial\eps}
\E_{H_\eps}\!\bigl[ Z_F(X,\boeta_0) \otimes Z_F(X,\boeta_0) \bigr]
\right|_{\eps=0}
\\
&+
\left(\bB^Z_1-\bB^Z_2\bD_{\mu,\eta}\right)
\left.
\frac{\partial \boeta(H_\eps)}{\partial \eps}
\right|_{\eps=0}
-\bB^Z_2\left.
\frac{\partial}{\partial \eps}
\mu_{Z_F}(H_\eps,\boeta_0)
\right|_{\eps=0},
\end{align*}
with $\bB^Z_2:=\left(
\mu_{Z_F}(H_0,\boeta_0)\otimes \bI+
\bI\otimes \mu_{Z_F}(H_0,\boeta_0)\right)$.

Then,
\begin{align}
\left.\frac{\partial}{\partial\eps}\sigma_{\theta}(H_\eps)^2\right|_{\eps=0}
&=
(\bJ_0^T\otimes\bJ_0^T)
\left.\frac{\partial\vect(\bM_\eps)}{\partial\eps}\right|_{\eps=0}
+
2\bJ_0^T \bM_0
\left.\frac{\partial\bJ_{\eps}}{\partial\eps}\right|_{\eps=0}
\nonumber\\
&=
(\bJ_0^T\otimes\bJ_0^T)\,\Big\{
\bA_K\left.\frac{\partial\vect(\bK_\eps)}{\partial\eps}\right|_{\eps=0}
+
\bA_D\left.\frac{\partial\vect(\bD_\eps)}{\partial\eps}\right|_{\eps=0}\nonumber\\&\hspace{3cm}
+
\bA_\Sigma\left.\frac{\partial\vect(\bSigma_{ab,\eps})}{\partial\eps}\right|_{\eps=0}
\Big\}
+
2\bJ_0^T \bM_0 \bJ_{2,0}
\left.\frac{\partial\btheta(H_\eps)}{\partial\eps}\right|_{\eps=0},
\label{eq:dsig2_split}
\end{align}
where we have set
\begin{align*}
\bA_K
&:=
(\bK_0\bD_0 \bSigma_{ab,0} \bD_0^T \otimes \bI_{d+d(d+1)/2})
+
(\bI \otimes \bK_0 \bD_0 \bSigma_{ab,0} \bD_0^T)\,\bP_{K},
\\
\bA_D
&:=
(\bK_0 \bD_0 \bSigma_{ab,0} \otimes \bK_0)
+
(\bK_0 \otimes \bK_0 \bD_0 \bSigma_{ab,0})\,\bP_{D},
\\
\bA_\Sigma
&:=
(\bK_0 \bD_0 \otimes \bK_0 \bD_0).
\end{align*}
Moreover,
\begin{align}
\left.\frac{\partial\vect(\bK_\eps)}{\partial\eps}\right|_{\eps=0}
&=
-(\bK_0^T\otimes \bK_0)\bH_{\theta\theta}
\left.\frac{\partial\btheta(H_\eps)}{\partial\eps}\right|_{\eps=0}
-(\bK_0^T\otimes \bK_0)\bH_{\theta\eta}
\left.\frac{\partial\boeta(H_\eps)}{\partial\eps}\right|_{\eps=0},
\label{eq:dK_vec}\\[1ex]
\left.\frac{\partial \vect(\bD_\eps)}{\partial \eps}\right|_{\eps=0}
&=
\bB^D\left.
\frac{\partial}{\partial \eps}
\mu_{Z_F}(H_\eps,\boeta_0)
\right|_{\eps=0}
+\bB^D\bD_{\mu,\eta}
\left.
\frac{\partial \boeta(H_\eps)}{\partial \eps}
\right|_{\eps=0},
\label{eq:dD_vec}\\[1ex]
\left.
\frac{\partial\,\vect(\bSigma_{ab,\eps})}{\partial\eps}
\right|_{\eps=0}
&=
\left.
\frac{\partial}{\partial\eps}
\E_{H_\eps}\!\bigl[ Z_F(X,\boeta_0) \otimes Z_F(X,\boeta_0) \bigr]
\right|_{\eps=0}
+
\left(\bB^Z_1-\bB^Z_2\bD_{\mu,\eta}\right)
\left.
\frac{\partial \boeta(H_\eps)}{\partial \eps}
\right|_{\eps=0}
\nonumber\\
&\qquad
-\bB^Z_2\left.
\frac{\partial}{\partial \eps}
\mu_{Z_F}(H_\eps,\boeta_0)
\right|_{\eps=0}.
\label{eq:dSigma_vec}
\end{align}
Substituting \eqref{eq:dK_vec}-\eqref{eq:dSigma_vec} into \eqref{eq:dsig2_split} and
collecting terms yields
\begin{align}
\left.\frac{\partial}{\partial\eps}\sigma_{\theta}(H_\eps)^2\right|_{\eps=0}
&=
\bmed_\theta^T\bK_0
\left.
\frac{\partial}{\partial\eps}
P_F(H_\eps,\boeta_0)
\right|_{\eps=0}
+
\bmed_\eta^T
\left.\frac{\partial\boeta(H_\eps)}{\partial\eps}\right|_{\eps=0}
\nonumber\\&+
\bmed_\mu^T
\left.\frac{\partial}{\partial\eps}
\mu_{Z_F}(H_\eps,\boeta_0)\right|_{\eps=0}
+
\bmed_{\mathrm{dir}}^T\left.
\frac{\partial}{\partial\eps}
\E_{H_\eps}\!\bigl[ Z_F(X,\boeta_0) \otimes Z_F(X,\boeta_0) \bigr]
\right|_{\eps=0},
\label{eq:dsig2_collected2}
\end{align}
where
\begin{align*}
\bmed_\theta^T
&:=
-(\bJ_0^T\otimes\bJ_0^T)\,\bA_K\,(\bK_0^T\otimes \bK_0)\,\bH_{\theta\theta}
+2\bJ_0^T \bM_0 \bJ_{2,0},
\\
\bmed_\eta^T
&:=
-(\bJ_0^T\otimes\bJ_0^T)\,\bA_K\,(\bK_0^T\otimes \bK_0)\,\bH_{\theta\eta}
+(\bJ_0^T\otimes\bJ_0^T)\,\bA_D\,\bB^D\bD_{\mu,\eta}\\
&\qquad+(\bJ_0^T\otimes\bJ_0^T)\,\bA_\Sigma\,(\bB^Z_1-\bB^Z_2\bD_{\mu,\eta}),
\\
\bmed_\mu^T
&:=
(\bJ_0^T\otimes\bJ_0^T)\,\bA_D\,\bB^D
-(\bJ_0^T\otimes\bJ_0^T)\,\bA_\Sigma\,\bB^Z_2,
\\
\bmed_{\mathrm{dir}}^T
&:=
(\bJ_0^T\otimes\bJ_0^T)\,\bA_\Sigma\,
.
\end{align*}
We also used
\[
\left.\frac{\partial\btheta(H_\eps)}{\partial\eps}\right|_{\eps=0}
=
\bK_0
\left.
\frac{\partial}{\partial\eps}
P_F(H_\eps,\boeta_0)
\right|_{\eps=0}.
\]

Let
\[
g(X):= Z_F(X,\boeta_0)\otimes Z_F(X,\boeta_0).
\]
By \eqref{eq:meandecomp},
\[
\E_{H_{\eps}}[g(X)]
=
\sum_{j=0}^d\delta_j(\eps)\sum_{I\in\ell_j}\E_{H(I,\bc)}[g(X)].
\]

Under the FDCM we have $H_{\eps}=H(G_\eps^D,\bc)$, and 
$\delta_0(\eps)=(1-\eps),\; \delta_1(\eps)=\cdots=\delta_{d-1}(\eps)=0$, and 
$\delta_d(\eps)=\eps$.
Thus,
\[
\E_{H(G_\eps^D,\bc)}[g(X)]
=
(1-\eps)\E_{H_0}[g(X)]
+\eps\,g(\bc).
\]
and
\begin{equation*}
\left.
\frac{\partial}{\partial\eps}
\E_{H(G_\eps^D,\bc)}\!\bigl[ Z_F(X,\boeta_0)\otimes Z_F(X,\boeta_0)\bigr]
\right|_{\eps=0}
=
Z_F(\bc,\boeta_0)\otimes Z_F(\bc,\boeta_0)
-
\E_{H_0}\!\bigl[ Z_F(X,\boeta_0)\otimes Z_F(X,\boeta_0)\bigr].
\end{equation*}
Moreover from \eqref{eq:IFP_F}, we have
\begin{equation*}
\left.\frac{\partial}{\partial\eps}
P_F(H(G_\eps^D,\bc),\boeta_0)\right|_{\eps=0}
=
\bD_0Z_F(\bc,\boeta_0)
\end{equation*} and from \eqref{eq:IFmu_FDCM}
\begin{equation*}
\left.\frac{\partial}{\partial\eps}
\mu_{Z_F}(H(G_\eps^D,\bc),\boeta_0)\right|_{\eps=0}
=
Z_F(\bc,\boeta_0)-\mu_{Z_F}(H_0,\boeta_0).
\end{equation*}
Thus, from \eqref{eq:dsig2_collected2}, we have 
\begin{align}
\left.\frac{\partial}{\partial\eps}\sigma_{\theta}^2(H(G_\eps^D,\bc))\right|_{\eps=0}
=&\;
\bmed_\eta^T
\left.\frac{\partial\boeta(H(G_\eps^D,\bc))}{\partial\eps}\right|_{\eps=0}
-\bmed_\mu^T\,\mu_{Z_F}(H_0,\boeta_0)\\&\quad
-\bmed_{\mathrm{dir}}^T\,
\E_{H_0}\!\bigl[ Z_F(X,\boeta_0)\otimes Z_F(X,\boeta_0)\bigr]
\nonumber\\
&\quad
+(\bmed_\theta^T\bK_0\,\bD_0\,
+\bmed_\mu^T)\,Z_F(\bc,\boeta_0)
+\bmed_{\mathrm{dir}}^T\!\bigl(
Z_F(\bc,\boeta_0)\otimes Z_F(\bc,\boeta_0)
\bigr).
\end{align}
Thus
\begin{align*}
\IFu_{\case}(\bc,\sigma_{\theta},H_0)
&=
\frac{1}{2\,\sigma_{\theta,0}}
\left.\frac{\partial}{\partial\eps}\sigma_{\theta}^2(H(G_\eps^D,\bc))\right|_{\eps=0}
\\
&=\;
m_{s,1}
+\bmed_{s,2}^T\,\IFu_{\case}(\bc,\boeta,H_0)
+\bmed_{s,3}^T Z_F(\bc,\boeta_0)
+\bmed_{s,4}^T\!\bigl(
Z_F(\bc,\boeta_0)\otimes Z_F(\bc,\boeta_0)
\bigr),
\end{align*}
where \[
m_{s,1}:=
-\frac{1}{2\sigma_{\theta,0}}\bmed_{\mu}^T\,\mu_{Z_F}(H_0,\boeta_0)
-\frac{1}{2\sigma_{\theta,0}}\bmed_{\mathrm{dir}}^T\,
\E_{H_0}\!\bigl[ Z_F(X,\boeta_0)\otimes Z_F(X,\boeta_0)\bigr],
\]
\[
\bmed_{s,2}^T:=\frac{1}{2\sigma_{\theta,0}}\bmed_\eta^T,\qquad\bmed_{s,3}^T:=\frac{1}{2\sigma_{\theta,0}}\Big(\bmed_\theta^T\bK_0\,\bD_0+\bmed_\mu^T\Big),
\qquad
\bmed_{s,4}^T:=\frac{1}{2\sigma_{\theta,0}}\bmed_{\mathrm{dir}}^T.
\]

Under the FICM we have
\[
\delta_0(\eps)=(1-\eps)^d,\qquad
\delta_1(\eps)=d(1-\eps)^{d-1}\eps,
\]
so that $\delta_0(0)=1$, $\delta'_0(0)=-d$, $\delta_1(0)=0$, $\delta_1'(0)=d$, and
$\delta_j(0)=\delta_j'(0)=0$ for all $j\geqslant2$. 
Thus
\begin{align*}
\left.
\frac{\partial}{\partial\eps}
\E_{H_\eps}\!\bigl[
Z_F(X,\boeta_0)\otimes Z_F(X,\boeta_0)
\bigr]
\right|_{\eps=0}
=&\;
-d\,\E_{H_0}\!\bigl[
Z_F(X,\boeta_0)\otimes Z_F(X,\boeta_0)
\bigr]
\\
&\quad
+
d\sum_{j=1}^{d}
\E_{H(j,\bc)}\!\bigl[
Z_F(X,\boeta_0)\otimes Z_F(X,\boeta_0)
\bigr],
\end{align*}
where $H(j,\bc):=H(\{j\},\bc)$.
Using this together with
\[
\left.\frac{\partial}{\partial\eps}
P_F(H(G_\eps^I,\bc),\boeta_0)\right|_{\eps=0}
=
\bD_0
d\sum_{j=1}^d \E_{H(j,\bc)}[Z_F(X,\boeta_0)]
\]
and
\[
\left.\frac{\partial}{\partial\eps}
\mu_{Z_F}(H(G_\eps^I,\bc),\boeta_0)\right|_{\eps=0}
=
-d\,\mu_{Z_F}(H_0,\boeta_0)
+
d\sum_{j=1}^{d}\E_{H(j,\bc)}\!\big[Z_F(X,\boeta_0)\big],
\]
we obtain from \eqref{eq:dsig2_collected2} the expanded FICM expression
\begin{align*}
\left.\frac{\partial}{\partial\eps}\sigma_{\theta}^2(H(G_\eps^I,\bc))\right|_{\eps=0}
=&\;
\bmed_\theta^T\bK_0\bD_0
d
\sum_{j=1}^{d}\E_{H(j,\bc)}\!\big[Z_F(X,\boeta_0)\big]
+\bmed_\eta^T
\left.\frac{\partial\boeta(H(G_\eps^I,\bc))}{\partial\eps}\right|_{\eps=0}
\nonumber\\
&\quad
+\bmed_\mu^T
\left(
-d\,\mu_{Z_F}(H_0,\boeta_0)
+
d\sum_{j=1}^{d}\E_{H(j,\bc)}\!\big[Z_F(X,\boeta_0)\big]
\right)
\nonumber\\
&\quad
-d\,\bmed_{\mathrm{dir}}^T\E_{H_0}\!\bigl[
Z_F(X,\boeta_0)\otimes Z_F(X,\boeta_0)
\bigr]\\
&\quad
+
d\,\bmed_{\mathrm{dir}}^T\left(\sum_{j=1}^{d}
\E_{H(j,\bc)}\!\bigl[
Z_F(X,\boeta_0)\otimes Z_F(X,\boeta_0)
\bigr]\right)
.
\end{align*}
Thus
\begin{align*}
\IFu_{\cell}(\bc,\sigma_{\theta},H_0)
&=
\frac{1}{2\,\sigma_{\theta,0}}
\left.\frac{\partial}{\partial\eps}\sigma_{\theta}^2(H(G_\eps^I,\bc))\right|_{\eps=0}
\\
&=\;
dm_{s,1}
+\bmed_{s,2}^T\,\IFu_{\cell}(\bc,\boeta,H_0)
+d\bmed_{s,3}^T\left( \sum_{j=1}^{d}\E_{H(j,\bc)}\!\big[Z_F(X,\boeta_0)\big]\right)
\\
&+d\bmed_{s,4}^T\left(\sum_{j=1}^{d}
\E_{H(j,\bc)}\!\bigl[
Z_F(X,\boeta_0)\otimes Z_F(X,\boeta_0)
\bigr]\right).
\end{align*}

Combining these results with \eqref{eq:IF_lproof}, 
\eqref{eq:IF_mproof}, \eqref{eq:IF_lproof_cell}, 
and \eqref{eq:IF_mproof_cell} completes the proof of the proposition.
\end{proof}

\newpage
\section{Additional simulation results}
\label{app:addsim}

Figure~\ref{fig:results_NA1_perout2_reg}  in the 
main text shows the average MSE  attained by 
RIDGE, SEST, PENSE, CRM, REGCELL, SHOOT, and 
cellMR in the presence of cellwise outliers, 
casewise outliers, or both, without missing data.
When we also set 10\% of randomly selected cells
to NA we obtain 
Figure~\ref{fig:results_NA2_perout2_reg}.
Surprisingly, the curves of cellMR are almost the
same as in Figure~\ref{fig:results_NA1_perout2_reg}.
Of the other methods, only RIDGE is able to handle
missing values. Its curves are similar to before
in the setting $p = q = 10$, but worsen in the 
higher dimensions.

\begin{figure}[H]
\centering
\begin{tabular}{M{0.0005\textwidth}M{0.29\textwidth}M{0.29\textwidth}M{0.32\textwidth}}
   &\large \textbf{Cellwise}  & \large \textbf{Casewise} &\large{\textbf{Casewise \& Cellwise}} \\
[-4mm]

 \rotatebox{90}{\textbf{\footnotesize{$p=q=10$}}}&\includegraphics[width=.31\textwidth]
 {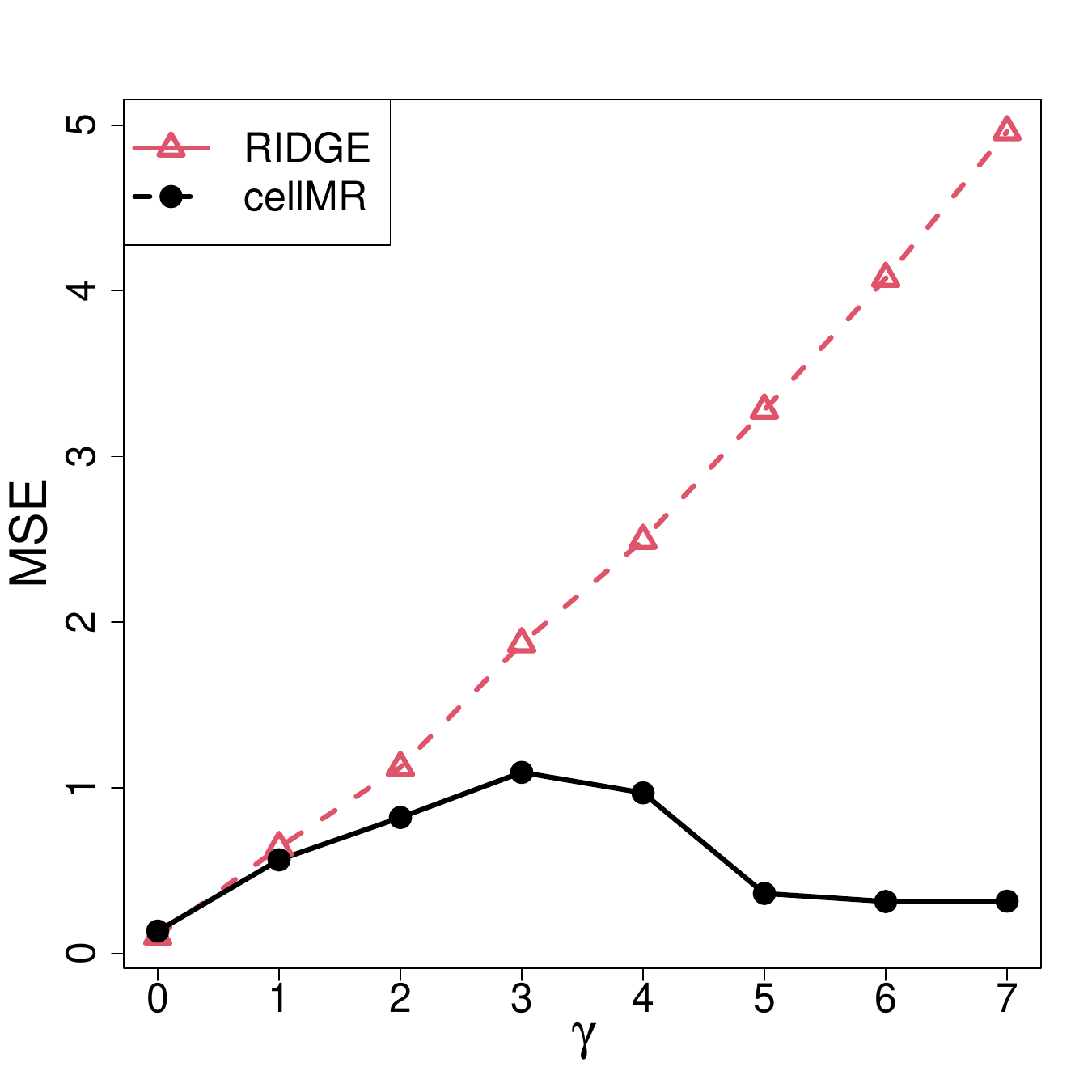} &\includegraphics[width=.31\textwidth]{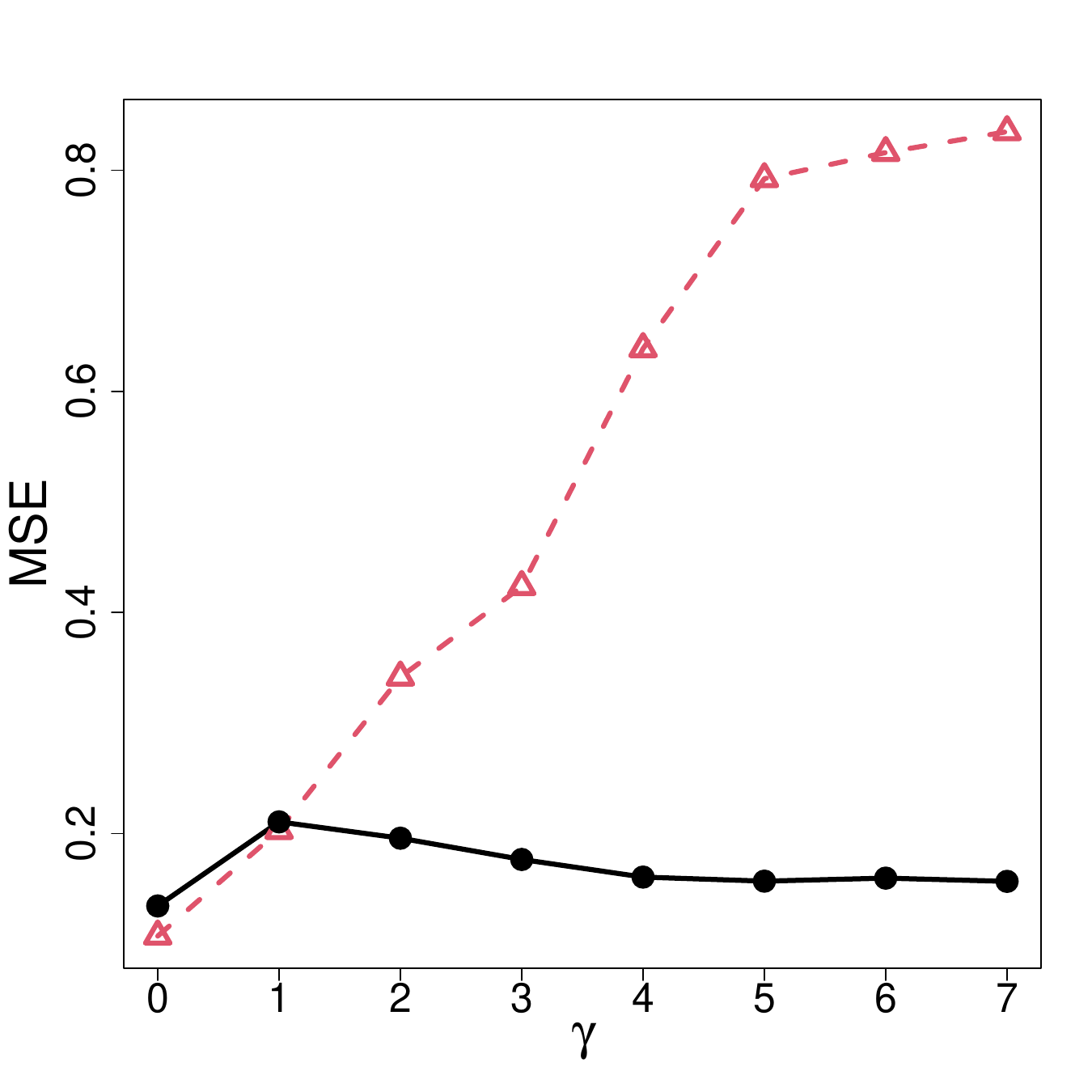} &\includegraphics[width=.31\textwidth]{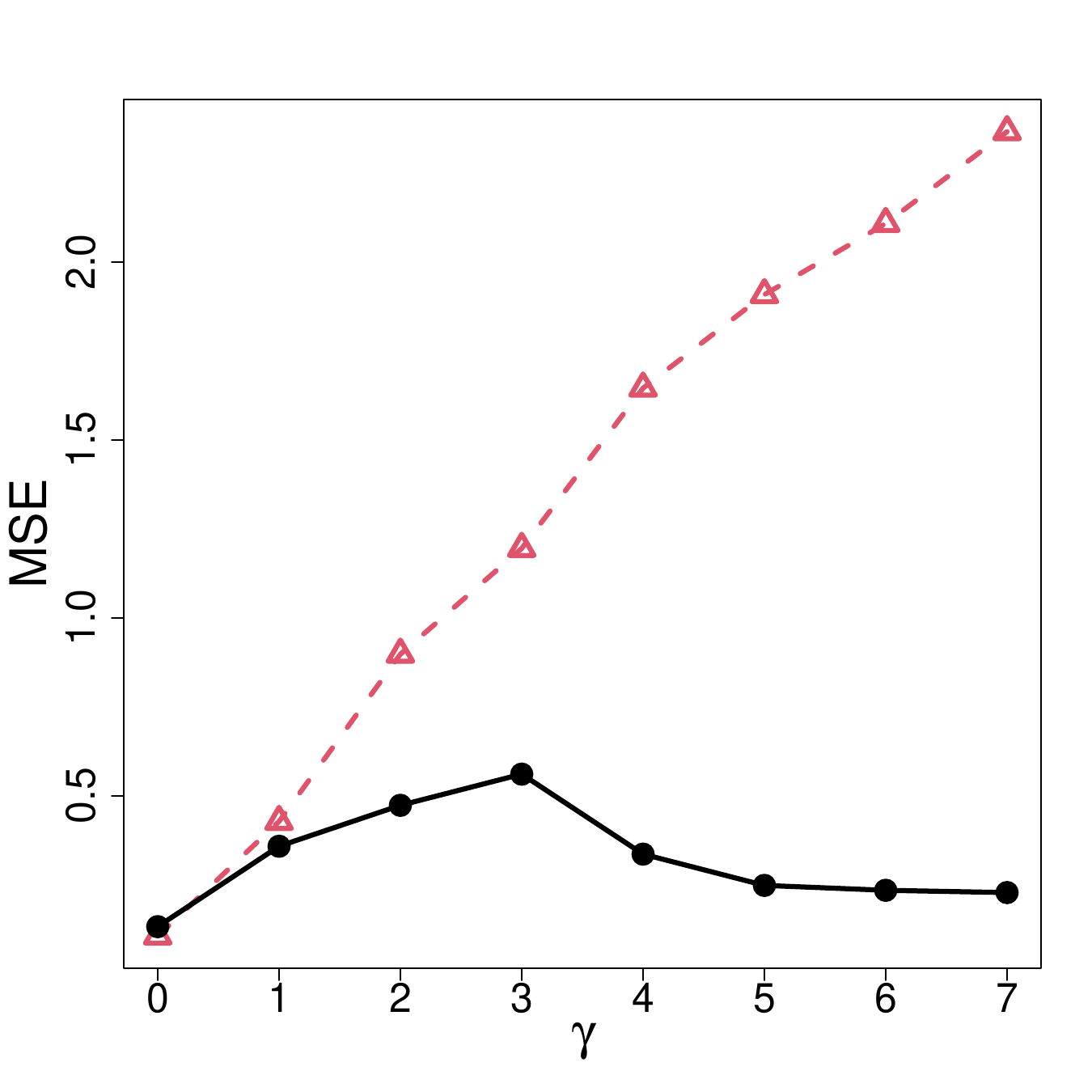} \\
 
 \rotatebox{90}{\textbf{\footnotesize{$p=q=25$}}}&\includegraphics[width=.31\textwidth]{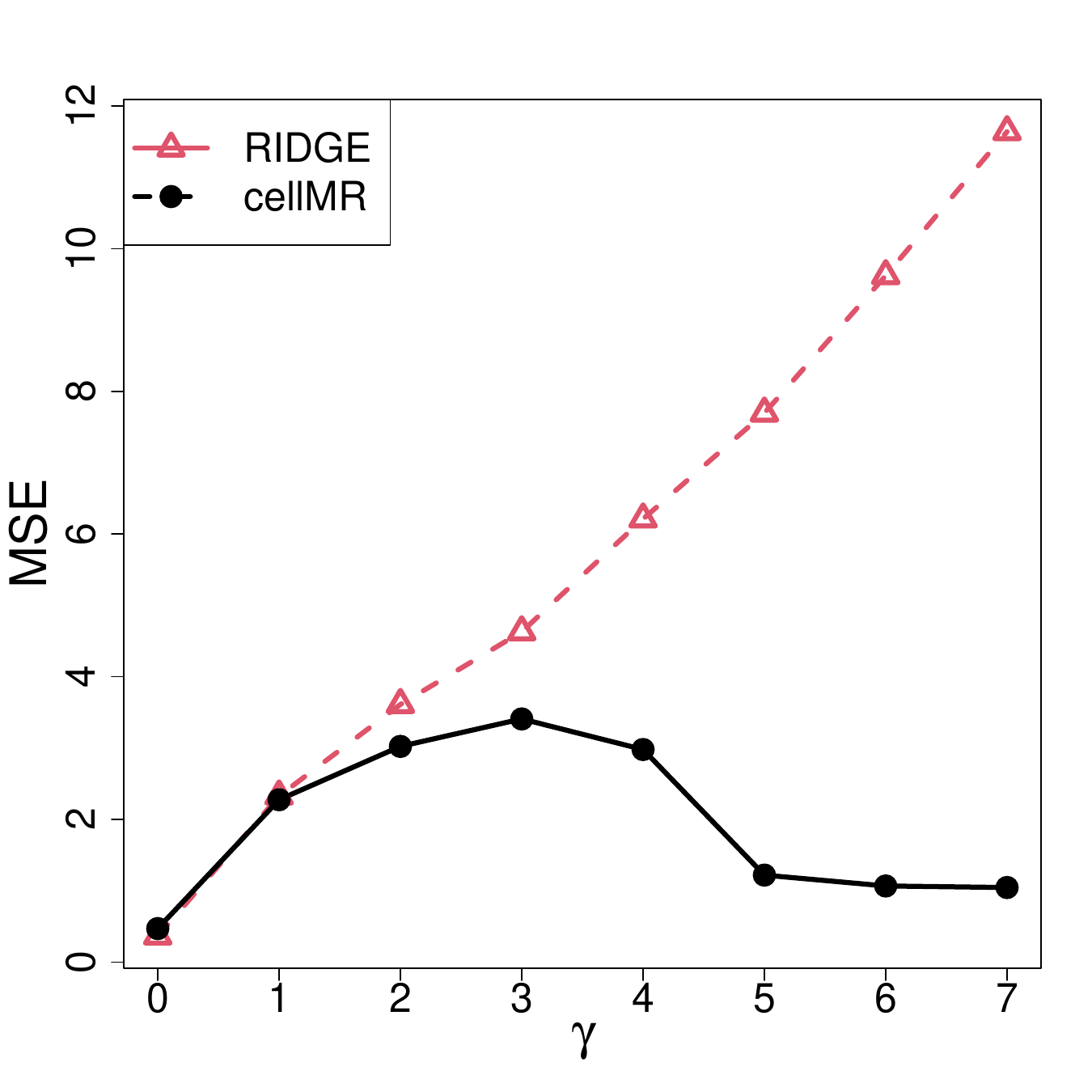} &\includegraphics[width=.31\textwidth]{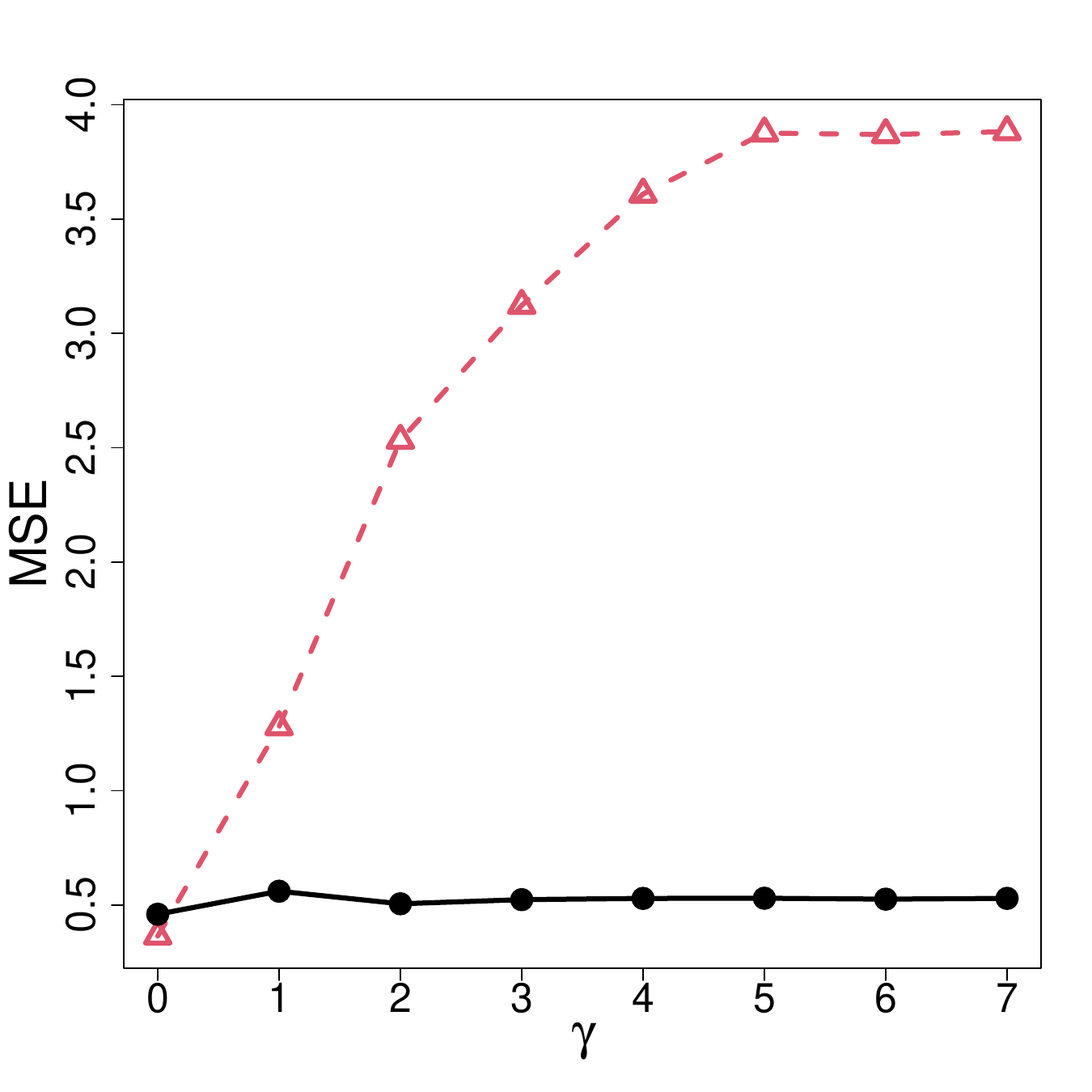} &\includegraphics[width=.31\textwidth]{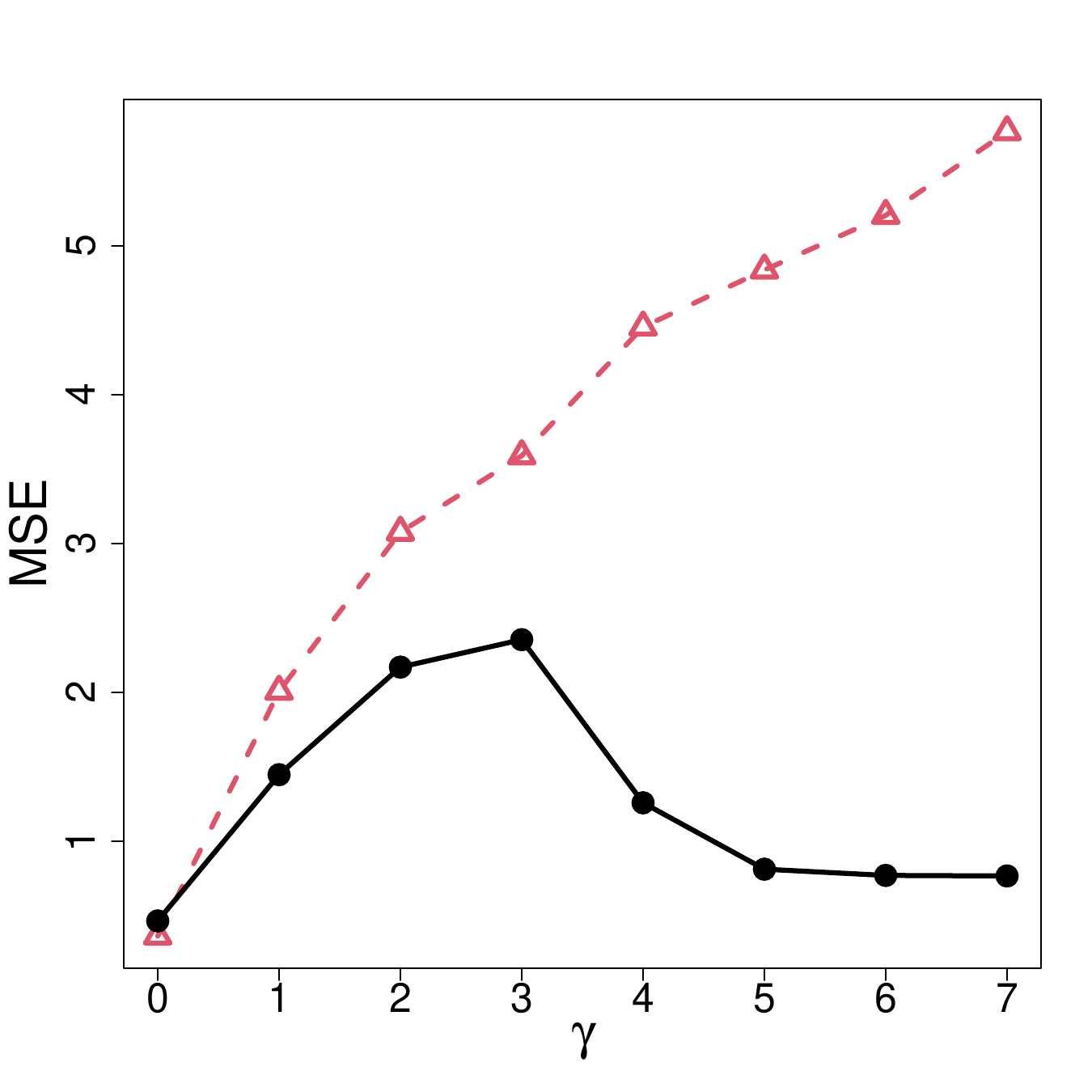}\\
 \rotatebox{90}{\textbf{\footnotesize{$p=q=50$}}}&\includegraphics[width=.31\textwidth]{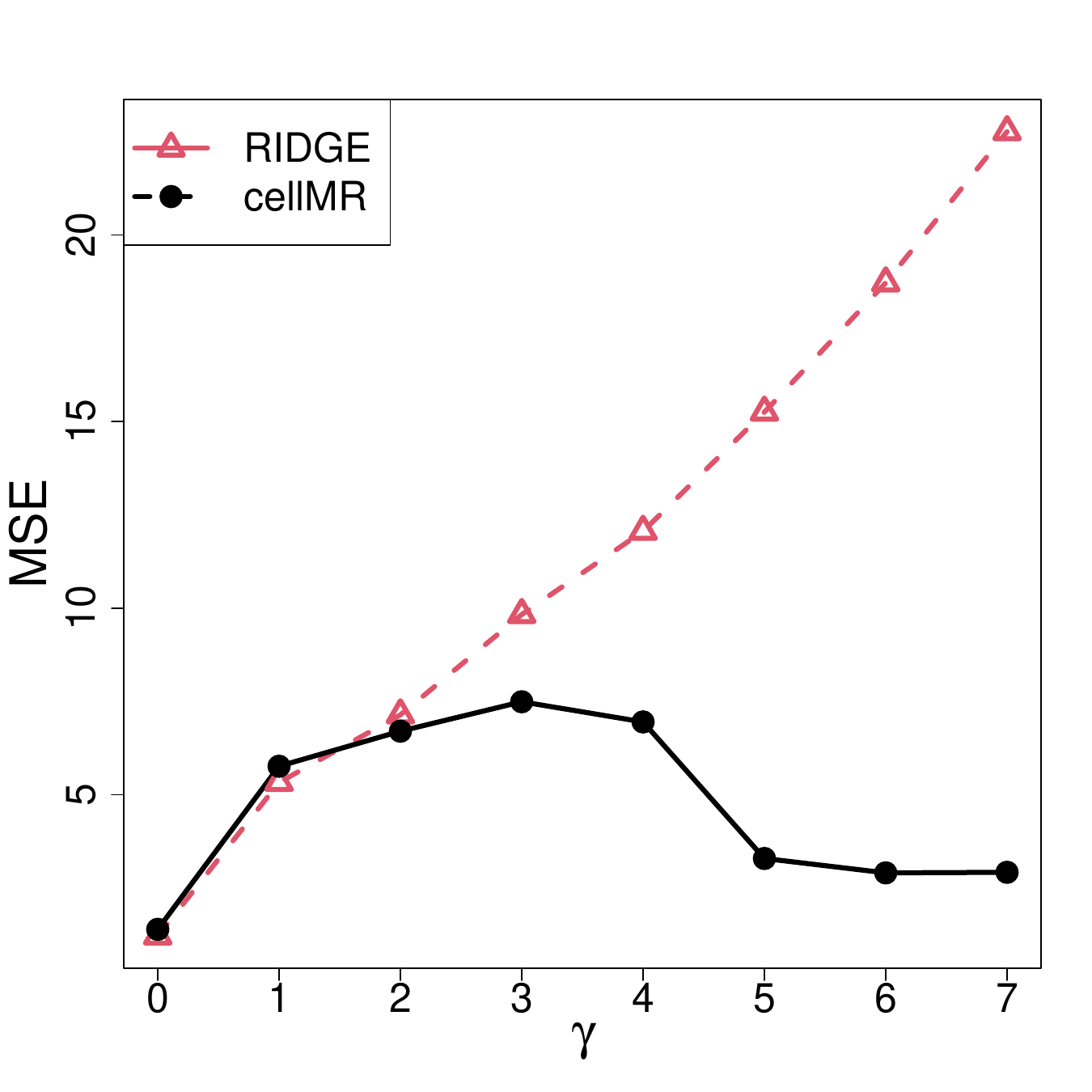} &\includegraphics[width=.31\textwidth]{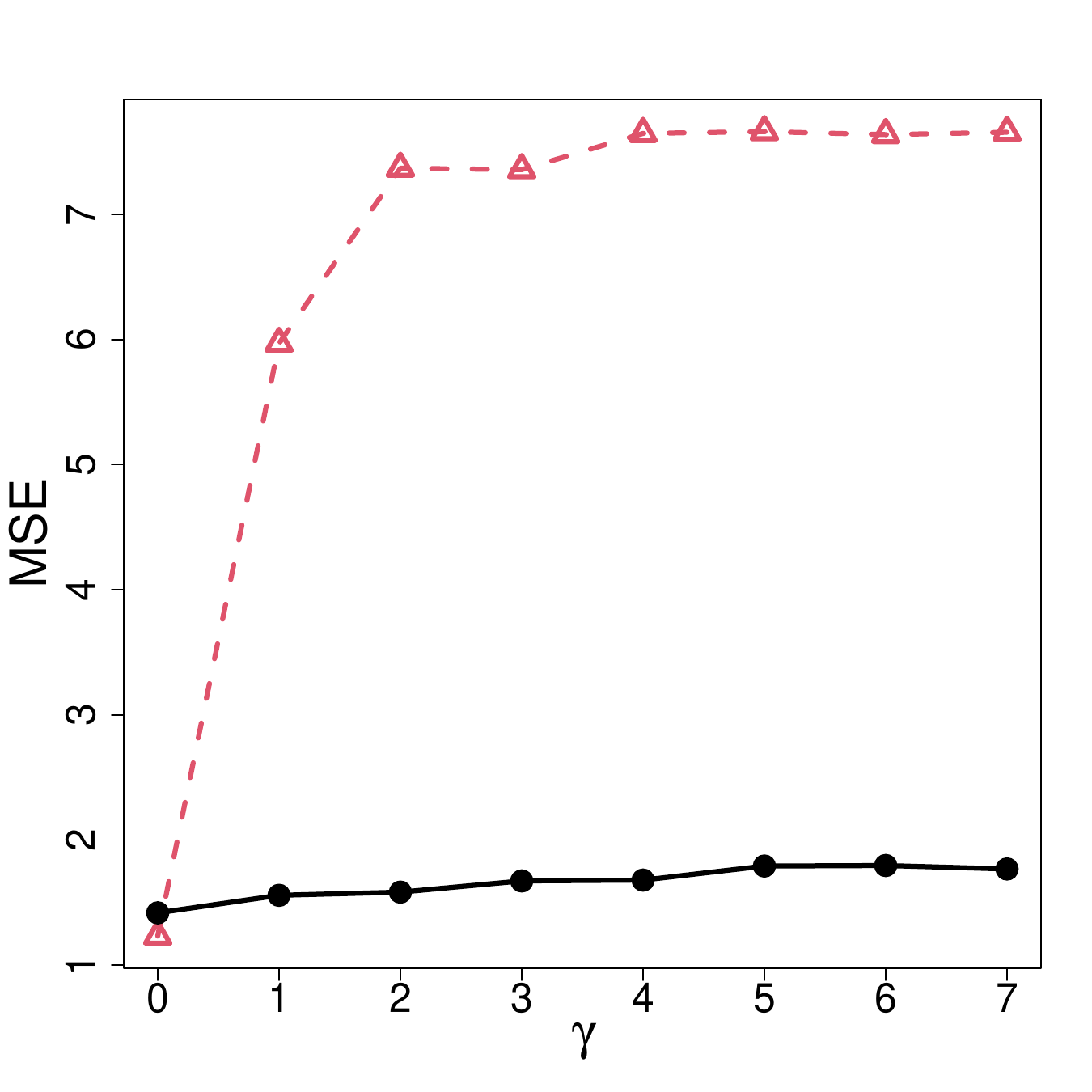} &\includegraphics[width=.31\textwidth]{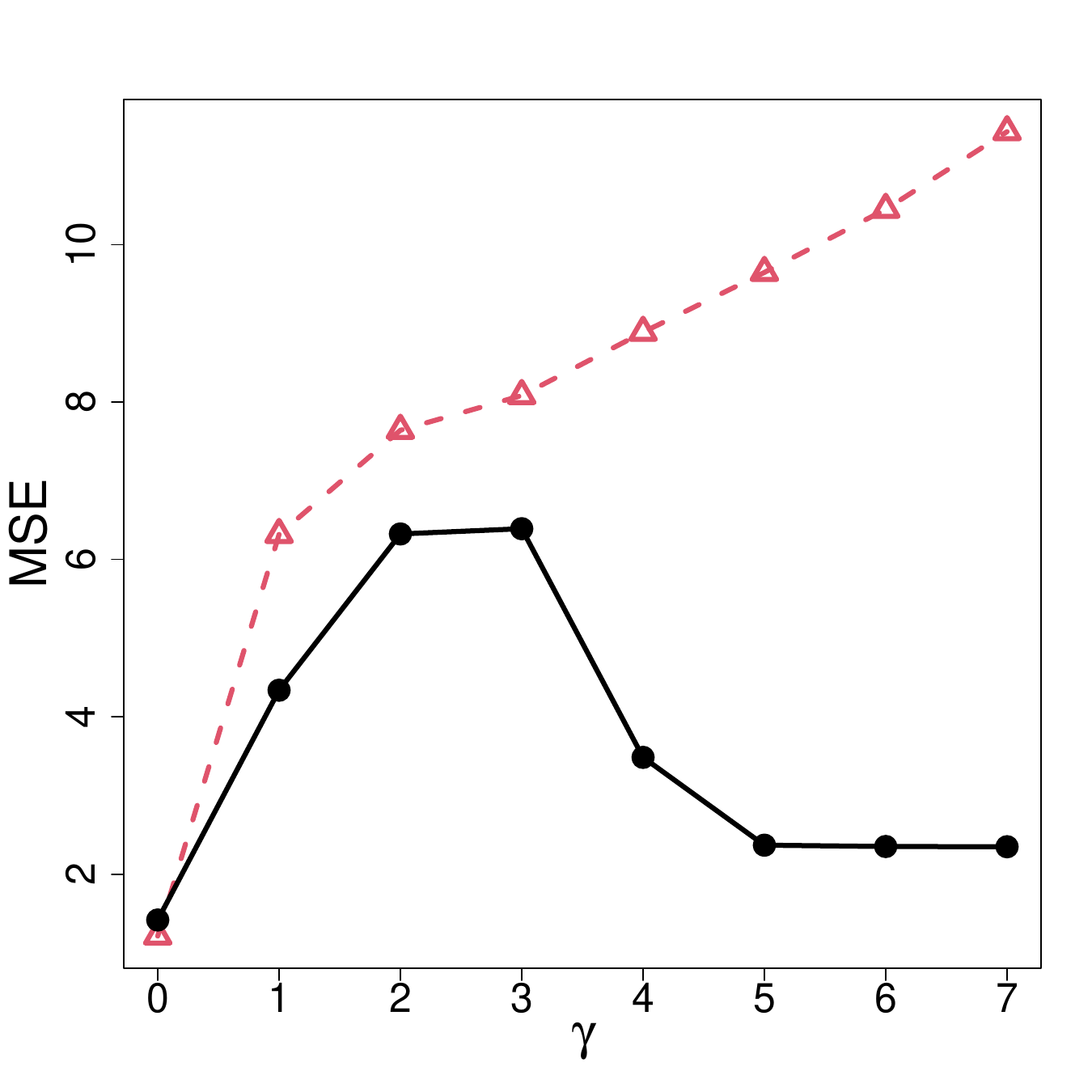}\\
\end{tabular}

\vspace{-5mm}
\caption{Average MSE attained by RIDGE and cellMR
in the presence of cellwise outliers, casewise 
outliers, or both, with 10\% of missing cells.}
\label{fig:results_NA2_perout2_reg}
\end{figure}

Figure~\ref{fig:results_NA2_perout1_coverage}
shows the empirical coverage of the nominally
90\% coverage intervals of the regression
coefficients, as obtained by OLS, FRB, Boot-cellMR and 
cellBoot, for $\eps = 10\%$ and 10\% of 
missing cells. (Note that FRB crashed for 
$p = q \geqslant 25$.) Also here the resulting 
curves closely resemble those obtained without 
missing data, that were shown in 
Section~\ref{sec:siminf}.

\vspace{5mm}
\begin{figure}[H]
\centering
\begin{tabular}{M{0.0005\textwidth}M{0.29\textwidth}M{0.29\textwidth}M{0.32\textwidth}}
   &\large \textbf{Cellwise}  & \large \textbf{Casewise} &\large{\textbf{Casewise \& Cellwise}} \\
[-4mm]

 \rotatebox{90}{\textbf{\footnotesize{$p=q=10$}}}&\includegraphics[width=.31\textwidth]{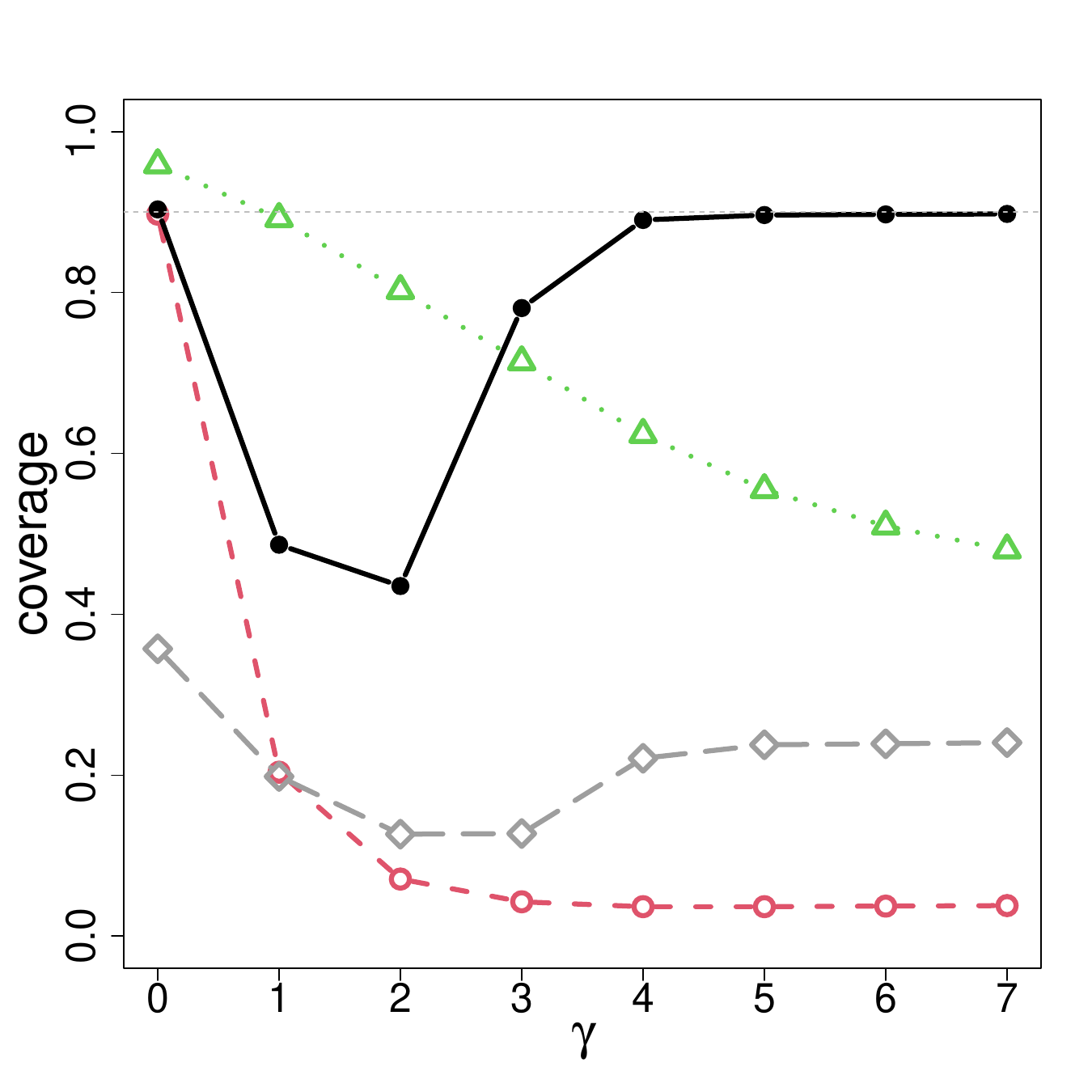} &\includegraphics[width=.31\textwidth]{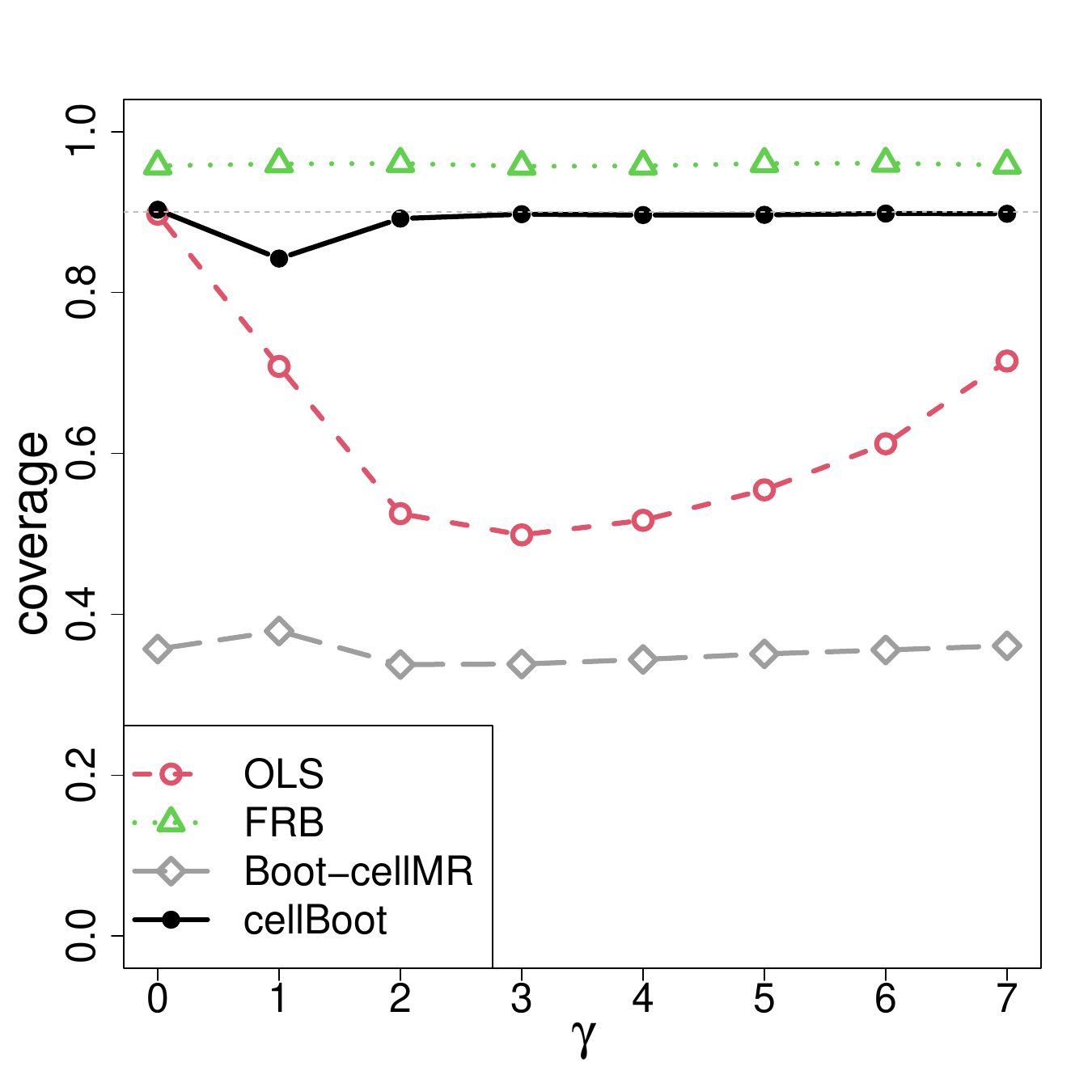} &\includegraphics[width=.31\textwidth]{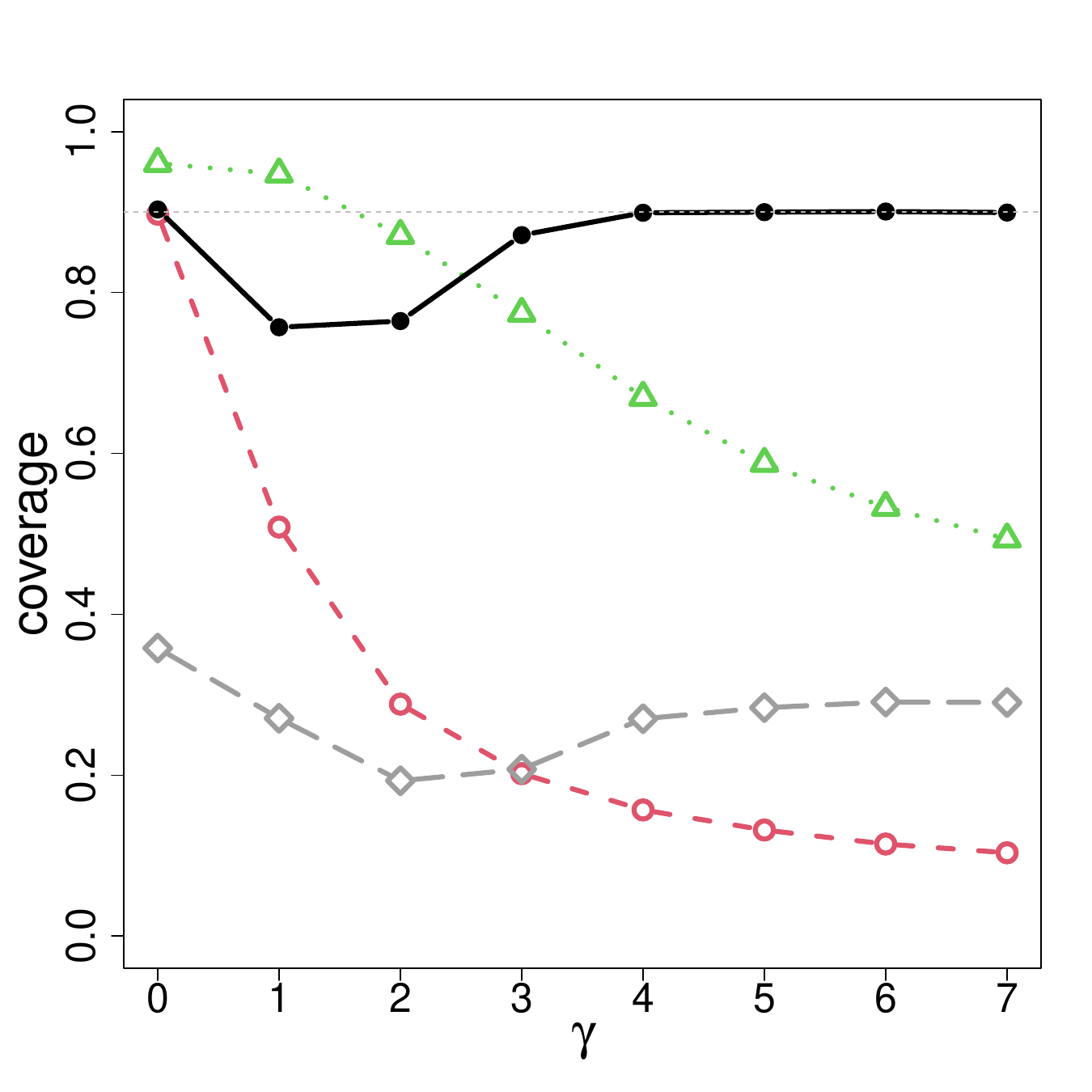} \\

 \rotatebox{90}{\textbf{\footnotesize{$p=q=25$}}}&\includegraphics[width=.31\textwidth]{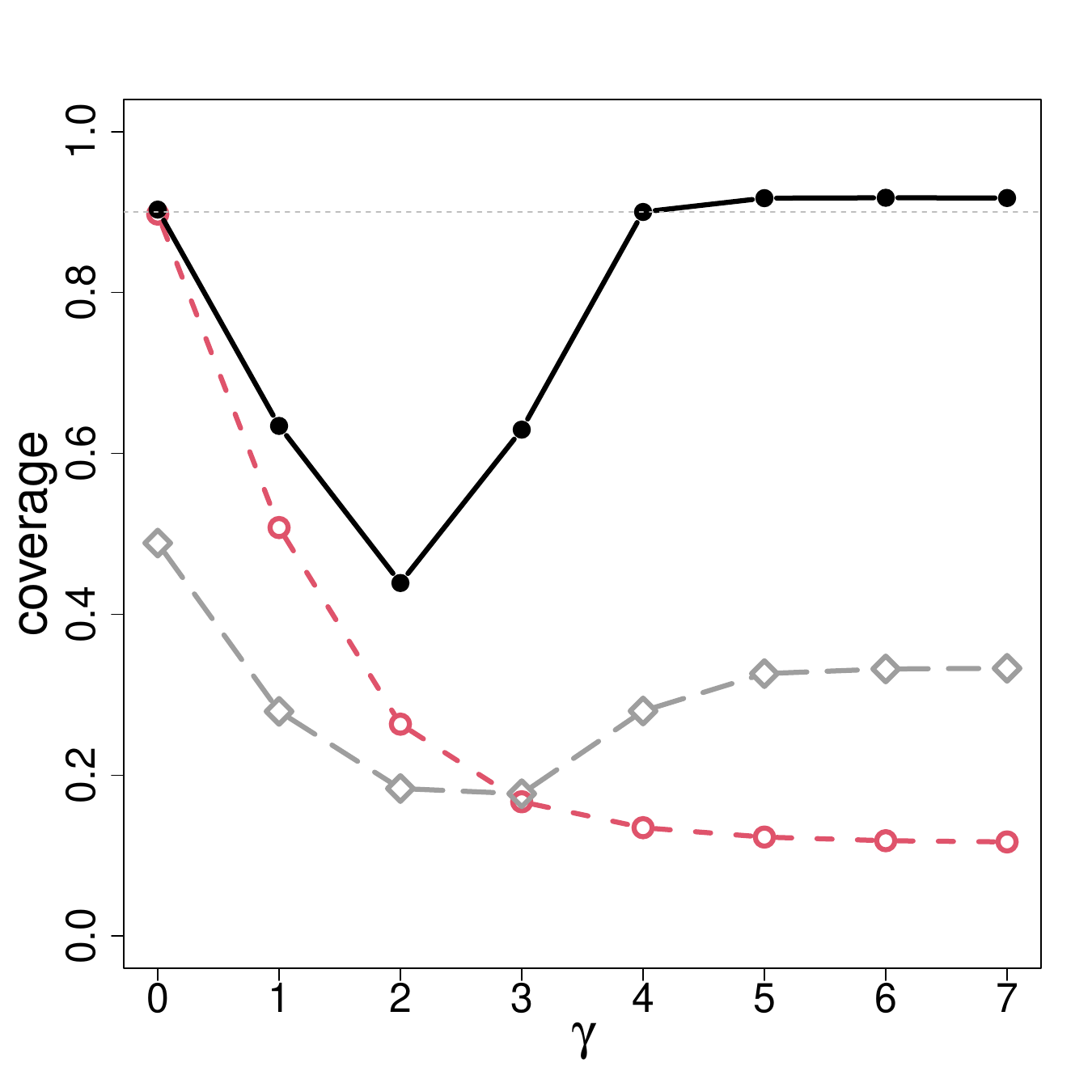} &\includegraphics[width=.31\textwidth]{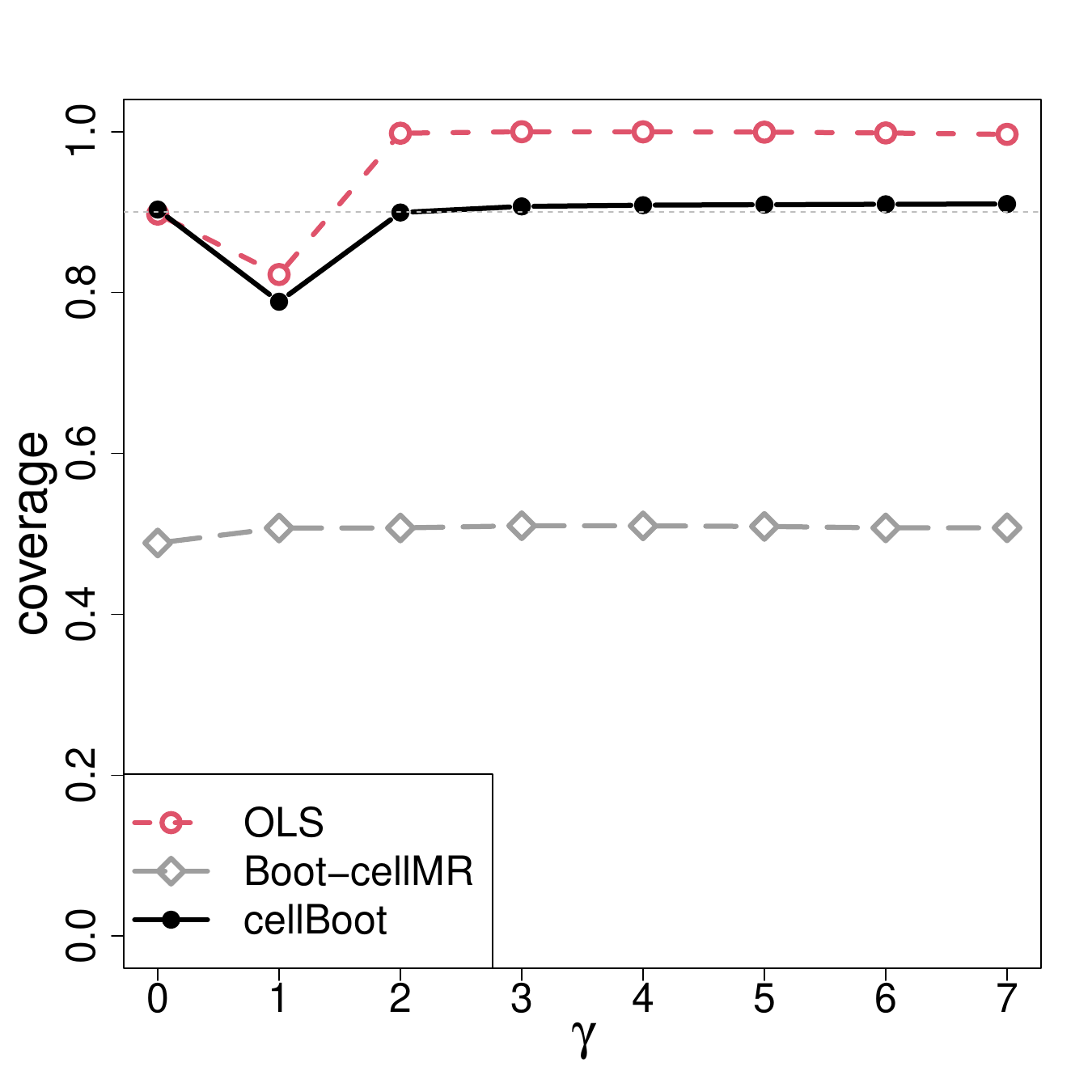} &\includegraphics[width=.31\textwidth]{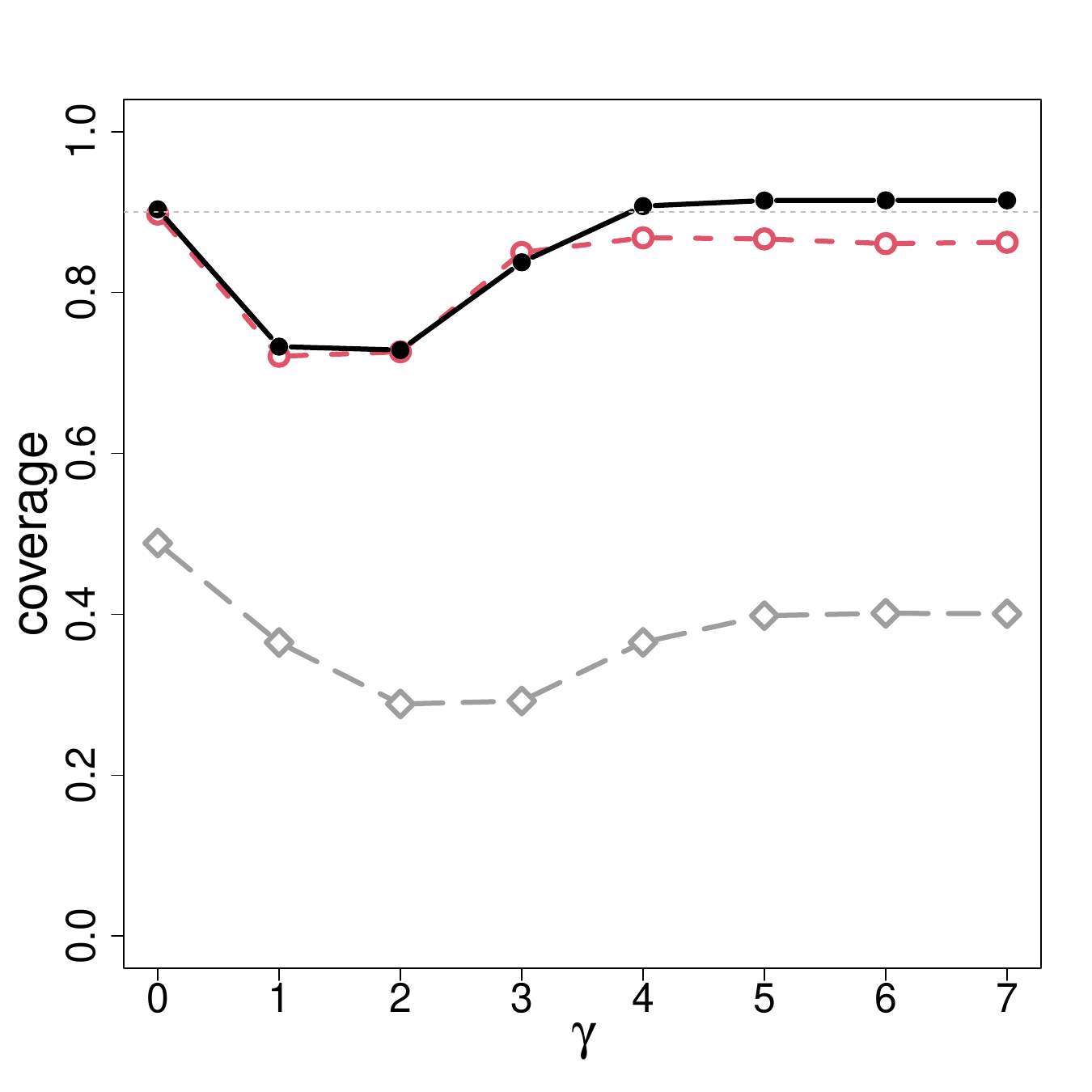}\\
 \rotatebox{90}{\textbf{\footnotesize{$p=q=50$}}}&\includegraphics[width=.31\textwidth]{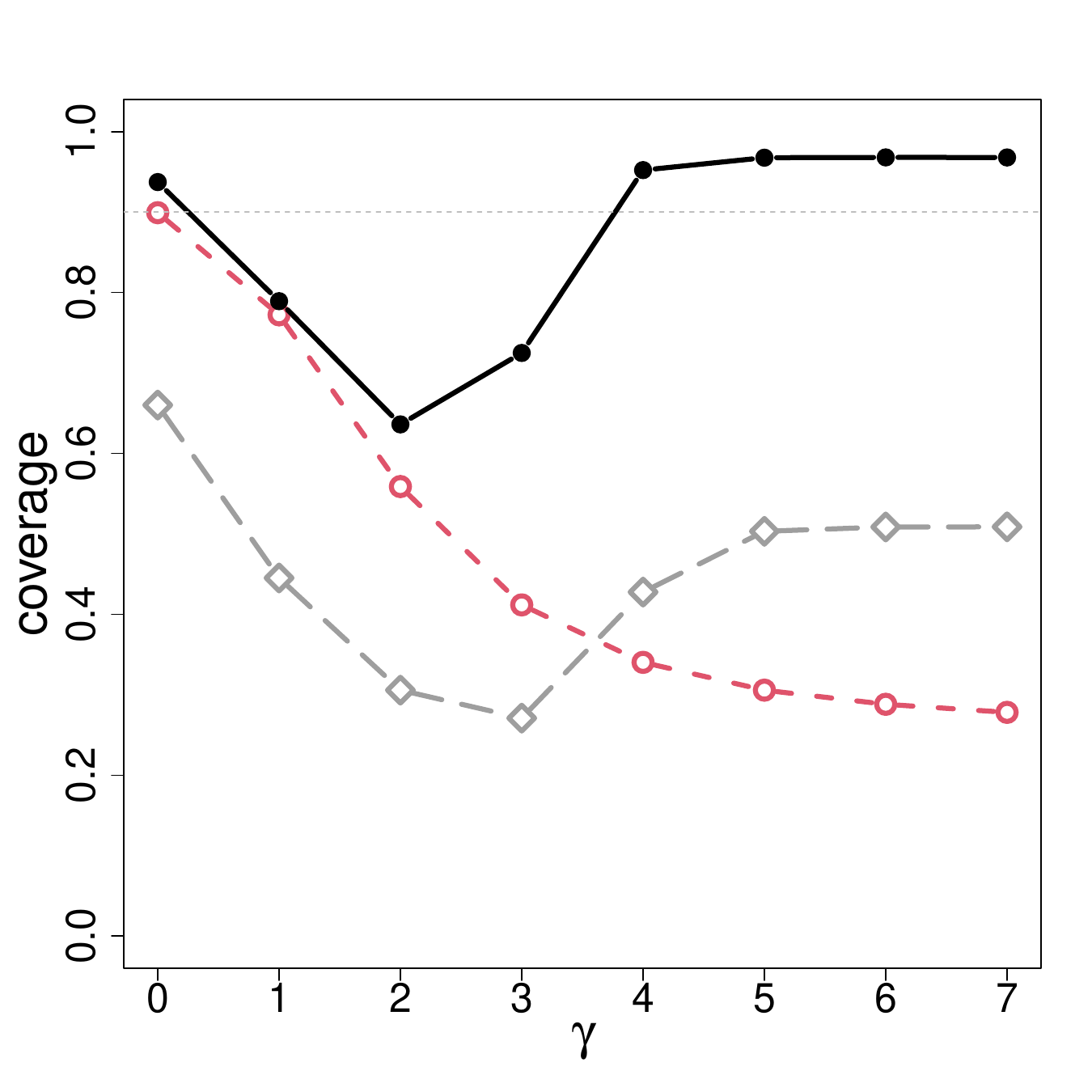} &\includegraphics[width=.31\textwidth]{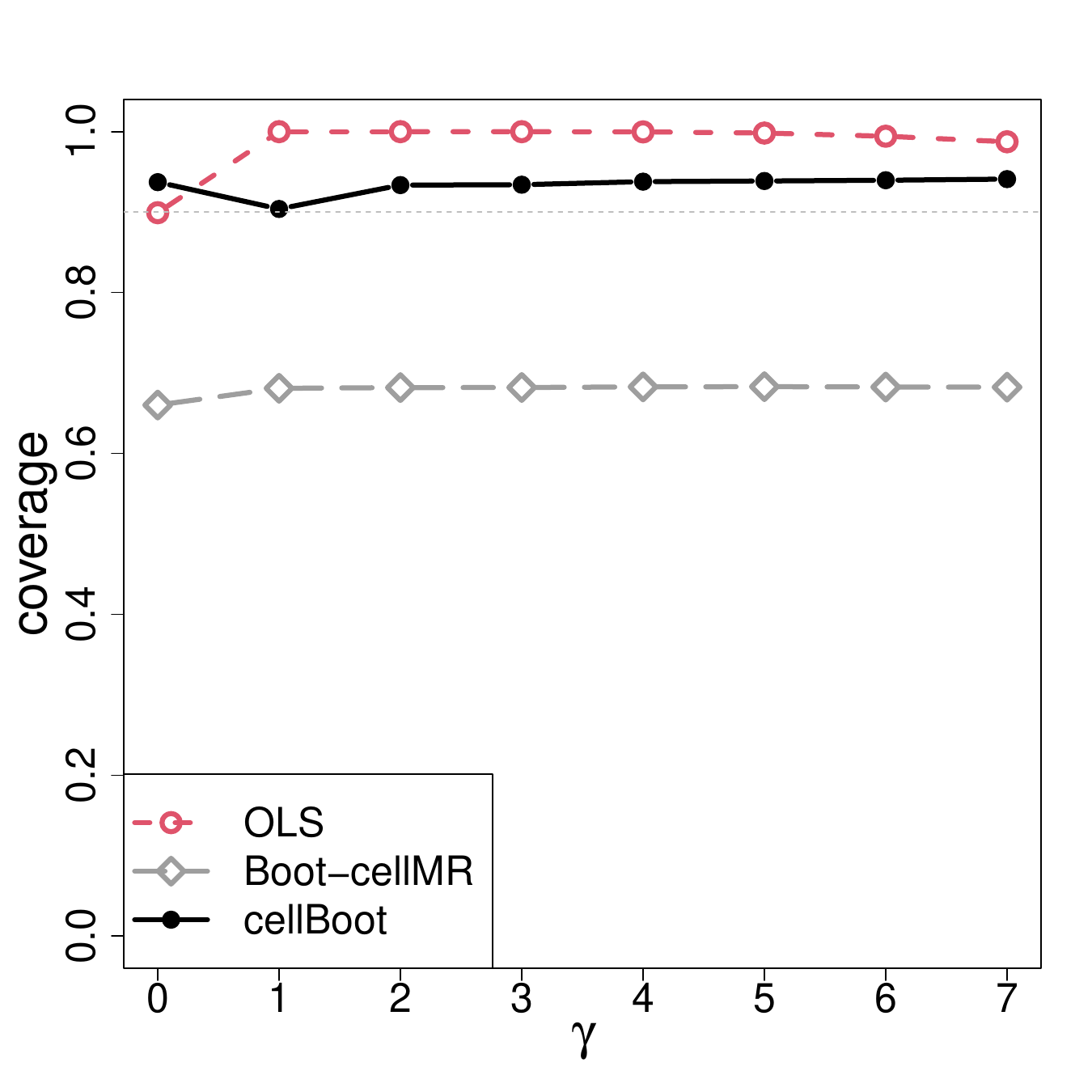} &\includegraphics[width=.31\textwidth]{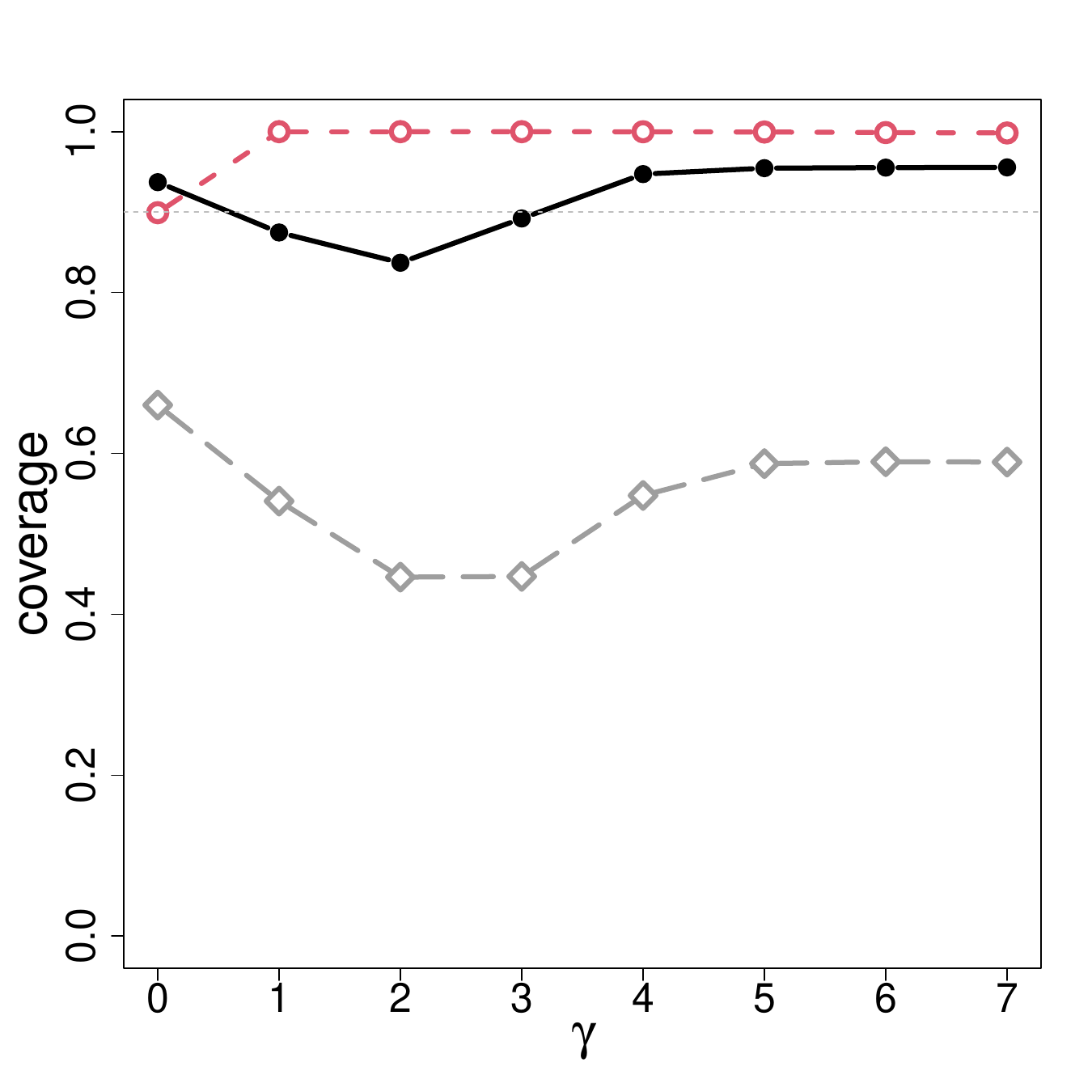}\\
\end{tabular}

\vspace{-4mm}
\caption{Average coverage  attained by OLS, 
FRB, Boot-cellMR, and cellBoot for the $0.9$-level
confidence intervals of the regression 
coefficients in the presence of cellwise outliers,
casewise outliers, or both, with 10\% of missing 
cells.}
\label{fig:results_NA2_perout1_coverage}
\end{figure}

\newpage
To study the sensitivity of cellMR and 
competing methods to the contamination fraction, 
Figure~\ref{fig:results_varyingeps_reg}
shows the MSE of the prediction as a function of
$\eps$ for data contaminated as in
Section~\ref{sec:simreg}, with $\gamma=7$ and
$p=q=10,25,50$. As expected, increasing $\eps$
generally leads to larger MSEs.
Under cellwise contamination, cellMR achieves 
the lowest MSE for $\eps\leqslant 0.3$ across 
all dimensions considered, but the curve
increases sharply at $\eps=0.4$.
Under casewise contamination, cellMR still 
performs very well at $\eps=0.4$. When both 
types of contamination are present, cellMR 
achieves the lowest MSE across all 
contamination levels and dimensions.
\begin{figure}[!ht]
\centering
\begin{tabular}{M{0.0005\textwidth}M{0.29\textwidth}M{0.29\textwidth}M{0.32\textwidth}}
   &\large \textbf{Cellwise}  & \large \textbf{Casewise} &\large{\textbf{Casewise \& Cellwise}} \\
[-4mm]

 \rotatebox{90}{\textbf{\footnotesize{$p=q=10$}}}&\includegraphics[width=.31\textwidth]{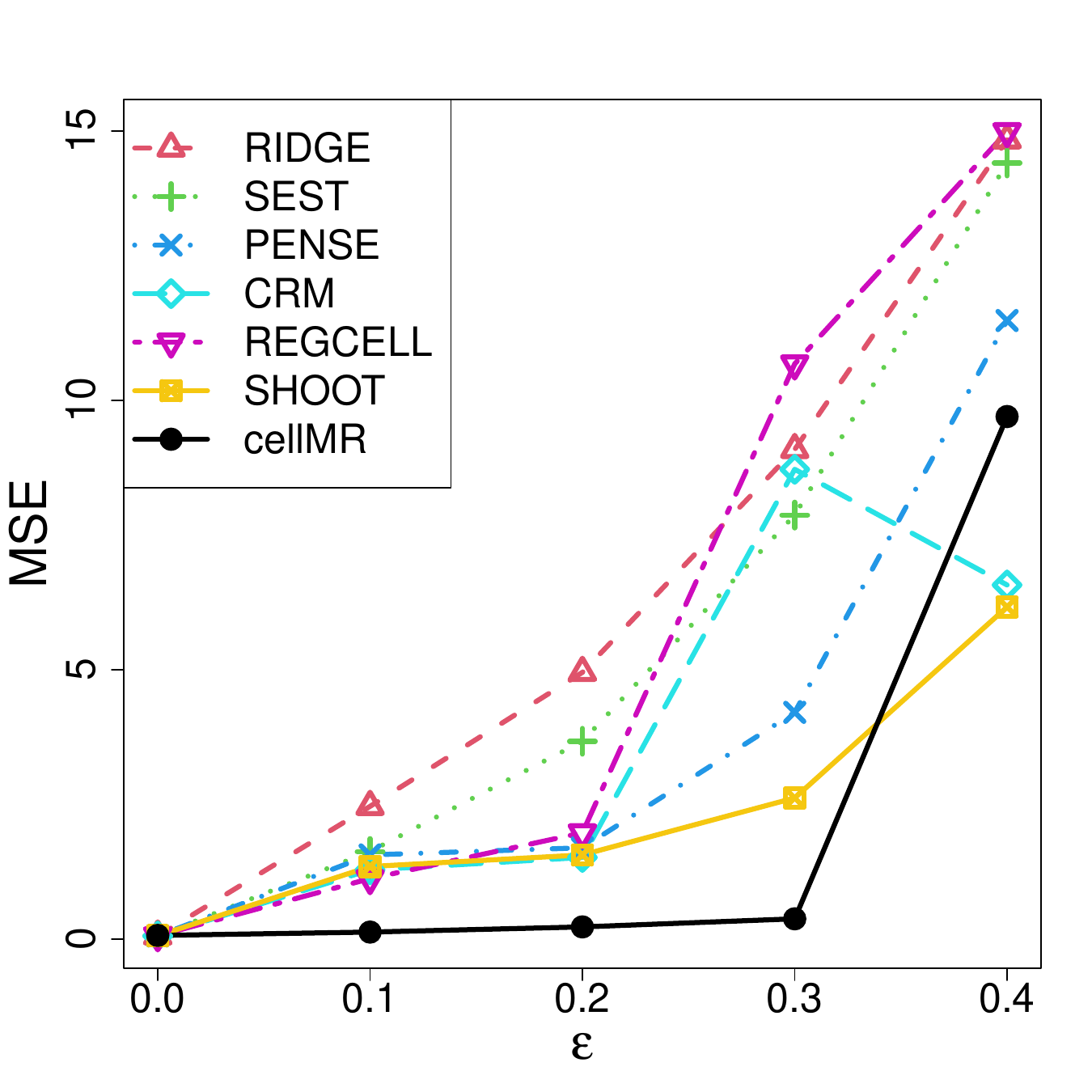} &\includegraphics[width=.31\textwidth]{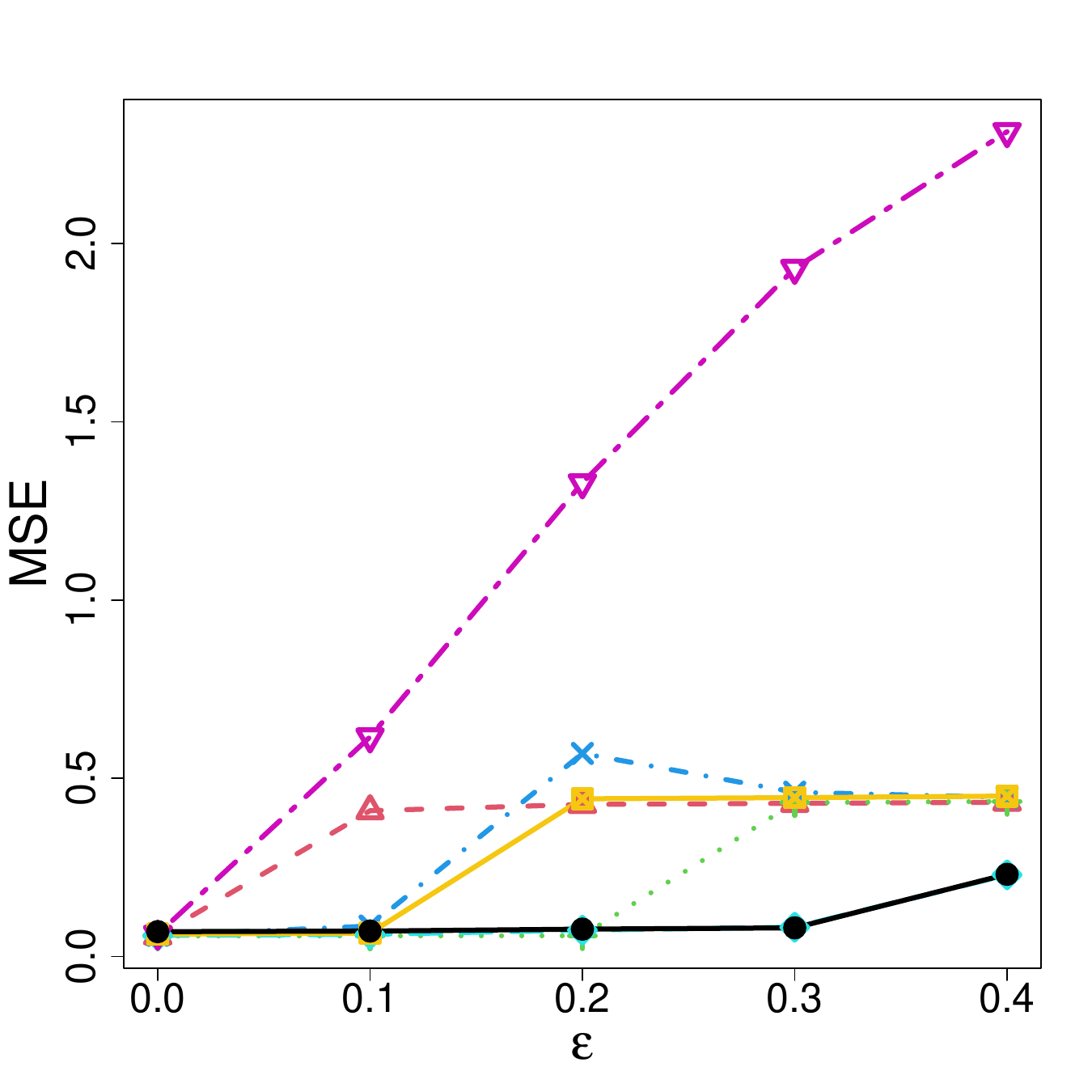} &\includegraphics[width=.31\textwidth]{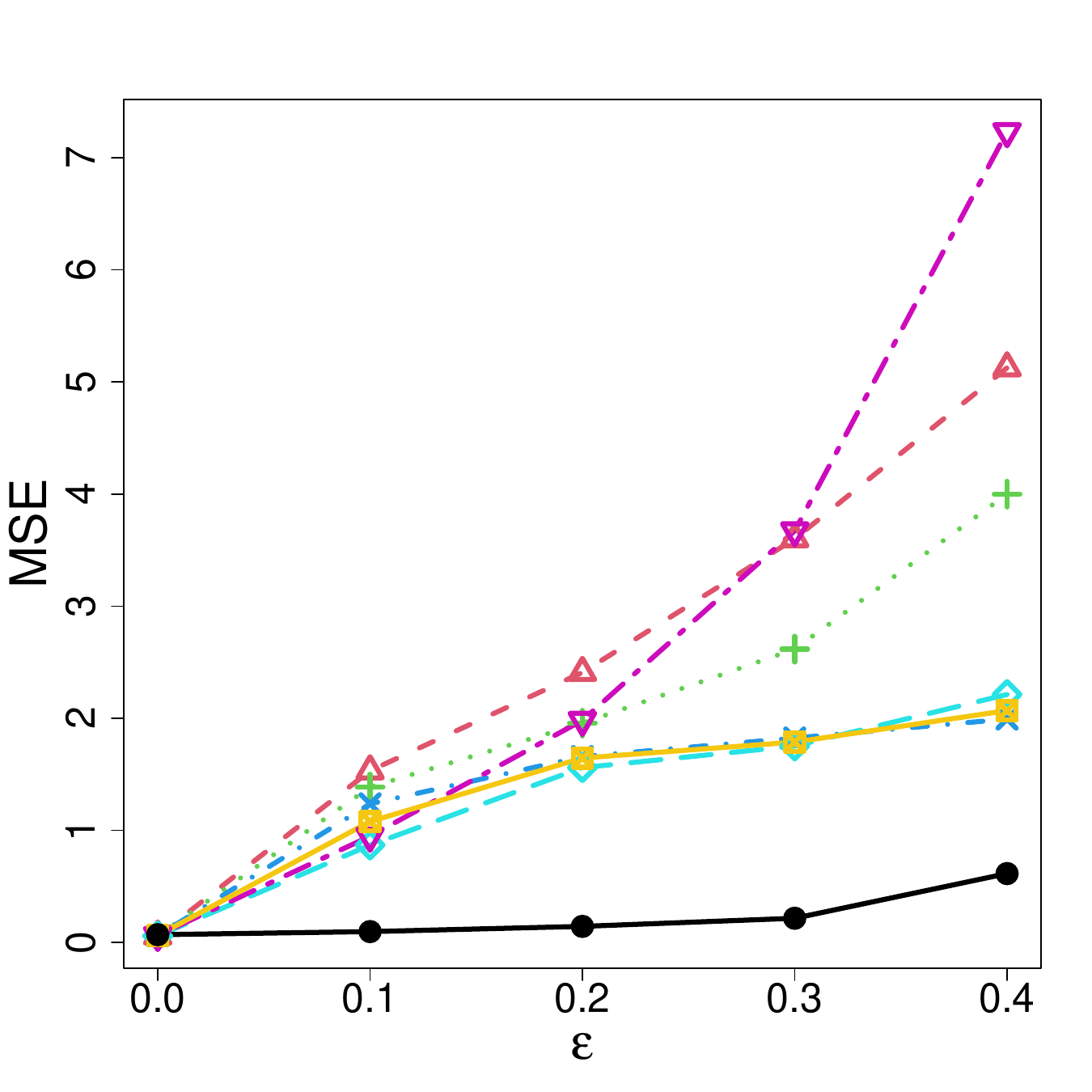} \\

 \rotatebox{90}{\textbf{\footnotesize{$p=q=25$}}}&\includegraphics[width=.31\textwidth]{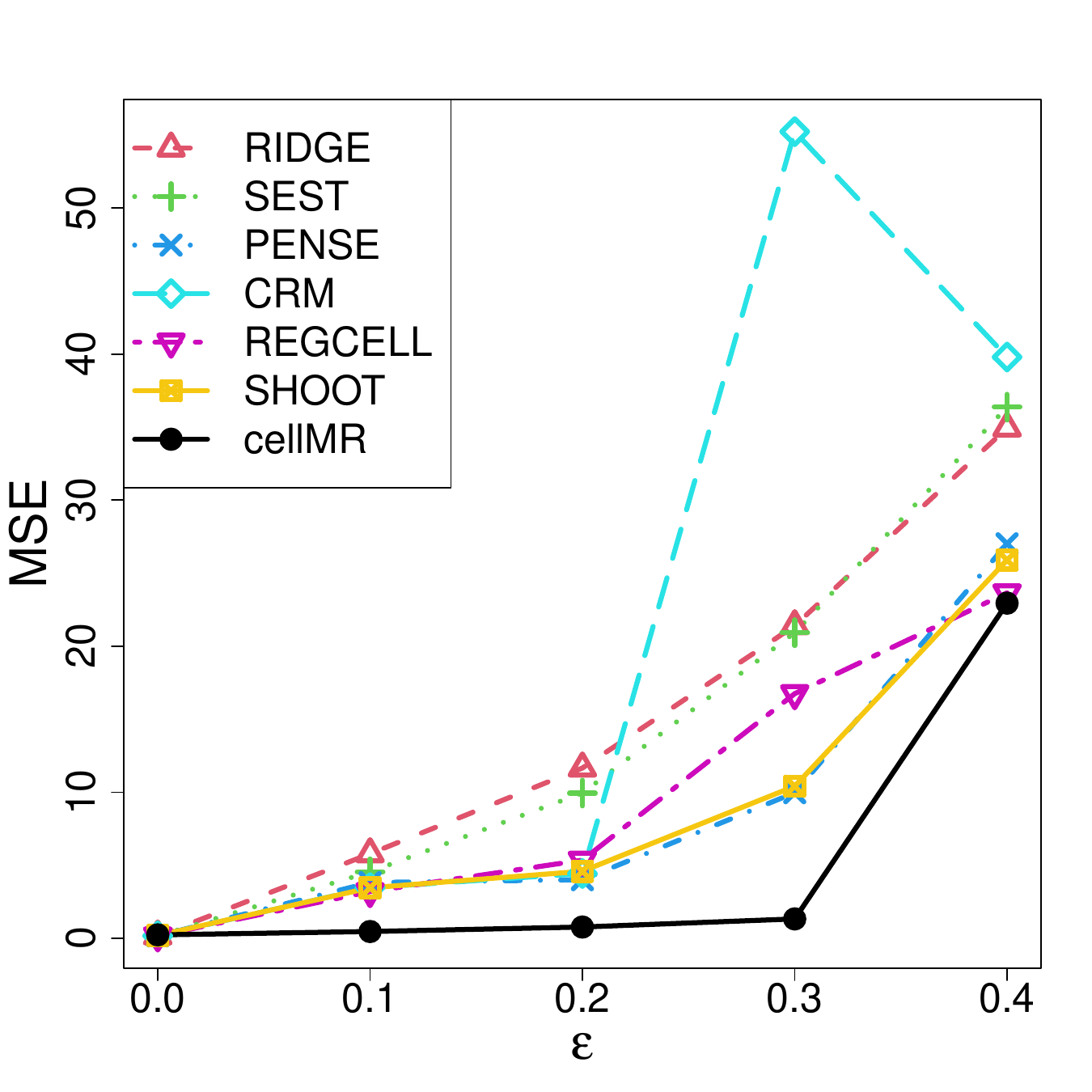} &\includegraphics[width=.31\textwidth]{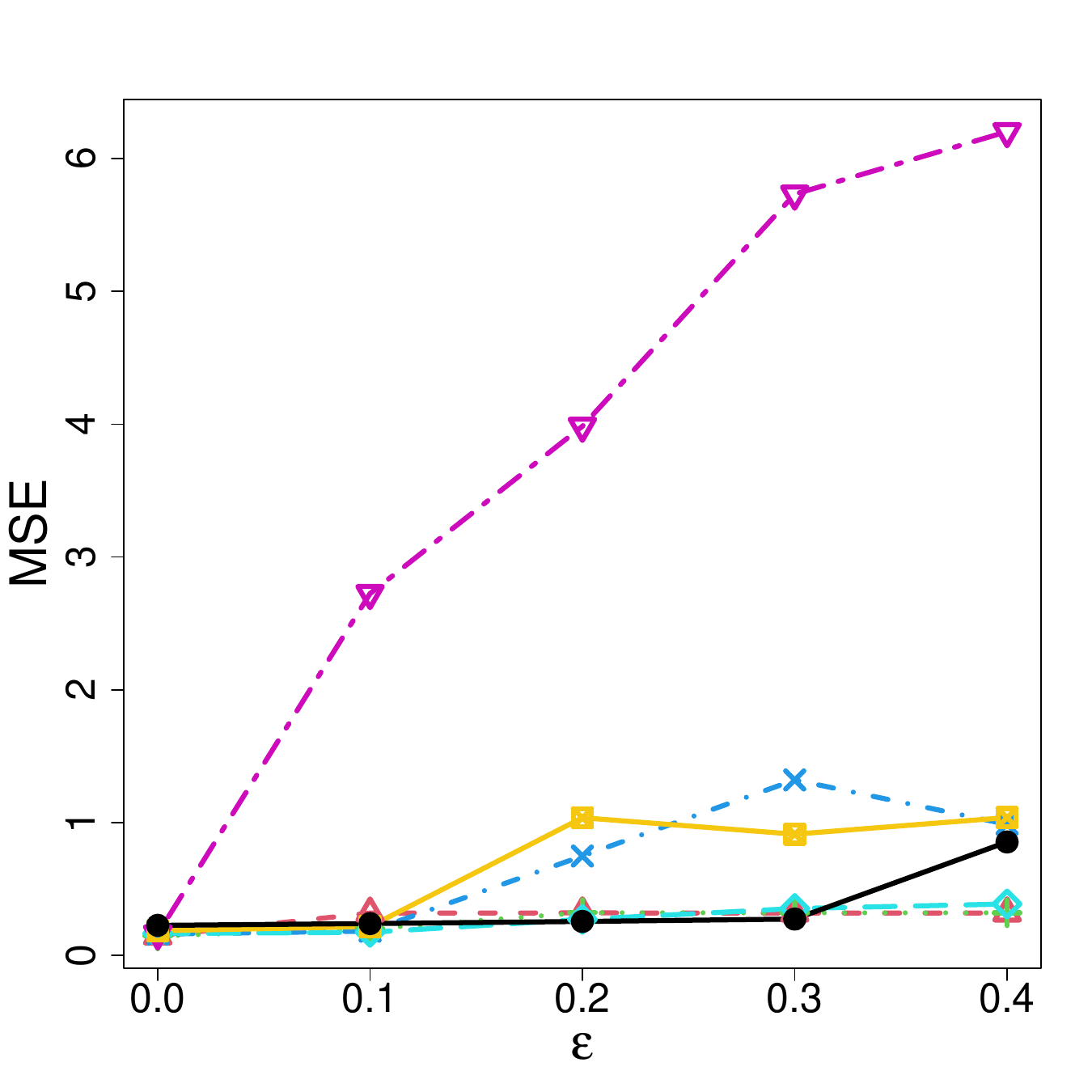} &\includegraphics[width=.31\textwidth]{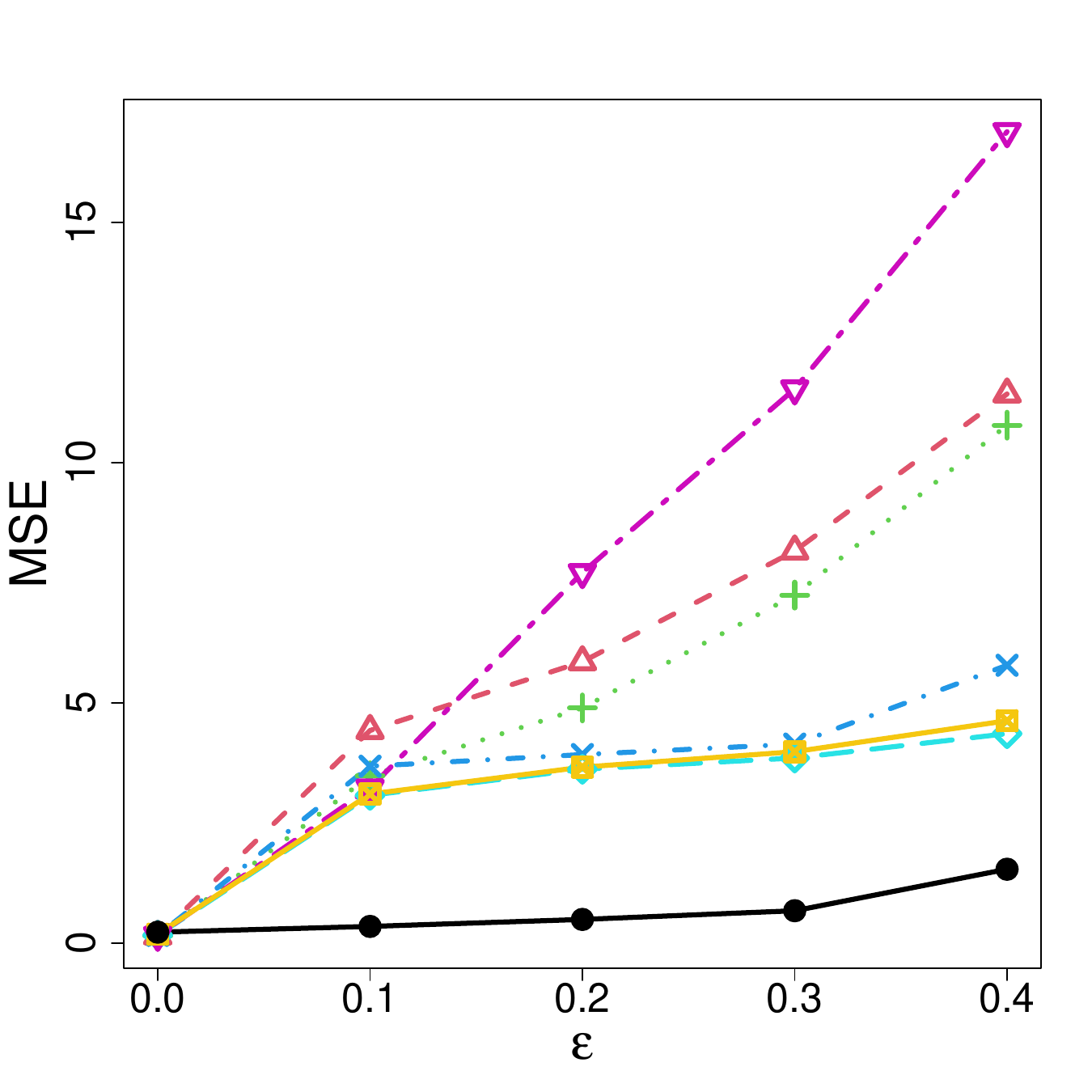}\\
 \rotatebox{90}{\textbf{\footnotesize{$p=q=50$}}}&\includegraphics[width=.31\textwidth]{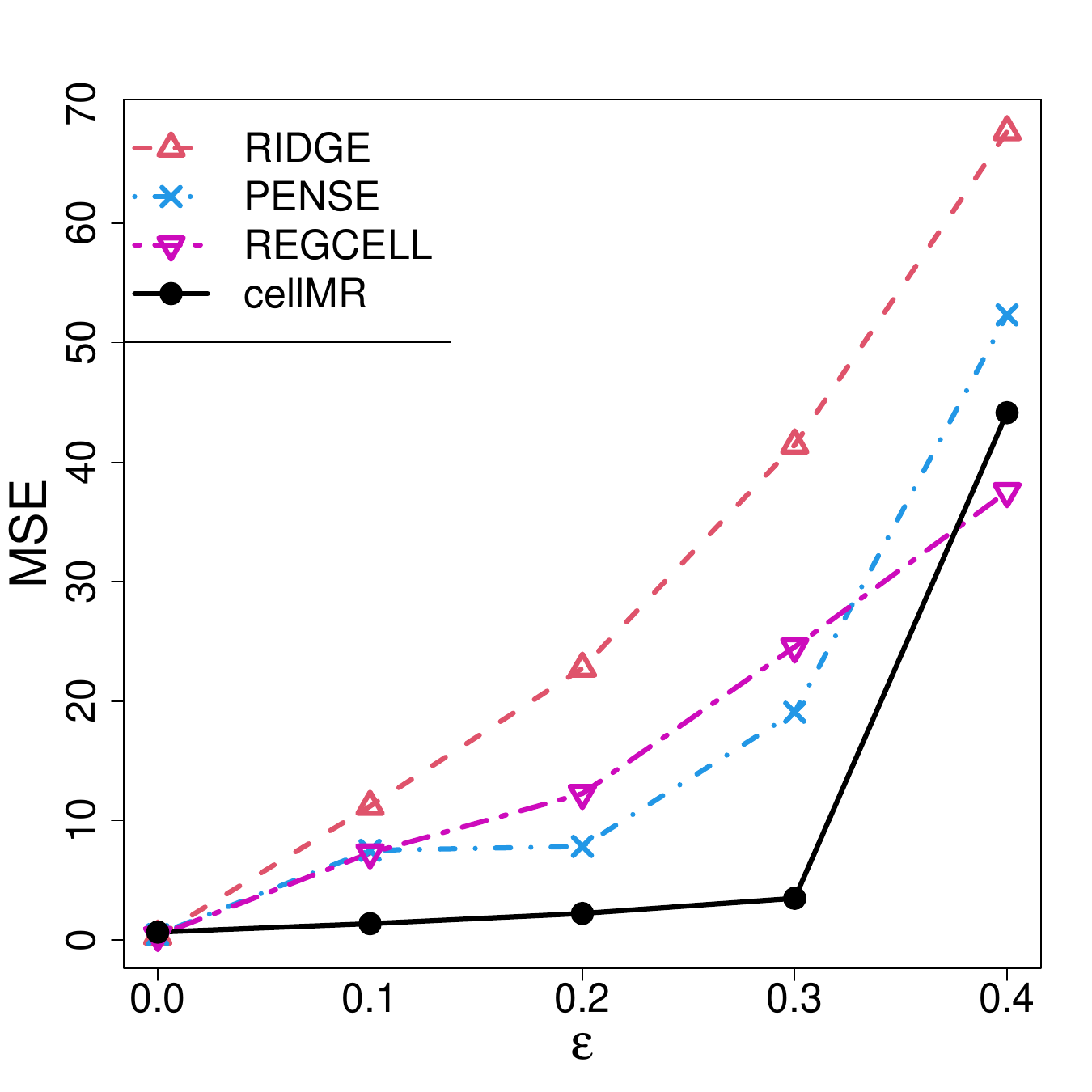} &\includegraphics[width=.31\textwidth]{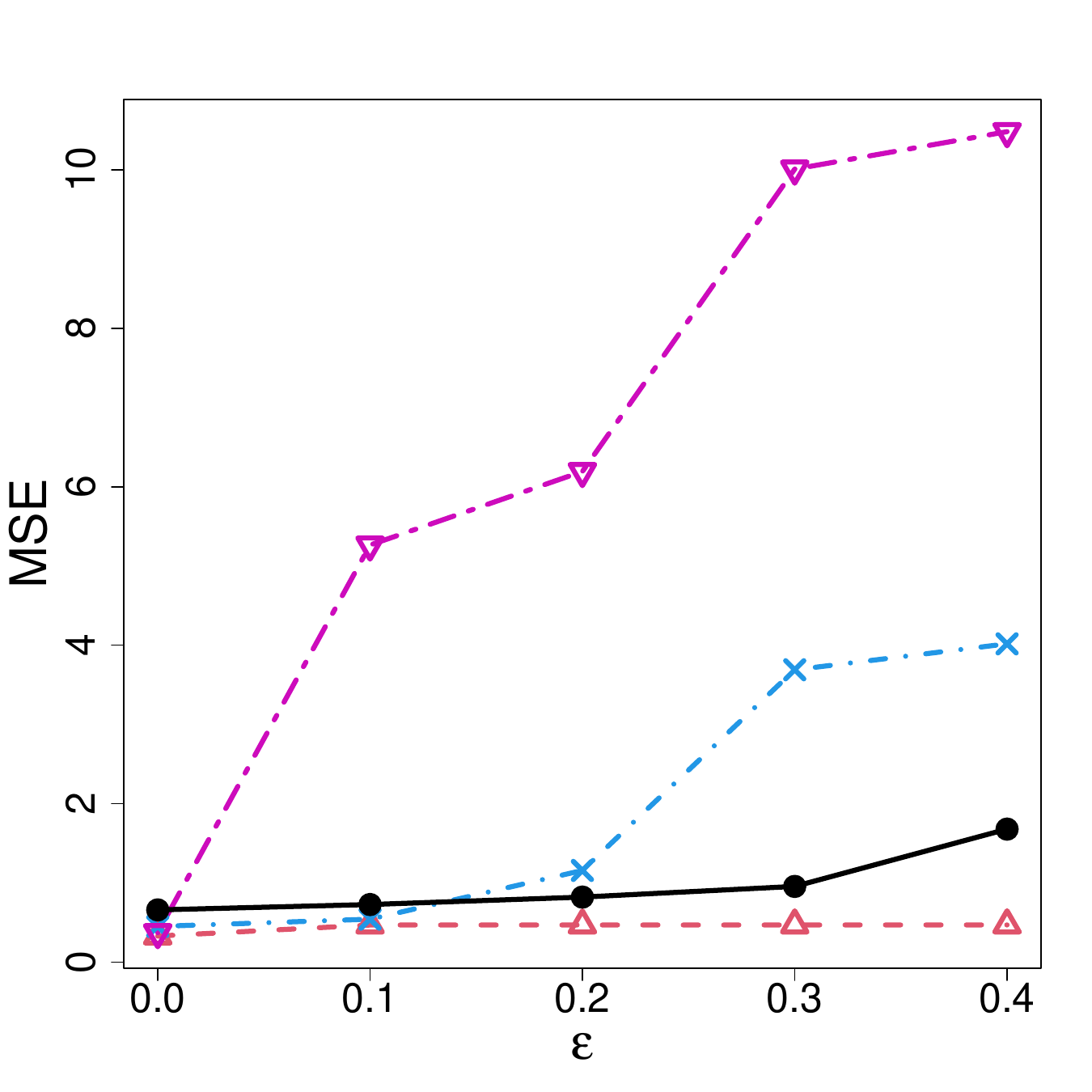} &\includegraphics[width=.31\textwidth]{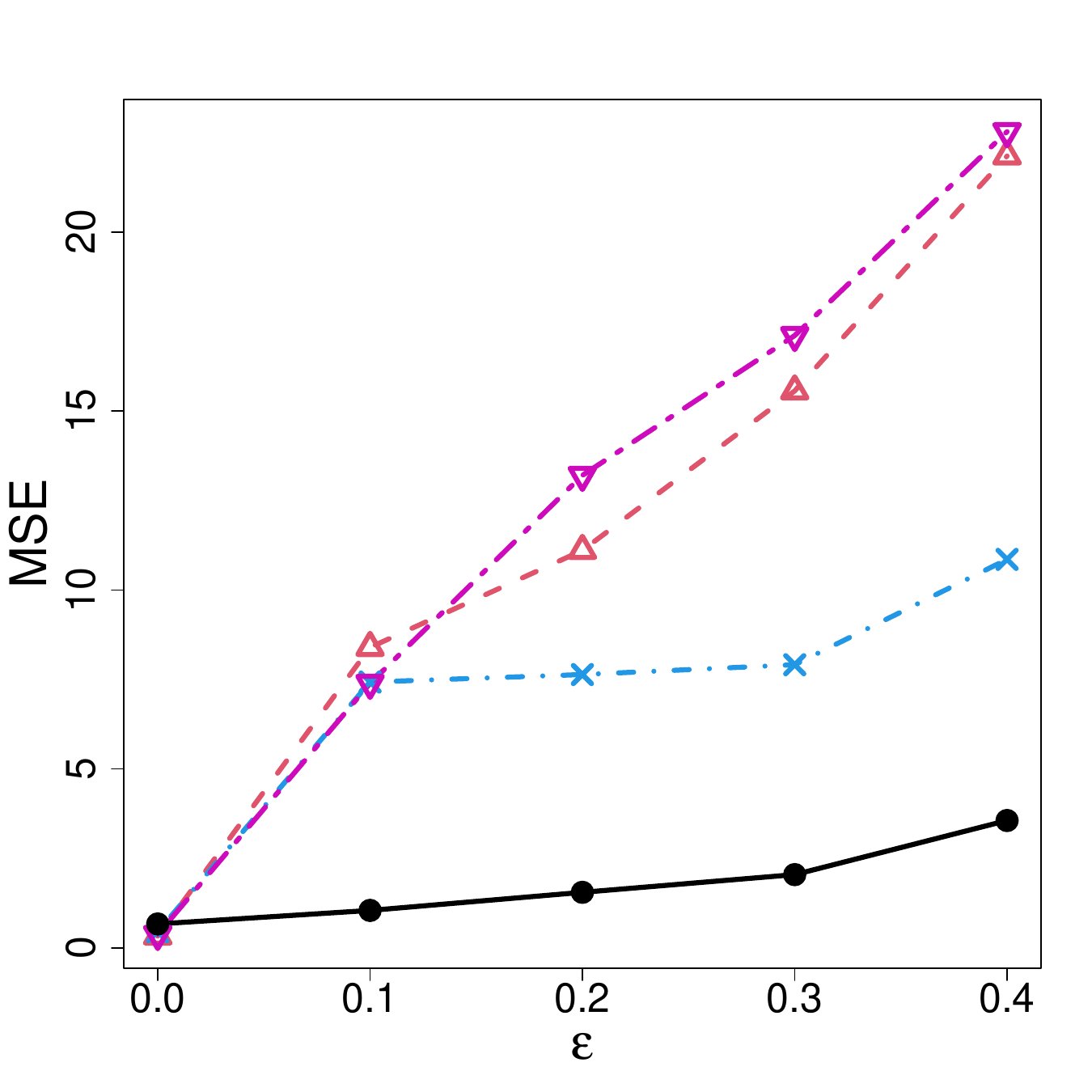}\\
\end{tabular}
\vspace{-6mm}
\caption{Average prediction MSE attained 
by RIDGE, SEST, PENSE, CRM, REGCELL, SHOOT, and 
cellMR as a function of $\eps$ in the presence of 
outliers generated with $\gamma=7$ and 
$p=q=10,25,50$.}
\label{fig:results_varyingeps_reg}
\end{figure}

Figure~\ref{fig:results_varyingeps_inf} 
shows the coverage attained by OLS, FRB, 
Boot-cellMR, and cellBoot for the nominal 90\%
confidence intervals of the regression 
coefficients as a function of $\eps$, with 
$\gamma=7$ and $p=q=10,25,50$. For $p=q=10$, 
cellBoot maintains coverage close to the 
nominal coverage under cellwise contamination 
up to $\eps=0.3$, but its coverage drops at 
$\eps=0.4$, which is consistent with the 
estimation results in
Figure~\ref{fig:results_varyingeps_reg}.
Under casewise and mixed contamination, its 
coverage remains relatively stable throughout the
investigated range, and it outperforms the
competing methods, including Boot-cellMR. Note
that FRB often exhibits severe undercoverage.

\vspace{-3mm}
\begin{figure}[H]
\centering
\begin{tabular}{M{0.0005\textwidth}M{0.29\textwidth}M{0.29\textwidth}M{0.32\textwidth}}
   &\large \textbf{Cellwise}  & \large \textbf{Casewise} &\large{\textbf{Casewise \& Cellwise}} \\
[-4mm]

 \rotatebox{90}{\textbf{\footnotesize{$p=q=10$}}}&\includegraphics[width=.31\textwidth]{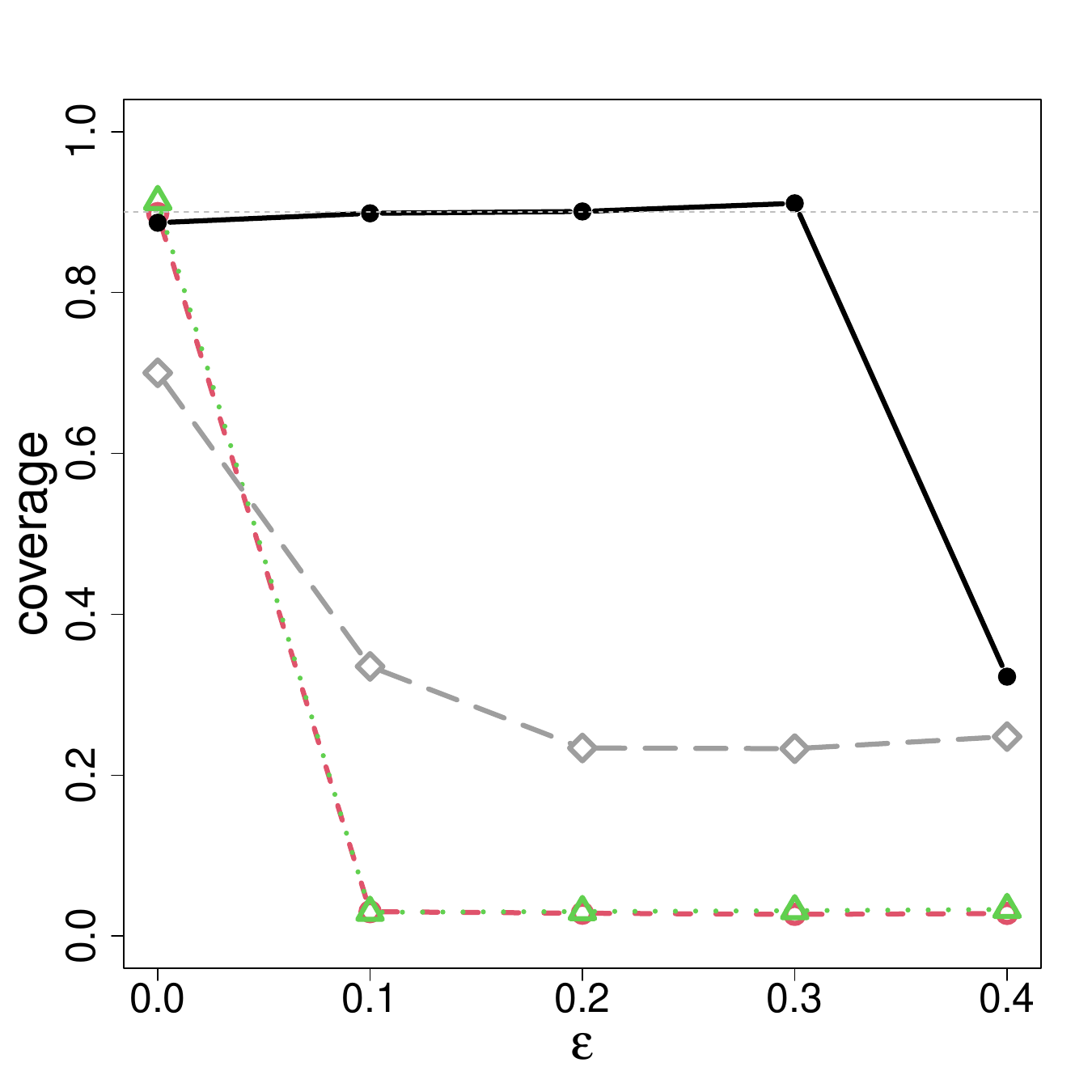} &\includegraphics[width=.31\textwidth]{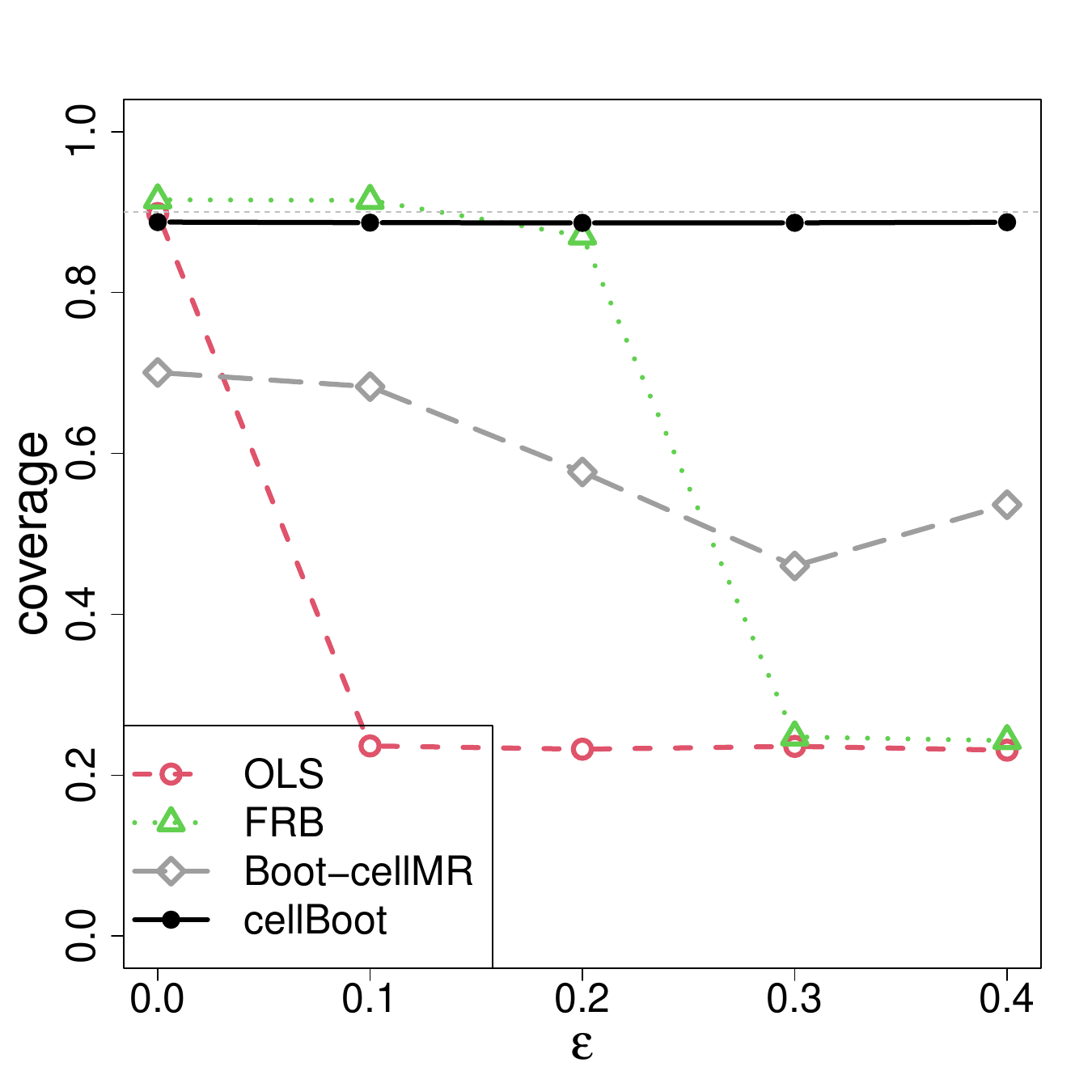} &\includegraphics[width=.31\textwidth]{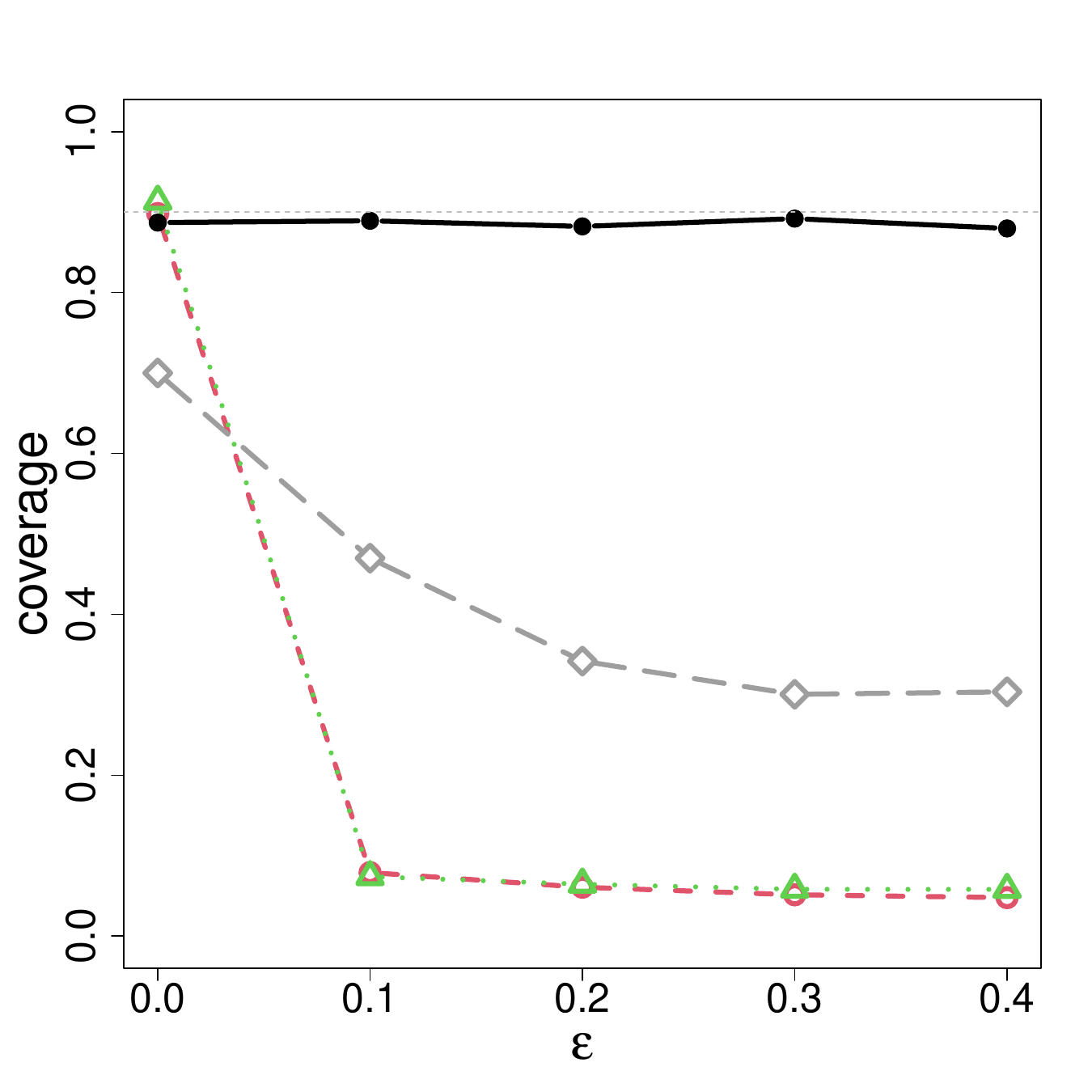} \\

 \rotatebox{90}{\textbf{\footnotesize{$p=q=25$}}}&\includegraphics[width=.31\textwidth]{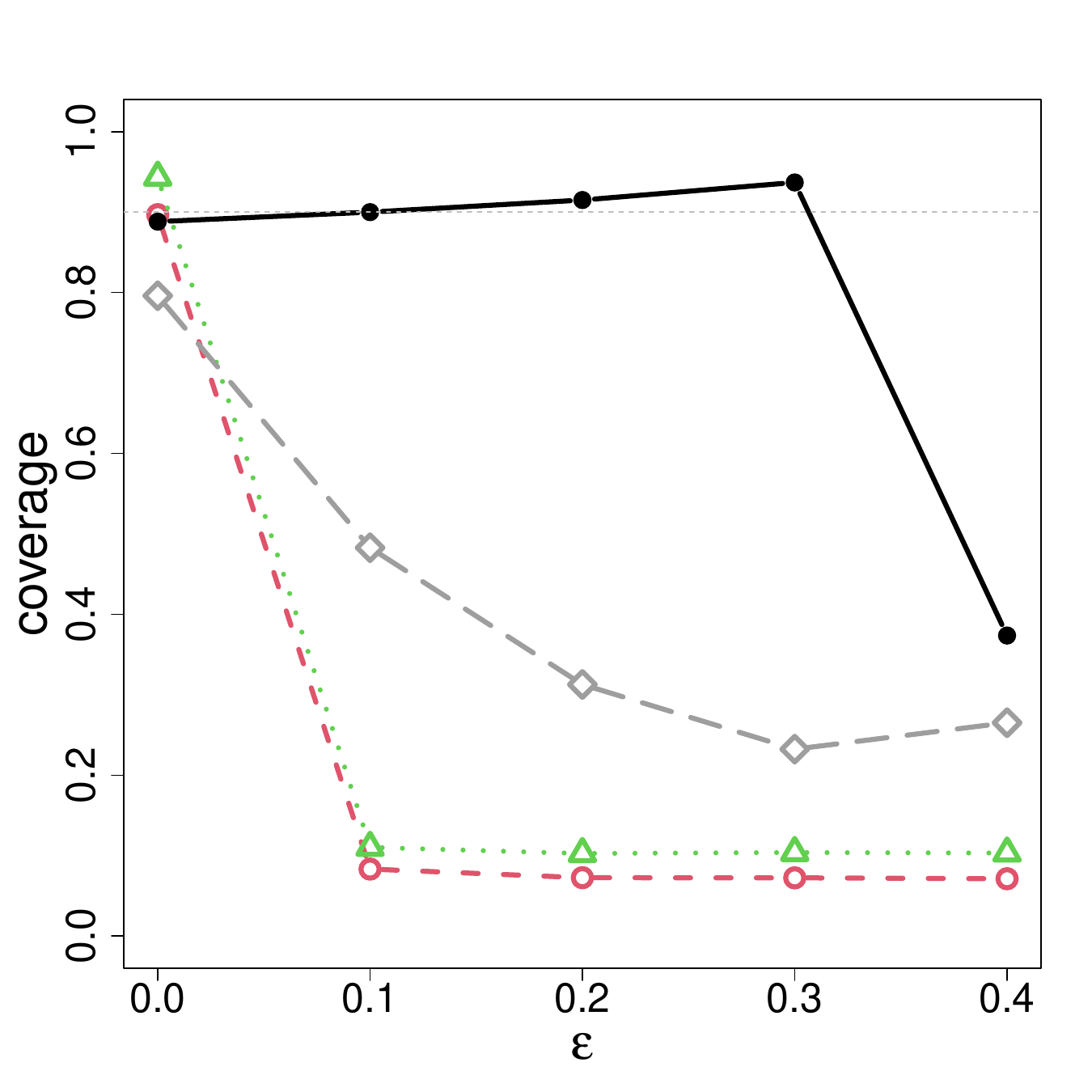} &\includegraphics[width=.31\textwidth]{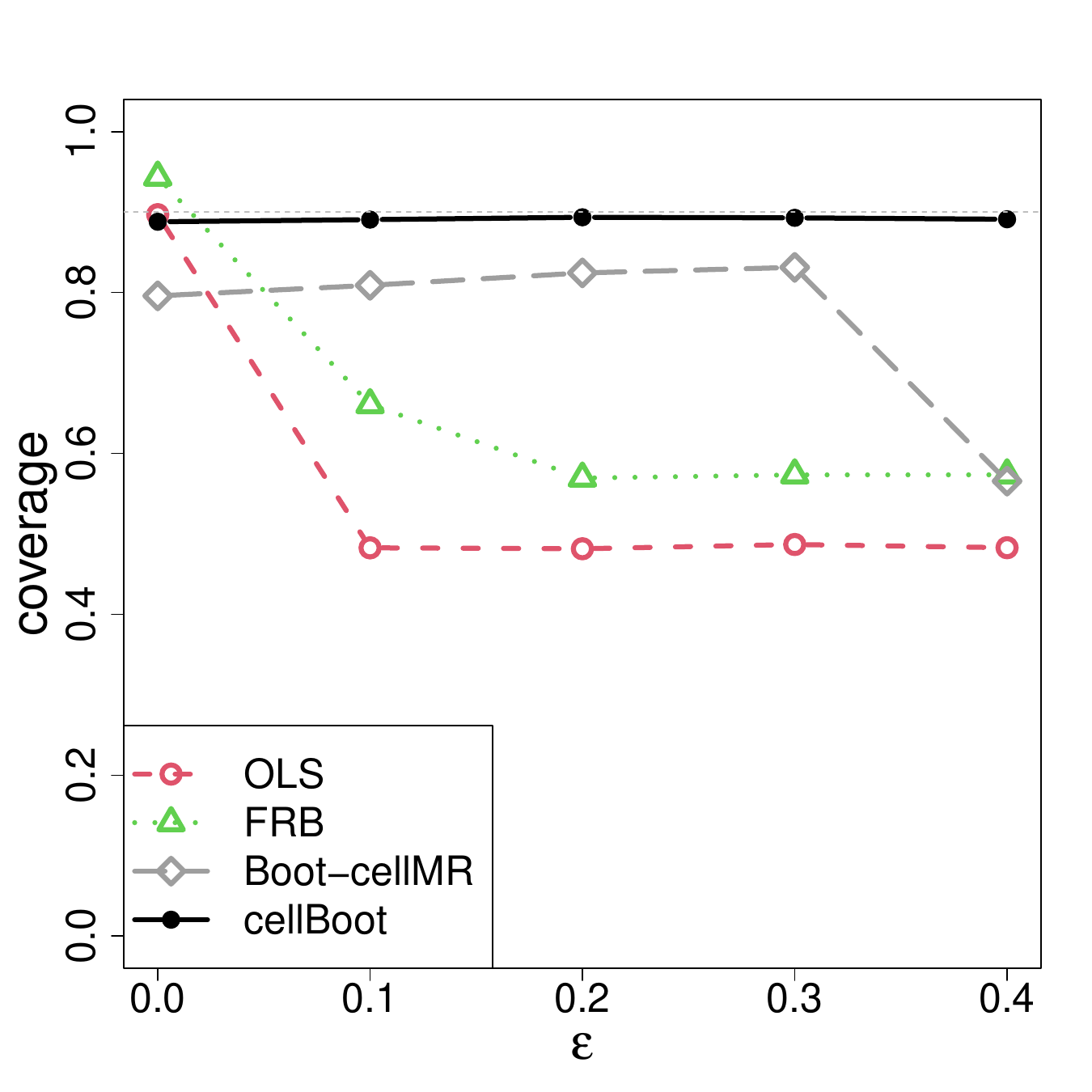} &\includegraphics[width=.31\textwidth]{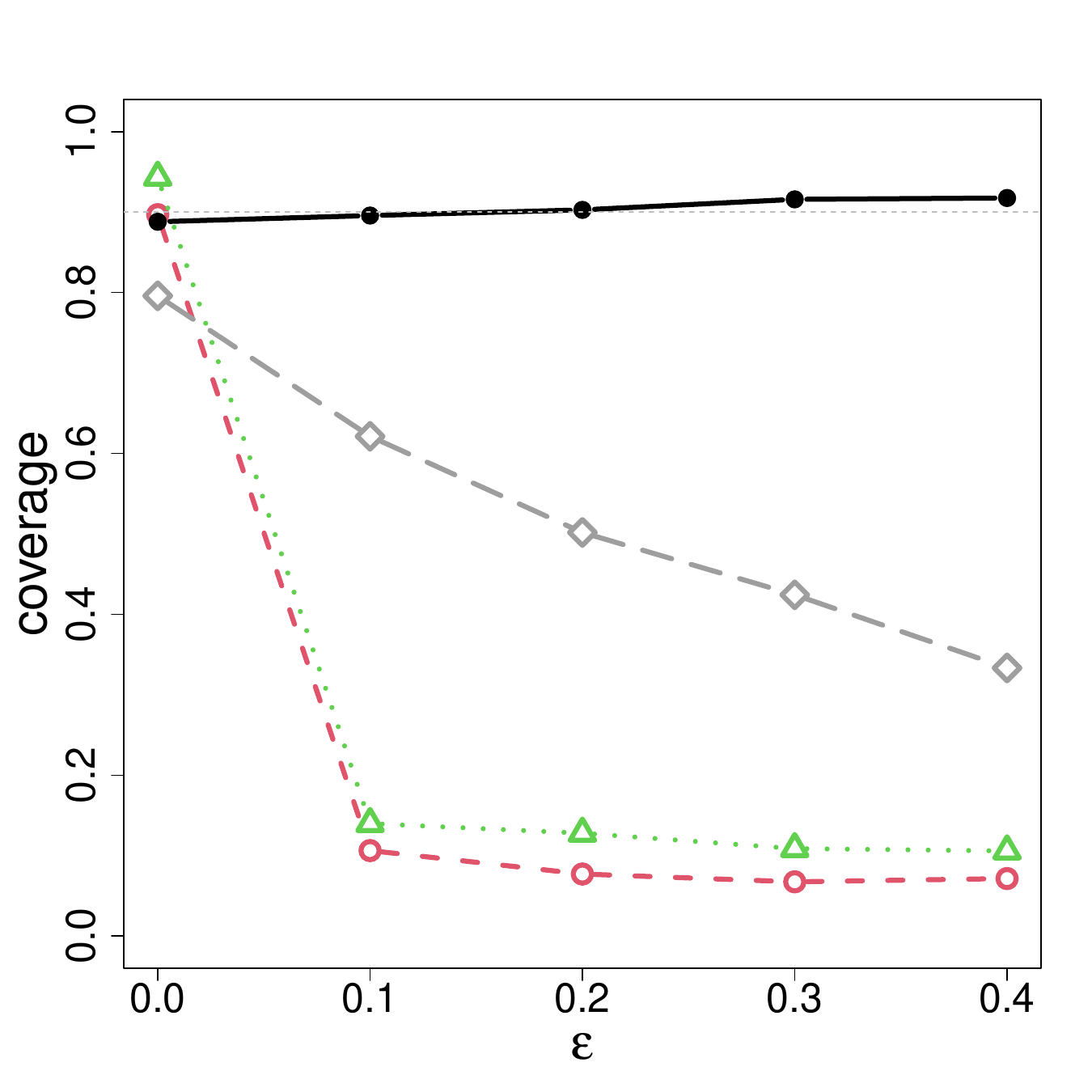}\\
 \rotatebox{90}{\textbf{\footnotesize{$p=q=50$}}}&\includegraphics[width=.31\textwidth]{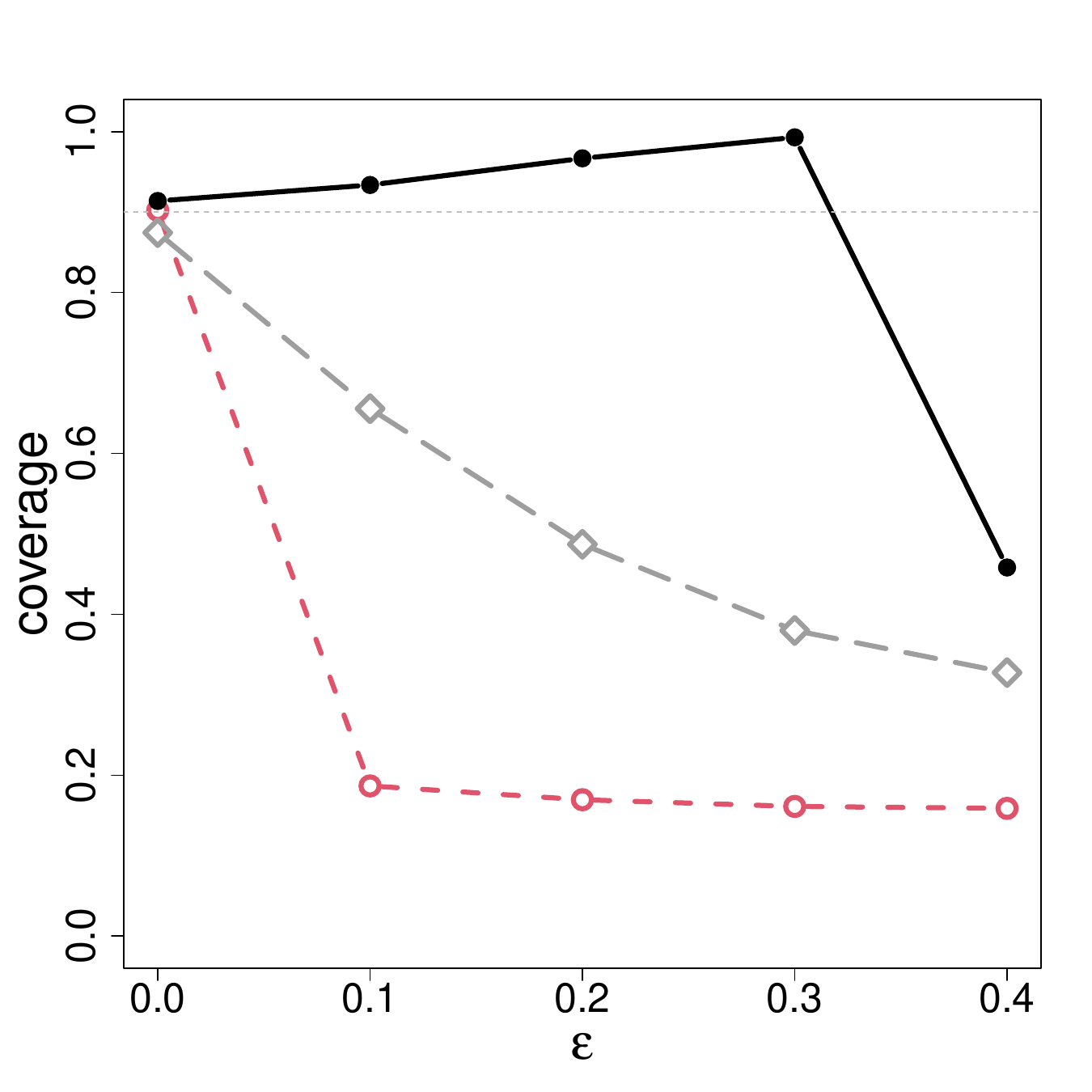} &\includegraphics[width=.31\textwidth]{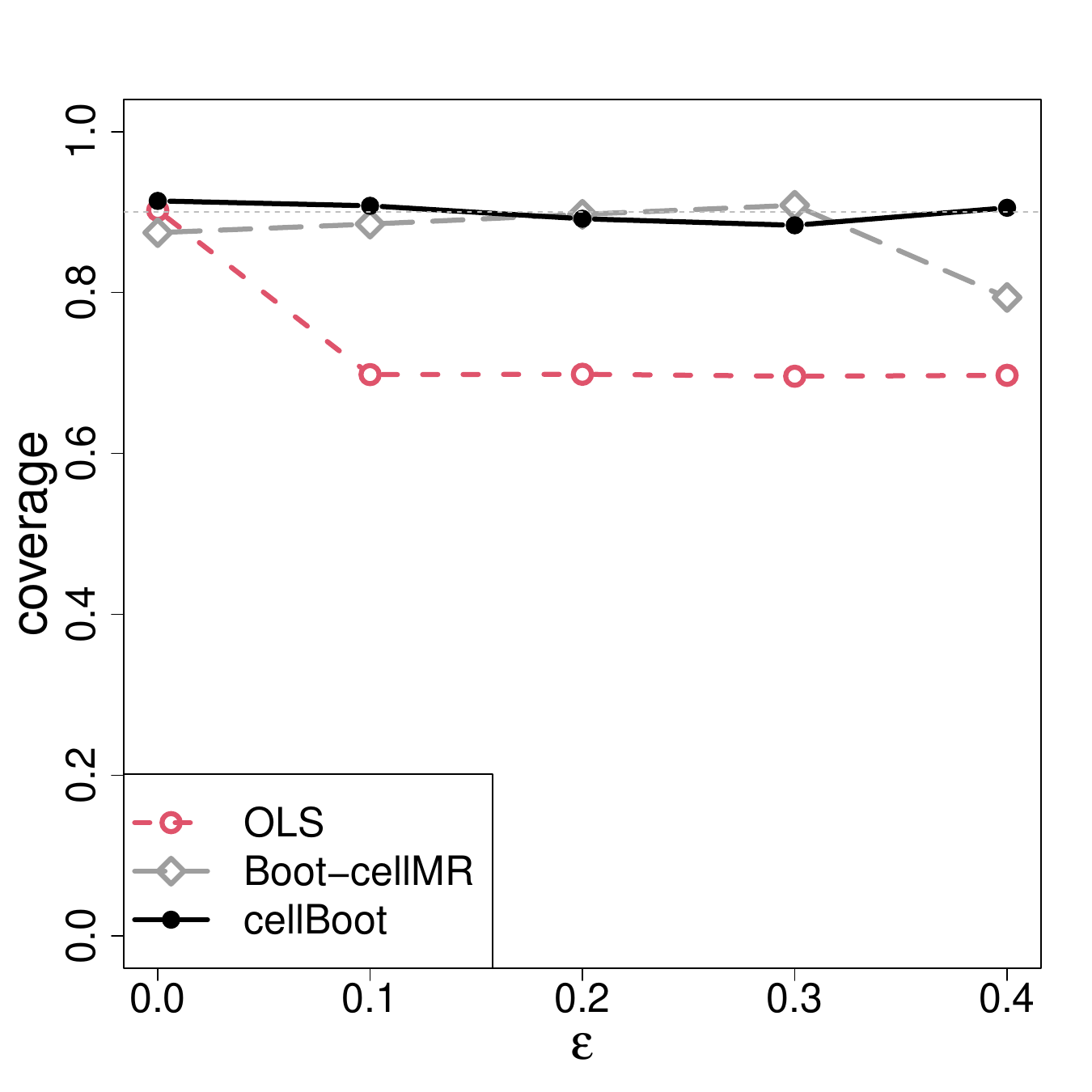} &\includegraphics[width=.31\textwidth]{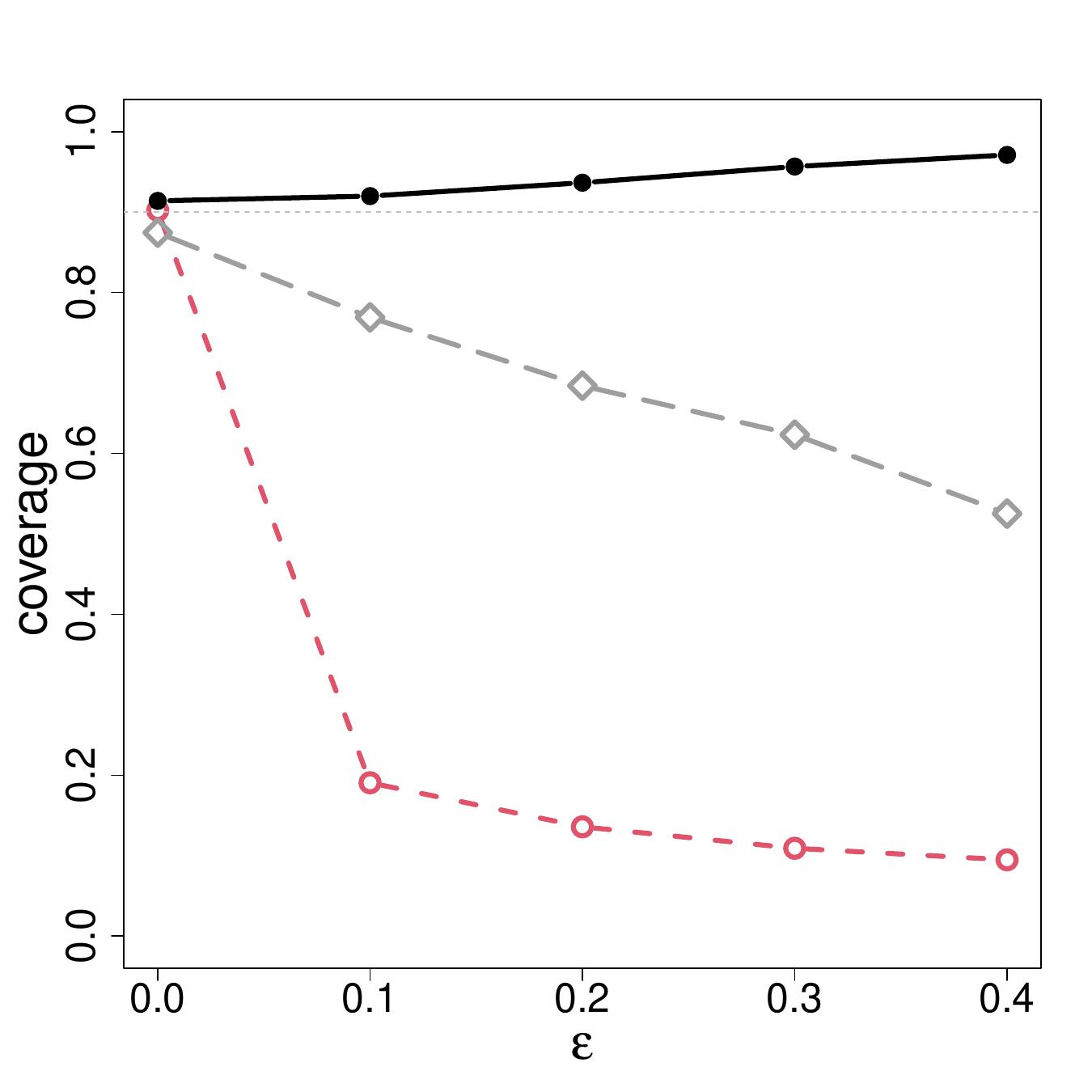}\\
\end{tabular}
\vspace{-5mm}
\caption{Average coverage  attained by OLS, 
FRB, Boot-cellMR, and cellBoot for the $0.9$-level confidence intervals of the regression coefficients as a function of $\eps$ in the presence of outliers generated with $\gamma=7$ and $p=q=\{10,25,50\}$.}
\label{fig:results_varyingeps_inf}
\end{figure}

\newpage
To study how the estimation error of the
slope changes with the contamination fraction, 
we consider the mean squared estimation error 
(MSEE), which is defined as 
$\text{MSEE}=\|\bhB-\bB\|_F^2$.
Figure~\ref{fig:results_varyingeps_reg_msee} shows
the average \text{MSEE} for data contaminated as in
Section~\ref{sec:simreg}, with $\gamma=7$ and
$p=q=10,25,50$. As expected, the pattern 
is qualitatively similar to that observed for the
prediction error in Figure~\ref{fig:results_varyingeps_reg}.

\vspace{4mm}
\begin{figure}[!ht]
\centering
\begin{tabular}{M{0.0005\textwidth}M{0.29\textwidth}M{0.29\textwidth}M{0.32\textwidth}}
   &\large \textbf{Cellwise}  & \large \textbf{Casewise} &\large{\textbf{Casewise \& Cellwise}} \\
[-4mm]

 \rotatebox{90}{\textbf{\footnotesize{$p=q=10$}}}&\includegraphics[width=.31\textwidth]{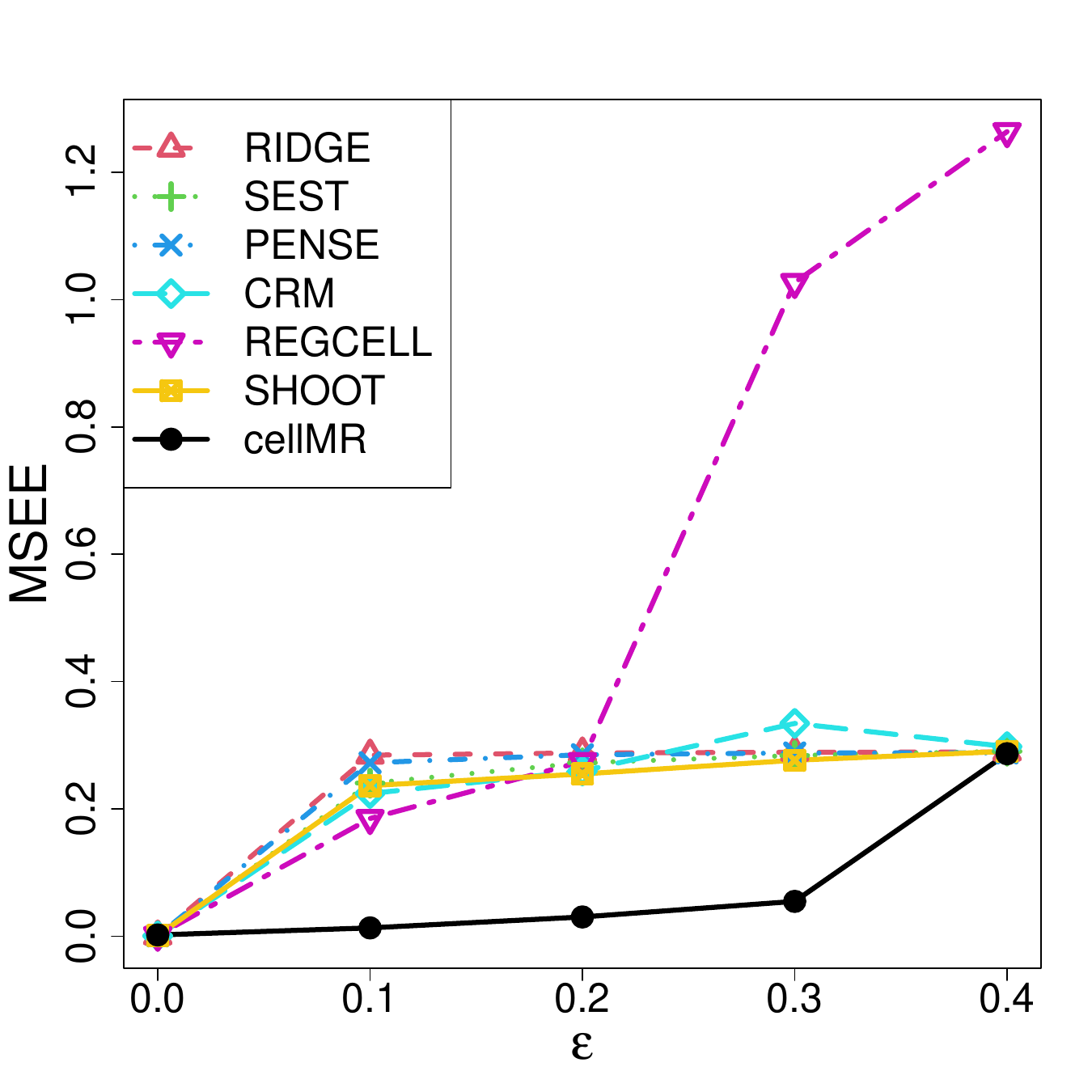} &\includegraphics[width=.31\textwidth]{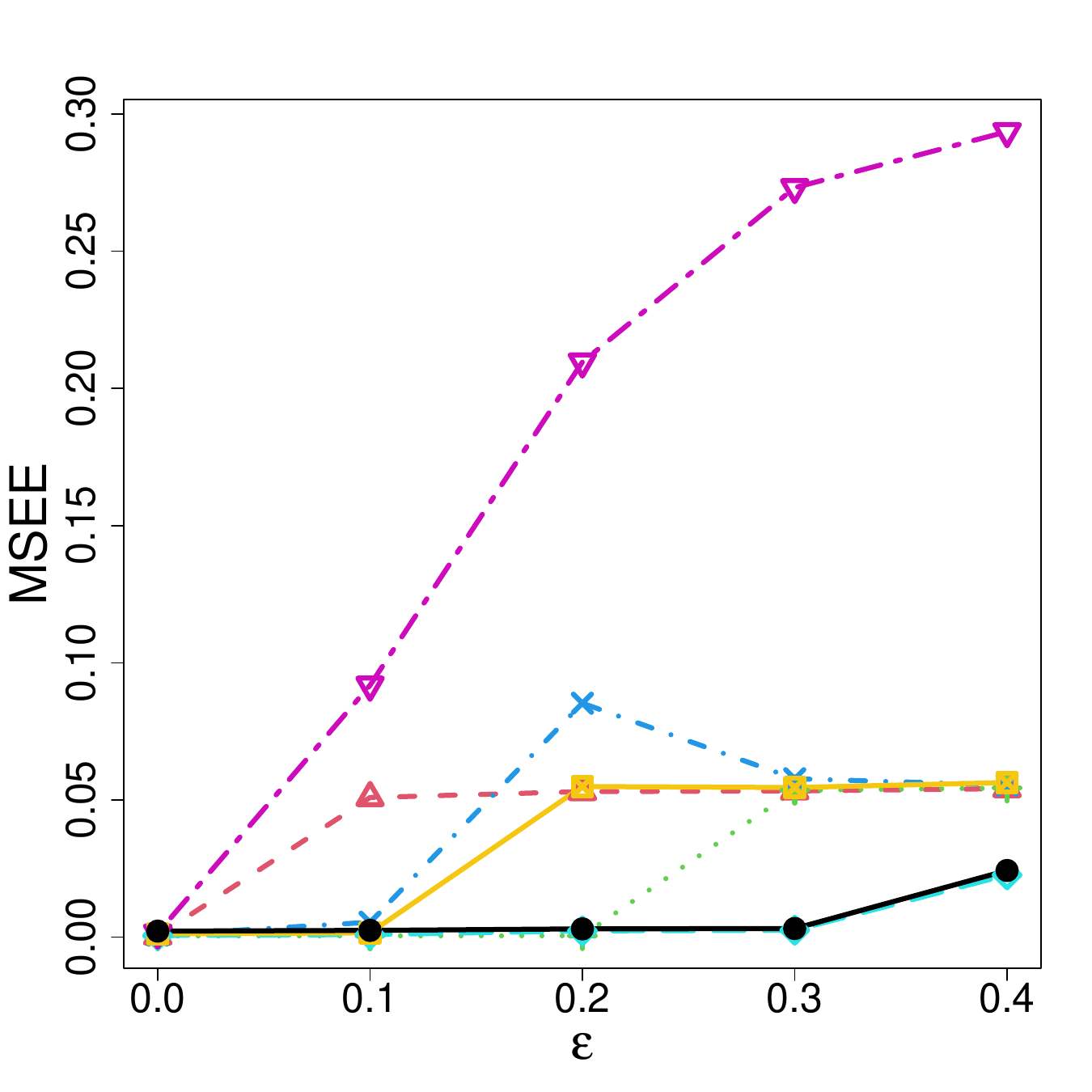} &\includegraphics[width=.31\textwidth]{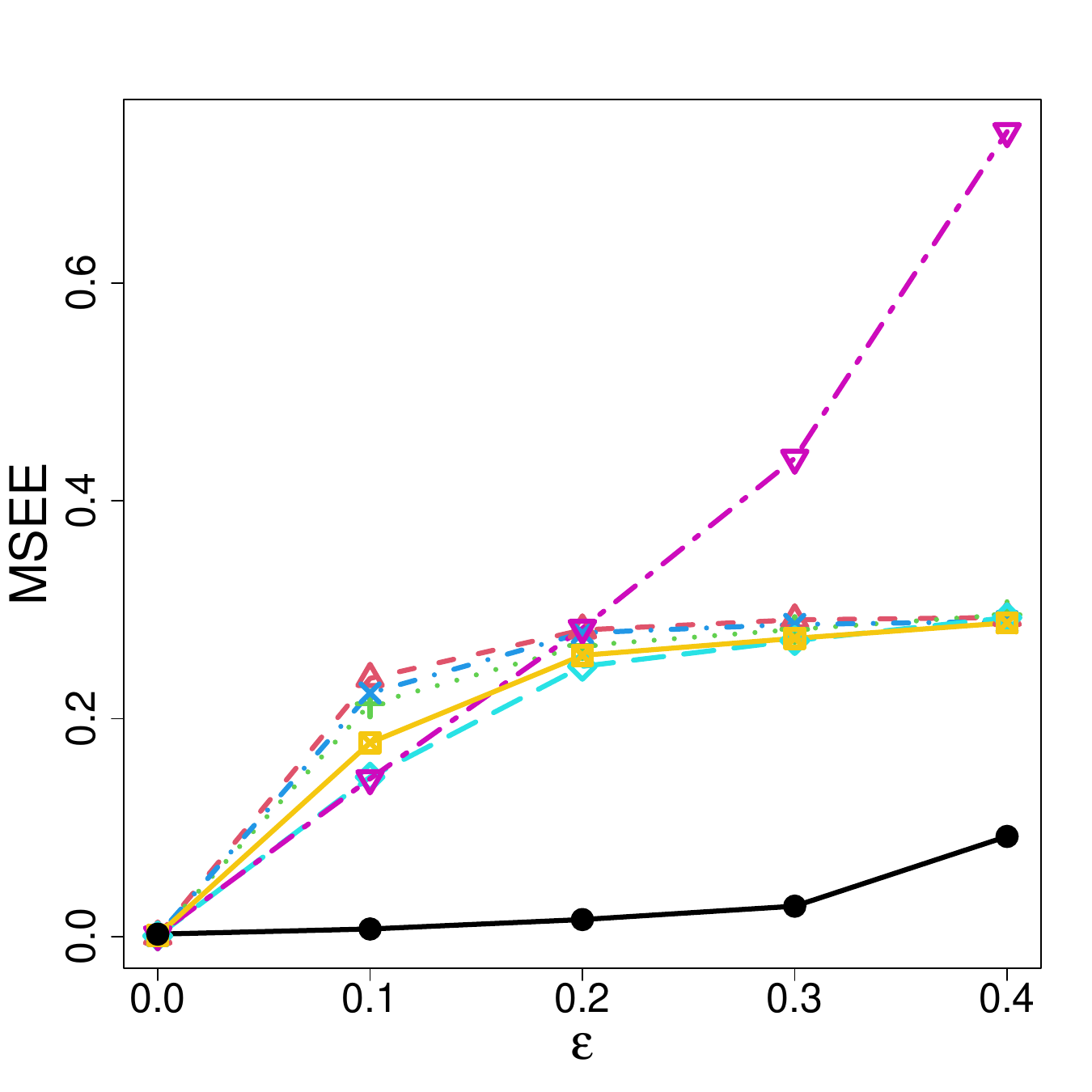} \\

 \rotatebox{90}{\textbf{\footnotesize{$p=q=25$}}}&\includegraphics[width=.31\textwidth]{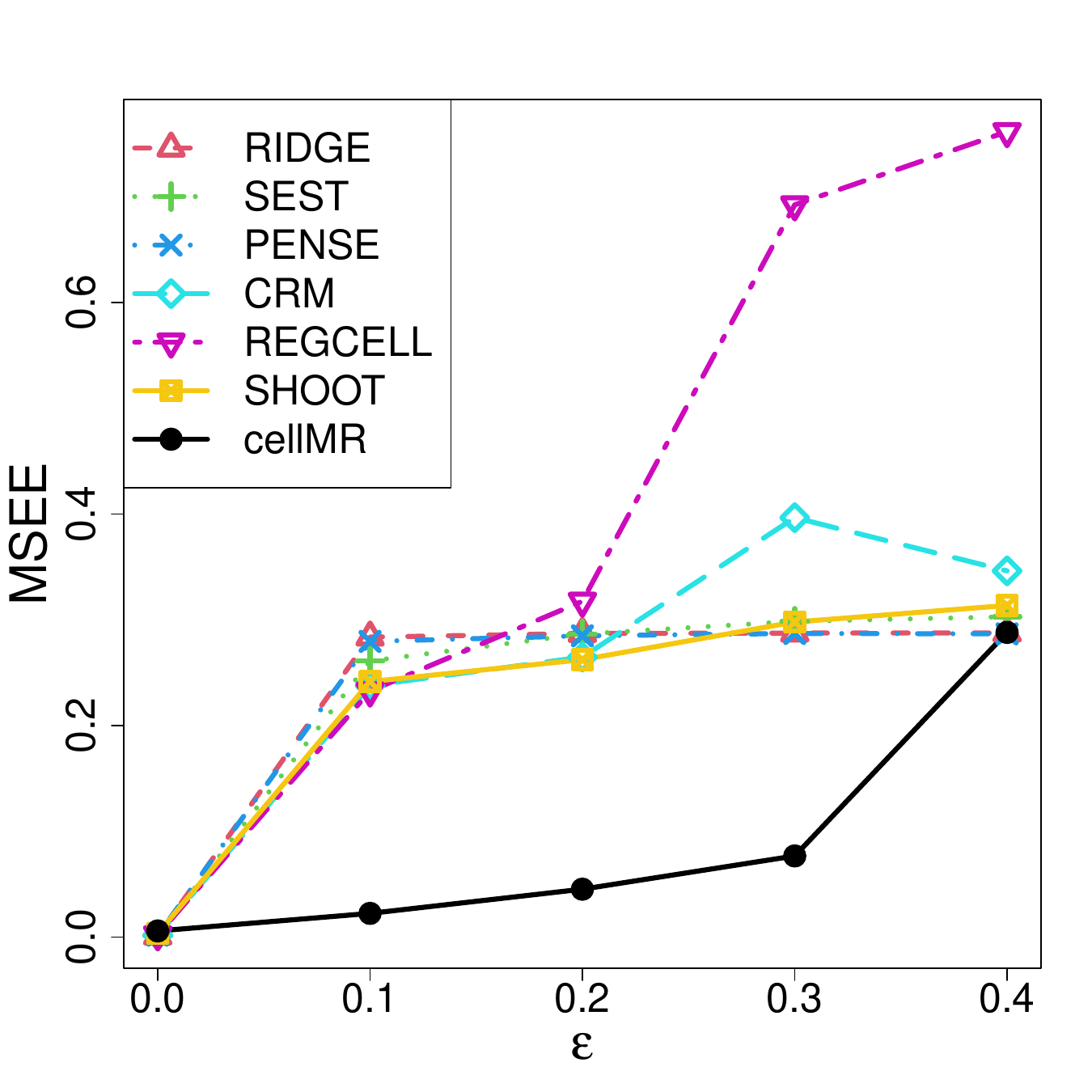} &\includegraphics[width=.31\textwidth]{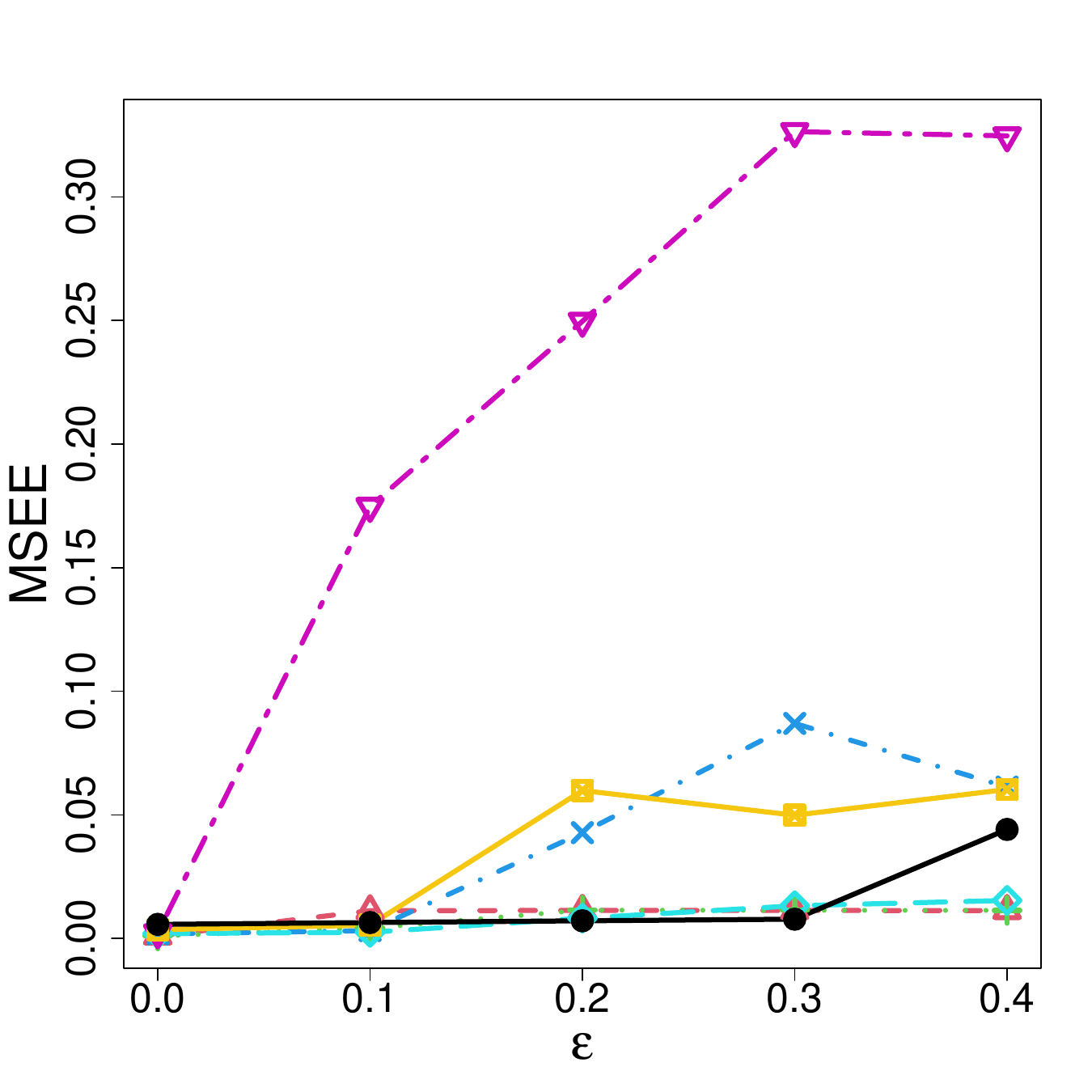} &\includegraphics[width=.31\textwidth]{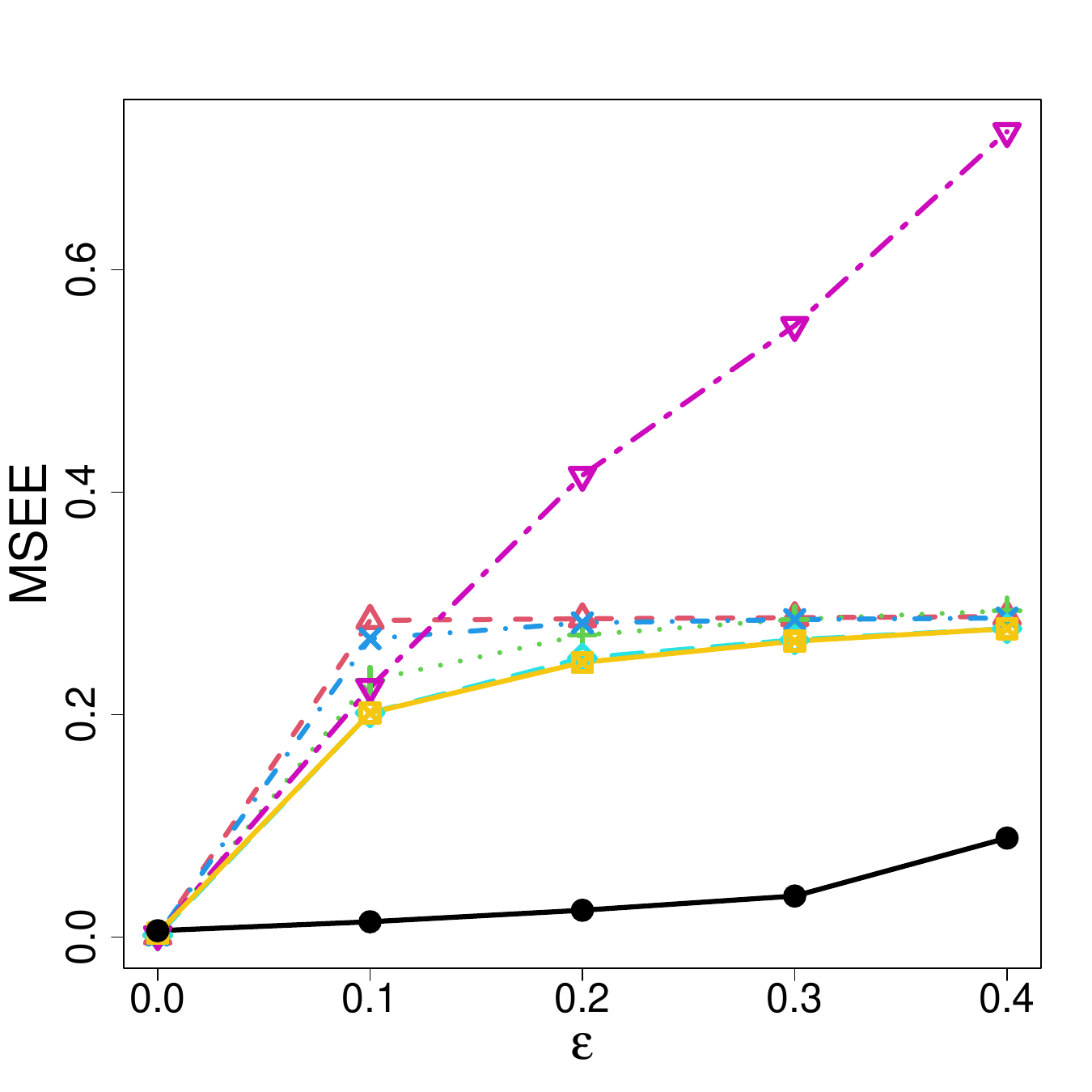}\\
 \rotatebox{90}{\textbf{\footnotesize{$p=q=50$}}}&\includegraphics[width=.31\textwidth]{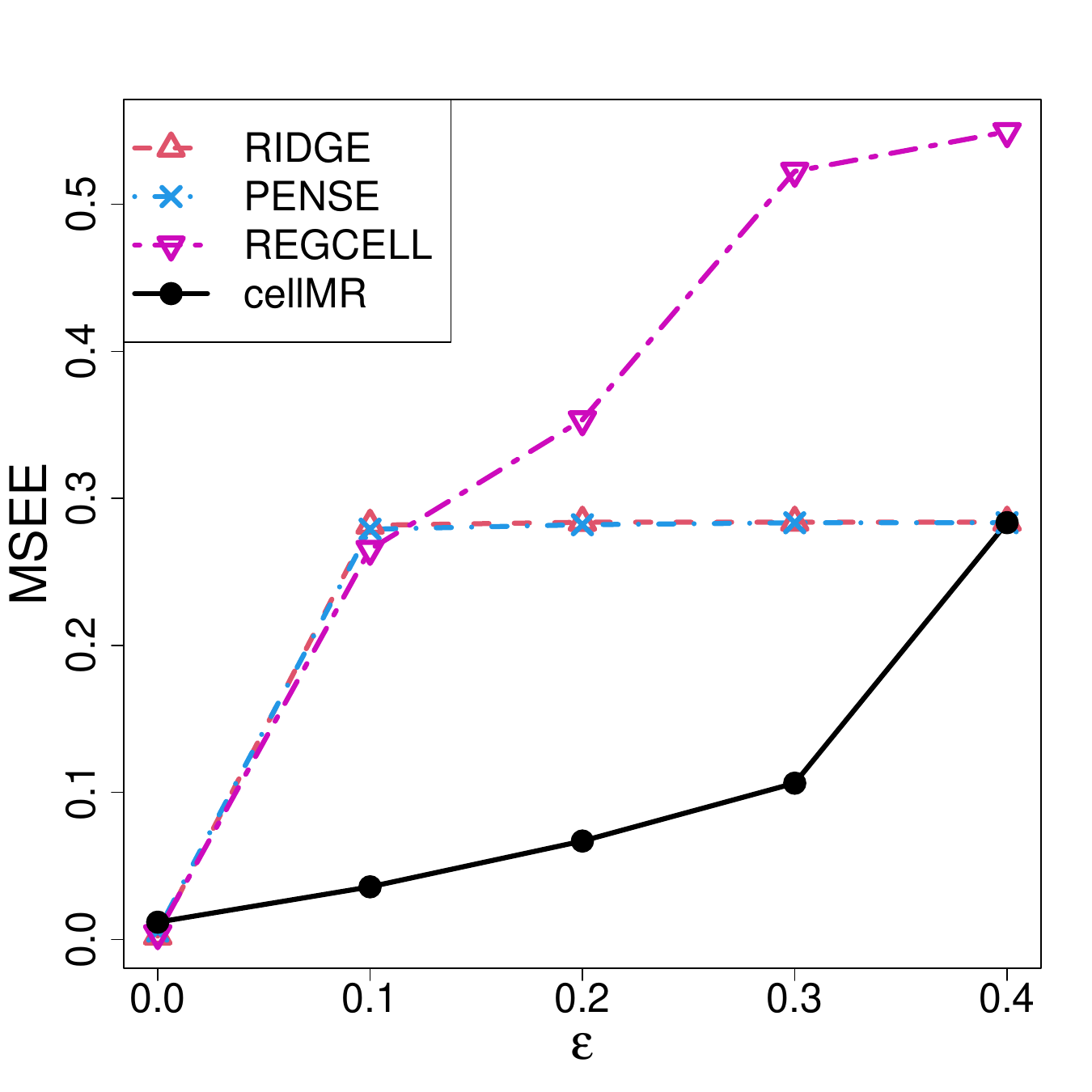} &\includegraphics[width=.31\textwidth]{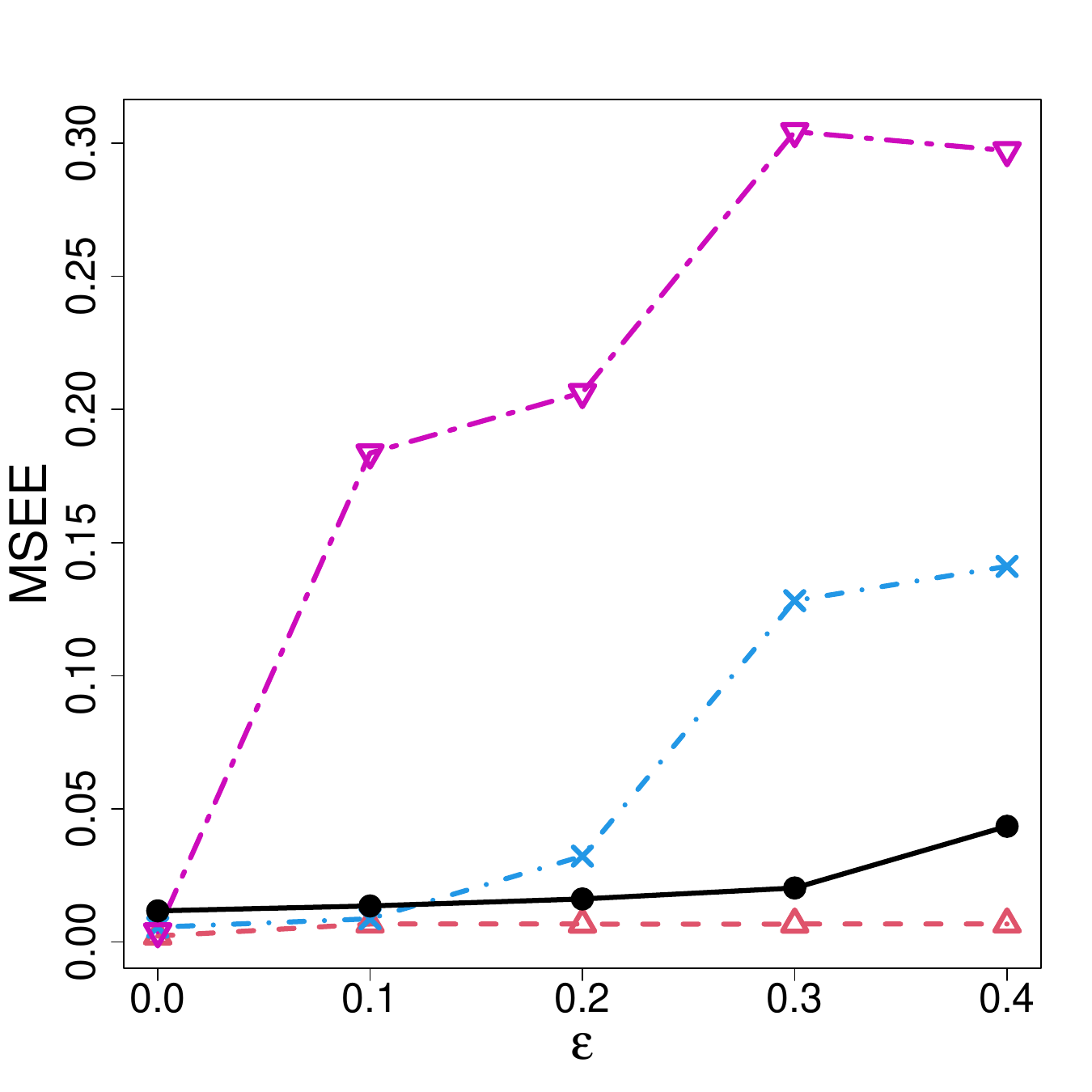} &\includegraphics[width=.31\textwidth]{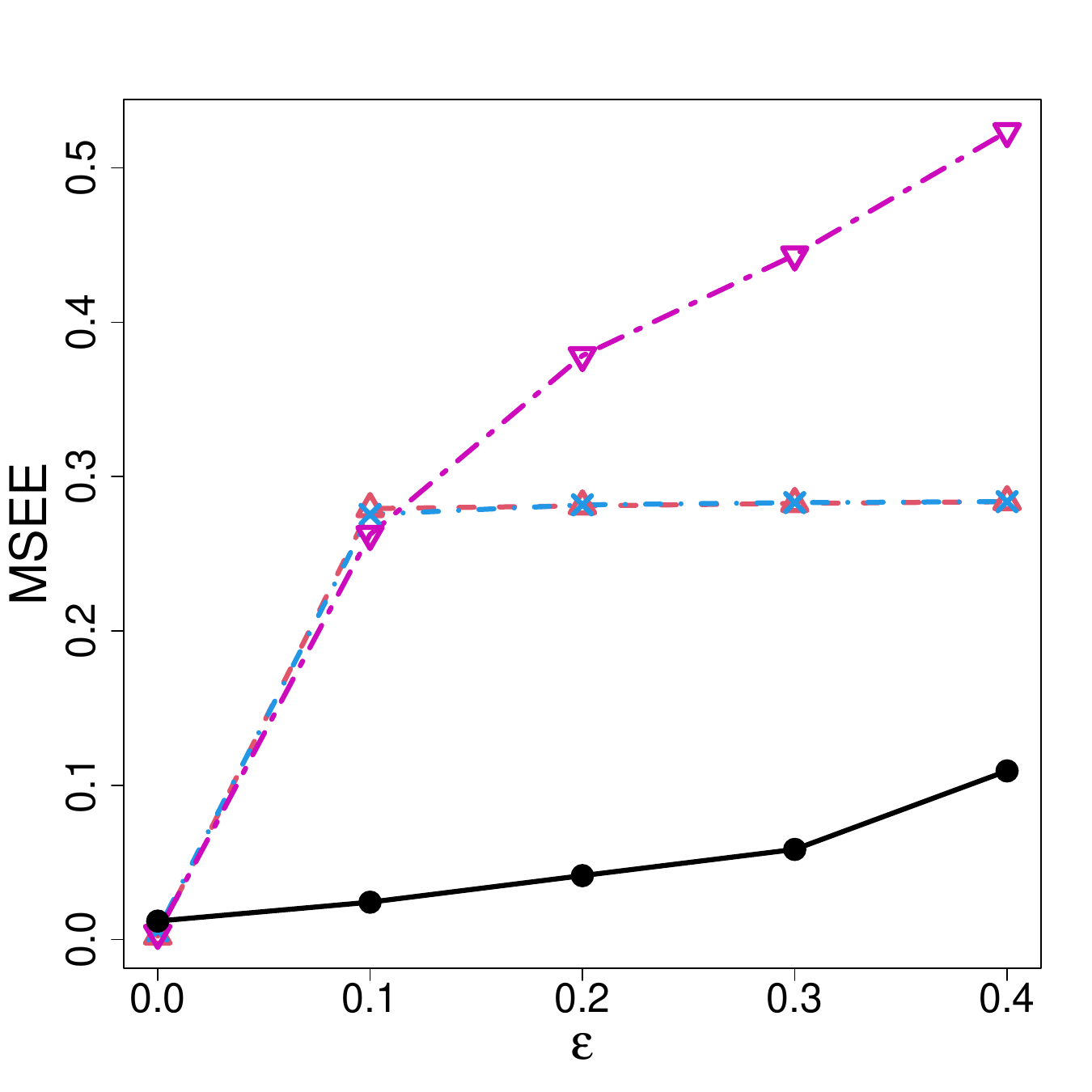}\\
\end{tabular}
\caption{Average MSEE attained by RIDGE, SEST,
PENSE, CRM, REGCELL, SHOOT, and cellMR as a function of
$\eps$ in the presence of outliers generated with
$\gamma=7$ and $p=q=10,25,50$.}
\label{fig:results_varyingeps_reg_msee}
\end{figure}

\newpage
To assess the contribution of individual components of cellMR, we conduct an ablation study
examining the benefits of ridge regularization in
\eqref{eq:ridge} and the cross-validation
strategy described in Section~\ref{sec:tuning}.
Figure~\ref{fig:results_lambda} compares the 
average MSE attained by cellMR with that of its
unregularized counterpart ($\lambda=0$)
in the presence of cellwise outliers, casewise
outliers, or both, without missing data.
The two estimators perform similarly for 
$p=q=10$, whereas the benefits of ridge 
regularization become more pronounced
as the dimension increases, as expected.
In particular, for $p=q=50$, regularization
substantially reduces the MSE.
\vspace{2mm}
\begin{figure}[H]
\centering
\begin{tabular}{M{0.0005\textwidth}M{0.29\textwidth}M{0.29\textwidth}M{0.32\textwidth}}
   &\large \textbf{Cellwise}  & \large \textbf{Casewise} &\large{\textbf{Casewise \& Cellwise}} \\
[-4mm]

 \rotatebox{90}{\textbf{\footnotesize{$p=q=10$}}}&\includegraphics[width=.31\textwidth]
 {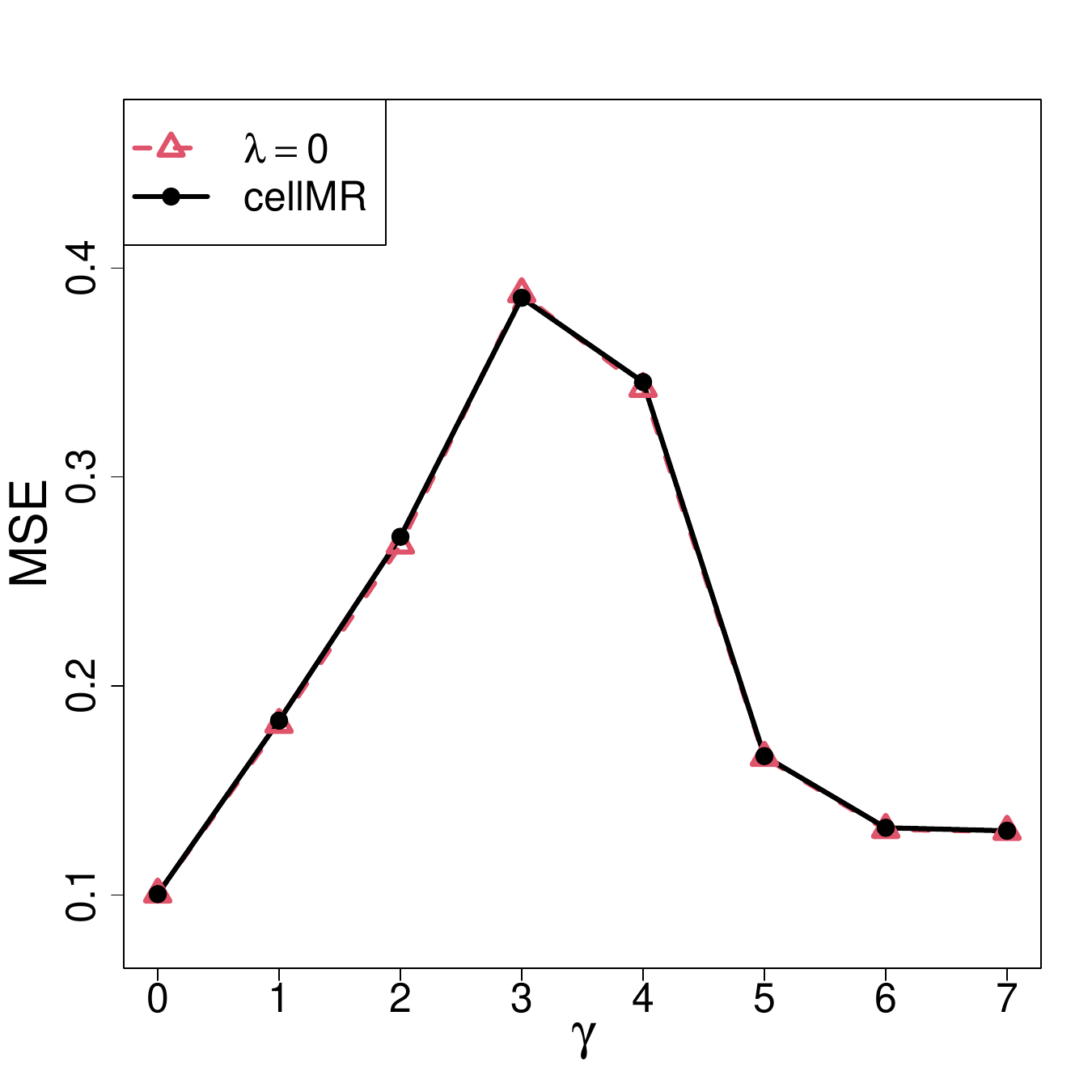} &\includegraphics[width=.31\textwidth]{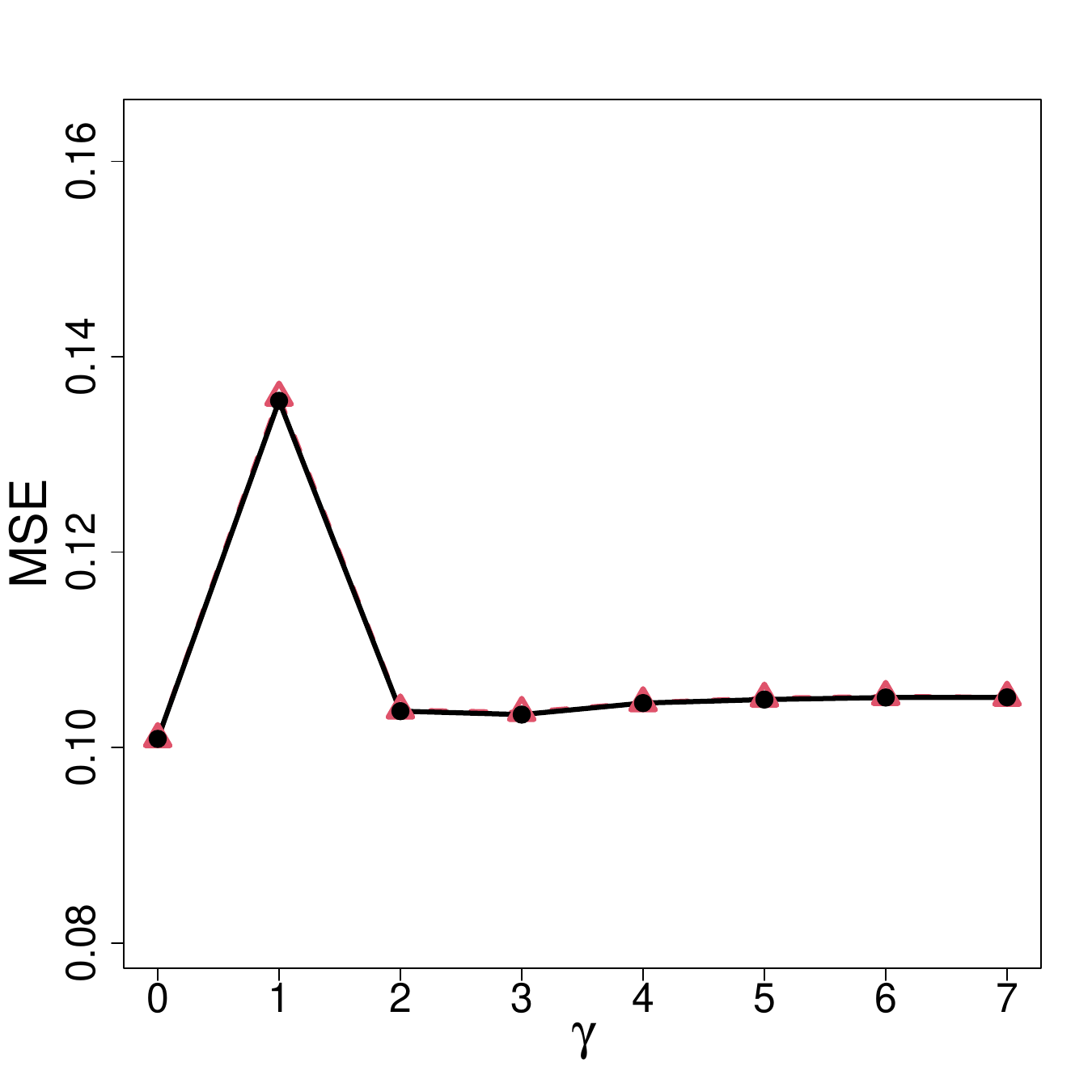} &\includegraphics[width=.31\textwidth]{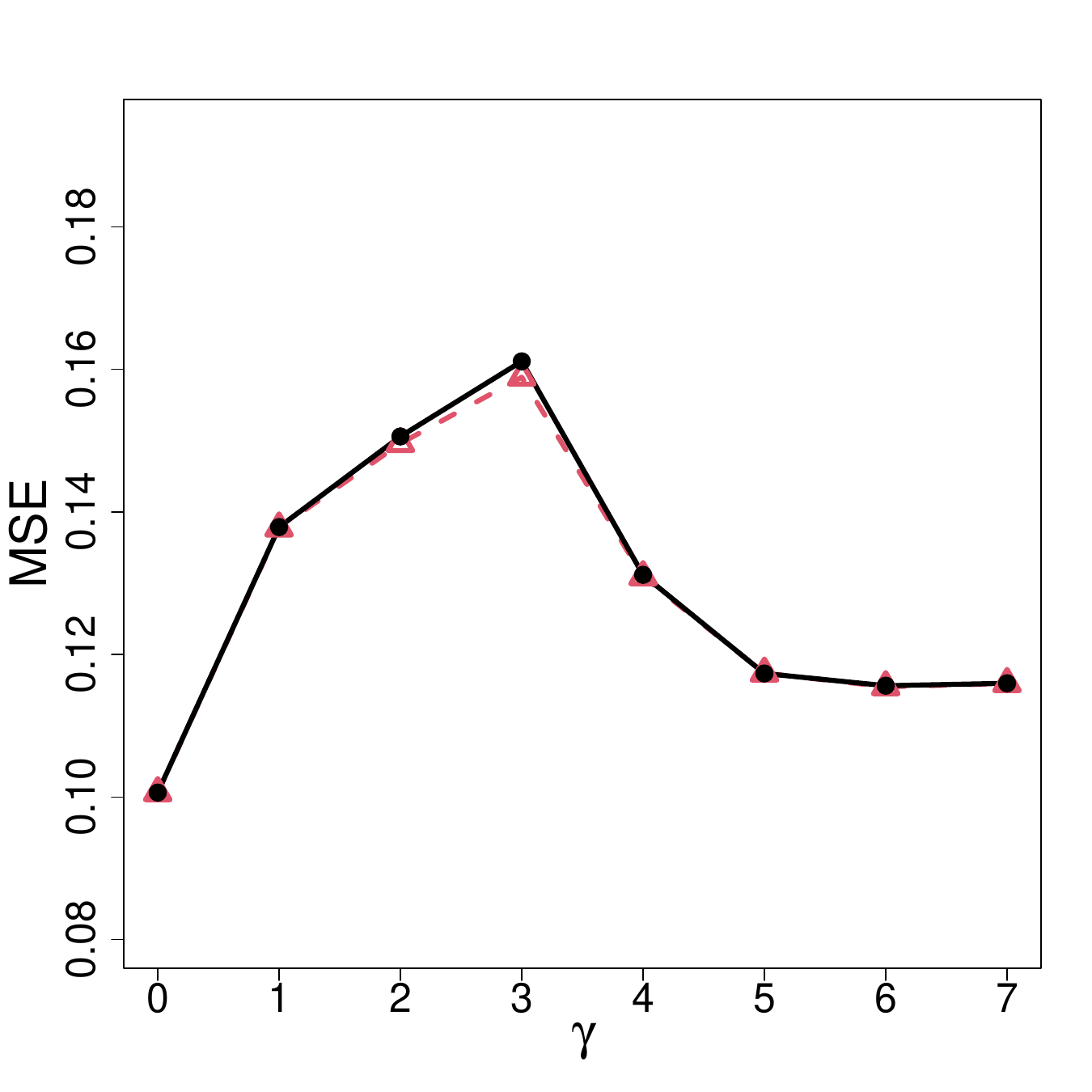} \\
 
 \rotatebox{90}{\textbf{\footnotesize{$p=q=25$}}}&\includegraphics[width=.31\textwidth]{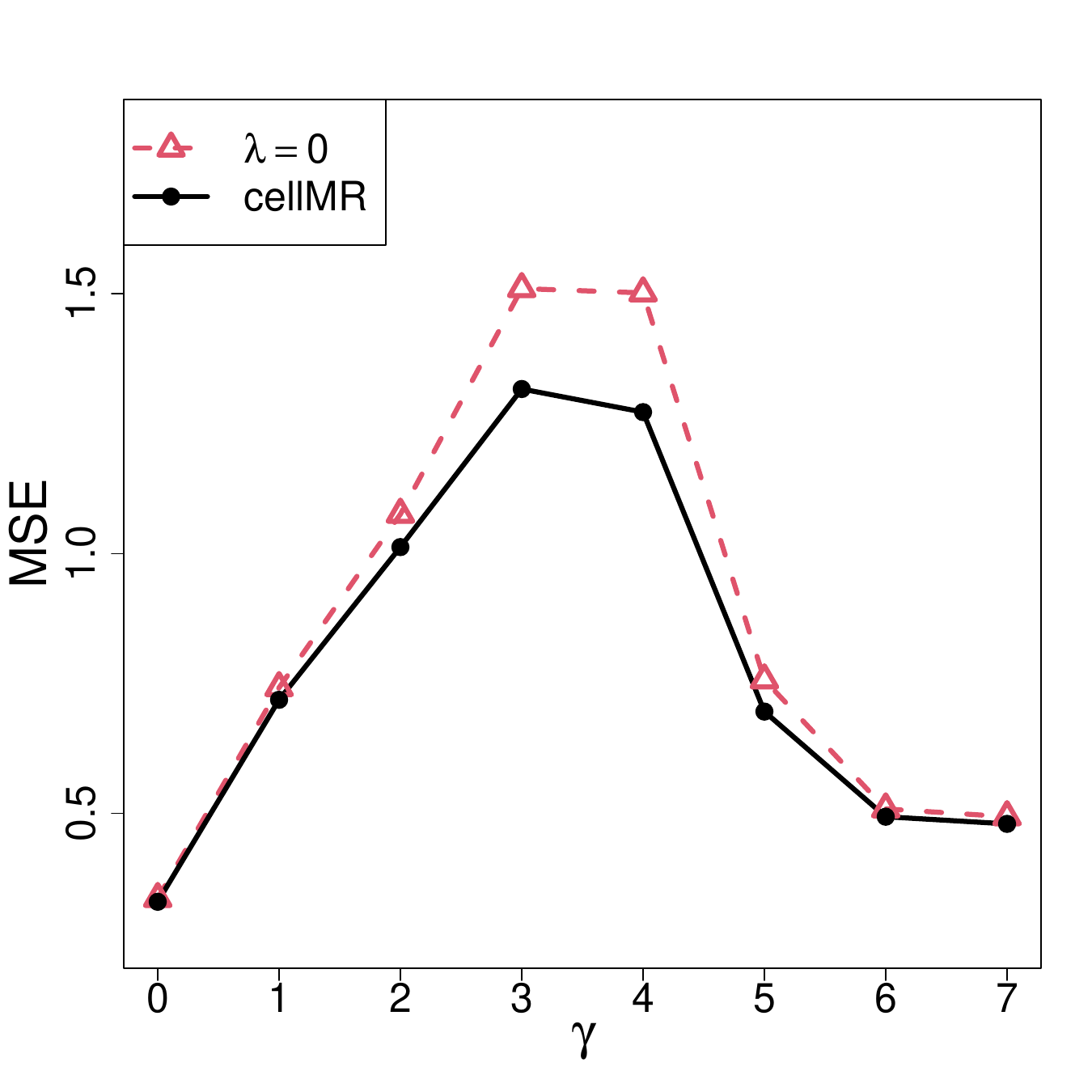} &\includegraphics[width=.31\textwidth]{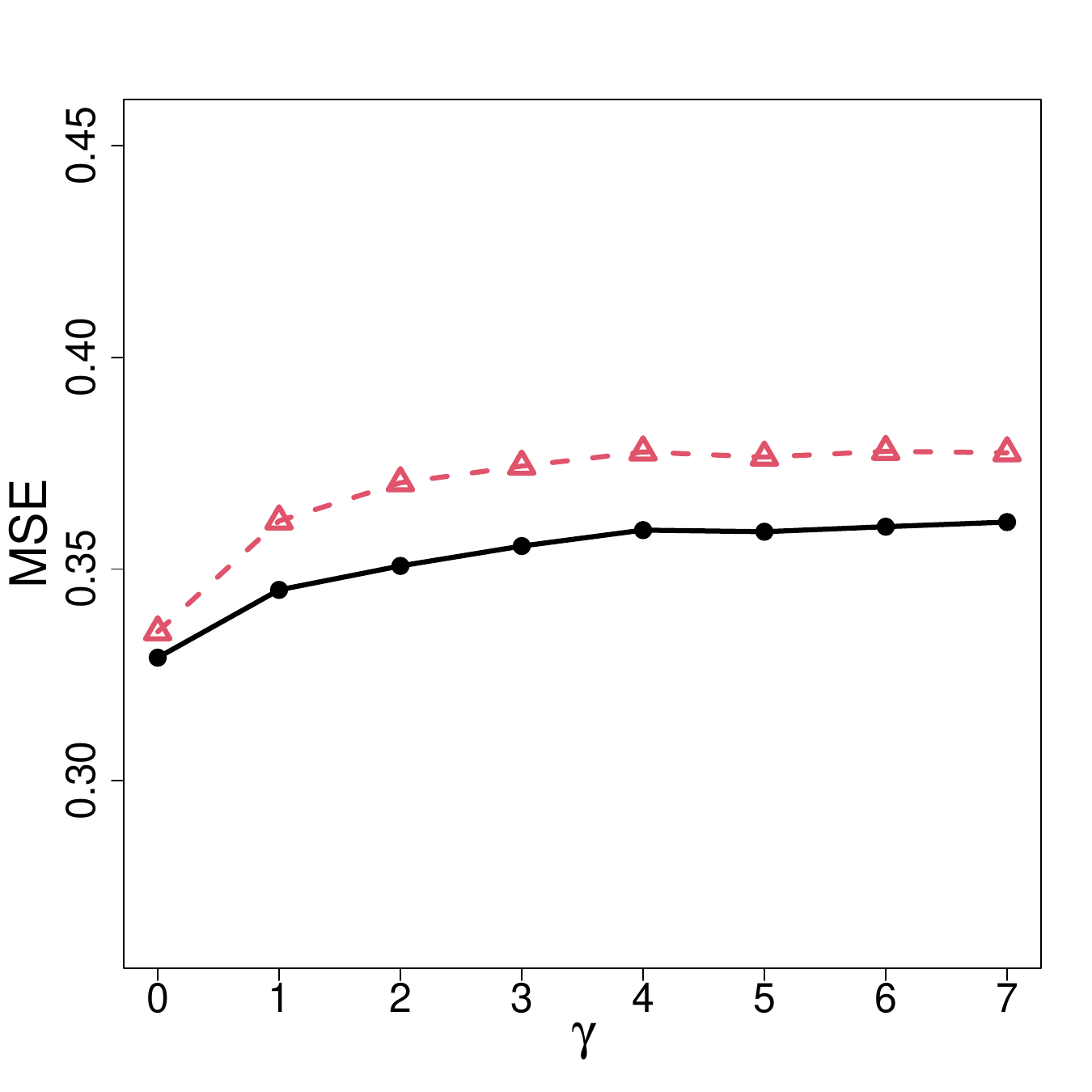} &\includegraphics[width=.31\textwidth]{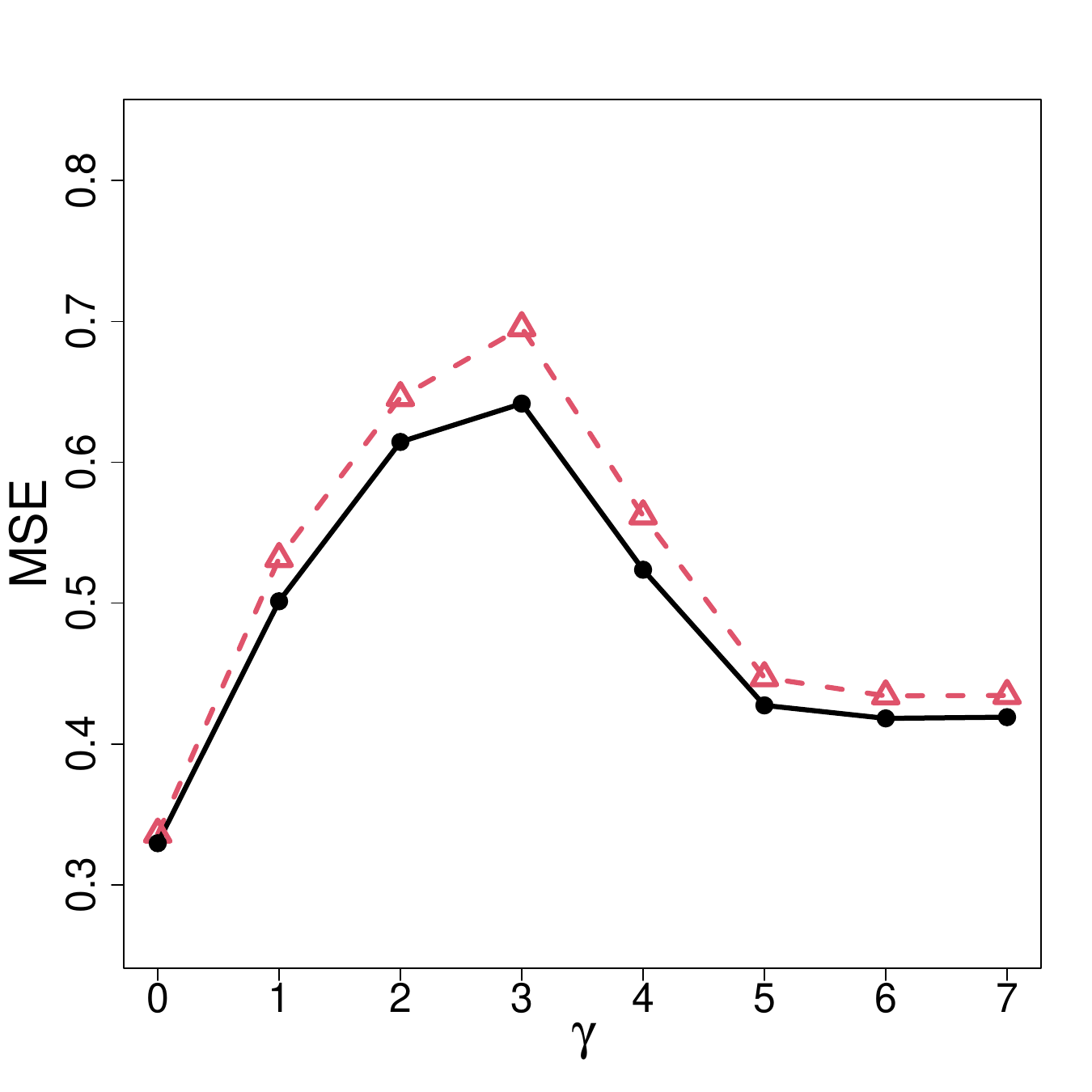}\\
 \rotatebox{90}{\textbf{\footnotesize{$p=q=50$}}}&\includegraphics[width=.31\textwidth]{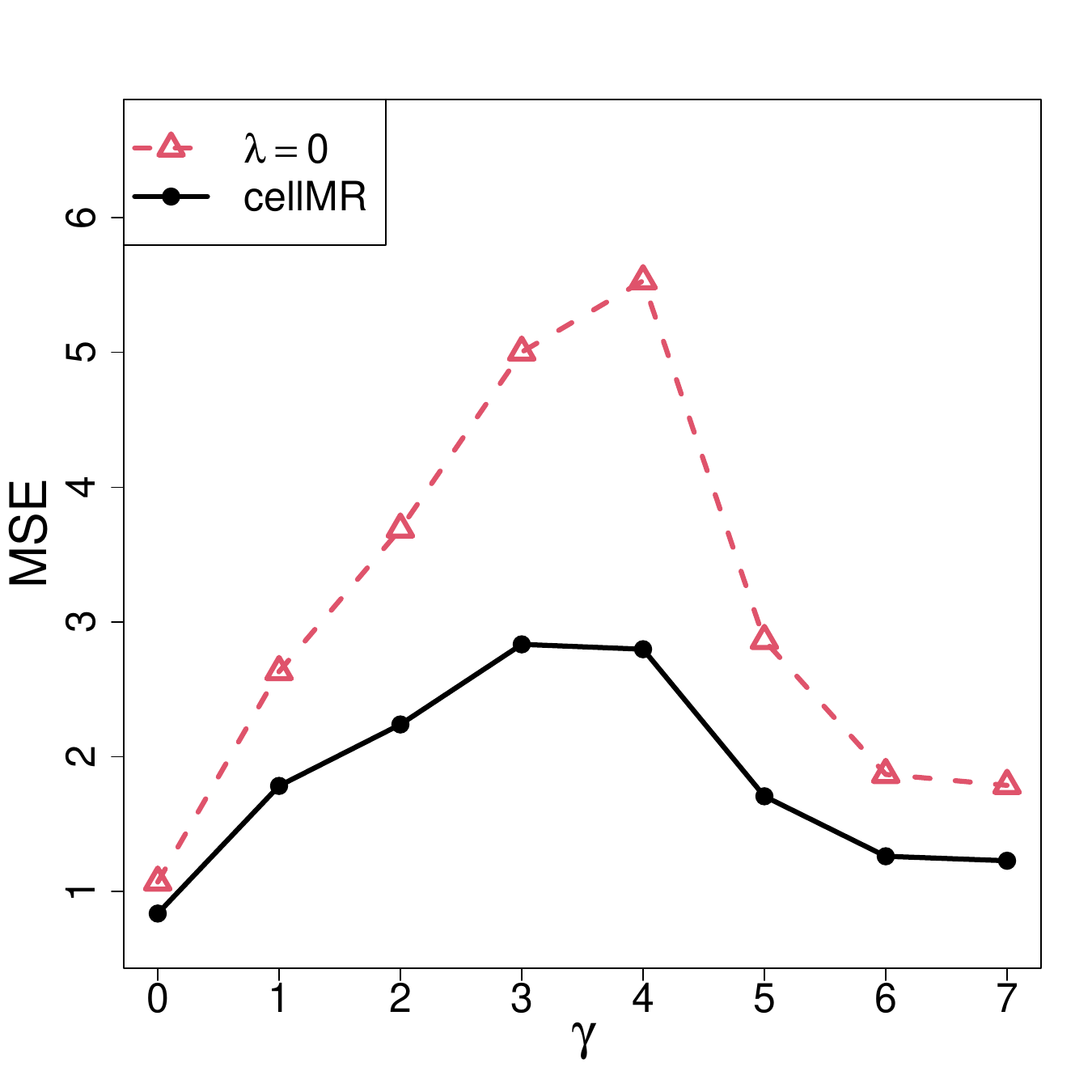} &\includegraphics[width=.31\textwidth]{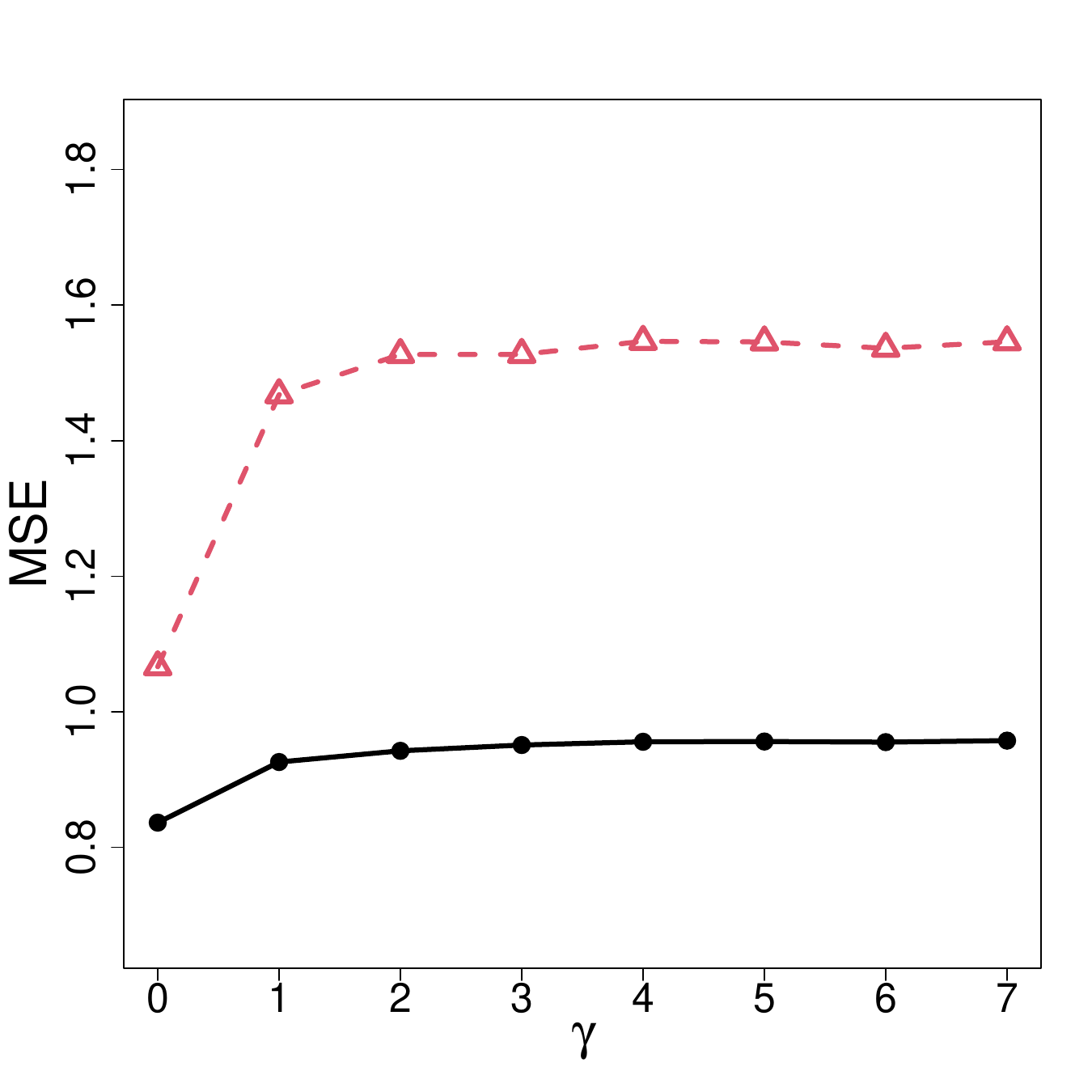} &\includegraphics[width=.31\textwidth]{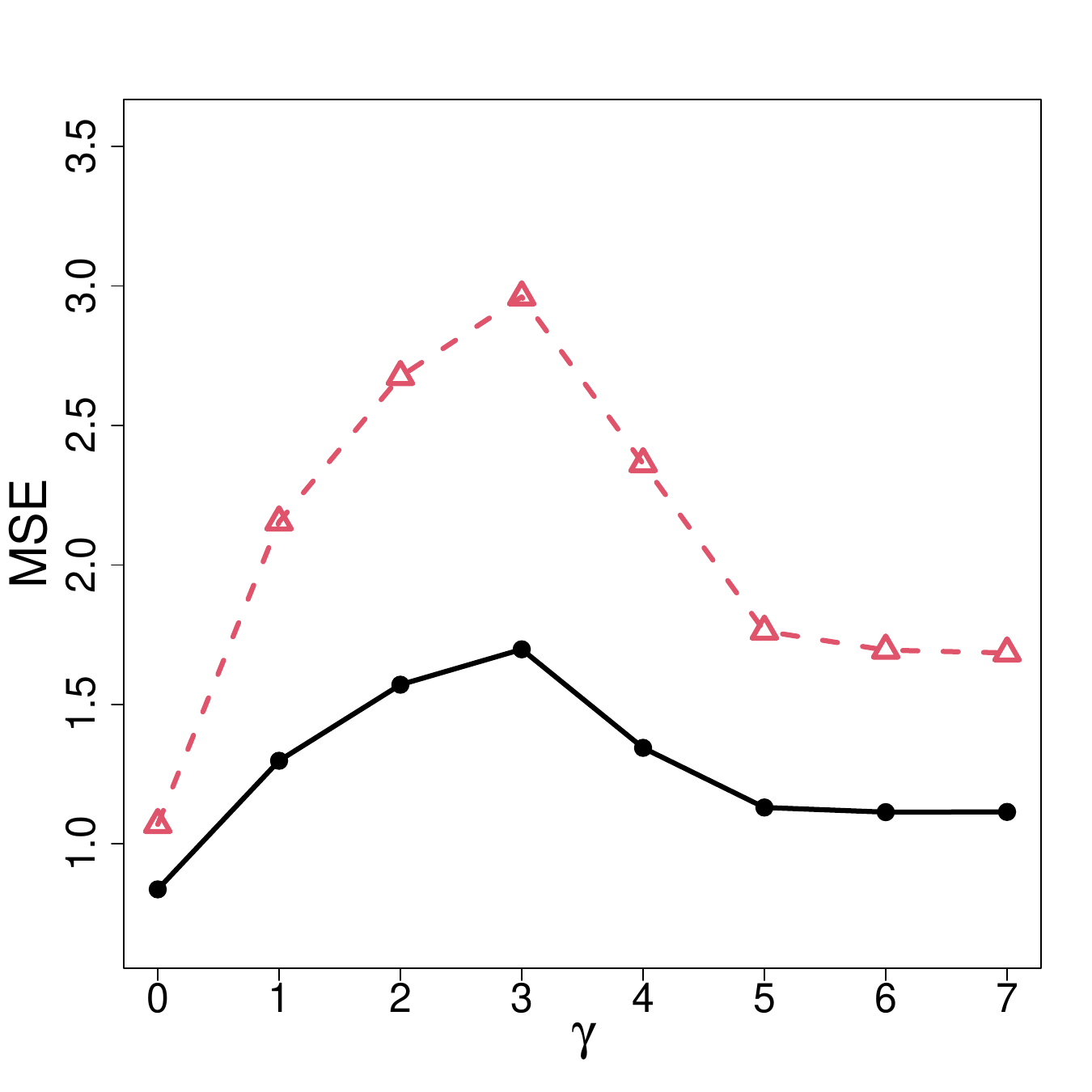}\\
\end{tabular}

\vspace{-5mm}
\caption{Average MSE attained by cellMR and 
its unregularized counterpart ($\lambda=0$) 
in the presence of cellwise outliers, casewise 
outliers, or both, without missing data.}
\label{fig:results_lambda}
\end{figure}

\newpage
Similarly, Figure~\ref{fig:results_CV} 
compares the average MSE attained by cellMR with 
that of the same estimator using nonrobust 
cross-validation to select the regularization 
parameter. The proposed cross-validation strategy
yields lower MSEs under contamination, with
particularly pronounced improvements for large 
outlier magnitudes and under casewise or mixed
contamination. These results highlight the 
importance of robust parameter selection.
\vspace{3mm}
\begin{figure}[H]
\centering
\begin{tabular}{M{0.0005\textwidth}M{0.29\textwidth}M{0.29\textwidth}M{0.32\textwidth}}
   &\large \textbf{Cellwise}  & \large \textbf{Casewise} &\large{\textbf{Casewise \& Cellwise}} \\
[-4mm]

 \rotatebox{90}{\textbf{\footnotesize{$p=q=10$}}}&\includegraphics[width=.31\textwidth]
 {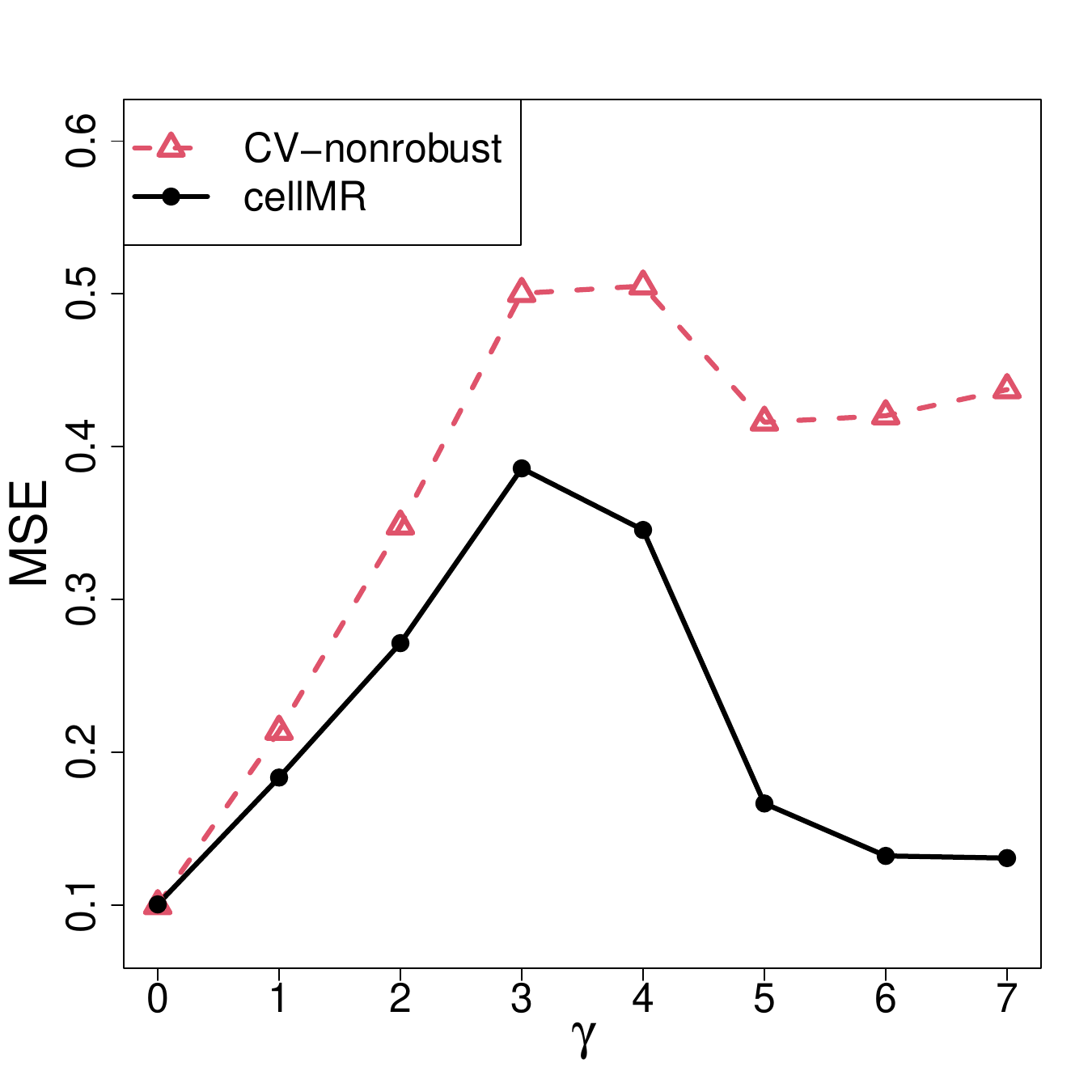} &\includegraphics[width=.31\textwidth]{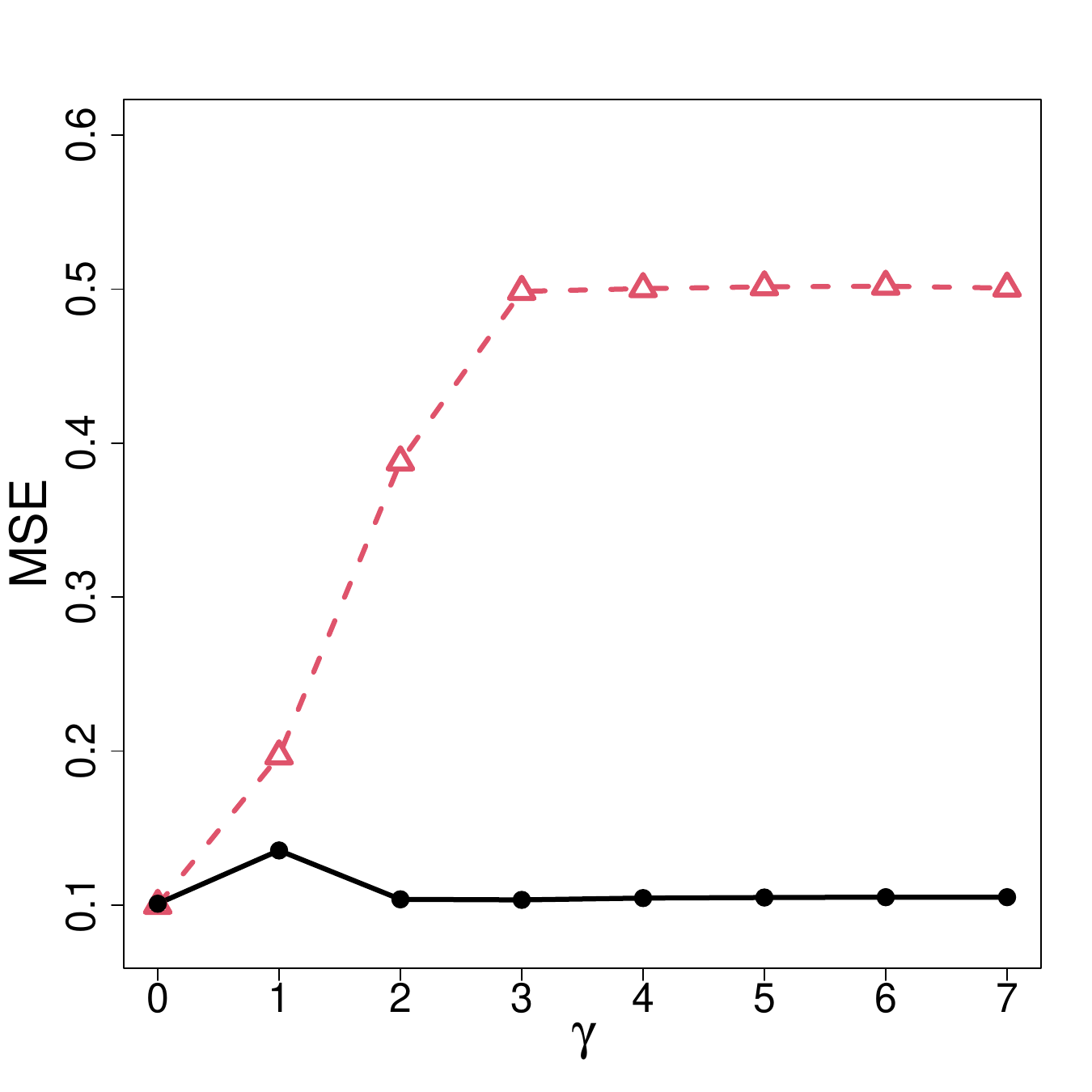} &\includegraphics[width=.31\textwidth]{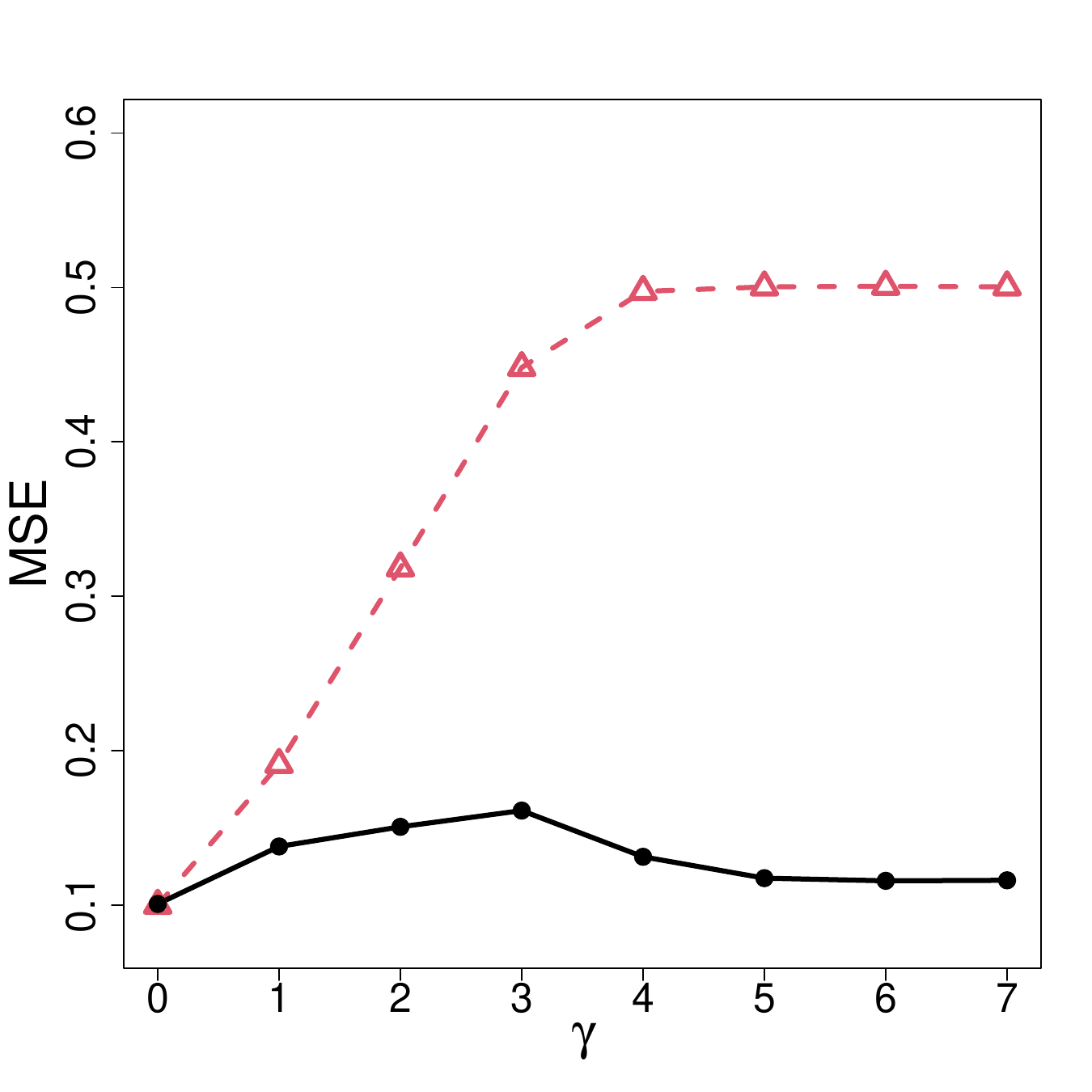} \\
 
 \rotatebox{90}{\textbf{\footnotesize{$p=q=25$}}}&\includegraphics[width=.31\textwidth]{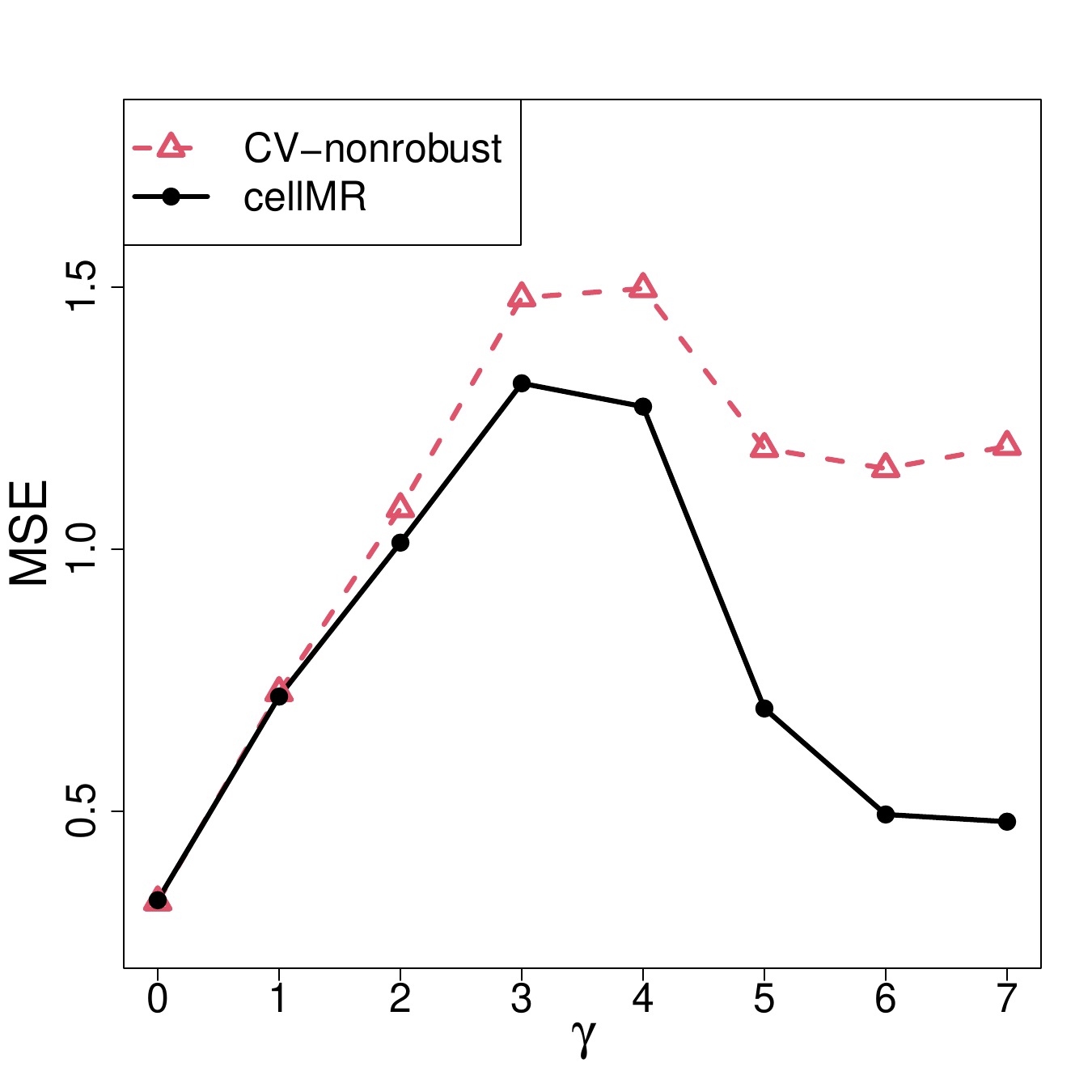} &\includegraphics[width=.31\textwidth]{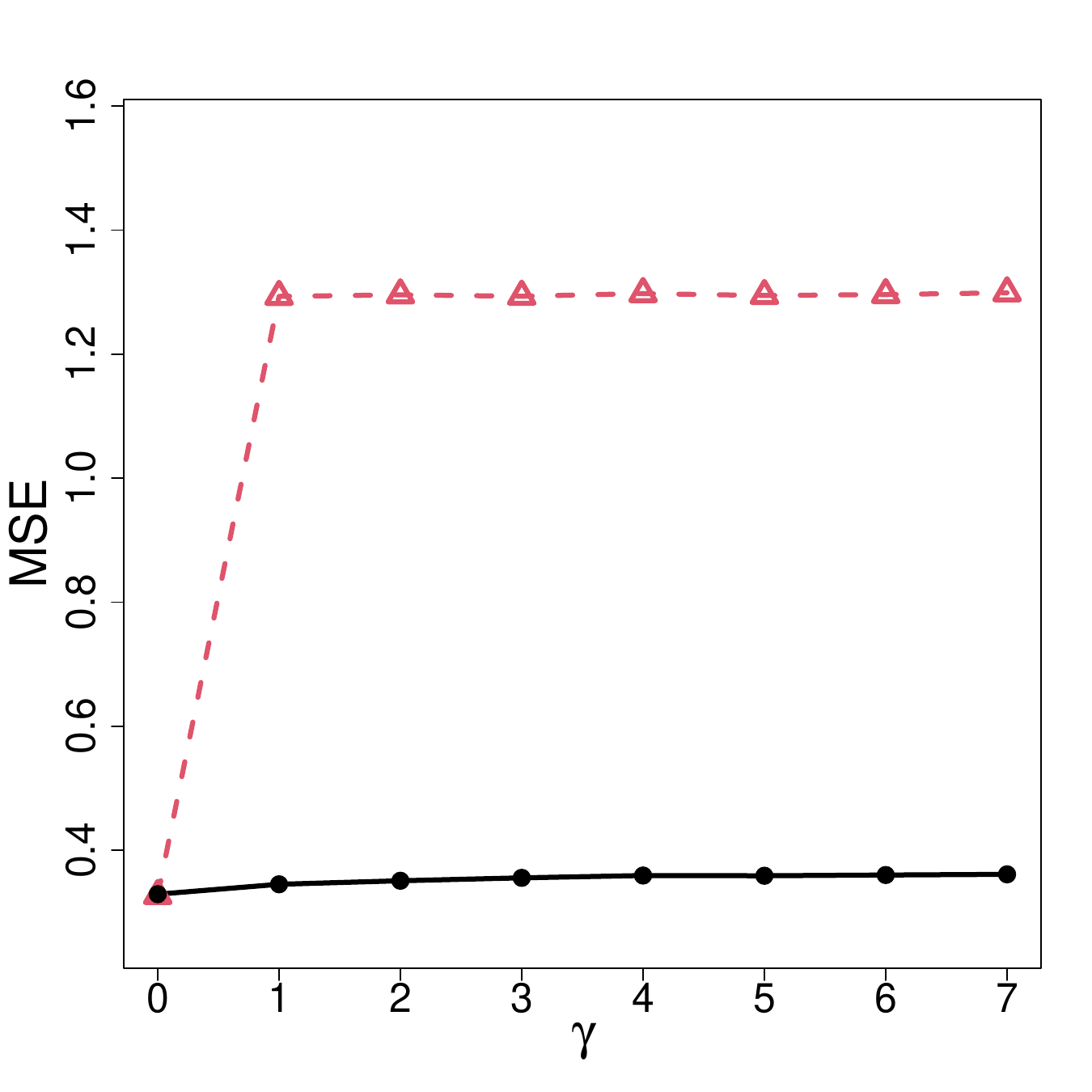} &\includegraphics[width=.31\textwidth]{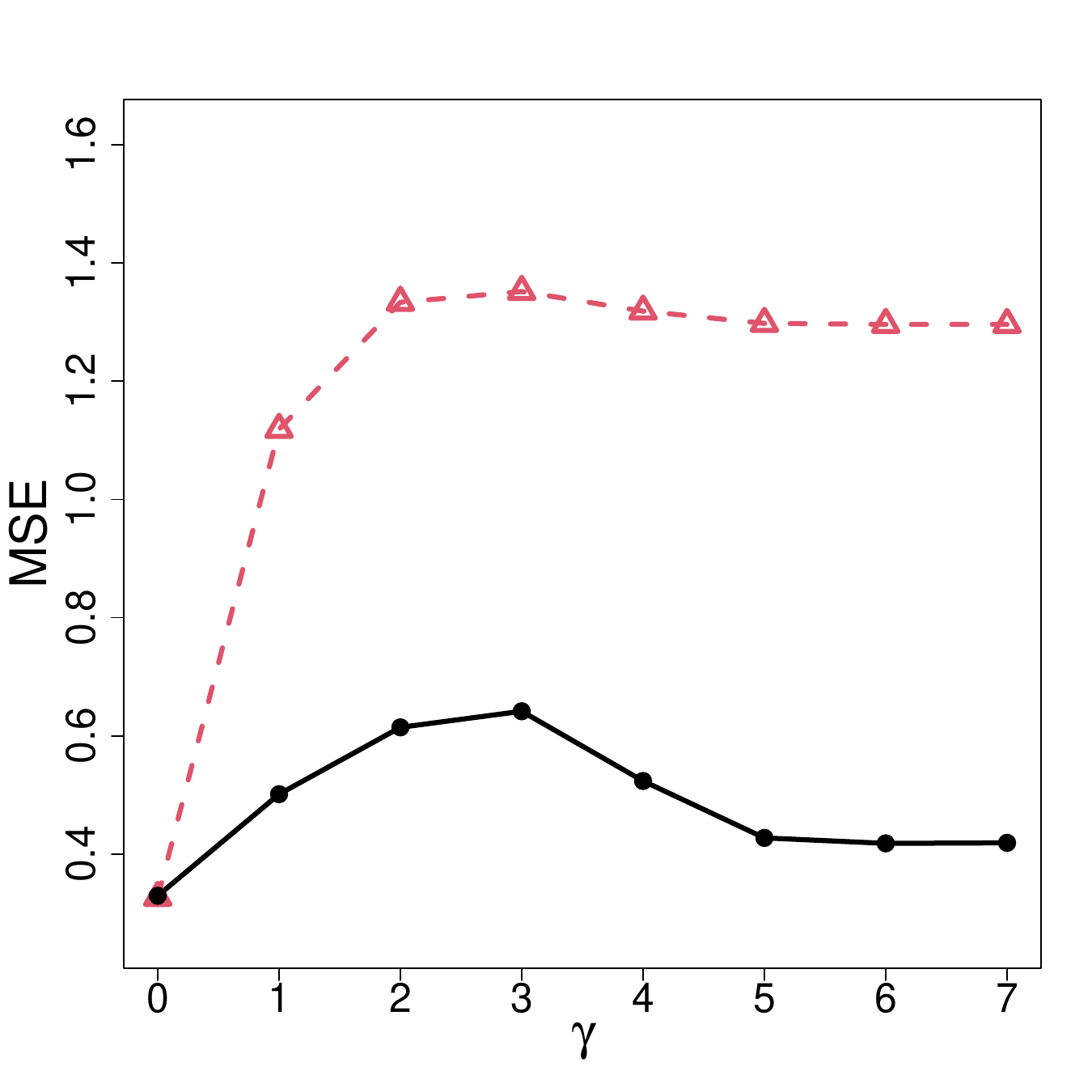}\\
 \rotatebox{90}{\textbf{\footnotesize{$p=q=50$}}}&\includegraphics[width=.31\textwidth]{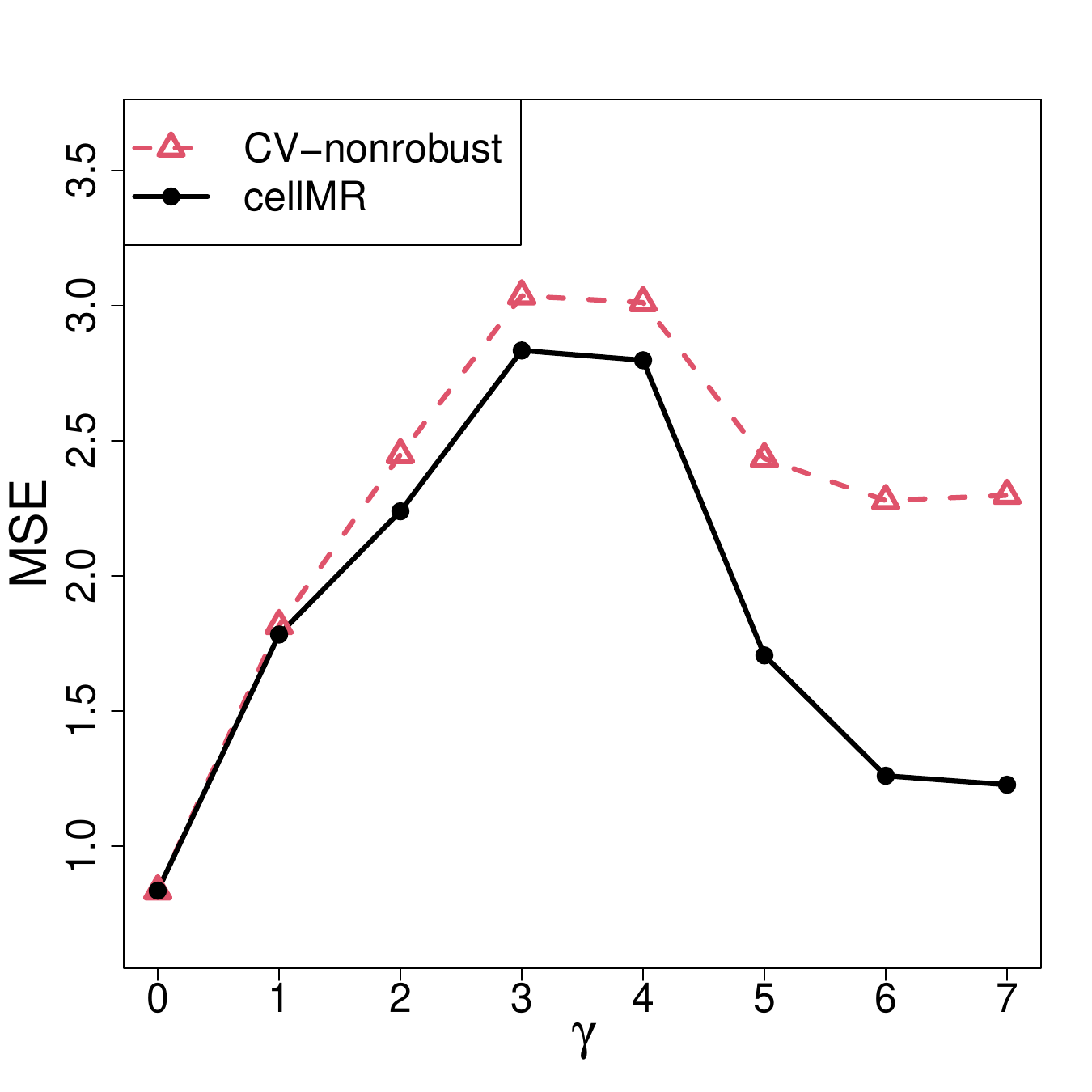} &\includegraphics[width=.31\textwidth]{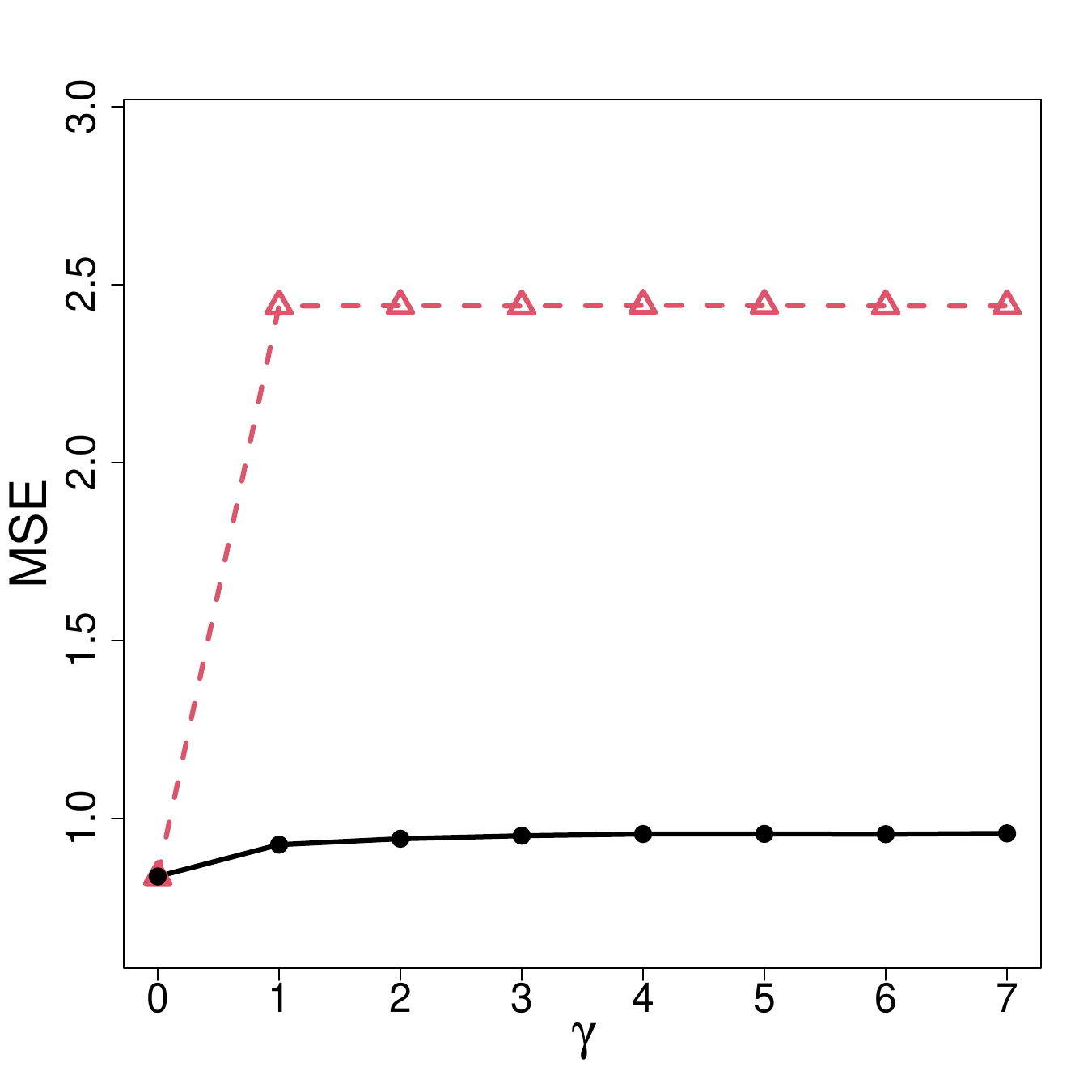} &\includegraphics[width=.31\textwidth]{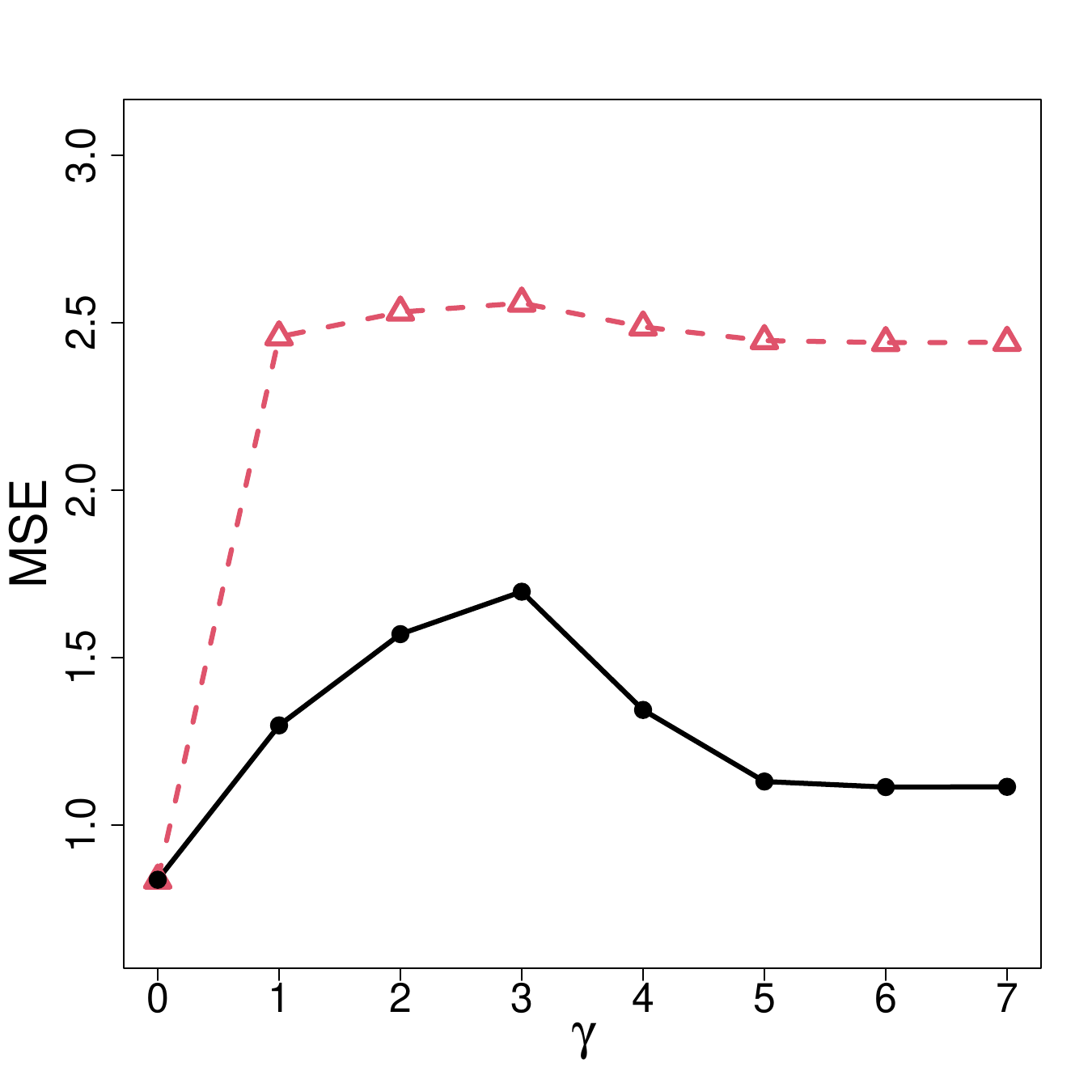}\\
\end{tabular}

\vspace{-5mm}
\caption{Average MSE attained by cellMR and its counterpart using nonrobust
cross-validation in the presence of cellwise outliers, casewise
outliers, or both, without missing data.}
\label{fig:results_CV}
\end{figure}

\section{Computational complexity and scalability of cellMR and cellBoot}
\label{sec:supp-complexity-cellRCov}
 \subsection{Theoretical complexity}
 \label{sec:supp-complexity-theor}

\begin{proof} [\textbf{Proof of Proposition~\ref{prop:complexity-cellMR}}]

The first step of cellMR consists of applying
cellCov to the combined $n \times d$ data matrix.
Its computational complexity follows from the
analysis of cellRCov in
\citet{centofanti2025cellwise}, and it is
$O\!\left(nd^2+nd\log(n)\right)$.

For a given $\lambda$, forming
$\bhSigma_x+\lambda\bI_p$ requires $O(p)$
operations, since only the diagonal entries
are modified.
Computing its inverse costs $O(p^3)$,
and multiplying the resulting $p \times p$
matrix by the $p \times q$ matrix
$\bhSigma_{xy}$ costs $O(p^2q)$.
To compute the intercept
$\bhb=\bhmu_y-\bhB^T\bhmu_x$,
multiplying the $q \times p$ matrix $\bhB^T$
by the $p$-dimensional vector $\bhmu_x$
costs $O(pq)$, and subtracting the result
from $\bhmu_y$ costs $O(q)$.

Combining these steps, the total computational
complexity of estimating the slope and intercept
for a given $\lambda$ is
\[
O\!\left(
nd^2+nd\log(n)+p+p^3+p^2q+pq+q
\right),
\]
which simplifies to
\[
O\!\left(
nd^2+nd\log(n)+p^3+p^2q
\right).
\]

When $p\leqslant n$, the term $p^3+p^2q=p^2d$
is dominated by $nd^2$. Therefore, the
computational complexity simplifies to
\[
O\!\left(nd^2+nd\log(n)\right).
\]
\end{proof}

For comparison, using dense $p\times p$
matrix inversion, the computational complexity
of classical ridge regression, for a given
$\lambda>0$, is $O\!\left(np^2+npq+p^3+p^2q\right)$.
When $p\leqslant n$, this simplifies to
$O(np^2+npq)$.
The additional cost of cellMR arises from the cellCov computation on the combined predictor
and response matrix.
In particular, when $p\leqslant n$ and $q=O(p)$,
these bounds become
$O(np^2)$ for classical ridge regression and
$O(np^2+np\log(n))$ for cellMR.

\begin{proof} [\textbf{Proof of Proposition~\ref{prop:complexity-cellBoot}}]
We first determine the cost of one FastCellCov
evaluation, and then combine this with the
remaining steps of cellBoot.

Computing $\bz_i^*=\bhD^{-1}\btz_i^*$
requires $O(d)$ operations per observation
because $\bhD$ is diagonal.
For all observations, this gives $O(nd)$ operations.
Similarly, computing
$\bz_{s,i}^*=\bhS^{-1}(\bz_i^*-\bhmu_z)$
and the filtering weights $w^f_{ij}$
requires $O(nd)$ operations.
The diagonal entries of $\bhS$ can be
computed once from the diagonal entries
of $\bhSigma$ and $\bhD$ in $O(d)$ time.

Next we compute the predictions
in~\eqref{eq:predicted}.
For each observation $i$ and variable $j$,
the numerator and denominator each require
a sum over $|H_j|$ variables, followed by
a scalar division.
Repeating this for all observations and
variables costs $O\!\left(nd+n\sum_{j=1}^d |H_j|\right)$.
Since $|H_j|\leqslant d$, this is at most $O(nd^2)$.
The sets $H_j$ can be constructed in $O(d^2)$ time.

Multiplying each predicted value
$\hz_{s,ij}^{\,*}$ by $a_j^{shr}$ requires
$O(nd)$ operations.
Computing the standardized residuals and
their weights, and constructing the imputed
observations $\bz_{i,\mathrm{imp}}^*$ are
also elementwise operations, each requiring
$O(nd)$ time.

The fitted values are $\bhz_i^{\,*}
=
\bhmu_z+
\bhV\bhV^T(\bz_{i,\mathrm{imp}}^*-\bhmu_z)$.
Computing $\bz_{i,\mathrm{imp}}^*-\bhmu_z$
requires $O(d)$ operations.
Multiplying $\bhV^T$ by
the centered observation
requires $O(dk)$ operations.
Multiplying the resulting score vector
by $\bhV$  also costs $O(dk)$,
and adding $\bhmu_z$ costs $O(d)$.
For all observations, the total cost is $O(ndk)$. As $\rk\leqslant \rk_{\max}$, this becomes $O(nd)$.

Computing the residuals, the cellwise weights
$w_{*,ij}^{\cell}$, and the matrix $\btW^*$
requires $O(nd)$ operations.
The casewise weights $w_{*,i}^{\case}$
also require $O(nd)$ operations because
their computation involves sums over
$d$ variables for each observation.

The subspace weights $w_{*,i}^{\mathrm{sub}}$
can be computed from the score vectors
$\bhV^T(\bz_{i,\mathrm{imp}}^*-\bhmu_z)$,
which are already available.
The MCD location and scatter estimates
of the original scores are kept fixed.
The matrix $\bhSigma_{\bhz}$ can be
inverted once in $O(k^3)$ time.
Evaluating the corresponding quadratic
form for each new score vector then
requires $O(k^2)$ operations, giving
$O(nk^2)$ for all subspace weights.
As $\rk\leqslant \rk_{\max}$, this becomes $O(n)$.

Next we compute the new center
in~\eqref{eq:btmuF}.
Since each $\btW_i^*$ is diagonal,
forming and summing the weighted observations
requires $O(nd)$ operations.
Computing the normalizing constants
$c_j^{\bmu}$ also costs $O(nd)$,
and dividing the resulting coordinates
by these constants costs $O(d)$.
The total cost of the new center is
therefore $O(nd)$.

To compute the covariance in~\eqref{eq:btSigmaF},
we first compute $\btW_i^*(\bz_i^*-\bhmu_z)$.
This costs $O(d)$ per observation because
$\btW_i^*$ is diagonal.
Forming each weighted outer product costs
$O(d^2)$, so computing and summing these
products over $n$ observations costs $O(nd^2)$.
Computing the normalizing matrix $\bC^{\bSigma}$
also requires $O(nd^2)$ operations.
The final elementwise division costs $O(d^2)$.
Thus, computing the new covariance costs
$O(nd^2)$.

Finally, transforming the center and
covariance back to the original scale
requires $O(d)$ and $O(d^2)$ operations,
respectively, because $\bhD$ is diagonal.

Combining these steps, the complexity
of one FastCellCov evaluation is $O(nd^2)$.
This proves the first assertion.

The initial cellMR fit requires
$O(nd^2+nd\log(n)+p^3+p^2q)$
operations and supplies the quantities
needed for the subsequent FastCellCov
evaluations.
For each bootstrap replication, constructing
a sample of $n$ observations costs $O(nd)$.
Computing its FastCellCov estimates
$\bhmu_{F,b}$ and $\bhSigma_{F,b}$
then costs $O(nd^2)$ by the first part
of the proof.

Next we apply indirect inference using
the update in~\eqref{eq:thetaupdate}.
At each iteration, generating Gaussian samples
requires a factorization of the current
$d\times d$ covariance matrix, which costs
$O(d^3)$ and is computed once for all $H$
simulated samples.
Each observation is generated by drawing $d$
independent standard normal variables, at
an expected cost of $O(d)$, multiplying the
resulting vector by the covariance factor,
at a cost of $O(d^2)$, and adding the current
location vector, at a cost of $O(d)$.
Generating $n$ observations for each of the
$H$ samples therefore costs $O(Hnd^2)$.
Including the factorization, the total cost
is $O(d^3+Hnd^2)$.
Applying FastCellCov to each of the
$H$ simulated samples costs $O(Hnd^2)$.

Each auxiliary estimate contains a location
vector of length $d$ and a scatter matrix
with $d(d+1)/2$ distinct entries.
Computing their average
$\pibar(\bhtheta_n^{(\ell-1)},\boheta_n,n)$
therefore costs $O(Hd^2)$.
The vector additions and subtraction
in~\eqref{eq:thetaupdate} cost $O(d^2)$.
The projection $\Pi_{\bTheta}$ requires
an eigendecomposition of a $d\times d$
symmetric matrix, clipping its eigenvalues,
and reconstructing the scatter matrix.
These operations together cost $O(d^3)$.
Clipping the location vector costs $O(d)$
and does not increase this order.

Adding these costs gives
\[
O(d^3+Hnd^2)+O(Hnd^2)+O(Hd^2)
+O(d^2)+O(d^3)
=
O(Hnd^2+d^3),
\]
per iteration.
Since $H$ and the maximum number of
iterations are fixed, the entire
indirect-inference correction for one
bootstrap sample costs $O(nd^2+d^3)$.

From the corrected location and scatter
estimates, computing the slope and intercept
requires $O(p^3+p^2q)$ operations.
Since
$p^3+p^2q=p^2(p+q)=p^2d\leqslant d^3$,
this cost is absorbed by $O(nd^2+d^3)$.
Evaluating the scalar bootstrap estimate
$\htheta_b^*$
requires $O(pq)$ operations and does not
increase the overall order.

Repeating these computations for a fixed
number $B$ of bootstrap samples still
costs $O(nd^2+d^3)$.
Computing the final empirical quantiles
by sorting costs $O(B\log(B))$, which
is constant.

Including the initial fit, the total
complexity is therefore
\begin{align*}
& O(nd^2+nd\log(n)+p^3+p^2q)
  +O(nd^2+d^3) \\
& \qquad =
O(nd^2+nd\log(n)+d^3).
\end{align*}
Finally, if $d\leqslant n$, then
$d^3\leqslant nd^2$, yielding $O\!\left(n(p+q)^2+n(p+q)\log(n)\right)$.
\end{proof}

FastCellCov avoids repeating the cellPCA
fit, its initialization, and MCD estimation
within each indirect-inference iteration.
Its per-sample complexity is $O(nd^2)$,
compared with $O(nd^2+nd\log(n))$ for cellCov.

\subsection{Empirical runtimes and scalability}
\label{app:runtime}
Figure~\ref{fig:results_time} shows the 
average computation times of cellMR in seconds 
over 100 replications.
The computations were run on 
a parallel machine with
IceLake thin nodes. Each node contains two 
Intel Xeon Platinum 8360Y CPUs at 2.4 GHz, 
with 36 cores per CPU and 256 GB RAM.
The figure shows computation times as a function
of $n$ for fixed choices of $p=q$, and as a 
function of $p=q$ for fixed choices of $n$.
For fixed $p=q$, the computation time increases
about linearly with $n$ over the investigated range.
For fixed $n$, it grows approximately cubically 
with $p=q$, consistent with the theoretical 
complexity analysis.
\begin{figure}[!ht]
\centering
\includegraphics[width=.4\textwidth]
  {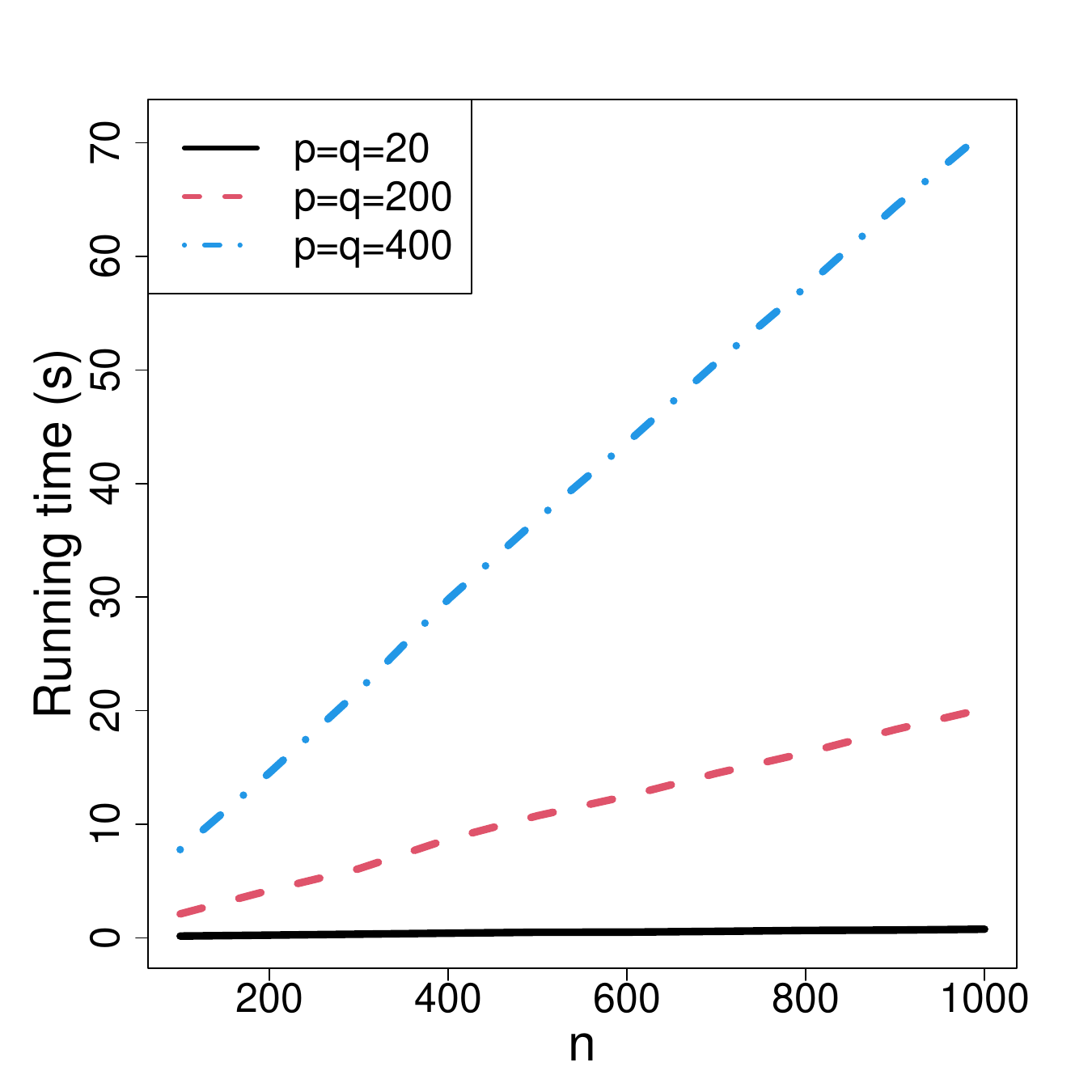}
\includegraphics[width=.4\textwidth]
  {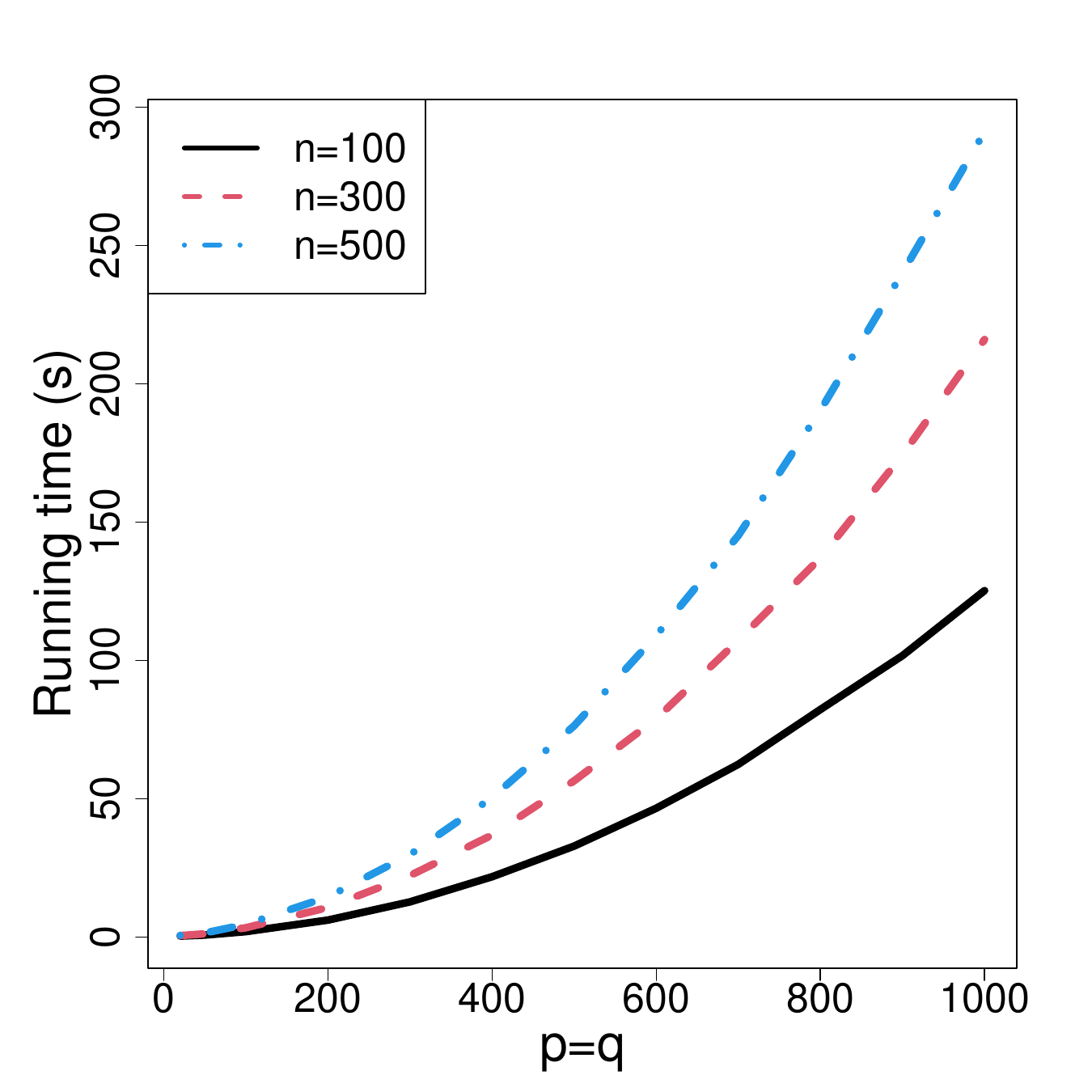}
\caption{Average cellMR runtimes in seconds,
as a function of the number of cases $n$ (left)
and the dimensions $p=q$ (right), over 100 
replications.}
\label{fig:results_time}
\end{figure}

To compare the runtime of cellMR with 
existing regression methods, 
Figure~\ref{fig:runtime_cellMR} shows boxplots of 
the runtimes in seconds over 100 Monte
Carlo replications of RIDGE, SEST, PENSE, CRM, 
REGCELL, SHOOT, and cellMR. Data are generated without 
contamination or missing values, with $n=100$ and 
$p=20$.
\begin{figure}[H]
\centering
\includegraphics[width=.75\textwidth]{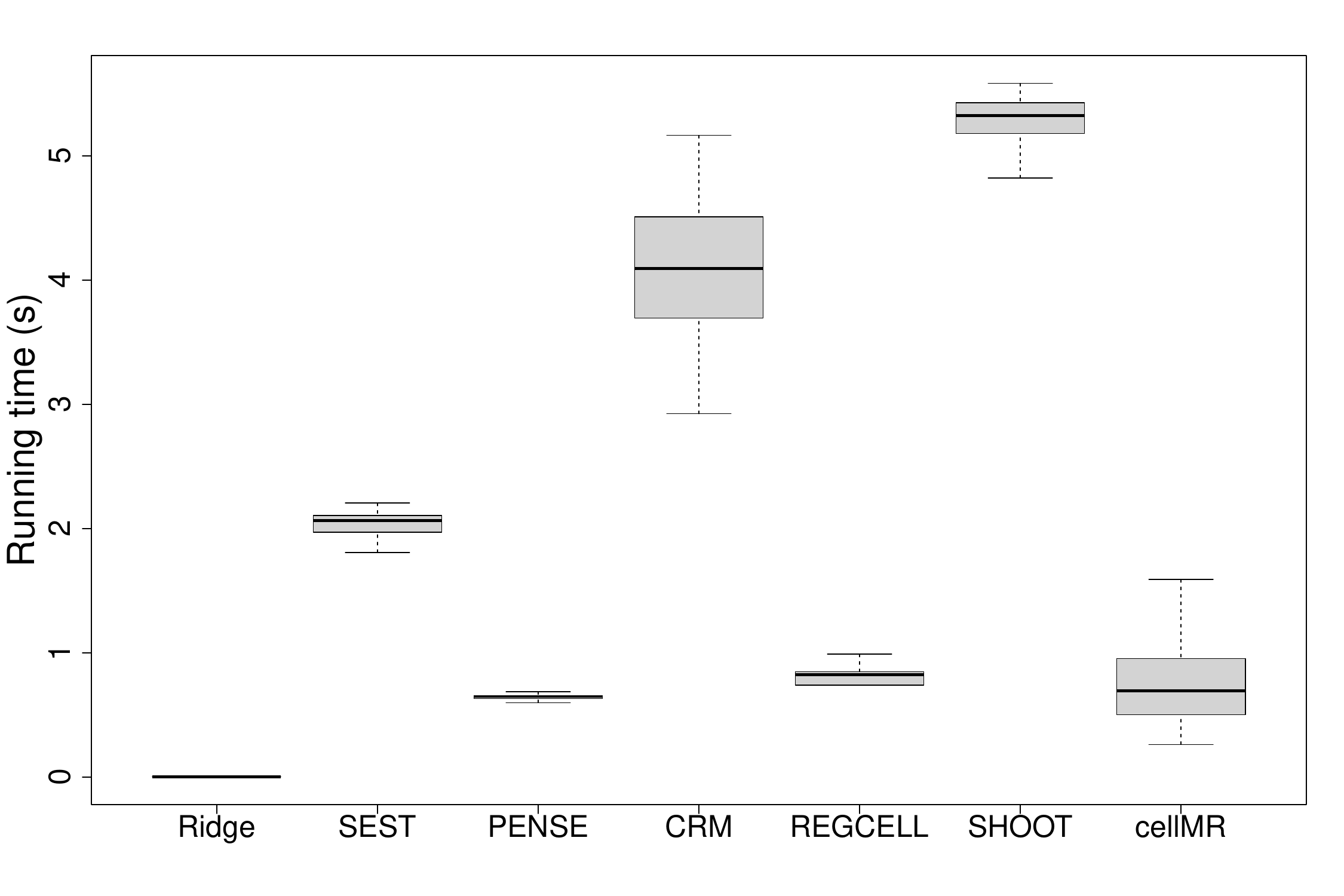}
\caption{Runtimes in seconds of
RIDGE, SEST, PENSE, CRM, REGCELL, SHOOT, and cellMR
over 100 Monte Carlo replications, for $n=100$ 
and $p=q=20$.}
\label{fig:runtime_cellMR}
\end{figure}
CellMR has a median runtime comparable 
to those of PENSE and REGCELL, and substantially 
shorter than those of SEST, CRM, and SHOOT. As 
expected, nonrobust RIDGE is the fastest method.

Finally, to highlight the computational 
advantages of using FastCellCov instead of cellCov 
as the auxiliary estimator in the II algorithm,
Figure~\ref{fig:results_time_fastcellcov} shows 
boxplots of their runtimes in seconds over 
100 Monte Carlo replications. Data were generated 
without contamination or missing values, with 
$n\in\{100,400\}$ and $p=q\in\{20,400\}$. 
FastCellCov is 
substantially faster than cellCov in all settings
considered. These results complement the 
complexity analysis by demonstrating the practical
computational gains achieved by using
FastCellCov as the auxiliary estimator.
\begin{figure}[!ht]
\centering
\begin{tabular}{M{0.0005\textwidth}M{0.41\textwidth}M{0.41\textwidth}}
 & \hspace{0.3cm} \textbf{$p=q=20$} & \hspace{0.3cm} \textbf{$p=q=400$} \\
[-4mm]
\rotatebox{90}{\textbf{\footnotesize{$n=100$}}}
 & \includegraphics[width=.4\textwidth]{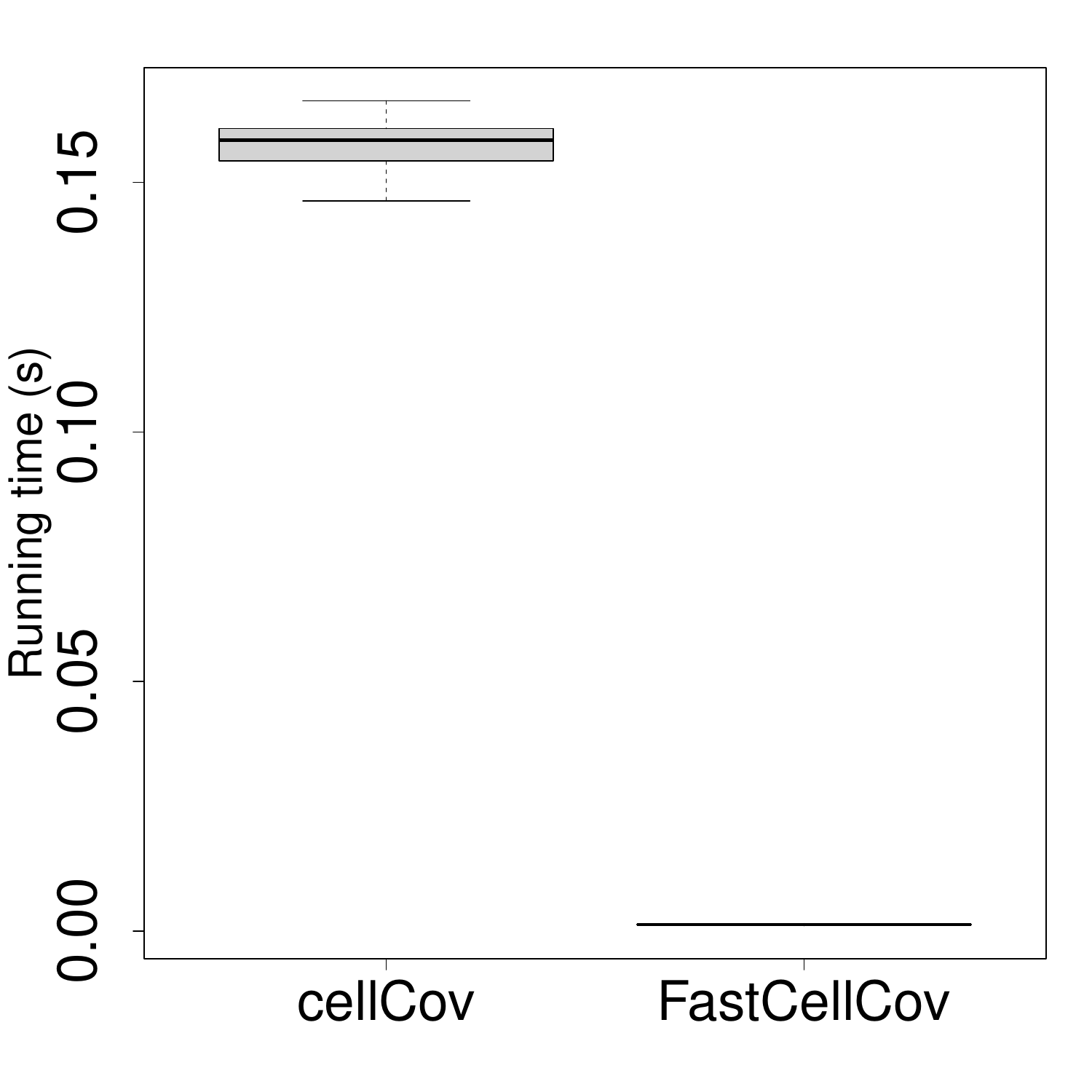}
 & \includegraphics[width=.4\textwidth]{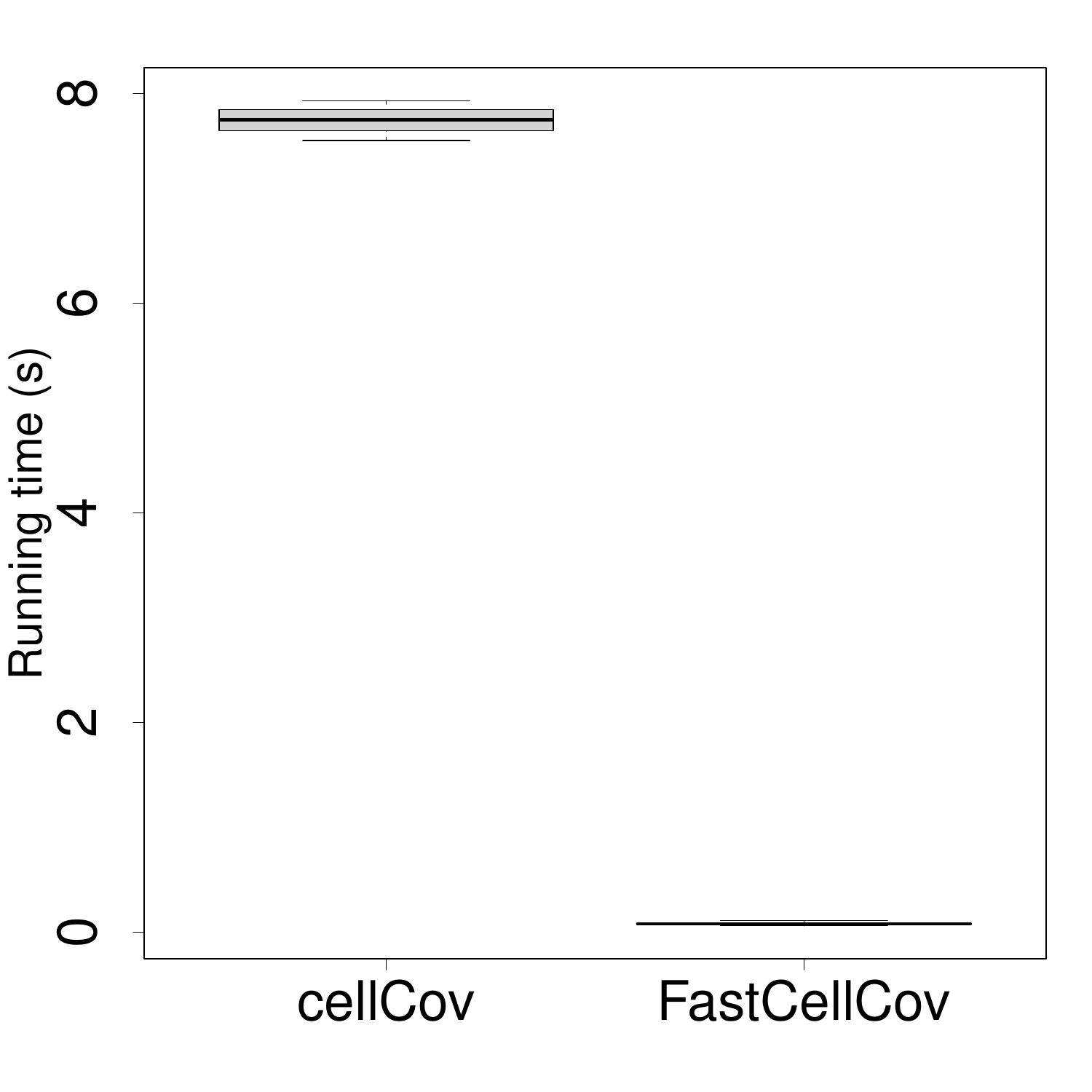} \\
\rotatebox{90}{\textbf{\footnotesize{$n=400$}}}
 & \includegraphics[width=.4\textwidth]{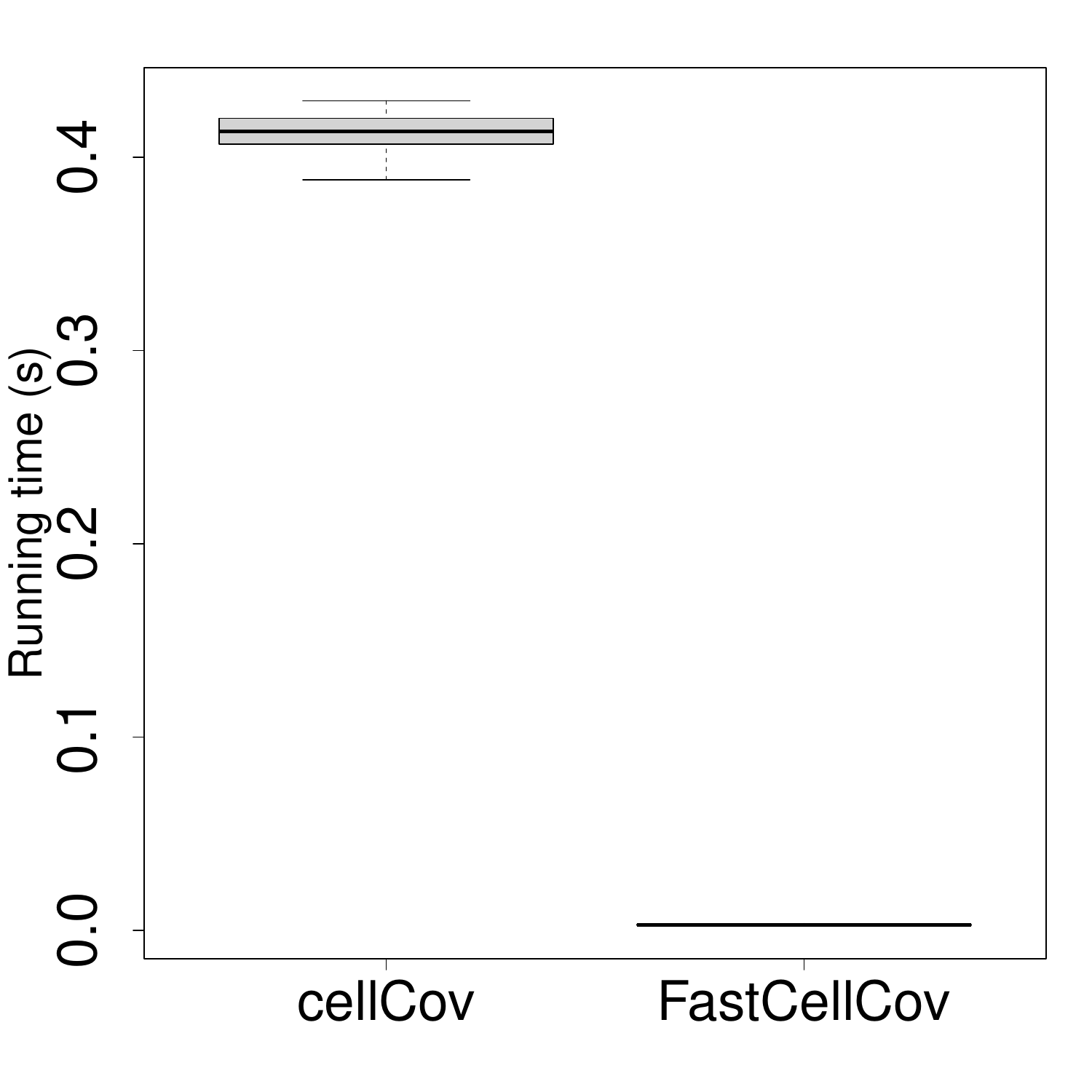}
 & \includegraphics[width=.4\textwidth]{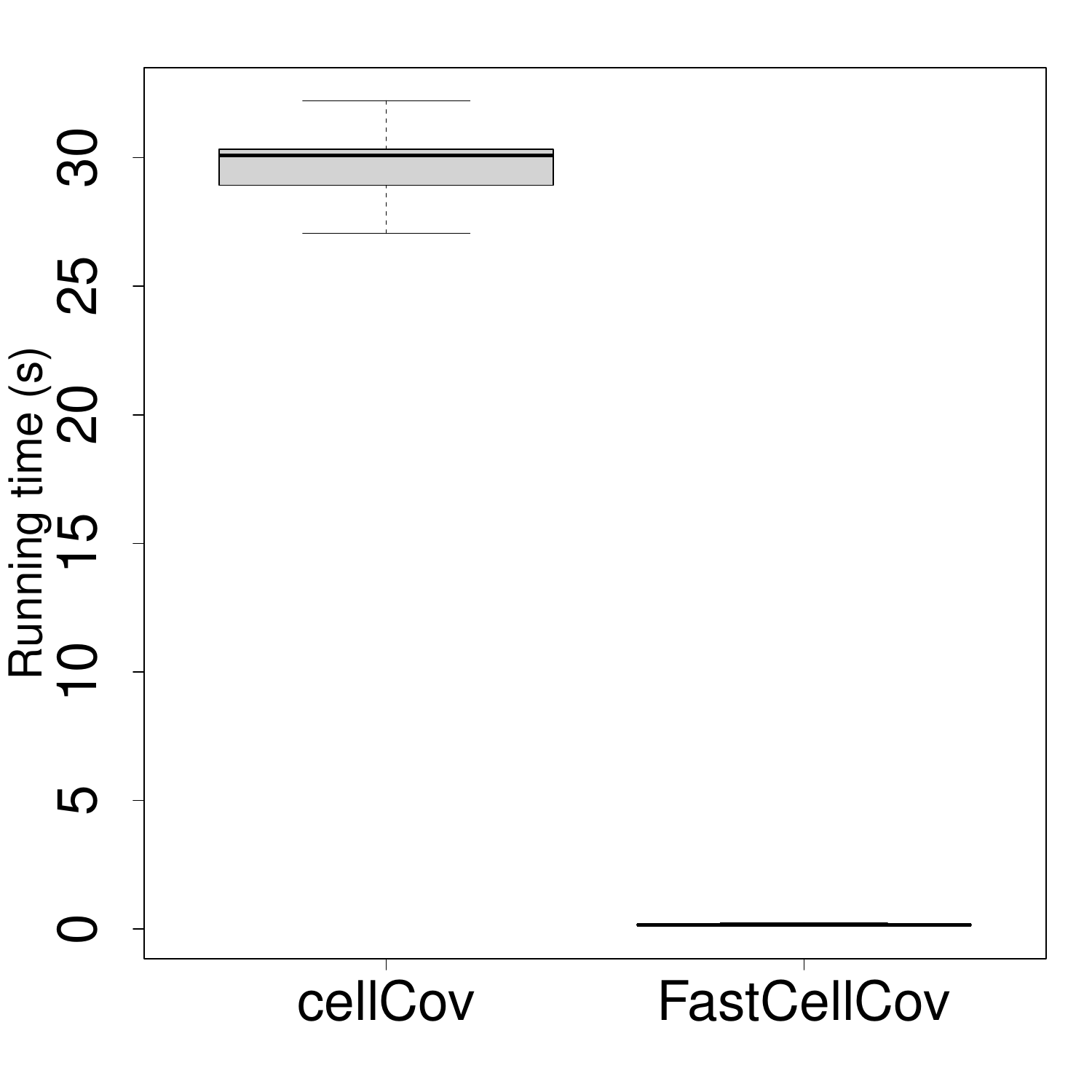}
\end{tabular}
\caption{Runtimes in seconds of cellCov and 
FastCellCov over 100 Monte Carlo replications,
for $n\in\{100,400\}$ and $p=q\in\{20,400\}$.}
\label{fig:results_time_fastcellcov}
\end{figure}

\section{Discussion on model misspecification}
\label{sec:misspecification}
To assess the performance of cellMR under departures
from linearity, we consider a simulation setting in which 
the response includes a quadratic component. We generate
\[
 \by_i
  = \bB^{T}\bx_i+a(\bB^T\bx_i)^{\circ 2}
  + \beps_i \quad \mbox{for} \quad
   i = 1,\ldots,n,
\]
where $\bx_i$, $\bB$, and $\beps_i$ are generated
as in Section~\ref{sec:sim}, and $\circ 2$ denotes
elementwise squaring. Writing 
$v_j=(\bB^{T}\bSigma_X\bB)_{jj}$\;, we set
\[
a=\sqrt{\frac{\delta\sum_{j=1}^{q}v_j}
{2\sum_{j=1}^{q}v_j^2}}\,.
\]
Therefore $\delta = 2a^2\sum_{j=1}^{q}v_j^2/\sum_{j=1}^{q}v_j
\geqslant 0$ is the ratio of the total variance of the
quadratic component to that of the linear component, with
$\delta=0$ recovering the linear model.
Since the predictors are centered Gaussian, the quadratic
component is uncorrelated with $\bX$. Therefore, when an
intercept is included, $\bB$ remains the coefficient matrix
of the best linear approximation even though the conditional
mean is nonlinear \citepsupp{buja2019models1}.

Figure~\ref{fig:results_misspecified} compares the average MSE
of the competing estimators under this nonlinear model, with
$p=q=10$ and $\delta\in\{0.1,0.2\}$.
CellMR retains good performance under contamination across
the levels of departure from linearity.
These results support the practical usefulness of cellMR under
the quadratic departures considered here.

\begin{figure}[H]
\centering
\begin{tabular}{M{0.0005\textwidth}M{0.29\textwidth}M{0.29\textwidth}M{0.32\textwidth}}
   &\large \textbf{Cellwise}  & \large \textbf{Casewise} &
   \large{\textbf{Casewise \& Cellwise}} \\
[-4mm]

 \rotatebox{90}{\textbf{\footnotesize{$\delta=0.1$}}}&
 \includegraphics[width=.31\textwidth]
 {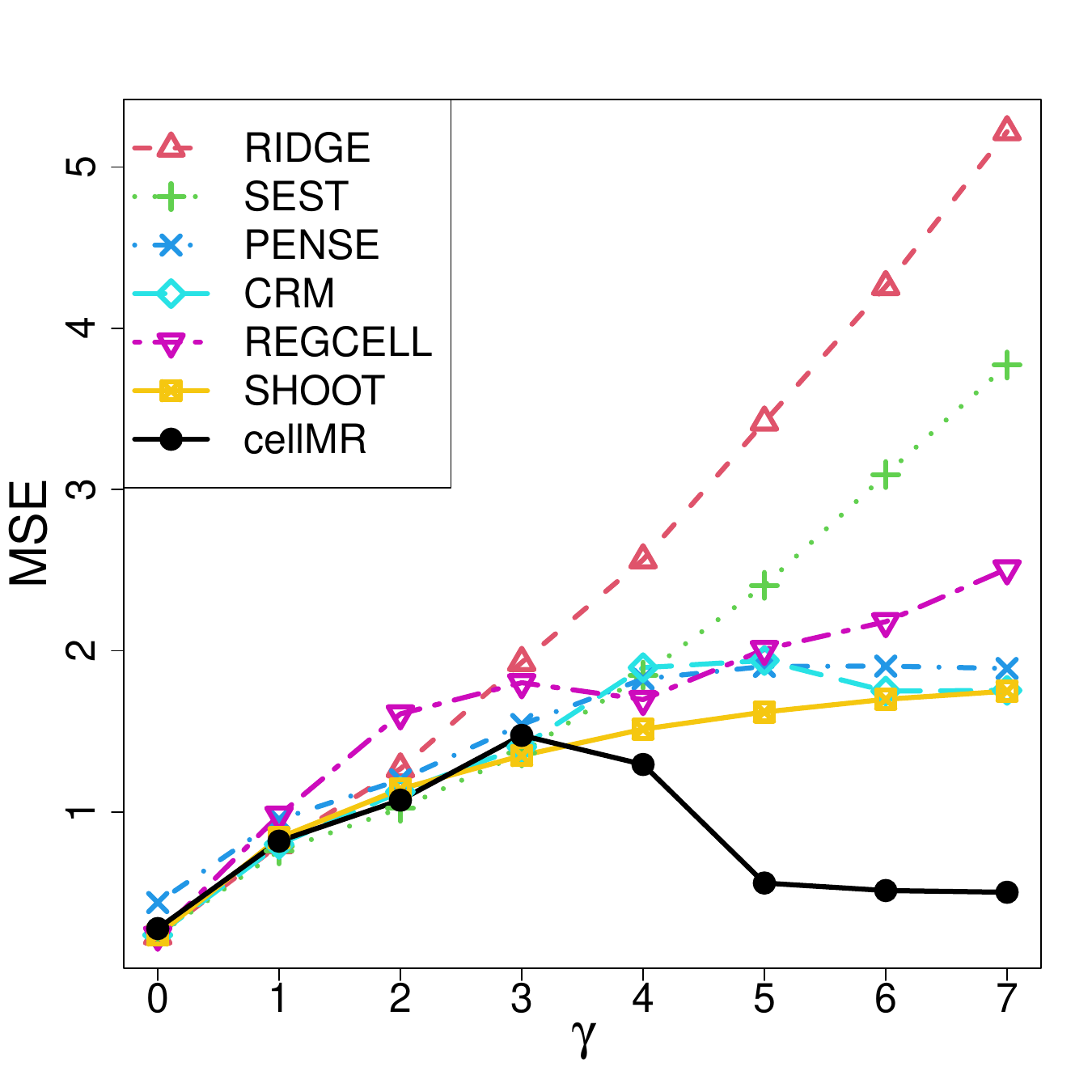} &
 \includegraphics[width=.31\textwidth]
 {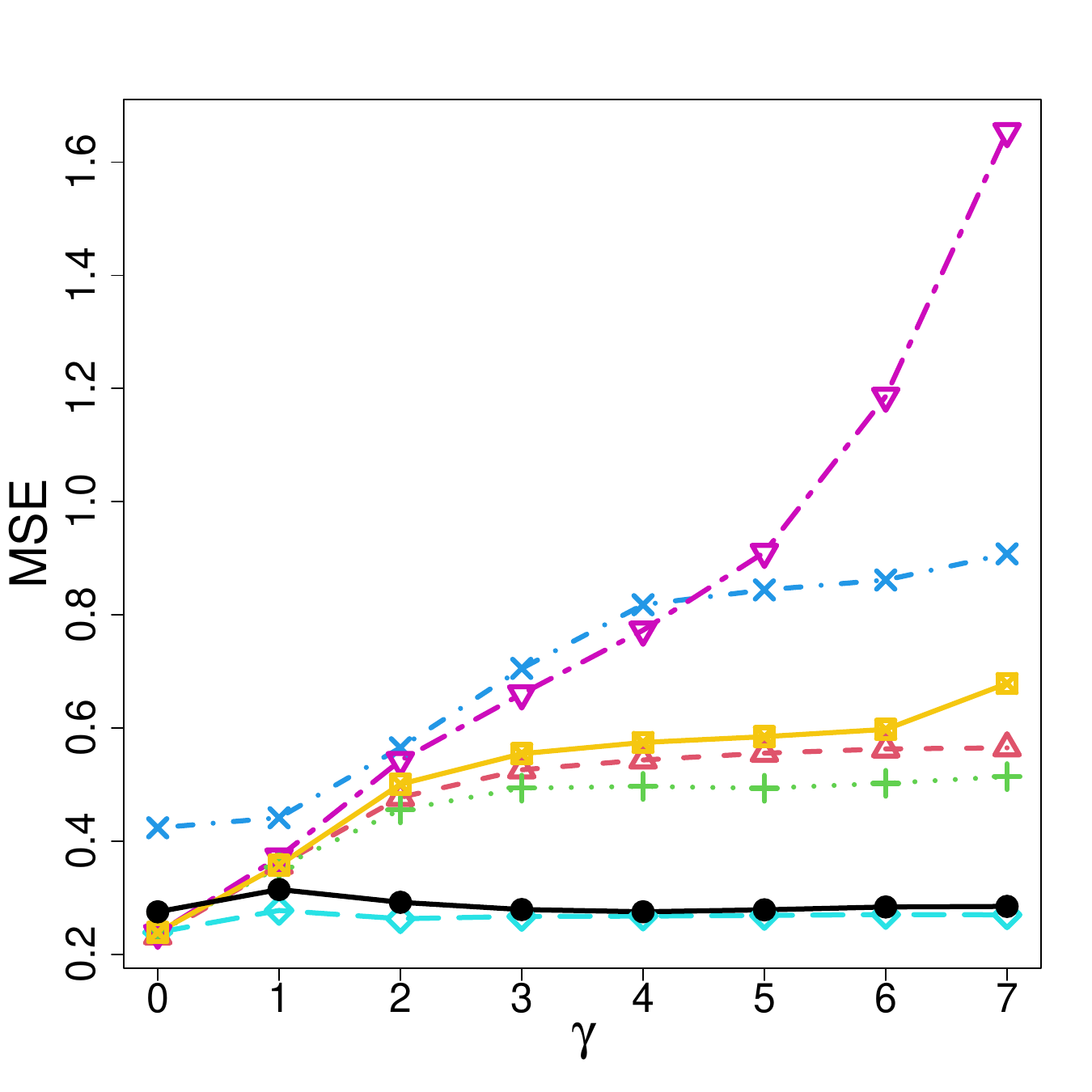} &
 \includegraphics[width=.31\textwidth]
 {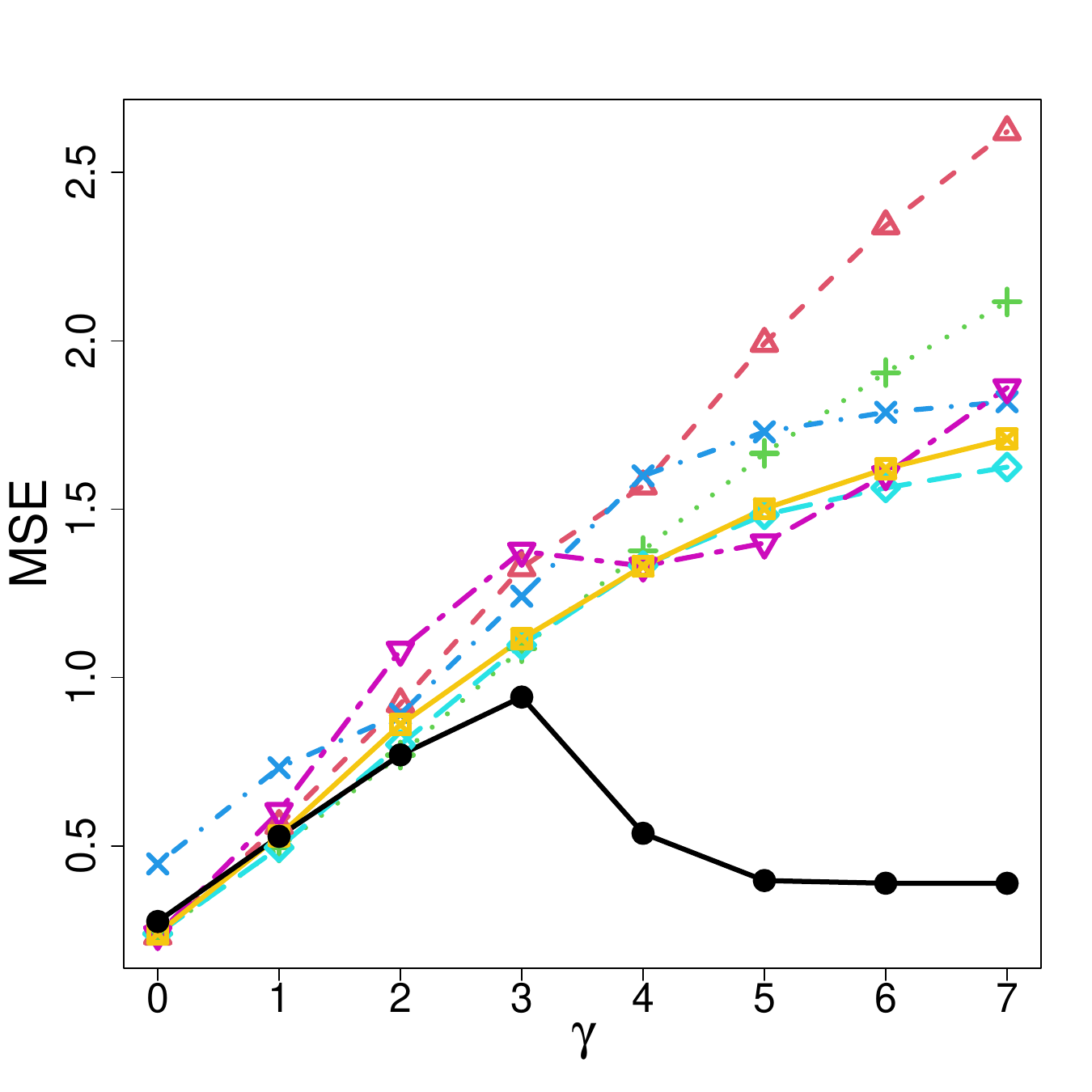} \\

 \rotatebox{90}{\textbf{\footnotesize{$\delta=0.2$}}}&
 \includegraphics[width=.31\textwidth]
 {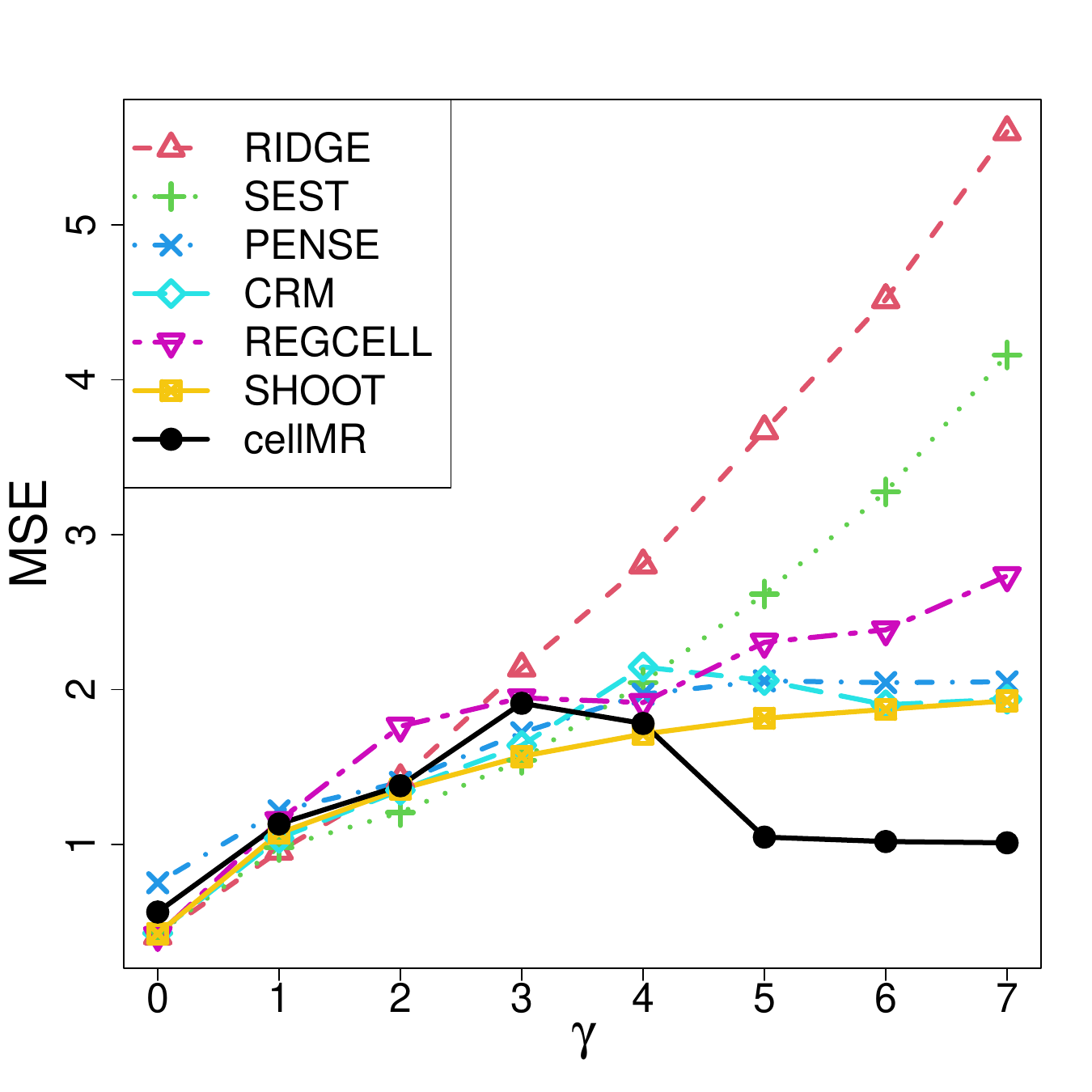} &
 \includegraphics[width=.31\textwidth]
 {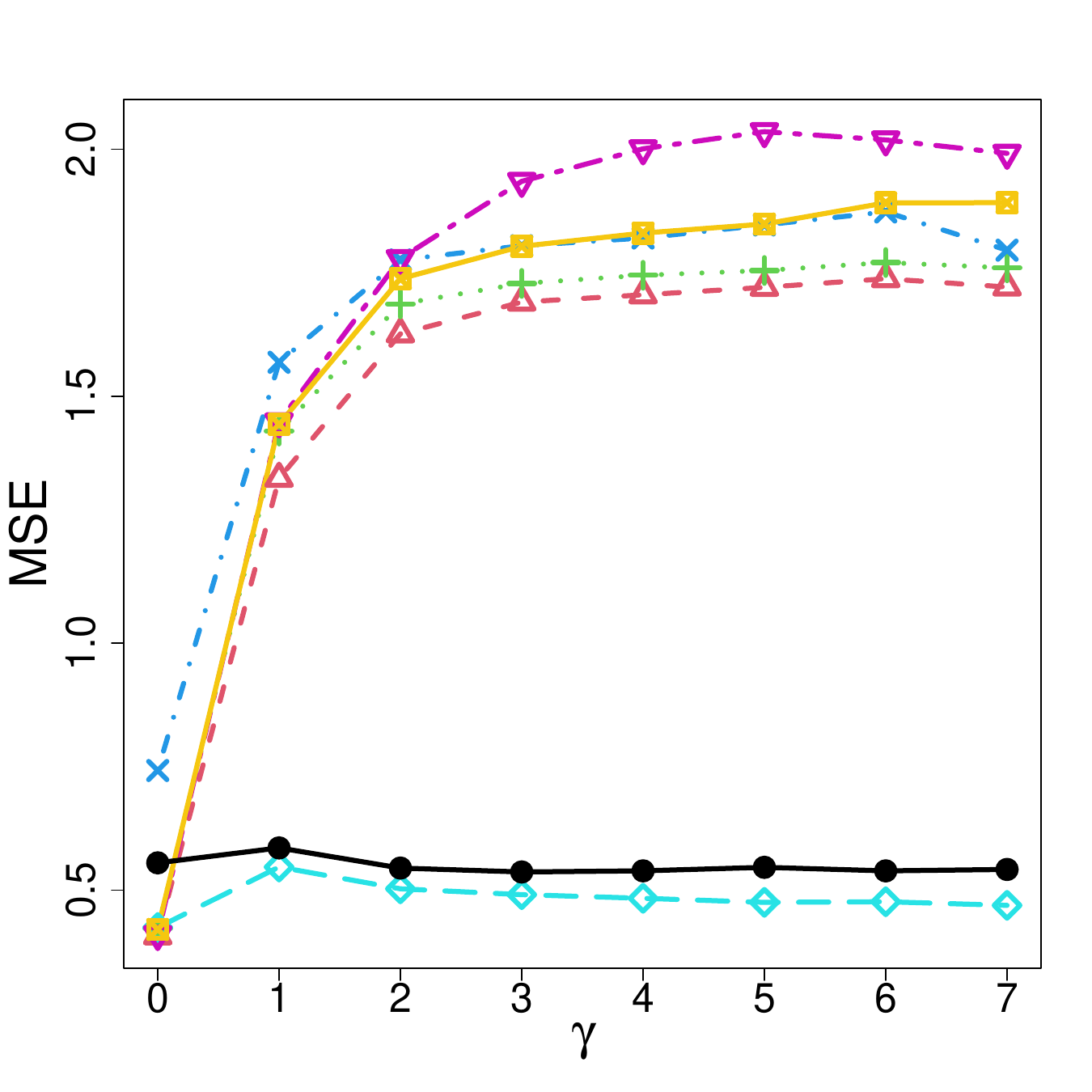} &
 \includegraphics[width=.31\textwidth]
 {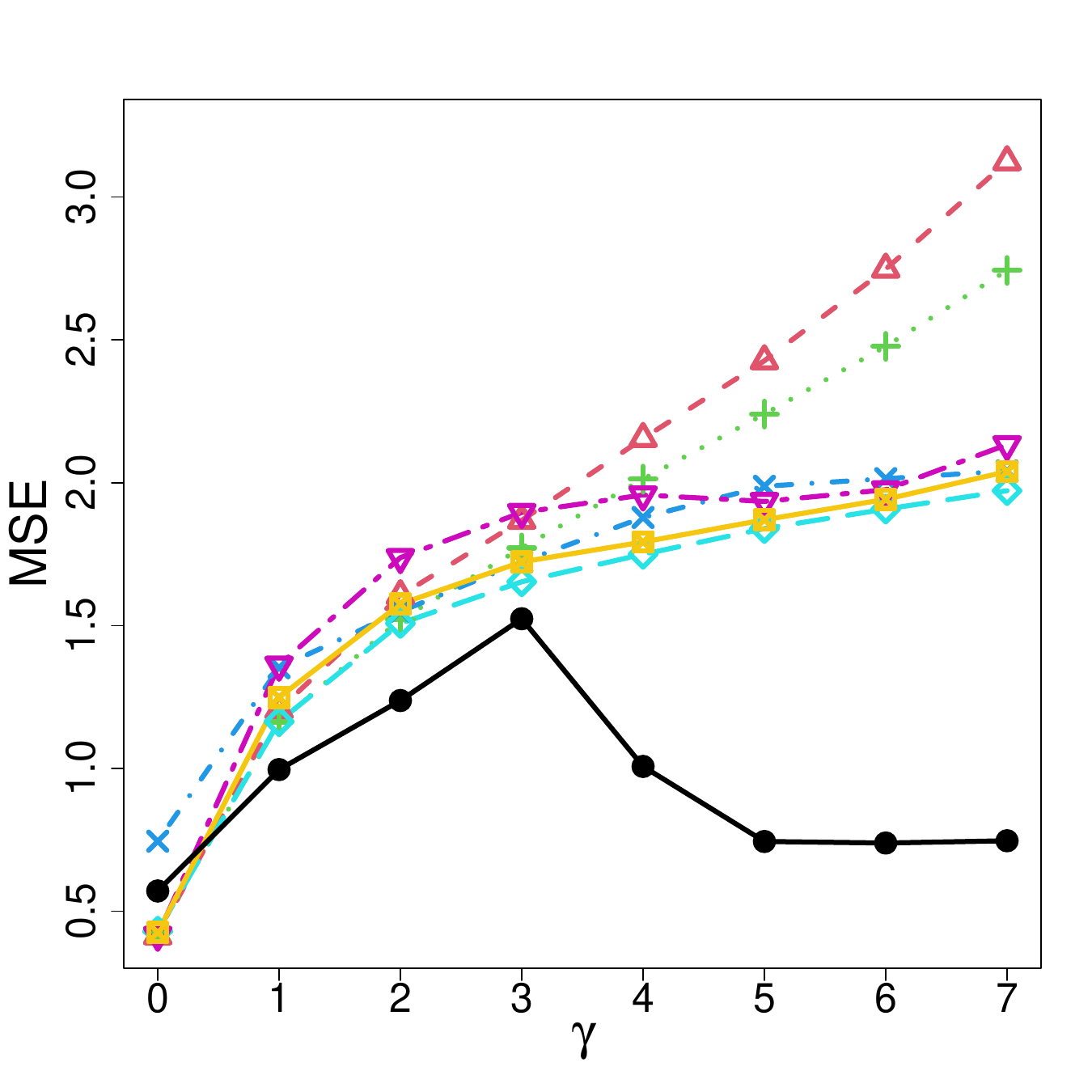}\\
\end{tabular}

\vspace{-5mm}
\caption{Average MSE attained by RIDGE, SEST,
PENSE, CRM, REGCELL, SHOOT, and cellMR under a 
quadratic departure from linearity with
$\delta\in\{0.1,0.2\}$ and $p=q=10$, in the presence 
of cellwise outliers, casewise outliers, or both, 
without missing data.}
\label{fig:results_misspecified}
\end{figure}

It is important to distinguish robustness to
contamination from robustness to model misspecification.
In robust statistics, outlyingness is inherently
defined relative to a reference model or structure, 
that is, robust methods aim to limit the effect of 
parts of the data that deviate substantially from 
the structure followed by the bulk of the data
\citep{rousseeuw2018robust}. Consequently,
when the reference structure is itself misspecified, 
genuine data generated from the clean distribution 
may also appear anomalous relative to that structure. 
For example, observations in regions of strong
curvature can behave like outliers when the fitted
regression structure does not adequately represent 
that curvature \citepsupp{mckean2000partial}. This is 
particularly relevant for cellMR because the
location and scatter of $\bZ$ are estimated
through cellCov, which in turn relies on cellPCA 
weights based on deviations from an estimated linear
principal subspace. Hence, sufficiently pronounced 
nonlinear structure in the clean joint distribution
may generate large reconstruction deviations
and may, in principle, be partially confounded 
with contamination. More generally, this issue is 
not specific to cellMR, but reflects a fundamental
aspect of robust statistics: outlyingness is defined
relative to a reference model, and robustness 
concerns departures from such an idealized model,
including both contamination and model 
misspecification \citep{hampel1986}.

At the same time, a departure from linearity
does not make a linear regression analysis 
meaningless. Even when $\E(\bY\mid\bX)$ is not 
exactly linear, the population least-squares
coefficient remains well defined as $\bB^\star=\bSigma_x^{-1}\bSigma_{xy}$,
which is the coefficient matrix of the best linear
approximation to the relationship between $\bY$ and 
$\bX$ in mean squared error \citep{white1980using}.
Thus, under model misspecification, $\bB^\star$ 
should not be interpreted as the coefficient 
matrix of the true conditional mean, which need not 
be linear, but rather as the population parameter 
of its best linear approximation. A comprehensive
discussion of the use and interpretation of
linear regression methods under model 
misspecification can be found in 
\citesupp{buja2019models1,BujaEtAl2019II}.

\section{More on the real data analysis}
\label{app_realdata}

Section~\ref{sec:realdata} analyzes the genomic 
dataset introduced in Section~\ref{sec:outdetect}, 
by modeling the expression levels of the three 
proteins MLH1, PRKCI, and RELA using the 50 
genetic predictors. The set of genes was selected
as those most strongly correlated with the protein responses. The problem is of interest because
proteins are the main drivers of cellular behavior 
and are frequently dysregulated in cancer.

There are both biological and technical reasons
for investigating potential outliers.
The data originate from the NCI-60 panel, a collection
of 60 human cancer cell lines from nine tissues 
maintained by the US National Cancer Institute.
One case had to be removed since all its values were
missing in the gene expression data, reducing the 
number of cases to $n = 59$.
The distinct molecular profiles in the remaining data
may lead some cell lines to depart from the 
gene-protein relationships followed by the majority.
Moreover, \citet{shankavaram2007transcript} discuss
potential sources of measurement discrepancies,
including incorrectly mapped oligonucleotide probes,
misidentified cDNA clones, and imperfect antibody
specificity. These issues could affect measurements
of particular genes or proteins without compromising
the entire expression profile of a cell line,
motivating consideration of cellwise as well as
casewise outliers.

Previous analyses provide evidence 
that robustness matters for this dataset.
\citet{alfons2013sparseLTS} modeled the expression
of a single protein, using 22,283 gene
expression predictors. Their method
flagged 13 of the 59 cell lines as outliers and
showed that including these cases changed
the predictor-response correlation structure.
These findings support the need for regression
methods that are robust to outliers, which can 
distort the estimated gene-protein relationships.

In Section~\ref{sec:realdata} we applied
cellMR and cellBoot in order to obtain 
prediction results, outlier diagnostics, and confidence
intervals. Together they provide insights into the 
gene-protein relationships and extend previous analyses
of this dataset. Whereas 
\citet{shankavaram2007transcript} studied
the correspondence between gene and protein
expression, and used expression profiles to predict
the tissue of origin of cancer cell lines,
we use gene expression measurements to predict the 
joint expression levels of three proteins. Relative to 
the casewise robust analysis with univariate response
by \citet{alfons2013sparseLTS}, our approach 
additionally distinguishes anomalous predictor cells 
from protein-specific departures from the regression
relationship. These distinctions help identify
whether a cell line has an unusual gene expression
profile or protein expression that is poorly
explained by that profile.
Furthermore, while \cite{alfons2013sparseLTS}
used sparse regression to select genes and obtain
an interpretable model, our cellBoot confidence
intervals quantify uncertainty in the estimated
gene-protein associations. In particular, intervals
excluding zero identify genes with evidence of
a conditional association with each protein response.
The predictive advantage of cellMR for retained
fractions $\alpha \leqslant 0.85$ is consistent with
the earlier finding that robust fitting can improve
prediction for the majority of cell lines.
However, the different responses, predictor sets,
and validation schemes preclude a direct numerical
comparison with the results of
\citet{alfons2013sparseLTS}.

{\small
\spacingset{1}
\renewcommand{\refname}{Supplementary References}

}

\stopcontents[supplement]

\end{document}